\documentclass{article}
\usepackage[utf8]{inputenc}
\usepackage{fullpage}
\usepackage{amsfonts}
\usepackage{amsmath}
\usepackage{amssymb}
\usepackage{amsthm}
\usepackage{enumitem}
\usepackage{bbm}

\setlist[itemize]{leftmargin=*}

\usepackage[colorlinks=true,linkcolor=blue,citecolor=blue]{hyperref}
\usepackage{natbib}
\usepackage{dashrule}

\usepackage[ruled,vlined]{algorithm2e}

\SetCommentSty{mycommfont}
\newenvironment{simulation}[1][htb]
  {% Update algorithm name
   \begin{algorithm}[#1]%
  }{\end{algorithm}}
  
\usepackage{longtable}

\usepackage{tikz}
\usetikzlibrary{bayesnet}

\newcommand{\sforall}{\;\forall\;}
\newcommand{\E}{\mathbb{E}}
\renewcommand{\Pr}{\mathbb{P}}
\newcommand{\Var}{\text{\normalfont Var}}

\newcommand{\argmin}[1]{\underset{#1}{\text{\normalfont arg min}}}
\newcommand{\argmax}[1]{\underset{#1}{\text{\normalfont arg max}}}

\newtheorem{assumption}{Assumption}
\newtheorem{theorem}{Theorem}
\newtheorem{proposition}{Proposition}
\newtheorem{corollary}{Corollary}
\newtheorem{remark}{Remark}
\newtheorem{definition}{Definition}
\newtheorem{lemma}{Lemma}

\usepackage{authblk}

\title{A Bayesian framework for change-point detection with uncertainty quantification}

\author[1]{Davis Berlind}
\author[2]{Lorenzo Cappello\thanks{L.C. acknowledges the support of Ayudas Fundacion BBVA a Proyectos de Investigacion Cientifica 2021, the Spanish Ministry of Economy and Competitiveness grant PID2022-138268NB-I00, financed by MCIN/AEI/10.13039/501100011033, FSE+MTM2015-67304-P, and FEDER, EU;  the grant Ramon y Cajal 2022 RYC2022-038467-I, financed by MCIN/AEI/10.13039/501100011033 and FSE+, and the Severo Ochoa Programme for Centres of Excellence in R\&D (Barcelona School of Economics CEX2019-000915-S), funded by MCIN/AEI/10.13039/50110001103}}
\author[3]{Oscar Hernan Madrid Padilla\thanks{O.H.M.P. was partially supported by the Hellman Fellowship, 2023-2024.}}
\affil[1]{Department of Statistics \& Data Science, University of California, Los Angeles}
\affil[2]{Department of Economics and Business, Universitat Pompeu Fabra; Data Science Center, Barcelona School of Economics}
\affil[3]{Department of Statistics \& Data Science, University of California, Los Angeles}

\date{\today}

\begin{document}

\numberwithin{equation}{section}

\maketitle

\begin{abstract}
    %Change-point detection is a well-studied problem in statistical inference dating back to the use of control charts during World War II. Despite this long history, there are only a few methods that provide measures of uncertainty for their estimated change-point locations. 
    We introduce a novel Bayesian method that can detect multiple structural breaks in the mean and variance of a length $T$ time-series. Our method quantifies uncertainty by returning $\alpha$-level credible sets around the estimated locations of the breaks. In the case of a single change in the mean and/or the variance of an independent sub-Gaussian sequence, we prove that our method attains a localization rate that is minimax optimal up to a $\log T$ factor. For an $\alpha$-mixing sequence with dependence, we prove this optimality holds up to $\log^2 T$ factor. For $d$-dimensional mean changes, we show that if $d \gg \log T$ and the mean signal is dense, then our method exactly recovers the location of the change. Our method detects multiple change-points by modularly ``stacking'' single change-point models and searching for a variational approximation to the posterior distribution. This approach is applicable to both continuous and count data. Extensive simulation studies demonstrate that our method is competitive with the state-of-the-art in terms of speed and performance, and we produce credible sets that are an order of magnitude smaller than our competitors without sacrificing nominal coverage guarantees. We apply our method to real data by detecting i) the gating of an ion channel in the outer membrane of a bacterial cell, and ii) changes in the lithological structure of an oil well. 

\end{abstract}

\section{Introduction}
\label{sec:intro}

Change-point detection (CPD) has been a perennial topic of interest in statistical inference since the introduction of the CUSUM algorithm by \cite{Page54}. Broadly, CPD involves identifying the locations of structural breaks in an ordered sequence of data $\mathbf{y}_{1:T}:= \{\mathbf{y}_t\}_{t=1}^T$ where $\mathbf{y}_t \in \mathbb{R}^d$. Suppose there are $L$ indices $\boldsymbol{\tau}_{1:L}\subset \{1,\ldots,T\}$, with $\tau_0:=1 < \tau_1 < \ldots< \tau_L < \tau_{L+1} := T+1$, and a collection of $L+1$ distributions $\{F_\ell\}_{\ell=0}^L$ such that $F_\ell \neq F_{\ell+1}$ and:
\begin{align}\label{eq:cdp-def}
    \mathbf{y}_t \sim F_\ell, \;\sforall t \in [\tau_\ell,\tau_{\ell+1}).
\end{align}
The aim of any CPD method is to consistently estimate the number of changes $L$ and their locations $\boldsymbol{\tau}_{1:L}$. Due to the generality of this construction, CPD problems appear in a variety of fields. Examples include: detecting adjustments to the real interest rate (\citealp{Bai03}), structural changes in a DNA sequence (\citealp{Muggeo11}), radiological anomalies \citep{madrid2019sequential}, and deforestation events (\citealp{Wendelberger21}). In addition to estimating the number and the locations of the changes underlying $\mathbf{y}_{1:T}$, a limited set of methods can return localized regions of $\{1,\ldots,T\}$ that contain change-points at a prescribed significance level. Such measures of uncertainty are essential in applications like the design of medical procedures. For example, \cite{Gao19} use a CPD method to determine whether a human liver is viable for transplant. In this case, knowing whether the estimated change-point is likely to be off by minutes versus hours will materially affect the outcome of the procedure. 

Early attempts to provide confidence sets for change-point estimates were either limited to the case of a single mean change (\citealp{Siegmund86, Worsley86,Jirak15,Horvath17}), required knowledge of $L$ (\citealp{Bai03}), or could only produce approximate sets based on the limiting distribution of the estimator (\citealp{Bai10}). Since the introduction of SMUCE by \cite{Frick14}, the state-of-the-art has advanced rapidly. There are now methods that can generate $\alpha$-level confidence sets in the presence of a variety of types of changes, data structures, and dependence settings (\citealp{Pein17, Eichinger18, Dette20, Fang20, fang2021detectionestimationlocalsignals, Chen22, Cho22, madrid2023change, Fryzlewicz24median, Fryzlewicz24}). Still, there remains room for improvement. \cite{chen2014discussion} note that the sets returned by SMUCE exhibit undesirable coverage properties as $\alpha$ decreases. \cite{Fryzlewicz24} attempted to address this issue, but like SMUCE, the large confidence sets returned by their method tend to be overly conservative. Additionally, methods capable of producing confidence sets in the case of multivariate data or changes in the variance of $\mathbf{y}_{1:T}$ remain underdeveloped.

Bayesian CPD (BCPD) traces its origins to \cite{chernoff1964estimating} and can offer a parsimonious framework for uncertainty quantification by fully characterizing the posterior distribution of $\{\boldsymbol{\tau}_{1:L}, L\}$. Important early developments came from \cite{Barry93}, who formulated (\ref{eq:cdp-def}) as a product partition model (\citealp{Barry92}) and used closed-form recursions to efficiently draw from the posterior distribution. Another popular approach to BCPD involves introducing a latent change-point indicator $z_t$ and using techniques from the HMM literature to simulate the posterior distribution of $\mathbf{z}_{1:T}$ (\citealp{Chib98, Fearnhead06, Nam12, Harle16, Fan17}). Despite these innovations, BCPD methods remain an order of magnitude slower than state-of-the-art approaches, and tend to lack theoretical development (a notable exception being \cite{Kim24}). Additionally, the posteriors returned by existing methods characterize the joint distribution of all the change-point locations; however, most methods do not offer a clear approach for condensing this information into meaningful summaries like credible sets. In practice, the majority of BCPD methods simply return the marginal posterior probability that any given index $t$ is a change-point. % We desire a method that inverts this logic, i.e. in place of the probability that $t$ is a change-point, we would like to know the probability that a given change-point is equal to $t$. 

\subsection{List of Contributions}

\textbf{MICH}: %We aim to design a method that is computationally efficient, can estimate and effectively localize a broad class of changes in the structure of $\mathbf{y}_{1:T}$, and has good theoretical guarantees. To this end, 
We introduce the Multiple Independent CHange-point (MICH) model, a Bayesian method that can identify changes in the mean and variance of $\mathbf{y}_{1:T}$. We construct MICH by modularly combining many single change-point models, a concept that was first introduced for mean changes by \cite{Wang20} and for variance changes by \cite{Cappello25}. This modular construction enables us to efficiently approximate the posterior distribution of the change-points using a variational algorithm (\citealp{Jordan99, Blei17}) that we implement as a deterministic backfitting procedure (\citealp{Breiman85}). Unlike other Bayesian methods, MICH requires virtually no hyperparameter tuning. In simulations, MICH is competitive with state-of-the-art models at little additional computational cost. In fact, MICH returns credible sets that are an order of magnitude smaller than the sets returned by competing methods without sacrificing nominal coverage guarantees.

%Each component of MICH maps to a single $\alpha$-level credible set, endowing the model with our desired localization behavior. Unlike other Bayesian methods, MICH requires virtually no hyperparameter tuning, making it readily accessible to analysts in a variety of fields. In simulations, MICH is competitive with state-of-the-art models at little additional computational cost. In fact, our method often returns credible sets that are an order of magnitude smaller than competitors without sacrificing nominal coverage guarantees.

\textbf{Minimax Optimal Localization}: We show that our promising empirical results are partially explained by the optimal consistency properties of the single change-point models that make up MICH. For univariate $\mathbf{y}_{1:T}$ with a single change-point, we show that for the cases of a mean-only change, a variance-only change, and a simultaneous mean and variance change, our model achieves the minimax localization rate up to a logarithmic factor. For the mean-only and variance-only cases, this result matches the optimal rates recently achieved by the Bayesian models of \cite{Liu17}, \cite{Cappello21}, and \cite{Kim24}. For the simultaneous mean and variance change case, this optimal rate is a novel result for a Bayesian CPD method, the proof of which relies on a careful treatment of the many non-linear terms in our model's likelihood.

\textbf{Non-Gaussian and Dependent Data}: Though our primary model assumes i.i.d. Gaussian data, we show that our idea to combine multiple change-point models extends to other data types. In Appendix \ref{app:poisson} we show how to adapt MICH to Poisson data that exhibits change-points in the rate. We also show that the single change-point localization rates remain unchanged when the entries of $\mathbf{y}_{1:T}$ are independent and sub-Gaussian, somewhat in-line with the robustness to misspecification observed for other Bayesian methods (\citealp{lai2011simple}). Beyond independence, we generalize our results to univariate $\alpha$-mixing sequences. Using novel concentration results from \cite{Padilla23}, we show that MICH achieves a localization rate of approximately order $\mathcal{O}(\log^2 T)$ for dependent data. 

\textbf{High-Dimensional Mean Changes}: In the multivariate setting, MICH can localize a single mean change at the minimax optimal rate when $\min\{\tau_1,T-\tau_1\}\lVert\E[\mathbf{y}_{\tau_1} - \mathbf{y}_{\tau_1 - 1}]\rVert^2_2 \gg d\log T$. In high-dimensions, this condition is satisfied when we observe a ``dense" signal, i.e. many coordinates of $\E[\mathbf{y}_t]$ undergo small changes. To the best of our knowledge, MICH is the first Bayesian method to achieve the optimal rate in high-dimensions in the presence of dense mean changes. Furthermore, when $d \gg \log T$, we recover the result of \cite{Bai10} by showing that our estimator converges to the true location of the change.

\subsection{Related Work}

Frequentist methods for detecting changes in the mean and variance of $\mathbf{y}_{1:T}$ broadly fall into two categories: i) greedy algorithms like Binary Segmentation and its variants that search for changes by recursively partitioning $\mathbf{y}_{1:T}$ (\citealp{Scott74, Sen75, Vostrikova81, Olshen04, Fryzlewicz14, Kovacs22}), and ii) exact search methods that jointly estimate the locations of the changes by solving an optimization problem, often using dynamic programming (\citealp{Auger89, Killick12, Baranowski19, Padilla22}). There is also a vast literature for the ``on-line" setting, where the task is to detect changes as each $\mathbf{y}_t$ is sequentially observed. We refer the reader to \cite{Namoano19} and \cite{Yu02102023} for recent surveys of on-line methods.

The high-dimensional CPD literature is still nascent and is mostly concerned with detecting sparse signals where $\tilde{d} = o(d)$ coordinates of the mean vector jump at a given change-point (\citealp{Cho14, Jirak15, Wang17, Enikeeva19, Yu20, Chen22}). Most high-dimensional methods employ $\ell^\infty$-norm based aggregations of the data to leverage this sparsity, which differs from the $\ell^2$-norm aggregation used by our method (see \cite{Li23} for a recent survey of $\ell^2$ methods). Only \cite{Jirak15} and \cite{Chen22} characterize the limiting distribution of their estimators and provide confidence sets for the high-dimensional regime. \cite{Kim24} recently developed the first Bayesian high-dimensional method capable of localizing sparse mean and covariance changes at the minimax optimal rate; however, they approach the problem from a testing perspective and do not provide credible sets.

While our study focuses on localizing change-points with credible sets, we note that uncertainty quantification for CPD is often approached as a problem of hypothesis testing and post-selection inference (\citealp{Horvath12,Preuss15,Duy20,Hyun21,Jewell22}). This approach has received considerable attention in the literature and we direct the reader to \cite{Fryzlewicz24} for a recent survey and discussion of the limitations of testing and post-selection inference for CPD. 

\subsection{Outline of the Paper}

The organization of this article is as follows: In Section \ref{sec:scp}, we introduce three Bayesian models for localizing a single change in either the mean, variance, or both. We also characterize the localization rates for each of these models. In Section \ref{sec:mich}, we introduce the MICH model for multiple change-points. We develop a backfitting procedure in Algorithm \ref{alg:mich} that facilitates inference and returns a variational approximation to the true posterior distribution under MICH. In Section \ref{sec:simulations}, we conduct an extensive simulation study to compare the performance of MICH with state-of-the art competitors. In Section \ref{sec:data}, we apply MICH to electric current recordings from an ion channel in a bacterial cell and to multivariate lithological data from oil well-log. In Section \ref{sec:discussion}, we conclude and discuss future directions. An \texttt{R} package implementing MICH is available at \url{https://github.com/davis-berlind/MICH}. All proofs are deferred to the \hyperlink{appendix}{Appendix}. 

\subsection{Notation}

Throughout this article, we denote the index set $[T] := \{1, \ldots, T\}$ for $T\in \mathbb{N}$. For a set $S$, we denote its cardinality by $|S|$. For an event $E$, we define the indicator $\mathbbm{1}_{\{E\}}$ function to be equal to one when $E$ occurs and zero otherwise. For a sequence $\{x_t\}_{t\in\mathbb{Z}}$ and $r,s \in \mathbb{Z}$ with $r<s$, we define $\mathbf{x}_{r:s} := \{x_t\}_{t=r}^s$. For $\mathbf{x}\in\mathbb{R}^d$ and $p \geq 1$, we denote the $p$-norm $\lVert \mathbf{x} \rVert_p := (\sum_{i=1}^d |x_i|^p)^{1/p}$ and let $\mathcal{S}^d := \{\mathbf{x}\in\mathbb{R}^d_+ \::\: \lVert \mathbf{x} \rVert_1 =1\}$ be the $d$-dimensional probability simplex. For $\boldsymbol{\pi} \in \mathcal{S}^d$, we say $\tau \sim \text{Categorical}(\boldsymbol{\pi})$ if $\tau \in [d]$ and $\Pr(\tau = i) = \pi_i$ for each $i \in [d]$. For a pair of random variables $X$ and $Y$, we write $\{X,Y\}\sim\text{Normal-Gamma}(\mu,\lambda,\alpha,\beta)$ if $X|Y\sim \mathcal{N}(\mu,(\lambda Y)^{-1})$ and $Y\sim\text{Gamma}(\alpha,\beta)$. We write $\E_{g} [X]$ to denote the expectation taken with respect to a generic distribution $g$. For $n> 0$, let $\lVert X \rVert_{\psi_n} := \inf\{t > 0 : \E[\exp[(X /t)^n]] \leq 2\}$ be the $n$-Orlicz norm of $X$ and let $X\in\mathcal{SG}(\sigma)$ and $X\in\mathcal{SE}(\nu,\alpha)$ respectively indicate that $X$ is sub-Gaussian with parameter $\sigma$ or sub-exponential with parameters $\nu$ and $\alpha$. For functions $f$ and $g$, we write $f(T) = \mathcal{O}(g(T))$ if there exists some constant $M > 0$ and some $T^* > 0$ so that $|f(T)| \leq M g(T)$ for all $T > T^*$. Similarly, we write $g(T) \gtrsim f(T)$ if there exists $M > 0$ so that $f(T) \leq M g(T)$ for any $T$. Lastly, we write $f(T) = o(g(T))$, or equivalently $g(T) \gg f(T)$, when $\lim_{T\to\infty} \frac{f(T)}{g(T)} = 0$ 
\section{Single Change-Point (SCP) Models}
\label{sec:scp}

We begin by introducing three single change-point (SCP) models that will form the basic components of the multiple change-point method in Section \ref{sec:mich}. Suppose that we have $T$ observations of a $d$-dimensional time-series $\mathbf{y}_{1:T}$ where:
\begin{align}\label{eq:dgp}
    \mathbf{y}_t \:|\: \boldsymbol{\mu}_t, \boldsymbol{\Lambda}_t \overset{\text{ind.}}{\sim} \mathcal{N}_d\left(\boldsymbol{\mu}_t, \boldsymbol{\Lambda}^{-1}_t\right), \;\sforall t \in [T].
\end{align}
We assume that either $\boldsymbol{\mu}_{1:T}$, $\boldsymbol{\Lambda}_{1:T}$, or both can be decomposed into a known trend and an unknown piece-wise constant structure with a single change occurring at some unknown time $\tau \in [T]$, where:
\begin{align}
    \tau \sim \text{Categorical}(\boldsymbol{\pi}_{1:T}), \; \boldsymbol{\pi}_{1:T} \in \mathcal{S}^T. \label{eq:tau-cat}
\end{align}
Then, the posterior distribution of $\tau$ is given by:
\begin{align}
    \tau \:|\: \mathbf{y}_{1:T} &\sim \text{Categorical}(\overline{\boldsymbol{\pi}}_{1:T}), \label{eq:gamma-post-cat1} \\ 
    \overline{\pi}_t &\propto  \pi_t p(\mathbf{y}_{1:T} \;|\; \tau = t). \label{eq:gamma-post-cat2}
\end{align}
By choosing conditionally conjugate prior distributions for the jumps in $\boldsymbol{\mu}_{1:T}$ and $\boldsymbol{\Lambda}_{1:T}$, we arrive at closed-form expressions for $\overline{\pi}_t$ (see Appendix \ref{app:posterior-parameters}). Given $\overline{\boldsymbol{\pi}}_{1:T}$, a natural point-estimate for the location of the change-point is the posterior mode of $\tau$, i.e. the maximum \textit{a posteriori} (MAP) estimator:
\begin{align}\label{eq:map}
    \hat{\tau}_{\text{MAP}} := \argmax{1 \leq t \leq T} \; \overline{\pi}_t.
\end{align}
We give a detailed characterization of the consistency properties of $\hat{\tau}_{\text{MAP}}$ in Section \ref{sec:localization}. The posterior probabilities $\overline{\boldsymbol{\pi}}_{1:T}$ also facilitate our goal of quantifying the uncertainty around $\tau$ with $\alpha$-level credible sets, which we construct by solving the following optimization problem:
\begin{align}\label{eq:cs}
    \mathcal{CS}(\alpha, \overline{\boldsymbol{\pi}}_{1:T}) := \argmin{S \subseteq[T]} |S| \;\text{ s.t. } \sum_{t \in S} \overline{\pi}_t \geq 1-\alpha.
\end{align}
\begin{remark}[Knapsack Problem]\label{rmk:knapsack}
    The optimization task (\ref{eq:cs}) is equivalent to solving the integer program $\mathbf{z} := \text{\normalfont arg min}_{\mathbf{x}\in\{0,1\}^T} \; \langle\mathbf{1}, \mathbf{x}\rangle \text{ s.t. } \langle \overline{\boldsymbol{\pi}}_{1:T}, \mathbf{x}\rangle \geq 1-\alpha$, and setting $\mathcal{CS}(\alpha, \overline{\boldsymbol{\pi}}_{1:T}) := \{t \in [T]: z_t = 1\}$. This is an example of a knapsack problem (\citealp{Santini24}), which we solve by adding indices to $\mathcal{CS}(\alpha, \overline{\boldsymbol{\pi}}_{1:T})$ in decreasing order of the value of $\overline{\boldsymbol{\pi}}_{1:T}$ until $\sum_{t \in \mathcal{CS}(\alpha, \overline{\boldsymbol{\pi}}_{1:T})} \overline{\pi}_t$ exceeds $1-\alpha$. 
\end{remark}
\vspace{-10pt}

\subsection{Mean Change-Point Model}
\label{sec:smcp}

Suppose that we have $T$ observations from (\ref{eq:dgp}) where the sequence of positive definite precision matrices $\boldsymbol{\Lambda}_{1:T}$ is known. Given $\tau$ from (\ref{eq:tau-cat}) and a prior precision parameter $\omega_0 > 0$, we place a Gaussian prior on the jump in $\boldsymbol{\mu}_{1:T}$ to complete the mean single change-point (mean-scp) model:
\begin{align} \label{eq:smcp-start}
    \boldsymbol{\mu}_t &= \mathbf{b}\mathbbm{1}_{\left\{t\geq \tau \right\}}, \\
    \mathbf{b} &\sim \mathcal{N}_d(\mathbf{0},\omega_0^{-1} \mathbf{I}_d). \label{eq:smcp-end}
\end{align}
Under (\ref{eq:smcp-start}), the coordinates of $\mathbf{y}_{t}$ begin centered at zero and jump by $\mathbf{b}\in\mathbb{R}^d$ at time $\tau$. Choosing to construct $\boldsymbol{\mu}_{1:T}$ in this way and include time varying precision terms $\boldsymbol{\Lambda}_{1:T}$ may initially appear convoluted, but this formulation will become key to facilitating multiple change-point inference in Section \ref{sec:mich}. For each $t \in [T]$, we arrive at the following closed forms for the posterior distribution and expectation of $\boldsymbol{\mu}_t$:
\begin{align}
    \mathbf{b} \:|\: \tau = t, \: \mathbf{y}_{1:T} &\sim \mathcal{N}_d\left(\overline{\mathbf{b}}_{t}, \overline{\boldsymbol{\Omega}}_{t}^{-1}\right), \label{eq:b-smcp} \\ 
    \E[\boldsymbol{\mu}_{t} \;|\; \mathbf{y}_{1:T}] &= \sum_{t'=1}^t \overline{\mathbf{b}}_{t'} \overline{\pi}_{t'}. \label{eq:mu-post-mean}
\end{align}
The posterior mean $\overline{\mathbf{b}}_{t}$, precision $\overline{\boldsymbol{\Omega}}_{t}$, and $\overline{\pi}_t$ are calculated as in (\ref{eq:mean-scp-post-omega})-(\ref{eq:mean-scp-post-pi}). When $d = 1$, the mean-scp model is identical to the SER model of \cite{Wang20} applied to a $T\times T$ covariate matrix $\mathbf{X}_T$ with $(\mathbf{X}_T)_{ij} = \mathbbm{1}_{\{i \geq j\}}$. Later, when we generalize to multiple change-points, it will be convenient to define a function that takes $\mathbf{y}_{1:T}$ and the model parameters as inputs and returns the posterior parameters: 
\begin{align}\label{eq:mean-scp-fn}
    \texttt{mean-scp}\left(\mathbf{y}_{1:T} \:;\: \boldsymbol{\Lambda}_{1:T}, \omega_0, \boldsymbol{\pi}_{1:T}\right) := \{\overline{\mathbf{b}}_t, \overline{\boldsymbol{\Omega}}_t, \overline{\pi}_t\}_{t=1}^T.
\end{align}

\subsection{Variance Change-Point Model}
\label{sec:sscp}

In this section, we restrict $d = 1$ so that $y_t$ is univariate and $\Var(y_t) := \lambda_t^{-1}$. We assume that $\boldsymbol{\mu}_{1:T}$ is known, therefore it is without loss of generality to assume $\mu_{t} \equiv 0$. We now introduce a positive vector $\boldsymbol{\omega}_{1:T}$ that represents the known trend component of $\boldsymbol{\lambda}_{1:T}$. The inclusion of $\boldsymbol{\lambda}_{1:T}$ will help facilitate multiple change-point detection in Section \ref{sec:mich}. Given $\tau$ from (\ref{eq:tau-cat}) and prior shape and rate parameters $u_0, v_0 > 0$, we place a gamma prior on the jump in $\boldsymbol{\lambda}_{1:T}$ to complete the variance single change-point (var-scp) model:
\begin{align}\label{eq:sscp-start}
    \lambda_t &= \omega_t s^{\mathbbm{1}\{t \geq \tau\}}, \\
    s &\sim \text{Gamma}(u_0,v_0).
    \label{eq:sscp-end}
\end{align}
Suppose that $\tau > t$, then $s^{\mathbbm{1}\{t \geq \tau\}} = 1$, and if $\tau \leq t$, then we have $s^{\mathbbm{1}\{t \geq \tau\}} = s$. Thus the construction in (\ref{eq:sscp-start}), decomposes $\lambda_t$ into the known component $\omega_{t}$ and a piecewise constant component equal to one or $s$. When $\omega_t \equiv 1$, then the var-scp model recovers the single variance change-point model of \cite{Cappello25}. Again, the conjugate prior on $s$ results in the following closed forms for the posterior distribution and expectation of $\lambda_t$:
\begin{align}
    s \:|\: \tau = t, \: \mathbf{y}_{1:T} &\sim \text{Gamma}\left(\overline{u}_{t}, \overline{v}_{t}\right), \label{eq:s-sscp} \\
    \E[ \lambda_{t}\;|\;\mathbf{y}_{1:T}] &= \omega_t\left[\sum_{t'=1}^{t} \frac{\overline{u}_{t'}\overline{\pi}_{t'}}{\overline{v}_{t'}}  + \sum_{t'=t+1}^T \overline{\pi}_{t'}\right]. \label{eq:lambda-post-mean}
\end{align}
The posterior shape $\overline{u}_{t}$, rate $\overline{v}_{t}$, and $\overline{\pi}_t$ are calculated as in (\ref{eq:var-scp-post-u})-(\ref{eq:var-scp-post-pi}). As with the mean-scp model, we define a function \texttt{var-scp}: 
\begin{align}\label{eq:var-scp-fn}
    \texttt{var-scp}\left(\mathbf{y}_{1:T} \:;\: \boldsymbol{\omega}_{1:T}, u_0, v_0, \boldsymbol{\pi}_{1:T}\right) := \{\overline{u}_t, \overline{v}_t, \overline{\pi}_t\}_{t=1}^T.
\end{align}

\subsection{Mean-Variance Change-Point Model}
\label{sec:smscp}

We now merge the mean-scp and var-scp settings by allowing $\mathbf{y}_{1:T}$ to have a single change-point where both the mean and variance shift simultaneously. We again restrict $d=1$ and define $\mu_t$ as in (\ref{eq:smcp-start}) and $\lambda_t$ as in (\ref{eq:sscp-start}), only now $\mu_t$ and $\lambda_t$ share the same $\tau$ from (\ref{eq:tau-cat}) and we place a joint prior on $\{b,s\}$ to specify the mean-variance single change-point (meanvar-scp) model:
\begin{align}
    \{b,s\} &\sim \text{Normal-Gamma}(0,\omega_0, u_0, v_0).
    \label{eq:smscp-end}
\end{align}
A version of this model with $\omega_t\equiv 1$ first appeared in \cite{Smith75}, though the constructions (\ref{eq:smcp-start}) and (\ref{eq:sscp-start}) were notably not present. In addition to the posterior distribution of $\{b,s\}$, we provide a closed form for the posterior correlation between $\mu_t$ and $\lambda_t$, given that they are now dependent under the meanvar-scp model:
\begin{align}
    \{b,s\} \:|\: \tau = t, \mathbf{y}_{1:T} &\sim \text{Normal-Gamma}(\overline{b}_t, \overline{\omega}_t, \overline{u}_t, \overline{v}_t), \label{eq:bs-smscp} \\
   \E[\lambda_t\mu_t\;|\; \mathbf{y}_{1:T}] &= \omega_t\sum_{t'=1}^t \frac{\overline{b}_{t'}\overline{u}_{t'}\overline{\pi}_{t'} }{\overline{v}_{t'}}. \label{eq:mu-lambda-post-mean}
\end{align}
Closed forms for the parameters $\overline{b}_t$, $\overline{\omega}_t$, $\overline{u}_{t}$, $\overline{v}_{t}$, and $\overline{\pi}_t$ are given in (\ref{eq:meanvar-scp-post-omega})-(\ref{eq:meanvar-scp-post-pi}). Note that the posterior means of $\mu_t$ and $\lambda_t$ are of the same form as in (\ref{eq:mu-post-mean}) and (\ref{eq:lambda-post-mean}). As before, we define a function \texttt{meanvar-scp}: 
\begin{align}\label{eq:meanvar-scp-fn}
    \texttt{meanvar-scp}\left(\mathbf{y}_{1:T} \:;\: \boldsymbol{\omega}_{1:T}, \omega_0, u_0, v_0, \boldsymbol{\pi}_{1:T}\right) := \{\overline{b}_t, \overline{\omega}_t, \overline{u}_t, \overline{v}_t, \overline{\pi}_t\}_{t=1}^T.
\end{align}

\subsection{Single Change-Point Theory}
\label{sec:localization}
 
In this section, we study the asymptotic behavior of $\hat{\tau}_{\text{MAP}}$ as defined in (\ref{eq:map}) for each of the SCP models. Suppose that $t_0 \in [T]$ is the true location of the change-point and define the \textit{minimum spacing condition} $\Delta_T := \min\{t_0,T-t_0 + 1\}$. Our goal is to characterize the smallest \textit{localization error} $\{\epsilon_T\}_{T\geq 1}$, as well as a minimal set of conditions on $\Delta_T$ and the size of the breaks in the mean and variance of $\mathbf{y}_{1:T}$ so that $\hat{\tau}_{\text{MAP}}$ is consistent in following sense:
\begin{align}
    \lim_{T\to\infty} \Pr\left(|t_0 - \hat{\tau}_{\text{MAP}}| \leq \epsilon_T\right) = 1 \text{ and } \lim_{T\to\infty} \frac{\epsilon_T}{\Delta_T} = 0. \label{def:loc-rate}
\end{align}
As is standard throughout the CPD literature, we refer to the quantity $\Delta_T^{-1}\epsilon_T$ as the \textit{localization rate}. 

\subsubsection{Mean-SCP Localization Rate in High-Dimensions}
\label{sec:smscp-theory}

Suppose that $\boldsymbol{\mu}_{1:T}$ jumps by $\mathbf{b}_0\in\mathbb{R}^d$ at time $t_0$, then the mean-scp model can detect this change provided that the aggregated jump-size is large enough as measured by the signal-to-noise ratio $\Delta_T\min_{1\leq t \leq T}\lVert \boldsymbol{\Lambda}_t^{1/2}\mathbf{b}_0 \rVert^2_2$. Assumption \ref{assumption:mean} formalizes this condition.

\begin{assumption}[Detectable Mean Change]\label{assumption:mean}  
    Suppose $\E[\mathbf{y}_t] = \mathbf{b}_0\mathbbm{1}_{\{t \geq t_0\}}$ for some $t_0 \in [T]$ and $\mathbf{b}_0\in\mathbb{R}^d$. Assume that $\Delta_T \gtrsim\log T$ and $\Delta_T\min_{1\leq t \leq T}\lVert \boldsymbol{\Lambda}_t^{\frac{1}{2}}\mathbf{b}_0 \rVert^2_2 \gg d\log T$, where $\normalfont{\Var}(\mathbf{y}_t) = \boldsymbol{\Lambda}^{-1}_t$.
\end{assumption}
\vspace{-5pt}

The assumption that $\E[\mathbf{y}_t] = \mathbf{0}$ for $t < t_0$ is made out of notational convenience. In reality, only the magnitude of the jump at $t_0$ matters and the results in this section still hold if $\E[\mathbf{y}_{t_0-1}]\neq \mathbf{0}$ and we replace $\mathbf{b}_0$ with $\E[\mathbf{y}_{t_0} -\mathbf{y}_{t_0-1}]$ in Assumption \ref{assumption:mean}. 

When $d=1$, \cite{Wang2020_localization} showed that consistent localization in the (\ref{def:loc-rate}) sense is not possible if $\Delta_T\min_{1\leq t \leq T}\lVert \boldsymbol{\Lambda}_t^{1/2}\mathbf{b}_0 \rVert^2_2 \lesssim d\log T$, making Assumption \ref{assumption:mean} a necessary condition for detection when only the mean of $\mathbf{y}_{1:T}$ changes. When $\min\{d,T\}\to\infty$ and $\lVert \mathbf{b}_0\rVert_\infty$ is bounded, Assumption \ref{assumption:mean} places a non-sparsity condition on $\mathbf{b}_0$ that has appeared elsewhere in the analysis of $\ell^2$-based methods (\citealp{Bai10, Horvath12, Li23}). To see this, suppose that $\boldsymbol{\Lambda}_t=\mathbf{I}_d$ and $\mathbf{b}_0$ is sparse in the sense that $\lVert\mathbf{b}_0\rVert_0 \leq d_0$ for some $d_0 \ll d$ and $\Delta_T = \log^{1+\varepsilon} T$ for some $\varepsilon > 0$. Then $\Delta_T\lVert \mathbf{b}_0 \rVert^2_2 \lesssim d_0\log^{1+\varepsilon}T$, meaning that Assumption \ref{assumption:mean} cannot be met even under the relatively weak sparsity condition $d_0 = \mathcal{O}(d\log^{-\varepsilon} T)$. Conversely, Assumption \ref{assumption:mean} will be met when $\mathbf{b}_0$ is dense in the sense that many coordinates undergo a change, even if $\lim_{T\to\infty}\lVert \mathbf{b}_0\rVert_\infty = 0$. For example, if some fraction of the coordinates of $\mathbf{b}_0$ are equal to $\log^{-\varepsilon/2}T$ and $\Delta_T \geq \log^{1+\varepsilon}T$.

\begin{theorem}[Mean-SCP Localization Rate]\label{theorem:smcp}
    Let $\{\mathbf{y}_t\}_{t=1}^T$ be a sequence of independent, sub-Gaussian random vectors with $\mathbf{y}_t \in \mathbb{R}^d$, $\normalfont{\Var}(\mathbf{y}_t) = \boldsymbol{\Lambda}^{-1}$ for all $t \in [T]$, and $\lVert \mathbf{y}_t\rVert_{\psi_2} =\mathcal{O}(1)$. Define $\hat{\tau}_{\normalfont \text{MAP}}$ as in (\ref{eq:map}) by fitting the mean-scp model in (\ref{eq:gamma-post-cat1}) and (\ref{eq:b-smcp}) with $\omega_0 > 0$. For any $\beta > 0$, if Assumption \ref{assumption:mean} holds and $\max_{t\in[T]} |\log \pi_{t}| \leq C_\pi \log T$ for some constant $C_\pi > 0$, then there exist constants $C,C_\beta > 0$ so that:
    \vspace{-5pt}
    \begin{align}
        \Pr\left(|\hat{\tau}_{\normalfont \text{MAP}}-t_0| \leq \frac{C_\beta \log T}{\lVert\boldsymbol{\Lambda}^{\frac{1}{2}}\mathbf{b}_0\rVert_2^2}\right) \geq 1-\frac{C}{T^{\beta}}. \label{eq:thm-1}
    \end{align}
\end{theorem}
\vspace{-5pt}

\begin{remark}[Mean-SCP Minimax Rate]\label{rmk:thm-1-minimax}
    Under the settings of Theorem \ref{theorem:smcp}, \cite{Wang17} and \cite{Wang2020_localization} have shown that the minimax optimal localization rate is proportional to $(\Delta_T\lVert\boldsymbol{\Lambda}^{\frac{1}{2}}\mathbf{b}_0\rVert_2^{2})^{-1}$. Theorem \ref{theorem:smcp} implies the localization rate for the mean-scp model is minimax optimal aside from a $\log T$ factor. 
\end{remark}

\begin{remark}[Exact Recovery]\label{rmk:exact-recover}
The localization error in (\ref{eq:thm-1}) vanishes when $\lVert\boldsymbol{\Lambda}^{\frac{1}{2}}\mathbf{b}_0\rVert_2^2 \gg \log T$. When $\lVert\mathbf{b}_0\rVert_\infty$ and the largest eigenvalue of $\boldsymbol{\Lambda}$ are both $\mathcal{O}(1)$, then exact recovery is only possible in the high-dimensional regime with $d \gg \log T$ as in Theorem 3.2 of \cite{Bai10}. To see this, note that $\lVert\boldsymbol{\Lambda}^{\frac{1}{2}}\mathbf{b}_0\rVert_2^2 \lesssim \lVert\mathbf{b}_0\rVert_2^2 \leq d\lVert\mathbf{b}_0\rVert_\infty^2$, so $\lVert\boldsymbol{\Lambda}^{\frac{1}{2}}\mathbf{b}_0\rVert_2^2 \lesssim \log T$ when $d \lesssim \log T$. 
\end{remark}

\begin{remark}[Bounded Prior]\label{rmk:prior-bd}
    The bound on $\boldsymbol{\pi}_{1:T}$ guarantees that our choice of prior does not overwhelm the evidence in the data. This assumption is trivially satisfied if each time is equally likely to be the location of the change \textit{a priori}, i.e. $\pi_t = T^{-1}$. We discuss the choice of $\boldsymbol{\pi}_{1:T}$ in greater detail in Appendix \ref{app:prior}. 
\end{remark}

\begin{remark}[Bounded Orlicz Norm]\label{rmk:sub-g}
    When $\min\{d,T\}\to\infty$, the assumption that $\lVert \mathbf{y}_t\rVert_{\psi_2}$ is uniformly bounded restricts the strength of the dependence between the $d$ time-series. Two separate sufficient conditions for this restriction include: i) each entry of $\mathbf{y}_t$ is independent and $\mathcal{SG}(\sigma)$ for some $\sigma < \infty$, and ii) $\mathbf{y}_t \sim \mathcal{N}_d(\boldsymbol{\mu}_t, \boldsymbol{\Lambda}^{-1})$ and $\lambda^{-1}_{\min} = \mathcal{O}(1)$, where $\lambda_{\min}$ is the smallest eigenvalue of $\boldsymbol{\Lambda}$. When $d$ is fixed, the coordinates of $\mathbf{y}_t$ can display arbitrary levels of dependence so long as they are sub-Gaussian (see Appendix \ref{app:thm1-events}).
\end{remark}
\vspace{-5pt}

% The precision matrix $\boldsymbol{\Lambda}$ is taken as known in Theorem \ref{theorem:smcp}. Corollary \ref{cor:lambda-hat} shows that under certain conditions on the growth of $d$ and $\lVert \mathbf{b}_0\rVert_0$, it is possible to construct a consistent estimator $\hat{\boldsymbol{\Lambda}}$ and use the normalized observations $\hat{\boldsymbol{\Lambda}}^{\frac{1}{2}}\mathbf{y}_t$ to localize $t_0$.
%
% \begin{corollary}
%     \label{cor:lambda-hat}
% \end{corollary}

\subsubsection{Var-SCP and MeanVar-SCP Localization Rates for Univariate Data}

To analyze the localization rates of the var-scp and meanvar-scp models, we once again restrict $d=1$ so that we are dealing with a univariate sequence. For a change in the variance of $\mathbf{y}_{1:T}$ to be detectable, we require a signal strength assumption analogous to Assumption \ref{assumption:mean}.
\begin{assumption}[Detectable Scale Change]\label{assumption:scale}
    Suppose $\normalfont{\Var}(y_t) = (s_0^2)^{\mathbbm{1}\{t\geq t_0\}}$ for some $t_0 \in [T]$ and that there are constants $\underline{s}, \overline{s}>0$ so that $0<\underline{s} < s_0<\overline{s} < \infty$. Assume that $\Delta_T \gtrsim\log T$ and $\Delta_T (s_0^2-1)^2 \gg \log T$. 
\end{assumption}
\vspace{-5pt}
Again, the assumption that $\Var(y_t) =1$ for $t < t_0$ is for notational convenience and the results in this section continue to hold when $\Var(y_{t_0-1}) \neq 1$. We simply replace $s_0^2$ in Assumption \ref{assumption:scale} with $\Var(y_{t_0}) / \Var(y_{t_0-1})$. Note that $(s_0^2-1)^2$ measures the signal strength of $s_0^2$ as $(s_0^2-1)^2$ decreases on $(0,1)$, increases on $(1,\infty)$, and achieves a unique minimum at zero when $s_0^2 = 1$. Assumption \ref{assumption:scale} also allows for a vanishing signal, e.g. if $\Delta_T = \log^{1+\varepsilon} T$ for some $\varepsilon > 0$ and $s_0^2 = 1 \pm \log^{-\epsilon/4}T$.
 
\begin{theorem}[Var-SCP Localization Rate]\label{theorem:sscp}
    Let $\{y_t\}_{t=1}^T$ be a sequence of independent, sub-Gaussian random variables with $\E[y_t]=0$, and $\lVert y_t\rVert_{\psi_2} = \mathcal{O}(1)$. Define $\hat{\tau}_{\normalfont \text{MAP}}$ as in (\ref{eq:map}) by fitting the var-scp model in (\ref{eq:gamma-post-cat1})-(\ref{eq:gamma-post-cat2}) and (\ref{eq:s-sscp}) with $u_0,v_0>0$. For any $\beta > 0$, there exist constants $C,C_{1,\beta}, C_{2,\beta} > 0$ so that if Assumption \ref{assumption:scale} holds, $\pi_t =0$ for each $t \geq C_{2,\beta}\log(\log T)$, and $\max_{t\in[T]} |\log \pi_{t}| \leq C_\pi \log T$ for some $C_\pi >0$, then:
    \vspace{-2.5pt}
    \begin{align}
        \Pr\left(|\hat{\tau}_{\normalfont \text{MAP}}- t_0| \leq \frac{C_{1,\beta} \log T}{(s_0^2-1)^2} \right) \geq 1 - \frac{C}{T^{\beta}} - \frac{C_{2,\beta}}{\log^{\beta}T}. \label{eq:thm-2}
    \end{align}
\end{theorem}
\vspace{-5pt}
\begin{remark}[Var-SCP Minimax Rate]\label{rmk:thm-2-minimax}
    Under the settings of Theorem \ref{theorem:sscp}, \cite{Wang21} showed that the minimax optimal localization rate is proportional to $(\Delta_T(s_0^2-1)^{2})^{-1}$, so Theorem \ref{theorem:sscp} shows that the localization rate for the var-scp model is minimax optimal aside from a $\log T$ factor. 
\end{remark}
\vspace{-10pt}
\cite{Cappello25} previously established a similar result to Theorem \ref{theorem:sscp} for the case of an independent Gaussian sequence with $s_0\neq 1$ fixed and $\Delta_T \geq cT$ for some $c \in (0,1/2)$. However, their localization error was of order $\mathcal{O}(\sqrt{T\log T})$, so Theorem \ref{theorem:sscp} both strengthens and generalizes their result. Theorem \ref{theorem:sscp} does require that $\hat{\tau}_{\normalfont \text{MAP}}$ be bounded away from $T$ by some number of indices proportional to $\log(\log T)$, but this requirement is weak seeing as Assumption \ref{assumption:scale} implies that $t_0$ cannot fall in this region.

The meanvar-scp model looks for changes in both the mean and variance of $\mathbf{y}_{1:T}$. Naturally, we require that a change occurs in at least one of $\boldsymbol{\mu}_{1:T}$ or $\boldsymbol{\lambda}_{1:T}$, and that the signal from this change is strong enough for the model to detect. This amounts to assuming at least one of Assumptions \ref{assumption:mean} or \ref{assumption:scale} holds.

\begin{theorem}[MeanVar-SCP Localization Rate]\label{theorem:smscp}
    Let $\{y_t\}_{t=1}^T$ be a sequence of independent, sub-Gaussian random variables with $\lVert y_t\rVert_{\psi_2}= \mathcal{O}(1)$. Define $\hat{\tau}_{\normalfont \text{MAP}}$ as in (\ref{eq:map}) by fitting the  meanvar-scp model in (\ref{eq:gamma-post-cat1})-(\ref{eq:gamma-post-cat2}) and (\ref{eq:bs-smscp}) with $\omega_0, u_0,v_0>0$. For any $\beta > 0$, there exist constants $C,C_{1,\beta}, C_{2,\beta} > 0$ so that if either Assumption \ref{assumption:mean} or Assumption \ref{assumption:scale} holds, $\pi_t =0$ for each $t \geq C_{2,\beta}\log(\log T)$, and $\max_{t\in[T]} |\log \pi_{t}| \leq C_\pi \log T$ for some $C_\pi >0$, then:
    \vspace{-5pt}
    \begin{align}
        \Pr\left(|\hat{\tau}_{\normalfont \text{MAP}}- t_0| \leq \frac{C_{1,\beta} \log T}{\max\{\min\{b_0^2,b_0^2/s_0^2\}, (s_0^2-1)^2\}} \right) \geq 1 - \frac{C}{T^{\beta}} - \frac{C_{2,\beta}}{\log^{\beta}T}. \label{eq:thm-3}
    \end{align}
\end{theorem}
\begin{remark}[MeanVar-SCP Minimax Rate]\label{rmk:thm-3-minimax}
    The localization error for the meanvar-scp model is the smaller of the respective errors for the mean-scp and var-scp models, which in were in turn minimax optimal for mean and variance changes, so the localization rate for the meanvar-scp model should also be minimax optimal regardless of whether $b_0=0$ or $s_0^2=1$. 
\end{remark}
\vspace{-10pt}

\subsubsection{Localization Rate with Dependent Data}

Thus far we have exclusively considered sequences of independent observations. We can weaken this assumption by allowing $\mathbf{y}_{1:T}$ to display auto-correlation. In particular, we assume that $\{y_t\}_{t \geq 1}$ is an $\alpha$-mixing process (see Appendix \ref{app:notation} for a definition of $\alpha$-mixing): 

\begin{assumption}\label{assumption:alpha-mixing}
    Given the stochastic process $\{y_t\}_{t\geq 1}$, assume that for some $t_0 \in \mathbb{N}$, $b_0 \in\mathbb{R}$, $s_0>0$, and some distributions $F_0$ and $F_1$, there are stationary processes $\{z_{0,t}\}_{t \geq 1}$ and $\{z_{1,t}\}_{t \geq 1}$ such that $z_{0,t} \sim F_0$, $z_{1,t} \sim F_1$, $\E[z_{0,t}] = \E[z_{1,t}] = 0$, $\Var(z_{0,t}) = \Var(z_{1,t}) = 1$ and $y_t := z_{0,t} \mathbbm{1}_{\{t < t_0\}}  + (s_0z_{1,t} + b_0)\mathbbm{1}_{\{t \geq t_0\}}$. Additionally, assume that:
    \vspace{-10pt}
    \begin{enumerate}[label=\normalfont(\roman*)]
        \item $\{z_{0,t}\}_{t \geq 1}$ and $\{z_{1,t}\}_{t \geq 1}$ are $\alpha$-mixing processes with respective coefficients $\{\alpha_{0,k}\}_{k\geq 1}$ and $\{\alpha_{1,k}\}_{k\geq 1}$ that satisfy $\max\{\alpha_{0,k}, \alpha_{1,k}\} \leq e^{-C k}$ for some $C > 0$. \vspace{-5pt}
        \item There exist constants $D_1, D_2 > 0$ such that $\sup_{t \geq 1} \max\{\E\left[|z_{0,t}|^{4+D_1}\right],\;\E\left[|z_{1,t}|^{4+D_1}\right]\}\leq D_2.$  
    \end{enumerate}
\end{assumption}
\vspace{-10pt}
Under Assumption \ref{assumption:alpha-mixing} each of the SCP models can still consistently localize $t_0$, although with a localization rate that is approximately order $\mathcal{O}(\log^{2} T)$:
\begin{theorem}[$\alpha$-Mixing Localization Rate]\label{theorem:alpha-mixing}
    Let $\{y_t\}_{t \geq 1}$ be a univariate stochastic process satisfying Assumption \ref{assumption:alpha-mixing} for some $t_0 \in \mathbb{N}$ and let $\{a_T\}_{T \geq 1}$ be any sequence such that $a_T \to \infty$. Let Assumption \ref{assumption:mean}* and Assumption \ref{assumption:scale}* denote Assumptions \ref{assumption:mean} and \ref{assumption:scale} with each $\log T$ term replaced by $a_T\log ^2T$. For each estimator $\hat{\tau}_{\normalfont \text{MAP}}$ and localization error $\epsilon_T$ corresponding to Theorems \ref{theorem:smcp}-\ref{theorem:smscp}, if Assumptions \ref{assumption:mean}* and \ref{assumption:scale}* hold in place of Assumptions \ref{assumption:mean} and \ref{assumption:scale}, then there exist constants $C,C_{1,a}$, $C_{2,a} > 0$ such that if $\pi_t =0$ for each $t \geq C_{2,a}a_T\log^2 T$ and $\max_{t\in[T]} |\log \pi_{t}| \leq C_\pi \log^2 T$ for some $C_\pi >0$, then $\Pr(|\hat{\tau}_{\text{\normalfont MAP}}-t_0|\leq \tilde{\epsilon}_T )\geq 1 - Ca^{-1}_T$, where $\tilde{\epsilon}_T = \epsilon_TC_{1,a} a_T\log T$.
\end{theorem}
%The proof of Theorem \ref{theorem:alpha-mixing} is given in Appendix \ref{app:alpha-mixing}.

\subsection{Credible Sets and Detection Rule}
\label{sec:cred-sets}

Even in the absence of a change-point, $\hat{\tau}_{\normalfont \text{MAP}}$ is well-defined as per (\ref{eq:map}). Therefore, we require some detection criterion to indicate whether $\hat{\tau}_{\normalfont \text{MAP}}$ does in fact identify a change-point. Figure \ref{fig:post-probs-plot} in Appendix \ref{app:prior} shows that the elements $\overline{\boldsymbol{\pi}}_{1:T}$ tend to be quite diffuse under the null model with no change to $\boldsymbol{\mu}_{1:T}$ or $\boldsymbol{\lambda}_{1:T}$. By this we mean that $\lim_{T\to\infty} \sum_{t\in\mathcal{T}} \overline{\pi}_t = 0$ for any fixed index set $\mathcal{T} \subseteq[T]$, so the credible sets constructed according to (\ref{eq:cs}) will contain a large subset of $[T]$ when no change is present. This observation motivated \cite{Cappello25} to adopt $|\mathcal{CS}(\alpha,\overline{\boldsymbol{\pi}}_{1:T})| \leq T/2$ as their detection criterion. This rule is somewhat \textit{ad hoc} and can lead to an inflated false positive rate in practice. To develop a more theoretically justified detection rule, we begin by noting that the appropriate localization error bounds the size of $\mathcal{CS}(\alpha, \overline{\boldsymbol{\pi}}_{1:T})$:

\begin{corollary} \label{cor:cred-sets}
    Let $\epsilon_T$ be the localization error corresponding to one of Theorems \ref{theorem:smcp}, \ref{theorem:sscp}, \ref{theorem:smscp}, or \ref{theorem:alpha-mixing}, then under the respective conditions of the these theorems, for any $\alpha > 0$, $\lim_{T \to \infty} \Pr\left(|\mathcal{CS}(\alpha, \overline{\boldsymbol{\pi}}_{1:T})| \leq 2 \epsilon_T \right) = 1.$
\end{corollary}
\vspace{-5pt}
Suppose we use $|\mathcal{CS}(\alpha,\overline{\boldsymbol{\pi}}_{1:T})| \leq \log^{1+\delta} T$ as our detection criterion for some small $\delta > 0$. Given independent observations and fixed $b_0$ and $s_0$, each localization error $\epsilon_T$ from Section \ref{sec:localization} is dominated by $\log^{1+\delta} T$, so Corollary \ref{cor:cred-sets} implies the false negative rate will converge to zero with high probability under this rule. At the same time, $\log^{1+\delta} T \ll T / 2$, which will result in fewer false positives, even for moderately large $T$. If the observations are dependent, then the false negative rate still goes to zero so long as $\delta > 1$.
\section{Multiple Independent Change-Point (MICH) Model}
\label{sec:mich}

We now extend the SCP models to the case of multiple change-points. Let $\mathbf{y}_{1:T}$ be a univariate sequence of independent Gaussian observations with mean $\boldsymbol{\mu}_{1:T}$ and precision $\boldsymbol{\lambda}_{1:T}$ (the case  of multivariate $\boldsymbol{\mu}_{1:T}$ is addressed in Appendix \ref{app:multi-mich}). Assume that both $\boldsymbol{\mu}_{1:T}$ and $\boldsymbol{\lambda}_{1:T}$ have piecewise constant structures such that there are $L^* \geq 0$ breaks in $\boldsymbol{\mu}_{1:T}$, $K^* \geq 0$ breaks in $\boldsymbol{\lambda}_{1:T}$, and an additional $J^* \geq 0$ simultaneous breaks in both $\boldsymbol{\mu}_{1:T}$ and $\boldsymbol{\lambda}_{1:T}$. The parameters $L^*$, $K^*$, and $J^*$ are generally unknown, but if we can specify some $L \geq L^*$, $K \geq K^*$, and $J \geq J^*$ with $N:= L+K+J$, then we can detect the breaks in $\boldsymbol{\mu}_{1:T}$ and $\boldsymbol{\lambda}_{1:T}$ using the following Multiple Independent CHange-point (MICH) model:
\begin{align}
    y_t\:|\: \mu_{t}, \lambda_t  &\overset{\text{ind.}}{\sim} \mathcal{N}(\mu_t, \lambda_t^{-1}), && t \in [T], \label{eq:y-mich}\\
    \mu_{t} &:= \sum_{i = 1}^{J+L} \mu_{i t} := \sum_{j = 1}^{J} b_j\mathbbm{1}_{\{t\geq\tau_j\}} + \sum_{\ell = J + 1}^{J+L} b_\ell\mathbbm{1}_{\{t\geq\tau_\ell\}}, && \label{eq:mu_t} \\ 
    \lambda_t &:= \prod_{i=1}^{J+K} \lambda_{it} := \prod_{j=1}^{J} 
 s_{j}^{\mathbbm{1}\{t\geq\tau_{j}\}}\prod_{k=J+L+1}^{N} 
 s_{k}^{\mathbbm{1}\{t\geq\tau_{k}\}},&& \label{eq:lambda_t} \\
    \tau_i &\overset{\text{ind.}}{\sim} \text{Categorical}(\boldsymbol{\pi}_{i,1:T}),&&  i \in [N], \label{eq:tau-mich}\\
    \{b_j,s_j\} &\overset{\text{ind.}}{\sim} \text{Normal-Gamma}(0,\omega_0,u_0,v_0), &&  1\leq j \leq J, \label{eq:bs-mich}\\
     b_\ell &\overset{\text{ind.}}{\sim} \mathcal{N}(0,\omega_0^{-1}), && J < \ell \leq  J+L, \label{eq:b-mich} \\
     s_k &\overset{\text{ind.}}{\sim} \text{Gamma}(u_0,v_0), && J+L<k \leq N. \label{eq:s-mich}
     \vspace{-5pt}
\end{align}
MICH employs a modular framework by stacking $N$ changes in the construction of $\boldsymbol{\mu}_{1:T}$ and $\boldsymbol{\lambda}_{1:T}$. Each indicator $\mathbbm{1}_{\{\tau_i \geq t\}}$ in the sum (\ref{eq:mu_t}) and product (\ref{eq:lambda_t}) contributes an additional break in the piecewise constant structure of either $\boldsymbol{\mu}_{1:T}$, $\boldsymbol{\lambda}_{1:T}$, or both at time $\tau_i$. We use distinct index labels $\ell$, $k$, and $j$ to clarify that the latent change-points $\tau_\ell$, $\tau_k$, and $\tau_j$ correspond to mean-only, variance-only, and joint mean and variance changes. When $\E[y_1]\neq0$ or $\Var(y_1)\neq 1$, we can also include an intercept $\mu_0$ in (\ref{eq:mu_t}) and initial precision $\lambda_0$ in (\ref{eq:lambda_t}) (see Appendix \ref{app:empirical-bayes}). 

Each distribution for the break sizes in (\ref{eq:bs-mich})-(\ref{eq:s-mich}) corresponds to one of the SCP models, so we can quantify the uncertainty around each $\tau_i$ using the same approach we took in Section \ref{sec:scp}. This is the key insight of MICH; the model integrates multiple SCP models in a unified framework, a strategy that will lead to computationally efficient inference for $\boldsymbol{\tau}_{1:N}$. To demonstrate the rationale, let $\boldsymbol{\Theta}$ denote the set of all model variables in (\ref{eq:tau-mich})-(\ref{eq:s-mich}). Pretend for the moment that we have access to $\boldsymbol{\Theta} \setminus\{b_j,s_j,\tau_j\}$ for some $j \in [J]$. We can use these known parameters to calculate partial residual terms $r_{-j t} := y_t - \mu_t + \mu_{j t}$ and $\lambda_{-jt} := \lambda_{jt}^{-1}\lambda_t$, in which case $p(b_j,s_j,\tau_j\:|\:\mathbf{r}_{-j,1:T}, \boldsymbol{\lambda}_{-j,1:T})$ is just the meanvar-scp posterior as defined in (\ref{eq:gamma-post-cat1}) and (\ref{eq:bs-smscp}) with parameters determined by the following call of \texttt{meanvar-scp} (\ref{eq:meanvar-scp-fn}):
\begin{align}
    \overline{\boldsymbol{\theta}}_{j} := \{\overline{b}_{jt}, \overline{\omega}_{jt}, \overline{u}_{jt}, \overline{v}_{jt}, \overline{\pi}_{jt}\}_{t=1}^T = \texttt{meanvar-scp}(\mathbf{r}_{-j,1:T}\:;\: \boldsymbol{\lambda}_{-j,1:T}, \omega_0, u_0, v_0, \boldsymbol{\pi}_{1:T}).
\end{align}
Analogously, we can define $r_{-\ell t}$ and $\lambda_{-kt}$ for $J < \ell\leq J+L$ and $J + L< k \leq N$, then use these partial residuals to calculate $\overline{\boldsymbol{\theta}}_{\ell} := \{\overline{b}_{\ell t}, \overline{\omega}_{\ell t}, \overline{\pi}_{\ell t}\}_{t=1}^T$ and $\overline{\boldsymbol{\theta}}_{k} := \{ \overline{u}_{kt}, \overline{v}_{kt}, \overline{\pi}_{kt}\}_{t=1}^T$ with respective calls of \texttt{mean-scp} (\ref{eq:mean-scp-fn}) and \texttt{var-scp} (\ref{eq:var-scp-fn}). In reality, $\boldsymbol{\Theta}$ is unknown, so we cannot calculate $r_{-i t}$ or $\lambda_{-it}$ for any $i \in [N]$. Nonetheless, we will show that it is still possible to approximate the posterior $p(\boldsymbol{\Theta}\:|\:\mathbf{y}_{1:T})$ by fitting individual SCP models to some notion of a residual. Before detailing this result, we briefly point out that MICH recovers CPD version SuSiE with $L$ change-points when $K=J=0$ (see Section 6 of \cite{Wang20}), and when $L=J=0$, then MICH reduces to the $K$ component PRISCA model (\citealp{Cappello25}).

\subsection{Fitting MICH}
\label{sec:fit-mich}

The goal of Bayesian inference for MICH is to find a way to approximate $p(\boldsymbol{\Theta}\:|\:\mathbf{y}_{1:T})$. Implementing a Gibbs sampler (\citealp{Geman84}) is conceptually straightforward; however, the change-point locations $\boldsymbol{\tau}_{1:N}$ are discrete, highly correlated, and belong to a high-dimensional space as $N$ and $T$ increase. Thus, the sampler will often suffer from poor mixing and may even fail to converge (\citealp{Smith93, Fearnhead06, Cappello21}). With Algorithm \ref{alg:mich}, we circumvent the challenges of MCMC based inference by implementing a deterministic backfitting procedure (\citealp{Friedman81, Breiman85}) that returns an approximation to  $p(\boldsymbol{\Theta}\:|\:\mathbf{y}_{1:T})$. We first describe the algorithm, then characterize the nature of the approximation. 

\begin{algorithm}[!b]
\label{alg:mich}
\caption{Variational Bayes Approximation to MICH Posterior.}

\footnotesize
\SetAlgoLined
  Inputs: $\mathbf{y}_{1:T},\:L,\:K,\:J,\:\omega_0,\:u_0,v_0,\{\boldsymbol{\pi}_{i,1:T}\}_{i=1}^N$\:; \\
  Initialize: $\overline{\boldsymbol{\theta}}_{1:N}$\:;

  \Repeat {Convergence} {
  \For{$j=1$ \KwTo $J$} {
      $\tilde{r}_{-jt} := y_t - \sum_{j' \neq j} \E_{q_{j'}}[\lambda_{j't} \mu_{j't}] / \E_{q_{j'}}[\lambda_{j't}]- \sum_{\ell =J+ 1}^{J+L} \E_{q_\ell}[\mu_{\ell t}]$ \tcp*{j\textsuperscript{th} partial mean residual} 
      $\overline{\lambda}_{-jt} := \prod_{j' \neq j}\E_{q_{j'}}[\lambda_{j't}] \prod_{k=J+L+1}^N \E_{q_k}[\lambda_{kt}]$ \tcp*{j\textsuperscript{th} partial scale residual}
      %%$\delta_{-jt} := \sum_{\ell=1}^L  \Var(\mu_{\ell t} ) + \sum_{j' \neq j}[\E[\lambda_{j't} \mu^2_{j't}] / \E[\lambda_{j't}] - (\E[\lambda_{j't} \mu_{j't}] / \E[\lambda_{j't}])^2 ]$ \tcp*{j\textsuperscript{th} variance correction term} 
      Compute $\delta_{-j t}$, $\tilde{v}_{jt}$, and $\tilde{\pi}_{jt}$ by (\ref{eq:delta-j}), (\ref{eq:v-j-corrected}), and (\ref{eq:pi-j-corrected}) \tcp*{variance corrected priors} 
      $\overline{\boldsymbol{\theta}}_j := \texttt{meanvar-scp}(\tilde{\mathbf{r}}_{-j,1:T} \:;\:\overline{\boldsymbol{\lambda}}_{-j,1:T}, \omega_0, u_0, \tilde{\mathbf{v}}_{j,1:T}, \tilde{\boldsymbol{\boldsymbol{\pi}}}_{j,1:T})$ \tcp*{update meanvar-scp parameters}
    }
    \For{$\ell=J+1$ \KwTo $J+L$} {
      $\tilde{r}_{-\ell t} := y_t - \sum_{j =1}^J \E_{q_j}[\lambda_{jt} \mu_{jt}] / \E_{q_j}[\lambda_{jt}]- \sum_{\ell' \neq \ell} \E_{q_{\ell'}}[\mu_{\ell' t}]$ \tcp*{l\textsuperscript{th} partial mean residual}
      $\overline{\lambda}_{t} := \prod_{j=1}^J\E_{q_j}[\lambda_{jt}]\prod_{k=J+L+1}^N \E_{q_k}[\lambda_{kt}]$ \tcp*{precision of residual} 
      $\overline{\boldsymbol{\theta}}_\ell := \texttt{mean-scp}(\tilde{\mathbf{r}}_{-\ell,1:T} \:;\: \overline{\boldsymbol{\lambda}}_{1:T}, \omega_0, \boldsymbol{\boldsymbol{\pi}}_{1:T})$ \tcp*{update mean-scp parameters}
    }
    \For{$k=J+L+1$ \KwTo $N$} {
      $\tilde{r}_{t} := y_t  - \sum_{j=1}^J \E_{q_j}[\lambda_{jt} \mu_{jt}] / \E_{q_j}[\lambda_{jt}]- \sum_{\ell = J+1}^{J+L} \E_{q_\ell}[\mu_{\ell t}] $ \tcp*{mean residual} 
      $\overline{\lambda}_{-kt} := \prod_{j=1}^J\E_{q_j}[\lambda_{jt}]\prod_{k' \neq k} \E_{q_{k'}}[\lambda_{k't}] $ \tcp*{k\textsuperscript{th} partial scale residual} 
      %% $\delta_{t} := \sum_{\ell=1}^L  \Var(\mu_{\ell t} ) + \sum_{j=1}^J[\E[\lambda_{jt} \mu^2_{jt}] / \E[\lambda_{jt}] - (\E[\lambda_{jt} \mu_{jt}] / \E[\lambda_{jt}])^2 ];$ \tcp*{variance correction term}\\
      Compute $\delta_t$, $\tilde{v}_{kt}$, and $\tilde{\pi}_{kt}$ by (\ref{eq:delta}), (\ref{eq:v-k-corrected}), and (\ref{eq:pi-k-corrected}) \tcp*{variance corrected priors} 
      $\overline{\boldsymbol{\theta}}_k := \texttt{var-scp}(\tilde{\mathbf{r}}_{1:T} 
      \:;\:\overline{\boldsymbol{\lambda}}_{-k,1:T}, u_0, \tilde{\mathbf{v}}_{k,1:T}, \tilde{\boldsymbol{\boldsymbol{\pi}}}_{k,1:T})$ \tcp*{update var-scp parameters}
    }
  }
  \Return{Posterior Parameters: $\overline{\boldsymbol{\theta}}_{1:N}$}.
\end{algorithm}

Algorithm \ref{alg:mich} begins by initializing some values for $\overline{\boldsymbol{\theta}}_{1:N}$, which are then used to approximate the residual $\tilde{r}_{t}$ and expected precision $\overline{\lambda}_t$ as per (\ref{eq:mod-resid}) and (\ref{eq:lambda-bar}). Each $\overline{\boldsymbol{\theta}}_i$ is then updated iteratively by fitting the appropriate SCP model with either $\tilde{r}_{t}$ or its partial analogue $\tilde{r}_{-it}$ in place of $y_t$, and $\overline{\lambda}_{t}$ or its partial $\overline{\lambda}_{-it}$ in place of the scale component $\omega_t$ (see (\ref{eq:sscp-start})). Fitting each SCP model returns a distribution, so given the value of $\overline{\boldsymbol{\theta}}_{1:N}$ at the current iteration of Algorithm \ref{alg:mich}, we have:
\begin{align}
    \{b_j,s_j\} \:|\: \tilde{\mathbf{r}}_{-j,1:T},  \overline{\boldsymbol{\lambda}}_{-j,1:T},\tau_j = t &\sim \text{\normalfont Normal-Gamma}(\overline{b}_{j t}, \overline{\omega}_{j t}, \overline{u}_{j t}, \overline{v}_{j t}), & 1 \leq j \leq J, \label{eq:mich-bs-post} \\
    b_\ell \:|\:  \tilde{\mathbf{r}}_{-\ell,1:T},  \overline{\boldsymbol{\lambda}}_{1:T}, \tau_\ell = t &\sim \mathcal{N}(\overline{b}_{\ell t}, 1/\overline{\omega}_{\ell t}), & J < \ell \leq J+L \label{eq:mich-b-post}\\
    s_k \:|\:  \tilde{\mathbf{r}}_{1:T},  \overline{\boldsymbol{\lambda}}_{-k,1:T},\tau_k = t &\sim \text{\normalfont Gamma}(\overline{u}_{k t}, \overline{v}_{k t}), & J+L < k  \leq N. \label{eq:mich-s-post}
\end{align}
Let $\{q_i\}_{i=1}^N$ stand in for the $N$ conditional distributions in (\ref{eq:mich-bs-post})-(\ref{eq:mich-s-post}) with $q_i(\tau_i = t) = \overline{\pi}_{it}$. We can use $\{q_i,\overline{\boldsymbol{\theta}}_i\}_{i=1}^N$ to efficiently compute the expected terms in Algorithm \ref{alg:mich} as well as any other statistics of interests for change-point detection. For example, we can use $\overline{\boldsymbol{\pi}}_{i,1:T}$ to construct estimators $\hat{\tau}_{\text{MAP}, i}$ and $\alpha$-level credible sets $\mathcal{CS}(\alpha, \overline{\boldsymbol{\pi}}_{i,1:T})$ as per (\ref{eq:map}) and (\ref{eq:cs}). 
Adopting the same detection rule as in Section \ref{sec:cred-sets}, we can estimate the total number of changes with $\hat{N}:=|\{ i\in [N] : |\mathcal{CS}(\alpha, \overline{\boldsymbol{\pi}}_{i,1:T})| \leq \log^{1+\delta} T \}|$ for some small $\delta > 0$. We can individually estimate $L^*$, $K^*$, and $J^*$ following the same approach.

\begin{remark}[Complexity of Algorithm \ref{alg:mich}] \label{rmk:computational-complexity}
     Fitting each SCP model and calculating each expectation in Algorithm \ref{alg:mich} requires taking the cumulative sum of $T$ terms (see Appendix \ref{app:posterior-parameters} and equations (\ref{eq:mu-post-mean}), (\ref{eq:lambda-post-mean}), and (\ref{eq:mu-lambda-post-mean})). Therefore, each outer loop of Algorithm \ref{alg:mich} is $\mathcal{O}(NT)$. The actual number of iterations required for Algorithm \ref{alg:mich} to converge is indeterminate and depends on the chosen stopping criterion (see Appendix \ref{app:convergence}). 
\end{remark}
\vspace{-10pt}

\subsection{Mean-Field Variational Approximation}
\label{sec:variational-bayes}
%The independence assumption in (\ref{eq:mean-field}) is strong, but still allows $q$ to capture arbitrary dependencies within each parameter block $\boldsymbol{\theta}_i$. This is key to the success of MICH and in Section \ref{sec:simulations} we show that $q$ effectively localizes changes in simulations by characterizing the relationship between $\tau_i$ and the rest of $\boldsymbol{\theta}_i$. 

We now show that Algorithm \ref{alg:mich} is equivalent to a Variational Bayes (VB) algorithm that returns an approximation of $p:=p(\boldsymbol{\Theta}\:|\:\mathbf{y}_{1:T})$ under a certain Mean-Field assumption. Recall that VB methods specify a family of distributions $\mathcal{Q}$, then search for a $q^* \in \mathcal{Q}$ that minimizes the Kullback-Leibler divergence with $p$, i.e., $q^* =\text{arg}\min_{q \in  \mathcal{Q}}  \; \text{KL}( q \:\lVert\: p)$ (see \cite{Blei17} for a recent overview of VB inference). In practice, $\text{KL}( q \:\lVert\: p)$ tends to be intractable, but we can define the evidence lower bound (ELBO):
\begin{align}
    \text{ELBO}(q) &:= \int q(\boldsymbol{\Theta}) \log \frac{ p(\mathbf{y}_{1:T},\boldsymbol{\Theta})}{q(\boldsymbol{\Theta})} \; d\boldsymbol{\Theta} \label{eq:elbo}\\
    &\;= \log p(\mathbf{y}_{1:T}) - \text{KL}( q \:\lVert\: p)\label{eq:kl-decomp},
\end{align}
and solve $\text{arg}\max_{q \in  \mathcal{Q}}  \; \text{ELBO}( q)$ since the marginal likelihood in (\ref{eq:kl-decomp}) does not depend on $q$. We can simplify further by selecting a family $\mathcal{Q}$ that makes this optimization task tractable, e.g., by making a Mean-Field assumption so that $q$ factorizes over the blocks of $\boldsymbol{\Theta}$:
\begin{align}\label{eq:mean-field}
    \mathcal{Q}_{\text{MF}} := \left\{q: q(\boldsymbol{\Theta}) =  \prod_{j=1}^{J} q_j(b_j,s_j,\tau_j)\prod_{\ell=J+1}^{J+L} q_\ell(b_\ell,\tau_\ell)\prod_{k=J+L+1}^{N} q_k(s_k,\tau_k)\right\}.
\end{align}
Under (\ref{eq:mean-field}), it is possible to maximize the ELBO via coordinate ascent by iteratively updating $q_i$ for each $i \in [N]$. Proposition \ref{prop:coord-ascent} shows that this is exactly what Algorithm \ref{alg:mich} accomplishes.

\begin{proposition}
\label{prop:coord-ascent}
Assume that $\omega_0, u_0, v_0,\pi_{t} > 0$. Let $\{q^{(n)}_i\}_{i=1}^N$ be the distributions (\ref{eq:mich-bs-post})-(\ref{eq:mich-s-post}) returned by the $n^{\text{th}}$ iteration of Algorithm \ref{alg:mich} and define $q^{(n)}:= \prod_{i=1}^Nq^{(n)}_i$. Then Algorithm \ref{alg:mich} is a coordinate ascent procedure that maximizes $\text{\normalfont ELBO}(q)$ on $ \mathcal{Q}_{\text{MF}}$, %Equivalently, Algorithm \ref{alg:mich} is a coordinate descent procedure for solving $\min_{q \in \mathcal{Q}_{\text{MF}}}  \; \text{KL}( q \:\lVert\: p)$. 
and $q^{(n)}$ converges to a stationary point of $\text{\normalfont ELBO}(q)$.
\end{proposition}
\vspace{-10pt}

\subsection{Computational Details}

\subsubsection{Choice of \texorpdfstring{$L$}{L}, \texorpdfstring{$K$}{K}, and \texorpdfstring{$J$}{J}}
\label{sec:LKJ}

While it is technically feasible to fit MICH when $L,K,J>0$, we stress that the inclusion of both mean-only, variance-only, and simultaneous mean and variance changes is intended to facilitate our presentation of the model. In practice, we are typically interested in cases where every change in $\boldsymbol{\mu}_{1:T}$ and $\boldsymbol{\lambda}_{1:T}$ is simultaneous ($J^*>0$ and $L^*=K^*=0$), or where every change occurs independently ($L^*,K^*>0$, and $J^*=0$). Beyond this domain specific restriction, we must have $L \geq L^*$, $K \geq K^*$, and $J \geq J^*$ for MICH to correctly recover all the changes in $\boldsymbol{\mu}_{1:T}$ and $\boldsymbol{\lambda}_{1:T}$. As shown in both \cite{Wang20} and \cite{Cappello25}, the model is robust to overstating $L$, $K$, and $J$, as the extra $N - N^*$ components tend to capture null effects with diffuse posteriors $\overline{\boldsymbol{\pi}}_{i,1:T}$ that fail the detection criterion (see Appendix \ref{app:prior} for examples). Some prior knowledge is still required to ensure we do not underestimate the true number of components. In the absence of any information, we could set $L=K=J=\lceil T /\log T \rceil$, i.e. proportional $T$ over the localization error from Theorems \ref{theorem:smcp}-\ref{theorem:smscp}. This conservative default increases the complexity of Algorithm \ref{alg:mich} to $\mathcal{O}(T^2/\log T)$. To avoid including so many redundant components in the model, we propose using the value of $\text{ELBO}(q)$ to automatically select $L$, $K$, and $J$ (see Appendix \ref{app:LKJ-choice} for detail). 

\subsubsection{Prior Parameters}

Outside of choosing $L$, $K$, and $J$, MICH requires virtually no hyperparameter tuning. The sensitivity analysis in Appendix \ref{app:prior_sensitivity} indicates that the choice of prior parameters $\omega_0$, $u_0$, and $v_0$ does not meaningfully impact the output of Algorithm \ref{alg:mich} (so long as these values are not too large). This result has theoretical support, as the finite sample results in the proofs of Theorems \ref{theorem:smcp}-\ref{theorem:alpha-mixing} each hold for any small positive value of $\omega_0$, $u_0$, and $v_0$, and the asymptotic localization results hold for any fixed values for these parameters. By default, we set $\omega_0 = u_0 = v_0 = 0.001$. We can also vary these parameters across each component of MICH by introducing $\{\omega_i,u_i,b_i\}_{i=1}^N$. For the location priors, we choose $\pi_{it}$ so that $\overline{\pi}_{it} \approx T^{-1}$ under the null model (see Appendix \ref{app:prior} for more detail).

\subsubsection{Model Asymmetry}

By construction, the variable $\tau_i \in [T]$ is used in MICH to ``activate" the $i^{\text{th}}$ component of the model via the indicator $\mathbbm{1}_{\{t \geq \tau_i\}}$. The direction of this inequality is arbitrary, and we could just as easily have decided to ``deactivate" components with  $\mathbbm{1}_{\{t < \tau_i\}}$ instead. In Appendix \ref{app:prior}, we show that the direction of the inequality builds an asymmetry into the model, which may artificially inflate the elements of $\overline{\boldsymbol{\pi}}_{i,1:T}$ closer to index $T$. The default prior we have chosen for MICH attempts to mitigate the effects of this asymmetry, but the model still occasionally misses change-points that it would have identified if the order of $\mathbf{y}_{1:T}$ were reversed (see Figure \ref{fig:reverse-fit}). To overcome this undesirable behavior, we fit MICH using $\mathbf{y}_{T:1}$, then reverse the fitted parameters and use them to restart Algorithm \ref{alg:mich}. We return the resulting fit if it improves the ELBO over directly inputting $\mathbf{y}_{1:T}$ into Algorithm \ref{alg:mich}. In simulations, we observe that this procedure leads to a marked improvement in the performance of MICH. Since we can fit the model to $\mathbf{y}_{1:T}$ and $\mathbf{y}_{T:1}$ in parallel, this improvement comes at little additional cost. 

\subsubsection{Duplicate Components}
\label{sec:merge-procedure}

The mean field assumption (\ref{eq:mean-field}) is key to unlocking the computational simplicity of Algorithm \ref{alg:mich}, but because each block of $\boldsymbol{\Theta}$ is independent \textit{a priori}, there is a non-zero probability that $\tau_i = \tau_{i'}$ for some $i \neq i'$, i.e. two changes of the same class may be ``stacked" on top of each other. This construction differs from the product partition model of \cite{Barry93}, which rules out overlapping change-points by placing a joint prior on all possible fixed-sized partitions of $[T]$. Although Algorithm \ref{alg:mich} typically avoids assigning multiple components to a single change-point, we do observe instances of a single change splitting across multiple blocks of $\boldsymbol{\Theta}$ (see Figure \ref{fig:duplicate}), and if $L$, $K$, or $J$ are not chosen large enough to allow for some degree of redundancy, this can result in an underestimation of the true number of changes. To avoid this scenario and eliminate duplicates, we propose merging components $i$ and $i'$ if $\langle\overline{\boldsymbol{\pi}}_{i,1:T}, \overline{\boldsymbol{\pi}}_{i',1:T} \rangle$ exceeds some threshold $\beta > 0$. The details of this merge procedure are given in Appendix \ref{app:merge-procedure}.
\section{Simulation Study}
\label{sec:simulations}

In this section, we study MICH's ability to identify joint mean and variance changes by adapting experiment from \cite{Pein17}. We generate 5,000 replicates of $\mathbf{y}_{1:T}$ from Simulation \ref{sim:main} for $C = \sqrt{200}$ and $T \in\{\text{100, 500, 1,000}\}$, $\Delta_T \in\{\text{15, 30, 50}\}$, and $J^* \in \{2,5,10\}$. As noted by \cite{Pein17}, Simulation \ref{sim:main} ensures each jump in $\boldsymbol{\mu}_{1:T}$ and $\boldsymbol{\lambda}_{1:T}$ is similarly difficult to find. For each replicate, we fit MICH using the non-informative prior parameters $\omega_0=u_0=v_0=0.001$, the default prior $\boldsymbol{\pi}_{1:T}$ described in Appendix \ref{app:prior}, error tolerance $\epsilon = 10^{-7}$, and $\delta = 0.5$ for the detection rule described in Section \ref{sec:cred-sets} (see Appendix \ref{app:simulations} for sensitivity analyses of each parameter). Let MICH-Auto denote the version of the model with $J$ selected based on the ELBO (see Appendix \ref{app:LKJ-choice}), and let MICH-Ora be the oracle model with $J=J^*$. For comparison, we fit two state-of-the-art methods that can also detect joint mean and variance changes and return confidence intervals: H-SMUCE (\citealp{Pein17}) and the MOSUM method of \cite{Cho22}. We also compare to the pruned exact linear time (PELT; \citealp{Killick12}) and narrowest over the threshold (NOT; \citealp{Baranowski19}) methods, both of which can detect joint mean and variance changes, but do not provide measures of uncertainty quantification. Lastly, we include narrowest significance pursuit (NSP; \citealp{Fryzlewicz24}). Instead of estimating $\boldsymbol{\tau}_{1:J^*}$, NSP returns localized regions of $\mathbf{y}_{1:T}$ that contain a change-point at a prescribed global significance level. We compare these sets to the credible sets returned by MICH. Details regarding the implementation of each method are provided in Appendix \ref{app:simulations}.
\vspace{-5pt}

\setcounter{algocf}{0}
\begin{simulation}[!h]
\begin{enumerate}
\small
    \itemsep0em 
    \item Fix the number of observations $T$, the number of change-points $J^*$, the minimum spacing condition $\Delta_T$, and a constant $C > 0$.
    \item Draw $\boldsymbol{\tau}_{1:J^*}$ uniformly from $[T]$ subject to the minimum spacing condition $|\tau_{j+1} - \tau_j| \geq \Delta_T$ with $\tau_0 = 1$ and $\tau_{J^*+1} = T+1$. 
    \item Set $\mu_0 :=0$ and $s_0:=1$ and for $j >0$, draw $U_j \sim \text{Uniform}(-2,2)$ and $\xi_j \sim \text{Bernoulli}(0.5)$ and set: 
    \vspace{-5pt}
    \begin{align*}
        s_j &:= 2^{U_j} \\
        \mu_{j} &:= \mu_{j-1}  + (1-2\xi_j)C\left(\min\{s_{j}^{-1}\sqrt{\tau_{j+1} - \tau_j}, s_{j-1}^{-1}\sqrt{\tau_j - \tau_{j-1}}\}\right)^{-1}.
    \end{align*}
    \vspace{-15pt}
    \item Draw $e_t \overset{\text{ind.}}{\sim} \mathcal{N}\left(0,1 \right)$ and set $y_t:= \sum_{j=0}^{J^*} \mu_j\mathbbm{1}_{\{\tau_j \leq t < \tau_{j+1}\}} + e_t \sum_{j=0}^{J^*} \sigma_j\mathbbm{1}_{\{\tau_j \leq t < \tau_{j+1}\}}$.
    \vspace{-10pt}
\normalsize
\end{enumerate}
\caption{Mean-Variance Simulation Study}
\label{sim:main}
\end{simulation}

For each method we evaluate the accuracy of the returned estimates $\hat{\boldsymbol{\tau}}_{1:\hat{J}}$ by calculating the Hausdorff distance (\ref{eq:hausdorff}) between the estimated change-points and the true locations $\boldsymbol{\tau}_{1:J^*}$, as well as the false positive sensitive location error (FPSLE; \ref{eq:fpsle}) and false negative sensitive location error (FNSLE; \ref{eq:fnsle}). We also calculate the bias $|J^* - \hat{J}|$ to assess the ability of each method to recover the correct number of changes. For any method that returns $\alpha$-level confidence or credible sets $\{\mathcal{CS}(\alpha, \hat{\tau}_j)\}_{j=1}^{\hat{J}}$, we calculate the average set size as well as the coverage conditional on detection (CCD; \ref{eq:ccd}), i.e. the coverage of $\tau_j$ given that $\tau_j$ is within a $\min\{T^{1/2}/2, 15\}$ sized window of some estimated $\hat{\tau}_i$ \citep{Killick12}. 

\begin{figure}[!h]
    \centering
    \includegraphics[scale =0.42]{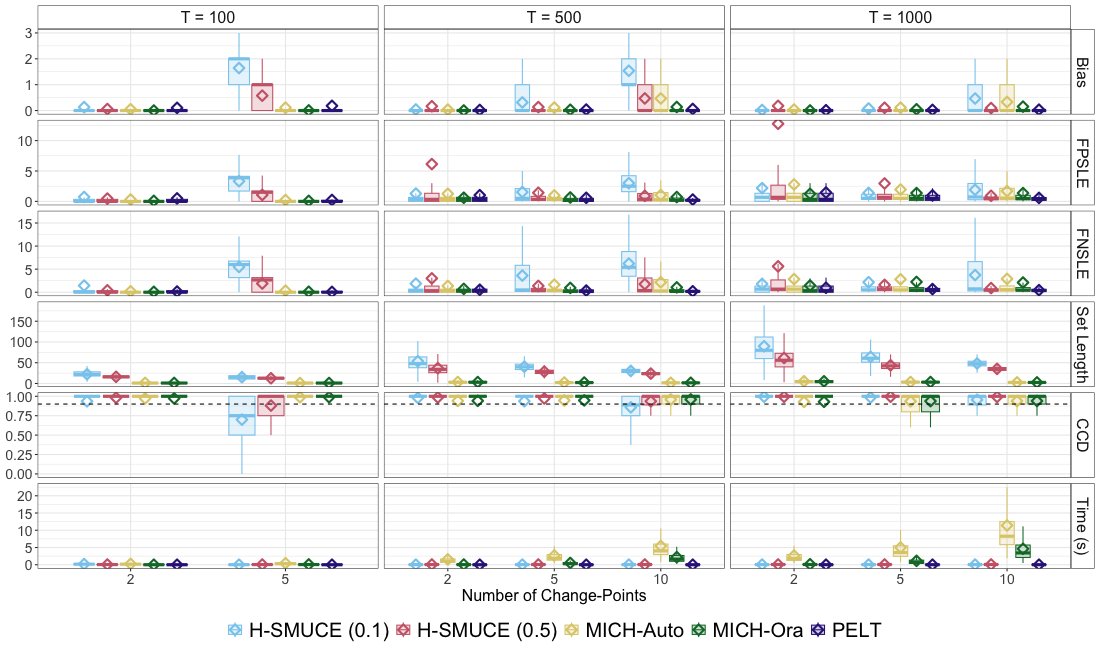}
    \caption{\small \textbf{Simulation \ref{sim:main} Results.} Box-plots of evaluation metrics aggregated over each value of $\Delta_T$ for 5,000 replicates of Simulation \ref{sim:main}. Diamonds ($\Diamond$) display mean of each statistic. Bias $|J^* - J|$ assesses each model's ability to estimate correct number of changes (lower is better). FPSLE and FNSLE assess each model's ability to accurately estimate the locations of the changes (lower is better). Set Length and CCD report the average size and coverage of credible/confidence sets for methods that provide uncertainty quantification (dashed line indicates nominal coverage level for $\alpha = 0.1$). Time reports run-time for each method in seconds.}
    \label{fig:sim_plot}
\end{figure}

Figure \ref{fig:sim_plot} summarizes the results of this simulation study. Note that for visual clarity we have omitted the Hausdorff statistic as well as the results for NOT, MOSUM, and NSP. Complete results for all methods are available in Appendix \ref{app:main_sim_results}. Our results show that MICH is on par with the state-of-the-art in terms of ability to estimate the number and locations of the changes this study. We see that MICH-Ora, MICH-Auto, and PELT are all virtually unbiased for the true number of changes $J^*$, with the average bias increasing negligibly as $J^*$ increases relative to $T$. Both MOSUM and NSP display significant bias for the number of changes. For H-SMUCE, $\alpha$ upper bounds the probability of overestimating $J^*$, but we see that when $J^*$ is large relative to $T$, this false discovery control results in H-SMUCE exhibiting much larger bias relative to MICH-Auto and PELT. Relaxing $\alpha$ from  0.1 to 0.5 decreases the magnitude of the error across each setting.

Regarding the accuracy of the estimated change-points, both MICH-Ora and MICH-Auto are indistinguishable from PELT in terms of FPSLE, while MICH slightly lags behind PELT in terms of the average FNSLE, indicating MICH may be missing some true changes even when $J=J^*$. Among the methods that provide uncertainty quantification, MICH and H-SMUCE clearly outperform the others. However, the accuracy of H-SMUCE heavily depends on the sample size $T$ and choice of $\alpha$. When $J^*/T$ is large and  $\alpha = 0.1$, H-SMUCE performs worse than either PELT or MICH across each metric, particularly for small $T$. On the other hand, if we are willing to compare MICH and PELT at different significance levels, then H-SMUCE beats MICH in terms of the FNSLE at $\alpha = 0.5$ and is even competitive with PELT. The reverse of this relationship with $\alpha$ is observed when $J^*/T$ is small. Thus, unlike PELT or MICH, H-SMUCE requires either a way to tune $\alpha$ or knowledge of $J^*/T$. 

Both MICH-Ora and MICH-Auto achieve the nominal coverage level of $90\%$ in each setting, while at the same time returning credible sets that are smaller than those returned by H-SMUCE, MOSUM, and NSP by a factor of 10-50. This empirical result reflects the desirable localization properties we established in Section \ref{sec:localization}. Though H-SMUCE achieves the desired coverage level for both $\alpha = 0.1$ and $0.5$, we see the counterintuitive result that the coverage of the H-SMUCE intervals \textit{increases} for larger $\alpha$. As a final note, we point out that although MICH converges much faster than a Gibbs sampler implementing (\ref{eq:y-mich})-(\ref{eq:s-mich}) would, the time needed to fit MICH still lags behind all other models except for NSP. Computation time decreases as we relax the convergence criterion, and the sensitivity analysis in Appendix \ref{app:simulations} shows that setting $\epsilon=10^{-5}$ brings the time to fit MICH in line with the other methods without significantly harming performance. 

\subsection{Additional Experiments}

In Appendix \ref{app:simulations}, we modify Simulation \ref{sim:main} to allow for autocorrelation and heavy-tailed errors. Our results indicate that MICH is somewhat robust to violations of independence and Gaussianity, particularly in comparison to PELT and NOT. We also introduce Simulation \ref{sim:multi} to test the multivariate version of MICH described in Algorithm \ref{alg:mich-multi} and compare its performance to other methods designed to detect multivariate mean changes, including: E-Divisive \citep{James15}, the informative sparse projection method (Inspect; \citealp{Wang17}), and the Two-Way MOSUM ($\ell^2$-HD) method of \cite{Li23}. MICH is the clear winner in Simulation \ref{sim:multi}, outperforming the other methods across each setting considered. As per the discussion of Assumption \ref{assumption:mean}, MICH is effective at detecting changes in dense signals. Conversely, Inspect is intended for sparse signals, making it somewhat surprising that MICH is more effective than Inspect at localizing sparse changes in our study. On the other hand, $\ell^2$-HD is also designed for dense signals, but even in the setting where every coordinate of $\mathbf{y}_{1:T}$ changes, MICH is still more effective at identifying those changes. 
\section{Applications}

\label{sec:data}

\subsection{Ion Channel Gating}

Porins are $\beta$-barrel shaped proteins present in the outer membrane of Gram-negative bacteria, as well as the mitochondria and plastids of eukaryotic cells. These proteins form molecule-specific channels that allow ions and certain other small molecules to permeate the outer membrane of the cell via passive diffusion (\citealp{SCHIRMER1998101}). A loop structure inside the porin controls whether the channel is in an open or partially blockaded state. Antibiotics enter bacteria via porins, so understanding the permeability of these channels, including being able to detect when they are in an open or closed state, plays an important role in the development of new drugs (\citealp{Hermansen22}). We seek to detect the gating of the porin by applying MICH to a sequence of ionic current recordings that were collected by the Steinem lab (Institute of Organic and Biomolecular Chemistry, University of G\"ottingen) using the voltage clamp technique (\citealp{neher1978extracellular}). We restrict our analysis to the same subsequence of 32,511 readings originally analyzed by \cite{Cappello21} and subsample every 11\textsuperscript{th} observation to account for local dependencies introduced during the collection process for this data (\citealp{Hotz13,Pein17}), leaving $T =$ 2,956 ionic current readings. 

\begin{figure}[!h]
    \centering
    \includegraphics[scale=0.37]{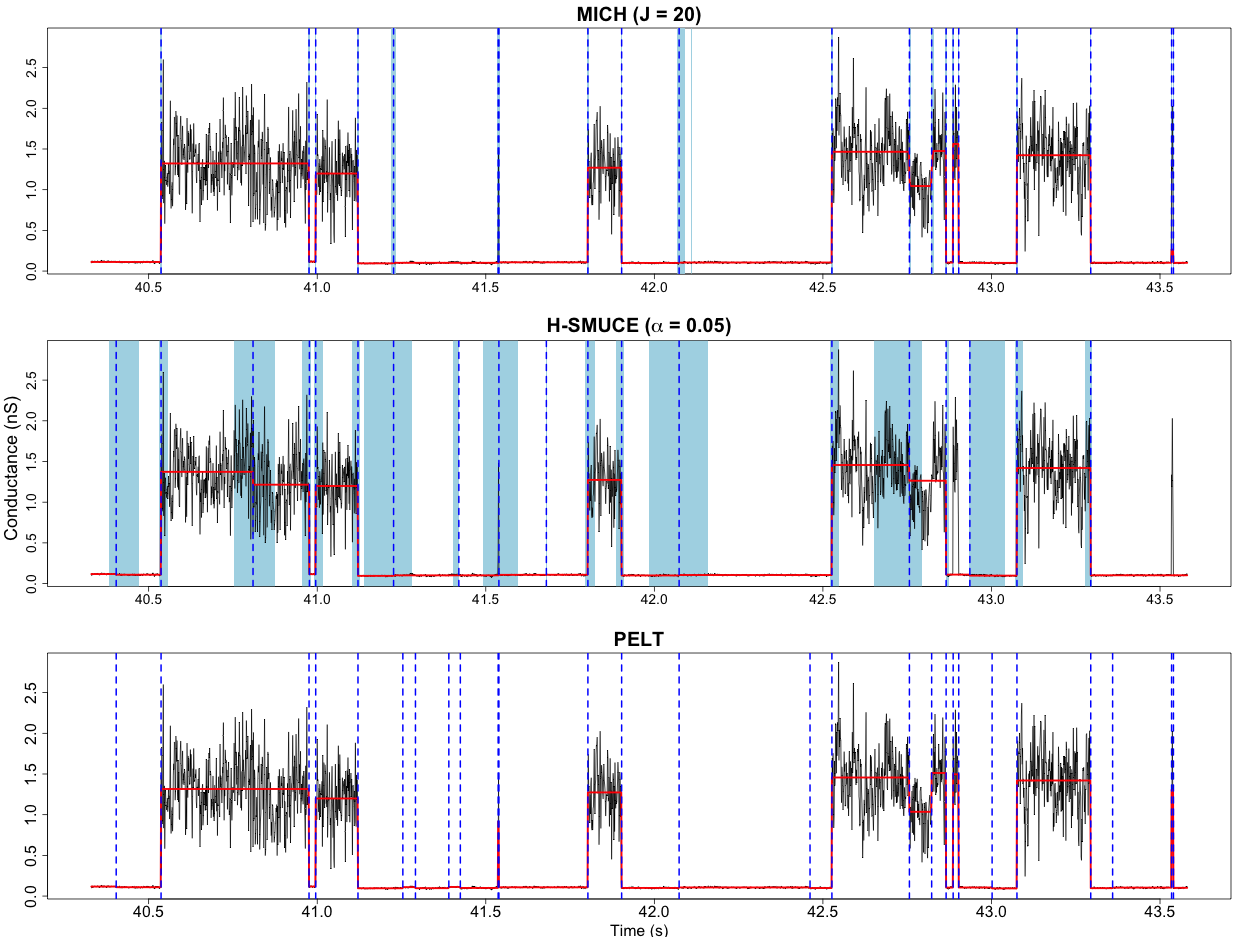}
    \caption{\small \textbf{Ion Channel Data}. Estimated mean signal (\textcolor{red}{\raisebox{0.5ex}{\rule{0.3cm}{1.5pt}}}) and change-points (\textcolor{blue}{\makebox[10pt][l]{\hdashrule[0.5ex]{10pt}{1.5pt}{4pt 1pt}}}) for 2,956 subsampled ionic current recordings (\textcolor{black}{\raisebox{0.5ex}{\rule{0.3cm}{1.5pt}}}) from a PorB porin collected by the Steinem lab. \textbf{Top}: MICH fit with $J = 20$ components and $\hat{J} = 20$ with 95\% credible sets shaded (\textcolor{cyan}{\raisebox{0.25ex}{\rule{0.3cm}{4pt}}}). \textbf{Middle}: H-SMUCE fit with $\hat{J} = 19$ and 95\% confidence intervals shaded (\textcolor{cyan}{\raisebox{0.25ex}{\rule{0.3cm}{4pt}}}). \textbf{Bottom}: PELT fit with $\hat{J} = 27$.}
    \label{fig:ion1}
\end{figure}

A key feature of this data set is that the readings exhibit open-channel noise, i.e. the variance increases significantly when the channel is open (see Figure \ref{fig:ion1}). To account for the simultaneous changes to $\boldsymbol{\mu}_{1:T}$ and $\boldsymbol{\lambda}_{1:T}$, we fit MICH with the restriction $L=K=0$, the non-informative prior $\omega_0=u_0=v_0=10^{-3}$, and use the ELBO to automatically select $J = 20$. Figure \ref{fig:ion1} displays the estimated mean signal and $\hat{J} = 20$ change-points with 95\% credible sets. MICH is able to detect each of the 16 visible gating events, including each of the three brief channel openings around seconds 41.54, 42.9, and 43.54. PELT and H-SMUCE are also capable of detecting simultaneous mean and variance changes. PELT estimates $\hat{J} = 27$ change-points, including all of the visible gating events, while H-SMUCE with $\alpha = 0.05$ estimates $\hat{J} = 19$ changes, but misses two of the short-lived changes. 

All three methods detect a number of potentially spurious changes, e.g. the when the channel is closed between seconds 42 and 42.5 and open around 42.7. \cite{Pein17} suggested that these detections indicate possible violations of the piece-wise constant assumption due to small holes in the cell membrane. We note that MICH detects about half as many of these changes as compared to PELT and H-SMUCE, and the uncertainty measures returned by MICH and H-SMUCE indicate that the models are least confident about the locations of these changes. In general, the credible sets returned by MICH are much smaller than the corresponding confidence intervals returned by H-SMUCE, as was the case in the simulation results of Section \ref{sec:simulations}. We also fit MICH with $J=K=0$ and $L=17$ (selected using the ELBO), as well as H-SMUCE with $\alpha = 0.5$ (see Figure \ref{fig:ion2}). Even at $\alpha = 0.5$, H-SMUCE still misses the channel openings around seconds 42.9 and 43.54. The performance of the $L=17$ MICH model is surprising in that it only misses the small spike near 41.54 despite not accounting for the noise heterogeneity, giving further credence to the robustness of Bayesian methods to misspecification. (\citealp{lai2011simple}).

\subsection{Shankle Oil Well Lithology}

Characterizing the lithology of the Earth's crust is a key task for resource extraction. When drilling for oil and gas, knowing the boundaries between lithofacies (stratified rock layers) and adjusting the pressure in the borehole accordingly can prevent catastrophic blow-outs (\citealp{FearnheadClifford03}), and a map of the well's structure can predict permeability and thus total gas production (\citealp{Bohling03}). We use MICH to detect changes in the facies of the B1-C formations of the Shankle well located in the Panoma oil and gas field of southwest Kansas. The data were collected by the Kansas Geological Survey by lowering a probe into a borehole and taking $T = 346$ recordings of $d=6$ petrophysical measures every half foot, including: average neutron and density porosity (AvgPhi), difference of neutron and density porosity (DeltaPhi), gamma ray emissions (GR), a marine/non-marine indicator (MnM), photoelectric factor (PE), and deep resistivity (Rt). A core sample taken from the well provides ground truth for the change-points between the facies in this case. We fit the multivariate version of MICH (see Section \ref{app:multi-mich}) with the non-informative prior $\omega_0 =0.001$ and again use the ELBO to automatically select $L = 42$. Figure \ref{fig:well_mich} shows the lithology of the well along with the estimated change-points and their 99\% credible sets. 

\begin{figure}[!h]
    \centering
    \includegraphics[scale=0.35]{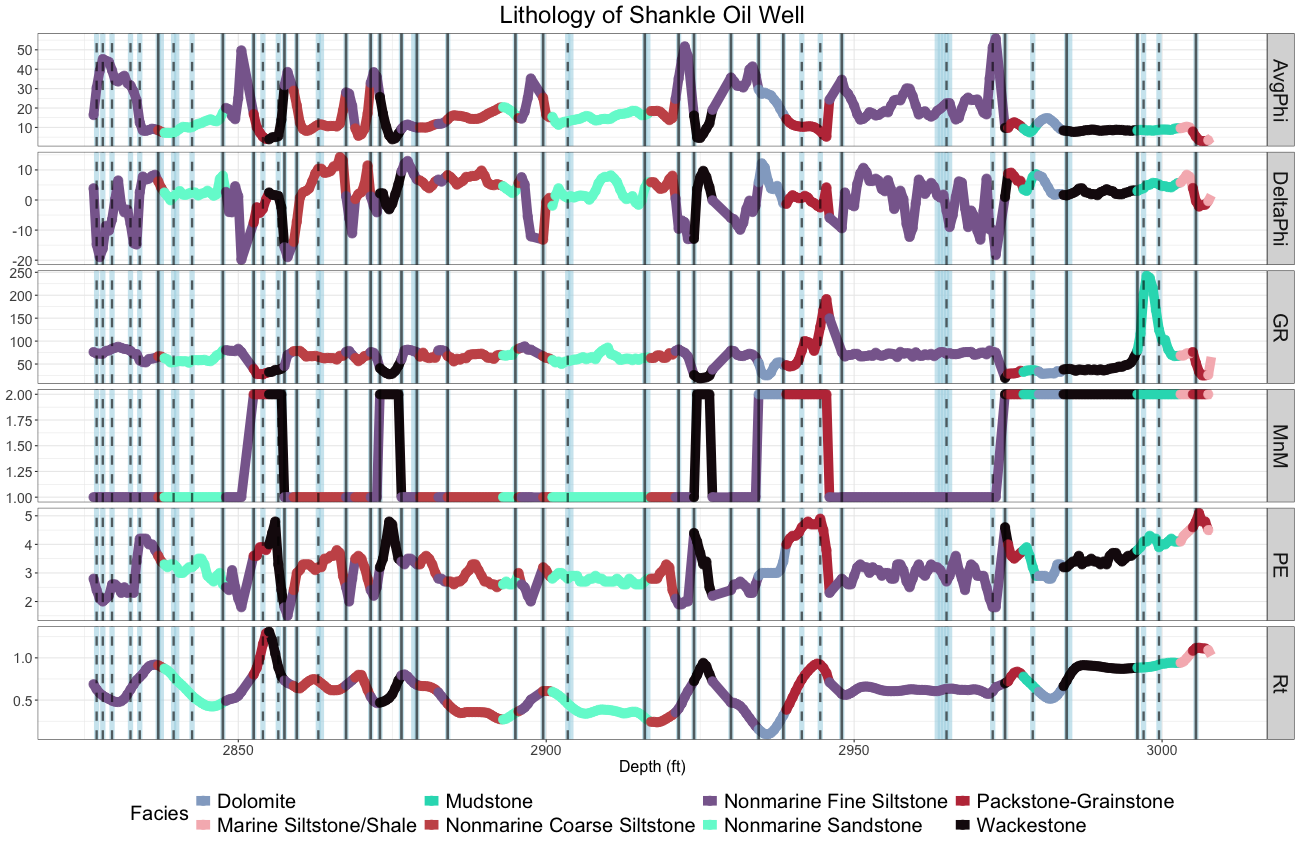}
    \caption{\small \textbf{Shankle Well Log.} Recordings from six petrophysical measures in the Shankle oil well. Color of segment indicates the corresponding lithofacie with changes in color indicating a true change-point. $\hat{L} =42$ estimated changes (\textcolor{black}{\makebox[10pt][l]{\hdashrule[0.5ex]{10pt}{1.5pt}{4pt 1pt}}}) with 99\% credible sets shaded (\textcolor{cyan}{\raisebox{0.25ex}{\rule{0.3cm}{4pt}}}). 26 of the estimated changes or their 99\% credible sets are within one index of a true change (\raisebox{0.5ex}{\rule{0.3cm}{1pt}}) and 24 of the true changes are covered by a credible set.}
    \label{fig:well_mich}
\end{figure}

Due to the many change-points in short succession, detecting changes in the lithology of the well is a particularly challenging task. Nonetheless, MICH is able to pick out most of the salient structural breaks in the data. Of the 36 true changes, 13 are covered by our 99\% credible sets and another 13 are missed only by a single index. Many of the false positives appear to stem from violations of the piece-wise constant assumption, e.g. MICH detects two changes in the final mudstone layer near 3,000 feet due to a trend in the GR measure. Similarly, MICH detects a change at 2,979 feet, which just misses the transition from mudstone to dolomite at 2,980 feet, likely due to the smoothness of the transition between these two layers. As a comparison, we apply the sparse projection method of \cite{Wang17} and the two-way MOSUM method of \cite{Li23} to the data. The results are displayed in Figures \ref{fig:well_inspect} and \ref{fig:well_l2}. Both methods perform considerably worse than MICH. The method of \cite{Wang17} detects $\hat{L} = 140$ change-points, perhaps due to the fact that this method is designed for sparse changes in the mean. Like MICH, the method of \cite{Li23} is also designed for dense changes. Despite this, their method lacks power in this application, detecting a single change-point that is also a false positive. 
\section{Discussion}
\label{sec:discussion}

Bayesian methods are not commonly used in change-point detection, as they are generally slower than competing approaches, less well understood from a theoretical perspective, and often produce posterior distributions that are challenging to summarize in terms of practical uncertainty measures. Our proposal addresses each of these limitations. First, our Bayesian SCP models possess theoretical properties that are optimal in minimax sense across a range of settings, including changes in the mean for high-dimensional $\mathbf{y}_{1:T}$, changes in the mean, variance, or both for univariate $\mathbf{y}_{1:T}$, and in the presence of dependence in $\mathbf{y}_{1:T}$. Second, these SCP models can be combined in a way that enables efficient inference for more complex models. Through empirical illustrations, we demonstrate that our method is highly competitive in terms of point estimation, while also producing credible sets that are narrower than existing alternatives--without compromising coverage.

%We introduced a Bayesian method to detect and quantify the uncertainty around change-points in the mean and variance of a time-series. In Section \ref{sec:scp} we showed that MICH can localize a single change at the minimax optimal rate. In Section \ref{sec:simulations}, we showed that MICH is competitive in simulations with state-of-the-art methods like PELT and H-SMUCE. In our study, MICH returned credible sets that were an order of magnitude smaller than competitors without sacrificing nominal coverage guarantees, indicating that the optimal properties of MICH seem to carry over to the multiple change-point setting, at least empirically. Finally, in Section \ref{sec:data} we applied MICH to real world data from the biological and geological domains. We saw that our method is on par with existing methods, and perhaps much better in the multivariate setting. 

Though we have focused on Gaussian data in this article, our general framework extends to other data types. As a proof of concept, when $y_t\:|\: \lambda_t \sim \text{Poisson}(\lambda_t)$, we can readily adapt MICH to detect changes in $\boldsymbol{\lambda}_{1:T}$ (see Appendix \ref{app:poisson}). The construction of Algorithm \ref{alg:mich-poisson} for this setting mirrors our approach for Gaussian scale parameters, providing evidence that MICH's rationale extends naturally to changes in both scale and location parameters across a broader class of models. Our theoretical results did not require an assumption of Gaussian data, further supporting the generality of the approach. %\textcolor{red}{Do we want to mention this? [L: agree, I would leave out]: } Beyond time-series, $\mathbf{y}_{1:T}$ can represent the leaf nodes of a binary tree with a known structure. The tree could represent the evolutionary path of a virus with $\mathbf{y}_{1:T}$ being outcome data like severity or the number of hospitalizations associated with each specific strain. \cite{Behr20} adapted H-SMUCE to this setting in order identify nodes on the tree that represent change-points in the signal underlying $\mathbf{y}_{1:T}$. Our results indicate that adapting MICH to this setting may lead to improved uncertainty quantification. \textcolor{red}{we could mention binary data instead}

MICH also shows promise for the on-line change-point detection setting. For example, after fitting MICH to $\mathbf{y}_{1:T}$ with $L$ mean changes, we use this fit to initialize an $L+1$ component model for $\mathbf{y}_{1:(T+1)}$, continuing in this manner until a new change-point is detected with the $L+1$ component. To control for false discoveries in this setting, we can adapt the approach to constructing credible sets described in \cite{Spector24}. This would allow us to apply MICH problems like detecting the occurrence deforestation events (\citealp{Wendelberger21}) or solar flares (\citealp{Xu21}) using satellite imagery. We leave these and other questions for future research. 

% As a concluding remark, we remind the reader that MICH only returns an approximation to the posterior distribution of the change-points. The quality of this approximation as well as any coverage guarantees for our credible sets remain unknown in the multiple change-point setting. Recent work by \cite{Zhang20} establishes a framework for studying the convergence of variational approximations. Further study is warranted to see if we can leverage these results to characterize the output of MICH. Furthermore, our promising empirical results suggest that it may be possible to extend the single-change point theory in Section \ref{sec:localization} to the multiple change-point case. We leave these and other questions for future research. 

\setlength{\bibsep}{5pt}
\bibliographystyle{chicago}
\bibliography{MICH/references}

\newpage
\appendix
\hypertarget{appendix}{\textbf{\LARGE Appendix}}
\section{Proofs of Localization Rates}
\label{app:localization-rates}

\subsection{Notation and Definitions}
\label{app:notation}

\subsubsection{Asymptotic Analysis} 

In the following proofs, for some functions $f$ and $g$ we use the notation $f(T) = \mathcal{O}(g(T))$ to mean that there exists some constant $M > 0$ and some $T^* > 0$ so that for all $T > T^*$ we have: $$|f(T)| \leq M g(T).$$ Similarly, we use $g(T) \gtrsim f(T)$ to indicate that there exists $M > 0$ so that for any $T$ we have $f(T) \leq M g(T).$ We also use $f(T) = o(g(T))$ to mean:$$\lim_{T\to\infty} \frac{f(T)}{g(T)} = 0$$ and $f(T) \sim g(T)$ to mean: $$\lim_{T\to\infty} \frac{f(T)}{g(T)} = 1.$$ We also denote the ball $\mathbb{B}_{r}^{p}(\mathbf{x}_0) := \{\mathbf{x}\in\mathbb{R}^d : \lVert \mathbf{x}_0 - \mathbf{x} \rVert_p \leq r\}$ with $\mathbb{B}_{r}^{p} \equiv \mathbb{B}_{r}^{p}(\mathbf{0})$.

\subsubsection{Random Variables} 

For a generic real-valued random variable $X$, we use $p(x)$ to denote the density of $X$ when it exists. If $\sigma(X)$ is the $\sigma$-algebra generated by $X$, then we use $\Pr(\cdot)$ to indicate the measure induced by $X$ on the measurable space $(\mathbb{R},\sigma(X))$.

\subsubsection{Sub-Gaussian Distributions}

We write $X\in\mathcal{SG}(\sigma)$ to denote that the random variable $X$ has a sub-Gaussian distribution with parameter $\sigma$, i.e. for all $\lambda \in \mathbb{R}$: 
\begin{align} 
    \E[\exp(\lambda X)] \leq \exp\left[\frac{\lambda^2\sigma^2}{2}\right]. \label{def:sub-gaussian}
\end{align}
For $t, \lambda > 0$, we have:
\begin{align}
     \Pr(|X| > t) &\leq \Pr(X > t) + \Pr(-X > t) \tag{union bound} \\
     &= \Pr(e^{\lambda X} > e^{\lambda t}) + \Pr(e^{-\lambda X} > e^{\lambda t}) \notag \\
     &\leq \E[\exp(\lambda X - \lambda t)] + \E[\exp(-\lambda X - \lambda t)] \tag{Markov inequality} \\
     &\leq 2\exp\left(\frac{\lambda^2\sigma^2}{2} - \lambda t\right). \tag{by (\ref{def:sub-gaussian})} 
\end{align}
The last bound is minimized be setting $\lambda = \frac{t}{\sigma^2}$, which yields the Chernoff bound:
\begin{align} \label{eq:chernoff}
    \Pr(|X| > t) \leq 2\exp\left[-\frac{t^2}{2\sigma^2}\right]
\end{align}
The sub-Gaussian norm of a random variable is defined as $\lVert X \rVert_{\psi_2} := \inf\{t > 0 \::\: \E[\exp(X^2 /t^2)] \leq 2\}$. By Proposition 2.5.2 of \cite{Vershynin18}, $X$ is sub-Gaussian if and only if $\lVert X \rVert_{\psi_2} < \infty$.

\subsubsection{Sub-Exponential Distributions}

We write $X\in\mathcal{SE}(\nu, \alpha)$ to denote that the random variable $X$ has a sub-Exponential distribution with parameters $\nu$ and $\alpha$, i.e. for all $|\lambda| \leq \frac{1}{\alpha}$: 
\begin{align*}
    \E[\exp(\lambda X)] \leq \exp\left[\frac{\lambda^2\nu^2}{2}\right].
\end{align*}
If $X\in\mathcal{SE}(\nu, \alpha)$, then by Proposition 2.9 of \cite{Wainwright19}:
\begin{align}\label{eq:wainwright_prop_2.9}
    \Pr(|X - \E[x]| \geq t) \leq 
    \begin{cases}
        2 \exp\left[-\frac{t^2}{2\nu^2}\right], & \text{if } t \in \left(0, \frac{\nu^2}{\alpha}\right], \\
        2 \exp\left[-\frac{t}{2\alpha}\right], & \text{if } t \in \left(\frac{\nu^2}{\alpha}, \infty\right).
    \end{cases}
\end{align}
As per Example 2.11 of \cite{Wainwright19}, if $X \sim \chi^2_n$, then $X\in\mathcal{SE}(2\sqrt{n}, 4)$, and thus: 
\begin{align} \label{eq:chi2-ineq}
    \Pr\left(\frac{|X - n|}{n} \geq t\right) \leq 
    \begin{cases}
        2 \exp\left[-\frac{n t^2}{8}\right], & \text{if } t \in \left(0, 1\right], \\
        2 \exp\left[-\frac{n t}{8}\right], & \text{if } t \in \left(1, \infty\right).
    \end{cases}
\end{align}
The sub-exponential norm of a random variable is defined as $\lVert X \rVert_{\psi_1} := \inf\{t > 0 \::\: \E[\exp(X /t)] \leq 2\}$. By Proposition 2.7.1 of \cite{Vershynin18}, $X$ is sub-exponential if and only if $\lVert X \rVert_{\psi_1} < \infty$.

\subsubsection{$\alpha$-mixing}

\begin{definition}\label{def:alpha-mixing}
Let $\{X_t\}_{t\in \mathbb{Z}}$ be a stochastic process on the probability space $(\Omega, \mathcal{F}, \Pr)$, then $\{X_t\}_{t\in \mathbb{Z}}$ is said to be $\alpha$-mixing if:
\begin{align*}
    \lim_{k\to\infty} \alpha_k(\{X_t\}_{t\in\mathbb{Z}}) = 0,
\end{align*}
where:
\begin{align*}
    \alpha_k(\{X_t\}_{t\in\mathbb{Z}}) := \sup_{t\in\mathbb{Z}} \; \alpha\left(\sigma(\{X_s\}_{s \leq t}), \; \sigma(\{X_s\}_{s \geq t + k})\right).
\end{align*}
Here $\sigma(Y)$ is the $\sigma$-algebra generated by $Y$ and the strong mixing, or $\alpha$-mixing, coefficient between two $\sigma$-algebras $\mathcal{A}, \mathcal{B} \subseteq \mathcal{F}$ is defined as:
\begin{align*}
    \alpha(\mathcal{A}, \mathcal{B}) := \sup_{A\in\mathcal{A}, B\in\mathcal{B}} |\Pr(A \cap B) - \Pr(A)\Pr(B)|.
\end{align*}
To simplify notation, we will often write $\alpha_k$ in place of $\alpha_k(\{X_t\}_{t\in\mathbb{Z}})$.
\end{definition}

\subsection{Concentration Results}
\label{app:concentration}

\subsubsection{Sums of sub-Gaussian and sub-Exponential Processes}

\begin{lemma}\label{lemma:sum-sub-gaussian}
If $\{X_i\}_{i=1}^n$ is a collection of random variables such that for each $i \in [n]$, $\E[X_i] = 0$ and $X_i \in \mathcal{SG}(\sigma_i)$ for some $\sigma_i > 0$, then letting $\boldsymbol{\sigma} = \{\sigma_i\}_{i=1}^n$, we have:
\begin{align*}
    \sum_{i=1}^n X_i \in 
    \begin{cases}
        \mathcal{SG}(\lVert \boldsymbol{\sigma}\rVert_2), &\text{if $\{X_i\}_{i=1}^n$ are mutually independent,} \\
        \mathcal{SG}(\lVert \boldsymbol{\sigma}\rVert_1), &\text{otherwise.}
    \end{cases}
\end{align*}
Similarly, if $X_i \in \mathcal{SE}(\nu_i, \alpha_i)$ for some $\sigma_i,\:\alpha_i > 0$, then letting $\boldsymbol{\nu} = \{\nu_i\}_{i=1}^n$ and $\boldsymbol{\alpha} = \{\alpha_i\}_{i=1}^n$ , we have:
\begin{align*}
    \sum_{i=1}^n X_i \in 
    \begin{cases}
        \mathcal{SE}(\lVert \boldsymbol{\nu}\rVert_2, \lVert \boldsymbol{\alpha}\rVert_\infty), &\text{if $\{X_i\}_{i=1}^n$ are mutually independent,} \\
        \mathcal{SE}(\lVert \boldsymbol{\nu}\rVert_1,  \lVert \boldsymbol{\alpha}\rVert_\infty), &\text{otherwise.}
    \end{cases}
\end{align*}
\end{lemma}

\begin{proof}
First suppose that $X_i \in \mathcal{SE}(\nu_i, \alpha_i)$ and that $\{X_i\}_{i=1}^n$ is a collection of independent random variables, then by independence:
\begin{align*}
    \E\left[\exp\left(\lambda \sum_{i=1}^n X_i\right)\right] &=  \prod_{i=1}^n \E\left[\exp\left(\lambda X_i\right)\right].
\end{align*}
Next, note that if $|\lambda| < \lVert \boldsymbol{\alpha}\rVert^{-1}_{\infty}$, then $|\lambda| < \alpha_i^{-1}$ for each $i \in [n]$ and thus:
\begin{align*}
    \E\left[\exp\left(\lambda X_i\right)\right] \leq \exp\left[\frac{\lambda^2\nu_{i}^2}{2}\right]
\end{align*}
which implies:
\begin{align*}
    \E\left[\exp\left(\lambda \sum_{i=1}^n X_i\right)\right] &\leq \exp\left[\frac{\lambda\sum_{i=1}^n\nu^2_i}{2}\right].
\end{align*}
So $\sum_{i=1}^n X_i \in \mathcal{SE}(\lVert \boldsymbol{\nu}\rVert_2, \lVert \boldsymbol{\alpha}\rVert_\infty)$. When $\{X_i\}_{i=1}^n$ are not independent, then for any $|\lambda| \leq \lVert \boldsymbol{\alpha}\rVert^{-1}_{\infty}$ and any collection of parameters $\{p_i\}_{i=1}^n$ such that $p_i > 0$ and $\sum_{i=1}^n p_i^{-1} = 1$, we have:
\begin{align*}
    \E\left[\exp\left(\lambda \sum_{i=1}^n X_i\right)\right] &= \E\left[ \prod_{i=1}^n \exp\left(\lambda X_i\right)\right] \\
    &\leq \prod_{i=1}^n \left(\E[\exp\left(p_i \lambda X_i\right)]\right)^{1/p_i} \tag{H\"older's inequality} \\
    &\leq \prod_{i=1}^n \left(\E\left[\exp\left(\frac{p^2_i\nu_i^2 \lambda^2}{2}\right)\right]\right)^{1/p_i} \tag{$X_i \in \mathcal{SE}(\nu_i,\alpha_i)$} \\
    &= \exp\left[\frac{\lambda^2\sum_{i=1}^n p_i \nu_i^2}{2}\right].
\end{align*}
Since our choice of $\{p_i\}_{i=1}^n$ was arbitrary so long as $p_i > 0$ and $\sum_{i=1}^n p_i^{-1} = 1$, we can set each $p_i = \frac{\lVert\boldsymbol{\nu}\rVert_1}{\nu_i},$ then we get:
\begin{align*}
    \E\left[\exp\left(\lambda \sum_{i=1}^n X_i\right)\right] &\leq \exp\left[\frac{\lambda \lVert \boldsymbol{\nu}\rVert^2_1}{2}\right]
\end{align*}
showing that $\sum_{i=1}^n X_i \in \mathcal{SE}(\lVert\boldsymbol{\nu}\rVert_1, \lVert \boldsymbol{\alpha}\rVert_\infty).$ The proof for the case where $X_i \in \mathcal{SG}(\sigma_i)$ follows an identical argument as above with $\sigma_i$ replacing $\nu_i$ and ignoring any constraints on $\lambda$. 

\end{proof}

\subsubsection{Sums of $\alpha$-Mixing Processes}

\begin{lemma}[\citealp{Padilla23} Lemma 3]\label{lemma:Padilla23}
Let $\nu > 0$ be given. Suppose that $\{X_t\}_{t=1}^\infty$ is a stationary $\alpha$-mixing time-series with mixing coefficients $\{\alpha_k\}_{k=0}^K$. Suppose that $\E[X_t] = 0$ and that there exist constants $\delta, \Delta, D_1, D_2 > 0$ such that:
\begin{align*}
    \sup_{t \geq 1}\; \E\left[\left|X_t\right|^{2+\delta+\Delta}\right] \leq D_1 
\end{align*}
and:
\begin{align*}
    \sum_{k=0}^\infty (k+1)^{\frac{\delta}{2}} \alpha_k^{\frac{\Delta}{2+\delta+\Delta}} \leq D_2.
\end{align*}
Then there exists some universal constant $C > 0$ such that for any $a \in (0,1)$:
\begin{align*}
    \Pr \left(\left|\sum_{t'=1}^t X_{t'}\right| \leq a^{-1}C\sqrt{t}\left[\log (\nu t) + 1\right], \;\sforall t\geq \nu^{-1}\right) 
    \geq 1 - a^2.
\end{align*}

\end{lemma}

\begin{lemma}\label{lemma:alpha-mix-sum}
Suppose that $\{X_t\}_{t\geq 1}$ is a stationary, $\alpha$-mixing processes with mixing coefficients and moments that satisfy the conditions of Lemma \ref{lemma:Padilla23}. Then for $T \geq 1$, $t_0 \in [T]$, and sequence $\{a_T\}_{T\geq 1}$ such that $a_T > 0$ and $\lim_{T\to\infty} a_T = \infty$, there exists a universal constant $C > 0$ such that:
\begin{align*}
    \Pr \left(\bigcup_{t=t_0 + 1}^{T}\left\{ \left|\sum_{t'=t_0}^{t-1} X_{t'}\right| > Ca_T\sqrt{t-t_0}\log T\right\}\right) 
    &\leq \frac{1}{a_T^2},
\end{align*}
and:
\begin{align*}
    \Pr \left(\bigcup_{t=1}^{t_0-1}\left\{ \left|\sum_{t'=t}^{t_0-1} X_{t'}\right| > Ca_T\sqrt{t_0-t}\log T\right\}\right) 
    &\leq \frac{1}{a^2_T},
\end{align*}
and:
\begin{align*}
    \Pr \left(\bigcup_{t=1}^{T}\left\{ \left|\sum_{t'=t}^{T} X_{t'}\right| > Ca_T\sqrt{T-t+1}\log T\right\}\right) 
    &\leq \frac{1}{a_T^2}.
\end{align*}

\end{lemma}

\begin{proof}
By Lemma \ref{lemma:Padilla23}, there is some constant $C > 0$ so that if we let $\nu = 1$ and $a = a_T^{-1}$, then:
\begin{align*}
    \Pr \left(\bigcup_{t=1}^\infty\left\{\left|\sum_{t'=1}^t X_{t'}\right| > Ca_T\sqrt{t}\left[\log t + 1\right]\right\}\right) 
    \leq \frac{1}{a_T^2}.
\end{align*}
Note that since $\{X_t\}_{t=1}^\infty$ is a stationary process, so we can shift the time indices of the sum by $t_0-1$ to get:
\small
\begin{align*}
    \frac{1}{a_T^2} &\geq \Pr \left(\bigcup_{t=1}^\infty\left\{\left|\sum_{t'=1}^t X_{t'}\right| > Ca_T\sqrt{t}\left[\log t + 1\right]\right\}\right) \\
    &= \Pr \left(\bigcup_{t=t_0+1}^\infty\left\{\left|\sum_{t'=t_0}^{t-1} X_{t'}\right| > C a_T\sqrt{t-t_0}\left[\log (t-t_0) + 1\right]\right\}\right) \tag{stationarity} \\
    &\geq \Pr \left(\bigcup_{t=t_0+1}^{T}\left\{\left|\sum_{t'=t_0}^{t-1} X_{t'}\right| > C a_T\sqrt{t-t_0}\log (t-t_0)\right\}\right)\\
    &\geq \Pr \left(\bigcup_{t=t_0+1}^{T}\left\{\left|\sum_{t'=t_0}^{t-1} X_{t'}\right| > C a_T\sqrt{t-t_0}\log T\right\}\right). \tag{$T > t-t_0$}
\end{align*}
\normalsize
Similarly:
\small
\begin{align*}
    \frac{1}{a_T^2} &\geq \Pr \left(\bigcup_{t=1}^\infty\left\{\left|\sum_{t'=1}^t X_{t'}\right| > Ca_T\sqrt{t} \left[\log t + 1\right]\right\}\right) \\
    &\geq \Pr \left(\bigcup_{t=1}^{t_0-1}\left\{\left|\sum_{t'=1}^t X_{t'}\right| > Ca_T\sqrt{t}\log t\right\}\right) \\
    &= \Pr \left(\bigcup_{t=1}^{t_0-1}\left\{\left|\sum_{t'=1}^{t_0 - t} X_{t'}\right| > Ca_T\sqrt{t_0- t}\log (t_0 - t)\right\}\right)  \\
    &\geq \Pr \left(\bigcup_{t=1}^{t_0-1}\left\{\left|\sum_{t'=t}^{t_0-1} X_{t'}\right| > Ca_T\sqrt{t_0-t}\log T\right\}\right). \tag{$T > t_0-t$}
\end{align*}
\normalsize
Lastly:
\small
\begin{align*}
     \frac{1}{a_T^2} &\geq \Pr \left(\bigcup_{t=1}^\infty\left\{\left|\sum_{t'=1}^t X_{t'}\right| > Ca_T\sqrt{t} \left[\log t + 1\right]\right\}\right) \\
     &\geq \Pr \left(\bigcup_{t=1}^{T}\left\{\left|\sum_{t'=1}^{t} X_{t'}\right| > Ca_T\sqrt{t}\log t\right\}\right) \\
    &= \Pr \left(\bigcup_{t=1}^{T}\left\{\left|\sum_{t'=1}^{T - t + 1} X_{t'}\right| > Ca_T\sqrt{T - t + 1}\log (T - t+1)\right\}\right)  \\
    &\geq \Pr \left(\bigcup_{t=1}^{T}\left\{\left|\sum_{t'=t}^{T} X_{t'}\right| > Ca_T\sqrt{T-t+1}\log T\right\}\right). \tag{$T > T-t+1$}
\end{align*}
\end{proof}
\subsection{Theorem \ref{theorem:smcp} Event Bounds}
\label{app:thm1-events}

In Lemma \ref{lemma:thm1-event-bound}, we assume $\lVert \mathbf{y}_t\rVert_{\psi_2}$ is uniformly bounded over all $t, d \geq 1$. As per Remark \ref{rmk:sub-g}, this assumption holds when either: i) each entry of $\mathbf{y}_t$ is independent and $\mathcal{SG}(\sigma)$ for some $\sigma < \infty$, ii) $\mathbf{y}_t \sim \mathcal{N}_d(\boldsymbol{\mu}_t, \boldsymbol{\Lambda}^{-1})$ and $\inf_{d\geq 1}  \lambda_{\min} > 0$, where $\lambda_{\min}$ is the smallest eigenvalue of $\boldsymbol{\Lambda}$, or iii) $y_{t,j} \in \mathcal{SG}(\sigma_j)$ for each $j \in \{1,\ldots, d\}$ with $\sup_{d \geq 1} \lVert\boldsymbol{\sigma}\rVert_2 < \infty$, $\lVert\boldsymbol{\sigma}\rVert_1 = \mathcal{O}(d^{-\frac{1}{2}})$, or $\lVert\boldsymbol{\sigma}\rVert_\infty = \mathcal{O}(d^{-\frac{1}{2}})$. 

\begin{enumerate}[label=\roman*)]
    \item When each entry of $\mathbf{y}_t$ is independent and $y_{t,j} \in \mathcal{SG}(\sigma)$ for each $j \in \{1, \ldots, d\}$, then for any $\mathbf{v} \in \mathbb{B}^2_1$, by Lemma \ref{lemma:sum-sub-gaussian} we have $\langle \mathbf{v}, \mathbf{y}_t \rangle \in \mathcal{SG}(\sigma)$.
    \item For any $\mathbf{v} \in \mathbb{B}^2_1$, we have:
    \begin{align*}
        \langle \mathbf{v}, \mathbf{y}_t \rangle &= \langle \boldsymbol{\Lambda}^{-\frac{1}{2}}\mathbf{v}, \boldsymbol{\Lambda}^{\frac{1}{2}}\mathbf{y}_t \rangle \\
        &= \lVert\boldsymbol{\Lambda}^{-\frac{1}{2}}\mathbf{v}\rVert_{2} \left\langle \frac{\boldsymbol{\Lambda}^{-\frac{1}{2}}\mathbf{v}}{\lVert\boldsymbol{\Lambda}^{-\frac{1}{2}}\mathbf{v}\rVert_{2}}, \boldsymbol{\Lambda}^{\frac{1}{2}}\mathbf{y}_t \right\rangle.
    \end{align*}
    Since $\mathbf{y}_t$ is Gaussian, the coordinates of $\boldsymbol{\Lambda}^{\frac{1}{2}}\mathbf{y}_t$ are independent and belong to $\mathcal{SG}(1)$, so by Lemma \ref{lemma:sum-sub-gaussian}, the inner-product above is $\mathcal{SG}(1)$. Also, $\inf_{d\geq 1}  \lambda_{\min} > 0$ by assumption so:
    \begin{align*}
        \lVert\boldsymbol{\Lambda}^{-\frac{1}{2}}\mathbf{v}\rVert_{2} \leq \sup_{\mathbf{v} \in \mathbb{B}^2_1} \lVert\boldsymbol{\Lambda}^{-\frac{1}{2}}\mathbf{v}\rVert_{2} = \frac{1}{\sqrt{\lambda_{\min}}} < \infty.
    \end{align*}
    \item If $y_{t,j} \in \mathcal{SG}(\sigma_j)$, then for any $\mathbf{v} \in \mathbb{B}^2_1$: 
    \begin{align*}
        \langle \mathbf{v}, \mathbf{y}_t \rangle &\in \mathcal{SG}(\langle \mathbf{v}, \boldsymbol{\sigma} \rangle) \tag{Lemma \ref{lemma:sum-sub-gaussian}}\\
        &\subseteq \mathcal{SG}(\lVert\mathbf{v}\rVert_2 \lVert\boldsymbol{\sigma}\rVert_2). \tag{Cauchy-Schwarz inequality} \\
        &\subseteq \mathcal{SG}(\sqrt{d} \lVert \boldsymbol{\sigma}\rVert_1). \tag{$\lVert \mathbf{v}\rVert_2 = 1$ and $\lVert \boldsymbol{\sigma}\rVert_2 \leq \sqrt{d} \lVert \boldsymbol{\sigma}\rVert_1 $}
    \end{align*}    
    In the last line we can replace $\lVert \boldsymbol{\sigma}\rVert_1$ with $\lVert \boldsymbol{\sigma}\rVert_\infty$.
\end{enumerate}

\begin{lemma}\label{lemma:thm1-event-bound}

Let $\{\mathbf{y}_t\}_{t=1}^T$ be a sequence of independent random vectors with $\mathbf{y}_t \in \mathbb{R}^d$, $\normalfont{\Var}(\mathbf{y}_t) = \boldsymbol{\Lambda}^{-1}$ and $\sup_{t\geq 1, d\geq 1} \lVert \mathbf{y}_t\rVert_{\psi_2} < \infty$. Let $\lambda_{\max}$ and $\lambda_{\min}$ be the largest and smallest eigenvalues of $\boldsymbol{\Lambda}$ respectively and assume that $\sup_{d\geq 1} \lambda_{\min} < \infty$ and $\inf_{d\geq 1} \lambda_{\min} > 0$. Under Assumption \ref{assumption:mean}, if we define the normalized terms $\mathbf{z}_t := \boldsymbol{\Lambda}^{\frac{1}{2}}(\mathbf{y}_t - \mathbf{b}_0\mathbbm{1}\{t\geq t_0\})$, then for any $\beta > 0$, there exist universal constants $C_\beta, C_1,C_2>0$ such that we can define the set:
\begin{align}
    \mathcal{T}_{\beta} := \left\{1\leq t \leq T\::\:|t_0 - t| > \frac{C_\beta \log T}{\lVert \boldsymbol{\Lambda}^{1/2}\mathbf{b}_0\rVert_2^2}\right\} \label{eq:thm1-index}
\end{align}
and the events:
\begin{align*}
    \mathcal{E}_1 &:= \bigcap_{t \in \mathcal{T}_{\beta} } \left\{ 
    \left|\left\langle\boldsymbol{\Lambda}^{\frac{1}{2}}\mathbf{b}_0, \sum_{t'=\min\{t_0,t\}}^{\max\{t_0,t\}-1}\mathbf{z}_{t'} \right\rangle\right| < \frac{|t_0-t| \lVert\boldsymbol{\Lambda}^{\frac{1}{2}} \mathbf{b}_0\rVert_2^2}{8}\right\}, \\
    \mathcal{E}_2 &:= \left\{\left|\left\langle \boldsymbol{\Lambda}^{\frac{1}{2}}\mathbf{b}_0, \sum_{t'=t_0}^T\mathbf{z}_{t'}\right\rangle\right| < \frac{(T-t_0+1) \lVert\boldsymbol{\Lambda}^{\frac{1}{2}} \mathbf{b}_0\rVert_2^2}{8\sqrt{d}} \right\}, \\
    \mathcal{E}_3 &:=  \bigcap_{t \in \mathcal{T}_{\beta}} \left\{\left|\left\langle \sum_{t'=\min\{t_0,t\}}^{\max\{t_0,t\}-1}\frac{\mathbf{z}_{t'}}{\sqrt{|t_0-t|}}, \frac{\sum_{t'=\max\{t_0,t\}}^{T}\mathbf{z}_{t'}}{\lVert\sum_{t'=\max\{t_0,t\}}^{T}\mathbf{z}_{t'}\rVert_2}\right\rangle\right| < \sqrt{C_2\log T}\right\}, \\
    \mathcal{E}_4 &:=\bigcap_{t = t_0}^T \left\{\left|\frac{1}{T-t+1}\left\lVert\sum_{t'=t}^T \mathbf{z}_{t'}\right\rVert_2^2 -d \right| < C_2 \sqrt{d} \log T \right\} \\
    \mathcal{E}_5 &:= \bigcap_{1\leq t\leq T\::\:  t \neq t_0} \left\{\left|\frac{1}{|t_0 - t|}\left\lVert\sum_{t'=\min\{t_0,t\}}^{\max\{t_0,t\}-1} \mathbf{z}_{t'}\right\rVert_2^2 -d \right| < C_2 \sqrt{d} \log T \right\}, \\
    \mathcal{E}_6 &:= \bigcap_{t = t_0}^T \left\{\left|\frac{1}{\sqrt{T-t+1}}\left\lVert\sum_{t'=t}^T \mathbf{z}_{t'}\right\rVert_2 -\sqrt{d} \right| < \sqrt{C_2d \log T} \right\},
\end{align*}
and the joint event $\mathcal{E} := \cap_{i=1}^6 \mathcal{E}_i$, and we will have $\Pr(\mathcal{E}) \geq 1 - C_1T^{-\beta}$.

\end{lemma}

\begin{proof}

By the union bound $\Pr(\mathcal{E}^c) = 1 - P(\cup_{i=1}^6 \mathcal{E}^c_i) < 1 - \sum_{i=1}^6 \Pr(\mathcal{E}^c_i)$, so it is sufficient to show that for some $C_i > 0$, $\Pr(\mathcal{E}^c_i) \leq C_i T^{-\beta}$ for each $i$. By assumption, there exists some $K < \infty$ such that $\sup_{t\geq 1, d\geq 1} \lVert \mathbf{y}_t\rVert_{\psi_2} \leq K$, so for any $\mathbf{v} \in \mathbb{B}^2_1$ we have:
\begin{align*}
    \left\lVert\left\langle\mathbf{v}, \mathbf{z}_t \right\rangle\right\rVert_{\psi_2} &= \lVert \boldsymbol{\Lambda}^{\frac{1}{2}}\mathbf{v}\rVert_{2}\left\lVert\left\langle\frac{\boldsymbol{\Lambda}^{\frac{1}{2}}\mathbf{v}}{\lVert \boldsymbol{\Lambda}^{\frac{1}{2}}\mathbf{v} \rVert_2}, \mathbf{y}_t - \E[\mathbf{y}_t] \right\rangle\right\rVert_{\psi_2} \\
    &\leq \sup_{\mathbf{v} \in  \mathbb{B}^2_1}\lVert\boldsymbol{\Lambda}^{\frac{1}{2}}\mathbf{v}\rVert_{2} \sup_{\mathbf{v} \in  \mathbb{B}^2_1} \left\lVert\left\langle\mathbf{v}, \mathbf{y}_t - \E[\mathbf{y}_t] \right\rangle\right\rVert_{\psi_2}\\
    &= \sqrt{\lambda_{\max}} \left\lVert\mathbf{y}_t - \E[\mathbf{y}_t] \right\rVert_{\psi_2} \\
    &\lesssim \sqrt{\lambda_{\max}} \left\lVert\mathbf{y}_t \right\rVert_{\psi_2}. \tag{Lemma 2.6.8 of \cite{Vershynin18}} \\
    &\lesssim \sqrt{\lambda_{\max}} K
\end{align*}
Since $\lambda_{\max}$ is finite by assumption, and $\{\mathbf{z}_t\}_{t=1}^T$ is a independent sequence, then by Propositions 2.5.2 and 2.6.1 of \cite{Vershynin18} there exists some constant $C>0$ such that for any index set $\mathcal{T} \subseteq [T]$ with cardinality $|\mathcal{T}|$ we have: 
\begin{align*}
    \sum_{t \in \mathcal{T}} \left\langle\mathbf{v}, \mathbf{z}_t \right\rangle \in \mathcal{SG}\left(C\sqrt{\lambda_{\max}|\mathcal{T}|} K\right).
\end{align*}
So letting $\mathbf{v} = \frac{\boldsymbol{\Lambda}^{\frac{1}{2}}\mathbf{b}_0}{\lVert \boldsymbol{\Lambda}^{\frac{1}{2}}\mathbf{b}_0 \rVert_2}$, then by the Chernoff bound (\ref{eq:chernoff}): 
\begin{align*}
    \Pr\left(\left|\left\langle\boldsymbol{\Lambda}^{\frac{1}{2}}\mathbf{b}_0, \sum_{t'=\min\{t_0,t\}}^{\max\{t_0,t\}-1}\mathbf{z}_{t'} \right\rangle\right| \geq \frac{|t_0-t| \lVert\boldsymbol{\Lambda}^{\frac{1}{2}} \mathbf{b}_0\rVert_2^2}{8}\right) &\leq 2\exp\left[-\frac{|t_0-t|\lVert\boldsymbol{\Lambda}^{\frac{1}{2}} \mathbf{b}_0\rVert^2_2}{128\lambda_{\max}C^2K^2}\right].
\end{align*}
For all $t \in \mathcal{T}_{\beta}$ and $C_\beta \geq 128(1+\beta)\lambda_{\max}C^2K^2$, we then have: 
\begin{align*}
    \Pr\left(\left|\left\langle\boldsymbol{\Lambda}^{\frac{1}{2}}\mathbf{b}_0, \sum_{t'=\min\{t_0,t\}}^{\max\{t_0,t\}-1}\mathbf{z}_{t'} \right\rangle\right| \geq \frac{|t_0-t| \lVert\boldsymbol{\Lambda}^{\frac{1}{2}} \mathbf{b}_0\rVert_2^2}{8}\right) &\leq \frac{2}{T^{1+\beta}}.
\end{align*}
Therefore:
\begin{align*}
    \Pr(\mathcal{E}^C_1) &= \Pr\left(\bigcup_{t \in \mathcal{T}_{\beta}} \left\{\left|\left\langle\boldsymbol{\Lambda}^{\frac{1}{2}}\mathbf{b}_0, \sum_{t'=\min\{t_0,t\}}^{\max\{t_0,t\}-1}\mathbf{z}_{t'} \right\rangle\right| \geq \frac{|t_0-t| \lVert\boldsymbol{\Lambda}^{\frac{1}{2}} \mathbf{b}_0\rVert_2^2}{8}\right\} \right) \\
    &\leq \sum_{t \in \mathcal{T}_{\beta}}\Pr\left(\left|\left\langle\boldsymbol{\Lambda}^{\frac{1}{2}}\mathbf{b}_0, \sum_{t'=\min\{t_0,t\}}^{\max\{t_0,t\}-1}\mathbf{z}_{t'} \right\rangle\right| \geq \frac{|t_0-t| \lVert\boldsymbol{\Lambda}^{\frac{1}{2}} \mathbf{b}_0\rVert_2^2}{8}\right)\tag{union bound} \\
    &\leq \sum_{t \in \mathcal{T}_{\beta}} \frac{2}{T^{1+\beta}} \\
    &= \frac{2}{T^\beta}.\tag{$| \mathcal{T}_{\beta}| \leq T$} 
\end{align*}
We can again use the Chernoff bound (\ref{eq:chernoff}) to get: 
\begin{align*}
    \Pr\left(\left|\left\langle\boldsymbol{\Lambda}^{\frac{1}{2}}\mathbf{b}_0, \sum_{t'=t_0}^T \mathbf{z}_{t'} \right\rangle\right| \geq \frac{(T-t_0+1) \lVert\boldsymbol{\Lambda}^{\frac{1}{2}} \mathbf{b}_0\rVert_2^2}{8\sqrt{d}}\right) &\leq 2\exp\left[-\frac{(T-t_0+1)\lVert\boldsymbol{\Lambda}^{\frac{1}{2}} \mathbf{b}_0\rVert^2_2}{128\lambda_{\max}C^2K^2d}\right].
\end{align*}
We have $\lVert\boldsymbol{\Lambda}^{\frac{1}{2}} \mathbf{b}_0\rVert^2_2 \geq \lambda_{\min} \lVert \mathbf{b}_0\rVert^2_2$ and by Assumption \ref{assumption:mean}, there is some positive sequence $\{a_T\}_{T\geq1}$ such that $a_T \to \infty$ and $(T-t_0+1)\lVert \mathbf{b}_0\rVert^2_2 \geq a_Td\log T$. Therefore:
\begin{align*}
    \Pr\left(\left|\left\langle\boldsymbol{\Lambda}^{\frac{1}{2}}\mathbf{b}_0, \sum_{t'=t_0}^T \mathbf{z}_{t'} \right\rangle\right| \geq \frac{(T-t_0+1) \lVert\boldsymbol{\Lambda}^{\frac{1}{2}} \mathbf{b}_0\rVert_2^2}{8\sqrt{d}}\right) &\leq 2\exp\left[-\frac{\lambda_{\min} a_T\log T}{128\lambda_{\max}C^2K^2}\right].
\end{align*}
So for $T$ large enough that $a_T \geq\frac{128\lambda_{\max}C^2K^2 \beta}{\lambda_{\min}}$, we again have $\Pr(\mathcal{E}_2^c) \leq 2T^{-\beta}$.

Again, since $\frac{\sum_{t'=\max\{t_0,t\}}^{T}\mathbf{z}_{t'}}{\lVert\sum_{t'=\max\{t_0,t\}}^{T}\mathbf{z}_{t'}\rVert_2} \in \mathbb{B}^2_1$ for each $t$, so by the Chernoff bound (\ref{eq:chernoff}) we have:
\begin{align*}
    \Pr\left(\left|\left\langle \sum_{t'=\min\{t_0,t\}}^{\max\{t_0,t\}-1}\frac{\mathbf{z}_{t'}}{\sqrt{|t_0-t|}}, \frac{\sum_{t'=\max\{t_0,t\}}^{T}\mathbf{z}_{t'}}{\lVert\sum_{t'=\max\{t_0,t\}}^{T}\mathbf{z}_{t'}\rVert_2}\right\rangle\right| \geq \sqrt{C_2\log T}\right) \leq 2\exp\left[-\frac{C_2\log T}{2C^2K^2\lambda_{\max}}\right] 
\end{align*}
So for $C_2 \geq 2(1+\beta)C^2K^2\lambda_{\max}$, we can use the same argument as we did for $\mathcal{E}_1$ above to get $\Pr(\mathcal{E}_3^c) \leq 2T^{-\beta}$.

Next, since our observations are independent across $t$, we have $\E[\mathbf{z}'_{t'}\mathbf{z}_{t}] = \E[\mathbf{z}'_{t'}]\E[\mathbf{z}_{t}] = 0$ for $t \geq t'$, and since $\E[z_{t,j}^2] = 1$ by construction, then:
\begin{align*}
    \E\left[ \frac{1}{T-t+1}\left\lVert\sum_{t'=t}^T\mathbf{z}_{t'}\right\rVert_2^2\right] &=  \frac{1}{T-t+1}\sum_{t'=t}^T\E[\mathbf{z}'_{t'}\mathbf{z}_{t'}] = d
\end{align*}
Since:
\begin{align*}
    \left\lVert\frac{1}{\sqrt{|\mathcal{T}|}}\sum_{t'\in\mathcal{T}}\mathbf{z}_{t'}\right\rVert_{\psi_2} \lesssim K\sqrt{\lambda_{\max}}
\end{align*}
then by the Hanson–Wright inequality (see e.g. exercise 6.2.5 of \citealp{Vershynin18}), there is a universal constant $K_1>0$ so that for any $t \geq t_0$:
\begin{align*}
    \frac{1}{T-t+1}\left\lVert\sum_{t'=t}^T\mathbf{z}_{t'}\right\rVert_2^2 &\in \mathcal{SE}(K_1\sqrt{\lambda_{\max}d},K_1), 
\end{align*}
so by (\ref{eq:wainwright_prop_2.9}):
\begin{align*}
    \Pr\left(\left| \frac{1}{T-t+1}\left\lVert\sum_{t'=t}^T\mathbf{z}_{t'}\right\rVert_2^2 - d\right| \geq C_2 \sqrt{d} \log T \right) & \leq 2 \exp\left[-\frac{C_2\log T}{2K_1}\min\left\{\sqrt{d}, \frac{C_2 \log T}{K_1\lambda_{\max}}\right\}\right].
\end{align*}
Since $\sqrt{d} \geq 1$, for $C_2 > 2K_1(1+\beta)$ and $T$ large enough that $\log T \geq \frac{K_1\lambda_{\max}}{C_2}$, we have: 
\begin{align*}
    \Pr\left(\left| \frac{1}{T-t+1}\left\lVert\sum_{t'=t}^T\mathbf{z}_{t'}\right\rVert_2^2 - d\right| \geq C_2 \sqrt{d} \log T \right) \leq \frac{1}{T^{1+\beta}}.
\end{align*}
Thus:
\begin{align*}
    \Pr(\mathcal{E}^c_4) &= \Pr\left(\bigcup_{t= t_0}^T \left\{\left| \frac{1}{T-t+1}\left\lVert\sum_{t'=t}^T\mathbf{z}_{t'}\right\rVert_2^2 - d\right| \geq  C_2 \sqrt{d}\log T \right\} \right) \\
    &\leq \sum_{t = t_0}^{T} \Pr\left(\left| \frac{1}{T-t+1}\left\lVert\sum_{t'=t}^T\mathbf{z}_{t'}\right\rVert_2^2 - d\right| \geq  C_2 \sqrt{d}\log T\right) \tag{union bound} \\
    &\leq \sum_{t = t_0}^{T} \frac{2}{T^{(1+\beta)}} \\
    &\leq \frac{2}{T^{\beta}}.
\end{align*}
An identical argument gives $\Pr(\mathcal{E}^c_5) \leq 2T^{-\beta}$. Next, for any $x,\delta \geq 0$ we have $|x - 1| \geq \delta \implies |x^2 - 1| \geq \max\{\delta, \delta^2\}$ (see 3.2 in \citealp{Vershynin18}). Using this fact and that $d^{1/4}\sqrt{C_2 \log T} \leq C_2 \sqrt{d}\log T$ for large $T$, we have:
\begin{align*}
    \left\{\left|\frac{1}{\sqrt{T-t+1}}\left\lVert\sum_{t'=t}^T \mathbf{z}_{t'}\right\rVert_2 -1 \right| \geq d^{1/4}\sqrt{ C_2 \log T}\right\} &\subseteq \left\{\left|\frac{1}{T-t+1}\left\lVert\sum_{t'=t}^T \mathbf{z}_{t'}\right\rVert^2_2 -1 \right| \geq  C_2\sqrt{d} \log T \right\} 
\end{align*}
and thus $\mathcal{E}_6^c \subseteq \mathcal{E}_4^c \implies \Pr(\mathcal{E}_6^c) \leq 2T^{-\beta}$.
\end{proof}

\subsection{Proof of Theorem \ref{theorem:smcp}}
\label{app:localization-smcp}

If we take the ratio of the posterior probabilities of the change-point location from the mean-scp model, we get:
\begin{align*}
    \log \frac{\Pr(\gamma = t_0  \:|\: \mathbf{y}_{1:T} ; \omega_0,\boldsymbol{\pi}_{1:T})}{\Pr(\gamma = t  \:|\: \mathbf{y}_{1:T} \:; \omega_0,\boldsymbol{\pi}_{1:T})} &=  \log \frac{p(\mathbf{y}_{1:T} \:|\:\gamma = t_0 \:; \omega_0)\Pr(\gamma = t_0 \:; \boldsymbol{\pi}_{1:T})/ p(\mathbf{y}_{1:T}\:;\omega_0,\boldsymbol{\pi}_{1:T})}{p(\mathbf{y}_{1:T} \:|\:\gamma = t \:; \omega_0)\Pr(\gamma = t \:;\boldsymbol{\pi}_{1:T})/p(\mathbf{y}_{1:T}\:;\omega_0,\boldsymbol{\pi}_{1:T})} \tag{Bayes' rule} \\
    &= \log \frac{p(\mathbf{y}_{1:T} \:|\:\tau = t_0 \:; \omega_0)}{p(\mathbf{y}_{1:T} \:|\:\tau = t \:; \omega_0)} + \log \frac{\pi_{t_0}}{\pi_t}.
\end{align*}
By the assumptions of Theorem \ref{theorem:smcp}, for large $T$ there is some constant $C_\pi > 0$ such that for all $t \in [T]$: 
\begin{align}\label{eq:thm1-prior-bd}
    \log \frac{\pi_t}{\pi_{t_0}} \leq C_\pi \log T.
\end{align}
Therefore, if we can show that for any $\beta > 0$, there are some constants $C,C_\beta > 0$ such that:
\begin{align}
    \Pr\left(\bigcap_{t\in\mathcal{T}_\beta}\left\{ \log \frac{p(\mathbf{y}_{1:T} \:|\:\tau = t_0 \:; \omega_0)}{p(\mathbf{y}_{1:T} \:|\:\tau = t \:; \omega_0)} > C_\pi \log T\right\}\right) \geq 1 - CT^{-\beta} \label{eq:thm1-result}
\end{align}
where $\mathcal{T}_\beta$ is the set defined in (\ref{eq:thm1-index}), then with probability at least $1 - CT^{-\beta}$, we will have:
\begin{align*}
    \max_{t\::\: |t - t_0| > \frac{C_\beta\log T}{\lVert\boldsymbol{\Lambda}^{1/2}\mathbf{b}_0\rVert_2^2}} \; \Pr(\tau = t  \;|\; \mathbf{y}_{1:T} ; \omega_0,\boldsymbol{\pi}_{1:T}) < \Pr(\tau = t_0  \;|\; \mathbf{y}_{1:T} ; \omega_0,\boldsymbol{\pi}_{1:T}).
\end{align*}
Before proving (\ref{eq:thm1-result}), note that: 
\begin{align*}
    \lambda_{\min} = \min_{\mathbf{x}\in\mathbb{R}^d}\frac{\lVert\boldsymbol{\Lambda}^{\frac{1}{2}}\mathbf{x}\rVert_2^2}{\lVert\mathbf{x}\rVert_2^2} \implies \lVert\boldsymbol{\Lambda}^{\frac{1}{2}} \mathbf{b}_0\rVert^2_2 \geq \lambda_{\min} \lVert \mathbf{b}_0\rVert^2_2 \implies \frac{C_\beta \log T}{\lVert \boldsymbol{\Lambda}^{1/2}\mathbf{b}_0\rVert_2^2} \leq 
    \frac{C_\beta \log T}{\lambda_{\min}\lVert\mathbf{b}_0\rVert_2^2},
\end{align*}
so the localization error vanishes when $\lambda_{\min}\lVert \mathbf{b}_0\rVert_2^2 \gg \log T$. Since $|t - t_0| > 1$ for all $t\neq t_0$, in this case we have:
\begin{align*}
    \lim_{T\to\infty}\Pr\left(\bigcap_{t \neq t_0}\left\{ \log \frac{p(\mathbf{y}_{1:T} \:|\:\tau = t_0 \:; \omega_0)}{p(\mathbf{y}_{1:T} \:|\:\tau = t \:; \omega_0)} > C_\pi \log T\right\}\right) = 1
\end{align*}
i.e. $\lim_{T\to\infty}\Pr(\hat{\tau}_{\text{MAP}} = t_0) = 1$. Note that when $\lambda_{\min}$ is fixed and $\lVert\mathbf{b}_0\rVert_\infty$ is bounded over all $d$, then it is only possible for $\lVert \mathbf{b}_0\rVert_2^2$ to grow like $\log T$ when $d \gtrsim \log T$, i.e. when $\lVert\mathbf{b}_0\rVert_\infty = \mathcal{O}(1)$, then $d \gg \log T$ is a necessary condition for exact recovery of $t_0$. When $d \geq a_T\log T$ for some $a_T\to\infty$, then exact recovery is possible when at least $b_T \log T$ elements of $\mathbf{b}_0$ are bounded below in absolute value by $c_T$ for some $b_T = \mathcal{O}(a_T)$. This holds even when $c_T \to 0$ so long as $b_T c_T \to \infty$. 

\begin{proof}
We begin by directly calculating the log-marginal evidence of $\mathbf{y}_{1:T}$: 
\scriptsize
\begin{align*}
    \log p(\mathbf{y}_{1:T} \:|\:\tau = t ; \omega_0) &= \log \int_{\mathbb{R}^d} \Pr(\mathbf{y}_{1:T} \:|\:\mathbf{b},\tau = t) \:\; d\Pr(\mathbf{b};\omega_0) \\
    &= \log \int_{\mathbb{R}^d} (2\pi)^{-\frac{dT}{2}}|\boldsymbol{\Lambda}|^\frac{T}{2} \exp\left[-\frac{\sum_{t'=1}^{t-1} \lVert\boldsymbol{\Lambda}^{\frac{1}{2}}\mathbf{y}_{t'}\rVert_2^2 + \sum_{t'=t}^{T} \lVert\boldsymbol{\Lambda}^{\frac{1}{2}}(\mathbf{y}_{t'} - \mathbf{b})\rVert_2^2}{2}\right] \\
    &\quad\quad\quad\quad\quad\times \left(\frac{\omega_0}{2\pi}\right)^{\frac{d}{2}}\exp\left[-\frac{\omega_0\lVert\mathbf{b}\rVert_2^2}{2}\right] \; d\mathbf{b} \\
    &= \frac{T}{2}\left[\log|\boldsymbol{\Lambda}| - d\log(2\pi)\right] + \frac{d}{2}\log \left(\frac{\omega_0}{2\pi}\right)\\
    &\quad + \log \int_{\mathbb{R}^d} \exp\left[-\frac{\sum_{t'=1}^{t-1} \lVert\boldsymbol{\Lambda}^{\frac{1}{2}}\mathbf{y}_{t'}\rVert_2^2 + \sum_{t'=t}^{T} \lVert\boldsymbol{\Lambda}^{\frac{1}{2}}\mathbf{b}\rVert_2^2 +\omega_0 \lVert\mathbf{b}\rVert_2^2 + \sum_{t'=t}^{T} \lVert\boldsymbol{\Lambda}^{\frac{1}{2}}\mathbf{y}_{t'}\rVert_2^2 + -2\langle \mathbf{b}, \boldsymbol{\Lambda}^{\frac{1}{2}} \sum_{t'=t}^T\mathbf{y}_{t'}\rangle}{2}\right] d\mathbf{b} \\
    &=  \frac{T}{2}\left[\log|\boldsymbol{\Lambda}| - d\log(2\pi)\right] + \frac{d}{2}\log \left(\omega_0\right) - \frac{\sum_{t'=1}^T \lVert\boldsymbol{\Lambda}^{\frac{1}{2}}\mathbf{y}_{t'}\rVert_2^2}{2} \\
    &\quad + \log \int_{\mathbb{R}^d} (2\pi)^{-\frac{d}{2}}\exp\left[-\frac{\lVert[(T-t+1)\boldsymbol{\Lambda} + \omega_0\mathbf{I}_d]^{\frac{1}{2}}\mathbf{b}\rVert_2^2-2\langle \mathbf{b}, \boldsymbol{\Lambda} \sum_{t'=t}^T\mathbf{y}_{t'}\rangle}{2}\right] \; d\mathbf{b}\\ 
    &= \log \int_{\mathbb{R}^d} (2\pi)^{-\frac{d}{2}}\exp\left[-\frac{\lVert[(T-t+1)\boldsymbol{\Lambda} + \omega_0\mathbf{I}_d]^{\frac{1}{2}}\mathbf{b}\rVert_2^2-2\langle \mathbf{b}, \boldsymbol{\Lambda} \sum_{t'=t}^T\mathbf{y}_{t'}\rangle}{2}\right] \; d\mathbf{b} + C \tag{$C$ constant independent of $\tau$}
\end{align*}
\normalsize
If we define $\overline{\boldsymbol{\Omega}}_t := (T-t+1)\boldsymbol{\Lambda} + \omega_0\mathbf{I}_d$ and $\overline{\mathbf{b}}_t := \overline{\boldsymbol{\Omega}}_t^{-1} \boldsymbol{\Lambda} \sum_{t'=t}^T \mathbf{y}_{t'}$, then we have:
\begin{align*}
    \log p(\mathbf{y}_{1:T} \:|\:\tau = t ; \omega_0)
    &=  -\frac{1}{2}\log |\overline{\boldsymbol{\Omega}}_t| + \log \int_{\mathbb{R}^d} \big|\overline{\boldsymbol{\Omega}}^{\frac{1}{2}}_t\big|(2\pi)^{-\frac{d}{2}}\exp\left[-\frac{\lVert\overline{\boldsymbol{\Omega}}^{\frac{1}{2}}_t(\mathbf{b} - \overline{\mathbf{b}}_t)\rVert_2^2}{2} + \frac{\lVert \overline{\boldsymbol{\Omega}}^{\frac{1}{2}}_t \overline{\mathbf{b}}_t\rVert_2^2}{2}\right] \; d\mathbf{b} + C\\
    &= -\frac{1}{2}\log |\overline{\boldsymbol{\Omega}}_t| + \frac{1}{2}\lVert \overline{\boldsymbol{\Omega}}^{\frac{1}{2}}_t \overline{\mathbf{b}}_t\rVert_2^2 + C.
\end{align*}
Letting $\omega_0 \to 0$ we have $\overline{\boldsymbol{\Omega}}_t \to (T-t + 1)\boldsymbol{\Lambda}$, therefore: 
\begin{align*}
    \lim_{\omega_0\to 0} \log p(\mathbf{y}_{1:T} \:|\: \tau = t; \omega_0) &= -\frac{d}{2}\log(T - t + 1) - \frac{1}{2}\log |\boldsymbol{\Lambda}| + \frac{\lVert \boldsymbol{\Lambda}^{\frac{1}{2}}\sum_{t'=t}^T\mathbf{y}_{t'}\rVert_2^2}{2(T-t+1)}+ C \\
    &:= \log \alpha_t.
\end{align*}
Since the determinant and $\lVert \cdot \rVert_2$ are continuous, for small $\omega_0$, we can use $\log \alpha_t$ to approximate the log-evidence in (\ref{eq:thm1-result}).\footnote{Formally, for any $\epsilon > 0$, we can show that there is a choice of $\omega_0$ small enough so that for each $t$, $|\frac{p(\mathbf{y}_{1:T} \:|\:\tau = t ; \omega_0)}{\Delta_{t}}-1| < \epsilon$ with high probability as $T \to \infty$. So if we can show that for each $t$ we consider, $\log \frac{\alpha_{t_0}}{\alpha_t} > C_\pi \log T$ with high probability, we have $\log \frac{p(\mathbf{y}_{1:T} \:|\:\tau = t_0 ; \omega_0)}{p(\mathbf{y}_{1:T} \:|\:\tau = t ; \omega_0)} = \log\frac{p(\mathbf{y}_{1:T} \:|\:\tau = t_0 ; \omega_0)}{\alpha_{t_0}} + \log \frac{\alpha_{t_0}}{\alpha_t} + \log \frac{\alpha_t}{p(\mathbf{y}_{1:T} \:|\:\tau = t ; \omega_0)} > C_\pi\log T -2 \log(1-\epsilon)$ with high probability, which is the desired result. \label{fn:approx}} Let $C_\beta$ be large enough so that the statement of Lemma \ref{lemma:thm1-event-bound} holds. Then we will be on the event $\mathcal{E}$ defined in Lemma \ref{lemma:thm1-event-bound} with probability approaching one as $T\to\infty$. We can show that the event in (\ref{eq:thm1-result}) holds on $\mathcal{E}$. To do so, we define the standardized terms $\mathbf{z}_t := \boldsymbol{\Lambda}^{\frac{1}{2}}(\mathbf{y}_t - \mathbf{b}_0\mathbbm{1}\{t\geq t_0\})$ and consider the following two cases:

\subsubsection*{Case 1: $t > t_0$ and $t\in\mathcal{T}_\beta$.}

For $t \geq t_0$ we have:
\begin{align*}
    \left\lVert\boldsymbol{\Lambda}^{\frac{1}{2}} \sum_{t'=t}^T\mathbf{y}_{t'}\right\rVert_2^2 &= \left\lVert\boldsymbol{\Lambda}^{\frac{1}{2}} \sum_{t'=t}^T(\mathbf{y}_{t'} - \mathbf{b}_0 + \mathbf{b}_0)\right\rVert_2^2 \\
    &= \left\lVert \sum_{t'=t}^T \mathbf{z}_{t'}\right\rVert_2^2 + 2(T-t+1) \left\langle\boldsymbol{\Lambda}^{\frac{1}{2}}\mathbf{b}_0, \sum_{t'=t}^T\mathbf{z}_{t'} \right\rangle + (T-t+1)^2\left\lVert\boldsymbol{\Lambda}^{\frac{1}{2}}\mathbf{b}_0\right\rVert_2^2 
\end{align*}
and thus:
\begin{align*}
    \log \frac{\alpha_{t_0}}{\alpha_t} &> -\frac{d}{2}\log\left(\frac{T-t_0+1}{T-t+1}\right) +\frac{(t-t_0)\left\lVert\boldsymbol{\Lambda}^{\frac{1}{2}}\mathbf{b}_0\right\rVert_2^2}{2} + \frac{\lVert\sum_{t'=t_0}^T\mathbf{z}_{t'}\rVert_2^2}{2(T-t_0+1)} - \frac{\lVert\sum_{t'=t}^T\mathbf{z}_{t'}\rVert_2^2}{2(T-t+1)} - \left| \left\langle\boldsymbol{\Lambda}^{\frac{1}{2}}\mathbf{b}_0, \sum_{t'=t_0}^{t-1}\mathbf{z}_{t'} \right\rangle\right|. 
\end{align*}
Let $C_\beta > 0$ be large enough so that we are on the event $\mathcal{E}$ defined in Lemma \ref{lemma:thm1-event-bound} with high probability, then for $t \in \mathcal{T}_\beta$ on $\mathcal{E}_1\subseteq\mathcal{E}$ we have:
\begin{align*}
    \left|\left\langle\boldsymbol{\Lambda}^{\frac{1}{2}}\mathbf{b}_0, \sum_{t'=t_0}^{t-1}\mathbf{z}_{t'} \right\rangle\right| < \frac{(t-t_0) \lVert\boldsymbol{\Lambda}^{\frac{1}{2}} \mathbf{b}_0\rVert_2^2}{4}
\end{align*}
and thus:
\begin{align}
    \log \frac{\alpha_{t_0}}{\alpha_t} &> -\frac{d}{2}\log\left(\frac{T-t_0+1}{T-t+1}\right) +\frac{(t-t_0)\left\lVert\boldsymbol{\Lambda}^{\frac{1}{2}}\mathbf{b}_0\right\rVert_2^2}{4} + \frac{\lVert\sum_{t'=t_0}^T\mathbf{z}_{t'}\rVert_2^2}{2(T-t_0+1)} - \frac{\lVert\sum_{t'=t}^T\mathbf{z}_{t'}\rVert_2^2}{2(T-t+1)}.\label{eq:thm1-bd1}
\end{align}
For the remaining random terms, we have:
\small
\begin{align}
    \frac{\lVert\sum_{t'=t_0}^T\mathbf{z}_{t'}\rVert_2^2}{T-t_0+1} - \frac{\lVert\sum_{t'=t}^T\mathbf{z}_{t'}\rVert_2^2}{T-t+1} &=
    \frac{\lVert\sum_{t'=t_0}^{t-1}\mathbf{z}_{t'}\rVert_2^2 + 2\langle \sum_{t'=t_0}^{t-1}\mathbf{z}_{t'}, \sum_{t'=t}^{T}\mathbf{z}_{t'}\rangle + \lVert\sum_{t'=t}^{T}\mathbf{z}_{t'}\rVert_2^2}{T-t_0+1}  - \frac{\lVert\sum_{t'=t}^T\mathbf{z}_{t'}\rVert_2^2}{T-t+1} \notag \\
    &= \left(\frac{t-t_0}{T-t_0+1}\right)\left[\frac{\lVert\sum_{t'=t_0}^{t-1}\mathbf{z}_{t'}\rVert_2^2}{t-t_0} - d\right] - \left(\frac{t-t_0}{T-t_0+1}\right)\left[\frac{\lVert\sum_{t'=t}^T\mathbf{z}_{t'}\rVert_2^2}{T-t+1} -d \right] \notag \\
    &\quad +\frac{2\sqrt{(t-t_0)(T-t+1)}\left(\lVert\sum_{t'=t}^{T}\frac{\mathbf{z}_{t'}}{\sqrt{T-t+1}}\rVert_2 -\sqrt{d} + \sqrt{d}\right)\left\langle \sum_{t'=t_0}^{t-1}\frac{\mathbf{z}_{t'}}{\sqrt{t-t_0}}, \frac{\sum_{t'=t}^{T}\mathbf{z}_{t'}}{\lVert\sum_{t'=t}^{T}\mathbf{z}_{t'}\rVert_2}\right\rangle}{T-t_0+1} \notag \\
    &> -\left(\frac{t-t_0}{T-t_0+1}\right)\left|\frac{\lVert\sum_{t'=t_0}^{t-1}\mathbf{z}_{t'}\rVert_2^2}{t-t_0} - d\right| - \left(\frac{t-t_0}{T-t_0+1}\right)\left|\frac{\lVert\sum_{t'=t}^T\mathbf{z}_{t'}\rVert_2^2}{T-t+1} -d \right| \notag \\
    &\quad -\frac{2\sqrt{(t-t_0)(T-t+1)}\left|\lVert\sum_{t'=t}^{T}\frac{\mathbf{z}_{t'}}{\sqrt{T-t+1}}\rVert_2 -\sqrt{d}\right|\left|\left\langle \sum_{t'=t_0}^{t-1}\frac{\mathbf{z}_{t'}}{\sqrt{t-t_0}}, \frac{\sum_{t'=t}^{T}\mathbf{z}_{t'}}{\lVert\sum_{t'=t}^{T}\mathbf{z}_{t'}\rVert_2}\right\rangle\right|}{T-t_0+1} \notag \\
    &\quad-\frac{2\sqrt{(t-t_0)(T-t+1)d}\left|\left\langle \sum_{t'=t_0}^{t-1}\frac{\mathbf{z}_{t'}}{\sqrt{t-t_0}}, \frac{\sum_{t'=t}^{T}\mathbf{z}_{t'}}{\lVert\sum_{t'=t}^{T}\mathbf{z}_{t'}\rVert_2}\right\rangle\right|}{T-t_0+1}. \label{eq:thm1-bd2}
\end{align}
\normalsize
There exists some $C_1 \geq 1$ so that on the event $\mathcal{E}$ defined in Lemma \ref{lemma:thm1-event-bound} we have:
\begin{align*}
    \left|\left\langle \sum_{t'=t_0}^{t-1}\frac{\mathbf{z}_{t'}}{\sqrt{t-t_0}}, \frac{\sum_{t'=t}^{T}\mathbf{z}_{t'}}{\lVert\sum_{t'=t}^{T}\mathbf{z}_{t'}\rVert_2}\right\rangle\right| &< \sqrt{C_1\log T} \tag{on $\mathcal{E}_3\subseteq\mathcal{E}$} \\
    \left|\frac{\left\lVert\sum_{t'=t}^T\mathbf{z}_{t'}\right\rVert_2^2}{T-t+1} - d\right| &< C_1\sqrt{d}\log T \tag{on $\mathcal{E}_4\subseteq\mathcal{E}$} \\
    \left|\frac{\left\lVert\sum_{t'=t_0}^{t-1}\mathbf{z}_{t'}\right\rVert_2^2}{t-t_0} - d\right| &< C_1\sqrt{d}\log T \tag{on $\mathcal{E}_5\subseteq\mathcal{E}$} \\
    \left|\frac{\left\lVert\sum_{t'=t}^{T}\mathbf{z}_{t'}\right\rVert_2}{\sqrt{T-t+1}} -\sqrt{d}\right| &< \sqrt{C_1 d\log T}. \tag{on $\mathcal{E}_6\subseteq\mathcal{E}$}
\end{align*}
Returning to (\ref{eq:thm1-bd2}), we can use the bounds above to write:
\begin{align*}
    \frac{\lVert\sum_{t'=t_0}^T\mathbf{z}_{t'}\rVert_2^2}{2(T-t_0+1)} - \frac{\lVert\sum_{t'=t}^T\mathbf{z}_{t'}\rVert_2^2}{2(T-t+1)} &> -\frac{C_1(t-t_0)\sqrt{d}\log T}{T-t_0+1} - \frac{2C_1\sqrt{d(t-t_0)(T-t+1)}\log T}{T-t_0+1}.
\end{align*} 
By Assumption \ref{assumption:mean}, there is some sequence $\{a_T\}_{t \geq 1}$ such that: 
\begin{align}
    a_T\to\infty \text{ and } \lVert \mathbf{b}_0\rVert_2^2(T-t_0+1)\geq a_Td\log T, \label{eq:mean-assumption-bd}
\end{align}
and thus:
\begin{align*}
    -\frac{C_1(t-t_0)\sqrt{d}\log T}{T-t_0+1} &= -\frac{C_1(t-t_0)\lVert\boldsymbol{\Lambda}^{\frac{1}{2}}\mathbf{b}_0\rVert_2^2\sqrt{d}\log T}{(T-t_0+1)\lVert\boldsymbol{\Lambda}^{\frac{1}{2}}\mathbf{b}_0\rVert_2^2} \\
    &\geq -\frac{C_1(t-t_0)\lVert\boldsymbol{\Lambda}^{\frac{1}{2}}\mathbf{b}_0\rVert_2^2\sqrt{d}\log T}{\lambda_{\min}(T-t_0+1)\lVert\mathbf{b}_0\rVert_2^2} \tag{$\lVert\boldsymbol{\Lambda}^{\frac{1}{2}}\mathbf{b}_0\rVert_2^2\geq\lambda_{\min} \lVert\mathbf{b}_0\rVert_2^2$} \\
    &\geq -\frac{C_1(t-t_0)\lVert\boldsymbol{\Lambda}^{\frac{1}{2}}\mathbf{b}_0\rVert_2^2}{\lambda_{\min}a_T \sqrt{d}}. \tag{by \ref{eq:mean-assumption-bd}}
\end{align*}
For $T$ large enough so that $a_T\sqrt{d} \geq \frac{8C_1}{\lambda_{\min}}$, we can rewrite (\ref{eq:thm1-bd1}) as: 
\begin{align}
    \log \frac{\alpha_{t_0}}{\alpha_t} &> -\frac{d}{2}\log\left(\frac{T-t_0+1}{T-t+1}\right) +\frac{(t-t_0)\left\lVert\boldsymbol{\Lambda}^{\frac{1}{2}}\mathbf{b}_0\right\rVert_2^2}{8} - \frac{2\sqrt{d(t-t_0)(T-t+1)}C_1\log T}{T-t_0+1}.\label{eq:thm1-bd3}
\end{align}
Focusing on the last term, we can write:
\begin{align*}
    \frac{2\sqrt{d(t-t_0)(T-t+1)}C_1\log T}{T-t_0+1} &= \frac{2(t-t_0)C_1\sqrt{d}\log T}{T-t_0+1}\sqrt{\frac{T-t+1}{t-t_0}}.
\end{align*}
Since $\sqrt{(T-t+1)/(t-t_0)}$ is decreasing as a function of $t$ and $ \inf \{t \geq t_0 \:|\:t\in\mathcal{T}_\beta\} = t_0 + \frac{C_\beta\log T}{\lVert\boldsymbol{\Lambda}^{1/2}\mathbf{b}_0\rVert_2^2}$, we have:
\begin{align*}
    \frac{2\sqrt{d(t-t_0)(T-t+1)}C_1\log T}{T-t_0+1} &\leq \frac{2C_1(t-t_0)\lVert\boldsymbol{\Lambda}^{1/2}\mathbf{b}_0\rVert_2\sqrt{d\log T}}{T-t_0+1}\sqrt{\frac{T-t_0 +1 - \frac{C_\beta\log T}{\lVert\boldsymbol{\Lambda}^{1/2}\mathbf{b}_0\rVert_2^2}}{C_\beta}} \\
    &\leq \frac{2C_1(t-t_0)\lVert\boldsymbol{\Lambda}^{1/2}\mathbf{b}_0\rVert_2\sqrt{d\log T}}{\sqrt{C_\beta(T-t_0+1)}} \tag{$\frac{C_\beta\log T}{\lVert\boldsymbol{\Lambda}^{1/2}\mathbf{b}_0\rVert_2^2} > 0$}\\
    &\leq \frac{2C_1(t-t_0)\lVert\boldsymbol{\Lambda}^{1/2}\mathbf{b}_0\rVert^2_2\sqrt{d\log T}}{\sqrt{C_\beta\lambda_{\min}\lVert\mathbf{b}_0\rVert^2_2(T-t_0+1)}} \\
    &\leq \frac{2C_1(t-t_0)\lVert\boldsymbol{\Lambda}^{1/2}\mathbf{b}_0\rVert^2_2}{\sqrt{C_\beta\lambda_{\min}a_T}}. \tag{by \ref{eq:mean-assumption-bd}}
\end{align*}
So for $T$ large enough that $a_T > \frac{32^2C_1^2}{C_\beta\lambda_{\min}}$, we can bound (\ref{eq:thm1-bd3}) with:
\begin{align}
    \log \frac{\alpha_{t_0}}{\alpha_t} &> -\frac{d}{2}\log\left(\frac{T-t_0+1}{T-t+1}\right) +\frac{(t-t_0)\left\lVert\boldsymbol{\Lambda}^{\frac{1}{2}}\mathbf{b}_0\right\rVert_2^2}{16}. \label{eq:thm1-bd4}
\end{align}
Note that the RHS expression above is concave as a function of $t$, therefore it is minimized at one of the extremal values $t =T$ or $t = t_0 + \frac{C_\beta\log T}{\lVert\boldsymbol{\Lambda}^{1/2}\mathbf{b}_0\rVert_2^2}$. When $t = T$ and $T$ is large enough that $T-t_0+1 \geq 2$ and $a_T \geq 16C_\pi + 8$, the lower bound in (\ref{eq:thm1-bd4}) becomes:
\begin{align*}
    -\frac{d}{2}\log\left(T-t_0+1\right) +\frac{(T-t_0)\left\lVert\boldsymbol{\Lambda}^{\frac{1}{2}}\mathbf{b}_0\right\rVert_2^2}{16} &\geq -\frac{d}{2}\log\left(T-t_0+1\right) +\frac{(T-t_0+1)\left\lVert\boldsymbol{\Lambda}^{\frac{1}{2}}\mathbf{b}_0\right\rVert_2^2}{16} \\
    &\geq -\frac{d}{2}\log\left(T-t_0+1\right) +\frac{a_Td\log T}{16} \tag{by (\ref{eq:mean-assumption-bd})} \\
    &\geq \frac{(a_T-8)d\log T}{16}
    \tag{$T \geq T-t_0+1$} \\
    &\geq C_\pi \log T. \tag{$a_T \geq 16C_\pi + 8$}
\end{align*}
On the other hand, when $t = t_0 + \frac{C_\beta\log T}{\lVert\boldsymbol{\Lambda}^{1/2}\mathbf{b}_0\rVert_2^2}$, the lower bound in (\ref{eq:thm1-bd4}) becomes: 
\begin{align*}
    \frac{d}{2}\log\left(1 - \frac{C_\beta\log T}{\lVert\boldsymbol{\Lambda}^{1/2}\mathbf{b}_0\rVert_2^2(T-t_0+1)}\right) +\frac{C_\beta\log T}{16} &\geq \frac{d}{2}\log\left(1 - \frac{C_\beta}{a_Td}\right) +\frac{C_\beta\log T}{16} \tag{by (\ref{eq:mean-assumption-bd})}.
\end{align*}
Noting that $\frac{C_\beta}{a_Td} \to 0$, we can take a first order Taylor expansion of $\log(1-x)$ around zero to get:
\begin{align*}
    \frac{d}{2}\log\left(1 - \frac{C_\beta}{a_Td}\right) +\frac{C_\beta\log T}{16} &= -\frac{C_\beta}{2a_T} + \mathcal{O}\left(\frac{1}{a^2_Td}\right) +\frac{C_\beta\log T}{16} 
\end{align*}
The first two terms on the RHS above vanish as $T\to\infty$, so for $C_\beta > 16 C_\pi$, we can return to (\ref{eq:thm1-bd4}) to get $\log \frac{\alpha_{t_0}}{\alpha_t} \geq C_\pi\log T$, i.e. the event in (\ref{eq:thm1-result}) holds for this case.

\subsubsection*{Case 2: $t < t_0$ and $t\in\mathcal{T}_\beta$.}

As in the previous case, we have:
\begin{align*}
    \frac{\lVert\boldsymbol{\Lambda}^{\frac{1}{2}} \sum_{t'=t_0}^T\mathbf{y}_{t'}\rVert_2^2}{T-t_0+1} &= \frac{\lVert\sum_{t'=t_0}^T\mathbf{z}_{t'}\rVert_2^2}{T-t_0+1} - 2\left\langle \boldsymbol{\Lambda}^{\frac{1}{2}}\mathbf{b}_0, \sum_{t'=t_0}^T\mathbf{z}_{t'}\right\rangle + (T-t_0+1)\lVert\boldsymbol{\Lambda}^{\frac{1}{2}}\mathbf{b}_0\rVert_2^2
\end{align*}
and for $t < t_0$, we can write:
\begin{align*}
    \frac{\lVert\boldsymbol{\Lambda}^{\frac{1}{2}} \sum_{t'=t}^T\mathbf{y}_{t'}\rVert_2^2}{T-t+1} &= \frac{\lVert \sum_{t'=t}^T\mathbf{z}_{t'} + (T-t_0+1)\boldsymbol{\Lambda}^{\frac{1}{2}}\mathbf{b}_0\rVert_2^2}{T-t+1} \\
    &= \frac{\lVert \sum_{t'=t}^T\mathbf{z}_{t'}\rVert_2^2}{T-t+1} -2\left(\frac{T-t_0+1}{T-t+1}\right) \left\langle \boldsymbol{\Lambda}^{\frac{1}{2}}\mathbf{b}_0, \sum_{t'=t}^T\mathbf{z}_{t'}\right\rangle  + \frac{(T-t_0+1)^2}{T-t+1}\lVert\boldsymbol{\Lambda}^{\frac{1}{2}}\mathbf{b}_0\rVert_2^2.
\end{align*}
After rearranging terms we get:
\begin{align*}
    \frac{\lVert\boldsymbol{\Lambda}^{\frac{1}{2}} \sum_{t'=t_0}^T\mathbf{y}_{t'}\rVert_2^2}{T-t_0+1} - \frac{\lVert\boldsymbol{\Lambda}^{\frac{1}{2}} \sum_{t'=t}^T\mathbf{y}_{t'}\rVert_2^2}{T-t+1} &=\frac{\lVert \sum_{t'=t_0}^T\mathbf{z}_{t'}\rVert_2^2}{T-t_0+1} - \frac{\lVert\sum_{t'=t}^T\mathbf{z}_{t'}\rVert_2^2}{T-t+1} + \frac{(T-t_0+1)(t_0 - t) }{T-t+1}\lVert\boldsymbol{\Lambda}^{\frac{1}{2}} \mathbf{b}_0\rVert^2_2 \\
    &\quad - 2\left(\frac{t_0 - t}{T-t+1}\right)\left\langle \boldsymbol{\Lambda}^{\frac{1}{2}}\mathbf{b}_0, \sum_{t'=t_0}^T\mathbf{z}_{t'}\right\rangle + 2\left(\frac{T-t_0+1}{T-t+1}\right)\left\langle \boldsymbol{\Lambda}^{\frac{1}{2}}\mathbf{b}_0, \sum_{t'=t}^{t_0-1}\mathbf{z}_{t'}\right\rangle.
\end{align*}
Focusing on the second line above, again let $C_\beta$ be large enough so that Lemma \ref{lemma:thm1-event-bound} holds and we are on the event $\mathcal{E}$ with high probability, then for $t \in \mathcal{T}_\beta$ on $\mathcal{E}$ we have:
\begin{align*}
    \left|\left\langle\boldsymbol{\Lambda}^{\frac{1}{2}}\mathbf{b}_0, \sum_{t'=t}^{t_0-1}\mathbf{z}_{t'} \right\rangle\right| < \frac{(t_0-t) \lVert\boldsymbol{\Lambda}^{\frac{1}{2}} \mathbf{b}_0\rVert_2^2}{8} \tag{on $\mathcal{E}_1 \subseteq \mathcal{E}$}
\end{align*}
and:
\begin{align*}
    \left|\left\langle\boldsymbol{\Lambda}^{\frac{1}{2}}\mathbf{b}_0, \sum_{t'=t_0}^T \mathbf{z}_{t'} \right\rangle\right| < \frac{(T-t_0+1) \lVert\boldsymbol{\Lambda}^{\frac{1}{2}} \mathbf{b}_0\rVert_2^2}{8} \tag{on $\mathcal{E}_2 \subseteq \mathcal{E}$}
\end{align*}
so:
\begin{align*}
    - 2\left(\frac{t_0 - t}{T-t+1}\right)\left|\left\langle \boldsymbol{\Lambda}^{\frac{1}{2}}\mathbf{b}_0, \sum_{t'=t_0}^T\mathbf{z}_{t'}\right\rangle\right| - 2\left(\frac{T-t_0+1}{T-t+1}\right)\left|\left\langle \boldsymbol{\Lambda}^{\frac{1}{2}}\mathbf{b}_0, \sum_{t'=t}^{t_0-1}\mathbf{z}_{t'}\right\rangle\right| 
    &> -\frac{(t_0-t)(T-t_0+1)\lVert\boldsymbol{\Lambda}^{\frac{1}{2}} \mathbf{b}_0\rVert_2^2}{2(T-t+1)}.
\end{align*}
So on $\mathcal{E}$ we have:
\begin{align*}
    \frac{\lVert\boldsymbol{\Lambda}^{\frac{1}{2}} \sum_{t'=t_0}^T\mathbf{y}_{t'}\rVert_2^2}{T-t_0+1} - \frac{\lVert\boldsymbol{\Lambda}^{\frac{1}{2}} \sum_{t'=t}^T\mathbf{y}_{t'}\rVert_2^2}{T-t+1} &> \frac{\lVert \sum_{t'=t_0}^T\mathbf{z}_{t'}\rVert_2^2}{T-t_0+1} - \frac{\lVert\sum_{t'=t}^T\mathbf{z}_{t'}\rVert_2^2}{T-t+1} + \frac{(T-t_0+1)(t_0 - t) }{2(T-t+1)}\lVert\boldsymbol{\Lambda}^{\frac{1}{2}} \mathbf{b}_0\rVert^2_2 .
\end{align*}
For the remaining random terms, we have:
\small
\begin{align}
    \frac{\lVert\sum_{t'=t_0}^T\mathbf{z}_{t'}\rVert_2^2}{T-t_0+1} - \frac{\lVert\sum_{t'=t}^T\mathbf{z}_{t'}\rVert_2^2}{T-t+1} &= \frac{\lVert\sum_{t'=t_0}^T\mathbf{z}_{t'}\rVert_2^2}{T-t_0+1} - 
    \frac{\lVert\sum_{t'=t}^{t_0-1}\mathbf{z}_{t'}\rVert_2^2 + 2\langle \sum_{t'=t}^{t_0-1}\mathbf{z}_{t'}, \sum_{t'=t_0}^{T}\mathbf{z}_{t'}\rangle + \lVert\sum_{t'=t_0}^{T}\mathbf{z}_{t'}\rVert_2^2}{T-t+1} \notag \\
    &= \left(\frac{t_0-t}{T-t+1}\right)\left[\frac{\lVert\sum_{t'=t_0}^T\mathbf{z}_{t'}\rVert_2^2}{T-t_0 + 1} - d\right] - \left(\frac{t_0-t}{T-t+1}\right)\left[\frac{\lVert\sum_{t'=t}^{t_0-1}\mathbf{z}_{t'}\rVert_2^2}{t_0-t} -d \right] \notag \\
    &\quad +\frac{2\sqrt{(t_0-t)(T-t_0+1)}\left(\lVert\sum_{t'=t_0}^{T}\frac{\mathbf{z}_{t'}}{\sqrt{T-t_0+1}}\rVert_2 -\sqrt{d} + \sqrt{d}\right)\left\langle \sum_{t'=t}^{t_0-1}\frac{\mathbf{z}_{t'}}{\sqrt{t_0-t}}, \frac{\sum_{t'=t_0}^{T}\mathbf{z}_{t'}}{\lVert\sum_{t'=t_0}^{T}\mathbf{z}_{t'}\rVert_2}\right\rangle}{T-t+1} \notag \\
    &> -\left(\frac{t_0-t}{T-t+1}\right)\left|\frac{\lVert\sum_{t'=t_0}^T\mathbf{z}_{t'}\rVert_2^2}{T-t_0 + 1} - d\right| - \left(\frac{t_0-t}{T-t+1}\right)\left|\frac{\lVert\sum_{t'=t}^{t_0-1}\mathbf{z}_{t'}\rVert_2^2}{t_0-t} -d \right| \notag \\
    &\quad -\frac{2\sqrt{(t_0-t)(T-t_0+1)}\left|\lVert\sum_{t'=t_0}^{T}\frac{\mathbf{z}_{t'}}{\sqrt{T-t_0+1}}\rVert_2 -\sqrt{d}\right|\left|\left\langle \sum_{t'=t}^{t_0-1}\frac{\mathbf{z}_{t'}}{\sqrt{t_0-t}}, \frac{\sum_{t'=t_0}^{T}\mathbf{z}_{t'}}{\lVert\sum_{t'=t_0}^{T}\mathbf{z}_{t'}\rVert_2}\right\rangle\right|}{T-t+1} \notag \\
    &\quad-\frac{2\sqrt{(t_0-t)(T-t_0+1)d}\left|\left\langle \sum_{t'=t}^{t_0-1}\frac{\mathbf{z}_{t'}}{\sqrt{t_0-t}}, \frac{\sum_{t'=t_0}^T\mathbf{z}_{t'}}{\lVert\sum_{t'=t_0}^T\mathbf{z}_{t'}\rVert_2}\right\rangle\right|}{T-t+1}. \label{eq:thm1-bd5}
\end{align}
\normalsize
There exists some $C_1 \geq 1$ so that on the event $\mathcal{E}$ defined in Lemma \ref{lemma:thm1-event-bound} we have:
\begin{align*}
    \left|\left\langle \sum_{t'=t}^{t_0-1}\frac{\mathbf{z}_{t'}}{\sqrt{t_0-t}}, \frac{\sum_{t'=t_0}^{T}\mathbf{z}_{t'}}{\lVert\sum_{t'=t_0}^{T}\mathbf{z}_{t'}\rVert_2}\right\rangle\right| &< \sqrt{C_1\log T} \tag{on $\mathcal{E}_3\subseteq\mathcal{E}$} \\
    \left|\frac{\left\lVert\sum_{t'=t_0}^T\mathbf{z}_{t'}\right\rVert_2^2}{T-t_0+1} - d\right| &< C_1\sqrt{d}\log T \tag{on $\mathcal{E}_4\subseteq\mathcal{E}$} \\
    \left|\frac{\left\lVert\sum_{t'=t}^{t_0-1}\mathbf{z}_{t'}\right\rVert_2^2}{t_0-t} - d\right| &< C_1\sqrt{d}\log T \tag{on $\mathcal{E}_5\subseteq\mathcal{E}$} \\
    \left|\frac{\left\lVert\sum_{t'=t_0}^{T}\mathbf{z}_{t'}\right\rVert_2}{\sqrt{T-t_0+1}} -\sqrt{d}\right| &< \sqrt{C_1 \log T}. \tag{on $\mathcal{E}_6\subseteq\mathcal{E}$} 
\end{align*}
Returning to (\ref{eq:thm1-bd5}), the bounds above give:
\begin{align*}
    \frac{\lVert\sum_{t'=t_0}^T\mathbf{z}_{t'}\rVert_2^2}{2(T-t_0+1)} - \frac{\lVert\sum_{t'=t}^T\mathbf{z}_{t'}\rVert_2^2}{2(T-t+1)} &> -\frac{C_1(t_0-t)\sqrt{d}\log T}{T-t+1} - \frac{2C_1\sqrt{(t_0-t)(T-t_0+1)d}\log T}{T-t+1}.
\end{align*}
On $\mathcal{E}$, we therefore have:
\begin{align*}
    \log \frac{\alpha_{t_0}}{\alpha_t} &> \frac{d}{2}\log\left(\frac{T-t+1}{T-t_0+1}\right) + \frac{(t_0 - t)}{4(T-t+1)}\left[(T-t_0+1)\lVert\boldsymbol{\Lambda}^{\frac{1}{2}} \mathbf{b}_0\rVert^2_2 -4C_1\sqrt{d}\log T - 8C_1\sqrt{\frac{(T-t_0+1)d}{t_0-t}}\log T\right] \\
    &> \frac{(t_0 - t) }{4(T-t+1)}\left[(T-t_0+1)\lVert\boldsymbol{\Lambda}^{\frac{1}{2}} \mathbf{b}_0\rVert^2_2 -4C_1\sqrt{d}\log T - 8C_1\sqrt{\frac{(T-t_0+1)d}{t_0-t}}\log T\right]. \tag{$T-t+1 > T - t_0+1$}
\end{align*}
Note that $\sqrt{(T-t_0+1)/(t_0-t)}$ is increasing as a function of $t$ and $\sup \{t \leq t_0 \:|\:t\in\mathcal{T}_\beta\} = t_0 - \frac{C_\beta\log T}{\lVert\boldsymbol{\Lambda}^{1/2}\mathbf{b}_0\rVert_2^2}$, so plugging in the value for $t$ we have:
\begin{align*}
    \log \frac{\alpha_{t_0}}{\alpha_t} &> \frac{(t_0 - t) }{4(T-t+1)}\left[(T-t_0+1)\lVert\boldsymbol{\Lambda}^{\frac{1}{2}} \mathbf{b}_0\rVert^2_2 -4C_1\sqrt{d}\log T - \frac{8C_1}{\sqrt{C_\beta}}\sqrt{(T-t_0+1)\lVert\boldsymbol{\Lambda}^{\frac{1}{2}} \mathbf{b}_0\rVert_2^2d\log T}\right] \\
    &= \frac{(t_0 - t) (T-t_0+1)\lVert\boldsymbol{\Lambda}^{\frac{1}{2}} \mathbf{b}_0\rVert^2_2}{4(T-t+1)}\left[1 - \frac{4C_1\sqrt{d}\log T}{(T-t_0+1)\lVert\boldsymbol{\Lambda}^{\frac{1}{2}} \mathbf{b}_0\rVert^2_2} - 8C_1\sqrt{\frac{d\log T}{C_\beta (T-t_0+1)\lVert\boldsymbol{\Lambda}^{\frac{1}{2}} \mathbf{b}_0\rVert_2^2}}\right] \\
    &\geq  \frac{(t_0 - t) (T-t_0+1)\lVert\boldsymbol{\Lambda}^{\frac{1}{2}} \mathbf{b}_0\rVert^2_2}{4(T-t+1)}\left[1 - \frac{4C_1}{a_T\sqrt{d}} - \frac{8C_1}{\sqrt{C_\beta a_T}}\right]. \tag{by (\ref{eq:mean-assumption-bd})}
\end{align*}
So fo $T$ large enough that $a_T \geq \max\{16C_1,C_\beta^{-1}(32C_1)^2\}$, we have: 
\begin{align*}
      \log \frac{\alpha_{t_0}}{\alpha_t} &>\frac{(t_0 - t) (T-t_0+1)\lVert\boldsymbol{\Lambda}^{\frac{1}{2}} \mathbf{b}_0\rVert^2_2}{8(T-t+1)}
\end{align*}
The term $(t_0 - t)/ (T-t+1)$ is decreasing as a function of $t$, so again plugging in $t = t_0 - \frac{C_\beta\log T}{\lVert\boldsymbol{\Lambda}^{1/2}\mathbf{b}_0\rVert_2^2}$ gives:
\begin{align*}
    \log \frac{\alpha_{t_0}}{\alpha_t} &> \frac{ C_\beta\log T}{8\left(1 + \frac{C_\beta\log T}{(T-t_0+1)\lVert\boldsymbol{\Lambda}^{1/2}\mathbf{b}_0\rVert_2^2} \right)} \\
    &\geq  \frac{ C_\beta\log T}{8\left(1 + \frac{C_\beta}{da_T} \right)}. \tag{by (\ref{eq:mean-assumption-bd})}
\end{align*}
For $C_\beta \geq 16 C_\pi$ and $a_T \geq C_\beta$ we get $\log \frac{\alpha_{t_0}}{\alpha_t} \geq C_\pi\log T$, i.e. the event in (\ref{eq:thm1-result}) holds for this case.
\end{proof}
\subsection{Theorem \ref{theorem:smscp} Event Bounds}

\begin{lemma}[Theorem \ref{theorem:smscp} Event Bounds]\label{lemma:thm3-event-bound}
Let $\{y_t\}_{t=1}^T$ be a sequence of independent, sub-Gaussian random variables with $\sup_{t \geq 1} \lVert y_t\rVert_{\psi_2} < \infty$. Let $\E[y_t] = b_0\mathbbm{1}\{t\geq t_0\}$ and $\normalfont{\Var}(y_t) = (s_0^2)^{\mathbbm{1}\{t\geq t_0\}}$, with $b_0 \in \mathbb{R}$ and $s_0 > 0$ so that it is possible that $b_0 = 0$ or $s_0 = 1$. Define the standardized term $z_t:= (y_t - b_0\mathbbm{1}\{t\geq t_0\})/s_0^{\mathbbm{1}\{t\geq t_0\}}$ and suppose that either Assumption \ref{assumption:mean}, \ref{assumption:scale}, or both hold. 
Then for any $\beta > 0$, there exist some constants $C_1, C_2, C_{1,\beta}, C_{2,\beta} > 0$ such that if we define the set:
\begin{align*}
    \mathcal{T}_{\beta} &:= \left\{ 1 \leq t \leq  T - C_{2,\beta}\log(\log T)\::\: |t_0 - t| > \frac{C_{1,\beta}\log T}{\max\{\min\{b_0^2,b_0^2/s_0^2\},(s_0^2-1)^2\}} \right\} 
\end{align*}
and the events: 
\begin{align*}
    \mathcal{E}_1 &:= \bigcap_{t \in \mathcal{T}_{\beta}} \left\{\left|\sum_{t'=\min\{t_0,t\}}^{\max\{t_0,t\}-1} z_{t'}\right| < C_1\sqrt{|t_0-t|\log T}\right\}, \\
    \mathcal{E}_2 &:= \bigcap_{t=1}^T \left\{\left|\sum_{t'=t}^T z_{t'}\right| < C_1\sqrt{(T-t+1)\log T}\right\}, \\
    \mathcal{E}_3 &:= \bigcap_{t \in \mathcal{T}_\beta} \left\{\left|\sum_{t'=t}^T z_{t'}\right| < \frac{T-t+1}{2}\right\}, \\
    \mathcal{E}_4 &:= \bigcap_{t \in\mathcal{T}_\beta}  \left\{\left|\sum_{t'=t}^T z_{t'}\right| < C_1^{1/2}(T-t+1)^{3/4}\log^{1/4} T\right\},\\ 
    \mathcal{E}_5 &:= \bigcap_{t \in \mathcal{T}_{\beta}} \left\{ \left| \sum_{t'=\min\{t_0,t\}}^{\max\{t_0,t\}-1} (z_{t'}^2-1)\right| \leq C_1\sqrt{|t_0-t|\log T} + C_1\log T\right\}, \\
    \mathcal{E}_6 &:= \bigcap_{t \in \mathcal{T}_\beta} \left\{\left|\sum_{t'=t}^T (z_{t'}^2 - 1)\right| < \frac{T-t+1}{4}\right\},\\
    \mathcal{E}_7 &:= \bigcap_{t =1}^T \left\{\left|\sum_{t'=t}^T (z_{t'}^2 - 1)\right| < C_1\sqrt{(T-t+1)\log T}\right\},\\
    \mathcal{E}_{8} &:=\left\{\left|b_0\sum_{t'=t_0}^T z_{t'}\right| \leq \frac{(T-t_0+1)b_0^2 }{4}\right\},
\end{align*}
and the joint event $\mathcal{E} := \bigcap_{i=1}^{12} \mathcal{E}_i$, then we will have $\Pr(\mathcal{E}) \geq 1 - \frac{C_2}{T^{\beta}} - \frac{C_{2,\beta}}{\log^\beta T}$. 
\end{lemma}

\begin{proof}
By the union bound $\Pr(\mathcal{E}^c) = 1 - P(\cup_{i=1}^{12} \mathcal{E}^c_i) < 1 - \sum_{i=1}^{12} \Pr(\mathcal{E}^c_i)$, so it is sufficient to show that $\Pr(\mathcal{E}^c_i) \leq \frac{C_i}{T^{\beta}} + \frac{C_{i,\beta}}{\log^\beta T}$ for each $i$. 

If Assumption \ref{assumption:scale} does not hold, then we assume $\Var(y_t)$ is constant and it is without loss of generality to assume that $s_0=1$. If Assumption \ref{assumption:scale} does hold, then there is some $\underline{s} > 0$ such that $s_0 > \underline{s}$. Let $M_s = \min\{1, \underline{s}\}$, then by Lemma 2.6.8 of \cite{Vershynin18} we have: 
\begin{align*}
    \lVert z_t \rVert_{\psi_2} &\leq 
    \frac{\left\lVert y_t - b_0\mathbbm{1}\{t\geq t_0\}\right\rVert_{\psi_2}}{M_s} \lesssim \frac{\left\lVert y_t\right\rVert_{\psi_2}}{M_s}.
\end{align*}
By assumption, there is some $K > 0$ so that $\sup_{t \geq 1} \lVert y_t\rVert_{\psi_2} < K$. Therefore, $\lVert z_t \rVert_{\psi_2} \lesssim M_s^{-1}K$, i.e. there exists some constant $K_1 \geq  M_s^{-1}K$ so that $z_t \in \mathcal{SG}(K_1)$. Then by Lemma \ref{lemma:sum-sub-gaussian}, for any index set $\mathcal{T}$, we have:
\begin{align*}
    \sum_{t\in\mathcal{T}} z_t \in\mathcal{SG}\left(\sqrt{|\mathcal{T}|}K_1\right)
\end{align*}
and thus, for any index set $\mathcal{T}_t$ that depends on $t$, we have:
\begin{align*}
    \Pr\left(\left|\sum_{t'\in\mathcal{T}_t} z_{t'}\right| \geq C_1 \sqrt{|\mathcal{T}_t|\log T}\right) \leq 2 \exp\left[-\frac{C_1^2\log T}{2K_1^2}\right]. \tag{Chernoff bound (\ref{eq:chernoff})} 
\end{align*}
So for $C_1 \geq \sqrt{2K_1^2(1+\beta)}$, we have: 
\begin{align*}
    \Pr\left(\bigcup_{t\in\mathcal{T}} \left\{\left|\sum_{t'\in\mathcal{T}} z_{t'}\right| \geq C_1\sqrt{ |\mathcal{T}_t|\log T}\right\} \right) &\leq \sum_{t\in\mathcal{T}}\Pr\left( \left|\sum_{t'\in\mathcal{T}_t} z_{t'}\right| \geq C_1 \sqrt{|\mathcal{T}_t|\log T} \right) \tag{union bound} \\
    &\leq \sum_{t\in\mathcal{T}}  2 \exp\left[-\frac{C^2_1\log T}{2K_1^2}\right] \\
    &\leq \frac{2}{T^\beta} \tag{$\mathcal{T}\subset[T]$ and $C^2_1 \geq 2K_1^2(1+\beta)$}
\end{align*}
which proves $\Pr(\mathcal{E}^c_1), \Pr(\mathcal{E}^c_2) \leq \frac{2}{T^\beta}$.

Next, we have: 
\begin{align*}
    \Pr(\mathcal{E}_3^c) &= \Pr\left(\bigcup_{t\in\mathcal{T}_\beta}\left\{\left|\sum_{t'=t}^T z_{t'}\right| \geq \frac{T-t+1}{4}\right\}\right) \\
    &\leq \sum_{t\in\mathcal{T}_\beta}\Pr\left(\left|\sum_{t'=t}^T z_{t'}\right| \geq \frac{T-t+1}{4}\right) \tag{union bound} \\
    &\leq \sum_{t\in\mathcal{T}_\beta} 2\exp\left[-\frac{T-t+1}{32K^2_1}\right] \tag{Chernoff bound (\ref{eq:chernoff})} \\
    &\leq \sum_{1 \leq t \leq T-C_{2,\beta} \log(\log T)} 2\exp\left[-\frac{T-t+1}{32K^2_1}\right] 
 \tag{$\mathcal{T}_\beta \subset [T-C_{2,\beta} \log(\log T)]$} \\
    &\leq \sum_{1\leq t \leq T-C_{2,\beta}\log T} 2\exp\left[-\frac{T-t+1}{32K^2_1}\right]\\
    &\quad+ \sum_{T- C_{2,\beta} \log T \leq t \leq T-C_{2,\beta} \log(\log T)} 2\exp\left[-\frac{T-t+1}{32K^2_1}\right].
\end{align*}
Note that for the indices of the first sum $T-t+1\geq C_{2,\beta}\log T$ and there are at most $T$ terms in the sum. Similarly, for the indices of the second sum $T-t+1 \geq C_{2,\beta} \log(\log T)$ and there are at most $C_{2,\beta} \log T$. Therefore:
\begin{align*}
    \sum_{1\leq t \leq T-C_{2,\beta}\log T} 2\exp\left[-\frac{T-t+1}{32K^2_1}\right] &\leq 2\exp\left[\log T-\frac{C_{2,\beta}\log T}{32K_1^2}\right] \\
    \sum_{T- C_{2,\beta} \log T \leq t \leq T-C_{2,\beta} \log(\log T)} 2\exp\left[-\frac{T-t+1}{32K^2_1}\right] &\leq  2C_{2,\beta} \exp\left[\log(\log T) - \frac{C_{2,\beta} \log(\log T)}{32 K_1^2} \right].
\end{align*}
So for $ C_{2,\beta}\geq 32K_1^2(1+\beta)$ we get:
\begin{align*}
    \Pr(\mathcal{E}_3^c) &\leq \frac{2}{T^{\beta}} + \frac{2C_{2,\beta}}{\log^\beta T}.
\end{align*}

Next, we have:
\begin{align*}
    \Pr(\mathcal{E}_4^c) &\leq \Pr\left(\bigcup_{t\in\mathcal{T}_\beta}\left\{\left|\sum_{t'=t}^T z_{t'}\right| \geq C_1^{1/2}(T-t+1)^{3/4}\log^{1/4} T\right\}\right) \\
    &\leq  \sum_{t\in\mathcal{T}_\beta} \Pr\left(\left|\sum_{t'=t}^T z_{t'}\right| \geq C_1^{1/2}(T-t+1)^{3/4}\log^{1/4} T\right)\tag{union bound} \\
    &\leq \sum_{t\in\mathcal{T}_\beta} 2\exp\left[- \frac{C_1\sqrt{(T-t+1)\log T}}{2K_1^2}\right] \tag{Chernoff bound (\ref{eq:chernoff})} \\
    &\leq  \sum_{1\leq t \leq T-\log T} 2\exp\left[- \frac{C_1\sqrt{(T-t+1)\log T}}{2K_1^2}\right] \\
    &\quad \sum_{T-\log T < t \leq T} 2\exp\left[- \frac{C_1\sqrt{(T-t+1)\log T}}{2K_1^2}\right] \\
    &\leq  2\exp\left[\log T- \frac{C_1\log T}{2K_1^2}\right] \\
    &\quad  2\exp\left[\log(\log (T))- \frac{C_1\sqrt{\log T}}{2K_1^2}\right]
\end{align*}
We can pick $C_1 \geq 2K_1^2(1+\beta)$, then since $\sqrt{T} \gg \log (\log T)$, an identical argument as before gives:
\begin{align*}
    \Pr(\mathcal{E}_4^c) = &\leq \frac{2}{T^{\beta}} + \frac{2}{\log^\beta T}.
\end{align*}

Next, since $\E[z^2_t] =1$, then by Exercise 2.7.10 \citealp{Vershynin18}, $\lVert z^2_t - 1 \rVert_{\psi_1} \lesssim \lVert z^2_t \rVert_{\psi_1}$, and by Lemma 2.7.6 of \citealp{Vershynin18}, $\lVert z^2_t \rVert_{\psi_1} = \lVert z_t \rVert_{\psi_2}^2 \leq K_1^2$. So by Proposition 2.7.1 of \citealp{Vershynin18} there exists some constant $K_2\geq K_1^2$ so that $z_t^2\in\mathcal{SE}(K_2,K_2)$. Then again by Lemma \ref{lemma:sum-sub-gaussian}:
\begin{align*}
    \sum_{t\in\mathcal{T}} (z^2_t -1) \in\mathcal{SE}\left(\sqrt{|\mathcal{T}|}K_2,K_2\right).
\end{align*}
Using Bernstein’s inequality (see Theorem 2.8.1 of \cite{Vershynin18}) for any $t \in [T]$ we get:
\footnotesize
\begin{align*}
     \Pr(\mathcal{E}_5^c)&\leq \Pr\left(\left|\sum_{t'=\min\{t_0,t\}}^{\max\{t_0,t\}-1} (z_{t'}^2-1)\right| >C_1\sqrt{|t_0-t| \log T}  + C_1\log T\right) \\
     &\leq \exp\left[-\frac{1}{2K_2} \min\left\{ C_1\sqrt{|t_0-t| \log T} + C_1\log T, \frac{C_1^2(\log T +\log^2T/|t_0-t|)}{K_2}\right\}\right] \\
     &\leq \exp\left[-\frac{C_1\log T}{2K^2_2} \right]. \tag{$C_1,K_2 > 1 $ WLOG}
\end{align*}
\normalsize
So for $C_1 >2K_2^2(1+\beta)$, we can use the same union bounding argument to get $\Pr(\mathcal{E}_5^c) < T^{-\beta}$.

Note that since $z_{t'}^2 \in \mathcal{SE}(K_2,K_2)$ and $K_1 \geq 1$ WLOG, then:
\begin{align*}
    \Pr\left(\left|\sum_{t'=t}^T (z_{t'}^2 - 1)\right| \geq \frac{T-t+1}{4}\right) \leq 2\exp\left[-\frac{T-t+1}{32K_2^2}\right].
\end{align*}
So by an identical argument as we used for $\mathcal{E}_3$ we can pick $ C_{2,\beta}\geq 32K_2^2(1+\beta)$ to get:
\begin{align*}
    \Pr(\mathcal{E}_6^c) = \Pr\left(\bigcup_{t\in\mathcal{T}_\beta}\left\{\left|\sum_{t'=t}^T z^2_{t'}-1\right| \geq \frac{T-t+1}{4}\right\}\right) &\leq \frac{2}{T^{\beta}} + \frac{2C_{2,\beta}}{\log^\beta T}.
\end{align*}
Similarly, we have:
\begin{align*}
    \Pr(\mathcal{E}_7^c) &= \Pr\left(\bigcup_{t=1}^T\left\{\left|\sum_{t'=t}^T (z_{t'}^2 - 1)\right| \geq C_1\sqrt{(T-t+1)\log T}\right\}\right) \\
    &\leq \sum_{t=1}^T\Pr\left(\left|\sum_{t'=t}^T (z_{t'}^2 - 1)\right| \geq C_1\sqrt{(T-t+1)\log T}\right) \tag{union bound} \\
    &\leq\sum_{t=1}^T 2\exp\left[-\frac{1}{2K_2}\min\left\{C_1\sqrt{(T-t+1)\log T}, \frac{C_1^2\log T}{K_2}\right\}\right] \tag{Bernstein inequality} \\
    &\leq \sum_{1\leq t \leq T-\log T} 2\exp\left[-\frac{1}{2K_2}\min\left\{C_1\sqrt{(T-t+1)\log T}, \frac{C_1^2\log T}{K_2}\right\}\right] \\
    &\quad + \sum_{T-\log T < t \leq T} 2\exp\left[-\frac{1}{2K_2}\min\left\{C_1\sqrt{(T-t+1)\log T}, \frac{C_1^2\log T}{K_2}\right\}\right] \\
    &\leq  2\exp\left[\log T-\frac{C_1\log T}{2K^2_2}\right] \tag{$T-t+1 \geq \log T$}\\
    &\quad + \sum_{T-\log T < t \leq T} 2\exp\left[\log(\log T)-\frac{C_1\sqrt{\log T}}{2K^2_2}\right]. \tag{$T-t+1 \geq 1$}
\end{align*}
So for $C_1 \geq 2K_2^2(1+\beta)$ the same argument we used for $\mathcal{E}_4$ gives:
\begin{align*}
    \Pr(\mathcal{E}_7^c) = &\leq \frac{2}{T^{\beta}} + \frac{2}{\log^\beta T}.
\end{align*}
Finally, note that event $\mathcal{E}_8$ holds trivially if $b_0 =0$, while if Assumption \ref{assumption:mean} holds, then $(T-t_0+1)b_0^2 \gg \log T$, so for large $T$ we have $(T-t_0+1)b_0^2\geq 32K_1^2(1+\beta)\log T$ we have:
\begin{align*}
    \Pr\left(\left|b_0\sum_{t'=t_0}^T z_{t'}\right| >\frac{(T-t_0+1)b_0^2 }{4}\right) &\leq 2\exp\left[- \frac{(T-t_0+1)b_0^2}{32K_1^2}\right] \tag{Chernoff bound (\ref{eq:chernoff})} \\
    &\leq \frac{2}{T^\beta}
\end{align*}
\end{proof}
\subsection{Proof of Theorems \ref{theorem:sscp} and \ref{theorem:smscp}}
\label{app:localization-smscp}

\cite{Cappello25} proved Theorem \ref{theorem:sscp} under the stronger assumptions of Gaussian data, $(s^2_0-1)^2>B$ for some fixed $B>0$, and the minimum spacing condition $\Delta_T > cT$ for some $c \in (0,1/2)$. The generalization of Theorem \ref{theorem:sscp} to sub-Gaussian data, and the signal strength and the minimum spacing conditions implied by Assumption \ref{assumption:scale} follows the same argument we use in proving Theorem \ref{theorem:smscp} and is therefore omitted.

If we take the ratio of the posterior probabilities of the change-point location from the meanvar-scp model, we get:
\begin{align*}
    \log \frac{\Pr(\tau = t_0  \:|\: \mathbf{y}_{1:T} ; \omega_0, u_0, v_0,\boldsymbol{\pi}_{1:T})}{\Pr(\tau = t  \:|\: \mathbf{y}_{1:T} \:; \omega_0, u_0, v_0,\boldsymbol{\pi})} &=  \log \frac{p(\mathbf{y}_{1:T} \:|\:\tau = t_0 \:; \omega_0, u_0, v_0)\Pr(\tau = t_0 \:; \boldsymbol{\pi}_{1:T})/ p(\mathbf{y}_{1:T}\:;\omega_0, u_0, v_0,\boldsymbol{\pi}_{1:T})}{p(\mathbf{y}_{1:T} \:|\:\tau = t \:; \omega_0, u_0, v_0)\Pr(\tau = t \:;\boldsymbol{\pi}_{1:T})/ p(\mathbf{y}_{1:T}\:;\omega_0, u_0, v_0,\boldsymbol{\pi}_{1:T})} \tag{Bayes' rule} \\
    &= \log \frac{p(\mathbf{y}_{1:T} \:|\:\tau = t_0 \:; \omega_0, u_0, v_0)}{p(\mathbf{y}_{1:T} \:|\:\tau = t \:; \omega_0, u_0, v_0)} + \log\frac{\pi_{t_0}}{\pi_t}.
\end{align*}
For some $C_{1,\beta},C_{2,\beta} > 0$ we can again define the set:
\begin{align*}
    \mathcal{T}_{\beta} &:= \left\{ 1 \leq t \leq  T - C_{2,\beta}\log(\log T)\::\: |t_0 - t| > \frac{C_{1,\beta}\log T}{\max\{\min\{b_0^2,b_0^2/s_0^2\},(s_0^2-1)^2\}} \right\} 
\end{align*}
By the assumptions of Theorem \ref{theorem:smscp}, for large $T$ there is some constant $C_\pi > 0$ such that for all $t \in [T]$: 
\begin{align}
    \log \frac{\pi_t}{\pi_{t_0}} \leq C_\pi \log T.
\end{align}
Therefore, if we can show that for large enough $C_{1,\beta}$ and $C_{2,\beta}$:
\begin{align}
    \lim_{T\to\infty}\Pr\left(\bigcap_{t\in\mathcal{T}_{\beta}}\left\{ \log \frac{p(\mathbf{y}_{1:T} \:|\:\tau = t_0 \:; \omega_0, u_0, v_0)}{p(\mathbf{y}_{1:T} \:|\:\tau = t \:; \omega_0, u_0, v_0)} > C_\pi \log T\right\}\right) = 1 \label{eq:thm3-result}
\end{align}
then with probability approaching one as $T \to \infty$ we will have:
\begin{align*}
    \max_{t\in\mathcal{T}_{\beta}} \; \Pr(\tau = t  \;|\; \mathbf{y}_{1:T} ; \omega_0, u_0, v_0,\boldsymbol{\pi}_{1:T}) < \Pr(\tau = t_0  \;|\; \mathbf{y}_{1:T} ; \omega_0, u_0, v_0,\boldsymbol{\pi}_{1:T}).
\end{align*}
We now prove that there exists a $C_{1,\beta},C_{2,\beta}> 0$ such that the event in (\ref{eq:thm3-result}) holds on the event $\mathcal{E}$ defined in Lemma \ref{lemma:thm3-event-bound}.

\begin{proof} 
We begin by directly calculating the marginal evidence of $\mathbf{y}_{1:T}$: 
\begin{align*}
    p(\mathbf{y}_{1:T} \;|\; \tau = t; \omega_0,u_0,v_0) &= \int_{\mathbb{R}^2} \mathbb{P}(\mathbf{y}_{1:T} \:|\:b,s,\tau = t) \:\; d\mathbb{P}(b,s;\omega_0, u_0, v_0) \\
    &= \int_{\mathbb{R}^2}  \prod_{t=1}^Tp(y_t \;|\; b, s, \tau =t) p(b|s;\omega_0)p(s;u_0,v_0)\;db\;ds\\
    &= \int_{0}^\infty\int_{-\infty}^\infty (2\pi)^{-T/2} s^{(T-t+1)/2} \exp\left[-\frac{\sum_{t'=1}^{t-1} y_{t'}^2 + \sum_{t'=t}^{T} s(y_{t'} - b)^2}{2}\right]  \\
    &\quad\quad\quad\quad\quad\quad\quad  \times (2\pi s\omega_0)^{\frac{1}{2}}\exp\left[-\frac{\omega_0 s b^2}{2}\right] \\
    &\quad\quad\quad\quad\quad\quad\quad  \times\frac{v_0^{u_0}}{\Gamma(u_0)}s^{u_0- 1}\exp(-v_0 s) \; db\; ds \\
    &\propto \exp\left[-\frac{1}{2}\sum_{t'=1}^{t-1} y_{t'}^2\right] \int_{0}^\infty\int_{-\infty}^\infty s^{\frac{1}{2}}\exp\left[-\frac{s}{2}\left((\omega_0 + T-t+1)b^2 - 2b\sum_{t'=t}^{T}y_{t'}\right)\right] \\
    &\quad\quad\quad\quad\quad\quad\quad\quad\quad\quad\quad\quad\quad   \times s^{u_0 + \frac{T-t+1}{2} - 1}\exp\left[-s \left(v_0 + \frac{1}{2}\sum_{t=t'}^Ty_{t'}^2\right)\right] \; db\; ds \\
    &= \exp\left[-\frac{1}{2}\sum_{t'=1}^{t-1} y_{t'}^2\right] \int_{0}^\infty\int_{-\infty}^\infty s^{\frac{1}{2}}\exp\left[-\frac{s\overline{\omega}_t}{2}\left(b - \overline{b}_{t}\right)^2\right] s^{\overline{u}_t  - 1}\exp\left[-s \overline{v}_t\right] \; db\; ds \tag{$\overline{\omega}_t = \omega_0 + T - t+1, \; \overline{b}_t = \overline{\omega}_t^{-1}\sum_{t'=t}^T y_{t'}, \;\overline{u}_t = u_0 + \frac{T-t+1}{2}, \;\overline{v}_t = v_0 +\frac{1}{2} \sum_{t'=t}^T y_{t'}^2 - \frac{\overline{\omega}_t\overline{b}_t^2}{2}$} \\ 
    &\propto \exp\left[-\frac{1}{2}\sum_{t'=1}^{t-1} y_{t'}^2\right] \overline{\omega}_t^{-\frac{1}{2}} \frac{\Gamma(\overline{u}_t)}{\overline{v}_t^{\overline{u}_T}} \int_{0}^\infty\int_{-\infty}^\infty p(b|s;\overline{b}_t,\overline{\omega}_t)p(s;\overline{u}_t,\overline{v}_t)\;db\;ds \\
    &= \exp\left[-\frac{1}{2}\sum_{t'=1}^{t-1} y_{t'}^2\right] \overline{\omega}_t^{-\frac{1}{2}} \frac{\Gamma(\overline{u}_t)}{\overline{v}_t^{\overline{u}_T}}. 
\end{align*}
Taking the log of the expression above we get:
\begin{align*}
    \log p(\mathbf{y}_{1:T} \;|\; \tau = t; \omega_0,u_0,v_0) &= -\frac{1}{2}\sum_{t'=1}^{t-1} y_{t'}^2 - \frac{1}{2} \log \overline{\omega}_t + \log \Gamma(\overline{u}_t) - \overline{u}_t\log \overline{v}_t \\
    &= -\frac{1}{2}\log(\omega_0 + T - t + 1) - \frac{1}{2}\sum_{t'=1}^{t-1} y_{t'}^2 + \log \Gamma\left(u_0 + \frac{T - t +1}{2}\right)  \\
    &\quad\: - \left(u_0 + \frac{T - t +1}{2}\right)\log\left[v_0 +\frac{1}{2}\left(\sum_{t'=t}^T y_{t'}^2 - \frac{\left(\sum_{t'=t}^T y_{t'}\right)^2 }{\omega_0 + T- t +1}\right)\right] \\
    &\quad\: + C
\end{align*}
where $C$ is a constant that does not depend on the value of $\tau$. Recall Binet's first formula for the log gamma function for $x > 0$ (see e.g. p. 249 of \citealp{Whittaker96}): 
\begin{align*}
    \log \Gamma(x) = \left(x - \frac{1}{2}\right) \log x - x + \frac{1}{2} \log 2 \pi + \int_{0}^\infty\left(\frac{1}{2} - \frac{1}{t} + \frac{1}{e^t - 1} \right)\frac{e^{-tx}}{t} \;dt.
\end{align*}
For $x \geq 1/2$ we have:
\begin{align*}
    \int_{0}^\infty\left(\frac{1}{2} - \frac{1}{t} + \frac{1}{e^t - 1} \right)\frac{e^{-tx}}{t} \;dt \leq \int_{0}^\infty\left(\frac{1}{2} - \frac{1}{t} + \frac{1}{e^t - 1} \right)\frac{e^{-t/2}}{t} \;dt \approx 0.153
\end{align*}
and since $\frac{T-t+1}{2} \geq \frac{1}{2}$, we can write:
\begin{align*}
    \log \Gamma\left(u_0 + \frac{T - t +1}{2}\right) = \left(u_0 + \frac{T - t}{2}\right) \log\left(u_0 + \frac{T - t +1}{2}\right) - \frac{T - t +1}{2} +\mathcal{O}(1)
\end{align*}
and thus:
\small
\begin{align*}
    \log p(\mathbf{y}_{1:T} \;|\; \tau = t; \omega_0,u_0,v_0) &= -\frac{1}{2}\log(\omega_0 + T - t + 1) - \frac{1}{2}\sum_{t'=1}^{t-1} y_{t'}^2  \\
    &\quad\: +\left(u_0 + \frac{T-t}{2}\right) \log \left(u_0 + \frac{T-t+1}{2}\right) - \frac{(T-t+1)}{2}  \\
    &\quad\: - \left(u_0 + \frac{T - t +1}{2}\right)\log\left[v_0 +\frac{1}{2}\left(\sum_{t'=t}^T y_{t'}^2 - \frac{\left(\sum_{t'=t}^T y_{t'}\right)^2}{\omega_0 + T- t +1}\right) \right] \\
    &\quad\:+ \mathcal{O}(1).
\end{align*}
\normalsize
Letting $\{\omega_0,u_0, v_0\} \to \mathbf{0}$, we have:
\small
\begin{align}
    \lim_{\omega_0,u_0, v_0 \to \mathbf{0}}\log p(\mathbf{y}_{1:T} \;|\; \tau = t; \omega_0,u_0,v_0) &= -\frac{1}{2}\log(T - t + 1) - \frac{1}{2}\sum_{t'=1}^{t-1} y_{t'}^2  \notag \\
    &\quad\: + \left(\frac{T-t}{2}\right) \log \left(\frac{T-t+1}{2}\right) - \frac{(T-t+1)}{2}  \notag\\
    &\quad\: - \left(\frac{T - t +1}{2}\right)\log\left[\frac{1}{2}\left(\sum_{t'=t}^T y_{t'}^2 - \frac{\left(\sum_{t'=t}^T y_{t'}\right)^2}{T- t +1}\right) \right] \notag \\
    &\quad\:+ \mathcal{O}(1)\notag \\
    &= -\log(T - t + 1) - \frac{1}{2}\sum_{t'=1}^{t-1} y_{t'}^2  - \frac{(T-t+1)}{2}  \notag\\
    &\quad\: - \left(\frac{T - t +1}{2}\right)\log\left[\frac{\sum_{t'=t}^T y_{t'}^2}{T-t+1} - \left(\frac{\sum_{t'=t}^T y_{t'}}{T-t+1}\right)^2 \right] \notag \\
    &\quad\:+ \mathcal{O}(1)\notag \\
    &:= \log \alpha_t. \label{eq:thm3-delta_t}
\end{align}
\normalsize
Note that when $t = T$, then:
\begin{align*}
    \frac{\sum_{t'=t}^T y_{t'}^2}{T-t+1} = \left(\frac{\sum_{t'=t}^T y_{t'}}{T-t+1}\right)^2  
\end{align*}
but $T \not\in \mathcal{T}_\beta$, so $\alpha_t$ is still well-defined for each $t\in\mathcal{T}_\beta$. As in the proof of Theorem \ref{theorem:smcp}, we can use $\alpha_t$ to approximate the log-evidence for small values of $\omega_0$, $u_0$, and $v_0$ (see Footnote \ref{fn:approx}). Let $C_{1,\beta}$ be large enough so that the statement of Lemma \ref{lemma:thm3-event-bound} holds. Then we will be on the event $\mathcal{E}$ defined in Lemma \ref{lemma:thm3-event-bound} with probability approaching one as $T\to\infty$. We can show that the event in (\ref{eq:thm3-result}) holds on $\mathcal{E}$. To do so, we define the standardized terms: 
\begin{align*}
    z_t = 
    \begin{cases}
        y_t, & \text{if } t < t_0, \\
        \frac{y_t - b_0}{s_0},& \text{if } t \geq t_0.
    \end{cases}
\end{align*}
There are two cases to consider: 

\subsubsection*{Case 1: $t > t_0$, and $t \in \mathcal{T}_{\beta}$.}

From (\ref{eq:thm3-delta_t}) we have:
\begin{align*}
    \log \frac{\alpha_{t_0}}{\alpha_t} &= \log \left(\frac{T - t + 1}{T - t_0 +1} \right) + \frac{1}{2} \sum_{t' = t_0}^{t - 1} y_{t'}^2 - \frac{(t-t_0)}{2} \\
    &\quad\: + \left(\frac{T - t +1}{2}\right)\log\left[\frac{\sum_{t'=t}^T y_{t'}^2}{T-t+1} - \left(\frac{\sum_{t'=t}^T y_{t'}}{T-t+1}\right)^2 \right] \\
    &\quad\: -\left(\frac{T - t_0 +1}{2}\right)\log\left[\frac{\sum_{t'=t_0}^T y_{t'}^2}{T-t_0+1} - \left(\frac{\sum_{t'=t_0}^T y_{t'}}{T-t_0+1}\right)^2 \right] \\
    &\quad\: + \mathcal{O}(1).
\end{align*}
Since $t > t_0$, we can standardize each term in the sums to get:
\begin{align*}
    \sum_{t' = t_0}^{t - 1} y_{t'}^2 &= \sum_{t' = t_0}^{t - 1} (s_0z_{t'} + b_0)^2 \\
    &= s_0^2 \sum_{t' = t_0}^{t - 1} (z_{t'}^2 -1) + 2 b_0 s_0 \sum_{t' = t_0}^{t - 1} z_{t'} + (t-t_0)(b_0^2 +s_0^2)
\end{align*}
and:
\begin{align*}
    \frac{\sum_{t'=t}^T y_{t'}^2}{T-t+1} - \left(\frac{\sum_{t'=t}^T y_{t'}}{T-t+1}\right)^2 &= \frac{\sum_{t'=t}^T (s_0z_{t'} + b_0)^2}{T-t+1} - \left(\frac{\sum_{t'=t}^T (s_0z_{t'} + b_0)}{T-t+1}\right)^2 \\
    &= s_0^2 \left[\frac{\sum_{t'=t}^T z_{t'}^2}{T-t+1} - \left(\frac{\sum_{t'=t}^T z_{t'}}{T-t+1}\right)^2\right].
\end{align*}
If we define the function $f(x) := x - \log(x)-1$,\footnote{Note that this function previously appeared in \cite{Bai10} as a measure of the signal strength of $s_0^2$.} then we now have:
\begin{align}
    \log \frac{\alpha_{t_0}}{\alpha_t} &= \log \left(\frac{T - t + 1}{T - t_0 +1} \right) + \frac{s_0^2}{2} \sum_{t' = t_0}^{t - 1} (z_{t'}^2-1) + b_0 s_0 \sum_{t' = t_0}^{t - 1} z_{t'} + \left(\frac{t-t_0}{2}\right)[b_0^2+ f(s_0^2)] \notag \\
    &\quad\: + \left(\frac{T - t +1}{2}\right)\log\left[\frac{\sum_{t'=t}^T z_{t'}^2}{T-t+1} - \left(\frac{\sum_{t'=t}^T z_{t'}}{T-t+1}\right)^2\right] \notag \\
    &\quad\: -\left(\frac{T - t_0 +1}{2}\right)\log\left[\frac{\sum_{t'=t_0}^T z_{t'}^2}{T-t_0+1} - \left(\frac{\sum_{t'=t_0}^T z_{t'}}{T-t_0+1}\right)^2\right] \notag \\
    &\quad\: + \mathcal{O}(1). \label{eq:thm3-cs1-bd0}
\end{align}
For any $t < T$, because $x^2$ is strictly convex on $(0,\infty)$, then with probability one we have:
\begin{align}
      \frac{\sum_{t'=t}^T z_{t'}^2}{T-t+1} > \left(\frac{\sum_{t'=t}^T z_{t'}}{T-t+1}\right)^2 \implies \frac{\sum_{t'=t}^T (z_{t'}^2 -1)}{T-t+1} - \left(\frac{\sum_{t'=t}^T z_{t'}}{T-t+1}\right)^2 > -1. \label{eq:neg1-bd}
\end{align}
Therefore, we can use the inequality $\frac{x}{1+x} \leq \log(1+x) \leq x$ for $x > -1$ to get:
\small
\begin{align*}
    -\log\left[1 + \frac{\sum_{t'=t_0}^T (z_{t'}^2 -1)}{T-t_0+1} - \left(\frac{\sum_{t'=t_0}^T z_{t'}}{T-t_0+1}\right)^2\right] &\geq -\frac{\sum_{t'=t_0}^T (z_{t'}^2 -1)}{T-t_0+1} + \left(\frac{\sum_{t'=t_0}^T z_{t'}}{T-t_0+1}\right)^2  >-\frac{\sum_{t'=t_0}^T (z_{t'}^2 -1)}{T-t_0+1} 
\end{align*}
\normalsize
and: 
\small
\begin{align*}
    \log\left[1 + \frac{\sum_{t'=t}^T (z_{t'}^2-1)}{T-t+1} - \left(\frac{\sum_{t'=t}^T z_{t'}}{T-t+1}\right)^2\right] &\geq \frac{\frac{\sum_{t'=t}^T (z_{t'}^2-1)}{T-t+1} - \left(\frac{\sum_{t'=t}^T z_{t'}}{T-t+1}\right)^2}{1 + \frac{\sum_{t'=t}^T (z_{t'}^2-1)}{T-t+1} - \left(\frac{\sum_{t'=t}^T z_{t'}}{T-t+1}\right)^2 } \\
    &= \frac{\sum_{t'=t}^T (z_{t'}^2-1)}{T-t+1} - \left(\frac{\sum_{t'=t}^T z_{t'}}{T-t+1}\right)^2 - \frac{\left[\frac{\sum_{t'=t}^T (z_{t'}^2-1)}{T-t+1} - \left(\frac{\sum_{t'=t}^T z_{t'}}{T-t+1}\right)^2 \right]^2}{1 + \frac{\sum_{t'=t}^T (z_{t'}^2-1)}{T-t+1} - \left(\frac{\sum_{t'=t}^T z_{t'}}{T-t+1}\right)^2 }.
\end{align*}
\normalsize
Combining these bounds with (\ref{eq:thm3-cs1-bd0}) gives: 
\begin{align}
    \log \frac{\alpha_{t_0}}{\alpha_t} &\geq  \log \left(\frac{T - t + 1}{T - t_0 +1} \right) + \left(\frac{t-t_0}{2}\right)\left[b_0^2 + f(s_0^2)\right]  + \frac{(s_0^2-1)}{2} \sum_{t' = t_0}^{t - 1} (z_{t'}^2-1) + b_0 s_0 \sum_{t' = t_0}^{t - 1} z_{t'} \notag \\
    &\quad\: - \frac{\left(\sum_{t'=t}^T z_{t'}\right)^2}{2(T-t+1)} - \frac{\frac{(T-t+1)}{2}\left[\frac{\sum_{t'=t}^T (z_{t'}^2-1)}{T-t+1} - \left(\frac{\sum_{t'=t}^T z_{t'}}{T-t+1}\right)^2 \right]^2}{1 + \frac{\sum_{t'=t}^T (z_{t'}^2-1)}{T-t+1} - \left(\frac{\sum_{t'=t}^T z_{t'}}{T-t+1}\right)^2 } + \mathcal{O}(1). \label{eq:thm3-cs1-bd1}
\end{align}
On the event $\mathcal{E}$ defined in Lemma \ref{lemma:thm3-event-bound}, for some constant $C_1 > 0$ we have:
\begin{align*}
    \left|\sum_{t'=t_0}^{t-1} z_{t'}\right| &\leq C_1\sqrt{(t-t_0)\log  T}, \tag{on $\mathcal{E}_1\subset\mathcal{E}$} \\
    \left|\sum_{t' = t_0}^{t - 1} (z_{t'}^2-1) \right| &\leq C_1\sqrt{(t-t_0)\log  T} +C_1\log T \tag{on $\mathcal{E}_5\subset\mathcal{E}$} 
\end{align*}
and since $-C_1|s_0^2 - 1| \log T \geq -C_1(\overline{s}^2 + 1) \log T = \mathcal{O}(\log T)$ and $\log \left(\frac{T - t + 1}{T - t_0 +1} \right) = \mathcal{O}(\log T)$, we have
\begin{align}
    \log \frac{\alpha_{t_0}}{\alpha_t} &\geq \left(\frac{t-t_0}{2}\right)\left[b_0^2 + f(s_0^2)\right]  - \left(\frac{|s_0^2-1|}{2}  + |b_0| s_0\right)C_1\sqrt{(t-t_0)\log  T} \notag \\
    &\quad\:  - \frac{\left(\sum_{t'=t}^T z_{t'}\right)^2}{2(T-t+1)} - \frac{\frac{(T-t+1)}{2}\left[\frac{\sum_{t'=t}^T (z_{t'}^2-1)}{T-t+1} - \left(\frac{\sum_{t'=t}^T z_{t'}}{T-t+1}\right)^2 \right]^2}{1 + \frac{\sum_{t'=t}^T (z_{t'}^2-1)}{T-t+1} - \left(\frac{\sum_{t'=t}^T z_{t'}}{T-t+1}\right)^2 } + \mathcal{O}(\log T). \label{eq:thm3-cs1-bd2}
\end{align}
Next, on the event $\mathcal{E}$ defined in Lemma \ref{lemma:thm3-event-bound}, for the same universal constant $C_1 > 0$ we have:
\begin{align*}
    \left|\sum_{t' = t}^{T} z_{t'}\right| &< C_1 \sqrt{(T-t+1)\log T} \tag{on $\mathcal{E}_2\subset\mathcal{E}$} 
\end{align*}
and thus:
\begin{align*}
    -\frac{\left(\sum_{t'=t}^T z_{t'}\right)^2}{2(T-t+1)} > - \frac{C_1^2}{2}\log T.
\end{align*}
Next, for $t\in\mathcal{T}_\beta$ on the event $\mathcal{E}$, we have:
\begin{align*}
    \left|\frac{\sum_{t'=t}^T z_{t'}}{T-t+1}\right| &\leq \frac{1}{2}, \tag{on $\mathcal{E}_3\subset\mathcal{E}$} \\
    \left|\sum_{t'=t}^T z_{t'}\right| &< C_1^{1/2}(T-t+1)^{3/4} \log^{1/4} T \tag{on $\mathcal{E}_4\subset\mathcal{E}$} \\
    \left| \frac{\sum_{t'=t}^T (z_{t'}^2-1)}{T-t+1}\right| &\leq \frac{1}{4}, \tag{on $\mathcal{E}_6\subset\mathcal{E}$} \\
    \left|\sum_{t'=t}^T (z_{t'}^2-1)\right| &< C_1\sqrt{(T-t+1)\log T}, \tag{on $\mathcal{E}_7\subset\mathcal{E}$}
\end{align*}
so on $\mathcal{E}$:
\begin{align*}
    \frac{\frac{(T-t+1)}{2}\left[\frac{\sum_{t'=t}^T (z_{t'}^2-1)}{T-t+1} - \left(\frac{\sum_{t'=t}^T z_{t'}}{T-t+1}\right)^2 \right]^2}{1 + \frac{\sum_{t'=t}^T (z_{t'}^2-1)}{T-t+1} - \left(\frac{\sum_{t'=t}^T z_{t'}}{T-t+1}\right)^2 } \leq 4C_1^2\log T.
\end{align*}
Combining these results, we can rewrite (\ref{eq:thm3-cs1-bd2}) on $\mathcal{E}$ as the deterministic bound:
\begin{align}
    \log \frac{\alpha_{t_0}}{\alpha_t} &\geq  \left(\frac{t-t_0}{2}\right)\left[b_0^2 + f(s_0^2)\right]  - \left(\frac{|s_0^2-1|}{2}  + |b_0| s_0\right)C_1\sqrt{(t-t_0)\log  T} + \mathcal{O}(\log T).\label{eq:thm3-cs1-bd3}
\end{align}
We now show that bound (\ref{eq:thm3-cs1-bd3}) is increasing as a function of $t-t_0$. Taking the derivative of this bound with respect to $t-t_0$ yields:
\begin{align}
    \frac{b_0^2 + f(s_0^2)}{2}  - \frac{C_1}{2}\left(\frac{|s_0^2-1|}{2}  + |b_0| s_0\right)\sqrt{\frac{\log T}{t-t_0}}. \label{eq:thm3-deriv1}
\end{align}
Note that (\ref{eq:thm3-deriv1}) is decreasing as a function of $t-t_0$. Suppose that: 
\begin{align*}
    \min\left\{b_0^2,\frac{b_0^2}{s_0^2}\right\} \geq (s_0^2-1)^2 \implies |t-t_0| \geq \frac{C_{1,\beta} \log T}{\min\{b_0^2,b_0^2/s_0^{2}\}}, \;\forall\; t\in\mathcal{T}_{\beta}
\end{align*}
and thus (\ref{eq:thm3-deriv1}) is bounded below by: 
\begin{align*}
    \frac{b_0^2}{2} - \frac{C_1}{2}\left(\frac{|s_0^2-1|}{2}  + |b_0| s_0\right)\frac{\min\{|b_0|,|b_0|/s_0\}}{\sqrt{C_{1,\beta}}} &\geq \frac{b_0^2}{2} - \frac{C_1}{2\sqrt{C_{1,\beta}}}\left(\frac{\min\{b^2_0,b^2_0/s^2_0\}}{2} + \min\{b^2_0s_0,b_0^2\}\right) \tag{$\min\left\{b_0^2,b_0^2/s_0^2\right\} \geq (s_0^2-1)^2$} \\
    &\geq\frac{b_0^2}{2}\left(1 - \frac{3C_1}{2\sqrt{C_{1,\beta}}}\right).
\end{align*}
So (\ref{eq:thm3-deriv1}) will be positive if we set $C_{1,\beta} > 9C_1^2/4$. On the other hand, if:
\begin{align*}
    \min\left\{b_0^2,\frac{b_0^2}{s_0^2}\right\} \leq (s_0^2 - 1)^2 \implies |t-t_0| \geq \frac{C_{1,\beta} \log T}{ (s_0^2 - 1)^2}, \;\forall\; t\in\mathcal{T}_{\beta}
\end{align*}
then (\ref{eq:thm3-deriv1}) is bounded below by: 
\begin{align*}
    \frac{f(s_0^2)}{2}  - \frac{C_1}{2}\left(\frac{|s_0^2-1|}{2}  + |b_0| s_0\right)\frac{|s_0^2-1|}{\sqrt{C_{1,\beta}}}.
\end{align*}
Suppose that $s_0^2 > 1$ so that $(s_0^2 - 1)^2 \geq b_0^2/s_0^2$, then:
\begin{align*}
    \frac{f(s_0^2)}{2}  - \frac{C_1}{2\sqrt{C_{1,\beta}}}\left(\frac{(s_0^2-1)^2}{2}  +  s^2_0(s_0^2-1)^2\right) &\geq \frac{f(s_0^2)}{2}  - \frac{C_1\overline{s}^2 (s_0^2-1)^2}{\sqrt{C_{1,\beta}}}. \tag{$s_0^2 \leq \overline{s}^2 < \infty$}
\end{align*}
Suppose that $s_0^2 > \overline{\nu}^2 > 1$ for all $T$, then: 
\begin{align*}
    \frac{f(s_0^2)}{2}  - \frac{C_1\overline{s}^2 (s_0^2-1)^2}{\sqrt{C_{1,\beta}}} \geq \frac{f(\overline{\nu}^2)}{2}  - \frac{C_1\overline{s}^2 (\overline{s}^2-1)^2}{\sqrt{C_{1,\beta}}} \tag{$s_0^2 \leq \overline{s}^2 < \infty$}
\end{align*}
and we can pick $C_{1,\beta}$ large enough so that the RHS is positive. If $s_0^2 < \underline{\nu}^2$ for some $\underline{\nu}^2 < 1$ for all $T$, then we can use the same argument and the fact that $s_0 > \underline{s} > 0$ to show (\ref{eq:thm3-deriv1}) is positive for large $C_{1,\beta}$. On the other hand, if $\lim_{T\to\infty}s_0^2=1$, then:
\begin{align*}
    \lim_{T\to\infty} \frac{f(s_0^2)}{(s_0^2 
- 1)^2} = \frac{1}{2}
\end{align*}
so for large $T$, we will have: 
\begin{align*}
    \frac{f(s_0^2)}{(s_0^2-1)^2} - \frac{1}{2} > - \frac{1}{4} \implies f(s_0^2) > \frac{(s_0^2-1)^2}{4}
\end{align*}
and thus:
\begin{align*}
    \frac{f(s_0^2)}{2}  - \frac{C_1\overline{s}^2 (s_0^2-1)^2}{\sqrt{C_{1,\beta}}} \geq \frac{f(s_0^2)}{2}\left(1  - \frac{C_1\overline{s}^2 }{2\sqrt{C_{1,\beta}}}\right).
\end{align*}
So for $C_{1,\beta} > C_1^2 \overline{s}^2/4$ bound (\ref{eq:thm3-deriv1}) is positive.

Having established that bound (\ref{eq:thm3-cs1-bd3}) is increasing as a function of $t-t_0$, then (\ref{eq:thm3-cs1-bd3}) is minimized by setting:
\begin{align*}
    t &= \inf \;\{ t\in\mathcal{T}_{\beta} : t>t_0\} \\
    &= t_0 + \frac{C_{1,\beta}\log T}{\max\{\min\{b_0^2,b_0^2/s_0^2\}, (s_0^2-1)^2\}}.
\end{align*}
Suppose first that $\min\{b_0^2,b_0^2/s_0^2\} \geq (s_0^2-1)^2$, then returning to bound (\ref{eq:thm3-cs1-bd3}), there is some $B_1 > 0$ so that: 
\begin{align*}
    \log \frac{\alpha_{t_0}}{\alpha_t} &\geq \left[\frac{C_{1,\beta}}{2\min\{1,s_0^{-2}\}} - \left(\frac{|s_0^2-1|}{2\min\{|b_0|,|b_0|/s_0\}} + \frac{s_0}{\min\{1,s_0^{-1}\}}\right)C_1\sqrt{C_{1,\beta}}\right]\log T - B_1 \log T \\
    &\geq \left[\frac{\max\{1,s_0^2\}(C_{1,\beta}-3C_1\sqrt{C_{1,\beta}})}{2} - B_1\right] \log T. \tag{$|s_0^2-1| \leq \min\{|b_0|,|b_0|/s_0\}$}
\end{align*}
So if we set $C_{1,\beta} > \max\{(3C_1 + 1)^2, 4(C_\pi + B_1)^2\}$, then (\ref{eq:thm3-result}) holds. On the other hand, if $\min\{b_0^2,b_0^2/s_0^2\} \leq (s_0^2-1)^2$, then: 
\begin{align*}
    \log \frac{\alpha_{t_0}}{\alpha_t} &\geq \left[\frac{C_{1,\beta}f(s_0^2)}{2(s_0^2-1)^2} - \left(\frac{1}{2} + \frac{|b_0|s_0}{|s_0^2-1|}\right)C_1\sqrt{C_{1,\beta}}\right] \log T - B_1 \log T \\
    &\geq \left[\frac{C_{1,\beta}f(s_0^2)}{2(s_0^2-1)^2} - \frac{3\max\{s_0^2,1\}C_1\sqrt{C_{1,\beta}}}{2} - B_1\right] \log T. \tag{$|s^2_0-1| \geq \min\{|b_0|,|b_0|/s_0\}$} 
\end{align*}
Again, if $s_0^2 > \overline{\nu}^2 > 1$ for all $T$, then: 
\begin{align*}
    \log \frac{\alpha_{t_0}}{\alpha_t} &\geq \left[\frac{C_{1,\beta}f(\overline{\nu}^2)}{2(\overline{s}^2-1)^2} - \frac{3\overline{s}^2C_1\sqrt{C_{1,\beta}}}{2} - B_1\right] \log T
\end{align*}
and we can pick $C_{1,\beta}$ large enough so that the term in square brackets is greater than $C_\pi$. An identical argument gives the same result when $s_0^2 < \underline{\nu}^2$ for some $\underline{\nu}<1$. When $\lim_{T\to\infty}s_0^2=1$, then again for large $T$, we will have: 
\begin{align*}
    \frac{f(s_0^2)}{(s_0^2-1)^2} \to \frac{1}{2} > - \frac{1}{4} \implies \frac{f(s_0^2)}{(s_0^2-1)^2} - \frac{1}{2} > \frac{1}{4}
\end{align*}
and $\max\{s_0^2,1\} < 3/2$, so: 
\begin{align*}
    \log \frac{\alpha_{t_0}}{\alpha_t} &\geq\left[\frac{C_{1,\beta}}{8} - \frac{9C_1\sqrt{C_{1,\beta}}}{4} - B_1\right] \log T
\end{align*}
and again we can pick $C_{1,\beta}$ large enough so that the term in square brackets is greater than $C_\pi$. 

\subsubsection*{Case 2: $t < t_0$ and $t\in\mathcal{T}_\beta$.}

From (\ref{eq:thm3-delta_t}) we again have:
\begin{align*}
    \log \frac{\alpha_{t_0}}{\alpha_t} &= \log \left(\frac{T - t + 1}{T - t_0 +1} \right) - \frac{1}{2} \sum_{t' = t}^{t_0 - 1} y_{t'}^2 + \frac{(t_0-t)}{2} \\
    &\quad\: + \left(\frac{T - t +1}{2}\right)\log\left[\frac{\sum_{t'=t}^T y_{t'}^2}{T-t+1} - \left(\frac{\sum_{t'=t}^T y_{t'}}{T-t+1}\right)^2 \right] \\
    &\quad\: -\left(\frac{T - t_0 +1}{2}\right)\log\left[\frac{\sum_{t'=t_0}^T y_{t'}^2}{T-t_0+1} - \left(\frac{\sum_{t'=t_0}^T y_{t'}}{T-t_0+1}\right)^2 \right] \\
    &\quad\: + \mathcal{O}(1).
\end{align*}
After standardizing terms we have:
\small
\begin{align*}
    \sum_{t' = t}^{t_0 - 1} y_{t'}^2 &= (t_0-t) + \sum_{t' = t}^{t_0 - 1} (z_{t'}^2-1), \\
    \frac{\sum_{t'=t_0}^T y_{t'}^2}{T-t_0+1} - \left(\frac{\sum_{t'=t_0}^T y_{t'}}{T-t_0+1}\right)^2
    &= s_0^2 \left[1 + \frac{\sum_{t'=t_0}^T (z_{t'}^2-1)}{T-t_0+1} - \left(\frac{\sum_{t'=t_0}^T z_{t'}}{T-t_0+1}\right)^2\right].
\end{align*}
\normalsize
Therefore, we can write:
\small
\begin{align*}
    \log \frac{\alpha_{t_0}}{\alpha_t} &> \log \left(\frac{T - t + 1}{T - t_0 +1} \right)- \frac{1}{2} \sum_{t' = t}^{t_0 - 1} (z_{t'}^2-1) - \frac{(T-t_0+1)}{2}\log s_0^2 \notag  \\
    &\quad\: + \left(\frac{T - t +1}{2}\right)\log\left[\frac{\sum_{t'=t}^T y_{t'}^2}{T-t+1} - \left(\frac{\sum_{t'=t}^T y_{t'}}{T-t+1}\right)^2 \right] \notag \\
    &\quad\: -\left(\frac{T - t_0 +1}{2}\right)\log \left[1 + \frac{\sum_{t'=t_0}^T (z_{t'}^2-1)}{T-t_0+1} - \left(\frac{\sum_{t'=t_0}^T z_{t'}}{T-t_0+1}\right)^2\right] \notag \\
    &\quad\: + \mathcal{O}(1).
\end{align*}
\normalsize
By the same argument we used to show (\ref{eq:neg1-bd}), with probability one we have:
\small
\begin{align*}
    -\left(\frac{T - t_0 +1}{2}\right)\log \left[1 + \frac{\sum_{t'=t_0}^T (z_{t'}^2-1)}{T-t_0+1} - \left(\frac{\sum_{t'=t_0}^T z_{t'}}{T-t_0+1}\right)^2\right] \geq -\frac{1}{2}\sum_{t'=t_0}^T (z_{t'}^2-1) +\frac{1}{2(T-t_0+1)} \left(\sum_{t'=t_0}^T z_{t'}\right)^2
\end{align*}
\normalsize
and since $t < t_0$ implies $\log\left(\frac{T - t + 1}{T - t_0 +1} \right) > 0$, we have:
\small
\begin{align}
    \log \frac{\alpha_{t_0}}{\alpha_t} &> - \frac{1}{2} \sum_{t' = t}^{t_0 - 1} (z_{t'}^2-1) - \frac{(T-t_0+1)}{2}\log s_0^2 \notag  \\
    &\quad\: + \left(\frac{T - t +1}{2}\right)\log\left[\frac{\sum_{t'=t}^T y_{t'}^2}{T-t+1} - \left(\frac{\sum_{t'=t}^T y_{t'}}{T-t+1}\right)^2 \right] \notag \\
    &\quad\: -\frac{1}{2}\sum_{t'=t_0}^T (z_{t'}^2-1) +\frac{1}{2(T-t_0+1)} \left(\sum_{t'=t_0}^T z_{t'}\right)^2 \notag \\
    &\quad\: + \mathcal{O}(1). \label{eq:thm3-cs2-bd1}
\end{align}
\normalsize
For notational convenience, we can define the error term:
\footnotesize
\begin{align}
    \xi(T,t_0,t,s_0) &:= \frac{\sum_{t'=t}^{t_0-1} (z_{t'}^2 -1)}{T-t+1} - \left(\frac{\sum_{t'=t}^{t_0-1} z_{t'}}{T-t+1}\right)^2 + s_0^2\left[\frac{\sum_{t'=t_0}^{T} (z_{t'}^2-1)}{T-t+1} - \left(\frac{\sum_{t'=t_0}^{T} z_{t'}}{T-t+1}\right)^2\right] - \frac{2s_0\sum_{t'=t}^{t_0-1} z_{t'}\sum_{t'=t_0}^{T} z_{t'}}{(T-t+1)^2} \label{eq:thm3-error}
\end{align}
\normalsize
then we can write:
\small
\begin{align*}
    \frac{\sum_{t'=t}^T y_{t'}^2}{T-t+1} - \left(\frac{\sum_{t'=t}^T y_{t'}}{T-t+1}\right)^2 &= \frac{\sum_{t'=t}^{t_0-1} z_{t'}^2 + \sum_{t'=t_0}^{T} (s_0z_{t'} + b_0)^2}{T-t+1} - \left[\frac{\sum_{t'=t}^{t_0-1} z_{t'} + \sum_{t'=t_0}^{T} (s_0z_{t'} + b_0)}{T-t+1}\right]^2 \notag \\
    &= 1 + \frac{(T-t_0+1)(t_0-t)b_0^2}{(T-t+1)^2} + \frac{(T-t_0+1)(s_0^2-1)}{T-t+1} \notag\\
    &\quad + \frac{\sum_{t'=t}^{t_0-1} (z_{t'}^2 -1)}{T-t+1} - \left(\frac{\sum_{t'=t}^{t_0-1} z_{t'}}{T-t+1}\right)^2 +s_0^2\left[\frac{\sum_{t'=t_0}^{T} (z_{t'}^2-1)}{T-t+1} - \left(\frac{\sum_{t'=t_0}^{T} z_{t'}}{T-t+1}\right)^2\right] \notag\\
    &\quad + \frac{2}{(T-t+1)^2}\left[(t_0-t)b_0s_0\sum_{t'=t_0}^Tz_{t'} - s_0\sum_{t'=t}^{t_0-1} z_{t'}\sum_{t'=t_0}^{T} z_{t'} - (T-t_0+1)b_0\sum_{t'=t}^{t_0-1} z_{t'}\right]  \\
    &= 1 + \frac{(T-t_0+1)(t_0-t)b_0^2}{(T-t+1)^2} + \frac{(T-t_0+1)(s_0^2-1)}{T-t+1} \\
    &\quad + \frac{2}{(T-t+1)^2}\left[(t_0-t)b_0s_0\sum_{t'=t_0}^Tz_{t'} - (T-t_0+1)b_0\sum_{t'=t}^{t_0-1} z_{t'}\right] + \xi(T,t_0,t,s_0)
\end{align*}
\normalsize
On $\mathcal{E}$, for some constant $C_1>0$ we have:
\begin{align*}
    \left|\sum_{t'=t}^{t_0-1} z_{t'}\right| &\leq C_1\sqrt{(t_0-t)\log T } \tag{on $\mathcal{E}_1\subset\mathcal{E}$} \\
    \left|b_0\sum_{t'=t_0}^T z_{t'}\right| &\leq \frac{(T-t_0+1)b_0^2 }{4} \tag{on $\mathcal{E}_8\subset\mathcal{E}$}
\end{align*}
and thus:
\small
\begin{align*}
    \frac{\sum_{t'=t}^T y_{t'}^2}{T-t+1} - \left(\frac{\sum_{t'=t}^T y_{t'}}{T-t+1}\right)^2 &\geq 1 + \frac{(T-t_0+1)(t_0-t)b_0^2}{2(T-t+1)^2} + \frac{(T-t_0+1)(s_0^2-1)}{T-t+1} \\
    &\quad - \frac{2(T-t_0+1)|b_0|C_1\sqrt{(t_0-t)\log T}}{(T-t+1)^2}  + \xi(T,t_0,t,s_0)
\end{align*}
\normalsize
Returning to bound (\ref{eq:thm3-cs2-bd1}), we now have:
\begin{align}
    \log \frac{\alpha_{t_0}}{\alpha_t} &> - \frac{1}{2} \sum_{t' = t}^{t_0 - 1} (z_{t'}^2-1) - \frac{(T-t_0+1)}{2}\log s_0^2  + \left(\frac{T - t +1}{2}\right)\log\left[1 + \frac{\frac{(T-t_0+1)(s_0^2-1)}{T-t+1}}{1 + \xi(T,t_0,t,s_0)}\right] \notag  \\
    &\quad\: + \left(\frac{T - t +1}{2}\right)\log\left[1 + \frac{\frac{(T-t_0+1)(t_0-t)b_0^2 }{2(T-t+1)^2} - \frac{2(T-t_0+1)|b_0|C_1\sqrt{(t_0-t)\log T} }{(T-t+1)^2}}{1 + \frac{(T-t_0+1)(s_0^2-1)}{T-t+1} + \xi(T,t_0,t,s_0)}\right] \notag \\
    &\quad\: -\frac{1}{2}\sum_{t'=t_0}^T (z_{t'}^2-1) +\frac{1}{2(T-t_0+1)} \left(\sum_{t'=t_0}^T z_{t'}\right)^2 + \left(\frac{T - t +1}{2}\right)\log\left[1 + \xi(T,t_0,t,s_0)\right] \notag \\
    &\quad\: + \mathcal{O}(1). \label{eq:thm3-cs2-bd2}
\end{align}

We now treat the case of $t_0 - t \geq T-t_0+1$ and $t_0-t < T-t_0+1$ separately. 

\subsubsection*{Case 2.a: $t_0 - t \leq T-t_0+1$.}
We can start by rewriting bound (\ref{eq:thm3-cs2-bd2}) as:
\begin{align}
    \log \frac{\alpha_{t_0}}{\alpha_t} &> - \frac{1}{2} \sum_{t' = t}^{t_0 - 1} (z_{t'}^2-1) + \frac{(t_0-t)}{2}\log s_0^2  + \left(\frac{T - t +1}{2}\right)\log\left[1 + \frac{\frac{(t_0-t)(s_0^{-2}-1)}{T-t+1}}{1 + s_0^{-2}\xi(T,t_0,t,s_0)}\right] \notag  \\
    &\quad\: + \left(\frac{T - t +1}{2}\right)\log\left[1 + \frac{\frac{(T-t_0+1)(t_0-t)b_0^2 }{2(T-t+1)^2} - \frac{2(T-t_0+1)|b_0|C_1\sqrt{(t_0-t)\log T} }{(T-t+1)^2}}{1 + \frac{(T-t_0+1)(s_0^2-1)}{T-t+1} + \xi(T,t_0,t,s_0)}\right] \notag \\
    &\quad\: -\frac{1}{2}\sum_{t'=t_0}^T (z_{t'}^2-1) +\frac{1}{2(T-t_0+1)} \left(\sum_{t'=t_0}^T z_{t'}\right)^2 + \left(\frac{T - t +1}{2}\right)\log\left[1 + s_0^{-2}\xi(T,t_0,t,s_0)\right] \notag \\
    &\quad\: + \mathcal{O}(1). \label{eq:thm3-cs2-a-bd1}
\end{align}
Note that if $s_0^2=1$, then:
\begin{align*}
    s_0^{-2}\xi(T,t_0,t,s_0) =  \frac{\sum_{t'=t}^{T} (z_{t'}^2 -1)}{T-t+1} - \left(\frac{\sum_{t'=t}^{T} z_{t'}}{T-t+1}\right)^2
\end{align*}
and on $\mathcal{E}$:
\begin{align*}
    \left|\frac{\sum_{t'=t}^T z_{t'}}{T-t+1}\right| &\leq \frac{1}{2}, \tag{on $\mathcal{E}_3\subset\mathcal{E}$} \\
    \left| \frac{\sum_{t'=t}^T (z_{t'}^2-1)}{T-t+1}\right| &\leq \frac{1}{4}, \tag{on $\mathcal{E}_6\subset\mathcal{E}$} \\
\end{align*}
so $ \left|s_0^{-2}\xi(T,t_0,t,s_0)\right| \leq 1/2$. On the other hand, if $s_0^2\neq 1$, then since $(s_0^2-1)^2 = \mathcal{O}(1)$, Assumption \ref{assumption:scale} implies $T-t_0+1 \gg \log T$. On $\mathcal{E}$ we have: 
\begin{align*}
    \max\left\{\left|\sum_{t'=t_0}^{T} (z_{t'}^2 -1)\right|, \left|\sum_{t'=t}^{t_0-1} z_{t'}\right|, \left|\sum_{t'=t_0}^{T} z_{t'} \right|\right\} &\leq C_1\sqrt{(T-t+1)\log T} \tag{on $\mathcal{E}_1\cap\mathcal{E}_2\cap\mathcal{E}_7\subset\mathcal{E}$}
\end{align*}
Additionally since $T-t+1\geq T-t_0+1\gtrsim\log T$, on $\mathcal{E}_5$ for large for some $C_2 > C_1$ we can write:
\begin{align*}
    \left|\sum_{t'=t}^{t_0-1} (z_{t'}^2 -1)\right| \leq C_1 \sqrt{(t_0-1)\log T} + C_1\log T &\leq C_2 \sqrt{(T-t+1)\log T}. \tag{$T-t+1\geq t_0-t$ and $T-t+1\gtrsim \log T$}
\end{align*}
Then on $\mathcal{E}$:
\begin{align*}
    \left|s_0^{-2}\xi(T,t_0,t,s_0)\right| &\leq (1+s_0^{-2} + 2s_0^{-1})\left[\frac{C_2\sqrt{\log T}}{\sqrt{T-t+1}} + \frac{C_2^2\log T}{T-t+1}\right] \\
    &\leq8\underline{s}^2C_2^2\left[\frac{\sqrt{\log T}}{\sqrt{T-t_0+1}} + \frac{\log T}{T-t_0+1}\right]. \tag{$s_0^2 > \underline{s}^2$ and $T-t+1\geq T-t_0+1$}
\end{align*}
So $\left|s_0^{-2}\xi(T,t_0,t,s_0)\right| \to 0$ and we can use the inequality $\log(1+x) > x/(x+1)$ for $x > -1$ to write:
\footnotesize
\begin{align*}
    \left(\frac{T - t +1}{2}\right)\log\left[1 + s_0^{-2}\xi(T,t_0,t,s_0)\right] &\geq \left(\frac{T - t +1}{2}\right) \left[s_0^{-2}\xi(T,t_0,t,s_0) - \frac{s_0^{-4}\xi(T,t_0,t,s_0)^2}{1+s_0^{-2}\xi(T,t_0,t,s_0)}\right] \\
    &= \frac{1}{2s_0^2}\sum_{t'=t}^{t_0-1}(z_{t'}^2 - 1) - \frac{\left(\sum_{t'=t}^{t_0-1}z_{t'}\right)^2 }{2s_0^2(T-t+1)} + \frac{1}{2}\sum_{t'=t_0}^T (z_{t'}^2 - 1) -  \frac{\left(\sum_{t'=t_0}^{T}z_{t'}\right)^2 }{2(T-t+1)}  \\
    &\quad\; - \frac{\sum_{t'=t_0}^{T}z_{t'} \sum_{t'=t}^{t_0-1}z_{t'}}{s_0(T-t+1)} - \left(\frac{T - t +1}{2}\right)\frac{s_0^{-4}\xi(T,t_0,t,s_0)^2}{1+s_0^{-2}\xi(T,t_0,t,s_0)}.
\end{align*}
\normalsize
We showed above that if $s_0^2=1$, then $\left|s_0^{-2}\xi(T,t_0,t,s_0)\right| \leq 1/2$ on $\mathcal{E}$. Furthermore, since:
\begin{align*}
    \left|\sum_{t'=t}^T z_{t'}\right| &< C_1^{1/2}(T-t+1)^{3/4} \log^{1/4} T \tag{on $\mathcal{E}_4\subset\mathcal{E}$} \\
    \left|\sum_{t'=t}^T (z_{t'}^2-1)\right| &< C_1\sqrt{(T-t+1)\log T}, \tag{on $\mathcal{E}_7\subset\mathcal{E}$}
\end{align*}
then for this case on $\mathcal{E}$:
\begin{align*}
    s_0^{-4}\xi(T,t_0,t,s_0)^2 = \left[\frac{\sum_{t'=t}^{T} (z_{t'}^2 -1)}{T-t+1} - \left(\frac{\sum_{t'=t}^{T} z_{t'}}{T-t+1}\right)^2\right]^2 &\leq \frac{4C_1^2\log T}{T-t+1}
\end{align*}
and thus:
\begin{align*}
    \left(\frac{T - t +1}{2}\right)\frac{s_0^{-4}\xi(T,t_0,t,s_0)^2}{1+s_0^{-2}\xi(T,t_0,t,s_0)} &\leq 4C_1^2\log T.
\end{align*}
On the other hand, when $s_0^2\neq 1$, we showed above that on $\mathcal{E}$:
\begin{align*}
    \left|s_0^{-2}\xi(T,t_0,t,s_0)\right| &\leq 16\underline{s}^2C_2^2\sqrt{\frac{\log T}{T-t+1}} \to 0.
\end{align*}
So for $T$ large enough that $16\underline{s}^2C_2^2\sqrt{\frac{\log T}{T-t+1}} \leq 1/2$, we have:
\begin{align*}
    \left(\frac{T - t +1}{2}\right)\frac{s_0^{-4}\xi(T,t_0,t,s_0)^2}{1+s_0^{-2}\xi(T,t_0,t,s_0)} &\leq (16\underline{s}^2C_2^2)^2\log T.
\end{align*}
On $\mathcal{E}$ we also have:
\begin{align*}
    \frac{\left(\sum_{t'=t}^{t_0-1}z_{t'}\right)^2 }{2s_0^2(T-t+1)} &\leq \frac{C_1^2}{2\underline{s}^2} \log T \\
    \frac{\sum_{t'=t_0}^{T}z_{t'} \sum_{t'=t}^{t_0-1}z_{t'}}{s_0(T-t+1)} &\leq  \frac{C_1^2}{\underline{s}}\log T
\end{align*}
and, since $T-t+1\geq T-t_0+1$, we have:
\begin{align*}
    \frac{\left(\sum_{t'=t_0}^{T}z_{t'}\right)^2 }{2(T-t_0+1)} -  \frac{\left(\sum_{t'=t_0}^{T}z_{t'}\right)^2 }{2(T-t+1)} \geq 0
\end{align*}
so returning to bound (\ref{eq:thm3-cs2-a-bd1}) we have: 
\begin{align}
    \log \frac{\alpha_{t_0}}{\alpha_t} &>  \frac{(s_0^{-2}-1)}{2} \sum_{t' = t}^{t_0 - 1} (z_{t'}^2-1) + \frac{(t_0-t)}{2}\log s_0^2 + \left(\frac{T - t +1}{2}\right)\log\left[1 + \frac{\frac{(t_0-t)(s_0^{-2}-1)}{T-t+1}}{1 + s_0^{-2}\xi(T,t_0,t,s_0)}\right] \notag  \\
    &\quad\: + \left(\frac{T - t +1}{2}\right)\log\left[1 + \frac{\frac{(T-t_0+1)(t_0-t)b_0^2 }{2(T-t+1)^2} - \frac{2(T-t_0+1)|b_0|C_1\sqrt{(t_0-t)\log T} }{(T-t+1)^2}}{1 + \frac{(T-t_0+1)(s_0^2-1)}{T-t+1} + \xi(T,t_0,t,s_0)}\right] \notag \\
    &\quad\: + \mathcal{O}(\log T). \label{eq:thm3-cs2-a-bd2}
\end{align}
Next we can write:
\footnotesize
\begin{align*}
    \left(\frac{T - t +1}{2}\right)\log\left[1 + \frac{\frac{(t_0-t)(s_0^{-2}-1)}{T-t+1}}{1 + s_0^{-2}\xi(T,t_0,t,s_0)}\right]&= \left(\frac{T - t +1}{2}\right)\log\left[1 + \frac{(t_0-t)(s_0^{-2}-1)}{T-t+1}\right] \\
    &\quad\:+ \left(\frac{T - t +1}{2}\right)\log\left[1 - \left(\frac{ s_0^{-2}\xi(T,t_0,t,s_0)}{1 + s_0^{-2}\xi(T,t_0,t,s_0)}\right)\left(\frac{\frac{(t_0-t)(s_0^{-2}-1)}{T-t+1}}{1+\frac{(t_0-t)(s_0^{-2}-1)}{T-t+1}}\right)\right]
\end{align*}
\normalsize
and have just shown that:
\begin{align*}
    \frac{ s_0^{-2}\xi(T,t_0,t,s_0)}{1 + s_0^{-2}\xi(T,t_0,t,s_0)} &= \mathcal{O}\left(\frac{\sqrt{\log T}}{\sqrt{T-t_0+1}}\right), \\
    \frac{\frac{(t_0-t)(s_0^{-2}-1)}{T-t+1}}{1+\frac{(t_0-t)(s_0^{-2}-1)}{T-t+1}} &= \mathcal{O}\left(\frac{(t_0-t)(s_0^{-2}-1)}{T-t+1}\right)= \mathcal{O}(1). \tag{$s_0^{-2} = \mathcal{O}(1)$}
\end{align*}
Note that if $s_0^2 = 1$, then we can completely ignore the first line of (\ref{eq:thm3-cs2-a-bd2}), on the other hand, if $s_0^2\neq 1$, then we showed above that $\left|s_0^{-2}\xi(T,t_0,t,s_0)\right| \to 0$ on $\mathcal{E}$, so for some constant $C_3 > 0$, a first order Taylor approximation of $\log(1+x)$ around zero gives:
\footnotesize
\begin{align*}
     \left(\frac{T - t +1}{2}\right)\log\left[1 - \left(\frac{ \xi(T,t_0,t,s_0)}{1 + \xi(T,t_0,t,s_0)}\right)\left(\frac{\frac{(t_0-t)(s_0^{-2}-1)}{T-t+1}}{1+\frac{(t_0-t)(s_0^{-2}-1)}{T-t_0+1}}\right)\right] &= (t_0-t)(s_0^{-2}-1)\mathcal{O}\left(\sqrt{\frac{\log T}{T-t_0+1}}\right) + \mathcal{O}(\log T) \\
     &= (s_0^{-2}-1)\mathcal{O}\left(\sqrt{(t_0-t)\log T}\right) + \mathcal{O}(\log T) \tag{$t_0-t\leq T-t_0+1$} \\
     &\geq -\frac{C_3|s_0^{-2}-1|}{2}\sqrt{(t_0-t)\log T} + \mathcal{O}(\log T)
\end{align*}
\normalsize
We also have:
\begin{align*}
    \left|\sum_{t' = t}^{t_0 - 1} (z_{t'}^2-1)\right| < C_1\sqrt{(t_0-t)\log T} +\log T \tag{on $\mathcal{E}_5\subset\mathcal{E}$}
\end{align*}
and since $C_1 > C_3$ and $C_1|s_0^{-2}-1|\log T \leq C_1(\underline{s}^{-2}+1)\log T$ WLOG, then we can now lower bound (\ref{eq:thm3-cs2-a-bd2}) with:
\begin{align}
    \log \frac{\alpha_{t_0}}{\alpha_t} &>  -C_1|s_0^{-2}-1|\sqrt{(t_0-t)\log T} + \frac{(t_0-t)}{2}\log s_0^2 + \left(\frac{T - t +1}{2}\right)\log\left[1 + \frac{(t_0-t)(s_0^{-2}-1)}{T-t+1}\right] \notag  \\
    &\quad\: + \left(\frac{T - t +1}{2}\right)\log\left[1 + \frac{\frac{(T-t_0+1)(t_0-t)b_0^2 }{2(T-t+1)^2} - \frac{2(T-t_0+1)|b_0|C_1\sqrt{(t_0-t)\log T}}{(T-t+1)^2}}{1 + \frac{(T-t_0+1)(s_0^2-1)}{T-t+1} + \xi(T,t_0,t,s_0)}\right] \notag \\
    &\quad\: + \mathcal{O}(\log T). \label{eq:thm3-cs2-a-bd3}
\end{align}

We now split our analysis into the case of $\min\left\{b_0^2,b_0^2/s_0^2\right\} \geq (s_0^2-1)^2$ and $\min\left\{b_0^2,b_0^2/s_0^2\right\} < (s_0^2-1)^2$. However in both cases, we are assuming that $\max\{\min\left\{b_0^2,b_0^2/s_0^2\right\},(s_0^2-1)^2\} >0$ (i.e. $\mathcal{T}_\beta$ is non-empty for large enough $T$), and when $\min\left\{b_0^2,b_0^2/s_0^2\right\} >0$ we assume Assumption \ref{assumption:mean} holds, while when $(s_0^2-1)^2 >0$ we assume  Assumption \ref{assumption:scale} holds, so it is without loss of generality to assume Assumption \ref{assumption:mean} holds when $\min\left\{b_0^2,b_0^2/s_0^2\right\} \geq (s_0^2-1)^2$ and that Assumption \ref{assumption:scale} holds when  $\min\left\{b_0^2,b_0^2/s_0^2\right\} < (s_0^2-1)^2$.

\subsubsection*{Case 2.a.i: $\min\left\{b_0^2,b_0^2/s_0^2\right\} \geq (s_0^2-1)^2$}

Suppose that:
\begin{align*}
    \min\left\{b_0^2,\frac{b_0^2}{s_0^2}\right\} \geq (s_0^2-1)^2 &\implies |t-t_0| \geq \frac{C_{a,1} \log T}{\min\{b_0^2,b_0^2/s_0^{2}\}}, \;\forall\; t\in\mathcal{T}_{\beta}.
\end{align*}
For $C_{1,\beta} > 64C_1^2$, we can write:
\small
\begin{align*}
    \frac{(T-t_0+1)(t_0-t)b_0^2 }{2(T-t+1)^2} - \frac{2(T-t_0+1)|b_0|C_1\sqrt{(t_0-t)\log T} }{(T-t+1)^2} &= \frac{(T-t_0+1)(t_0-t)b_0^2 }{(T-t+1)^2} \left[\frac{1}{2} - \frac{2C_1\sqrt{\log T}}{|b_0|\sqrt{t_0-t}}\right] \\
    &\geq \frac{(T-t_0+1)(t_0-t)b_0^2 }{(T-t+1)^2} \left[\frac{1}{2} - \frac{2C_1}{\sqrt{C_{1,\beta}}}\right] \tag{$|t-t_0| \geq \frac{C_{1,\beta} \log T}{\min\{b_0^2,b_0^2/s_0^{2}\}}$} \\
    &> \frac{(T-t_0+1)(t_0-t)b_0^2 }{4(T-t+1)^2} . \tag{$C_{1,\beta} > 64C_1^2$}
\end{align*}
\normalsize
Furthermore, we have shown that $\left|\xi(T,t_0,t,s_0)\right| \leq 1$ for large $T$ on $\mathcal{E}$, and we have:
\begin{align*}
    -1 <\underline{s}^2 - 1 <\frac{(T-t_0+1)(s_0^2-1)}{T-t+1} < \overline{s}^2 -1 < \infty
\end{align*}
so we can replace (\ref{eq:thm3-cs2-a-bd3}) with the deterministic bound: 
\begin{align*}
    \log \frac{\alpha_{t_0}}{\alpha_t} &>  -C_1|s_0^{-2}-1|\sqrt{(t_0-t)\log T} + \frac{(t_0-t)}{2}\log s_0^2 + \left(\frac{T - t +1}{2}\right)\log\left[1 + \frac{(t_0-t)(s_0^{-2}-1)}{T-t+1}\right] \notag  \\
    &\quad\: + \left(\frac{T - t +1}{2}\right)\log\left[1 + \frac{(T-t_0+1)(t_0-t)b_0^2 }{16\overline{s}^2(T-t+1)^2}\right] + \mathcal{O}(\log T)
\end{align*}
Note that the term:
\begin{align*}
    \frac{(t_0-t)}{2}\log s_0^2 + \left(\frac{T - t +1}{2}\right)\log\left[1 + \frac{(t_0-t)(s_0^{-2}-1)}{T-t+1}\right]
\end{align*}
is decreasing in $t$ and is clearly equal to zero when $t=t_0$. To see this we can take the derivative with respect $t$ to get:
\begin{align*}
    \frac{1}{2}\left[ \frac{\frac{(T-t_0+1)(s_0^2 -1)}{T-t+1}}{1 + \frac{(T-t_0+1)(s_0^2 -1)}{T-t+1}}- \log\left[1 + \frac{(T-t_0+1)(s_0^2 -1)}{T-t+1}\right] \right] &\leq 0. \tag{$\log(1+x) > x/(1+x)$ for $x>-1$}
\end{align*}
Therefore:
\small
\begin{align}
    \log \frac{\alpha_{t_0}}{\alpha_t} &>  -C_1|s_0^{-2}-1|\sqrt{(t_0-t)\log T} + \left(\frac{T - t +1}{2}\right)\log\left[1 + \frac{(T-t_0+1)(t_0-t)b_0^2 }{16\overline{s}^2(T-t+1)^2}\right] + \mathcal{O}(\log T) \label{eq:thm3-cs2-a-i-bd1}
\end{align}
\normalsize
Suppose first that $\lim_{T\to\infty}b_0^2 = 0$, then for large $T$:
\footnotesize
\begin{align*}
    \lim_{T\to\infty}\frac{(T-t_0+1)(t_0-t)b_0^2 }{16\overline{s}^2(T-t+1)^2} =0 \implies  \log\left[1 + \frac{(T-t_0+1)(t_0-t)b_0^2 }{16\overline{s}^2(T-t+1)^2}\right] &= \frac{(T-t_0+1)(t_0-t)b_0^2 }{16\overline{s}^2(T-t+1)^2} + \mathcal{O}\left(\frac{(T-t_0+1)(t_0-t)^2b_0^4}{(T-t+1)^2}\right) \\
    &\geq \frac{(T-t_0+1)(t_0-t)b_0^2 }{32\overline{s}^2(T-t+1)^2}.
\end{align*}
\normalsize
So for large $T$, we can lower bound (\ref{eq:thm3-cs2-a-i-bd1}) in this case with:
\begin{align*}
     \log \frac{\alpha_{t_0}}{\alpha_t} &>  -C_1|s_0^{-2}-1|\sqrt{(t_0-t)\log T} +\frac{(T-t_0+1)(t_0-t)b_0^2 }{64\overline{s}^2(T-t+1)} + \mathcal{O}(\log T) \\
     &\geq (t_0-t)b_0^2\left[\frac{(T-t_0+1)}{64\overline{s}^2(T-t+1)} -\frac{C_1s_0^{-2} \min\{1,s_0^{-1}\}\sqrt{\log T}}{|b_0|\sqrt{t_0-t}}\right] + \mathcal{O}(\log T) \tag{$(s_0^2-1)^2 \leq\min\{b_0^2,b_0^2/s_0^2\}$} \\
     &\geq (t_0-t)b_0^2\left[\frac{1}{128\overline{s}^2} -\frac{C_1\underline{s}^{-2}}{\sqrt{C_{1,\beta}}}\right] + \mathcal{O}(\log T) \tag{$|t_0-t| \geq \frac{C_{1,\beta}\log^2 T}{\min\{b_0^2,b_0^2/s_0^2\}}$}
\end{align*}
so for $C_{1,\beta} > (128C_1 \overline{s}^2/\underline{s}^2)^2$ and some $B_1 > 0$, we have:
\begin{align*}
     \log \frac{\alpha_{t_0}}{\alpha_t} &> (t_0-t)b_0^2 + \mathcal{O}(\log T) \\
     &\geq (C_{1,\beta} - B_1)\log T  \tag{$|t_0-t| \geq \frac{C_{1,\beta}\log T}{\min\{b_0^2,b_0^2/s_0^2\}}$}
\end{align*}
and thus (\ref{eq:thm3-result}) holds for $C_{1,\beta} > B_1+C_\pi$.

Next, if $b_0^2 > \underline{b}^2$ for some $\underline{b}^2 > 0$, then we can show that the bound (\ref{eq:thm3-cs2-a-i-bd1}) is decreasing in $t$. We verify below that the second term is decreasing in $t$, and the first term disappears when $s_0^2=1$, so it is without loss of generality to assume that $s_0^2\neq 1$. Taking the derivative with respect to $t$ yields:
\small
\begin{align}
    &\frac{1}{2}\left[\frac{C_1|s_0^{-2}-1|\sqrt{\log T}}{\sqrt{t_0-t}} - \log\left[1 + \frac{(T-t_0+1)(t_0-t)b_0^2 }{16\overline{s}^2(T-t+1)^2}\right] +\frac{\frac{(T-t_0+1)(t_0-t)b_0^2 }{8\overline{s}^2(T-t+1)^2} - \frac{(T-t_0+1)b_0^2 }{16\overline{s}^2(T-t+1)}}{1 + \frac{(T-t_0+1)(t_0-t)b_0^2 }{16\overline{s}^2(T-t+1)^2}}\right] \label{eq:thm3-case-2-i-deriv}
\end{align}
\normalsize
Suppose that $b_0^2\to\infty$ as $T\to\infty$, then since we can write:
\begin{align*}
     - \log\left[1 + \frac{(T-t_0+1)(t_0-t)b_0^2 }{16\overline{s}^2(T-t+1)^2}\right] +\frac{\frac{(T-t_0+1)(t_0-t)b_0^2 }{8\overline{s}^2(T-t+1)^2} - \frac{(T-t_0+1)b_0^2 }{16\overline{s}^2(T-t+1)}}{1 + \frac{(T-t_0+1)(t_0-t)b_0^2 }{16\overline{s}^2(T-t+1)^2}} \leq -\frac{\frac{(T-t_0+1)^2b_0^2 }{16\overline{s}^2(T-t+1)^2}}{1 + \frac{(T-t_0+1)(t_0-t)b_0^2 }{16\overline{s}^2(T-t+1)^2}} \tag{$\log(1+x)\geq x/(1+x)$ for $x>-1$}
\end{align*}
and:
\begin{align*}
    \frac{(T-t_0+1)^2b_0^2 }{(T-t+1)^2} = \frac{b_0^2}{\left(1+\frac{t_0-t}{T-t_0+1}\right)^2} \geq \frac{b_0^2}{4}. \tag{$t_0-t \leq T-t_0+1$}
\end{align*}
Then for large $T$: 
\begin{align*}
    - \log\left[1 + \frac{(T-t_0+1)(t_0-t)b_0^2 }{16\overline{s}^2(T-t+1)^2}\right] +\frac{\frac{(T-t_0+1)(t_0-t)b_0^2 }{8\overline{s}^2(T-t+1)^2} - \frac{(T-t_0+1)b_0^2 }{16\overline{s}^2(T-t+1)}}{1 + \frac{(T-t_0+1)(t_0-t)b_0^2 }{16\overline{s}^2(T-t+1)^2}} \lesssim -\frac{T-t_0+1}{t_0-t}.
\end{align*}
At the same time we have:
\begin{align*}
    \frac{\frac{T-t_0+1}{t_0-t}}{\frac{C_1|s_0^{-2}-1|\sqrt{\log T}}{\sqrt{t_0-t}}} \geq \frac{\sqrt{T-t_0+1}}{C_1|\underline{s}^{-2}-1|\sqrt{\log T}}. \tag{$t_0-t \leq T-t_0+1$}
\end{align*}
Since $s_0^2\neq0$, then again $T-t_0+1 \gg \log T$ and the lower bound above is diverging as $T\to\infty$, and thus the first term in (\ref{eq:thm3-case-2-i-deriv}) is dominated by the remaining terms and we get that (\ref{eq:thm3-case-2-i-deriv}) is negative for large $T$. 

On the other hand, if $b_0^2 =\mathcal{O}(1)$ so that $b_0^2 < \overline{b}^2$ for some $\overline{b}^2 < \infty$, then since the first term of (\ref{eq:thm3-case-2-i-deriv}) is decreasing in $t$, we have:
\begin{align*}
    \frac{C_1|s_0^{-2}-1|\sqrt{\log T}}{\sqrt{t_0-t}}  &\leq \frac{C_1\min\{|b_0|,|b_0|/s_0\}\sqrt{\log T}}{s_0^{2}\sqrt{t_0-t}}  \tag{$(s_0^2-1)^2 \leq\min\{b_0^2,b_0^2/s_0^2\}$} \\
    &\leq \frac{C_1\min\{|b_0|,|b_0|/s_0\}\sqrt{\log T}}{s_0^{2}\sqrt{t_0-t}} \bigg|_{t=\sup\{t\in\mathcal{T}_{\beta}:t\leq t_0,\;t_0-t\leq T-t_0+1\}} \\
    &= \frac{C_1\min\{b_0^2,b_0^2/s^2_0\}}{s_0^{2}\sqrt{C_{1,\beta}}} \\
    &\leq \frac{C_1 \overline{b}^2 }{\underline{s}^4\sqrt{C_{1,\beta}}} \tag{$|b_0|\leq \overline{b} < \infty$}
\end{align*}
while defining:
\begin{align*}
     h(t):=\frac{(T-t_0+1)(t_0-t)b_0^2 }{16\overline{s}^2(T-t+1)}
\end{align*}
and taking the derivative of the other terms with respect to $t$ yields: 
\begin{align*}
    \frac{h''(t)}{1+ \frac{h(t)}{T-t+1}} - \frac{\left[\frac{h(t)}{T-t+1} + h'(t)\right]^2}{(T-t+1)\left[1+\frac{h(t)}{T-t+1}\right]^2}. 
\end{align*}
Since:
\begin{align*}
    h''(t) &= -\frac{(T-t_0+1)^2b_0^2 }{8\overline{s}^2(T-t+1)^3} < 0
\end{align*}
then:
\small
\begin{align*}
    &- \log\left[1 + \frac{(T-t_0+1)(t_0-t)b_0^2 }{16\overline{s}^2(T-t+1)^2}\right] +\frac{\frac{(T-t_0+1)(t_0-t)b_0^2 }{8\overline{s}^2(T-t+1)^2} - \frac{(T-t_0+1)b_0^2 }{16\overline{s}^2(T-t+1)}}{1 + \frac{(T-t_0+1)(t_0-t)b_0^2 }{16\overline{s}^2(T-t+1)^2}} \\
    &\leq - \log\left[1 + \frac{(T-t_0+1)(t_0-t)b_0^2 }{16\overline{s}^2(T-t+1)^2}\right] +\frac{\frac{(T-t_0+1)(t_0-t)b_0^2 }{8\overline{s}^2(T-t+1)^2} - \frac{(T-t_0+1)b_0^2 }{16\overline{s}^2(T-t+1)}}{1 + \frac{(T-t_0+1)(t_0-t)b_0^2 }{16\overline{s}^2(T-t+1)^2}}\bigg|_{t=\inf\{t\in\mathcal{T}_{\beta}:t\leq t_0,\;t_0-t\leq T-t_0+1\}} \\
    &= - \log\left[1 + \frac{b_0^2 }{32\overline{s}^2}\right] +\frac{\frac{b_0^2 }{32\overline{s}^2}}{1 + \frac{b_0^2 }{32\overline{s}^2}} \\
    &\leq - \log\left[1 + \frac{\underline{b}^2 }{32\overline{s}^2}\right] +\frac{\frac{\underline{b}^2 }{32\overline{s}^2}}{1 + \frac{\underline{b}^2 }{32\overline{s}^2}} \tag{$0<\underline{b} < |b_0|$}
\end{align*}
\normalsize
Once more, $\log(1+x) > x/(1+x)$ for $x>-1$, and since the last bound above is fixed with respect to $T$, then there is some $b^* > 0$ so that: 
\begin{align*}
     - \log\left[1 + \frac{\underline{b}^2 }{32\overline{s}^2}\right] +\frac{\frac{\underline{b}^2 }{32\overline{s}^2}}{1 + \frac{\underline{b}^2 }{32\overline{s}^2}} < -b^*.
\end{align*}
So if $C_{1,\beta} > (C_1 \overline{b}^2/ \underline{s}^4b^*)^2$, then (\ref{eq:thm3-case-2-i-deriv}) is negative and the bound (\ref{eq:thm3-cs2-a-i-bd1}) is minimized at:
\begin{align*}
    t &= \sup\;\{t \in \mathcal{T}_{\beta}:t<t_0\} \\
    &\leq t_0 - \frac{C_{1,\beta}\log T}{\min\{b_0^2,b_0^2/s_0^2\}}.
\end{align*}
and thus:
\begin{align*}
    \log \frac{\alpha_{t_0}}{\alpha_t} &> -\frac{C_1|s_0^2-1|\sqrt{C_{1,\beta}} \log T}{s_0^2\min\{|b_0|,|b_0|/s_0\}} \\
    &\quad\:+ \frac{T - t_0 +1}{2}\left(1 + \frac{C_{1,\beta}\log T}{\min\{b_0^2,b_0^2/s_0^2\}(T-t_0+1)}\right)\log\left[ 1 + \frac{\frac{C_{1,\beta}\log T}{16\overline{s}^2(T-t_0+1)\min\{1,s_0^{-2}\}}}{\left(1 + \frac{C_{1,\beta} \log T}{\min\{b_0^2,b_0^2/s_0^2\}(T-t_0+1)}\right)^2}\right] \\
    &\quad\:+ \mathcal{O}(\log T).
\end{align*}
By Assumption \ref{assumption:mean} we have:
\begin{align*}
    \min\{b_0^2,b_0^2/s_0^2\}(T-t_0+1) &\gg \log T
\end{align*}
so for any $C_{1,\beta} > 0$, for large $T$ we have:
\begin{align*}
    \frac{C_{1,\beta} \log T}{\min\{b_0^2,b_0^2/s_0^2\}(T-t_0+1)} \leq \sqrt{2}-1
\end{align*}
and:
\begin{align*}
    \log \frac{\alpha_{t_0}}{\alpha_t} &> -\frac{C_1|s_0^2-1|\sqrt{C_{1,\beta}}\log T}{s_0^2\min\{|b_0|,|b_0|/s_0\}}+\frac{T - t_0 +1}{2}\log\left[ 1 + \frac{C_{1,\beta}\log T}{32\overline{s}^2(T-t_0+1)}\right] + \mathcal{O}(\log T) 
\end{align*}
Suppose that $b_0^2 = \mathcal{O}(1)$, then it must be the case that $T-t_0+1 \gg \log T$, in which case we can use another first order Taylor approximation to write:
\begin{align*}
    \frac{T - t_0 +1}{2}\log\left[ 1 + \frac{C_{1,\beta}\log T}{32\overline{s}^2(T-t_0+1)}\right] = \frac{C_{1,\beta}\log T}{32\overline{s}^2} + \mathcal{O}\left(\log T\right)
\end{align*}
and since $(s_0^2 - 1)^2 \leq\min\{|b_0|,|b_0|/s_0\}$, for some $B_1 > 0$ we have:
\begin{align*}
    \log \frac{\alpha_{t_0}}{\alpha_t} &> -\frac{C_1\sqrt{C_{1,\beta}}\log T}{\underline{s}^2}+ \frac{C_{1,\beta}\log T}{32\overline{s}^2} - B_1 \log T 
\end{align*}
so we can select $C_{1,\beta} > 0$ large so that (\ref{eq:thm3-result}) holds. On the other hand, if $b_0^2 \to \infty$, then by Assumption \ref{assumption:mean}:
\begin{align*}
    T-t_0+1 \gtrsim \log T
\end{align*}
so we can write:
\begin{align*}
    \log \frac{\alpha_{t_0}}{\alpha_t} &> -\frac{C_1\overline{s}(\underline{s}^{-2} + 1)\sqrt{C_{1,\beta}}\log T}{|b_0|}+\frac{B_2\log T}{2}\log\left[ 1 + \frac{C_{1,\beta}}{32\overline{s}^2B_2}\right] -B_1 \log T 
\end{align*}
for some $B_2 > 0$. Since $|b_0| \to \infty$, for large $T$ we
\begin{align*}
    \frac{C_1\overline{s}(\underline{s}^{-2} + 1)}{|b_0|} \leq B_1
\end{align*}
and clearly we can pick $C_{1,\beta} > 0$ large so that:
\begin{align*}
    \log\left[ 1 + \frac{C_{1,\beta}}{32\overline{s}^2B_2}\right] > 2B_1 + C_\pi
\end{align*}
so again (\ref{eq:thm3-result}) holds for this case.

\subsubsection*{Case 2.a.ii: $\min\left\{b_0^2,b_0^2/s_0^2\right\} \leq (s_0^2-1)^2$}

Returning to bound (\ref{eq:thm3-cs2-a-bd3}), we can write:
\begin{align}
    \log \frac{\alpha_{t_0}}{\alpha_t} &>  -C_1|s_0^{-2}-1|\sqrt{(t_0-t)\log T} + \frac{(t_0-t)}{2}\log s_0^2 + \left(\frac{T - t +1}{2}\right)\log\left[1 + \frac{(t_0-t)(s_0^{-2}-1)}{T-t+1}\right] \notag  \\
    &\quad\: + \left(\frac{T - t +1}{2}\right)\log\left[1 - \frac{ \frac{2(T-t_0+1)|b_0|C_1\sqrt{(t_0-t)\log T} }{(T-t+1)^2}}{1 + \frac{(T-t_0+1)(s_0^2-1)}{T-t+1} + \xi(T,t_0,t,s_0)}\right] \notag \\
    &\quad\: + \mathcal{O}(\log T). \label{eq:thm3-cs2-a-ii-bd1}
\end{align}
Since $(s_0^2-1)^2 = \mathcal{O}(1)$ by assumption, then again Assumption \ref{assumption:scale} implies $ T-t_0+1 \gg \log T$. We have shown that on $\mathcal{E}$:
\begin{align*}
    \left|\xi(T,t_0,t,s_0)\right| &\leq (1+s_0^2 + 2s_0)\left[\frac{C_1\sqrt{\log T}}{\sqrt{T-t+1}} + \frac{C_1^2\log T}{T-t+1}\right]  \\
    &\leq  4 \overline{s}^2\left[\frac{C_1\sqrt{\log T}}{\sqrt{T-t_0+1}} + \frac{C_1^2\log T}{T-t_0+1}\right] 
\end{align*}
so for large $T$ we have $\xi(T,t_0,t,s_0) > - \underline{s}^2/2$, and we have:
\begin{align*}
    -1 <\underline{s}^2 - 1 <\frac{(T-t_0+1)(s_0^2-1)}{T-t+1} < \overline{s}^2 -1 < \infty
\end{align*}
so: 
\begin{align}
    \log \frac{\alpha_{t_0}}{\alpha_t} &>  -C_1|s_0^{-2}-1|\sqrt{\log T(t_0-t)} + \frac{(t_0-t)}{2}\log s_0^2 + \left(\frac{T - t +1}{2}\right)\log\left[1 + \frac{(t_0-t)(s_0^{-2}-1)}{T-t+1}\right] \notag  \\
    &\quad\: + \left(\frac{T - t +1}{2}\right)\log\left[1 - \frac{4(T-t_0+1)|b_0|C_1\sqrt{(t_0-t)\log T} }{\underline{s}^2(T-t+1)^2}\right] \notag \\
    &\quad\: + \mathcal{O}(\log T). \label{eq:thm3-cs2-a-ii-bd2}
\end{align}
Next, since $(s_0^2-1)^2 = \mathcal{O}(1)$ and $\min\{b_0^2,b_0^2/s_0^2\}\leq (s_0^2-1)^2$, then then there is some $\overline{b} > 0$ so that:
\begin{align*}
    \frac{4(T-t_0+1)|b_0|C_1\sqrt{(t_0-t)\log T} }{\underline{s}^2(T-t+1)^2} \leq \frac{4\overline{b}C_1\sqrt{\log T}}{\underline{s}^2\sqrt{T-t_0+1}}. \tag{$T-t+1 \geq T-t_0+1\geq t_0-t$}
\end{align*}
Since $T-t_0+1\gg \log T$ in this case, then the RHS above is converging to zero and we can use another first order Taylor approximation to write:
\footnotesize
\begin{align*}
    \left(\frac{T - t +1}{2}\right)\log\left[1 - \frac{4(T-t_0+1)|b_0|C_1\sqrt{(t_0-t)\log T}}{\underline{s}^2(T-t+1)^2}\right] &= -\frac{4(T-t_0+1)|b_0|C_1\sqrt{(t_0-t)\log T }}{\underline{s}^2(T-t+1)} + \mathcal{O}\left(\log T\right) \\
    &\geq -\frac{4\overline{s}C_1|s_0^2-1|}{\underline{s}^2} \sqrt{(t_0-t)\log T}\tag{$\min\{b_0^2,b_0^2/s_0^2\}\leq (s_0^2-1)^2$} + \mathcal{O}\left(\log T\right).
\end{align*}
\normalsize
So for some $C_2 > 8\overline{s}C_1/\underline{s}^2$, we can write:
\small
\begin{align*}
    \log \frac{\alpha_{t_0}}{\alpha_t} &>  -C_2|s_0^{2}-1|\sqrt{(t_0-t)\log T} + \frac{(t_0-t)}{2}\log s_0^2 + \left(\frac{T - t +1}{2}\right)\log\left[1 + \frac{(t_0-t)(s_0^{-2}-1)}{T-t+1}\right] +\mathcal{O}(\log T) \\
    &=h_T(t,t_0) 
\end{align*}
\normalsize
We now show that for any $t_0$, $h_T$ is decreasing in $t$. Analyzing the first derivative of $h_T$ we have:
\small
\begin{align*}
    \frac{\partial h_T}{\partial t} &= \frac{1}{2}\left[ C_2|s_0^{2}-1|\sqrt{\frac{\log T}{t_0-t}}- \log\left[1 + \frac{(s_0^2-1)(T-t_0+1)}{T-t+1}\right] + \frac{\frac{(T-t_0+1)(s_0^{2}-1)}{T-t+1}}{1 + \frac{(s_0^2-1)(T-t_0+1)}{T-t+1} }\right] \\
    &= \frac{1}{2}\left[ C_2|s_0^{2}-1|\sqrt{\frac{\log }{(T-t_0+1)x}}- \log\left[1 + \frac{(s_0^2-1)}{1 + x}\right] + \frac{\frac{(s_0^{2}-1)}{1+x}}{1 + \frac{(s_0^2-1)}{1 + x} }\right] \tag{$x = \frac{t_0-t}{T-t_0+1}$}
\end{align*}
\normalsize
Next, note that the term:
\begin{align*}
     g_1(x) := C_2|s_0^{2}-1|\sqrt{\frac{\log T}{(T-t_0+1)x}}
\end{align*}
is convex and decreasing as a function of $x$. On the other hand, if we define:
\begin{align*}
    g_2(x) &:= - \log\left[1 + \frac{(s_0^2-1)}{1 + x}\right] + \frac{\frac{(s_0^{2}-1)}{1+x}}{1 + \frac{(s_0^2-1)}{1 + x} }
\end{align*}
then we have:
\begin{align*}
    \frac{d g_2}{dx} = \frac{\frac{(s_0^{2}-1)^2}{(1+x)^3}}{\left(1 + \frac{(s_0^2-1)}{1 + x} \right)^2} > 0
\end{align*}
so this term is increasing as a function of $x$. Therefore:
\begin{align*}
    g_1(x) + g_2(x) &\leq g_1(x)\bigg|_{x = \inf\:\{x\in[0,1]\::\:t_0-t\leq T-t_0+1, \; t\in\mathcal{T}_{\beta}\} } +  g_2(x)\bigg|_{x = \sup\:\{x\in[0,1]\::\:t_0-t\leq T-t_0+1, \; t\in\mathcal{T}_{\beta}\} } \\
    &\leq g_1\left(\frac{C_{1,\beta}\log T}{(T-t_0+1)(s_0^2-1)^2}\right) + g_2\left(1\right) \\
    &= \frac{C_2(s_0^{2}-1)^2}{\sqrt{C_{1,\beta}}}- \log\left[1 + \frac{(s_0^2-1)}{2}\right] + \frac{\frac{(s_0^{2}-1)}{2}}{1 + \frac{(s_0^2-1)}{2} }.
\end{align*}
Taking the derivative of the second and third term with respect to $s_0$ gives:
\begin{align}
    -\frac{s_0(s_0^2-1)}{2\left(1 + \frac{s_0^2-1}{2}\right)^2} \label{eq:thm3-deriv2}
\end{align}
Suppose that there is some $\overline{\nu} > 1$ so that $s_0^2 > \overline{\nu}^2$ for all $T$, then (\ref{eq:thm3-deriv2}) is negative, implying:
\begin{align*}
    \frac{C_2(s_0^{2}-1)^2}{\sqrt{C_{1,\beta}}}- \log\left[1 + \frac{(s_0^2-1)}{2}\right] + \frac{\frac{(s_0^{2}-1)}{2}}{1 + \frac{(s_0^2-1)}{2}} &\leq \frac{C_2(\overline{s}^{2}-1)^2}{\sqrt{C_{1,\beta}}} - \log\left[1 + \frac{(\overline{\nu}^2-1)}{2}\right] + \frac{\frac{(\overline{\nu}^{2}-1)}{2}}{1 + \frac{(\overline{\nu}^2-1)}{2}}\\
    &:= \frac{C_2(\overline{s}^{2}-1)^2}{\sqrt{C_{1,\beta}}} -\nu^*.
\end{align*}
The inequality $\log(1+x)> x/(x+1)$ for $x > -1$ and the fact that $\overline{\nu}^2 > 1$ imply that $\nu^* > 0$, so if we pick $C_{1,\beta} > (C_1(\overline{s}^2-1)^2/\nu^*)^2$, then we will have $\partial h_T/\partial t < 0$. On the other hand, if there is some $\underline{\nu} < 1$ so that $s_0^2 < \overline{\nu}^2$ for all $T$, then (\ref{eq:thm3-deriv2}) is positive, implying:
\begin{align*}
    \frac{C_2(s_0^{2}-1)^2}{\sqrt{C_{1,\beta}}}- \log\left[1 + \frac{(s_0^2-1)}{2}\right] + \frac{\frac{(s_0^{2}-1)}{2}}{1 + \frac{(s_0^2-1)}{2}} &\leq \frac{C_2(\overline{s}^{2}-1)^2}{\sqrt{C_{1,\beta}}} - \log\left[1 + \frac{(\underline{\nu}^2-1)}{2}\right] + \frac{\frac{(\underline{\nu}^{2}-1)}{2}}{1 + \frac{(\underline{\nu}^2-1)}{2}} 
\end{align*}
and we can use the same argument to pick $C_{1,\beta}>0$ large enough so that $\partial h_T/\partial t < 0$. On the other hand, if $\lim_{T\to\infty} s_0^2 = 1$, we can again use the inequality $\log(1+x) > x/(1+x)$ for $x > -1$ to write:
\begin{align*}
    - \log\left[1 + \frac{(s_0^2-1)}{2}\right] + \frac{\frac{(s_0^{2}-1)}{2}}{1 + \frac{(s_0^2-1)}{2} } \leq 0
\end{align*}
and at the same time, a first order Taylor approximation of $\log(1+x)$ gives:
\begin{align*}
    - \log\left[1 + \frac{(s_0^2-1)}{2}\right] + \frac{\frac{(s_0^{2}-1)}{2}}{1 + \frac{(s_0^2-1)}{2} } &= \frac{(s_0^2-1)}{2} + \mathcal{O}\left((s_0^2-1)^2\right) - \frac{\frac{(s_0^2-1)}{2}}{1 + \frac{(s_0^2-1)}{2}} \\
    &= \mathcal{O}\left((s_0^2-1)^2\right)  + \frac{\frac{(s_0^2-1)^2}{4}}{1 + \frac{(s_0^2-1)}{2}} \\
    &= \mathcal{O}\left((s_0^2-1)^2\right)
\end{align*}
so this term is negative and converges to zero at rate $\mathcal{O}((s_0^2-1)^2)$. Therefore, there is some $B_1>0$ such that:
\begin{align*}
    \lim_{T\to\infty} \frac{\frac{C_2(s_0^{2}-1)^2}{\sqrt{C_{1,\beta}}}}{\log\left[1 + \frac{(s_0^2-1)}{2}\right] - \frac{\frac{(s_0^{2}-1)}{2}}{1 + \frac{(s_0^2-1)}{2} }} = \frac{B_1}{\sqrt{C_{1,\beta}}}.
\end{align*}
So if we pick $C_{1,\beta} > B_1^2$, then for large $T$ we will also have $\partial h_T/\partial t < 0$ in this case, confirming that $h_T(t,t_0)$ is decreasing in $t$, and thus:
\begin{align*}
    \log \frac{\alpha_{t_0}}{\alpha_t} &> h_T(t,t_0)\big|_{t = \sup\;\{t\in\mathcal{T}_{\beta}:t\leq t_0\}} \\
    &\geq h_T\left(t_0-\frac{C_{1,\beta} \log T}{(s_0^2-1)^2},t_0\right) \\
    &= -C_2\sqrt{C_{1,\beta}} \log T + \frac{C_{1,\beta}\log s_0^2}{2(s_0^2-1)^2}\log T \\
    &\quad\:+ \left(\frac{C_{1,\beta} \log T}{2(s_0^2-1)^2}\right)\left(1 +\frac{(T-t_0+1)(s_0^2-1)^2}{C_{1,\beta}\log T} \right)\log\left[1 + \frac{(s_0^{-2}-1)}{1 +\frac{(T-t_0+1)(s_0^2-1)^2}{C_{1,\beta} \log T}}\right] \\
    &\quad\:+ \mathcal{O}(\log T)
\end{align*}
By Assumption \ref{assumption:scale}:
\begin{align*}
    \lim_{T\to\infty}\frac{C_{1,\beta}\log T}{(T-t_0+1)(s_0^2-1)^2} = 0 \implies \lim_{T\to\infty} \frac{(s_0^{-2}-1)}{1 +\frac{(T-t_0+1)(s_0^2-1)^2}{C_{1,\beta}\log T}} =0.
\end{align*}
Therefore, another first order Taylor approximation gives:
\small
\begin{align*}
    &\left(\frac{C_{1,\beta}\log T}{2(s_0^2-1)^2}\right)\left(1 +\frac{(T-t_0+1)(s_0^2-1)^2}{C_{1,\beta}\log T} \right)\log\left[1 + \frac{(s_0^{-2}-1)}{1 +\frac{(T-t_0+1)(s_0^2-1)^2}{C_{1,\beta}\log T}}\right]\\
    &= \frac{C_{1,\beta}(s_0^{-2} -1)}{2(s_0^2-1)^2}\log T + \mathcal{O}\left(\frac{C_{1,\beta}^2\log^2 T}{(T-t_0+1)(s_0^2-1) + C_{1,\beta}\log T}\right) \\
    &= \frac{C_{1,\beta}(s_0^{-2} -1)}{2(s_0^2-1)^2} \log T + \mathcal{O}\left(\log T\right) \tag{$(T-t_0+1)(s_0^2-1)^2 \gg \log T$}
\end{align*}
\normalsize
and thus for some $B_1 > 0$:
\begin{align*}
    \log \frac{\alpha_{t_0}}{\alpha_t} &> \left(\frac{C_{1,\beta}f(s_0^{-2})}{2(s_0^2-1)^2}-C_2\sqrt{C_{1,\beta}} - B_1\right) \log T 
\end{align*}
Once again, if $s_0^2 \geq \overline{\nu}^2 > 1$ for all $T$, then $f(s_0^{-2})/(s_0^2-1)^2 \geq f(\overline{\nu}^{-2})/(\overline{s}^2-1)^2 > 0$, while if $s_0^2 \leq \underline{\nu}^2 < 1$, then $f(s_0^{-2})/(s_0^2-1)^2 \geq f(\underline{\nu}^{-2}) > 0$, and since:
\begin{align*}
    \lim_{s_0^2\to 1} \frac{f(s_0^{-2})}{(s_0^2-1)^2} = \frac{1}{2}
\end{align*}
then if $\lim_{T\to\infty} s_0^2 = 1$, then for large enough $T$, $f(s_0^{-2})/(s_0^2-1)^2 > 1/4$, so again we can pick $C_{1,\beta}>0$ large enough that:
\begin{align*}
    \frac{C_{1,\beta}f(s_0^{-2})}{2(s_0^2-1)^2}-C_2\sqrt{C_{1,\beta}} - B_1 
    > C_\pi
\end{align*}
proving (\ref{eq:thm3-result}) holds for this case.

\subsubsection*{Case 2.b: $t_0 - t \leq T-t_0+1$.}

By the same argument we used in Case 2.a, on $\mathcal{E}$ we have $|\xi(T,t_0,t,s_0)| > -1$ and we can use the inequality $\log(1+x) > x/(x+1)$ for $x > -1$ to write:
\small
\begin{align*}
    \left(\frac{T - t +1}{2}\right)\log\left[1 + \xi(T,t_0,t,s_0)\right] &\geq \left(\frac{T - t +1}{2}\right) \left[\xi(T,t_0,t,s_0) - \frac{\xi(T,t_0,t,s_0)^2}{1+\xi(T,t_0,t,s_0)}\right] \\
    &= \frac{1}{2}\sum_{t'=t}^{t_0-1}(z_{t'}^2 - 1) - \frac{\left(\sum_{t'=t}^{t_0-1}z_{t'}\right)^2 }{2(T-t+1)} + \frac{s_0^2}{2}\left[\sum_{t'=t_0}^T (z_{t'}^2 - 1) -  \frac{\left(\sum_{t'=t_0}^{T}z_{t'}\right)^2 }{(T-t+1)}\right]  \\
    &\quad\; - \frac{s_0\sum_{t'=t_0}^{T}z_{t'} \sum_{t'=t}^{t_0-1}z_{t'}}{(T-t+1)} - \left(\frac{T - t +1}{2}\right)\frac{\xi(T,t_0,t,s_0)^2}{1+\xi(T,t_0,t,s_0)}\\
    &\geq  \frac{1}{2}\sum_{t'=t}^{t_0-1}(z_{t'}^2 - 1) + \frac{s_0^2}{2} \sum_{t'=t_0}^T (z_{t'}^2 - 1) + \mathcal{O}(\log T)
\end{align*}
\normalsize
Returning to bound (\ref{eq:thm3-cs2-bd2}), we now have: 
\begin{align}
    \log \frac{\alpha_{t_0}}{\alpha_t} &>  \frac{(s_0^2-1)}{2}\sum_{t'=t_0}^T (z_{t'}^2-1)- \frac{(T-t_0+1)}{2}\log s_0^2  + \left(\frac{T - t +1}{2}\right)\log\left[1 + \frac{\frac{(T-t_0+1)(s_0^2-1)}{T-t+1}}{1 + \xi(T,t_0,t,s_0)}\right] \notag  \\
    &\quad\: + \left(\frac{T - t +1}{2}\right)\log\left[1 + \frac{\frac{(T-t_0+1)(t_0-t)b_0^2 }{2(T-t+1)^2} - \frac{2(T-t_0+1)|b_0|C_1\sqrt{(t_0-t)\log T }}{(T-t+1)^2}}{1 + \frac{(T-t_0+1)(s_0^2-1)}{T-t+1} + \xi(T,t_0,t,s_0)}\right] \notag \\
    &\quad\: + \mathcal{O}(\log T). \label{eq:thm3-cs2-b-bd1}
\end{align}
As in Case 2.a, we can completely ignore the first line of (\ref{eq:thm3-cs2-b-bd1}) if $s_0^2=1$, and if $s_0^2\neq 1$, then Assumption \ref{assumption:scale} implies $T-t_0+1\gg \log T$. We have shown that on $\mathcal{E}$:
\begin{align*}
    \left|\xi(T,t_0,t,s_0)\right| &\leq (1+s_0^{2} + 2s_0)\left[\frac{C_1\sqrt{\log T }x}{\sqrt{T-t+1}} + \frac{C_1^2\log T}{T-t+1}\right] \\
    &\leq 4\overline{s}^2\left[\frac{C_1\sqrt{\log T}}{\sqrt{T-t_0+1}} + \frac{C_1^2\log T}{T-t_0+1}\right] \tag{$T-t+1\geq T-t_0+1$}
\end{align*}
so again, $\left|\xi(T,t_0,t,s_0)\right| \to 0$ on $\mathcal{E}$ when $s_0^2\neq 1$. Next, we can write:
\footnotesize
\begin{align*}
    \left(\frac{T - t +1}{2}\right)\log\left[1 + \frac{\frac{(T-t_0+1)(s_0^2-1)}{T-t+1}}{1 + \xi(T,t_0,t,s_0)}\right] &= \left(\frac{T - t +1}{2}\right)\log\left[1 + \frac{(T-t_0+1)(s_0^2-1)}{T-t+1}\right] \\
    &\quad\:+ \left(\frac{T - t +1}{2}\right)\log\left[1 - \left(\frac{ \xi(T,t_0,t,s_0)}{1 + \xi(T,t_0,t,s_0)}\right)\left(\frac{\frac{(T-t_0+1)(s_0^2-1)}{T-t+1}}{1+\frac{(T-t_0+1)(s_0^2-1)}{T-t+1}}\right)\right].
\end{align*}
\normalsize
We have just shown that:
\begin{align*}
    \frac{ \xi(T,t_0,t,s_0)}{1 + \xi(T,t_0,t,s_0)} &= \mathcal{O}\left(\frac{\sqrt{\log T}}{\sqrt{T-t+1}}\right) =  \mathcal{O}\left(\frac{\sqrt{\log T}}{\sqrt{T-t_0+1}}\right) \\
    \frac{\frac{(T-t_0+1)(s_0^2-1)}{T-t+1}}{1+\frac{(T-t_0+1)(s_0^2-1)}{T-t+1}} &= \mathcal{O}\left(\frac{(s_0^2-1)(T-t_0+1)}{T-t+1}\right)= \mathcal{O}(1)
\end{align*}
so another first order Taylor approximation of $\log(1+x)$ around zero gives:
\footnotesize
\begin{align*}
     \left(\frac{T - t +1}{2}\right)\log\left[1 - \left(\frac{ \xi(T,t_0,t,s_0)}{1 + \xi(T,t_0,t,s_0)}\right)\left(\frac{\frac{(T-t_0+1)(s_0^2-1)}{T-t+1}}{1+\frac{(T-t_0+1)(s_0^2-1)}{T-t+1}}\right)\right] &= (s_0^2-1)\mathcal{O}\left(\sqrt{(T-t_0+1)\log T}\right) + \mathcal{O}(\log T).
\end{align*}
\normalsize
Similarly, for some $C_1>0$ on $\mathcal{E}$ we have:
\begin{align*}
    \left|\sum_{t' = t_0}^T (z_{t'}^2-1)\right|  \leq C_1\sqrt{(T-t_0+1)\log T} \tag{on $\mathcal{E}_7\subset\mathcal{E}$}
\end{align*}
So for some $C_2 \geq  C_1/2$, we can now rewrite (\ref{eq:thm3-cs2-b-bd1}) on $\mathcal{E}$ as:
\small
\begin{align}
    \log \frac{\alpha_{t_0}}{\alpha_t} &> -C_2|s_0^2-1|\sqrt{(T-t_0+1)\log T}- \frac{(T-t_0+1)}{2}\log s_0^2  + \left(\frac{T - t +1}{2}\right)\log\left[1 + \frac{(T-t_0+1)(s_0^2-1)}{T-t+1}\right] \notag  \\
    &\quad\: + \left(\frac{T - t +1}{2}\right)\log\left[1 + \frac{\frac{(T-t_0+1)(t_0-t)b_0^2 }{2(T-t+1)^2} - \frac{2(T-t_0+1)|b_0|C_1\sqrt{(t_0-t)\log T}}{(T-t+1)^2}}{1 + \frac{(T-t_0+1)(s_0^2-1)}{T-t+1} + \xi(T,t_0,t,s_0)}\right] \notag \\
    &\quad\: + \mathcal{O}(\log T). \label{eq:thm3-cs2-b-bd2}
\end{align}
\normalsize

\subsubsection*{Case 2.b.i: $\min\left\{b_0^2,b_0^2/s_0^2\right\} \geq (s_0^2-1)^2$}

As in case 2.a.i, we can assume that Assumption \ref{assumption:mean} holds in this case. Using the same argument from Case 2.a.i, it is straight forward to show that the term:
\begin{align*}
    - \frac{(T-t_0+1)}{2}\log s_0^2  + \left(\frac{T - t +1}{2}\right)\log\left[1 + \frac{(T-t_0+1)(s_0^2-1)}{T-t+1}\right]
\end{align*}
is decreasing in $t$ and equal to zero when $t = t_0$, and that on $\mathcal{E}$:
\begin{align*}
     \log\left[1 + \frac{\frac{(T-t_0+1)(t_0-t)b_0^2 }{2(T-t+1)^2} - \frac{2(T-t_0+1)|b_0|C_1\sqrt{(t_0-t)\log T}}{(T-t+1)^2}}{1 + \frac{(T-t_0+1)(s_0^2-1)}{T-t+1} + \xi(T,t_0,t,s_0)}\right] \geq \log\left[1 + \frac{(T-t_0+1)(t_0-t)b_0^2 }{16\overline{s}^2(T-t+1)^2}\right].
\end{align*}
So in this case we can write:
\small
\begin{align*}
    \log \frac{\alpha_{t_0}}{\alpha_t} &> -C_2|s_0^2-1|\sqrt{(T-t_0+1)\log T}  + \left(\frac{T - t +1}{2}\right)\log\left[1 + \frac{(T-t_0+1)(t_0-t)b_0^2 }{16\overline{s}^2(T-t+1)^2}\right] + \mathcal{O}(\log T). 
\end{align*}
\normalsize
The second term is decreasing in $t$ and $t_0 -t \geq T-t_0+1$ implies $t \leq 2t_0-T-1$ in this case so for some $B_1 > 0$:
\small
\begin{align*}
    \log \frac{\alpha_{t_0}}{\alpha_t} &> -C_2|s_0^2-1|\sqrt{(T-t_0+1)\log T} + (T-t_0+1)\log\left[1 + \frac{b_0^2 }{16^2\overline{s}^2}\right] - B_1 \log T. 
\end{align*}
\normalsize
Next we have:
\footnotesize
\begin{align*}
    -C_2|s_0^2-1|\sqrt{(T-t_0+1)\log T} + (T-t_0+1)\log\left[1 + \frac{b_0^2 }{16^2\overline{s}^2}\right] &\geq (T-t_0+1)\left(\log\left[1 + \frac{b_0^2 }{16^2\overline{s}^2}\right]-C_2(\overline{s}^2+1) \sqrt{\frac{\log T}{T-t_0+1}}\right).
\end{align*}
\normalsize
Suppose that $b_0^2\to\infty$ as $T\to\infty$, then since $T-t_0+1\gtrsim \log T$, there is some $B_2 > 0$ so that:
\begin{align*}
    \frac{\log T}{B_2}\left(\log\left[1 + \frac{b_0^2 }{16\overline{s}^2}\right]-\frac{C_2(\overline{s}^2+1)}{\sqrt{B_2}}\right)
\end{align*}
So (\ref{eq:thm3-result}) holds for $T$ large enough that:
\begin{align*}
    \log\left[1 + \frac{b_0^2 }{16^2\overline{s}^2}\right] > \frac{C_2(\overline{s}^2+1)}{\sqrt{B_2}} + B_2(B_1 + C_\pi).
\end{align*}
On the other hand, if $b_0^2 = \mathcal{O}(1)$, then Assumption \ref{assumption:mean} implies $T-t_0+1\gg \log T$. If $b_0^2 > \underline{b}^2 > 0$ for all $T$, then for large $T$ we have:
\begin{align*}
    \log\left[1 + \frac{b_0^2 }{16^2\overline{s}^2}\right]-C_2(\overline{s}^2+1) \sqrt{\frac{\log T}{T-t_0+1}} >\frac{1}{2}\log\left[1 + \frac{\underline{b}^2 }{16^2\overline{s}^2}\right] > 0
\end{align*}
and:
\begin{align*}
    \frac{(T-t_0+1)}{2}\log\left[1 + \frac{\underline{b}^2 }{16^2\overline{s}^2}\right] - \log T \geq C_\pi \log T.
\end{align*}
Finally, if $\lim_{T\to\infty}b_0^2=0$, then a first order Taylor expansion of $\log(1+x)$ gives: 
\begin{align*}
    \log\left[1 + \frac{b_0^2 }{16^2\overline{s}^2}\right] = \frac{b_0^2 }{16^2\overline{s}^2} + \mathcal{O}(b_0^4) \geq \frac{b_0^2 }{512\overline{s}^2} 
\end{align*}
where the last inequality holds for large $T$, and thus:
\begin{align*}
    \log \frac{\alpha_{t_0}}{\alpha_t} &> -C_2(\overline{s}^2+1)\sqrt{(T-t_0+1)\log T} + \frac{(T-t_0+1)b_0^2 }{512\overline{s}^2}  - B_1 \log T. 
\end{align*}
But Assumption \ref{assumption:mean} implies $(T-t_0+1)b_0^2 \gg \log T$, so the second term dominates for large $T$ we again get that (\ref{eq:thm3-result}) holds for this case.

\subsubsection*{Case 2.b.ii: $\min\left\{b_0^2,b_0^2/s_0^2\right\} < (s_0^2-1)^2$}

As in case 2.a.ii, we can assume that Assumption \ref{assumption:scale} holds in this case. Returning to bound (\ref{eq:thm3-cs2-b-bd2}), we can use the same argument from Case 2.a.ii to get:
\small
\begin{align*}
    \log \frac{\alpha_{t_0}}{\alpha_t} &> -C_2|s_0^2-1|\sqrt{(T-t_0+1)\log T}- \frac{(T-t_0+1)}{2}\log s_0^2  + \left(\frac{T - t +1}{2}\right)\log\left[1 + \frac{(T-t_0+1)(s_0^2-1)}{T-t+1}\right] \notag  \\
    &\quad\: + \left(\frac{T - t +1}{2}\right)\log\left[1 - \frac{4(T-t_0+1)|b_0|C_1\sqrt{(t_0-t)\log T }}{\underline{s}^2(T-t+1)^2}\right]  + \mathcal{O}(\log T).
\end{align*}
\normalsize
and:
\footnotesize
\begin{align*}
    \left(\frac{T - t +1}{2}\right)\log\left[1 - \frac{4(T-t_0+1)|b_0|C_1\sqrt{(t_0-t)\log T}}{\underline{s}^2(T-t+1)^2}\right] &= -\frac{4(T-t_0+1)|b_0|C_1\sqrt{(t_0-t)\log T}}{\underline{s}^2(T-t+1)} + \mathcal{O}\left(\log T\right) \\
    &\geq -\frac{4\overline{s}C_1|s_0^2-1|}{\underline{s}^2} \sqrt{(T-t_0+1)\log T} \tag{$T-t+1\geq t_0-t \geq T-t_0+1$ and $\min\{b_0^2,b_0^2/s_0^2\}\leq (s_0^2-1)^2$} + \mathcal{O}\left(\log T\right).
\end{align*}
\normalsize
So for $C_3 > C_2 + 4\overline{s}C_1/\underline{s}^2$ we have:
\small
\begin{align}
    \log \frac{\alpha_{t_0}}{\alpha_t} &> -C_3|s_0^2-1|\sqrt{(T-t_0+1)\log T}- \frac{(T-t_0+1)}{2}\log s_0^2  + \left(\frac{T - t +1}{2}\right)\log\left[1 + \frac{(T-t_0+1)(s_0^2-1)}{T-t+1}\right] \notag \\
    &\quad\:+ \mathcal{O}(\log T). \label{eq:thm3-cs2-b-ii-bd1}
\end{align}
\normalsize
The bound (\ref{eq:thm3-cs2-b-ii-bd1}) is decreasing as a function of $t$ and in this case, since $t_0-t>T-t_0+1$, then $t < 2t_0-T + 1$. Therefore, we have:
\small
\begin{align*}
    \log \frac{\alpha_{t_0}}{\alpha_t} &> -C_3|s_0^2-1|\sqrt{(T-t_0+1)\log T}-\frac{(T-t_0+1)}{2}\log s_0^2 + (T-t_0+1)\log\left[1 + \frac{(s_0^2-1)}{2}\right]  + \mathcal{O}(\log T) \\
    &= -C_3|s_0^2-1|\sqrt{(T-t_0+1)\log T} + (T-t_0+1)\log\left[1 + \frac{s_0+\frac{1}{s_0}-2}{2}\right] + \mathcal{O}(\log T)
\end{align*}
\normalsize
If there is some $\underline{\nu} <1$ so that $s_0 \leq \underline{\nu}$ for all $T$, then: 
\small
\begin{align*}
    \log \frac{\alpha_{t_0}}{\alpha_t} &> -C_3\sqrt{(T-t_0+1)\log T}  + (T-t_0+1)\log\left[1 + \frac{\underline{\nu} +\underline{\nu}^{-1}-2}{2}\right] + \mathcal{O}(\log T)
\end{align*}
\normalsize
and if there is some $\overline{\nu} >1$ so that $s_0 \geq \underline{\nu}$ for all $T$, then: 
\small
\begin{align*}
    \log \frac{\alpha_{t_0}}{\alpha_t} &> -C_3|\overline{s}^2-1|\sqrt{(T-t_0+1)\log T}  + (T-t_0+1)\log\left[1 + \frac{\overline{\nu} +\overline{\nu}^{-1}-2}{2}\right] + \mathcal{O}(\log T)
\end{align*}
\normalsize
and since $ \min\{\underline{\nu} +\underline{\nu}^{-1},  \overline{\nu} +\overline{\nu}^{-1}\}-2 > 0$ and $T-t_0+1 \gg \log T$, the second term dominates in both the bounds above and for large enough $T$ we will have $\log \frac{\alpha_{t_0}}{\alpha_t} \geq C_\pi \log^2 T$. On the other hand, if $\lim_{T\to\infty} s_0^2 = 1$, then we can use a first order Taylor approximation of $\log(1+x)$ around zero to write: 
\begin{align*}
    \log\left[1 + \frac{s_0+\frac{1}{s_0}-2}{2}\right] = \frac{s_0+\frac{1}{s_0}-2}{2} + \mathcal{O}\left(\left(s_0+\frac{1}{s_0}-2\right)^2\right).
\end{align*}
Since $\lim_{s_0\to 1} s_0+\frac{1}{s_0}-2 = 0$, then for large enough $T$, we have:
\begin{align*}
    \log\left[1 + \frac{s_0+\frac{1}{s_0}-2}{2}\right] = \frac{s_0+\frac{1}{s_0}-2}{4}
\end{align*}
and thus:
\small
\begin{align*}
    \log \frac{\alpha_{t_0}}{\alpha_t} &> -C_3|s_0^2-1|\sqrt{(T-t_0+1)\log T }+ (T-t_0+1)\left( \frac{s_0+\frac{1}{s_0}-2}{4}\right) + \mathcal{O}(\log T)
\end{align*}
\normalsize
Next, since:
\begin{align*}
    \lim_{s_0\to1} \frac{s_0+\frac{1}{s_0}-2}{(s_0^{2} - 1)^2} = \frac{1}{4}
\end{align*}
then for large enough $T$ we have:
\begin{align*}
    s_0+\frac{1}{s_0}-2 > \frac{(s_0^{2} - 1)^2}{8}
\end{align*}
and thus: 
\begin{align*}
    \log \frac{\alpha_{t_0}}{\alpha_t} &> -C_3\sqrt{(s_0^2-1)^2(T-t_0+1)\log T} + \frac{(T-t_0+1)(s_0^2-1)^2}{8}+ \mathcal{O}(\log T)
\end{align*}
By Assumption \ref{assumption:scale}, $(T-t_0+1)(s_0^2-1)^2 \gg \log T$, so again the second term dominates the others and (\ref{eq:thm3-result}) holds on $\mathcal{E}$ for large $T$ in this case.

\end{proof}
\subsection{Theorem \ref{theorem:alpha-mixing} Event Bounds}

\begin{lemma}[Theorem \ref{theorem:alpha-mixing} Event Bounds]\label{lemma:thm4-event-bound}

Let $\{y_t\}_{t\geq 1}$ be a univariate sequence of random variables satisfying Assumption \ref{assumption:alpha-mixing}. Let $\E[y_t] = b_0\mathbbm{1}\{t\geq t_0\}$ and $\normalfont{\Var}(y_t) = (s_0^2)^{\mathbbm{1}\{t\geq t_0\}}$, with $b_0 \in \mathbb{R}$ and $s_0 > 0$ so that it is possible that $b_0 = 0$ or $s^2_0 = 1$. Define the standardized term $z_t:= (y_t - b_0\mathbbm{1}\{t\geq t_0\})/s_0^{\mathbbm{1}\{t\geq t_0\}}$ and for any sequence $\{a_T\}_{T \geq 1}$ such that $a_T \to \infty$ and constants $C, C_{1,a}, C_{2,a}, C_{3,a} > 0$, define the set:
\begin{align*}
    \mathcal{T}_{a_T} &:= \left\{ 1 \leq t \leq  T - C_{2,a}a_T\log^2 T\::\: |t_0 - t| > \frac{C_{1,a}a_T\log^2 T}{\max\{\min\{b_0^2,b_0^2/s_0^2\},(s_0^2-1)^2\} } \right\} 
\end{align*}
and the events: 
\begin{align*}
    \mathcal{E}_1 &:= \bigcap_{t \in \mathcal{T}_{a_T}} \left\{\left|\sum_{t'=\min\{t_0,t\}}^{\max\{t_0,t\}-1} z_{t'}\right| < C_{3,a}\sqrt{a_T|t_0-t|}\log T\right\}, \\
    \mathcal{E}_2 &:= \bigcap_{t \in \mathcal{T}_{a_T}} \left\{\left|\sum_{t'=\min\{t_0,t\}}^{\max\{t_0,t\}-1} (z_{t'}^2 - 1) \right| <  C_{3,a}\sqrt{a_T|t_0-t|}\log T  \right\},  \\
    \mathcal{E}_3 &:=  \bigcap_{t =1}^T \left\{\left|\sum_{t'=t}^{T} z_{t'}\right| < C_{3,a}\sqrt{a_T(T-t+1)}\log T  \right\}, \\
    \mathcal{E}_4 &:=  \bigcap_{t =1}^T \left\{\left|\sum_{t'=t}^{T} (z_{t'}^2 - 1) \right| <  C_{3,a}\sqrt{a_T(T-t+1)}\log T  \right\},\\
    \mathcal{E}_5 &= \left\{\left|b_0\sum_{t'=t_0}^T z_{t'}\right| \leq \frac{(T-t_0+1)b_0^2 }{4} \right\},
\end{align*}
and the joint event $\mathcal{E} := \bigcap_{i=1}^4 \mathcal{E}_i$. Then if:
\begin{align*}
    \min\{b_0^2,b_0^2/s_0^2\} \min\{t_0,T-t_0+1\} \gg a_T\log^2T
\end{align*}
or:
\begin{align*}
    (s_0^2-1)^2 \min\{t_0,T-t_0+1\} \gg a_T\log^2T
\end{align*}
then there exist $C, C_{1,a}, C_{2,a}, C_{3,a}$ large enough so that we will have $\Pr(\mathcal{E}) \geq 1 - Ca_T^{-1}$. 
\end{lemma}

\begin{proof}

By the union bound $\Pr(\mathcal{E}^c) = 1 - P(\cup_{i=1}^5 \mathcal{E}^c_i) < 1 - \sum^5_{i=1} \Pr(\mathcal{E}^c_i)$, so it is sufficient to show that $\Pr(\mathcal{E}^c_i) \leq \frac{C_i}{a_T}$ for each $i$. First, by Assumption \ref{assumption:alpha-mixing} there are stochastic processes $\{z_{0,t}\}_{t \geq 1}$ and $\{z_{1,t}\}_{t \geq 1}$ such that $z_{0,t} \sim F_0$, $z_{1,t} \sim F_1$, $\E[z_{0,t}] = \E[z_{1,t}] = 0$, $\Var(z_{0,t}) = \Var(z_{1,t}) = 1$, $y_t := z_{0,t} \mathbbm{1}_{\{t < t_0\}}  + (s_0z_{1,t} + b_0)\mathbbm{1}_{\{t \geq t_0\}}$, and:
    \begin{enumerate}[label=\normalfont(\roman*)]
        \item $\{z_{0,t}\}_{t \geq 1}$ and $\{z_{1,t}\}_{t \geq 1}$ are $\alpha$-mixing processes with respective coefficients $\{\alpha_{0,k}\}_{k\geq 1}$ and $\{\alpha_{1,k}\}_{k\geq 1}$ that satisfy $\max\{\alpha_{0,k}, \alpha_{1,k}\} \leq e^{-C k}$ for some $C > 0$. \vspace{-5pt}
        \item There exist constants $D_1, D_2 > 0$ such that $\sup_{t \geq 1} \max\{\E\left[|z_{0,t}|^{4+D_1}\right],\;\E\left[|z_{1,t}|^{4+D_1}\right]\}\leq D_2.$  
    \end{enumerate}
Therefore, the conditions of Lemma \ref{lemma:Padilla23} are met. Next we can write:
\small
\begin{align*}
    \Pr\left(\mathcal{E}_1^c\right) &= \Pr\left(\bigcup_{t \in \mathcal{T}_{a_T}} \left\{\left|\sum_{t'=\min\{t_0,t\}}^{\max\{t_0,t\}-1} z_{t'}\right| \geq C_{3,a}\sqrt{a_T|t_0-t|}\log T\right\} \right) \\
    &\leq \Pr\left(\bigcup_{t \in \mathcal{T}_{a_T} \::\: t < t_0} \left\{\left|\sum_{t'=t}^{t_0-1} z_{t'}\right| \geq C_{3,a}\sqrt{a_T(t_0-t)}\log T\right\} \right) + \Pr\left(\bigcup_{t \in \mathcal{T}_{a_T} \::\: t > t_0} \left\{\left|\sum_{t'=t_0}^{t-1} z_{t'}\right| \geq C_{3,a}\sqrt{a_T(t-t_0)}\log T\right\} \right) \tag{union bound} \\
    &\leq \Pr\left(\bigcup_{t = 1}^{t_0-1} \left\{\left|\sum_{t'=t}^{t_0-1} z_{t'}\right| \geq C_{3,a}\sqrt{a_T(t_0-t)}\log T\right\} \right) + \Pr\left(\bigcup_{t = t_0+1}^T \left\{\left|\sum_{t'=t_0}^{t-1} z_{t'}\right| \geq C_{3,a}\sqrt{a_T(t-t_0)}\log T\right\} \right) \\
    &= \Pr\left(\bigcup_{t = 1}^{t_0-1} \left\{\left|\sum_{t'=t}^{t_0-1} z_{0,t'}\right| \geq C_{3,a}\sqrt{a_T(t_0-t)}\log T\right\} \right) + \Pr\left(\bigcup_{t = t_0+1}^T \left\{\left|\sum_{t'=t_0}^{t-1} z_{1,t'}\right| \geq C_{3,a}\sqrt{a_T(t-t_0)}\log T\right\} \right) \tag{$z_{t'} = z_{0,t'} \mathbbm{1}\{t' < t_0\}  + z_{1,t'}\mathbbm{1}\{t' \geq t_0\}$}
\end{align*}
\normalsize
Since $\{z_{0,t}\}_{t\geq1}$ and $\{z_{1,t}\}_{t\geq1}$ satisfy the conditions of Lemma \ref{lemma:alpha-mix-sum}, for some $K_1,K_2 > 0$ we have:
\begin{align*}
    \Pr\left(\bigcup_{t = 1}^{t_0-1} \left\{\left|\sum_{t'=t}^{t_0-1} z_{0,t'}\right| \geq K_1\sqrt{a_T(t_0-t)}\log T\right\} \right) &\leq \frac{1}{a_T} \\
    \Pr\left(\bigcup_{t = t_0+1}^T \left\{\left|\sum_{t'=t_0}^{t-1} z_{1,t'}\right| \geq K_2\sqrt{a_T(t-t_0)}\log T\right\} \right) &\leq \frac{1}{a_T}
\end{align*}
So for $C_{3,a} > \max\{K_1,K_2\}$ we get $\Pr\left(\mathcal{E}_1^c\right) \leq \frac{2}{a_T}$. Since $\{z^2_{0,t}-1\}_{t\geq1}$ and $\{z^2_{1,t}-1\}_{t\geq1}$ also satisfy the conditions of Lemma \ref{lemma:alpha-mix-sum}, an identical argument shows we can pick $C_{3,a}$ large enough so that $\Pr\left(\mathcal{E}_2^c\right) \leq \frac{2}{a_T}$ as well. 

Next, note that if $t < t_0$ and: 
\begin{align*}
    \left|\sum_{t'=t}^{t_0-1} z_{0,t'}\right| < \frac{C_{3,a}}{2}\sqrt{a_T(t_0-t)}\log T
\end{align*}
and:
\begin{align*}
    \left|\sum_{t'=t_0}^{T} z_{1,t'}\right| < \frac{C_{3,a}}{2}\sqrt{a_T(T-t_0+1)}\log T
\end{align*}
then:
\begin{align*}
    \left|\sum_{t'=t}^{T} z_{t'}\right| < \frac{C_{3,a}}{2}\sqrt{a_T(T-t+1)}\log T + \frac{C_{3,a}}{2}\sqrt{a_T(t_0-t)}\log T \leq C_{3,a}\sqrt{a_T(T-t+1)}\log T. \tag{$T-t+1 > \max\{t_0-t,T-t_0+1\}$}
\end{align*}
So:
\small
\begin{align*}
     \Pr\left(\mathcal{E}^c_3\right) &\leq \Pr\left(\bigcup_{t=1}^{T} \left\{\left|\sum_{t'=t}^{T} z_{t'}\right| \geq C_{3,a}\sqrt{a_T(T-t+1)}\log T\right\} \right) \\
     &\leq\Pr\left(\bigcup_{t=1}^{t_0-1} \left\{\left|\sum_{t'=t}^{T} z_{t'}\right| \geq C_{3,a}\sqrt{a_T(T-t+1)}\log T\right\} \right) + \Pr\left(\bigcup_{t=t_0}^{T} \left\{\left|\sum_{t'=t}^{T} z_{1,t'}\right| \geq C_{3,a}\sqrt{a_T(T-t+1)}\log T\right\} \right) \tag{union bound, $z_t = z_{1,t}$ for $t\geq t_0$} \\
     &\leq\Pr\left(\bigcup_{t=1}^{t_0-1} \left\{\left|\sum_{t'=t}^{t_0-1} z_{0,t'}\right| \geq \frac{C_{3,a}}{2}\sqrt{a_T(t_0-t)}\log T\right\}\cup\left\{ \left|\sum_{t'=t_0}^{T} z_{1,t'}\right| \geq \frac{C_{3,a}}{2}\sqrt{a_T(T-t_0+1)}\log T\right\}  \right) \\
     &\quad\: + \Pr\left(\bigcup_{t=t_0}^{T} \left\{\left|\sum_{t'=t}^{T} z_{1,t'}\right| \geq C_{3,a}\sqrt{a_T(T-t+1)}\log T\right\} \right) \tag{union bound} \\
     &\leq \Pr\left(\bigcup_{t=1}^{t_0-1} \left\{\left|\sum_{t'=t}^{t_0-1} z_{0,t'}\right| \geq \frac{C_{3,a}}{2}\sqrt{a_T(t_0-t)}\log T\right\}  \right) + \Pr\left( \left|\sum_{t'=t_0}^{T} z_{1,t'}\right| \geq \frac{C_{3,a}}{2}\sqrt{a_T(T-t_0+1)}\log T  \right) \\
     &\quad\: + \Pr\left(\bigcup_{t=t_0}^{T} \left\{\left|\sum_{t'=t}^{T} z_{1,t'}\right| \geq C_{3,a}\sqrt{a_T(T-t+1)}\log T\right\} \right) \tag{union bound} \\
\end{align*}
\normalsize
Once again, by Lemma \ref{lemma:alpha-mix-sum} there is some $K_3 > 0$ such that:
\begin{align*}
    \Pr\left(\bigcup_{t=1}^{T} \left\{\left|\sum_{t'=t}^{T} z_{1,t'}\right| \geq K_3\sqrt{a_T(T-t+1)}\log T\right\} \right)  &\leq \frac{1}{a_T}
\end{align*}
So for $C_{3,a} \geq 2\max\{K_1,K_3\}$, we have $\Pr\left(\mathcal{E}_3^c\right) \leq \frac{3}{a_T}$. Again, we can use an identical argument to show that we can pick $C_{3,a}$ large enough so that $\Pr\left(\mathcal{E}_4^c\right) \leq \frac{3}{a_T}$.

Finally note that the event $\mathcal{E}_5$ holds trivially if $b_0=0$, otherwise if $b_0 \neq 0$ then by assumption $b_0^2(T-t_0+1)\gg a_T\log T$ and for $T$ large enough that $|b_0|\sqrt{T-t_0+1} \geq 4C_{3,a}\sqrt{a_T}\log T$:
\begin{align*}
    \Pr(\mathcal{E}_5^c) &= \Pr\left(\left|\sum_{t'=t_0}^T z_{t'}\right| > \frac{(T-t_0+1)|b_0| }{4}\right) \\
    &\leq \Pr\left(\left|\sum_{t'=t_0}^T z_{t'}\right| > C_{3,a}\sqrt{a_T(T-t_0+1)}\log T\right) \tag{$|b_0|\sqrt{T-t_0+1} \geq 4C_{3,a}\sqrt{a_T}\log T$} \\
    &\leq \Pr(\mathcal{E}^c_3).
\end{align*}

\end{proof}

\subsection{Proof of Theorem \ref{theorem:alpha-mixing}}
\label{app:alpha-mixing}

Below we show the proof of Theorem \ref{theorem:alpha-mixing} for the meanvar-scp model of Section \ref{sec:smscp}. Proving that Theorem \ref{theorem:alpha-mixing} holds for the mean-scp and var-scp models of Sections \ref{sec:smcp} and \ref{sec:sscp} respectively is a straightforward extension since the form of the log-evidence for the meanvar-scp model is essentially a combination of the log-evidence for the mean-scp and var-scp models. 

As we did for the proof of Theorem \ref{theorem:smscp} in Appendix \ref{app:localization-smscp}, for some sequence $\{a_T\}_{T\geq1}$ and constants $C_{1,a}, C_{2,a} > 0$, we define the set:
\begin{align*}
    \mathcal{T}_{a_T}:= \left\{1\leq t \leq T - \lfloor C_{2,a}a_T\log^{2} T \rfloor\::\: |t_0 - t| > \frac{C_{1,a}a_T\log^{2} T}{\max\{\min\{b^2_0,b_0^2/s_0^2\}, (s_0^2-1)^2\}} \right\}.
\end{align*}
We also assumed in Theorem \ref{theorem:alpha-mixing} that for large $T$ there is some constant $C_\pi > 0$ such that for all $t \in [T]$: 
\begin{align}
    \log \frac{\pi_t}{\pi_{t_0}} \leq C_\pi \log^2 T.
\end{align}
Therefore, by the same argument we used in the proof of Theorem \ref{theorem:smscp}, if we let $p(\mathbf{y}_{1:T} \:|\:\gamma = t_0 \:; \lambda_0, u_0, v_0)$ be the marginal likelihood for the meanvar-scp model and we show that under the conditions of Theorem \ref{theorem:alpha-mixing}, for any $a_T \to\infty$, we can pick $C_{1,a}, C_{2,a}, C_{3,a}>0$ large enough such that:
\begin{align}
    \lim_{T\to\infty}\Pr\left(\bigcap_{t\in\mathcal{T}_{a_T}}\left\{ \log \frac{p(\mathbf{y}_{1:T} \:|\:\gamma = t_0 \:; \lambda_0, u_0, v_0)}{p(\mathbf{y}_{1:T} \:|\:\gamma = t \:; \lambda_0, u_0, v_0)} > C_\pi \log^2 T\right\}\right) = 1 \label{eq:thm4-result}
\end{align}
then we will have proven the desired result.  We now prove that there exists a $C_{1,a}, C_{2,a} > 0$ such that the event in (\ref{eq:thm4-result}) holds on the event $\mathcal{E}$ defined in Lemma \ref{lemma:thm4-event-bound}.

\begin{proof}

As in the proof of Theorem \ref{theorem:smscp}, it is sufficient to show that $\log \frac{\alpha_{t_0}}{\alpha_t} > C_\pi \log^2 T$ for large $T$ on the event $\mathcal{E}$. To show this we again consider two cases:

\subsubsection*{Case 1: $t > t_0$, and $t \in \mathcal{T}_{a_T}$.}

Using the identical argument we used to get (\ref{eq:thm3-cs1-bd1}) in the proof of Theorem \ref{theorem:smscp} in Appendix \ref{app:localization-smscp} and the function $f(x) = x -\log(x) -1$, we can begin from:
\begin{align}
    \log \frac{\alpha_{t_0}}{\alpha_t} &\geq \log \left(\frac{T - t + 1}{T - t_0 +1} \right) + \left(\frac{t-t_0}{2}\right)\left[b_0^2 + f(s_0^2)\right]  + \frac{(s_0^2-1)}{2} \sum_{t' = t_0}^{t - 1} (z_{t'}^2-1) + b_0 s_0 \sum_{t' = t_0}^{t - 1} z_{t'} \notag \\
    &\quad\:  - \frac{\left(\sum_{t'=t}^T z_{t'}\right)^2}{2(T-t+1)} - \frac{\frac{(T-t+1)}{2}\left[\frac{\sum_{t'=t}^T (z_{t'}^2-1)}{T-t+1} - \left(\frac{\sum_{t'=t}^T z_{t'}}{T-t+1}\right)^2 \right]^2}{1 + \frac{\sum_{t'=t}^T (z_{t'}^2-1)}{T-t+1} - \left(\frac{\sum_{t'=t}^T z_{t'}}{T-t+1}\right)^2 } + \mathcal{O}(1). \label{eq:thm4-cs1-bd1}
\end{align}
On the event $\mathcal{E}$ defined in Lemma \ref{lemma:thm4-event-bound}, for some constant $C_1 > 0$ we have:
\begin{align*}
    \left|\sum_{t'=t_0}^{t-1} z_{t'}\right| &\leq C_1\sqrt{a_T(t-t_0)}\log  T, \tag{on $\mathcal{E}_1\subset\mathcal{E}$} \\
    \left|\sum_{t' = t_0}^{t - 1} (z_{t'}^2-1) \right| &\leq C_1\sqrt{a_T(t-t_0)}\log  T \tag{on $\mathcal{E}_2\subset\mathcal{E}$} 
\end{align*}
and thus:
\begin{align}
    \log \frac{\alpha_{t_0}}{\alpha_t} &\geq  \log\left(\frac{T - t + 1}{T - t_0 +1} \right)+\left(\frac{t-t_0}{2}\right)\left[b_0^2 + f(s_0^2)\right]  - \left(\frac{|s_0^2-1|}{2}  + |b_0| s_0\right)C_1\sqrt{a_T(t-t_0)}\log  T \notag \\
    &\quad\:  - \frac{\left(\sum_{t'=t}^T z_{t'}\right)^2}{2(T-t+1)} - \frac{\frac{(T-t+1)}{2}\left[\frac{\sum_{t'=t}^T (z_{t'}^2-1)}{T-t+1} - \left(\frac{\sum_{t'=t}^T z_{t'}}{T-t+1}\right)^2 \right]^2}{1 + \frac{\sum_{t'=t}^T (z_{t'}^2-1)}{T-t+1} - \left(\frac{\sum_{t'=t}^T z_{t'}}{T-t+1}\right)^2 } + \mathcal{O}(1). \label{eq:thm4-cs1-bd2}
\end{align}
Next, on the event $\mathcal{E}$ defined in Lemma \ref{lemma:thm4-event-bound}, for the same universal constant $C_1 > 0$ we have:
\begin{align*}
    \max\left\{\left|\sum_{t' = t}^{T} z_{t'}\right|,\left|\sum_{t' = t}^{T} (z^2_{t'}-1)\right|\;\right\} &< C_1 \sqrt{a_T(T-t+1)}\log T \tag{on $\mathcal{E}_3\cap\mathcal{E}_4\subset\mathcal{E}$} 
\end{align*}
and thus:
\begin{align*}
    -\frac{\left(\sum_{t'=t}^T z_{t'}\right)^2}{2(T-t+1)} > - \frac{C_1^2}{2}a_T\log^{2} T
\end{align*}
and:
\begin{align*}
     \left|\frac{\sum_{t'=t}^T (z_{t'}^2-1)}{T-t+1} - \left(\frac{\sum_{t'=t}^T z_{t'}}{T-t+1}\right)^2\right| &< \frac{C_1\sqrt{a_T}\log T}{\sqrt{T-t+1}} + \frac{C_1^2 a_T\log^{2} T}{T-t+1}. 
\end{align*}
Suppose we pick $C_{2,a} > 4C_1^2$, then since $t\in\mathcal{T}_{a_T}$ implies $T-t+1 \geq C_{2,a}a_T\log^2 T$, then: 
\begin{align*}
     \left|\frac{\sum_{t'=t}^T (z_{t'}^2-1)}{T-t+1} - \left(\frac{\sum_{t'=t}^T z_{t'}}{T-t+1}\right)^2\right| &< \frac{3}{4}
\end{align*}
and:
\begin{align*}
    \frac{\frac{(T-t+1)}{2}\left[\frac{\sum_{t'=t}^T (z_{t'}^2-1)}{T-t+1} - \left(\frac{\sum_{t'=t}^T z_{t'}}{T-t+1}\right)^2 \right]^2}{1 + \frac{\sum_{t'=t}^T (z_{t'}^2-1)}{T-t+1} - \left(\frac{\sum_{t'=t}^T z_{t'}}{T-t+1}\right)^2 } &< 2(T-t+1)\left[\frac{\sum_{t'=t}^T (z_{t'}^2-1)}{T-t+1} - \left(\frac{\sum_{t'=t}^T z_{t'}}{T-t+1}\right)^2 \right]^2 \\
    &< 2(T-t+1)\left[\frac{C_1\sqrt{a_T}\log T}{\sqrt{T-t+1}} + \frac{C_1^2 a_T\log^{2} T}{T-t+1}\right]^2 \\
    &= 2\left[C_1\sqrt{a_T}\log T + \frac{C_1^2 a_T\log^{2} T}{\sqrt{T-t+1}}\right]^2 \\
    &\leq 8 C_1^2 a_T\log^{2} T. \tag{$t\in\mathcal{T}_{a_T}\implies T-t+1 \geq C_{1}a_T\log^2 T$}
\end{align*}
We can now rewrite bound (\ref{eq:thm4-cs1-bd2}) as:
\begin{align}
    \log \frac{\alpha_{t_0}}{\alpha_t} &\geq \left(\frac{t-t_0}{2}\right)\left[b_0^2 + f(s_0^2)\right]  - \left(\frac{|s_0^2-1|}{2}  + |b_0| s_0\right)C_1\sqrt{a_T(t-t_0)}\log  T + \mathcal{O}(a_T\log^2 T) \label{eq:thm4-cs1-bd3}
\end{align}
where we have absorbed the $\log\left(\frac{T-t+1}{T-t_0+1}\right)$ term into the $\mathcal{O}(a_T\log^2 T)$ term. We now show that bound (\ref{eq:thm4-cs1-bd3}) is increasing as a function of $t-t_0$. Taking the derivative of this bound with respect to $t-t_0$ yields:
\begin{align}
    \frac{b_0^2 + f(s_0^2)}{2}  - \frac{C_1}{2}\left(\frac{|s_0^2-1|}{2}  + |b_0| s_0\right)\sqrt{\frac{a_T}{t-t_0}}\log T. \label{eq:thm4-deriv1}
\end{align}
Note that (\ref{eq:thm4-deriv1}) is decreasing as a function of $t-t_0$. Suppose that: 
\begin{align*}
    \min\left\{b_0^2,\frac{b_0^2}{s_0^2}\right\} \geq (s_0^2-1)^2 \implies |t-t_0| \geq \frac{C_{1,a} a_T\log^2T}{\min\{b_0^2,b_0^2/s_0^{2}\}}, \;\forall\; t\in\mathcal{T}_{a_T}
\end{align*}
and thus (\ref{eq:thm4-deriv1}) is bounded below by: 
\begin{align*}
    \frac{b_0^2}{2} - \frac{C_1}{2}\left(\frac{|s_0^2-1|}{2}  + |b_0| s_0\right)\frac{\min\{|b_0|,|b_0|/s_0\}}{\sqrt{C_{1,a}}} &\geq \frac{b_0^2}{2} - \frac{C_1}{2\sqrt{C_{1,a}}}\left(\frac{\min\{b^2_0,b^2_0/s^2_0\}}{2} + \min\{b^2_0s_0,b_0^2\}\right) \tag{$\min\left\{b_0^2,b_0^2/s_0^2\right\} \geq (s_0^2-1)^2$} \\
    &\geq\frac{b_0^2}{2}\left(1 - \frac{3C_1}{2\sqrt{C_{1,a}}}\right).
\end{align*}
So (\ref{eq:thm4-deriv1}) will be positive if we set $C_{1,a} > 9C_1^2/4$. On the other hand, if:
\begin{align*}
    \min\left\{b_0^2,\frac{b_0^2}{s_0^2}\right\} \leq (s_0^2 - 1)^2 \implies |t-t_0| \geq \frac{C_{1,a} a_T\log^2T}{ (s_0^2 - 1)^2}, \;\forall\; t\in\mathcal{T}_{a_T}
\end{align*}
then (\ref{eq:thm4-deriv1}) is bounded below by: 
\begin{align*}
    \frac{f(s_0^2)}{2}  - \frac{C_1}{2}\left(\frac{|s_0^2-1|}{2}  + |b_0| s_0\right)\frac{|s_0^2-1|}{\sqrt{C_{1,a}}}.
\end{align*}
Suppose that $s_0^2 > 1$ so that $(s_0^2 - 1)^2 \geq b_0^2/s_0^2$, then:
\begin{align*}
    \frac{f(s_0^2)}{2}  - \frac{C_1}{2\sqrt{C_{1,a}}}\left(\frac{(s_0^2-1)^2}{2}  +  s^2_0(s_0^2-1)^2\right) &\geq \frac{f(s_0^2)}{2}  - \frac{C_1\overline{s}^2 (s_0^2-1)^2}{\sqrt{C_{1,a}}}. \tag{$s_0^2 \leq \overline{s}^2 < \infty$}
\end{align*}
Suppose that $s_0^2 > \overline{\nu}^2 > 1$ for all $T$, then: 
\begin{align*}
    \frac{f(s_0^2)}{2}  - \frac{C_1\overline{s}^2 (s_0^2-1)^2}{\sqrt{C_{1,a}}} \geq \frac{f(\overline{\nu}^2)}{2}  - \frac{C_1\overline{s}^2 (\overline{s}^2-1)^2}{\sqrt{C_{1,a}}} \tag{$s_0^2 \leq \overline{s}^2 < \infty$}
\end{align*}
and we can pick $C_{1,a}$ large enough so that the RHS is positive. If $s_0^2 < \underline{\nu}^2$ for some $\underline{\nu}^2 < 1$ for all $T$, then we can use the same argument and the fact that $s_0 > \underline{s} > 0$ to show (\ref{eq:thm4-deriv1}) is positive for large $C_{1,a}$. On the other hand, if $\lim_{T\to\infty}s_0^2=1$, then:
\begin{align*}
    \lim_{T\to\infty} \frac{f(s_0^2)}{(s_0^2 
- 1)^2} = \frac{1}{2}
\end{align*}
so for large $T$, we will have: 
\begin{align*}
    \frac{f(s_0^2)}{(s_0^2-1)^2} - \frac{1}{2} > - \frac{1}{4} \implies f(s_0^2) > \frac{(s_0^2-1)^2}{4}
\end{align*}
and thus:
\begin{align*}
    \frac{f(s_0^2)}{2}  - \frac{C_1\overline{s}^2 (s_0^2-1)^2}{\sqrt{C_{1,a}}} \geq \frac{f(s_0^2)}{2}\left(1  - \frac{C_1\overline{s}^2 }{2\sqrt{C_{1,a}}}\right).
\end{align*}
So for $C_{1,a} > C_1^2 \overline{s}^2/4$ bound (\ref{eq:thm4-deriv1}) is positive.

Having established that bound (\ref{eq:thm4-cs1-bd3}) is increasing as a function of $t-t_0$, then (\ref{eq:thm4-cs1-bd3}) is minimized by setting:
\begin{align*}
    t &= \inf \;\{ t\in\mathcal{T}_{a_T} : t>t_0\} \\
    &= t_0 + \frac{C_{1,a}a_T\log^2T}{\max\{\min\{b_0^2,b_0^2/s_0^2\}, (s_0^2-1)^2\}}.
\end{align*}
Suppose first that $\min\{b_0^2,b_0^2/s_0^2\} \geq  (s_0^2-1)^2$, then returning to bound (\ref{eq:thm4-cs1-bd3}), there is some $B_1 > 0$ so that: 
\begin{align*}
    \log \frac{\alpha_{t_0}}{\alpha_t} &\geq \left[\frac{C_{1,a}}{2\min\{1,s_0^{-2}\}} - \left(\frac{|s_0^2-1|}{2\min\{|b_0|,|b_0|/s_0\}} + \frac{s_0}{\min\{1,s_0^{-1}\}}\right)C_1\sqrt{C_{1,a}}\right]a_T\log^2T - B_1a_T\log^2 T \\
    &\geq \left[\frac{\max\{1,s_0^2\}(C_{1,a}-3C_1\sqrt{C_{1,a}})}{2} - B_1\right]a_T\log^2T. \tag{$|s_0^2-1| \leq \min\{|b_0|,|b_0|/s_0\}$}
\end{align*}
So if we set $C_{1,a} > \max\{(3C_1 + 1)^2, 4(C_\pi + B_1)^2\}$, then (\ref{eq:thm4-result}) holds. On the other hand, if $\min\{b_0^2,b_0^2/s_0^2\} \leq (s_0^2-1)^2$, then: 
\begin{align*}
    \log \frac{\alpha_{t_0}}{\alpha_t} &\geq \left[\frac{C_{1,a}f(s_0^2)}{2(s_0^2-1)^2} - \left(\frac{1}{2} + \frac{|b_0|s_0}{|s_0^2-1|}\right)C_1\sqrt{C_{1,a}}\right]a_T\log^2T - B_1a_T\log^2 T \\
    &\geq \left[\frac{C_{1,a}f(s_0^2)}{2(s_0^2-1)^2} - \frac{3\max\{s_0^2,1\}C_1\sqrt{C_{1,a}}}{2} - B_1\right]a_T\log^2T. \tag{$|s^2_0-1| \geq \min\{|b_0|,|b_0|/s_0\}$} 
\end{align*}
Again, if $s_0^2 > \overline{\nu}^2 > 1$ for all $T$, then: 
\begin{align*}
    \log \frac{\alpha_{t_0}}{\alpha_t} &\geq \left[\frac{C_{1,a}f(\overline{\nu}^2)}{2(\overline{s}^2-1)^2} - \frac{3\overline{s}^2C_1\sqrt{C_{1,a}}}{2} - B_1\right]a_T\log^2T
\end{align*}
and we can pick $C_{1,a}$ large enough so that the term in square brackets is greater than $C_\pi$. An identical argument gives the same result when $s_0^2 < \underline{\nu}^2$ for some $\underline{\nu}<1$. When $\lim_{T\to\infty}s_0^2=1$, then again for large $T$, we will have: 
\begin{align*}
    \frac{f(s_0^2)}{(s_0^2-1)^2} \to \frac{1}{2} > - \frac{1}{4} \implies \frac{f(s_0^2)}{(s_0^2-1)^2} - \frac{1}{2} > \frac{1}{4}
\end{align*}
and $\max\{s_0^2,1\} < 3/2$, so: 
\begin{align*}
    \log \frac{\alpha_{t_0}}{\alpha_t} &\geq\left[\frac{C_{1,a}}{8} - \frac{9C_1\sqrt{C_{1,a}}}{4} - B_1\right]a_T\log^2T
\end{align*}
and again we can pick $C_{1,a}$ large enough so that the term in square brackets is greater than $C_\pi$. 

\subsubsection*{Case 2: $t < t_0$, and $t \in \mathcal{T}_{a_T}$.}

Using the identical argument we used to get (\ref{eq:thm3-cs2-bd1}) in the proof of Theorem \ref{theorem:smscp} in Appendix \ref{app:localization-smscp} we can begin from:
\begin{align}
    \log \frac{\alpha_{t_0}}{\alpha_t} &> - \frac{1}{2} \sum_{t' = t}^{t_0 - 1} (z_{t'}^2-1) - \frac{(T-t_0+1)}{2}\log s_0^2 \notag  \\
    &\quad\: + \left(\frac{T - t +1}{2}\right)\log\left[\frac{\sum_{t'=t}^T y_{t'}^2}{T-t+1} - \left(\frac{\sum_{t'=t}^T y_{t'}}{T-t+1}\right)^2 \right] \notag \\
    &\quad\: -\frac{1}{2}\sum_{t'=t_0}^T (z_{t'}^2-1) +\frac{1}{2(T-t_0+1)} \left(\sum_{t'=t_0}^T z_{t'}\right)^2 \notag \\
    &\quad\: + \mathcal{O}(1). \label{eq:thm4-cs2-bd1}
\end{align}
For notational convenience, we can define the error term:
\footnotesize
\begin{align}
    \xi(T,t_0,t,s_0) &:= \frac{\sum_{t'=t}^{t_0-1} (z_{t'}^2 -1)}{T-t+1} - \left(\frac{\sum_{t'=t}^{t_0-1} z_{t'}}{T-t+1}\right)^2 + s_0^2\left[\frac{\sum_{t'=t_0}^{T} (z_{t'}^2-1)}{T-t+1} - \left(\frac{\sum_{t'=t_0}^{T} z_{t'}}{T-t+1}\right)^2\right] - \frac{2s_0\sum_{t'=t}^{t_0-1} z_{t'}\sum_{t'=t_0}^{T} z_{t'}}{(T-t+1)^2} \label{eq:thm4-error}
\end{align}
\normalsize
then we can write:
\small
\begin{align*}
    \frac{\sum_{t'=t}^T y_{t'}^2}{T-t+1} - \left(\frac{\sum_{t'=t}^T y_{t'}}{T-t+1}\right)^2 &= \frac{\sum_{t'=t}^{t_0-1} z_{t'}^2 + \sum_{t'=t_0}^{T} (s_0z_{t'} + b_0)^2}{T-t+1} - \left[\frac{\sum_{t'=t}^{t_0-1} z_{t'} + \sum_{t'=t_0}^{T} (s_0z_{t'} + b_0)}{T-t+1}\right]^2 \notag \\
    &= 1 + \frac{(T-t_0+1)(t_0-t)b_0^2}{(T-t+1)^2} + \frac{(T-t_0+1)(s_0^2-1)}{T-t+1} \notag\\
    &\quad + \frac{\sum_{t'=t}^{t_0-1} (z_{t'}^2 -1)}{T-t+1} - \left(\frac{\sum_{t'=t}^{t_0-1} z_{t'}}{T-t+1}\right)^2 +s_0^2\left[\frac{\sum_{t'=t_0}^{T} (z_{t'}^2-1)}{T-t+1} - \left(\frac{\sum_{t'=t_0}^{T} z_{t'}}{T-t+1}\right)^2\right] \notag\\
    &\quad + \frac{2}{(T-t+1)^2}\left[(t_0-t)b_0s_0\sum_{t'=t_0}^Tz_{t'} - s_0\sum_{t'=t}^{t_0-1} z_{t'}\sum_{t'=t_0}^{T} z_{t'} - (T-t_0+1)b_0\sum_{t'=t}^{t_0-1} z_{t'}\right]  \\
    &= 1 + \frac{(T-t_0+1)(t_0-t)b_0^2}{(T-t+1)^2} + \frac{(T-t_0+1)(s_0^2-1)}{T-t+1} \\
    &\quad + \frac{2}{(T-t+1)^2}\left[(t_0-t)b_0s_0\sum_{t'=t_0}^Tz_{t'} - (T-t_0+1)b_0\sum_{t'=t}^{t_0-1} z_{t'}\right] + \xi(T,t_0,t,s_0)
\end{align*}
\normalsize
On $\mathcal{E}$, for some constant $C_1>0$ we have:
\begin{align*}
    \left|\sum_{t'=t}^{t_0-1} z_{t'}\right| &\leq C_1\sqrt{a_T(t_0-t)}\log T \tag{on $\mathcal{E}_1\subset\mathcal{E}$} \\
    \left|b_0\sum_{t'=t_0}^T z_{t'}\right| &\leq \frac{(T-t_0+1)b_0^2 }{4} \tag{on $\mathcal{E}_5\subset\mathcal{E}$}
\end{align*}
and thus:
\small
\begin{align*}
    \frac{\sum_{t'=t}^T y_{t'}^2}{T-t+1} - \left(\frac{\sum_{t'=t}^T y_{t'}}{T-t+1}\right)^2 &\geq 1 + \frac{(T-t_0+1)(t_0-t)b_0^2}{2(T-t+1)^2} + \frac{(T-t_0+1)(s_0^2-1)}{T-t+1} \\
    &\quad - \frac{2(T-t_0+1)|b_0|C_1\sqrt{a_T(t_0-t)}\log T }{(T-t+1)^2}  + \xi(T,t_0,t,s_0)
\end{align*}
\normalsize
Returning to bound (\ref{eq:thm4-cs2-bd1}), we now have:
\begin{align}
    \log \frac{\alpha_{t_0}}{\alpha_t} &> - \frac{1}{2} \sum_{t' = t}^{t_0 - 1} (z_{t'}^2-1) - \frac{(T-t_0+1)}{2}\log s_0^2  + \left(\frac{T - t +1}{2}\right)\log\left[1 + \frac{\frac{(T-t_0+1)(s_0^2-1)}{T-t+1}}{1 + \xi(T,t_0,t,s_0)}\right] \notag  \\
    &\quad\: + \left(\frac{T - t +1}{2}\right)\log\left[1 + \frac{\frac{(T-t_0+1)(t_0-t)b_0^2 }{2(T-t+1)^2} - \frac{2(T-t_0+1)|b_0|C_1\sqrt{a_T(t_0-t)}\log T }{(T-t+1)^2}}{1 + \frac{(T-t_0+1)(s_0^2-1)}{T-t+1} + \xi(T,t_0,t,s_0)}\right] \notag \\
    &\quad\: -\frac{1}{2}\sum_{t'=t_0}^T (z_{t'}^2-1) +\frac{1}{2(T-t_0+1)} \left(\sum_{t'=t_0}^T z_{t'}\right)^2 + \left(\frac{T - t +1}{2}\right)\log\left[1 + \xi(T,t_0,t,s_0)\right] \notag \\
    &\quad\: + \mathcal{O}(1). \label{eq:thm4-cs2-bd3}
\end{align}

As in the proof of Theorem \ref{theorem:smscp}, we treat the case of $t_0 - t \geq T-t_0+1$ and $t_0-t < T-t_0+1$ separately. 

\subsubsection*{Case 2.a: $t_0 - t \leq T-t_0+1$.}
We can start by rewriting bound (\ref{eq:thm4-cs2-bd3}) as:
\begin{align}
    \log \frac{\alpha_{t_0}}{\alpha_t} &> - \frac{1}{2} \sum_{t' = t}^{t_0 - 1} (z_{t'}^2-1) + \frac{(t_0-t)}{2}\log s_0^2  + \left(\frac{T - t +1}{2}\right)\log\left[1 + \frac{\frac{(t_0-t)(s_0^{-2}-1)}{T-t+1}}{1 + s_0^{-2}\xi(T,t_0,t,s_0)}\right] \notag  \\
    &\quad\: + \left(\frac{T - t +1}{2}\right)\log\left[1 + \frac{\frac{(T-t_0+1)(t_0-t)b_0^2 }{2(T-t+1)^2} - \frac{2(T-t_0+1)|b_0|C_1\sqrt{a_T(t_0-t)}\log T }{(T-t+1)^2}}{1 + \frac{(T-t_0+1)(s_0^2-1)}{T-t+1} + \xi(T,t_0,t,s_0)}\right] \notag \\
    &\quad\: -\frac{1}{2}\sum_{t'=t_0}^T (z_{t'}^2-1) +\frac{1}{2(T-t_0+1)} \left(\sum_{t'=t_0}^T z_{t'}\right)^2 + \left(\frac{T - t +1}{2}\right)\log\left[1 + s_0^{-2}\xi(T,t_0,t,s_0)\right] \notag \\
    &\quad\: + \mathcal{O}(1). \label{eq:thm4-cs2-a-bd1}
\end{align}
For some $C_1\geq1$, on $\mathcal{E}$ we have: 
\begin{align*}
    \max\left\{\left|\sum_{t'=t}^{t_0-1} (z_{t'}^2 -1)\right|, \left|\sum_{t'=t_0}^{T} (z_{t'}^2 -1)\right|, \left|\sum_{t'=t}^{t_0-1} z_{t'}\right|, \left|\sum_{t'=t_0}^{T} z_{t'} \right|\right\} &\leq C_1\sqrt{a_T(T-t+1)}\log T \tag{on $\cap_{i=1}^4\mathcal{E}_i\subset\mathcal{E}$}
\end{align*}
and thus:
\begin{align*}
    \left|s_0^{-2}\xi(T,t_0,t,s_0)\right| &\leq (1+s_0^{-2} + 2s_0^{-1})\left[\frac{C_1\sqrt{a_T}\log T}{\sqrt{T-t+1}} + \frac{C_1^2a_T\log^2T}{T-t+1}\right] \\
    &\leq\frac{8\underline{s}^2C_1^2}{C_{2,a}}. \tag{$s_0^2 > \underline{s}^2$ and $T-t+1\geq C_{2,a}a_T\log^2T$}
\end{align*}
So for $C_{2,a} > 8\underline{s}^2C_1^2$, we have $|s_0^{-2}\xi(T,t_0,t,s_0)| > -1$ and we can use the inequality $\log(1+x) > x/(x+1)$ for $x > -1$ to write:
\footnotesize
\begin{align*}
    \left(\frac{T - t +1}{2}\right)\log\left[1 + s_0^{-2}\xi(T,t_0,t,s_0)\right] &\geq \left(\frac{T - t +1}{2}\right) \left[s_0^{-2}\xi(T,t_0,t,s_0) - \frac{s_0^{-4}\xi(T,t_0,t,s_0)^2}{1+s_0^{-2}\xi(T,t_0,t,s_0)}\right] \\
    &= \frac{1}{2s_0^2}\sum_{t'=t}^{t_0-1}(z_{t'}^2 - 1) - \frac{\left(\sum_{t'=t}^{t_0-1}z_{t'}\right)^2 }{2s_0^2(T-t+1)} + \frac{1}{2}\sum_{t'=t_0}^T (z_{t'}^2 - 1) -  \frac{\left(\sum_{t'=t_0}^{T}z_{t'}\right)^2 }{2(T-t+1)}  \\
    &\quad\; - \frac{\sum_{t'=t_0}^{T}z_{t'} \sum_{t'=t}^{t_0-1}z_{t'}}{s_0(T-t+1)} - \left(\frac{T - t +1}{2}\right)\frac{s_0^{-4}\xi(T,t_0,t,s_0)^2}{1+s_0^{-2}\xi(T,t_0,t,s_0)}
\end{align*}
\normalsize
We also have:
\begin{align*}
    \frac{\left(\sum_{t'=t}^{t_0-1}z_{t'}\right)^2 }{2s_0^2(T-t+1)} &\leq \frac{C_1^2}{2\underline{s}^2}a_T\log^2T \\
    \frac{\sum_{t'=t_0}^{T}z_{t'} \sum_{t'=t}^{t_0-1}z_{t'}}{s_0(T-t+1)} &\leq  \frac{C_1^2}{\underline{s}}a_T\log^2T.
\end{align*}
For $C_{2,a} > 16\underline{s}^2C_1^2$:
\begin{align*}
    \left(\frac{T - t +1}{2}\right)\frac{s_0^{-4}\xi(T,t_0,t,s_0)^2}{1+s_0^{-2}\xi(T,t_0,t,s_0)} &\leq (T - t +1)(1+s_0^{-2} + 2s_0^{-1})^2\left[\frac{C_1\sqrt{a_T}\log T}{\sqrt{T-t+1}} + \frac{C_1^2a_T\log^2T}{T-t+1}\right]^2 \\
    &\leq  \frac{16 \underline{s}^2C_1^4a_T^2\log^4T}{T-t+1} \tag{$s_0^2 \geq \underline{s}^2$} \\
    &\leq C_1^2a_T\log^2T. \tag{$T-t+1\geq C_{2,a}a_T\log^2T > 16\underline{s}^2C_1^2a_T\log^2T$} 
\end{align*}
Furthermore, since $T-t+1\geq T-t_0+1$, we have:
\begin{align*}
    \frac{\left(\sum_{t'=t_0}^{T}z_{t'}\right)^2 }{2(T-t_0+1)} -  \frac{\left(\sum_{t'=t_0}^{T}z_{t'}\right)^2 }{2(T-t+1)} \geq 0
\end{align*}
so returning to bound (\ref{eq:thm4-cs2-a-bd1}) we have: 
\begin{align}
    \log \frac{\alpha_{t_0}}{\alpha_t} &>  \frac{(s_0^{-2}-1)}{2} \sum_{t' = t}^{t_0 - 1} (z_{t'}^2-1) + \frac{(t_0-t)}{2}\log s_0^2 + \left(\frac{T - t +1}{2}\right)\log\left[1 + \frac{\frac{(t_0-t)(s_0^{-2}-1)}{T-t+1}}{1 + s_0^{-2}\xi(T,t_0,t,s_0)}\right] \notag  \\
    &\quad\: + \left(\frac{T - t +1}{2}\right)\log\left[1 + \frac{\frac{(T-t_0+1)(t_0-t)b_0^2 }{2(T-t+1)^2} - \frac{2(T-t_0+1)|b_0|C_1\sqrt{a_T(t_0-t)}\log T }{(T-t+1)^2}}{1 + \frac{(T-t_0+1)(s_0^2-1)}{T-t+1} + \xi(T,t_0,t,s_0)}\right] \notag \\
    &\quad\: + \mathcal{O}(a_T\log^2 T). \label{eq:thm4-cs2-a-bd2}
\end{align}
Note that if $s_0^2 = 1$, then we can completely ignore the first line of (\ref{eq:thm4-cs2-a-bd2}), on the other hand, if $s_0^2\neq 1$, then by assumption $(s_0^2 -1)^2(T-t_0+1)\gg a_T\log^2T$. However, $0<\underline{s}^2\leq s_0^2\leq\overline{s}^2<\infty$ implies $(s_0^2 -1)^2 = \mathcal{O}(1)$, so we must have $T-t_0+1\gg a_T\log^2T$. We showed above that on $\mathcal{E}$:
\begin{align*}
    \left|s_0^{-2}\xi(T,t_0,t,s_0)\right| &\leq (1+s_0^{-2} + 2s_0^{-1})\left[\frac{C_1\sqrt{a_T}\log T}{\sqrt{T-t+1}} + \frac{C_1^2a_T\log^2T}{T-t+1}\right] \\
    &\leq 4\underline{s}^2\left[\frac{C_1\sqrt{a_T}\log T}{\sqrt{T-t_0+1}} + \frac{C_1^2a_T\log^2T}{T-t_0+1}\right] \tag{$T-t+1\geq T-t_0+1$}
\end{align*}
i.e., $\left|s_0^{-2}\xi(T,t_0,t,s_0)\right| \to 0$ on $\mathcal{E}$ when $s_0^2\neq 1$. Next, we can write:
\footnotesize
\begin{align*}
    \left(\frac{T - t +1}{2}\right)\log\left[1 + \frac{\frac{(t_0-t)(s_0^{-2}-1)}{T-t+1}}{1 + s_0^{-2}\xi(T,t_0,t,s_0)}\right]&= \left(\frac{T - t +1}{2}\right)\log\left[1 + \frac{(t_0-t)(s_0^{-2}-1)}{T-t+1}\right] \\
    &\quad\:+ \left(\frac{T - t +1}{2}\right)\log\left[1 - \left(\frac{ s_0^{-2}\xi(T,t_0,t,s_0)}{1 + s_0^{-2}\xi(T,t_0,t,s_0)}\right)\left(\frac{\frac{(t_0-t)(s_0^{-2}-1)}{T-t+1}}{1+\frac{(t_0-t)(s_0^{-2}-1)}{T-t+1}}\right)\right].
\end{align*}
\normalsize
We have just shown that:
\begin{align*}
    \frac{ s_0^{-2}\xi(T,t_0,t,s_0)}{1 + s_0^{-2}\xi(T,t_0,t,s_0)} &= \mathcal{O}\left(\frac{\sqrt{a_T}\log T}{\sqrt{T-t_0+1}}\right), \\
    \frac{\frac{(t_0-t)(s_0^{-2}-1)}{T-t+1}}{1+\frac{(t_0-t)(s_0^{-2}-1)}{T-t+1}} &= \mathcal{O}\left(\frac{(t_0-t)(s_0^{-2}-1)}{T-t+1}\right)= \mathcal{O}(1) \tag{$s_0^{-2} = \mathcal{O}(1)$}
\end{align*}
so for some constant $C_2 > 0$, a first order Taylor approximation of $\log(1+x)$ around zero gives:
\footnotesize
\begin{align*}
     \left(\frac{T - t +1}{2}\right)\log\left[1 - \left(\frac{ \xi(T,t_0,t,s_0)}{1 + \xi(T,t_0,t,s_0)}\right)\left(\frac{\frac{(t_0-t)(s_0^{-2}-1)}{T-t+1}}{1+\frac{(t_0-t)(s_0^{-2}-1)}{T-t_0+1}}\right)\right] &= (t_0-t)(s_0^{-2}-1)\mathcal{O}\left(\frac{\sqrt{a_T}\log T}{\sqrt{T-t_0+1}}\right) + \mathcal{O}(a_T\log^2 T) \\
     &= (s_0^{-2}-1)\mathcal{O}\left(\sqrt{a_T(t_0-t)}\log T\right) + \mathcal{O}(a_T\log^2 T) \tag{$t_0-t\leq T-t_0+1$} \\
     &\geq -\frac{C_2}{2}|s_0^{-2}-1|\sqrt{a_T(t_0-t)}\log T + \mathcal{O}(a_T\log^2 T)
\end{align*}
\normalsize
We also have:
\begin{align*}
    \left|\sum_{t' = t}^{t_0 - 1} (z_{t'}^2-1)\right| < C_1\sqrt{a_T(t_0-t)}\log T \tag{on $\mathcal{E}_2\subset\mathcal{E}$}
\end{align*}
and since $C_1 > C_2$ WLOG, then we can now lower bound (\ref{eq:thm4-cs2-a-bd2}) with:
\begin{align}
    \log \frac{\alpha_{t_0}}{\alpha_t} &>  -C_1|s_0^{-2}-1|\sqrt{a_T(t_0-t)}\log T + \frac{(t_0-t)}{2}\log s_0^2 + \left(\frac{T - t +1}{2}\right)\log\left[1 + \frac{(t_0-t)(s_0^{-2}-1)}{T-t+1}\right] \notag  \\
    &\quad\: + \left(\frac{T - t +1}{2}\right)\log\left[1 + \frac{\frac{(T-t_0+1)(t_0-t)b_0^2 }{2(T-t+1)^2} - \frac{2(T-t_0+1)|b_0|C_1\sqrt{a_T(t_0-t)}\log T }{(T-t+1)^2}}{1 + \frac{(T-t_0+1)(s_0^2-1)}{T-t+1} + \xi(T,t_0,t,s_0)}\right] \notag \\
    &\quad\: + \mathcal{O}(a_T\log^2 T). \label{eq:thm4-cs2-a-bd3}
\end{align}

\subsubsection*{Case 2.a.i: $\min\left\{b_0^2,b_0^2/s_0^2\right\} \geq (s_0^2-1)^2$}

Suppose that:
\begin{align*}
    \min\left\{b_0^2,\frac{b_0^2}{s_0^2}\right\} \geq (s_0^2-1)^2 &\implies |t-t_0| \geq \frac{C_{1,a} a_T\log^2T}{\min\{b_0^2,b_0^2/s_0^{2}\}}, \;\forall\; t\in\mathcal{T}_{a_T}.
\end{align*}
For $C_{1,a} > 64C_1^2$, we can write:
\small
\begin{align*}
    \frac{(T-t_0+1)(t_0-t)b_0^2 }{2(T-t+1)^2} - \frac{2(T-t_0+1)|b_0|C_1\sqrt{a_T(t_0-t)}\log T }{(T-t+1)^2} &= \frac{(T-t_0+1)(t_0-t)b_0^2 }{(T-t+1)^2} \left[\frac{1}{2} - \frac{2C_1\sqrt{a_T}\log T}{|b_0|\sqrt{t_0-t}}\right] \\
    &\geq \frac{(T-t_0+1)(t_0-t)b_0^2 }{(T-t+1)^2} \left[\frac{1}{2} - \frac{2C_1}{\sqrt{C_{1,a}}}\right] \tag{$|t-t_0| \geq \frac{C_{1,a} a_T\log^2T}{\min\{b_0^2,b_0^2/s_0^{2}\}}$} \\
    &> \frac{(T-t_0+1)(t_0-t)b_0^2 }{4(T-t+1)^2} . \tag{$C_{1,a} > 64C_1^2$}
\end{align*}
\normalsize
Furthermore, we have shown that on $\mathcal{E}$:
\begin{align*}
    \left|\xi(T,t_0,t,s_0)\right| &\leq (1+s_0^2 + 2s_0)\left[\frac{C_1\sqrt{a_T}\log T}{\sqrt{T-t+1}} + \frac{C_1^2a_T\log^2T}{T-t+1}\right]  \\
    &\leq  \frac{8\overline{s}^2C_1^2}{C_{2,a}} \tag{$s_0^2 < \overline{s}^2$ and $T-t+1\geq C_{2,a}a_T\log^2T$}
\end{align*}
and we have:
\begin{align*}
    -1 <\underline{s}^2 - 1 <\frac{(T-t_0+1)(s_0^2-1)}{T-t+1} < \overline{s}^2 -1 < \infty.
\end{align*}
So for $C_{2,a} > 8C_1^2$, in this case we have:
\begin{align*}
    \log \frac{\alpha_{t_0}}{\alpha_t} &>  -C_1|s_0^{-2}-1|\sqrt{a_T(t_0-t)}\log T + \frac{(t_0-t)}{2}\log s_0^2 + \left(\frac{T - t +1}{2}\right)\log\left[1 + \frac{(t_0-t)(s_0^{-2}-1)}{T-t+1}\right] \notag  \\
    &\quad\: + \left(\frac{T - t +1}{2}\right)\log\left[1 + \frac{(T-t_0+1)(t_0-t)b_0^2 }{16\overline{s}^2(T-t+1)^2}\right] + \mathcal{O}(a_T\log^2 T)
\end{align*}
Note that the term:
\begin{align*}
    \frac{(t_0-t)}{2}\log s_0^2 + \left(\frac{T - t +1}{2}\right)\log\left[1 + \frac{(t_0-t)(s_0^{-2}-1)}{T-t+1}\right]
\end{align*}
is decreasing in $t$ and is clearly equal to zero when $t=t_0$. To see this we can take the derivative with respect $t$ to get:
\begin{align*}
    \frac{1}{2}\left[ \frac{\frac{(T-t_0+1)(s_0^2 -1)}{T-t+1}}{1 + \frac{(T-t_0+1)(s_0^2 -1)}{T-t+1}}- \log\left[1 + \frac{(T-t_0+1)(s_0^2 -1)}{T-t+1}\right] \right] &\leq 0. \tag{$\log(1+x) > x/(1+x)$ for $x>-1$}
\end{align*}
Therefore:
\small
\begin{align}
    \log \frac{\alpha_{t_0}}{\alpha_t} &>  -C_1|s_0^{-2}-1|\sqrt{a_T(t_0-t)}\log T + \left(\frac{T - t +1}{2}\right)\log\left[1 + \frac{(T-t_0+1)(t_0-t)b_0^2 }{16\overline{s}^2(T-t+1)^2}\right] + \mathcal{O}(a_T\log^2 T) \label{eq:thm4-cs2-a-i-bd1}
\end{align}
\normalsize
Suppose first that $\lim_{T\to\infty}b_0^2 = 0$, then for large $T$:
\footnotesize
\begin{align*}
    \lim_{T\to\infty}\frac{(T-t_0+1)(t_0-t)b_0^2 }{16\overline{s}^2(T-t+1)^2} =0 \implies  \log\left[1 + \frac{(T-t_0+1)(t_0-t)b_0^2 }{16\overline{s}^2(T-t+1)^2}\right] &= \frac{(T-t_0+1)(t_0-t)b_0^2 }{16\overline{s}^2(T-t+1)^2} + \mathcal{O}\left(\frac{(T-t_0+1)(t_0-t)^2b_0^4}{(T-t+1)^2}\right) \\
    &\geq \frac{(T-t_0+1)(t_0-t)b_0^2 }{32\overline{s}^2(T-t+1)^2}.
\end{align*}
\normalsize
So for large $T$, we can lower bound (\ref{eq:thm4-cs2-a-i-bd1}) in this case with:
\begin{align*}
     \log \frac{\alpha_{t_0}}{\alpha_t} &>  -C_1|s_0^{-2}-1|\sqrt{a_T(t_0-t)}\log T +\frac{(T-t_0+1)(t_0-t)b_0^2 }{64\overline{s}^2(T-t+1)} + \mathcal{O}(a_T\log^2 T) \\
     &\geq (t_0-t)b_0^2\left[\frac{(T-t_0+1)}{64\overline{s}^2(T-t+1)} -\frac{C_1s_0^{-2} \min\{1,s_0^{-1}\}\sqrt{a_T}\log T}{|b_0|\sqrt{t_0-t}}\right] + \mathcal{O}(a_T\log^2 T) \tag{$(s_0^2-1)^2 \leq\min\{b_0^2,b_0^2/s_0^2\}$} \\
     &\geq (t_0-t)b_0^2\left[\frac{1}{128\overline{s}^2} -\frac{C_1\underline{s}^{-2}}{\sqrt{C_{1,a}}}\right] + \mathcal{O}(a_T\log^2 T) \tag{$|t_0-t| \geq \frac{C_{1,a}a_T\log^2 T}{\min\{b_0^2,b_0^2/s_0^2\}}$}
\end{align*}
so for $C_{1,a} > (128C_1 \overline{s}^2/\underline{s}^2)^2$ and some $B_1 > 0$, we have:
\begin{align*}
     \log \frac{\alpha_{t_0}}{\alpha_t} &> (t_0-t)b_0^2 + \mathcal{O}(a_T\log^2 T) \\
     &\geq (C_{1,a} - B_1)a_T\log^2 T  \tag{$|t_0-t| \geq \frac{C_{1,a}a_T\log^2 T}{\min\{b_0^2,b_0^2/s_0^2\}}$}
\end{align*}
and thus (\ref{eq:thm4-result}) holds for $C_{1,a} > B_1+C_\pi$.

Next, if $b_0^2 > \underline{b}^2$ for some $\underline{b}^2 > 0$, then we can show that the bound (\ref{eq:thm4-cs2-a-i-bd1}) is decreasing in $t$. We verify below that the second term is decreasing in $t$, and the first term disappears when $s_0^2=1$, so it is without loss of generality to assume that $s_0^2\neq 1$. Taking the derivative with respect to $t$ yields:
\small
\begin{align}
    &\frac{1}{2}\left[\frac{C_1|s_0^{-2}-1|\sqrt{a_T}\log T}{\sqrt{t_0-t}} - \log\left[1 + \frac{(T-t_0+1)(t_0-t)b_0^2 }{16\overline{s}^2(T-t+1)^2}\right] +\frac{\frac{(T-t_0+1)(t_0-t)b_0^2 }{8\overline{s}^2(T-t+1)^2} - \frac{(T-t_0+1)b_0^2 }{16\overline{s}^2(T-t+1)}}{1 + \frac{(T-t_0+1)(t_0-t)b_0^2 }{16\overline{s}^2(T-t+1)^2}}\right] \label{eq:case-2-i-deriv}
\end{align}
\normalsize
Suppose that $b_0^2\to\infty$ as $T\to\infty$, then since we can write:
\begin{align*}
     - \log\left[1 + \frac{(T-t_0+1)(t_0-t)b_0^2 }{16\overline{s}^2(T-t+1)^2}\right] +\frac{\frac{(T-t_0+1)(t_0-t)b_0^2 }{8\overline{s}^2(T-t+1)^2} - \frac{(T-t_0+1)b_0^2 }{16\overline{s}^2(T-t+1)}}{1 + \frac{(T-t_0+1)(t_0-t)b_0^2 }{16\overline{s}^2(T-t+1)^2}} \leq -\frac{\frac{(T-t_0+1)^2b_0^2 }{16\overline{s}^2(T-t+1)^2}}{1 + \frac{(T-t_0+1)(t_0-t)b_0^2 }{16\overline{s}^2(T-t+1)^2}} \tag{$\log(1+x)\geq x/(1+x)$ for $x>-1$}
\end{align*}
and:
\begin{align*}
    \frac{(T-t_0+1)^2b_0^2 }{(T-t+1)^2} = \frac{b_0^2}{\left(1+\frac{t_0-t}{T-t_0+1}\right)^2} \geq \frac{b_0^2}{4}. \tag{$t_0-t \leq T-t_0+1$}
\end{align*}
Then for large $T$: 
\begin{align*}
    - \log\left[1 + \frac{(T-t_0+1)(t_0-t)b_0^2 }{16\overline{s}^2(T-t+1)^2}\right] +\frac{\frac{(T-t_0+1)(t_0-t)b_0^2 }{8\overline{s}^2(T-t+1)^2} - \frac{(T-t_0+1)b_0^2 }{16\overline{s}^2(T-t+1)}}{1 + \frac{(T-t_0+1)(t_0-t)b_0^2 }{16\overline{s}^2(T-t+1)^2}} \lesssim -\frac{T-t_0+1}{t_0-t}.
\end{align*}
At the same time we have:
\begin{align*}
    \frac{\frac{T-t_0+1}{t_0-t}}{\frac{C_1|s_0^{-2}-1|\sqrt{a_T}\log T}{\sqrt{t_0-t}}} \geq \frac{\sqrt{T-t_0+1}}{C_1|\underline{s}^{-2}-1|\sqrt{a_T}\log T}. \tag{$t_0-t \leq T-t_0+1$}
\end{align*}
Since $s_0^2\neq0$, then by assumption $(s_0^2-1)^2(T-t_0+1) \gg a_T\log^2T$, but since $(s_0^2-1)^2 = \mathcal{O}(1)$, then $T-t_0+1 \gg a_T\log^2T$ and the lower bound above is diverging as $T\to\infty$, and thus the first term in (\ref{eq:case-2-i-deriv}) is dominated by the remaining terms and we get that (\ref{eq:case-2-i-deriv}) is negative for large $T$. 

On the other hand, if $b_0^2 =\mathcal{O}(1)$ so that $b_0^2 < \overline{b}^2$ for some $\overline{b}^2 < \infty$, then since the first term of (\ref{eq:case-2-i-deriv}) is decreasing in $t$, we have:
\begin{align*}
    \frac{C_1|s_0^{-2}-1|\sqrt{a_T}\log T}{\sqrt{t_0-t}}  &\leq \frac{C_1\min\{|b_0|,|b_0|/s_0\}\sqrt{a_T}\log T}{s_0^{2}\sqrt{t_0-t}}  \tag{$(s_0^2-1)^2 \leq\min\{b_0^2,b_0^2/s_0^2\}$} \\
    &\leq \frac{C_1\min\{|b_0|,|b_0|/s_0\}\sqrt{a_T}\log T}{s_0^{2}\sqrt{t_0-t}} \bigg|_{t=\sup\{t\in\mathcal{T}_{a_T}:t\leq t_0,\;t_0-t\leq T-t_0+1\}} \\
    &= \frac{C_1\min\{b_0^2,b_0^2/s^2_0\}}{s_0^{2}\sqrt{C_{1,a}}} \\
    &\leq \frac{C_1 \overline{b}^2 }{\underline{s}^4\sqrt{C_{1,a}}} \tag{$|b_0|\leq \overline{b} < \infty$}
\end{align*}
while defining:
\begin{align*}
     h(t):=\frac{(T-t_0+1)(t_0-t)b_0^2 }{16\overline{s}^2(T-t+1)}
\end{align*}
and taking the derivative of the other terms with respect to $t$ yields: 
\begin{align*}
    \frac{h''(t)}{1+ \frac{h(t)}{T-t+1}} - \frac{\left[\frac{h(t)}{T-t+1} + h'(t)\right]^2}{(T-t+1)\left[1+\frac{h(t)}{T-t+1}\right]^2}. 
\end{align*}
Since:
\begin{align*}
    h''(t) &= -\frac{(T-t_0+1)^2b_0^2 }{8\overline{s}^2(T-t+1)^3} < 0
\end{align*}
then:
\small
\begin{align*}
    &- \log\left[1 + \frac{(T-t_0+1)(t_0-t)b_0^2 }{16\overline{s}^2(T-t+1)^2}\right] +\frac{\frac{(T-t_0+1)(t_0-t)b_0^2 }{8\overline{s}^2(T-t+1)^2} - \frac{(T-t_0+1)b_0^2 }{16\overline{s}^2(T-t+1)}}{1 + \frac{(T-t_0+1)(t_0-t)b_0^2 }{16\overline{s}^2(T-t+1)^2}} \\
    &\leq - \log\left[1 + \frac{(T-t_0+1)(t_0-t)b_0^2 }{16\overline{s}^2(T-t+1)^2}\right] +\frac{\frac{(T-t_0+1)(t_0-t)b_0^2 }{8\overline{s}^2(T-t+1)^2} - \frac{(T-t_0+1)b_0^2 }{16\overline{s}^2(T-t+1)}}{1 + \frac{(T-t_0+1)(t_0-t)b_0^2 }{16\overline{s}^2(T-t+1)^2}}\bigg|_{t=\inf\{t\in\mathcal{T}_{a_T}:t\leq t_0,\;t_0-t\leq T-t_0+1\}} \\
    &= - \log\left[1 + \frac{b_0^2 }{32\overline{s}^2}\right] +\frac{\frac{b_0^2 }{32\overline{s}^2}}{1 + \frac{b_0^2 }{32\overline{s}^2}} \\
    &\leq - \log\left[1 + \frac{\underline{b}^2 }{32\overline{s}^2}\right] +\frac{\frac{\underline{b}^2 }{32\overline{s}^2}}{1 + \frac{\underline{b}^2 }{32\overline{s}^2}} \tag{$0<\underline{b} < |b_0|$}
\end{align*}
\normalsize
Once more, $\log(1+x) > x/(1+x)$ for $x>-1$, and since the last bound above is fixed with respect to $T$, then there is some $b^* > 0$ so that: 
\begin{align*}
     - \log\left[1 + \frac{\underline{b}^2 }{32\overline{s}^2}\right] +\frac{\frac{\underline{b}^2 }{32\overline{s}^2}}{1 + \frac{\underline{b}^2 }{32\overline{s}^2}} < -b^*.
\end{align*}
So if $C_{1,a} > (C_1 \overline{b}^2/ \underline{s}^4b^*)^2$, then (\ref{eq:case-2-i-deriv}) is negative and the bound (\ref{eq:thm4-cs2-a-i-bd1}) is minimized at:
\begin{align*}
    t &= \sup\;\{t \in \mathcal{T}_{a_T}:t<t_0\} \\
    &\leq t_0 - \frac{C_{1,a}a_T\log^2T}{\min\{b_0^2,b_0^2/s_0^2\}}.
\end{align*}
and thus:
\begin{align*}
    \log \frac{\alpha_{t_0}}{\alpha_t} &> -\frac{C_1|s_0^2-1|\sqrt{C_{1,a}}a_T\log^2T}{s_0^2\min\{|b_0|,|b_0|/s_0\}} \\
    &\quad\:+ \frac{T - t_0 +1}{2}\left(1 + \frac{C_{1,a}a_T\log^2 T}{\min\{b_0^2,b_0^2/s_0^2\}(T-t_0+1)}\right)\log\left[ 1 + \frac{\frac{C_{1,a}a_T\log^2 T}{16\overline{s}^2(T-t_0+1)\min\{1,s_0^{-2}\}}}{\left(1 + \frac{C_{1,a} a_T\log^2 T}{\min\{b_0^2,b_0^2/s_0^2\}(T-t_0+1)}\right)^2}\right] \\
    &\quad\:+ \mathcal{O}(a_T\log^2 T).
\end{align*}
By assumption we have:
\begin{align*}
    \min\{b_0^2,b_0^2/s_0^2\}(T-t_0+1) &\gg a_T\log^2 T
\end{align*}
so for any $C_{1,a} > 0$, for large $T$ we have:
\begin{align*}
    \frac{C_{1,a} a_T\log^2 T}{\min\{b_0^2,b_0^2/s_0^2\}(T-t_0+1)} \leq \sqrt{2}-1
\end{align*}
and:
\begin{align*}
    \log \frac{\alpha_{t_0}}{\alpha_t} &> -\frac{C_1|s_0^2-1|\sqrt{C_{1,a}}a_T\log^2T}{s_0^2\min\{|b_0|,|b_0|/s_0\}}+\frac{T - t_0 +1}{2}\log\left[ 1 + \frac{C_{1,a}a_T\log^2 T}{32\overline{s}^2(T-t_0+1)}\right] + \mathcal{O}(a_T\log^2 T) 
\end{align*}
Suppose that $b_0^2 = \mathcal{O}(1)$, then it must be the case that $T-t_0+1 \gg a_T\log^2T$, in which case we can use another first order Taylor approximation to write:
\begin{align*}
    \frac{T - t_0 +1}{2}\log\left[ 1 + \frac{C_{1,a}a_T\log^2 T}{32\overline{s}^2(T-t_0+1)}\right] = \frac{C_{1,a}a_T\log^2 T}{32\overline{s}^2} + \mathcal{O}\left(a_T\log^2T\right)
\end{align*}
and since $(s_0^2 - 1)^2 \leq\min\{|b_0|,|b_0|/s_0\}$, for some $B_1 > 0$ we have:
\begin{align*}
    \log \frac{\alpha_{t_0}}{\alpha_t} &> -\frac{C_1\sqrt{C_{1,a}}a_T\log^2T}{\underline{s}^2}+ \frac{C_{1,a}a_T\log^2 T}{32\overline{s}^2} - B_1 a_T\log^2 T 
\end{align*}
so we can select $C_{1,a} > 0$ large so that (\ref{eq:thm4-result}) holds. On the other hand, if $b_0^2 \to \infty$, then we make use of the assumption:
\begin{align*}
    T-t_0+1 \gtrsim a_T\log^2 T
\end{align*}
to write:
\begin{align*}
    \log \frac{\alpha_{t_0}}{\alpha_t} &> -\frac{C_1\overline{s}(\underline{s}^{-2} + 1)\sqrt{C_{1,a}}a_T\log^2T}{|b_0|}+\frac{B_2a_T\log^2T}{2}\log\left[ 1 + \frac{C_{1,a}}{32\overline{s}^2B_2}\right] -B_1 a_T\log^2 T 
\end{align*}
for some $B_2 > 0$. Since $|b_0| \to \infty$, for large $T$ we
\begin{align*}
    \frac{C_1\overline{s}(\underline{s}^{-2} + 1)}{|b_0|} \leq B_1
\end{align*}
and clearly we can pick $C_{1,a} > 0$ large so that:
\begin{align*}
    \log\left[ 1 + \frac{C_{1,a}}{32\overline{s}^2B_2}\right] > 2B_1 + C_\pi
\end{align*}
so again (\ref{eq:thm4-result}) holds for this case.

\subsubsection*{Case 2.a.ii: $\min\left\{b_0^2,b_0^2/s_0^2\right\} \leq (s_0^2-1)^2$}

Returning to bound (\ref{eq:thm4-cs2-a-bd3}), we can write:
\begin{align}
    \log \frac{\alpha_{t_0}}{\alpha_t} &>  -C_1|s_0^{-2}-1|\sqrt{a_T(t_0-t)}\log T + \frac{(t_0-t)}{2}\log s_0^2 + \left(\frac{T - t +1}{2}\right)\log\left[1 + \frac{(t_0-t)(s_0^{-2}-1)}{T-t+1}\right] \notag  \\
    &\quad\: + \left(\frac{T - t +1}{2}\right)\log\left[1 - \frac{ \frac{2(T-t_0+1)|b_0|C_1\sqrt{a_T(t_0-t)}\log T }{(T-t+1)^2}}{1 + \frac{(T-t_0+1)(s_0^2-1)}{T-t+1} + \xi(T,t_0,t,s_0)}\right] \notag \\
    &\quad\: + \mathcal{O}(a_T\log^2 T). \label{eq:thm4-cs2-a-ii-bd1}
\end{align}
Since $(s_0^2-1)^2 = \mathcal{O}(1)$ by assumption, then:
\begin{align*}
    (s_0^2-1)^2(T-t_0+1) \gg a_T\log^2 T \implies T-t_0+1 \gg a_T\log^2 T.
\end{align*}
We have shown that on $\mathcal{E}$:
\begin{align*}
    \left|\xi(T,t_0,t,s_0)\right| &\leq (1+s_0^2 + 2s_0)\left[\frac{C_1\sqrt{a_T}\log T}{\sqrt{T-t+1}} + \frac{C_1^2a_T\log^2T}{T-t+1}\right]  \\
    &\leq  4 \overline{s}^2\left[\frac{C_1\sqrt{a_T}\log T}{\sqrt{T-t_0+1}} + \frac{C_1^2a_T\log^2T}{T-t_0+1}\right] 
\end{align*}
so for large $T$ we have $\xi(T,t_0,t,s_0) > - \underline{s}^2/2$, and we have:
\begin{align*}
    -1 <\underline{s}^2 - 1 <\frac{(T-t_0+1)(s_0^2-1)}{T-t+1} < \overline{s}^2 -1 < \infty
\end{align*}
so: 
\begin{align}
    \log \frac{\alpha_{t_0}}{\alpha_t} &>  -C_1|s_0^{-2}-1|\sqrt{a_T(t_0-t)}\log T + \frac{(t_0-t)}{2}\log s_0^2 + \left(\frac{T - t +1}{2}\right)\log\left[1 + \frac{(t_0-t)(s_0^{-2}-1)}{T-t+1}\right] \notag  \\
    &\quad\: + \left(\frac{T - t +1}{2}\right)\log\left[1 - \frac{4(T-t_0+1)|b_0|C_1\sqrt{a_T(t_0-t)}\log T }{\underline{s}^2(T-t+1)^2}\right] \notag \\
    &\quad\: + \mathcal{O}(a_T\log^2 T). \label{eq:thm4-cs2-a-ii-bd2}
\end{align}
Next, since $(s_0^2-1)^2 = \mathcal{O}(1)$ and $\min\{b_0^2,b_0^2/s_0^2\}\leq (s_0^2-1)^2$, then then there is some $\overline{b} > 0$ so that:
\begin{align*}
    \frac{4(T-t_0+1)|b_0|C_1\sqrt{a_T(t_0-t)}\log T }{\underline{s}^2(T-t+1)^2} \leq \frac{4\overline{b}C_1\sqrt{a_T}\log T }{\underline{s}^2\sqrt{T-t_0+1}}. \tag{$T-t+1 \geq T-t_0+1\geq t_0-t$}
\end{align*}
Since $T-t_0+1\gg a_T\log^2T$ in this case, then the RHS above is converging to zero and we can use another first order Taylor approximation to write:
\footnotesize
\begin{align*}
    \left(\frac{T - t +1}{2}\right)\log\left[1 - \frac{4(T-t_0+1)|b_0|C_1\sqrt{a_T(t_0-t)}\log T }{\underline{s}^2(T-t+1)^2}\right] &= -\frac{4(T-t_0+1)|b_0|C_1\sqrt{a_T(t_0-t)}\log T }{\underline{s}^2(T-t+1)} + \mathcal{O}\left(a_T\log^2T\right) \\
    &\geq -\frac{4\overline{s}C_1|s_0^2-1|}{\underline{s}^2} \sqrt{a_T(t_0-t)}\log T\tag{$\min\{b_0^2,b_0^2/s_0^2\}\leq (s_0^2-1)^2$} + \mathcal{O}\left(a_T\log^2T\right).
\end{align*}
\normalsize
So for some $C_2 > 8\overline{s}C_1/\underline{s}^2$, we can write:
\small
\begin{align*}
    \log \frac{\alpha_{t_0}}{\alpha_t} &>  -C_2|s_0^{2}-1|\sqrt{a_T(t_0-t)}\log T + \frac{(t_0-t)}{2}\log s_0^2 + \left(\frac{T - t +1}{2}\right)\log\left[1 + \frac{(t_0-t)(s_0^{-2}-1)}{T-t+1}\right] +\mathcal{O}(a_T\log^2 T) \\
    &=h_T(t,t_0) 
\end{align*}
\normalsize
We now show that for any $t_0$, $h_T$ is decreasing in $t$. Analyzing the first derivative of $h_T$ we have:
\small
\begin{align*}
    \frac{\partial h_T}{\partial t} &= \frac{1}{2}\left[ C_2|s_0^{2}-1|\sqrt{\frac{a_T\log^2 T}{t_0-t}}- \log\left[1 + \frac{(s_0^2-1)(T-t_0+1)}{T-t+1}\right] + \frac{\frac{(T-t_0+1)(s_0^{2}-1)}{T-t+1}}{1 + \frac{(s_0^2-1)(T-t_0+1)}{T-t+1} }\right] \\
    &= \frac{1}{2}\left[ C_2|s_0^{2}-1|\sqrt{\frac{a_T\log^2 }{(T-t_0+1)x}}- \log\left[1 + \frac{(s_0^2-1)}{1 + x}\right] + \frac{\frac{(s_0^{2}-1)}{1+x}}{1 + \frac{(s_0^2-1)}{1 + x} }\right] \tag{$x = \frac{t_0-t}{T-t_0+1}$}
\end{align*}
\normalsize
Next, note that the term:
\begin{align*}
     g_1(x) := C_2|s_0^{2}-1|\sqrt{\frac{a_T\log^2 }{(T-t_0+1)x}}
\end{align*}
is convex and decreasing as a function of $x$. On the other hand, if we define:
\begin{align*}
    g_2(x) &:= - \log\left[1 + \frac{(s_0^2-1)}{1 + x}\right] + \frac{\frac{(s_0^{2}-1)}{1+x}}{1 + \frac{(s_0^2-1)}{1 + x} }
\end{align*}
then we have:
\begin{align*}
    \frac{d g_2}{dx} = \frac{\frac{(s_0^{2}-1)^2}{(1+x)^3}}{\left(1 + \frac{(s_0^2-1)}{1 + x} \right)^2} > 0
\end{align*}
so this term is increasing as a function of $x$. Therefore:
\begin{align*}
    g_1(x) + g_2(x) &\leq g_1(x)\bigg|_{x = \inf\:\{x\in[0,1]\::\:t_0-t\leq T-t_0+1, \; t\in\mathcal{T}_{a_T}\} } +  g_2(x)\bigg|_{x = \sup\:\{x\in[0,1]\::\:t_0-t\leq T-t_0+1, \; t\in\mathcal{T}_{a_T}\} } \\
    &\leq g_1\left(\frac{C_{1,a}a_T\log^2 T}{(T-t_0+1)(s_0^2-1)^2}\right) + g_2\left(1\right) \\
    &= \frac{C_2(s_0^{2}-1)^2}{\sqrt{C_{1,a}}}- \log\left[1 + \frac{(s_0^2-1)}{2}\right] + \frac{\frac{(s_0^{2}-1)}{2}}{1 + \frac{(s_0^2-1)}{2} }.
\end{align*}
Taking the derivative of the second and third term with respect to $s_0$ gives:
\begin{align}
    -\frac{s_0(s_0^2-1)}{2\left(1 + \frac{s_0^2-1}{2}\right)^2} \label{eq:thm4-deriv2}
\end{align}
Suppose that there is some $\overline{\nu} > 1$ so that $s_0^2 > \overline{\nu}^2$ for all $T$, then (\ref{eq:thm4-deriv2}) is negative, implying:
\begin{align*}
    \frac{C_2(s_0^{2}-1)^2}{\sqrt{C_{1,a}}}- \log\left[1 + \frac{(s_0^2-1)}{2}\right] + \frac{\frac{(s_0^{2}-1)}{2}}{1 + \frac{(s_0^2-1)}{2}} &\leq \frac{C_2(\overline{s}^{2}-1)^2}{\sqrt{C_{1,a}}} - \log\left[1 + \frac{(\overline{\nu}^2-1)}{2}\right] + \frac{\frac{(\overline{\nu}^{2}-1)}{2}}{1 + \frac{(\overline{\nu}^2-1)}{2}}\\
    &:= \frac{C_2(\overline{s}^{2}-1)^2}{\sqrt{C_{1,a}}} -\nu^*.
\end{align*}
The inequality $\log(1+x)> x/(x+1)$ for $x > -1$ and the fact that $\overline{\nu}^2 > 1$ imply that $\nu^* > 0$, so if we pick $C_{1,a} > (C_1(\overline{s}^2-1)^2/\nu^*)^2$, then we will have $\partial h_T/\partial t < 0$. On the other hand, if there is some $\underline{\nu} < 1$ so that $s_0^2 < \overline{\nu}^2$ for all $T$, then (\ref{eq:thm4-deriv2}) is positive, implying:
\begin{align*}
    \frac{C_2(s_0^{2}-1)^2}{\sqrt{C_{1,a}}}- \log\left[1 + \frac{(s_0^2-1)}{2}\right] + \frac{\frac{(s_0^{2}-1)}{2}}{1 + \frac{(s_0^2-1)}{2}} &\leq \frac{C_2(\overline{s}^{2}-1)^2}{\sqrt{C_{1,a}}} - \log\left[1 + \frac{(\underline{\nu}^2-1)}{2}\right] + \frac{\frac{(\underline{\nu}^{2}-1)}{2}}{1 + \frac{(\underline{\nu}^2-1)}{2}} 
\end{align*}
and we can use the same argument to pick $C_{1,a}$ large enough so that $\partial h_T/\partial t < 0$. On the other hand, if $\lim_{T\to\infty} s_0^2 = 1$, we can again use the inequality $\log(1+x) > x/(1+x)$ for $x > -1$ to write:
\begin{align*}
    - \log\left[1 + \frac{(s_0^2-1)}{2}\right] + \frac{\frac{(s_0^{2}-1)}{2}}{1 + \frac{(s_0^2-1)}{2} } \leq 0
\end{align*}
and at the same time, a first order Taylor approximation of $\log(1+x)$ gives:
\begin{align*}
    - \log\left[1 + \frac{(s_0^2-1)}{2}\right] + \frac{\frac{(s_0^{2}-1)}{2}}{1 + \frac{(s_0^2-1)}{2} } &= \frac{(s_0^2-1)}{2} + \mathcal{O}\left((s_0^2-1)^2\right) - \frac{\frac{(s_0^2-1)}{2}}{1 + \frac{(s_0^2-1)}{2}} \\
    &= \mathcal{O}\left((s_0^2-1)^2\right)  + \frac{\frac{(s_0^2-1)^2}{4}}{1 + \frac{(s_0^2-1)}{2}} \\
    &= \mathcal{O}\left((s_0^2-1)^2\right)
\end{align*}
so this term is negative and converges to zero at rate $\mathcal{O}((s_0^2-1)^2)$. Therefore, there is some $B_1>0$ such that:
\begin{align*}
    \lim_{T\to\infty} \frac{\frac{C_2(s_0^{2}-1)^2}{\sqrt{C_{1,a}}}}{\log\left[1 + \frac{(s_0^2-1)}{2}\right] - \frac{\frac{(s_0^{2}-1)}{2}}{1 + \frac{(s_0^2-1)}{2} }} = \frac{B_1}{\sqrt{C_{1,a}}}.
\end{align*}
So if we pick $C_{1,a} > B_1^2$, then for large $T$ we will also have $\partial h_T/\partial t < 0$ in this case, confirming that $h_T(t,t_0)$ is decreasing in $t$, and thus:
\begin{align*}
    \log \frac{\alpha_{t_0}}{\alpha_t} &> h_T(t,t_0)\big|_{t = \sup\;\{t\in\mathcal{T}_{a_T}:t\leq t_0\}} \\
    &\geq h_T\left(t_0-\frac{C_{1,a} a_T\log^2T}{(s_0^2-1)^2},t_0\right) \\
    &= -C_2\sqrt{C_{1,a}}a_T\log^2T + \frac{C_{1,a}\log s_0^2}{2(s_0^2-1)^2}a_T\log^2T \\
    &\quad\:+ \left(\frac{C_{1,a}a_T\log^2 T}{2(s_0^2-1)^2}\right)\left(1 +\frac{(T-t_0+1)(s_0^2-1)^2}{C_{1,a}a_T\log^2 T} \right)\log\left[1 + \frac{(s_0^{-2}-1)}{1 +\frac{(T-t_0+1)(s_0^2-1)^2}{C_{1,a}a_T\log^2 T}}\right] \\
    &\quad\:+ \mathcal{O}(a_T\log^2T)
\end{align*}
By assumption:
\begin{align*}
    \lim_{T\to\infty}\frac{C_{1,a}a_T\log^2 T}{(T-t_0+1)(s_0^2-1)^2} = 0 \implies \lim_{T\to\infty} \frac{(s_0^{-2}-1)}{1 +\frac{(T-t_0+1)(s_0^2-1)^2}{C_{1,a}a_T\log^2 T}} =0.
\end{align*}
Therefore, another first order Taylor approximation gives:
\small
\begin{align*}
    &\left(\frac{C_{1,a}a_T\log^2 T}{2(s_0^2-1)^2}\right)\left(1 +\frac{(T-t_0+1)(s_0^2-1)^2}{C_{1,a}a_T\log^2 T} \right)\log\left[1 + \frac{(s_0^{-2}-1)}{1 +\frac{(T-t_0+1)(s_0^2-1)^2}{C_{1,a}a_T\log^2 T}}\right]\\
    &= \frac{C_{1,a}(s_0^{-2} -1)}{2(s_0^2-1)^2}a_T\log^2 T + \mathcal{O}\left(\frac{C_{1,a}^2a_T\log^4 T}{(T-t_0+1)(s_0^2-1) + C_{1,a}a_T\log^2 T}\right) \\
    &= \frac{C_{1,a}(s_0^{-2} -1)}{2(s_0^2-1)^2}a_T\log^2 T + \mathcal{O}\left(a_T\log^2 T\right) \tag{$(T-t_0+1)(s_0^2-1)^2 \gg a_T\log^2T$}
\end{align*}
\normalsize
and thus for some $B_1 > 0$:
\begin{align*}
    \log \frac{\alpha_{t_0}}{\alpha_t} &> \left(\frac{C_{1,a}f(s_0^{-2})}{2(s_0^2-1)^2}-C_2\sqrt{C_{1,a}} - B_1\right)a_T\log^2T 
\end{align*}
Once again, if $s_0^2 \geq \overline{\nu}^2 > 1$ for all $T$, then $f(s_0^{-2})/(s_0^2-1)^2 \geq f(\overline{\nu}^{-2})/(\overline{s}^2-1)^2 > 0$, while if $s_0^2 \leq \underline{\nu}^2 < 1$, then $f(s_0^{-2})/(s_0^2-1)^2 \geq f(\underline{\nu}^{-2}) > 0$, and since:
\begin{align*}
    \lim_{s_0^2\to 1} \frac{f(s_0^{-2})}{(s_0^2-1)^2} = \frac{1}{2}
\end{align*}
then if $\lim_{T\to\infty} s_0^2 = 1$, then for large enough $T$, $f(s_0^{-2})/(s_0^2-1)^2 > 1/4$, so again we can pick $C_{1,a}$ large enough that:
\begin{align*}
    \frac{C_{1,a}f(s_0^{-2})}{2(s_0^2-1)^2}-C_2\sqrt{C_{1,a}} - B_1 
    > C_\pi
\end{align*}
proving (\ref{eq:thm4-result}) holds for this case.

\subsubsection*{Case 2.b: $t_0 - t \leq T-t_0+1$.}

For some $C_1\geq1$, on $\mathcal{E}$ we have: 
\begin{align*}
    \max\left\{\left|\sum_{t'=t}^{t_0-1} (z_{t'}^2 -1)\right|, \left|\sum_{t'=t_0}^{T} (z_{t'}^2 -1)\right|, \left|\sum_{t'=t}^{t_0-1} z_{t'}\right|, \left|\sum_{t'=t_0}^{T} z_{t'} \right|\right\} &\leq C_1\sqrt{a_T(T-t+1)}\log T \tag{on $\cap_{i=1}^4\mathcal{E}_i\subset\mathcal{E}$}
\end{align*}
and thus:
\begin{align*}
    \left|\xi(T,t_0,t,s_0)\right| &\leq (1+s_0^{2} + 2s_0)\left[\frac{C_1\sqrt{a_T}\log T}{\sqrt{T-t+1}} + \frac{C_1^2a_T\log^2T}{T-t+1}\right] \\
    &\leq\frac{8\overline{s}^2C_1^2}{C_{2,a}}. \tag{$s_0^2 < \overline{s}^2$ and $T-t+1\geq C_{2,a}a_T\log^2T$}
\end{align*}
So for $C_{2,a} > 8\overline{s}^2C_1^2$, we have $|\xi(T,t_0,t,s_0)| > -1$ and we can use the inequality $\log(1+x) > x/(x+1)$ for $x > -1$ to write:
\small
\begin{align*}
    \left(\frac{T - t +1}{2}\right)\log\left[1 + \xi(T,t_0,t,s_0)\right] &\geq \left(\frac{T - t +1}{2}\right) \left[\xi(T,t_0,t,s_0) - \frac{\xi(T,t_0,t,s_0)^2}{1+\xi(T,t_0,t,s_0)}\right] \\
    &= \frac{1}{2}\sum_{t'=t}^{t_0-1}(z_{t'}^2 - 1) - \frac{\left(\sum_{t'=t}^{t_0-1}z_{t'}\right)^2 }{2(T-t+1)} + \frac{s_0^2}{2}\left[\sum_{t'=t_0}^T (z_{t'}^2 - 1) -  \frac{\left(\sum_{t'=t_0}^{T}z_{t'}\right)^2 }{(T-t+1)}\right]  \\
    &\quad\; - \frac{s_0\sum_{t'=t_0}^{T}z_{t'} \sum_{t'=t}^{t_0-1}z_{t'}}{(T-t+1)} - \left(\frac{T - t +1}{2}\right)\frac{\xi(T,t_0,t,s_0)^2}{1+\xi(T,t_0,t,s_0)}
\end{align*}
\normalsize
We also have:
\begin{align*}
    \frac{\left(\sum_{t'=t}^{t_0-1}z_{t'}\right)^2 }{2(T-t+1)} &\leq \frac{C_1^2}{2}a_T\log^2T \\
    \frac{s_0\sum_{t'=t_0}^{T}z_{t'} \sum_{t'=t}^{t_0-1}z_{t'}}{(T-t+1)} &\leq \overline{s}C_1^2 a_T\log^2T \\
    \frac{s_0^2\left(\sum_{t'=t_0}^{T}z_{t'}\right)^2 }{2(T-t+1)} &\leq \frac{\overline{s}^2C_1^2}{2} a_T\log^2 T.
\end{align*}
For $C_{2,a} > 16\overline{s}^2C_1^2$:
\begin{align*}
    \left(\frac{T - t +1}{2}\right)\frac{\xi(T,t_0,t,s_0)^2}{1+\xi(T,t_0,t,s_0)} &\geq (T - t +1)(1+s_0^{2} + 2s_0)^2\left[\frac{C_1\sqrt{a_T}\log T}{\sqrt{T-t+1}} + \frac{C_1^2a_T\log^2T}{T-t+1}\right]^2 \\
    &\leq  \frac{16 \overline{s}^2C_1^4a_T^2\log^4T}{T-t+1} \tag{$s_0^2 \leq \overline{s}^2$} \\
    &\leq C_1^2a_T\log^2T. \tag{$T-t+1\geq C_{2,a}a_T\log^2T > 16\overline{s}^2C_1^2a_T\log^2T$} 
\end{align*}
Returning to bound (\ref{eq:thm4-cs2-bd3}), we now have: 
\begin{align}
    \log \frac{\alpha_{t_0}}{\alpha_t} &>  \frac{(s_0^2-1)}{2}\sum_{t'=t_0}^T (z_{t'}^2-1)- \frac{(T-t_0+1)}{2}\log s_0^2  + \left(\frac{T - t +1}{2}\right)\log\left[1 + \frac{\frac{(T-t_0+1)(s_0^2-1)}{T-t+1}}{1 + \xi(T,t_0,t,s_0)}\right] \notag  \\
    &\quad\: + \left(\frac{T - t +1}{2}\right)\log\left[1 + \frac{\frac{(T-t_0+1)(t_0-t)b_0^2 }{2(T-t+1)^2} - \frac{2(T-t_0+1)|b_0|C_1\sqrt{a_T(t_0-t)}\log T }{(T-t+1)^2}}{1 + \frac{(T-t_0+1)(s_0^2-1)}{T-t+1} + \xi(T,t_0,t,s_0)}\right] \notag \\
    &\quad\: + \mathcal{O}(a_T\log^2 T). \label{eq:thm4-cs2-b-bd1}
\end{align}
Note that if $s_0^2 = 1$, then we can completely ignore the first line of (\ref{eq:thm4-cs2-b-bd1}), on the other hand, if $s_0^2\neq 1$, then by assumption $(s_0^2 -1)^2(T-t_0+1)\gg a_T\log^2T$. However, $0<\underline{s}^2\leq s_0^2\leq\overline{s}^2<\infty$ implies $(s_0^2 -1)^2 = \mathcal{O}(1)$, so we must have $T-t_0+1\gg a_T\log^2T$. We showed above that on $\mathcal{E}$:
\begin{align*}
    \left|\xi(T,t_0,t,s_0)\right| &\leq (1+s_0^{2} + 2s_0)\left[\frac{C_1\sqrt{a_T}\log T}{\sqrt{T-t+1}} + \frac{C_1^2a_T\log^2T}{T-t+1}\right] \\
    &\leq 4\overline{s}^2\left[\frac{C_1\sqrt{a_T}\log T}{\sqrt{T-t_0+1}} + \frac{C_1^2a_T\log^2T}{T-t_0+1}\right] \tag{$T-t+1\geq T-t_0+1$}
\end{align*}
i.e., $\left|\xi(T,t_0,t,s_0)\right| \to 0$ on $\mathcal{E}$ when $s_0^2\neq 1$. Next, we can write:
\footnotesize
\begin{align*}
    \left(\frac{T - t +1}{2}\right)\log\left[1 + \frac{\frac{(T-t_0+1)(s_0^2-1)}{T-t+1}}{1 + \xi(T,t_0,t,s_0)}\right] &= \left(\frac{T - t +1}{2}\right)\log\left[1 + \frac{(T-t_0+1)(s_0^2-1)}{T-t+1}\right] \\
    &\quad\:+ \left(\frac{T - t +1}{2}\right)\log\left[1 - \left(\frac{ \xi(T,t_0,t,s_0)}{1 + \xi(T,t_0,t,s_0)}\right)\left(\frac{\frac{(T-t_0+1)(s_0^2-1)}{T-t+1}}{1+\frac{(T-t_0+1)(s_0^2-1)}{T-t+1}}\right)\right].
\end{align*}
\normalsize
We have just shown that:
\begin{align*}
    \frac{ \xi(T,t_0,t,s_0)}{1 + \xi(T,t_0,t,s_0)} &= \mathcal{O}\left(\frac{\sqrt{a_T}\log T}{\sqrt{T-t+1}}\right) =  \mathcal{O}\left(\frac{\sqrt{a_T}\log T}{\sqrt{T-t_0+1}}\right) \\
    \frac{\frac{(T-t_0+1)(s_0^2-1)}{T-t+1}}{1+\frac{(T-t_0+1)(s_0^2-1)}{T-t+1}} &= \mathcal{O}\left(\frac{(s_0^2-1)(T-t_0+1)}{T-t+1}\right)= \mathcal{O}(1)
\end{align*}
so another first order Taylor approximation of $\log(1+x)$ around zero gives:
\footnotesize
\begin{align*}
     \left(\frac{T - t +1}{2}\right)\log\left[1 - \left(\frac{ \xi(T,t_0,t,s_0)}{1 + \xi(T,t_0,t,s_0)}\right)\left(\frac{\frac{(T-t_0+1)(s_0^2-1)}{T-t+1}}{1+\frac{(T-t_0+1)(s_0^2-1)}{T-t+1}}\right)\right] &= (s_0^2-1)\mathcal{O}\left(\sqrt{(T-t_0+1)a_T}\log T\right) + \mathcal{O}(a_T\log^2 T).
\end{align*}
\normalsize
Similarly, for some $C_1>0$ on $\mathcal{E}$ we have:
\begin{align*}
    \left|\sum_{t' = t_0}^T (z_{t'}^2-1)\right|  \leq C_1\sqrt{(T-t_0+1)a_T}\log T \tag{on $\mathcal{E}_4\subset\mathcal{E}$}
\end{align*}
So for some $C_2 \geq  C_1/2$, we can now rewrite (\ref{eq:thm4-cs2-b-bd1}) on $\mathcal{E}$ as:
\small
\begin{align}
    \log \frac{\alpha_{t_0}}{\alpha_t} &> -C_2|s_0^2-1|\sqrt{(T-t_0+1)a_T}\log T- \frac{(T-t_0+1)}{2}\log s_0^2  + \left(\frac{T - t +1}{2}\right)\log\left[1 + \frac{(T-t_0+1)(s_0^2-1)}{T-t+1}\right] \notag  \\
    &\quad\: + \left(\frac{T - t +1}{2}\right)\log\left[1 + \frac{\frac{(T-t_0+1)(t_0-t)b_0^2 }{2(T-t+1)^2} - \frac{2(T-t_0+1)|b_0|C_1\sqrt{a_T(t_0-t)}\log T }{(T-t+1)^2}}{1 + \frac{(T-t_0+1)(s_0^2-1)}{T-t+1} + \xi(T,t_0,t,s_0)}\right] \notag \\
    &\quad\: + \mathcal{O}(a_T\log^2 T). \label{eq:thm4-cs2-b-bd2}
\end{align}
\normalsize

\subsubsection*{Case 2.b.i: $\min\left\{b_0^2,b_0^2/s_0^2\right\} \geq (s_0^2-1)^2$}

Using the same argument from Case 2.a.i, it is straight forward to show that the term:
\begin{align*}
    - \frac{(T-t_0+1)}{2}\log s_0^2  + \left(\frac{T - t +1}{2}\right)\log\left[1 + \frac{(T-t_0+1)(s_0^2-1)}{T-t+1}\right]
\end{align*}
is decreasing in $t$ and equal to zero when $t = t_0$, and that on $\mathcal{E}$:
\begin{align*}
     \log\left[1 + \frac{\frac{(T-t_0+1)(t_0-t)b_0^2 }{2(T-t+1)^2} - \frac{2(T-t_0+1)|b_0|C_1\sqrt{a_T(t_0-t)}\log T }{(T-t+1)^2}}{1 + \frac{(T-t_0+1)(s_0^2-1)}{T-t+1} + \xi(T,t_0,t,s_0)}\right] \geq \log\left[1 + \frac{(T-t_0+1)(t_0-t)b_0^2 }{16\overline{s}^2(T-t+1)^2}\right].
\end{align*}
So in this case we can write:
\small
\begin{align*}
    \log \frac{\alpha_{t_0}}{\alpha_t} &> -C_2|s_0^2-1|\sqrt{(T-t_0+1)a_T}\log T + \left(\frac{T - t +1}{2}\right)\log\left[1 + \frac{(T-t_0+1)(t_0-t)b_0^2 }{16\overline{s}^2(T-t+1)^2}\right] + \mathcal{O}(a_T\log^2 T). 
\end{align*}
\normalsize
The second term is decreasing in $t$ and $t_0 -t \geq T-t_0+1$ implies $t \leq 2t_0-T-1$ in this case so for some $B_1 > 0$:
\small
\begin{align*}
    \log \frac{\alpha_{t_0}}{\alpha_t} &> -C_2|s_0^2-1|\sqrt{(T-t_0+1)a_T}\log T + (T-t_0+1)\log\left[1 + \frac{b_0^2 }{16^2\overline{s}^2}\right] - B_1 a_T\log^2 T. 
\end{align*}
\normalsize
Next we have:
\footnotesize
\begin{align*}
    -C_2|s_0^2-1|\sqrt{(T-t_0+1)a_T}\log T + (T-t_0+1)\log\left[1 + \frac{b_0^2 }{16^2\overline{s}^2}\right] &\geq (T-t_0+1)\left(\log\left[1 + \frac{b_0^2 }{16^2\overline{s}^2}\right]-C_2(\overline{s}^2+1) \sqrt{\frac{a_T\log^2T}{T-t_0+1}}\right).
\end{align*}
\normalsize
Suppose that $b_0^2\to\infty$ as $T\to\infty$, then since $T-t_0+1\gtrsim a_T\log^2T$, there is some $B_2 > 0$ so that:
\begin{align*}
    \frac{a_T\log^2T}{B_2}\left(\log\left[1 + \frac{b_0^2 }{16\overline{s}^2}\right]-\frac{C_2(\overline{s}^2+1)}{\sqrt{B_2}}\right)
\end{align*}
So (\ref{eq:thm4-result}) holds for $T$ large enough that:
\begin{align*}
    \log\left[1 + \frac{b_0^2 }{16^2\overline{s}^2}\right] > \frac{C_2(\overline{s}^2+1)}{\sqrt{B_2}} + B_2(B_1 + C_\pi).
\end{align*}
On the other hand, if $b_0^2 = \mathcal{O}(1)$, then it must be the case that $T-t_0+1\gg a_T\log^2 T$. If $b_0^2 > \underline{b}^2 > 0$ for all $T$, then for large $T$ we have:
\begin{align*}
    \log\left[1 + \frac{b_0^2 }{16^2\overline{s}^2}\right]-C_2(\overline{s}^2+1) \sqrt{\frac{a_T\log^2T}{T-t_0+1}} >\frac{1}{2}\log\left[1 + \frac{\underline{b}^2 }{16^2\overline{s}^2}\right] > 0
\end{align*}
and:
\begin{align*}
    \frac{(T-t_0+1)}{2}\log\left[1 + \frac{\underline{b}^2 }{16^2\overline{s}^2}\right] - B_2a_T\log^2T \geq C_\pi a_T\log^2T.
\end{align*}
Finally, if $\lim_{T\to\infty}b_0^2=0$, then a first order Taylor expansion of $\log(1+x)$ gives: 
\begin{align*}
    \log\left[1 + \frac{b_0^2 }{16^2\overline{s}^2}\right] = \frac{b_0^2 }{16^2\overline{s}^2} + \mathcal{O}(b_0^4) \geq \frac{b_0^2 }{512\overline{s}^2} 
\end{align*}
where the last inequality holds for large $T$, and thus:
\begin{align*}
    \log \frac{\alpha_{t_0}}{\alpha_t} &> -C_2(\overline{s}^2+1)\sqrt{(T-t_0+1)a_T}\log T + \frac{(T-t_0+1)b_0^2 }{512\overline{s}^2}  - B_1 a_T\log^2 T. 
\end{align*}
But $(T-t_0+1)b_0^2 \gg a_T\log^2T$ by assumption, so the second term dominates for large $T$ we again get that (\ref{eq:thm4-result}) holds for this case.

\subsubsection*{Case 2.b.ii: $\min\left\{b_0^2,b_0^2/s_0^2\right\} < (s_0^2-1)^2$}

Returning to bound (\ref{eq:thm4-cs2-b-bd2}), we can use the same argument from Case 2.a.ii to get:
\small
\begin{align*}
    \log \frac{\alpha_{t_0}}{\alpha_t} &> -C_2|s_0^2-1|\sqrt{(T-t_0+1)a_T}\log T- \frac{(T-t_0+1)}{2}\log s_0^2  + \left(\frac{T - t +1}{2}\right)\log\left[1 + \frac{(T-t_0+1)(s_0^2-1)}{T-t+1}\right] \notag  \\
    &\quad\: + \left(\frac{T - t +1}{2}\right)\log\left[1 - \frac{4(T-t_0+1)|b_0|C_1\sqrt{a_T(t_0-t)}\log T }{\underline{s}^2(T-t+1)^2}\right]  + \mathcal{O}(a_T\log^2 T).
\end{align*}
\normalsize
and:
\footnotesize
\begin{align*}
    \left(\frac{T - t +1}{2}\right)\log\left[1 - \frac{4(T-t_0+1)|b_0|C_1\sqrt{a_T(t_0-t)}\log T }{\underline{s}^2(T-t+1)^2}\right] &= -\frac{4(T-t_0+1)|b_0|C_1\sqrt{a_T(t_0-t)}\log T }{\underline{s}^2(T-t+1)} + \mathcal{O}\left(a_T\log^2T\right) \\
    &\geq -\frac{4\overline{s}C_1|s_0^2-1|}{\underline{s}^2} \sqrt{a_T(T-t_0+1)}\log T\tag{$T-t+1\geq t_0-t \geq T-t_0+1$ and $\min\{b_0^2,b_0^2/s_0^2\}\leq (s_0^2-1)^2$} + \mathcal{O}\left(a_T\log^2T\right).
\end{align*}
\normalsize
So for $C_3 > C_2 + 4\overline{s}C_1/\underline{s}^2$ we have:
\small
\begin{align}
    \log \frac{\alpha_{t_0}}{\alpha_t} &> -C_3|s_0^2-1|\sqrt{(T-t_0+1)a_T}\log T- \frac{(T-t_0+1)}{2}\log s_0^2  + \left(\frac{T - t +1}{2}\right)\log\left[1 + \frac{(T-t_0+1)(s_0^2-1)}{T-t+1}\right] \notag \\
    &\quad\:+ \mathcal{O}(a_T\log^2 T). \label{eq:thm4-cs2-b-ii-bd1}
\end{align}
\normalsize
The bound (\ref{eq:thm4-cs2-b-ii-bd1}) is decreasing as a function of $t$ and in this case, since $t_0-t>T-t_0+1$, then $t < 2t_0-T + 1$. Therefore, we have:
\small
\begin{align*}
    \log \frac{\alpha_{t_0}}{\alpha_t} &> -C_3|s_0^2-1|\sqrt{(T-t_0+1)a_T}\log T-\frac{(T-t_0+1)}{2}\log s_0^2 + (T-t_0+1)\log\left[1 + \frac{(s_0^2-1)}{2}\right]  + \mathcal{O}(a_T\log^2 T) \\
    &= -C_3|s_0^2-1|\sqrt{(T-t_0+1)a_T}\log T + (T-t_0+1)\log\left[1 + \frac{s_0+\frac{1}{s_0}-2}{2}\right] + \mathcal{O}(a_T\log^2 T)
\end{align*}
\normalsize
If there is some $\underline{\nu} <1$ so that $s_0 \leq \underline{\nu}$ for all $T$, then: 
\small
\begin{align*}
    \log \frac{\alpha_{t_0}}{\alpha_t} &> -C_3\sqrt{(T-t_0+1)a_T}\log T  + (T-t_0+1)\log\left[1 + \frac{\underline{\nu} +\underline{\nu}^{-1}-2}{2}\right] + \mathcal{O}(a_T\log^2 T)
\end{align*}
\normalsize
and if there is some $\overline{\nu} >1$ so that $s_0 \geq \underline{\nu}$ for all $T$, then: 
\small
\begin{align*}
    \log \frac{\alpha_{t_0}}{\alpha_t} &> -C_3|\overline{s}^2-1|\sqrt{(T-t_0+1)a_T}\log T  + (T-t_0+1)\log\left[1 + \frac{\overline{\nu} +\overline{\nu}^{-1}-2}{2}\right] + \mathcal{O}(a_T\log^2 T)
\end{align*}
\normalsize
and since $ \min\{\underline{\nu} +\underline{\nu}^{-1},  \overline{\nu} +\overline{\nu}^{-1}\}-2 > 0$ and $T-t_0+1 \gg a_T\log^2 T$, the second term dominates in both the bounds above and for large enough $T$ we will have $\log \frac{\alpha_{t_0}}{\alpha_t} \geq C_\pi \log^2 T$. On the other hand, if $\lim_{T\to\infty} s_0^2 = 1$, then we can use a first order Taylor approximation of $\log(1+x)$ around zero to write: 
\begin{align*}
    \log\left[1 + \frac{s_0+\frac{1}{s_0}-2}{2}\right] = \frac{s_0+\frac{1}{s_0}-2}{2} + \mathcal{O}\left(\left(s_0+\frac{1}{s_0}-2\right)^2\right).
\end{align*}
Since $\lim_{s_0\to 1} s_0+\frac{1}{s_0}-2 = 0$, then for large enough $T$, we have:
\begin{align*}
    \log\left[1 + \frac{s_0+\frac{1}{s_0}-2}{2}\right] = \frac{s_0+\frac{1}{s_0}-2}{4}
\end{align*}
and thus:
\small
\begin{align*}
    \log \frac{\alpha_{t_0}}{\alpha_t} &> -C_3|s_0^2-1|\sqrt{(T-t_0+1)a_T}\log T + (T-t_0+1)\left( \frac{s_0+\frac{1}{s_0}-2}{4}\right) + \mathcal{O}(a_T\log^2 T)
\end{align*}
\normalsize
Next, since:
\begin{align*}
    \lim_{s_0\to1} \frac{s_0+\frac{1}{s_0}-2}{(s_0^{2} - 1)^2} = \frac{1}{4}
\end{align*}
then for large enough $T$ we have:
\begin{align*}
    s_0+\frac{1}{s_0}-2 > \frac{(s_0^{2} - 1)^2}{8}
\end{align*}
and thus: 
\begin{align*}
    \log \frac{\alpha_{t_0}}{\alpha_t} &> -C_3\sqrt{(s_0^2-1)^2(T-t_0+1)a_T}\log T + \frac{(T-t_0+1)(s_0^2-1)^2}{8}+ \mathcal{O}(a_T\log^2 T)
\end{align*}
By assumption, $(T-t_0+1)(s_0^2-1)^2 \gg a_T\log^2 T$, so again the second term dominates the others and (\ref{eq:thm4-result}) holds on $\mathcal{E}$ for large $T$ in this case.

\end{proof}

\subsection{Proof of Corollary \ref{cor:cred-sets}}
\label{app:cor-cred-sets}

\begin{proof}
In the proofs of Theorems \ref{theorem:smcp}, \ref{theorem:sscp}, and \ref{theorem:smscp}, we show that given the true change $t_0\in[T]$ and the posterior change-point probabilities $\overline{\boldsymbol{\pi}}_{1:T}$, we can pick $C_\beta > 0$ large enough when defining the localization rate $\epsilon_T$ so that:
\begin{align*}
    \lim_{T\to\infty} \Pr\left(\bigcap_{t \in [T] \;:\; |t_0-t| > \epsilon_T}\log\frac{\overline{\pi}_{t_0}}{\overline{\pi}_t} > 2 \log T\right) = 1.
\end{align*}
Thus, with high probability we will have:
\begin{align*}
    \sum_{t \in [T] \;:\; |t_0-t| > \epsilon_T} \overline{\pi}_t < \sum_{t \in [T] \;:\; |t_0-t| > \epsilon_T} \frac{\overline{\pi}_{t_0}}{T^2} \leq \frac{1}{T}.
\end{align*}
So all the mass of $\overline{\boldsymbol{\pi}}_{1:T}$ is concentrating in $\{t \in [T] : |t_0-t| \leq \epsilon_T\}$, meaning that for any $\alpha \in (0,1)$, (\ref{eq:cs}) will only pick coordinates within $\epsilon_T$ of $t_0$ with probability converging to one as as $T\to\infty$. Since $|\{t \in [T] : |t_0-t| \leq \epsilon_T\}| \leq 2\epsilon_T$ this proves the corollary. 

Note that in the argument above we ignored the end-point restrictions we built into the set $\mathcal{T}_\beta$ in the proofs of Theorems \ref{theorem:smcp}, \ref{theorem:sscp}, and \ref{theorem:smscp}, but if we change the $2 \epsilon_t$ bound to $3 \epsilon_T$, then result still holds since there are at most $\mathcal{O}(\log(\log T))$ end points of $[T]$ that could be added to $\mathcal{CS}(\alpha,\overline{\boldsymbol{\pi}}_{1:T})$. For the detection rule, this omission has no effect since $\log^{2+\delta} T$ still dominates $3\epsilon_T$. Proving that the result holds for Theorem \ref{theorem:alpha-mixing} follows an identical argument with the $\log T$ term above replaced by $\log^2 T$.
\end{proof}

\newpage
\section{Details of Variational Algorithms}
\subsection{Proof of Proposition \ref{prop:coord-ascent}}
\label{app:prop1-proof}

Define the $N$ variable and $N$ parameters blocks:
\begin{align*}
    \boldsymbol{\theta}_j &:=\{b_j,s_j, \tau_j\},&\overline{\boldsymbol{\theta}}_j &:= \{\overline{b}_{jt}, \overline{\omega}_{jt}, \overline{u}_{jt}, \overline{v}_{jt}, \overline{\pi}_{jt}\}_{t=1}^T, &  1\leq j\leq J,\\
    \boldsymbol{\theta}_\ell &:=\{b_\ell, \tau_\ell\}, & \overline{\boldsymbol{\theta}}_\ell &:= \{\overline{b}_{\ell t}, \overline{\omega}_{\ell t}, \overline{\pi}_{\ell t}\}_{t=1}^T, & J<\ell \leq J+L, \\
    \boldsymbol{\theta}_k &:=\{s_k, \tau_k\}, &\overline{\boldsymbol{\theta}}_k &:= \{ \overline{u}_{kt}, \overline{v}_{kt}, \overline{\pi}_{kt}\}_{t=1}^T, & J+L <k \leq N, 
\end{align*}
and let $\boldsymbol{\Theta} := \{{\boldsymbol{\theta}}_{1:L}, {\boldsymbol{\theta}}_{1:K},{\boldsymbol{\theta}}_{1:J}\}$. We see that Algorithm \ref{alg:mich} begins by initializing some $\overline{\boldsymbol{\theta}}_{1:N}$ that parametrizes a $q\in\mathcal{Q}_{\text{MF}}$, then Proposition \ref{prop:vb} below shows that the update step for $\overline{\boldsymbol{\theta}}_i$ in Algorithm \ref{alg:mich} is equivalent to solving $\max_{q_i} \text{\normalfont ELBO}(q)$ while holding all the other component distributions of $q$ fixed. Therefore, Algorithm \ref{alg:mich} is a coordinate ascent procedure for maximizing $\text{\normalfont ELBO}(q)$. Proposition \ref{prop:vb} also shows that the $\overline{\boldsymbol{\theta}}_i$ that parameterizes the solution to $\text{arg}\max_{q_i} \text{\normalfont ELBO}\left(q\right)$ is the unique maximizer of some continuously differentiable objective function $F$ when $\overline{\boldsymbol{\theta}}_{1:N}\setminus\overline{\boldsymbol{\theta}}_i$ is held fixed, so by Proposition 2.7.1 of \cite{Bertsekas97}, the sequence of parameters $\{\overline{\boldsymbol{\theta}}_{1:N}^{(n)}\}_{n\geq 1}$ returned by Algorithm \ref{alg:mich} is converging to some limit point. Because the set of distributions $\{q^{(n)}_i\}_{i=1}^N$ in (\ref{eq:mich-bs-post})-(\ref{eq:mich-s-post}) returned by the $n$\textsuperscript{th} iteration of Algorithm \ref{alg:mich} is fully characterized by $\overline{\boldsymbol{\theta}}^{(n)}_{1:N}$, the sequence of distributions $\{q^{(n)}\}_{n \geq 1}$, where $q^{(n)} := \prod_{i=1}^N q_i^{(n)}$, is converging to a stationary point of $\text{\normalfont ELBO}\left(q\right)$. Therefore, Proposition \ref{prop:coord-ascent} directly follows from Proposition \ref{prop:vb}.

\begin{proposition} 
\label{prop:vb}
Given the MICH model (\ref{eq:y-mich})-(\ref{eq:s-mich}), assume that $\omega_0,u_0,v_0,\pi_t >0$. Let $\{q_i\}_{i=1}^N$ be an arbitrary collection of distributions such that $q := \prod_{i=1}^N q_i$ and $\E_{q_j}[\lambda_{jt}], \E_{q_j}[\lambda_{jt}\mu_{j t}], \E_{q_j}[\lambda_{jt}\mu^2_{j t}]<\infty$ for any $1\leq j \leq J$, $\E_{q_\ell}[\mu_{\ell t}], \E_{q_\ell}[\mu^2_{\ell t}]<\infty$ for any $J<\ell \leq J+L$, and $\E_{q_k}[\lambda_{k t}]<\infty$ for any $J+L<k \leq N$. Define $\text{\normalfont ELBO}\left(q\right)$ as in (\ref{eq:elbo}), and the mean residual $\tilde{r}_t$, expected precision $\overline{\lambda}_t$ and variance correction term $\delta_t$ as in (\ref{eq:mod-resid})-(\ref{eq:delta}). Then: 
\vspace{-10pt}

\begin{enumerate}[label=\normalfont(\roman*)]
    \itemsep0em 

    \item For $J<\ell \leq J+L$, define the partial mean residual $\tilde{r}_{-\ell t} :=  \tilde{r}_{t} + \E_{q_{\ell}}[\mu_{\ell t}]$. Then $\text{\normalfont{arg}}\max_{q_\ell} \text{\normalfont ELBO}\left(q\right)$ is equivalent to the mean-scp posterior in (\ref{eq:gamma-post-cat1}) and (\ref{eq:b-smcp}) with parameters:
    \begin{align*}
        \overline{\boldsymbol{\theta}}_\ell = \normalfont{\texttt{mean-scp}}(\tilde{\mathbf{r}}_{-\ell,1:T} \:;\: \overline{\boldsymbol{\lambda}}_{1:T}, \omega_0, \boldsymbol{\pi}_{1:T})
    \end{align*}
    
    \item For $J+L <k \leq N$, define the partial scale residual $\overline{\lambda}_{-kt} := \E_{q_k}[\lambda_{k t}]^{-1}\overline{\lambda}_t$ and the variance corrected priors $\tilde{v}_{kt}$ and $\tilde{\pi}_{kt}$ as in (\ref{eq:v-k-corrected}) and (\ref{eq:pi-k-corrected}). Then $\text{\normalfont{arg}}\max_{q_k} \text{\normalfont ELBO}\left(q\right)$ is equivalent to the var-scp posterior in (\ref{eq:gamma-post-cat1}) and (\ref{eq:s-sscp}) with parameters:
    \begin{align}
        \overline{\boldsymbol{\theta}}_k = \normalfont{\texttt{var-scp}}(\tilde{\mathbf{r}}_{1:T} \:;\:\overline{\boldsymbol{\lambda}}_{-k,1:T}, u_0, \tilde{\mathbf{v}}_{k,1:T}, \tilde{\boldsymbol{\pi}}_{k,1:T})
    \end{align}

    \item For $1 \leq j \leq J$, define the partial mean residual $\tilde{r}_{-jt} := \tilde{r}_{t} +\E_{q_j}[\lambda_{jt}]^{-1} \E_{q_j}[\lambda_{jt} \mu_{j t}]$, the partial scale residual $\overline{\lambda}_{-jt} := \E_{q_j}[\lambda_{j t}]^{-1}\overline{\lambda}_t$, the partial correction term $\delta_{-jt}$ as per (\ref{eq:delta-j}), and the variance corrected priors $\tilde{v}_{jt}$ and $\tilde{\pi}_{jt}$ as in (\ref{eq:v-j-corrected}) and (\ref{eq:pi-j-corrected}). Then $\text{\normalfont{arg}}\max_{q_j} \text{\normalfont ELBO}\left(q\right)$ is equivalent to the meanvar-scp posterior in (\ref{eq:gamma-post-cat1}) and (\ref{eq:bs-smscp}) with parameters:
    \begin{align*}
        \overline{\boldsymbol{\theta}}_j := \normalfont{\texttt{meanvar-scp}}(\tilde{\mathbf{r}}_{-j,1:T} \:;\:\overline{\boldsymbol{\lambda}}_{-j,1:T}, \omega_0, u_0, \tilde{\mathbf{v}}_{j,1:T}, \tilde{\boldsymbol{\pi}}_{j,1:T})
    \end{align*}
\end{enumerate}
\vspace{-10pt}
Furthermore, when $q_\ell$, $q_k$, and $q_j$, have the parametric forms specified in (i)-(iii), then  $\text{\normalfont ELBO}\left(q\right)$ is equivalent to a continuously differentiable function $F(\overline{\boldsymbol{\theta}}_{1:N})$, and: i) $\overline{\boldsymbol{\theta}}_\ell$ is the unique maximizer of $F$ holding $\overline{\boldsymbol{\theta}}_{1:N}\setminus\overline{\boldsymbol{\theta}}_\ell$ fixed, ii) $\overline{\boldsymbol{\theta}}_k$ is the unique maximizer of $F$ holding $\overline{\boldsymbol{\theta}}_{1:N}\setminus\overline{\boldsymbol{\theta}}_k$ fixed, and iii) $\overline{\boldsymbol{\theta}}_j$ is the unique maximizer of $F$ holding $\overline{\boldsymbol{\theta}}_{1:N}\setminus\overline{\boldsymbol{\theta}}_j$ fixed, i.
\end{proposition}
\vspace{-10pt}
%The proof of Proposition \ref{prop:vb} is given in Appendix \ref{app:prop1-proof}. Backfitting procedures for models with additive effects typically work with the expected residual $\E_q[r_t]$ (\citealp{Breiman85,Hastie90,Friedman00,Wang20}). Algorithm \ref{alg:mich} on the other hand uses the modified residual $\tilde{r}_t$ to fit the individual SCP models, which is necessitated by the inclusion of both mean and variance changes in the MICH model. Similarly, the variance correction term $\delta_t$ plays an important role in Algorithm \ref{alg:mich} that did not previously appear in either the IBSS or ProSCALE algorithms. The variance of $\mathbf{y}_{1:T}$ is inflated by the uncertainty around each mean change, and Proposition \ref{prop:vb} shows how to use $\delta_t$ to account for the added noise from $\mu_{\ell t}$ and $\mu_{j t}$ when fitting the var-scp and meanvar-scp components.

\begin{proof}

Consider the following generic model:
\begin{align*}
    \mathbf{y}_{1:T} \:|\: \boldsymbol{\Theta} &\;\sim f(\mathbf{y}_{1:T}\:|\: \boldsymbol{\Theta};\eta), \\ 
    \boldsymbol{\Theta} &\;= \{\boldsymbol{\theta}_1, \ldots, \boldsymbol{\theta}_M\}, \\
    \boldsymbol{\theta}_i &\overset{\text{ind.}}{\sim} g_i(\boldsymbol{\theta}_i),\; \sforall i \in \{1, \ldots, M\},
\end{align*}
where $\eta$ is nuisance parameter and $g_i$ is the prior over the block $\boldsymbol{\theta}_i$. Define the mean-field family of distributions over $\boldsymbol{\Theta}$:
\begin{align*} 
    \mathcal{Q}_{\text{MF}} &:= \left\{q \::\: q(\boldsymbol{\Theta}) = \prod_{i=1}^M q_i(\boldsymbol{\theta}_i)\right\}.
\end{align*}
The task of variational Bayesian inference is to find the $q^* \in \mathcal{Q}_{\text{MF}}$ that maximizes $\text{ELBO}(q;\eta)$ as defined in (\ref{eq:elbo}). If $q \in \mathcal{Q}_{\text{MF}}$ and we define:
\begin{align*}
    q_{-i}(\boldsymbol{\theta}_{-i}) := \prod_{1\leq i' \leq M,\:i'\neq i} q_{i'}(\boldsymbol{\theta}_{i'}),
\end{align*}
then we can write:
\begin{align*}
    \text{ELBO}(q;\eta) &= \int q(\boldsymbol{\Theta}) \log \frac{f(\mathbf{y}_{1:T}\:|\: \boldsymbol{\Theta};\eta) g(\boldsymbol{\Theta})}{q(\boldsymbol{\Theta})} \; d\boldsymbol{\Theta} \\
    &= \int \prod_{i'=1}^M q_{i'}(\boldsymbol{\theta}_{i'}) \log \frac{ f(\mathbf{y}_{1:T}\:|\: \boldsymbol{\Theta};\eta) \prod_{i'=1}^M g_{i'}(\boldsymbol{\theta}_{i'})}{\prod_{i'=1}^M q_{i'}(\boldsymbol{\theta}_{i'})} \; d\boldsymbol{\Theta} \\
    &= \int q_{i}(\boldsymbol{\theta}_i) \log \frac{\exp\left\{\E_{q_{-i}}[\log f(\mathbf{y}_{1:T}\:|\: \boldsymbol{\Theta};\eta)]\right\}g_{i}(\boldsymbol{\theta}_i)}{q_{i}(\boldsymbol{\theta}_i)} \; d\boldsymbol{\theta}_i - \sum_{\substack{1 \leq i' \leq M \\ i' \neq i}} \text{KL}(q_{i'}(\boldsymbol{\theta}_{i'})\:\lVert\: g_{i'}(\boldsymbol{\theta}_{i'})).
\end{align*}
Note that only the first term in the last line above depends on $q_i$, so if we treat $q_{-i}$ as fixed, then: 
\begin{align}
    \max_{q_i}  \; \text{ELBO}( q \:;\: \eta) \label{eq:max-elbo-m}
\end{align}
is equivalent to solving:
\begin{align*}
    \max_{q_i}  \; \int q_{i}(\boldsymbol{\theta}_i) \log \frac{\exp\left\{\E_{q_{-i}}[\log f(\mathbf{y}_{1:T}\:|\: \boldsymbol{\Theta};\eta)]\right\}g_{i}(\boldsymbol{\theta}_i)}{q_{i}(\boldsymbol{\theta}_i)} \; d\boldsymbol{\theta}_i
\end{align*}
If we define the normalizing constant:
\begin{align*}
    C_i := \int \exp\left\{\E_{q_{-i}}[\log f(\mathbf{y}_{1:T}\:|\: \boldsymbol{\Theta};\eta)]\right\} \; d\boldsymbol{\theta}_i
\end{align*}
then we can define a new distribution over $\boldsymbol{\theta}_i$:
\begin{align*}
    \tilde{q}_i(\boldsymbol{\theta}_i) := C^{-1}_i\exp\left\{\E_{q_{-i}}[\log f(\mathbf{y}_{1:T}\:|\: \boldsymbol{\Theta};\eta)]\right\}g_{i}(\boldsymbol{\theta}_i)
\end{align*}
Because $C_i$ does not depend on our choice of $q_i$, we have 
\begin{align*}
    \argmax{q_i} \; \text{ELBO}( q \:;\: \eta) &= \argmax{q_i} \; \int q_i(\boldsymbol{\theta}_i) \log \frac{\exp\left\{\E_{q_{-i}}[\log f(\mathbf{y}_{1:T}\:|\: \boldsymbol{\Theta};\eta)]\right\}g_i(\boldsymbol{\theta}_i)}{q_i(\boldsymbol{\theta}_i)} \; d\boldsymbol{\theta}_i\\
    &= \argmax{q_i} \; \int q_i(\boldsymbol{\theta}_i) \log \frac{C^{-1}_i\exp\left\{\E_{q_{-i}}[\log f(\mathbf{y}_{1:T}\:|\: \boldsymbol{\Theta};\eta)]\right\}g_i(\boldsymbol{\theta}_i)}{q_i(\boldsymbol{\theta}_i)} \; d\boldsymbol{\theta}_i \\
    &=  \argmax{q_i} \; \text{KL}(q_i(\boldsymbol{\theta}_i) \:\lVert\: \tilde{q}_i(\boldsymbol{\theta}_i)).
\end{align*}
If we have a closed-form expression for $\tilde{q}_i$, then the last line above shows that the solution to (\ref{eq:max-elbo-m}) is to set $q_i \equiv \tilde{q}_i$, i.e. we will have:
\begin{align}
    q_i(\boldsymbol{\theta}_i) \propto \exp\left\{\E_{q_{-i}}[\log f(\mathbf{y}_{1:T}\:|\: 
    \boldsymbol{\Theta};\eta)]\right\}g_i(\boldsymbol{\theta}_i). \label{eq:q-soln}
\end{align}
Therefore, to prove Proposition \ref{prop:vb}, we must show that (\ref{eq:q-soln}) implies that each $q_\ell$, $q_k$, and $q_j$ have the functional forms specified in (i)-(iii) of Proposition \ref{prop:vb}. To show this, we begin by defining the key terms of the backfitting procedure:
\begin{align}
    \tilde{r}_t &:= y_t - \sum_{\ell = 1}^L \E_{q_\ell}[\mu_{\ell t}] - \sum_{j = 1}^J \frac{\E_{q_j}[\lambda_{jt} \mu_{jt}]}{\E_{q_j}[\lambda_{jt}]}, \label{eq:mod-resid} \\
    \overline{\lambda}_t &:= \prod_{k=1}^K \E_{q_k}[\lambda_{kt}] \prod_{j=1}^J \E_{q_j}[\lambda_{jt}], \label{eq:lambda-bar} \\
    \delta_t &:= \sum_{\ell=1}^L \Var_{q_\ell}\left(\mu_{\ell t} \right) +  \sum_{j=1}^J\left(\frac{\E_{q_j}[\lambda_{jt} \mu^2_{jt}]}{\E_{q_j}[\lambda_{jt}]} -\frac{\E_{q_j}[\lambda_{jt} \mu_{jt}]^2}{\E_{q_j}[\lambda_{jt}]^2} \right). \label{eq:delta}
\end{align}
Note the slight abuse of notation above, as we are now using the symbols $\ell$, $k$, and $j$ both to index the model components as well as to indicate which type of change-point the component represents (mean, variance, or mean and variance). We also introduce notation to help simplify the derivation of the distributions $q_\ell$, $q_j$, and $q_j$. Let $i \in \{\ell, k, j\}$ be a generic block index, then we define:
\small
\begin{align*}
    \overline{\pi}_{i t} &:= \E_{q_i}\left[\mathbbm{1}{\left\{\tau_i = t\right\}}\right],\\
    \overline{b}_{i t} &:= \E_{q_i}[b_i \:|\: \tau_i = t], \\
    \overline{\mu}_{it} &:= \E_{q_i}[\mu_{i t}] = \sum_{t'=1}^{T} \mathbbm{1}{\{t \geq t'\}}\E_{q_i}[b_i \:|\:\tau_i = t']q_i(\tau_i =t') = \sum_{t'=1}^t \overline{b}_{i t'} \overline{\pi}_{i t'}, \\
    \overline{\mu}_t &:= \E_q[\mu_t] = \sum_{\ell=1}^L \overline{\mu}_{\ell t} + \sum_{j=1}^J \overline{\mu}_{jt}, \quad \overline{\mu}_{-jt} := \overline{\mu}_t - \overline{\mu}_{jt}, \quad \overline{\mu}_{-\ell t} := \overline{\mu}_t - \overline{\mu}_{\ell t}, \\
    \overline{\sigma}^2_{it} &:= \E_{q_i}\left[\left(b_i - \overline{b}_{i t}\right)^2\:|\: \tau_i = t\right], \\
    \overline{\mu^2_{it}} &:= \E_{q_i}[\mu^2_{it}] = \sum_{t'=1}^{T} \mathbbm{1}{\{t \geq t'\}}\E_{q_i}[b^2_i \:|\:\tau_i = t']q_i(\tau_i =t') = \sum_{t'=1}^t (\overline{b}^2_{i t'} + \overline{\sigma}^2_{i t'}) \overline{\pi}_{it'}, \\
    \overline{s}_{i t} &:= \E_{q_i}[s_i \:|\: \tau_i = t], \\
    \overline{\lambda}_{it} &:= \E_{q_i}[\lambda_{it}] = \sum_{t'=1}^T \E_{q_i}[s_i \:|\: \tau_i = t']^{\mathbbm{1}\{t\geq t'\}} q_i(\tau_i = t') = \sum_{t'=1}^t \overline{s}_{i t'} \overline{\pi}_{i t'} + \sum_{t'=t+1}^T \overline{\pi}_{i t'}.
\end{align*}
\normalsize
Let $p\left(\mathbf{y}_{1:T} \:|\:\boldsymbol{ \Theta}\right)$ denote the marginal likelihood (\ref{eq:dgp}) under the MICH model. Let $p_\ell$, $p_k$, and $p_j$ be the respective priors for the parameter blocks $\boldsymbol{\theta}_\ell$, $\boldsymbol{\theta}_k$, and $\boldsymbol{\theta}_j$. Note that we can write:
\begin{align*}
    \text{ELBO}(q\:) &= \int q(\boldsymbol{\Theta}) \log \frac{ p(\mathbf{y}_{1:T},\boldsymbol{\Theta})}{q(\boldsymbol{\Theta})} \; d\boldsymbol{\Theta} \\
    &= \E_q\left[\log p(\mathbf{y}_{1:T},\boldsymbol{\Theta})\right] - \sum_{\ell=1}^L \text{KL}(q_\ell\:\lVert\: p_\ell) - \sum_{k=1}^K \text{KL}(q_k\:\lVert\: p_k) -\sum_{j=1}^J \text{KL}(q_j\:\lVert\: p_j).
\end{align*}
Note that if $X\sim\text{Gamma}(\alpha,\beta)$, then $\E[\log X] = \psi(\alpha) - \log \beta$, where $\psi$ is the digamma function. So if $q_\ell$, $q_k$, and $q_j$ do in fact take the forms claimed in (i)-(iii) of Proposition \ref{prop:vb}, then we have::
\small
\begin{align*}
    \E_q\left[\log p(\mathbf{y}_{1:T},\boldsymbol{\Theta})\right]&= - \frac{1}{2} \sum_{t=1}^{T} \E_q[\lambda_t(y_t -\mu_t)^2] +  \frac{1}{2}\sum_{t=1}^T \left[\sum_{j=1}^J \E_{q_j}[\log \lambda_{jt}]+ \sum_{k=1}^K \E_{q_k}[\log \lambda_{kt}]\right] +C \\
    &= - \frac{1}{2} \sum_{t=1}^{T} \overline{\lambda}_t(\Tilde{r}_t^2 + \delta_t) + \frac{1}{2}\sum_{t=1}^T \sum_{j=1}^J (T-t+1)\overline{\pi}_{jt}[\psi(\overline{u}_{jt}) - \log\overline{v}_{jt}]  \\
    &\quad + \frac{1}{2}\sum_{t=1}^T \sum_{k=1}^K(T-t+1)\overline{\pi}_{kt}[\psi(\overline{u}_{kt}) - \log\overline{v}_{kt}] + C,
\end{align*}
\normalsize
where $C$ is a constant that does not depend on $\overline{\boldsymbol{\theta}}_{1:N}$. Noting that $\tilde{r}_t$ and $\delta_t$ can be expressed as continuously differentiable functions of known constants $\overline{\boldsymbol{\theta}}_{1:N}$, then we see that $\E_q\left[\log p(\mathbf{y}_{1:T},\boldsymbol{\Theta})\right]$ is a continuously differentiable function of $\overline{\boldsymbol{\theta}}_{1:N}$. Next, when $q_\ell$ takes the form in (i):
\begin{align*}
     \text{KL}(q_\ell\:\lVert\: p_\ell) &= \E_{q_\ell}\left[\sum_{t=1}^T \mathbbm{1}\{\tau_\ell = t\}\left(\log \frac{q(b_\ell|\tau_\ell =t)}{p(b_\ell|\tau_\ell =t)}+  \log \frac{\overline{\pi}_{\ell t}}{\pi_{\ell t}}\right) \right] \\
     &= \E_{q_\ell}\left[\sum_{t=1}^T \mathbbm{1}\{\tau_\ell = t\}\left(\frac{1}{2}\log\frac{\overline{\omega}_{\ell t}}{\omega_\ell}- \frac{\overline{\omega}_{\ell t}(b_\ell - \overline{b}_{\ell t})^2 - \omega_\ell b_\ell^2}{2} +  \log \frac{\overline{\pi}_{\ell t}}{\pi_{\ell t}}\right) \right] \\
     &= \sum_{t=1}^T \overline{\pi}_{\ell t}\left[\frac{1}{2}\log\frac{\overline{\omega}_{\ell t}}{\omega_\ell}- \frac{\overline{\omega}_{\ell t}\E_{q_\ell}\left[(b_\ell - \overline{b}_{\ell t})^2|\tau_\ell=t\right] - \omega_\ell \E_{q_\ell}\left[b_\ell^2|\tau_\ell=t\right]}{2} +  \log \frac{\overline{\pi}_{\ell t}}{\pi_{\ell t}}\right] \\
     &=\sum_{t=1}^T \overline{\pi}_{\ell t}\left[\frac{1}{2}\log\frac{\overline{\omega}_{\ell t}}{\omega_\ell}- \frac{1 - \omega_\ell (\overline{\omega}_{\ell t}^{-1} + \overline{b}^2_{\ell t})}{2} +  \log \frac{\overline{\pi}_{\ell t}}{\pi_{\ell t}}\right]. 
\end{align*}
\normalsize
So $\text{KL}(q_\ell\:\lVert\: p_\ell)$ is a continuously differentiable function of $\overline{\boldsymbol{\theta}}_\ell$. Next, when $q_k$ takes the form in (ii):
\scriptsize
\begin{align*}
    \text{KL}(q_k\:\lVert\: p_k) &= \E_{q_k}\left[\sum_{t=1}^T \mathbbm{1}\{\tau_k = t\}\left(\log \frac{q(s_k|\tau_k =t)}{p(s_k|\tau_k =t)}+  \log \frac{\overline{\pi}_{k t}}{\pi_{k t}}\right) \right]  \\
    &= \E_{q_k}\left[\sum_{t=1}^T \mathbbm{1}\{\tau_k = t\}\left(\overline{u}_{kt}\log\overline{v}_{kt} -u_k \log v_k -\log \frac{\Gamma(\overline{u}_{kt})}{\Gamma(u_k)} + (\overline{u}_{kt} -u_k)\log s_k -(\overline{v}_{kt} - v_k)s_k +  \log \frac{\overline{\pi}_{k t}}{\pi_{k t}}\right) \right] \\
    &= \sum_{t=1}^T \overline{\pi}_{kt}\left(\overline{u}_{kt}\log\overline{v}_{kt} -u_k \log v_k -\log \frac{\Gamma(\overline{u}_{kt})}{\Gamma(u_k)} + (\overline{u}_{kt} -u_k)\E_{q_k}[\log s_k|\tau_k=t] -(\overline{v}_{kt} - v_k)\E_{q_k}[s_k|\tau_k=t] +  \log \frac{\overline{\pi}_{k t}}{\pi_{k t}}\right) \\
    &= \sum_{t=1}^T \overline{\pi}_{kt}\left[\overline{u}_{kt}\log\overline{v}_{kt} -u_k \log v_k -\log \frac{\Gamma(\overline{u}_{kt})}{\Gamma(u_k)} + (\overline{u}_{kt} -u_k)[\psi(\overline{u}_{kt}) - \log \overline{v}_{kt}] -\frac{(\overline{v}_{kt} - v_k)\overline{u}_{kt}}{\overline{v}_{kt}} +  \log \frac{\overline{\pi}_{k t}}{\pi_{k t}}\right] \\
    &= \sum_{t=1}^T \overline{\pi}_{kt}\left[u_k \log \frac{\overline{v}_{kt}}{v_k} -\log \frac{\Gamma(\overline{u}_{kt})}{\Gamma(u_k)} + (\overline{u}_{kt} -u_k)\psi(\overline{u}_{kt}) -\frac{(\overline{v}_{kt} - v_k)\overline{u}_{kt}}{\overline{v}_{kt}} +  \log \frac{\overline{\pi}_{k t}}{\pi_{k t}}\right].
\end{align*}
\normalsize
So $\text{KL}(q_k\:\lVert\: p_k)$ is a continuously differentiable function of $\overline{\boldsymbol{\theta}}_k$. Next, when $q_j$ takes the form in (iii):
\scriptsize
\begin{align*}
     \text{KL}(q_j\:\lVert\: p_j) &= \E_{q_j}\left[\sum_{t=1}^T \mathbbm{1}\{\tau_j = t\}\left(\log \frac{q(b_j|s_j,\tau_j =t)}{p(b_j|s_,j\tau_j =t)} +\log \frac{q(s_j|\tau_j =t)}{p(s_j|,\tau_j =t)} + \log \frac{\overline{\pi}_{j t}}{\pi_{j t}}\right) \right] \\
     &= \E_{q_j}\left[\sum_{t=1}^T \mathbbm{1}\{\tau_j = t\}\left(\frac{1}{2}\log\frac{\overline{\omega}_{j t}}{\omega_j}- \frac{\overline{\omega}_{j t}s_j(b_j - \overline{b}_{j t})^2 - \omega_js_j b_j^2}{2}\right) \right] \\
     &\quad + \E_{q_j}\left[\sum_{t=1}^T \mathbbm{1}\{\tau_j = t\}\left(\overline{u}_{jt}\log\overline{v}_{jt} -u_j \log v_j -\log \frac{\Gamma(\overline{u}_{jt})}{\Gamma(u_j)} + (\overline{u}_{jt} -u_j)\log s_j -(\overline{v}_{jt} - v_j)s_j + \log \frac{\overline{\pi}_{j t}}{\pi_{j t}}\right) \right] \\
     &= \sum_{t=1}^T\overline{\pi}_{jt}\left[\frac{1}{2}\log\frac{\overline{\omega}_{j t}}{\omega_j}- \frac{\overline{\omega}_{j t}\E_{q_j}[s_j\E_{q_j}[(b_j - \overline{b}_{j t})^2|s_j,\tau_j=t]|\tau_j=t] - \omega_j\E_{q_j}[s_j\E_{q_j}[ b_j^2|s_j,\tau_j=t]|\tau_j=t]}{2} \right] \\
     &\quad + \sum_{t=1}^T\overline{\pi}_{jt}\left[\overline{u}_{jt}\log\overline{v}_{jt} -u_j \log v_j -\log \frac{\Gamma(\overline{u}_{jt})}{\Gamma(u_j)} + (\overline{u}_{jt} -u_j)\E_{q_j}[\log s_j|\tau_j=t] -(\overline{v}_{jt} - v_j)\E_{q_j}[s_j|\tau_j=t] +   \log \frac{\overline{\pi}_{j t}}{\pi_{j t}}\right] \\
     &= \sum_{t=1}^T\overline{\pi}_{jt}\left[\frac{1}{2}\log\frac{\overline{\omega}_{j t}}{\omega_j}- \frac{1 - \omega_j(\overline{\omega}_{jt}^{-1} +\overline{b}_{jt}^2\E_{q_j}[s_j|\tau_j=t])}{2} \right] \\
     &\quad + \sum_{t=1}^T \overline{\pi}_{jt}\left[u_j \log \frac{\overline{v}_{jt}}{v_j} -\log \frac{\Gamma(\overline{u}_{jt})}{\Gamma(u_j)} + (\overline{u}_{jt} -u_j)\psi(\overline{u}_{jt}) -\frac{(\overline{v}_{jt} - v_j)\overline{u}_{jt}}{\overline{v}_{jt}} +  \log \frac{\overline{\pi}_{j t}}{\pi_{j t}}\right] \\
     &= \sum_{t=1}^T\overline{\pi}_{jt}\left[\frac{1}{2}\log\frac{\overline{\omega}_{j t}}{\omega_j}- \frac{1}{2} + \frac{\omega_j}{2\overline{\omega}_{jt}} + \frac{\omega_j\overline{u}_{jt}\overline{b}_{jt}^2}{2\overline{v}_{jt}} + u_j \log \frac{\overline{v}_{jt}}{v_j} -\log \frac{\Gamma(\overline{u}_{jt})}{\Gamma(u_j)} + (\overline{u}_{jt} -u_j)\psi(\overline{u}_{jt}) -\frac{(\overline{v}_{jt} - v_j)\overline{u}_{jt}}{\overline{v}_{jt}} +  \log \frac{\overline{\pi}_{j t}}{\pi_{j t}}\right] 
\end{align*}
\normalsize
So $\text{KL}(q_j\:\lVert\: p_j)$ is a continuously differentiable function of $\overline{\boldsymbol{\theta}}_j$. Therefore, $\text{ELBO}(q\:)$ can be expressed as a continuously differentiable function of $\overline{\boldsymbol{\theta}}_{1:N}$, i.e. we can define:
\begin{align}
    F(\overline{\boldsymbol{\theta}}_{1:N}) := \text{ELBO}(q)  \label{eq:F}
\end{align}
We now show that the marginal distributions $q_\ell(\boldsymbol{\theta}_\ell)$, $q_k(\boldsymbol{\theta}_k)$, and $q_j(\boldsymbol{\theta}_j)$ that are implied by the expression in (\ref{eq:q-soln}) are precisely the same distributions we claim solve the appropriate maximization problems in (i)-(iii) of Proposition \ref{prop:vb}. Throughout the proof we take a general approach and allow the hyperparameters $\omega_i,u_i,v_i$, and $\boldsymbol{\pi}_{i,1:T}$ to vary by component $i$.

\textbf{Case (i):} 

For fixed a $\ell' \in\{1,\ldots,L\}$, we treat the distributions $\{q_\ell\}_{\ell\neq\ell'}$, $\{q_k\}_{k=1}^K$, and $\{q_j\}_{j=1}^J$ as fixed and known. Then by (\ref{eq:q-soln}), the solution to:
\begin{align}
    \max_{q_{\ell'}} \text{\normalfont ELBO}\left(q\right), \label{eq:elbo-l}
\end{align}
is to pick $q_{\ell'}$ so that:
\begin{align*}
    \log q_{\ell'}(\boldsymbol{\theta}_{\ell'}) &\underset{\boldsymbol{\theta}_{\ell'}}{\propto} \E_{q_{-\ell'}} \left[\log p\left(\mathbf{y} \:|\: \boldsymbol{\Theta}\right)\right] + \log p_{\ell'}(\boldsymbol{\theta}_{\ell'}) \\
    &\underset{\boldsymbol{\theta}_{\ell'}}{\propto} - \frac{1}{2}\sum_{t=1}^{T}  \E_{q_{-\ell'}}\left[\lambda_{t}\left(y_t - \mu_t\right)^2\right] + \log p_{\ell'}(\boldsymbol{\theta}_{\ell'}).
\end{align*}
Define the partial residual:
\begin{align*}
    r_{-\ell' t} &:= y_t - \mu_t + \mu_{\ell' t}.
\end{align*}
If we fix $\tau_{\ell'} = t'$, then we have:
\begin{align*}
    \log q_{\ell'}(b_{\ell'},\tau_{\ell'}=t') 
    &\underset{\boldsymbol{\theta}_{\ell'}}{\propto}- \frac{1}{2}\sum_{t=t'}^{T}  \E_{q_{-\ell'}}\left[\lambda_{t}\left(b_{\ell'}^2 - 2r_{-\ell't}b_{\ell'} \right)\right] - \frac{b_{\ell'}^2\omega_{\ell'}}{2} + \log \pi_{\ell' t} \\
    &\;= - \frac{1}{2}\sum_{t=t'}^{T}  \left(\E_{q_{-\ell'}}\left[\lambda_{t}\right]b_{\ell'}^2 - 2\E_{q_{-\ell'}}\left[\lambda_{t}r_{-\ell't}\right]b_{\ell'}\right) - \frac{b_{\ell'}^2\omega_{\ell'}}{2} + \log \pi_{\ell' t}\\
    &\;= -\frac{1}{2}\left[b_{\ell'}^2\left(\omega_{\ell'} +\sum_{t=t'}^{T} \overline{\lambda}_t\right) - 2 b_{\ell'}\sum_{t=t'}^{T}\E_{q_{-\ell'}}\left[\lambda_{t}r_{-\ell't}\right]\right]  + \log \pi_{\ell' t}.
\end{align*}
We can now calculate the posterior parameters for $b_{\ell'}$ given $\tau_{\ell'} = t'$:
\begin{align*}
    \overline{\omega}_{\ell't'} &= \omega_{\ell'} +\sum_{t=t'}^{T} \overline{\lambda}_t \\
    \overline{b}_{\ell't'} &= \overline{\omega}_{\ell't'}^{-1}  \sum_{t=t'}^{T} \E_{q_{-\ell'}}\left[\lambda_{t}r_{-\ell't}\right].
\end{align*}
We now have:
\begin{align*}
    \log q_{\ell'}(b_{\ell'},\tau_{\ell'}=t') 
    &\underset{\boldsymbol{\theta}_{\ell'}}{\propto} \log \overline{\omega}_{\ell't'} - \frac{\overline{\omega}_{\ell't'}(b_{\ell'} - \overline{b}_{\ell't'})^2}{2} - \log \overline{\omega}_{\ell't'} + \frac{\overline{\omega}_{\ell't'}\overline{b}^2_{\ell't'}}{2}+ \log \pi_{\ell' t'}
\end{align*}
The first two terms above are just the log density of a Gaussian distribution with mean $\overline{b}_{\ell't'} $ and precision $\overline{\omega}_{\ell't'}$. Therefore, the $q_{\ell'}$ that solves (\ref{eq:elbo-l}) satisfies:
\begin{align*}
    b_{\ell'} \:|\: \tau_{\ell'} = t' \sim \mathcal{N}(\overline{b}_{j't'}, \overline{\omega}^{-1}_{j't'}),
\end{align*}
which matches (\ref{eq:b-smcp}) of the mean-scp posterior. Integrating the joint density $q_{\ell'}(\boldsymbol{\theta}_{\ell'})$ with respect to $b_{\ell'}$ and calculating the normalizing constant, we are left with: 
\begin{align}
    q_{\ell'}(\tau_{\ell'} = t') &= \frac{\pi_{\ell't'} \overline{\omega}_{\ell't'}^{-\frac{1}{2}} \exp[\overline{\omega}_{\ell't'} \overline{b}^2_{\ell't'}/2]}{\sum_{t=1}^T \pi_{\ell't} \overline{\omega}_{\ell't}^{-\frac{1}{2}} \exp[\overline{\omega}_{\ell't} \overline{b}^2_{\ell't}/2]} \label{eq:optimal-pi-l}\\ 
    &:= \overline{\pi}_{\ell't'}, \notag
\end{align}
which matches (\ref{eq:gamma-post-cat1}). Therefore, the $q_{\ell'}$ that solves (\ref{eq:elbo-l}) has the same functional form as the posterior distribution of the mean-scp with parameters $\overline{\boldsymbol{\theta}}_{\ell'} := \{\overline{b}_{\ell't},\overline{\omega}_{\ell't},\overline{\pi}_{\ell't}\}_{t=1}^T$. It remains to show $\overline{\boldsymbol{\theta}}_{\ell'}$ is the same set of parameters defined in (i) of Proposition \ref{prop:vb}. Define a partial mean residual $\tilde{r}_{-\ell t} :=  \tilde{r}_{t} + \E_{q_{\ell}}[\mu_{\ell t}]$. Then we have:
\begin{align*}
    \E_{q_{-\ell}}\left[\lambda_{t}r_{-\ell t}\right] &= \E_{q_{-\ell}}\left[\prod_{j=1}^J\lambda_{jt}\prod_{k=1}^K\lambda_{kt} \left(y_t  - \sum_{\ell'\neq \ell} \mu_{\ell't} - \sum_{j=1}^J \mu_{jt}\right)\right] \\
    &=  \prod_{k=1}^K\E_{q_{k}}\left[\lambda_{kt}\right]\left(y_t  - \sum_{\ell'\neq \ell} \E_{q_{\ell'}}\left[\mu_{\ell' t}\right] - \sum_{j=1}^J \E_{q_{-\ell}}\left[\lambda_{jt}\mu_{jt}\prod_{j'\neq j} \lambda_{j't}\right]\right) \tag{by (\ref{eq:mean-field})} \\
    &= \overline{\lambda}_{t}\left(y_t  - \overline{\mu}_{-\ell t} -  \sum_{j=1}^J\overline{\lambda}_{jt}^{-1} \E_{q_{j}}\left[\lambda_{jt}\mu_{jt}\right]\right) \\
    &=  \overline{\lambda}_{t} \tilde{r}_{-\ell t}.
\end{align*}
Returning to the posterior parameters for $b_{\ell}$ and $\tau_{\ell}$, we can now write: 
\begin{align}
    \overline{\omega}_{\ell t} &=  \omega_{\ell} +  \sum_{t'=t}^{T} \overline{\lambda}_{t'}, \label{eq:optimal-omega-l} \\
    \overline{b}_{\ell t}  &= \overline{\omega}_{\ell t}^{-1}  \sum_{t'=t}^{T} \overline{\lambda}_{t'}\tilde{r}_{-\ell t'}, \label{eq:optimal-b-l}
\end{align}
Examining the mean-scp posterior parameters (\ref{eq:mean-scp-post-omega})-(\ref{eq:mean-scp-post-pi}), it is clear that:
\begin{align*}
    \overline{\boldsymbol{\theta}}_\ell = \normalfont{\texttt{mean-scp}}(\tilde{\mathbf{r}}_{-\ell,1:T} \:;\: \overline{\boldsymbol{\lambda}}_{1:T}, \omega_{\ell}, \boldsymbol{\pi}_{\ell,1:T}).
\end{align*}
We now show that $\overline{\boldsymbol{\theta}}_\ell$ is the unique maximizer of $F(\overline{\boldsymbol{\theta}}_{1:N})$ holding $\overline{\boldsymbol{\theta}}_{1:N}\setminus\overline{\boldsymbol{\theta}}_\ell$ fixed. First note that we have:
\begin{align*}
     F(\overline{\boldsymbol{\theta}}_{1:N}) &\underset{\overline{\boldsymbol{\theta}}_\ell}{\propto} - \frac{1}{2} \sum_{t=1}^{T} \overline{\lambda}_t(\Tilde{r}_t^2 + \delta_t) - \sum_{t=1}^T \overline{\pi}_{\ell t}\left[\frac{1}{2}\log\frac{\overline{\omega}_{\ell t}}{\omega_\ell}- \frac{1 - \omega_\ell (\overline{\omega}_{\ell t}^{-1} + \overline{b}^2_{\ell t})}{2} +  \log \frac{\overline{\pi}_{\ell t}}{\pi_{\ell t}}\right] \\
     &\underset{\overline{\boldsymbol{\theta}}_\ell}{\propto} - \frac{1}{2} \sum_{t=1}^{T} \overline{\lambda}_t\left[\left(\tilde{r}_{-\ell t} -\sum_{t'=1}^t \overline{b}_{\ell t'} \overline{\pi}_{\ell t'}\right)^2 + \sum_{t'=1}^t (\overline{b}^2_{\ell t'} + \overline{\omega}^{-1}_{\ell t'}) \overline{\pi}_{\ell t'} - \left(\sum_{t'=1}^t \overline{b}_{\ell t'} \overline{\pi}_{\ell t'}\right)^2\right] \\
     &\quad- \sum_{t=1}^T \overline{\pi}_{\ell t}\left[\frac{1}{2}\log\frac{\overline{\omega}_{\ell t}}{\omega_\ell}- \frac{1 - \omega_\ell (\overline{\omega}_{\ell t}^{-1} + \overline{b}^2_{\ell t})}{2} +  \log \frac{\overline{\pi}_{\ell t}}{\pi_{\ell t}}\right] \\
     &\underset{\overline{\boldsymbol{\theta}}_\ell}{\propto} - \frac{1}{2} \sum_{t=1}^{T} \overline{\lambda}_t\left[ \sum_{t'=1}^t (\overline{b}^2_{\ell t'} + \overline{\omega}^{-1}_{\ell t'}) \overline{\pi}_{\ell t'} -2\tilde{r}_{-\ell t}\sum_{t'=1}^t \overline{b}_{\ell t'} \overline{\pi}_{\ell t'}\right] \\
     &\quad- \sum_{t=1}^T \overline{\pi}_{\ell t}\left[\frac{1}{2}\log\frac{\overline{\omega}_{\ell t}}{\omega_\ell}- \frac{1 - \omega_\ell (\overline{\omega}_{\ell t}^{-1} + \overline{b}^2_{\ell t})}{2} +  \log \frac{\overline{\pi}_{\ell t}}{\pi_{\ell t}}\right]   
\end{align*}
Assume that $\overline{\pi}_{\ell t}> 0$ for the moment (we will momentarily show that $\pi_{\ell t} > 0$ implies that this is true for the optimal value of $\overline{\pi}_{\ell t}$), then the expression above is strictly concave as a function of $\overline{b}_{\ell t}$ and:
\begin{align*}
    \frac{\partial F}{\partial \overline{b}_{\ell t}} = \overline{\pi}_{\ell t}\sum_{t'=t}^T \overline{\lambda}_{t'}\tilde{r}_{-\ell t'} -  \left(\omega_\ell + \sum_{t'=t}^T\overline{\lambda}_{t'}\right)\overline{\pi}_{\ell t}\overline{b}_{\ell t}.
\end{align*}
Since $\overline{\pi}_{\ell t}> 0$, $F$ is uniquely maximized by setting:
\begin{align*}
    \overline{b}_{\ell t} &= \frac{\sum_{t'=t}^T \overline{\lambda}_{t'}\tilde{r}_{-\ell t'}}{\omega_\ell + \sum_{t'=t}^T\overline{\lambda}_{t'}}, 
\end{align*}
which matches (\ref{eq:optimal-b-l}). Next:
\begin{align*}
    \frac{\partial F}{\partial \overline{\omega}_{\ell t}} = \frac{\overline{\pi}_{\ell t}\left(\omega_\ell + \sum_{t'=t}^T\overline{\lambda}_{t}\right)} {2\overline{\omega}_{\ell t}^2} - \frac{\overline{\pi}_{\ell t}} {2\overline{\omega}_{\ell t}}     
    \begin{cases}
        >0, &\text{if } \overline{\omega}_{\ell t} < \omega_\ell + \sum_{t'=t}^T\overline{\lambda}_{t'}, \\
        =0, &\text{if } \overline{\omega}_{\ell t} = \omega_\ell + \sum_{t'=t}^T\overline{\lambda}_{t'}, \\
        <0, &\text{if } \overline{\omega}_{\ell t} > \omega_\ell + \sum_{t'=t}^T\overline{\lambda}_{t'}.
    \end{cases}
\end{align*}
So when $\overline{\pi}_{\ell t}> 0$, $F$ is uniquely maximized by setting:
\begin{align*}
    \overline{\omega}_{\ell t} = \omega_\ell + \sum_{t'=t}^T\overline{\lambda}_{t'}, 
\end{align*}
which matches (\ref{eq:optimal-omega-l}). Lastly, for any constants $\{C_t\}_{t=1}^T$, we can solve:
\begin{align*}
    \max_{\overline{\boldsymbol{\pi}}_{\ell, 1:T} \in \mathcal{S}^T} \sum_{t=1}^T C_t\overline{\pi}_{\ell t} - \overline{\pi}_{\ell t} \log \overline{\pi}_{\ell t} \text{ s.t. } \sum_{t=1}^T \overline{\pi}_{\ell t} = 1
\end{align*}
by solving:
\begin{align*}
     \max_{\overline{\boldsymbol{\pi}}_{\ell, 1:T} \in \mathcal{S}^T} \sum_{t=1}^T C_t\overline{\pi}_{\ell t} - \overline{\pi}_{\ell t} \log \overline{\pi}_{\ell t} - \alpha \left(1 - \sum_{t=1}^T \pi_{\ell t}\right) 
\end{align*}
where $\alpha \geq 0$ is the Lagrange multiplier. Furthermore, since  $C_t\overline{\pi}_{\ell t} - \overline{\pi}_{\ell t} \log \overline{\pi}_{\ell t}$ is strictly concave as a function of $\overline{\pi}_{\ell t}$, and $\mathcal{S}^T$ is closed and convex, the solution will be unique. The solution is determined by the following KKT conditions:
\begin{align*}
    C_t - 1-\log \overline{\pi}_{\ell t}  + \alpha &=0, \sforall t \in \{1, \ldots, T\} \\
    1 - \sum_{t=1}^T \overline{\pi}_{\ell t} &= 0.
\end{align*}
The first condition implies $\overline{\pi}_{\ell t} = \exp[C_t + \alpha - 1]$, which combined with the the second condition implies:
\begin{align}
   \alpha = 1-\log\left(\sum_{t=1}^T\exp[C_t] \right) \implies \overline{\pi}_{\ell t} = \frac{\exp[C_t]}{\sum_{t'=1}^T \exp[C_{t'}]}. \label{eq:optimal-pi}
\end{align}
Reexamining $F$, we see that:
\begin{align*}
    C_t&= -\frac{\left(\overline{b}_{\ell t}^2 + \overline{\omega}_{\ell t}^{-1}\right)\sum_{t'=t}^T\overline{\lambda}_{t'}}{2} +\overline{b}_{\ell t}\sum_{t'=t}^T\overline{\lambda}_{t'}\tilde{r}_{-\ell t'} - \frac{1}{2}\log\frac{\overline{\omega}_{\ell t}}{\omega_\ell} + \frac{1 - \omega_\ell (\overline{\omega}_{\ell t}^{-1} + \overline{b}^2_{\ell t})}{2} +\log \pi_{\ell t}\\
    &= -\frac{\left(\overline{b}_{\ell t}^2 + \overline{\omega}_{\ell t}^{-1}\right)\left(\omega_\ell + \sum_{t'=t}^T\overline{\lambda}_{t'}\right)}{2} +\overline{b}_{\ell t}\sum_{t'=t}^T\overline{\lambda}_{t'}\tilde{r}_{-\ell t'} - \frac{1}{2}\log\overline{\omega}_{\ell t} +\log \pi_{\ell t} + \frac{1}{2}\left(1 +\log \omega_\ell\right) 
\end{align*}
Note that $\pi_{\ell t} > 0$ and the expression above imply the optimal value of $\overline{\pi}_{\ell t}$ is positive for any values of $\overline{b}_{\ell t}$ and $\overline{\omega}_{\ell t}$, so we are at an interior solution and the optimal values of $\overline{b}_{\ell t}$ and $\overline{\omega}_{\ell t}$ must take the form (\ref{eq:optimal-b-l}) and (\ref{eq:optimal-omega-l}), but for these values we have:
\begin{align*}
    C_t &= -\frac{\left(\overline{b}_{\ell t}^2 + \overline{\omega}_{\ell t}^{-1}\right) \overline{\omega}_{\ell t}}{2} + \overline{\omega}_{\ell t}\overline{b}^2_{\ell t} - \frac{1}{2}\log\overline{\omega}_{\ell t} +\log \pi_{\ell t} + \frac{1}{2}\left(1 +\log \omega_\ell\right) \\
    &= \frac{\overline{\omega}_{\ell t}\overline{b}^2_{\ell t}}{2} - \frac{1}{2}\log\overline{\omega}_{\ell t} +\log \pi_{\ell t} + \frac{1}{2}\log \omega_\ell, 
\end{align*}
which along with (\ref{eq:optimal-pi}) implies:
\begin{align*}
      \overline{\pi}_{\ell t} = \frac{\pi_{\ell t} \overline{\omega}_{\ell t }^{-\frac{1}{2}} \exp[\overline{\omega}_{\ell t } \overline{b}^2_{\ell t }/2]}{\sum_{t'=1}^T \pi_{\ell t'} \overline{\omega}_{\ell t'}^{-\frac{1}{2}} \exp[\overline{\omega}_{\ell t'} \overline{b}^2_{\ell t'}/2]},
\end{align*}
which matches (\ref{eq:optimal-pi-l}).

\textbf{Case (ii):} 

For fixed $k' \in\{1,\ldots,K\}$, we treat the distributions $\{q(\boldsymbol{\theta}_\ell)\}_{\ell=1}^L$, $\{q(\boldsymbol{\theta}_k)\}_{k \neq k'}$, and $\{q(\boldsymbol{\theta}_{j})\}_{j=1}^J$ as fixed and known. Then by (\ref{eq:q-soln}), the solution to:
\begin{align}
    \max_{q_{k'}} \text{\normalfont ELBO}\left(q\right), \label{eq:elbo-k}
\end{align}
is to pick $q_{k'}$ so that:
\begin{align*}
    \log q_{k'}(\boldsymbol{\theta}_{k'}) &\underset{\boldsymbol{\theta}_{k'}}{\propto} \E_{q_{-k'}} \left[\log p\left(\mathbf{y}_{1:T} \:|\: \boldsymbol{\Theta}\right)\right] + \log p_{k'}(\boldsymbol{\theta}_{k'}) \\
    &\underset{\boldsymbol{\theta}_{k'}}{\propto}  \E_{q_{-k'}}\left[\frac{1}{2}\sum_{t=1}^{T} \log\lambda_t - \frac{1}{2}\sum_{t=1}^{T} \lambda_{t}\left(y_t - \mu_t\right)^2\right] + \log p_{k'}(\boldsymbol{\theta}_{k'}).
\end{align*}
Define $\lambda_{-k't} := \lambda_{k't}^{-1}\lambda_t$. If we fix $\tau_{k'} = t'$, then we have:
\begin{align*}
    \log q_{k'}(s_{k'},\tau_{k'} = t') &\underset{\boldsymbol{\theta}_{k'}}{\propto} \frac{T - t' + 1}{2} \log s_{k'} - \frac{s_{k'}}{2}\sum_{t=t'}^{T} \E_{q_{-k'}}\left[\lambda_{-k't}\left(y_t - \mu_t\right)^2\right] \\
    &\quad - \frac{1}{2}\sum_{t=1}^{t'-1} \E_{q_{-k'}}\left[\lambda_{-k't}\left(y_t - \mu_t\right)^2\right]  + (u_{k'}-1)\log s_{k'} - v_{k'} s_{k'} + \log \pi_{k't'}.
\end{align*}
Define:
\begin{align*}
    \overline{u}_{k't'} &= u_{k'} + \frac{T-t'+1}{2} \\
    \overline{v}_{k't'} &= v_{k'} + \frac{1}{2}\sum_{t=t'}^{T} \E_{q_{-k'}}\left[\lambda_{-k't}\left(y_t - \mu_t\right)^2\right].
\end{align*}
Then we have:
\begin{align*}
    \log q_{k'}(s_{k'},\tau_{k'} = t') 
    &\underset{\boldsymbol{\theta}_{k'}}{\propto} -  \log\Gamma(\overline{u}_{k't'}) + \overline{u}_{k't'} \log \overline{v}_{k't'} + \left(\overline{u}_{k't'} -1\right)\log s_{k'}  - \overline{v}_{k't'} s_{k'} \\
    &\quad + \log \pi_{k't'} + \log\Gamma(\overline{u}_{k't'}) - \overline{u}_{k't'} \log \overline{v}_{k't'} - \frac{1}{2}\sum_{t=1}^{t'-1} \E_{q_{-k'}}\left[\lambda_{-k't}\left(y_t - \mu_t\right)^2\right]
\end{align*}
The first line above is just the log density of a Gamma distribution with shape and rate parameters $\overline{u}_{k't'}$, and $\overline{v}_{k't'}$. Therefore, the $q_{k'}$ that solves (\ref{eq:elbo-k}) satisfies:
\begin{align*}
    s_{k'} \:|\: \tau_{k'} = t' \sim \text{Gamma}(\overline{u}_{k't'}, \overline{v}_{k't'}),
\end{align*}
which matches (\ref{eq:s-sscp}) of the var-scp posterior. Integrating the joint density $q_{k'}(\boldsymbol{\theta}_{k'})$ with respect to $s_{k'}$ and calculating the normalizing constant, we are left with: 
\begin{align*}
    q_{k'}(\tau_{k'} = t) &= \frac{\pi_{k't}\Gamma(\overline{u}_{k't})\overline{v}_{k't}^{-\overline{u}_{k't}}\exp\left( -\frac{1}{2} \sum_{t'=1}^{t-1} \E_{q_{-k'}}\left[\lambda_{-k't'}(y_{t'} - \mu_{t'})^2\right] \right)}{\sum_{t'=1}^T \pi_{k't'}\Gamma(\overline{u}_{k't'})\overline{v}_{k't'}^{-\overline{u}_{k't'}}\exp\left( -\frac{1}{2} \sum_{t''=1}^{t'-1} \E_{q_{-k'}}\left[\lambda_{-k't''}(y_{t''} - \mu_{t''})^2\right] \right)} \\
    &:= \overline{\pi}_{k't}, 
\end{align*}
which matches (\ref{eq:gamma-post-cat1}). Therefore, the $q_{k'}$ that solves (\ref{eq:elbo-k}) has the same functional form as the posterior distribution of the var-scp model with parameters $\overline{\boldsymbol{\theta}}_{k'} := \{\overline{u}_{k't},\overline{v}_{k't},\overline{\pi}_{k't}\}_{t=1}^T$. It remains is to show $\overline{\boldsymbol{\theta}}_{k'}$ is the same set of parameters defined in (ii) of Proposition \ref{prop:vb}. Define a partial scale residual $\overline{\lambda}_{-k t} := \overline{\lambda}_{kt}^{-1}\overline{\lambda}_t$. Then using the residual $\tilde{r}_t$ (\ref{eq:mod-resid}) and the variance correction term $\delta_t$ (\ref{eq:delta}), we have:
\footnotesize
\begin{align*}
    \E_{q_{-k}}\left[\lambda_{-kt}(y_t - \mu_t)^2\right] &= \E_{q_{-k}}\left[\prod_{j =1}^J\lambda_{jt} \prod_{k' \neq k} \lambda_{k't}\left(y_t  - \sum_{j'=1}^J \mu_{j't} - \sum_{\ell=1}^L \mu_{\ell t}\right)^2\right] \\
    &= \prod_{j =1}^J\E_{q_{j}}\left[\lambda_{jt}\right] \prod_{k'\neq k} \E_{q_{k'}}\left[\lambda_{k't}\right]\E_{q_{-k}}\left[\left(y_t  - \sum_{\ell=1}^L \mu_{\ell t}\right)^2\right] \\
    &\quad -2  \prod_{k'\neq k} \E_{q_{k'}}\left[\lambda_{k't}\right]\E_{q_{-k}}\left[y_t   - \sum_{\ell=1}^L \mu_{\ell t}\right] \sum_{j=1}^J \E_{q_{j}}\left[\lambda_{jt} \mu_{jt}\right]\prod_{j' \neq j}\E_{q_{j'}}[\lambda_{j't}]  \\
    &\quad +  \prod_{k' \neq k} \E_{q_{k'}}\left[ \lambda_{k't}\right]\E_{q_{-k}}\left[\sum_{j=1}^J\sum_{j'=1}^J  \prod_{j''=1}^J \lambda_{j''t} \mu_{jt} \mu_{j't}\right] \\
    &= \overline{\lambda}_{-kt}\left[\left(y_t  - \sum_{\ell=1}^L \overline{\mu}_{\ell t}\right)^2 + \sum_{\ell=1}^L \left(\overline{\mu^2_{\ell t}} - \overline{\mu}^2_{\ell t} \right) \right] \\
    &\quad -2  \overline{\lambda}_{-kt}\left(y_t   - \sum_{\ell=1}^L\overline{\mu}_{\ell t}\right) \sum_{j=1}^J \overline{\lambda}_{jt}^{-1} \E_{q_{j}}\left[\lambda_{jt} \mu_{jt}\right] \\
    &\quad + \overline{\lambda}_{-kt} \sum_{j=1}^J\sum_{j'=1}^J \overline{\lambda}_{jt}^{-1}\E_{q_{-k}}\left[\lambda_{jt} \mu_{jt}\right]\overline{\lambda}_{j't}^{-1}\E_{q_{j'}}\left[ \lambda_{j't}\mu_{j't}\right] \\
    &\quad + \overline{\lambda}_{-kt} \sum_{j=1}^J\left(\overline{\lambda}_{jt}^{-1}\E_{q_{j}}\left[ \lambda_{jt}\mu_{jt}^2\right] -\overline{\lambda}_{jt}^{-2}\E_{q_{j}}\left[ \lambda_{jt}\mu_{jt}\right]^2\right) \\
    &=  \overline{\lambda}_{-kt}\left(y_t  - \sum_{j=1}^J \overline{\lambda}_{jt}^{-1} \E_{q_{j}}\left[\lambda_{jt} \mu_{jt}\right] - \sum_{\ell=1}^L\overline{\mu}_{\ell t}\right)^2 \\
    &\quad + \overline{\lambda}_{-kt}\left[\sum_{\ell=1}^L \left(\overline{\mu^2_{\ell t}} - \overline{\mu}^2_{\ell t} \right) + \sum_{j=1}^J\left(\overline{\lambda}_{jt}^{-1}\E_{q_{j}}\left[ \lambda_{jt}\mu_{jt}^2\right] -\overline{\lambda}_{jt}^{-2}\E_{q_{j}}\left[ \lambda_{jt}\mu_{jt}\right]^2\right)\right] \\
    &= \overline{\lambda}_{-kt}(\tilde{r}^2_t + \delta_t). \tag{by (\ref{eq:mod-resid}) and (\ref{eq:delta})} 
\end{align*}
\normalsize
We now define variance corrected prior parameters for the rate of $s_{k}$ and location of the change-point:
\begin{align}
    \tilde{v}_{kt} &:= v_k + \frac{1}{2}\sum_{t'=t}^{T}\overline{\lambda}_{-kt'}\delta_{t'}, \label{eq:v-k-corrected} \\
    \tilde{\pi}_{kt} &:= \frac{\pi_{kt} \exp\left(-\frac{1}{2}\sum_{t'=1}^{t-1}\overline{\lambda}_{-kt'}\delta_{t'}\right)}{\sum_{t'=1}^T \pi_{kt'} \exp\left(-\frac{1}{2}\sum_{t''=1}^{t'-1}\overline{\lambda}_{-kt''}\delta_{t''}\right)}. \label{eq:pi-k-corrected}
\end{align}
Then, returning to the posterior parameters for $s_{k}$ when $\tau_{k}=t$, we now have:
\begin{align}
    \overline{u}_{kt} &= u_{k} + \frac{T-t+1}{2}, \label{eq:optimal-u-k} \\
    \overline{v}_{kt} &= v_{k} + \frac{1}{2}\sum_{t'=t}^{T}\overline{\lambda}_{-kt'}\tilde{r}^2_{t'} + \frac{1}{2}\sum_{t'=t}^{T}\overline{\lambda}_{-kt'}\delta_{t'} \label{eq:optimal-v-k}\\
    &= \tilde{v}_{kt} + \frac{1}{2}\sum_{t'=t}^{T}\overline{\lambda}_{-kt'}\tilde{r}^2_{t'}, \notag\\
    \overline{\pi}_{kt}  &= \frac{\exp\left(-\frac{1}{2}\sum_{t'=1}^{t-1}\overline{\lambda}_{-kt'}\delta_{t'}\right)\pi_{kt}\Gamma(\overline{u}_{kt})\overline{v}_{kt}^{-\overline{u}_{kt}}\exp\left( -\frac{1}{2} \sum_{t'=1}^{t-1} \overline{\lambda}_{-kt'}\tilde{r}_{t'}^2 \right)}{\sum_{t'=1}^T \exp\left(-\frac{1}{2}\sum_{t''=1}^{t'-1}\overline{\lambda}_{-kt''}\delta_{t''}\right) \pi_{kt'}\Gamma(\overline{u}_{kt'})\overline{v}_{kt'}^{-\overline{u}_{kt'}}\exp\left( -\frac{1}{2} \sum_{t''=1}^{t'-1} \overline{\lambda}_{-kt''}\tilde{r}_{t''}^2\right)} \label{eq:optimal-pi-k}\\
    &= \frac{\tilde{\pi}_{kt}\Gamma(\overline{u}_{kt})\overline{v}_{kt}^{-\overline{u}_{kt}}\exp\left( -\frac{1}{2} \sum_{t'=1}^{t-1} \overline{\lambda}_{-kt'}\tilde{r}_{t'}^2 \right)}{\sum_{t'=1}^T \tilde{\pi}_{kt'}\Gamma(\overline{u}_{kt'})\overline{v}_{kt'}^{-\overline{u}_{kt'}}\exp\left( -\frac{1}{2} \sum_{t''=1}^{t'-1} \overline{\lambda}_{-kt''}\tilde{r}_{t''}^2\right)}.\notag
\end{align}
Examining the var-scp posterior parameters (\ref{eq:var-scp-post-u})-(\ref{eq:var-scp-post-pi}), it is clear that:
\begin{align*}
    \overline{\boldsymbol{\theta}}_k = \normalfont{\texttt{var-scp}}(\tilde{\mathbf{r}}_{1:T} \:;\:\overline{\boldsymbol{\lambda}}_{-k,1:T}, u_k, \tilde{\mathbf{v}}_{k,1:T}, \tilde{\boldsymbol{\pi}}_{k,1:T}).  
\end{align*}
We now show that $\overline{\boldsymbol{\theta}}_k$ is the unique maximizer of $F(\overline{\boldsymbol{\theta}}_{1:N})$ holding $\overline{\boldsymbol{\theta}}_{1:N}\setminus\overline{\boldsymbol{\theta}}_k$ fixed. First note that we have:
\begin{align*}
     F(\overline{\boldsymbol{\theta}}_{1:N}) &\underset{\overline{\boldsymbol{\theta}}_k}{\propto} - \frac{1}{2} \sum_{t=1}^{T} \overline{\lambda}_t(\Tilde{r}_t^2 + \delta_t) + \frac{1}{2}\sum_{t=1}^T (T-t+1)\overline{\pi}_{kt}[\psi(\overline{u}_{kt}) - \log\overline{v}_{kt}]\\
     &\quad - \sum_{t=1}^T \overline{\pi}_{kt}\left[u_k \log \frac{\overline{v}_{kt}}{v_k} -\log \frac{\Gamma(\overline{u}_{kt})}{\Gamma(u_k)} + (\overline{u}_{kt} -u_k)\psi(\overline{u}_{kt}) -\frac{(\overline{v}_{kt} - v_k)\overline{u}_{kt}}{\overline{v}_{kt}} +  \log \frac{\overline{\pi}_{k t}}{\pi_{k t}}\right] \\
     &= - \frac{1}{2} \sum_{t=1}^{T} \left(\sum_{t'=1}^t  \frac{\overline{\pi}_{k t'}\overline{u}_{kt'}}{\overline{v}_{kt'}} + \sum_{t'=t+1}^T \overline{\pi}_{k t'}\right)\overline{\lambda}_{-kt}(\Tilde{r}_t^2 + \delta_t) + \frac{1}{2}\sum_{t=1}^T (T-t+1)\overline{\pi}_{kt}[\psi(\overline{u}_{kt}) - \log\overline{v}_{kt}]\\
     &\quad - \sum_{t=1}^T \overline{\pi}_{kt}\left[u_k \log \frac{\overline{v}_{kt}}{v_k} -\log \frac{\Gamma(\overline{u}_{kt})}{\Gamma(u_k)} + (\overline{u}_{kt} -u_k)\psi(\overline{u}_{kt}) -\frac{(\overline{v}_{kt} - v_k)\overline{u}_{kt}}{\overline{v}_{kt}} +  \log \frac{\overline{\pi}_{k t}}{\pi_{k t}}\right]. 
\end{align*}
First, we have:
\begin{align*}
    \frac{\partial F}{\partial \overline{v}_{kt}} = \frac{\overline{\pi}_{kt}\overline{u}_{kt}}{\overline{v}_{kt}^2} \left[v_k + \sum_{t'=t}^T \frac{\overline{\lambda}_{-kt'}(\tilde{r}^2_{t'} + \delta_{t'})}{2}\right] - \frac{\overline{\pi}_{kt}}{\overline{v}_{kt}}\left[u_k + \frac{(T-t+1)}{2}\right].
\end{align*}
Assume that $\overline{\pi}_{k t}> 0$ for the moment (we will momentarily show that $\pi_{k t} > 0$ implies that this is true for the optimal value of $\overline{\pi}_{k t}$). For any given $\overline{u}_{kt}$, define:
\begin{align*}
    \overline{v}_{kt}\circ(\overline{u}_{kt}) := \frac{\overline{u}_{kt}}{u_k + \frac{(T-t+1)}{2}}\left[v_k + \frac{1}{2}\sum_{t'=t}^T \overline{\lambda}_{-kt'}(\tilde{r}^2_{t'} + \delta_{t'})\right].
\end{align*}
Then:
\begin{align*}
\frac{\partial F}{\partial \overline{v}_{kt}}
    \begin{cases}
        > 0, &\text{if } \overline{v}_{kt} < \overline{v}_{kt}\circ(\overline{u}_{kt}), \\
        = 0, &\text{if } \overline{v}_{kt} = \overline{v}_{kt}\circ(\overline{u}_{kt}), \\
        < 0, &\text{if } \overline{v}_{kt} > \overline{v}_{kt}\circ(\overline{u}_{kt}).
    \end{cases}
\end{align*}
So $F$ is uniquely maximized by setting $\overline{v}_{kt} = \overline{v}_{kt}\circ(\overline{u}_{kt})$. Next we have:
\begin{align*}
    \frac{\partial F}{\partial \overline{u}_{k t}}\bigg|_{\overline{v}_{kt} = \overline{v}_{kt}\circ(\overline{u}_{kt})} &= \overline{\pi}_{kt}\psi'(\overline{u}_{kt})\left(u_k + \frac{T-t+1}{2} -\overline{u}_{kt} \right) - \frac{\overline{\pi}_{kt}}{\overline{u}_{kt}}\left(u_k + \frac{T-t+1}{2}\right) +\overline{\pi}_{kt}.
\end{align*}
We can use the following bounds for the trigamma function when $u>0$:
\begin{align*}
    \frac{1}{u} + \frac{1}{2u^2} \leq \psi'(u) \leq \frac{1}{u} + \frac{1}{u^2}
\end{align*}
to write:
\begin{align*}
     \frac{\overline{\pi}_{kt}}{2\overline{u}_{kt}^2}\left(u_k + \frac{T-t+1}{2} -\overline{u}_{kt} \right)\leq \frac{\partial F}{\partial \overline{u}_{k t}}\bigg|_{\overline{v}_{kt} = \overline{v}_{kt}\circ(\overline{u}_{kt})} \leq \frac{\overline{\pi}_{kt}}{\overline{u}_{kt}^2}\left(u_k + \frac{T-t+1}{2} -\overline{u}_{kt} \right).
\end{align*}
Therefore:
\begin{align*}
    \frac{\partial F}{\partial \overline{u}_{k t}}\bigg|_{\overline{v}_{kt} = \overline{v}_{kt}\circ(\overline{u}_{kt})}
    \begin{cases}
        > 0, &\text{if } \overline{u}_{kt} < u_k + \frac{T-t+1}{2}, \\
        = 0, &\text{if } \overline{u}_{kt} = u_k + \frac{T-t+1}{2}, \\
        < 0, &\text{if } \overline{u}_{kt} > u_k + \frac{T-t+1}{2}.
    \end{cases}
\end{align*}
Showing that $F$ is uniquely maximized when we set:
\begin{align*}
    \overline{u}_{kt} &= u_k + \frac{T-t+1}{2} \\
    \overline{v}_{kt} &= \overline{v}_{kt} \circ\left(u_k + \frac{T-t+1}{2}\right), \\
    &= v_k + \frac{1}{2}\sum_{t'=t}^T \overline{\lambda}_{-kt'}(\tilde{r}^2_{t'} + \delta_{t'}),
\end{align*}
which match (\ref{eq:optimal-u-k}) and (\ref{eq:optimal-v-k}). Lastly we can write:
\begin{align*}
    F(\overline{\boldsymbol{\theta}}_{1:N}) &\underset{\overline{\pi}_{kt}}{\propto} C_t \overline{\pi}_{kt} + \overline{\pi}_{kt} \log\overline{\pi}_{kt}, \\
    C_t &:= -\frac{1}{2}\left[\sum_{t'=1}^{t-1} \overline{\lambda}_{-kt'}(\tilde{r}_{t'}^2 + \delta_{t'}) + \frac{\overline{u}_{kt}}{\overline{v}_{kt}}\sum_{t'=t}^T  \overline{\lambda}_{-kt'}(\tilde{r}_{t'}^2 + \delta_{t'}) - (T-t+1)[\psi(\overline{u}_{kt}) - \log\overline{v}_{kt}]\right] \\
    &\quad- \left[[u_k \log \frac{\overline{v}_{kt}}{v_k} -\log \frac{\Gamma(\overline{u}_{kt})}{\Gamma(u_k)} + (\overline{u}_{kt} -u_k)\psi(\overline{u}_{kt}) -\frac{(\overline{v}_{kt} - v_k)\overline{u}_{kt}}{\overline{v}_{kt}} - \log \pi_{kt}\right].
\end{align*}
Note that $\pi_{k t} > 0$ and the expression above imply the optimal value of $\overline{\pi}_{\ell t}$ is positive for any values of $\overline{u}_{k t}$ and $\overline{v}_{k t}$, so we are at an interior solution and the optimal values of $\overline{u}_{k t}$ and $\overline{v}_{\ell t}$ must take the form (\ref{eq:optimal-u-k}) and (\ref{eq:optimal-v-k}), but for these values we have:
\begin{align*}
    C_t &:= -\frac{1}{2}\sum_{t'=1}^{t-1} \overline{\lambda}_{-kt'}(\tilde{r}_{t'}^2 + \delta_{t'}) - \overline{u}_{kt} \log\overline{v}_{kt} + +\log\Gamma(\overline{u}_{kt}) + \log \pi_{kt} + u_k \log v_k -\log\Gamma(u_k).
\end{align*}
which along with (\ref{eq:optimal-pi}) implies:
\begin{align*}
      \overline{\pi}_{k t} = \frac{\pi_{kt}\Gamma(\overline{u}_{kt})\overline{v}_{kt}^{-\overline{u}_{kt}}\exp\left[ -\frac{1}{2} \sum_{t'=1}^{t-1} \overline{\lambda}_{-kt'}(\tilde{r}_{t'}^2+\delta_{t'}) \right]}{\sum_{t'=1}^T  \pi_{kt'}\Gamma(\overline{u}_{kt'})\overline{v}_{kt'}^{-\overline{u}_{kt'}}\exp\left[ -\frac{1}{2} \sum_{t''=1}^{t'-1} \overline{\lambda}_{-kt''}(\tilde{r}_{t''}^2+\delta_{t''})\right]} ,
\end{align*}
which matches (\ref{eq:optimal-pi-k}).

\textbf{Case (iii):} For fixed $j' \in\{1,\ldots,J\}$, we treat the distributions $\{q_\ell(\boldsymbol{\theta}_\ell)\}_{\ell=1}^L$, $\{q_k(\boldsymbol{\theta}_k)\}_{k=1}^K$, and $\{q_j(\boldsymbol{\theta}_j)\}_{j\neq j'}$ as fixed and known. Then by (\ref{eq:q-soln}), the solution to:
\begin{align}
    \max_{q_{j'}} \text{\normalfont ELBO}\left(q\right), \label{eq:elbo-j}
\end{align}
is to pick $q_{j'}$ so that:
\begin{align*}
    \log q_{j'}(\boldsymbol{\theta}_{j'}) &\underset{\boldsymbol{\theta}_{j'}}{\propto} \E_{q_{-j'}} \left[\log p\left(\mathbf{y}_{1:T} \:|\: \boldsymbol{\Theta}\right)\right]  + \log p_{j'}(\boldsymbol{\theta}_{j'}) \\
    &\underset{\boldsymbol{\theta}_{j'}}{\propto}  \E_{q_{-j'}}\left[\frac{1}{2}\sum_{t=1}^{T} \log\lambda_t - \frac{1}{2}\sum_{t=1}^{T} \lambda_{t}\left(y_t - \mu_t\right)^2\right] + \log p_{j'}(\boldsymbol{\theta}_{j'}).
\end{align*}
Define $\lambda_{-j't} := \lambda_{j't}^{-1}\lambda_t$. If we fix $\tau_{j'} = t'$, then we have:
\begin{align*}
     \log q_{j'}(b_{j'}, s_{j'}, \tau_{j'} = t')&\underset{\boldsymbol{\theta}_{j'}}{\propto} \frac{T - t' + 1}{2} \log s_{j'}  -\frac{s_{j'}}{2} \sum_{t=t'}^{T} \E_{q_{-j'}}\left[\lambda_{-j't}(r_{-j't} - b_{j'})^2\right] \\
    &\quad -\frac{1}{2} \sum_{t=1}^{t'-1} \E_{q_{-j'}}\left[\lambda_{-j't}r_{-j't}^2\right] \\
    &\quad + \frac{1}{2}\log s_{j'} - \frac{s_{j'}\omega_{j'}b_{j'}^2}{2} + (u_{j'}-1)\log s_{j'} - v_{j'}s_{j'} + \log \pi_{j't'}.  
\end{align*}
Define:
\begin{align*}
    \overline{\omega}_{j't'} &=  \omega_{j'} +  \sum_{t=t'}^{T} \E_{q_{-j'}}\left[\lambda_{-j't}\right], \\
    \overline{b}_{j't'}  &= \overline{\omega}_{j't'}^{-1}  \sum_{t=t'}^{T} \E_{q_{-j'}}\left[\lambda_{-j't}r_{-j't}\right], \\
    \overline{u}_{j't'} &= u_{j'} + \frac{T-t'+1}{2}, \\
    \overline{v}_{j't'} &= v_{j'} - \frac{\overline{\omega}_{j't'}\overline{b}^2_{j't'}}{2} + \frac{1}{2}\sum_{t=t'}^{T}\E_{q_{-j'}}\left[\lambda_{-j't}r^2_{-j't}\right].
\end{align*}
Then we have:
\begin{align*}
    \log q_{j'}(b_{j'}, s_{j'}, \tau_{j'} = t')
    &\underset{\boldsymbol{\theta}_{j'}}{\propto} -\log\Gamma(\overline{u}_{j't'}) + \overline{u}_{j't'} \log \overline{v}_{j't'}  + \left(\overline{u}_{j't'} -1\right)\log s_{j'}  - \overline{v}_{j't'} s_{j'} \\
    &\quad + \frac{1}{2}\log (\overline{\omega}_{j't'}s_{j'}) - \frac{\overline{\omega}_{j't'} s_{j'}}{2} (b_{j'} - \overline{b}_{j't'} )^2 \\
    &\quad + \log \pi_{j't'} + \log\Gamma(\overline{u}_{j't'}) - \overline{u}_{j't'} \log \overline{v}_{j't'} - \frac{1}{2}\log \overline{\omega}_{j't'} -\frac{1}{2} \sum_{t=1}^{t'-1} \E_{q_{-j'}}\left[\lambda_{-j't}r_{-j't}^2\right]. 
\end{align*}
The first two lines above are just the log density of a Normal-Gamma distribution with parameters $\overline{b}_{j't'} $, $\overline{\omega}_{j't'}$, $\overline{u}_{j't'}$, and $\overline{v}_{j't'}$. Therefore, the $q_{j'}$ that solves (\ref{eq:elbo-j}) satisfies:
\begin{align*}
    b_{j'}, s_{j'} \:|\: \tau_{j'} = t' \sim \text{Normal-Gamma}(\overline{b}_{j't'} , \overline{\omega}_{j't'}, \overline{u}_{j't'}, \overline{v}_{j't'}),
\end{align*}
which matches (\ref{eq:bs-smscp}) of the meanvar-scp posterior. Integrating the joint density $q_{j'}(\boldsymbol{\theta}_{j'})$ with respect to $b_{j'}$ and $s_{j'}$ and calculating the normalizing constant, we are left with:
\begin{align*}
    q_{j'}(\tau_{j'} = t) &= \frac{\pi_{j't}\Gamma(\overline{u}_{j't})\overline{v}_{j't}^{-\overline{u}_{j't}}\overline{\omega}_{j't}^{-\frac{1}{2}}\exp\left( -\frac{1}{2} \sum_{t'=1}^{t-1} \E_{q_{-j'}}\left[\lambda_{-j't'}r_{-j't'}^2\right] \right)}{\sum_{t'=1}^T \pi_{j't'}\Gamma(\overline{u}_{j't'})\overline{v}_{j't'}^{-\overline{u}_{j't'}}\overline{\omega}_{j't'}^{-\frac{1}{2}}\exp\left( -\frac{1}{2} \sum_{t''=1}^{t'-1} \E_{q_{-j'}}\left[\lambda_{-j't''}r_{-j't''}^2\right] \right)} \\
    &:= \overline{\pi}_{j't}, 
\end{align*}
which matches (\ref{eq:gamma-post-cat1}). Therefore, the $q_{j'}$ that solves (\ref{eq:elbo-j}) has the same functional form as the posterior distribution of the meavar-scp model with parameters $\overline{\boldsymbol{\theta}}_{j'} := \{\overline{b}_{j't}, \overline{\omega}_{j't},\overline{u}_{j't},\overline{v}_{j't},\overline{\pi}_{j't}\}_{t=1}^T$. It remains to show $\overline{\boldsymbol{\theta}}_{j'}$ is the same set of parameters defined in (iii) of Proposition \ref{prop:vb}. We have:
\footnotesize
\begin{align*}
    \E_{q_{-j'}}\left[\lambda_{-j't}\right] &= \E_{q_{-j'}}\left[\prod_{j \neq j'}\lambda_{jt} \prod_{k=1}^K \lambda_{kt}\right] \\
    &= \prod_{j \neq j'}\E_{q_{j}}\left[\lambda_{jt}\right] \prod_{k=1}^K \E_{q_{k}}\left[\lambda_{kt}\right] \tag{by (\ref{eq:mean-field})} \\
    &= \overline{\lambda}_{-jt}, \\
\end{align*}
\begin{align*}
    \E_{q_{-j'}}\left[\lambda_{-j't}r_{-j't}\right] &= \E_{q_{-j'}}\left[\prod_{j \neq j'}\lambda_{jt} \prod_{k=1}^K \lambda_{kt}\left(y_t  - \sum_{j''\neq j'} \mu_{j''t} - \sum_{\ell=1}^L \mu_{\ell t}\right)\right] \\
    &= \prod_{j \neq j'}\E_{q_{j}}\left[\lambda_{jt}\right] \prod_{k=1}^K \E_{q_{k}}\left[\lambda_{kt}\right]\left(y_t  -\sum_{\ell=1}^L\E_{q_{\ell}}\left[\mu_{\ell t}\right]\right) \\
    &\quad-\prod_{k=1}^K \E_{q_{k}}\left[\lambda_{kt}\right]\E_{q_{-j'}}\left[\sum_{j'' \neq j'}\prod_{j \neq j'} \mu_{j''t}\lambda_{jt}\right]  \tag{by (\ref{eq:mean-field})} \\
    &=  \overline{\lambda}_{-j't} \left(y_t  -\sum_{\ell=1}^L\overline{\mu}_{\ell t}\right) \\
    &\quad-\prod_{k=1}^K\overline{\lambda}_{kt} \sum_{j'' \neq j'}\E_{q_{-j'}}\left[ \mu_{j''t}\lambda_{j''t}\prod_{j \not\in\{j',j''\}} \lambda_{jt}\right]  \\
    &= \overline{\lambda}_{-j't} \left(y_t  -\sum_{\ell=1}^L\overline{\mu}_{\ell t}\right) \\
    &\quad-\prod_{k=1}^K\overline{\lambda}_{kt} \sum_{j'' \neq j'}\E_{q_{j''}}\left[ \mu_{j''t}\lambda_{j''t} \right] \prod_{j \not\in\{j',j''\}} \E_{q_{j}}\left[\lambda_{jt}\right] \\
    &= \overline{\lambda}_{-j't} \left(y_t  -\sum_{\ell=1}^L\overline{\mu}_{\ell t}\right) \\
    &\quad-\prod_{k=1}^K\overline{\lambda}_{kt} \sum_{j \neq j'}\E_{q_{j}}\left[ \mu_{jt}\lambda_{jt} \right] \overline{\lambda}_{jt}^{-1}\overline{\lambda}_{-j't}\tag{multiply and divide by $\overline{\lambda}_{jt}$}\\
    &=\overline{\lambda}_{-j't} \left(y_t  -  \sum_{j \neq j'}\overline{\lambda}_{jt}^{-1}\E_{q_{j}}\left[ \mu_{jt}\lambda_{jt} \right] -\sum_{\ell=1}^L\overline{\mu}_{\ell t}\right), \\
\end{align*}
\begin{align*} 
    \E_{q_{j}}\left[ \mu_{jt}\lambda_{jt} \right] &= \E_{q_{j}}\left\{\E_{q_{j}}\left[\E_{q_{j}}\left[\mu_{j t}\:|\: s_j, \tau_j\right]\lambda_{j t}\:|\: \tau_j\right]\right\} \label{eq:lambda-beta} \\
    &= \sum_{t'=1}^T  \E_{q_{j}}\left[\E_{q_{j}}\left[\mu_{j t}\:|\: s_j, \tau_j = t'\right]\lambda_{j t}\:|\: \tau_j = t'\right]\overline{\pi}_{j t'} \\
    &= \sum_{t'=1}^t \overline{b}_{j t'}\E_{q_j}\left[s_j^{\mathbbm{1}\{t \geq t'\}}\:|\: \tau_j= t' \right] \overline{\pi}_{jt'} \tag{$\E_{q_j}\left[\mu_{j t}\:|\: s_j, \tau_j = t'\right] = \mathbbm{1}{\{t \geq t'\}} \overline{b}_{j t'}$} \\
    &= \sum_{t'=1}^t \overline{b}_{j t'}\overline{s}_{jt'} \overline{\pi}_{jt'} \\
    &= \sum_{t'=1}^t \frac{\overline{\pi}_{jt'}\overline{b}_{j t'}\overline{u}_{jt'}}{\overline{v}_{jt'}} 
\end{align*}
\normalsize
Define the partial mean residual $\tilde{r}_{-j t} := \tilde{r}_{t} + \E_{q_{j}}[\mu_{j t}]$, the partial scale residual $\overline{\lambda}_{-j t} := \overline{\lambda}_{jt}^{-1}\overline{\lambda}_t$, and the partial correction term: 
\begin{align}
    \delta_{-j't} := \delta_t - \overline{\lambda}_{j't}^{-1}\E_{q_{j'}}\left[ \lambda_{jt}\mu_{j't}^2\right] +\overline{\lambda}_{j't}^{-2}\E_{q_{j'}}\left[ \lambda_{j't}\mu_{j't}\right]^2. \label{eq:delta-j}
\end{align}
Then we have:
\begin{align*}
     \E_{q_{-j'}}\left[\lambda_{-j't}r_{-j't}\right] &=\overline{\lambda}_{-j't}\tilde{r}_{-j' t} 
\end{align*}
We also have:
\footnotesize
\begin{align*}
    \E_{q_{-j'}}\left[\lambda_{-j't}r^2_{-j't}\right] &=  \E_{q_{-j'}}\left[\prod_{j \neq j'}\lambda_{jt} \prod_{k=1}^K \lambda_{kt}\left(y_t  - \sum_{j''\neq j'} \mu_{j''t} - \sum_{\ell=1}^L \mu_{\ell t}\right)^2\right] \\
    &= \prod_{j \neq j'}\E_{q_{j}}\left[\lambda_{jt}\right] \prod_{k=1}^K \E_{q_{k}}\left[\lambda_{kt}\right]\E_{q_{-j'}}\left[\left(y_t  - \sum_{\ell=1}^L \mu_{\ell t}\right)^2\right] \\
    &\quad -2  \prod_{k=1}^K \E_{q_{k}}\left[\lambda_{kt}\right]\E_{q_{-j'}}\left[y_t   - \sum_{\ell=1}^L \mu_{\ell t}\right] \sum_{j''\neq j'} \E_{q_{j''}}\left[\lambda_{j''t} \mu_{j''t}\right]\prod_{j \not\in \{j',j''\}}\E_{q_{j}}[\lambda_{jt}]  \\
    &\quad +  \prod_{k=1}^K \E_{q_{k}}\left[ \lambda_{kt}\right]\E_{q_{-j'}}\left[\sum_{j^*\neq j'}\sum_{j''\neq j'}  \prod_{j \neq j'} \lambda_{jt} \mu_{j''t} \mu_{j^*t}\right] \\
    &= \overline{\lambda}_{-j't}\left[\left(y_t  - \sum_{\ell=1}^L \overline{\mu}_{\ell t}\right)^2 + \sum_{\ell=1}^L \left(\overline{\mu^2_{\ell t}} - \overline{\mu}^2_{\ell t} \right) \right] \\
    &\quad -2  \overline{\lambda}_{-j't}\left(y_t   - \sum_{\ell=1}^L\overline{\mu}_{\ell t}\right) \sum_{j\neq j'} \overline{\lambda}_{jt}^{-1} \E_{q_{j}}\left[\lambda_{jt} \mu_{jt}\right] \\
    &\quad + \overline{\lambda}_{-j't} \sum_{j\neq j'}\sum_{j''\neq j'} \overline{\lambda}_{j''t}^{-1}\E_{q_{j''}}\left[\lambda_{j''t} \mu_{j''t}\right]\overline{\lambda}_{jt}^{-1}\E_{q_{j}}\left[ \lambda_{jt}\mu_{jt}\right] \\
    &\quad + \overline{\lambda}_{-j't} \sum_{j\neq j'}\left(\overline{\lambda}_{jt}^{-1}\E_{q_{j}}\left[ \lambda_{jt}\mu_{jt}^2\right] -\overline{\lambda}_{jt}^{-2}\E_{q_{j}}\left[ \lambda_{jt}\mu_{jt}\right]^2\right) \\
    &=  \overline{\lambda}_{-j't}\left(y_t  - \sum_{j\neq j'} \overline{\lambda}_{j't}^{-1} \E_{q_{j}}\left[\lambda_{jt} \mu_{jt}\right] - \sum_{\ell=1}^L \overline{\mu}_{\ell t}\right)^2 \\
    &\quad + \overline{\lambda}_{-j't}\left[\sum_{\ell=1}^L \left(\overline{\mu^2_{\ell t}} - \overline{\mu}^2_{\ell t} \right) + \sum_{j\neq j'}\left(\overline{\lambda}_{jt}^{-1}\E_{q_{j}}\left[ \lambda_{jt}\mu_{jt}^2\right] -\overline{\lambda}_{jt}^{-2}\E_{q_{j}}\left[ \lambda_{jt}\mu_{jt}\right]^2\right)\right] \\
    &= \overline{\lambda}_{-j't}\left[\tilde{r}^2_{-j' t} +\sum_{\ell=1}^L \Var_{q_\ell} \left(\mu_{\ell t} \right) + \sum_{j\neq j'}\left(\overline{\lambda}_{jt}^{-1}\E_{q_{j}}\left[ \lambda_{jt}\mu_{jt}^2\right] -\overline{\lambda}_{jt}^{-2}\E_{q_{j}}\left[ \lambda_{jt}\mu_{jt}\right]^2\right)\right]\\
    &= \overline{\lambda}_{-j't}\left[\tilde{r}^2_{-j' t} + \delta_{-j't}\right]
\end{align*}
\normalsize
Finally, we have:
\small
\begin{align*}
    \E_{q_{j}}\left[\lambda_{jt}\mu^2_{jt}\right] &= \E_{q_{j}}\left\{\E_{q_{j}}\left[\E_{q_{j}}\left[\mu^2_{jt}\:|\: s_{j}, \tau_{j}\right]\lambda_{jt}\:|\: \tau_{j}\right]\right\} \\
    &= \sum_{t'=1}^T  \E_{q_{j}}\left[\E_{q_{j}}\left[\mu^2_{jt}\:|\: s_{j}, \tau_{j} = t'\right]\lambda_{jt}\:|\: \tau_{j}= t'\right] \overline{\pi}_{jt'} \\
    &= \sum_{t'=1}^t \E_{q_{j}}\left[(\overline{b}_{j t'}^2 + (\overline{\omega}_{j t'}s_{j})^{-1})s_j^{\mathbbm{1}\{t \geq t'\}}\:|\: \tau_j= t'\right] \overline{\pi}_{j t'} \tag{$\E_{q_{j}}[\mu^2_{jt}\:|\: s_j, \tau_j = t'] = \mathbbm{1}{\{t \geq t'\}} (\overline{b}_{jt'}^2 + (\overline{\omega}_{jt'}s_{j})^{-1})$} \\
    &= \sum_{t'=1}^t \left(\overline{b}^2_{jt'}\overline{s}_{jt'} +\overline{\omega}_{jt'}^{-1}\right) \overline{\pi}_{jt'} \\
    &= \sum_{t'=1}^t \overline{\pi}_{jt'}\left(\frac{\overline{b}^2_{jt'}\overline{u}_{jt'}}{\overline{v}_{jt'}} +\frac{1}{\overline{\omega}_{jt'}}\right) .
\end{align*}
\normalsize
We now define variance corrected prior parameters for the rate of $s_{j}$ and location of the change-point:
\begin{align}
    \tilde{v}_{jt} &:= v_j + \frac{1}{2}\sum_{t'=t}^{T}\overline{\lambda}_{-jt'}\delta_{-jt'}, \label{eq:v-j-corrected} \\
    \tilde{\pi}_{jt} &:= \frac{\pi_{jt} \exp\left(-\frac{1}{2}\sum_{t'=1}^{t-1}\overline{\lambda}_{-jt'}\delta_{-jt'}\right)}{\sum_{t'=1}^T \pi_{jt'} \exp\left(-\frac{1}{2}\sum_{t''=1}^{t'-1}\overline{\lambda}_{-jt''}\delta_{-jt''}\right)}. \label{eq:pi-j-corrected}
\end{align}
Then, returning to the posterior parameters for $b_{j}$ and $s_{j}$ when $\tau_{j}=t$, we now have:
\begin{align}
    \overline{\omega}_{jt} &=  \omega_{j} +  \sum_{t'=t}^{T} \overline{\lambda}_{-jt'}, \label{eq:optimal-omega-j}\\
    \overline{b}_{jt}  &= \overline{\omega}_{jt}^{-1}  \sum_{t'=t}^{T} \overline{\lambda}_{-jt'}\tilde{r}_{-j t'}, \label{eq:optimal-b-j}\\
    \overline{u}_{jt} &= u_{j} + \frac{T-t+1}{2}, \label{eq:optimal-u-j}\\
    \overline{v}_{jt} &= v_{j} - \frac{\overline{\omega}_{jt}\overline{b}^2_{jt}}{2} + \frac{1}{2}\sum_{t'=t}^{T}\overline{\lambda}_{-jt'}\tilde{r}^2_{jt'} + \frac{1}{2}\sum_{t'=t}^{T}\overline{\lambda}_{-jt'}\delta_{-jt'} \label{eq:optimal-v-j}\\
    &= \tilde{v}_{jt} - \frac{\overline{\omega}_{jt}\overline{b}^2_{jt}}{2} + \frac{1}{2}\sum_{t'=t}^{T}\overline{\lambda}_{-jt'}\tilde{r}^2_{jt'} + \frac{1}{2}\sum_{t'=t}^{T}\overline{\lambda}_{-jt'}\delta_{-jt'}, \notag \\
    \overline{\pi}_{jt}  &= \frac{\exp\left(-\frac{1}{2}\sum_{t'=1}^{t-1}\overline{\lambda}_{-jt'}\delta_{-jt'}\right)\pi_{jt}\Gamma(\overline{u}_{jt})\overline{v}_{jt}^{-\overline{u}_{jt}}\overline{\omega}_{jt}^{-\frac{1}{2}}\exp\left( -\frac{1}{2} \sum_{t'=1}^{t-1} \overline{\lambda}_{-jt'}\tilde{r}_{-jt'}^2 \right)}{\sum_{t'=1}^T \exp\left(-\frac{1}{2}\sum_{t''=1}^{t'-1}\overline{\lambda}_{-jt''}\delta_{-j't''}\right) \pi_{jt'}\Gamma(\overline{u}_{jt'})\overline{v}_{jt'}^{-\overline{u}_{jt'}}\overline{\omega}_{jt'}^{-\frac{1}{2}}\exp\left( -\frac{1}{2} \sum_{t''=1}^{t'-1} \overline{\lambda}_{-jt''}\tilde{r}_{-jt''}^2\right)} \label{eq:optimal-pi-j} \\
    &= \frac{\tilde{\pi}_{jt}\Gamma(\overline{u}_{jt})\overline{v}_{jt}^{-\overline{u}_{jt}}\overline{\omega}_{jt}^{-\frac{1}{2}}\exp\left( -\frac{1}{2} \sum_{t'=1}^{t-1} \overline{\lambda}_{-jt'}\tilde{r}_{-jt'}^2 \right)}{\sum_{t'=1}^T \tilde{\pi}_{jt'}\Gamma(\overline{u}_{jt'})\overline{v}_{jt'}^{-\overline{u}_{jt'}}\overline{\omega}_{jt'}^{-\frac{1}{2}}\exp\left( -\frac{1}{2} \sum_{t''=1}^{t'-1} \overline{\lambda}_{-jt''}\tilde{r}_{-jt''}^2\right)}. \notag
\end{align}
Examining the meanvar-scp posterior parameters (\ref{eq:meanvar-scp-post-omega})-(\ref{eq:meanvar-scp-post-pi}), it is clear that:
\begin{align*}
    \overline{\boldsymbol{\theta}}_j := \normalfont{\texttt{meanvar-scp}}(\tilde{\mathbf{r}}_{-j,1:T} \:;\:\overline{\boldsymbol{\lambda}}_{-j,1:T}, \omega_j, u_j, \tilde{\mathbf{v}}_{j,1:T}, \tilde{\boldsymbol{\pi}}_{j,1:T}).
\end{align*}
We now show that $\overline{\boldsymbol{\theta}}_j$ is the unique maximizer of $F(\overline{\boldsymbol{\theta}}_{1:N})$ holding $\overline{\boldsymbol{\theta}}_{1:N}\setminus\overline{\boldsymbol{\theta}}_j$ fixed. First note that we have:
\begin{align*}
     F(\overline{\boldsymbol{\theta}}_{1:N}) &\underset{\overline{\boldsymbol{\theta}}_j}{\propto} - \frac{1}{2} \sum_{t=1}^{T} \overline{\lambda}_t(\Tilde{r}_t^2 + \delta_t) + \frac{1}{2}\sum_{t=1}^T (T-t+1)\overline{\pi}_{jt}[\psi(\overline{u}_{jt}) - \log\overline{v}_{jt}]\\
     &\quad - \sum_{t=1}^T\overline{\pi}_{jt}\left[\frac{1}{2}\log\frac{\overline{\omega}_{j t}}{\omega_j}- \frac{1}{2} + \frac{\omega_j}{2\overline{\omega}_{jt}} + \frac{\omega_j\overline{u}_{jt}\overline{b}_{jt}^2}{2\overline{v}_{jt}}\right]\\
     &\quad - \sum_{t=1}^T\overline{\pi}_{jt}\left[u_j \log \frac{\overline{v}_{jt}}{v_j} -\log \frac{\Gamma(\overline{u}_{jt})}{\Gamma(u_j)} + (\overline{u}_{jt} -u_j)\psi(\overline{u}_{jt}) -\frac{(\overline{v}_{jt} - v_j)\overline{u}_{jt}}{\overline{v}_{jt}} +  \log \frac{\overline{\pi}_{j t}}{\pi_{j t}}\right].
\end{align*}
Note that: 
\small
\begin{align*}
    \overline{\lambda}_t(\Tilde{r}_t^2 + \delta_t) &= \overline{\lambda}_{jt}\overline{\lambda}_{-jt}\left[\left(\Tilde{r}_{-jt} - \frac{\E_{q_j}[\lambda_{jt}\mu_{jt}]}{\overline{\lambda}_{jt}}\right)^2 + \delta_{-jt} + \frac{\E_{q_j}[\lambda_{jt} \mu^2_{jt}]}{\overline{\lambda}_{jt}} -\frac{\E_{q_j}[\lambda_{jt} \mu_{jt}]^2}{\overline{\lambda}_{jt}^2}\right] \\
    &= \overline{\lambda}_{jt}\overline{\lambda}_{-jt}(\Tilde{r}_{-jt}^2 + \delta_{-jt}) + \overline{\lambda}_{-jt}(\E_{q_j}[\lambda_{jt} \mu^2_{jt}] - 2 \Tilde{r}_{-jt} \E_{q_j}[\lambda_{jt}\mu_{jt}]) \\
    &= \left(\sum_{t'=1}^t  \frac{\overline{\pi}_{j t'}\overline{u}_{jt'}}{\overline{v}_{jt'}} + \sum_{t'=t+1}^T \overline{\pi}_{j t'}\right)\overline{\lambda}_{-jt}(\Tilde{r}_{-jt}^2 + \delta_{-jt}) \\
    &\quad + \overline{\lambda}_{-jt}\left(\sum_{t'=1}^t \overline{\pi}_{jt'}\left(\frac{\overline{b}^2_{jt'}\overline{u}_{jt'}}{\overline{v}_{jt'}} +\frac{1}{\overline{\omega}_{jt'}}\right) - 2 \Tilde{r}_{-jt} \sum_{t'=1}^t \frac{\overline{\pi}_{jt'}\overline{b}_{j t'}\overline{u}_{jt'}}{\overline{v}_{jt'}} \right).
\end{align*}
\normalsize
Thus:
\begin{align*}
    F(\overline{\boldsymbol{\theta}}_{1:N}) &\underset{\overline{\boldsymbol{\theta}}_j}{\propto} - \frac{1}{2} \left(\sum_{t'=1}^t  \frac{\overline{\pi}_{j t'}\overline{u}_{jt'}}{\overline{v}_{jt'}} + \sum_{t'=t+1}^T \overline{\pi}_{j t'}\right)\overline{\lambda}_{-jt}(\Tilde{r}_{-jt}^2 + \delta_{-jt}) \\
    &\quad- \frac{1}{2} \sum_{t=1}^{T} \overline{\lambda}_{-jt}\left(\sum_{t'=1}^t \overline{\pi}_{jt'}\left(\frac{\overline{b}^2_{jt'}\overline{u}_{jt'}}{\overline{v}_{jt'}} +\frac{1}{\overline{\omega}_{jt'}}\right) - 2 \Tilde{r}_{-jt} \sum_{t'=1}^t \frac{\overline{\pi}_{jt'}\overline{b}_{j t'}\overline{u}_{jt'}}{\overline{v}_{jt'}} \right) \\
    &\quad+ \frac{1}{2}\sum_{t=1}^T (T-t+1)\overline{\pi}_{jt}[\psi(\overline{u}_{jt}) - \log\overline{v}_{jt}]\\
    &\quad - \sum_{t=1}^T\overline{\pi}_{jt}\left[\frac{1}{2}\log\frac{\overline{\omega}_{j t}}{\omega_j}- \frac{1}{2} + \frac{\omega_j}{2\overline{\omega}_{jt}} + \frac{\omega_j\overline{u}_{jt}\overline{b}_{jt}^2}{2\overline{v}_{jt}}\right]\\
    &\quad - \sum_{t=1}^T\overline{\pi}_{jt}\left[u_j \log \frac{\overline{v}_{jt}}{v_j} -\log \frac{\Gamma(\overline{u}_{jt})}{\Gamma(u_j)} + (\overline{u}_{jt} -u_j)\psi(\overline{u}_{jt}) -\frac{(\overline{v}_{jt} - v_j)\overline{u}_{jt}}{\overline{v}_{jt}} +  \log \frac{\overline{\pi}_{j t}}{\pi_{j t}}\right]. 
\end{align*}
Assume that $\overline{\pi}_{jt} > 0$ for the moment (we will momentarily show that $\pi_{j t} > 0$ implies that this is true for the optimal value of $\overline{\pi}_{j t}$), then the expression above is strictly concave as a function of $\overline{b}_{j t}$ and:
\begin{align*}
    \frac{\partial F}{\partial \overline{b}_{j t}} = \frac{\overline{\pi}_{j t}\overline{u}_{jt}}{\overline{v}_{jt}}\sum_{t'=t}^T \overline{\lambda}_{-jt'}\tilde{r}_{-j t'} -  \left(\omega_j + \sum_{t'=t}^T\overline{\lambda}_{-jt'}\right)\frac{\overline{\pi}_{j t} \overline{u}_{j t}\overline{b}_{j t}}{\overline{v}_{j t}},
\end{align*}
Since $\overline{\pi}_{j t}, \overline{u}_{j t}, \overline{v}_{j t}  > 0$, $F$ is uniquely maximized by setting:
\begin{align*}
    \overline{b}_{jt} &= \frac{\sum_{t'=t}^T \overline{\lambda}_{-jt'}\tilde{r}_{-j t'}}{\omega_j + \sum_{t'=t}^T\overline{\lambda}_{-jt'}}, 
\end{align*}
which matches (\ref{eq:optimal-b-j}). Next:
\begin{align*}
    \frac{\partial F}{\partial \overline{\omega}_{j t}} = \frac{\overline{\pi}_{jt}\left(\omega_j + \sum_{t'=t}^T\overline{\lambda}_{-jt}\right)} {2\overline{\omega}_{j t}^2} - \frac{\overline{\pi}_{j t}} {2\overline{\omega}_{jt}} 
    \begin{cases}
        >0, &\text{if } \overline{\omega}_{j t} < \omega_j + \sum_{t'=t}^T\overline{\lambda}_{-jt'}, \\
        =0, &\text{if } \overline{\omega}_{j t} = \omega_j + \sum_{t'=t}^T\overline{\lambda}_{-jt'}, \\
        <0, &\text{if } \overline{\omega}_{j t} > \omega_j + \sum_{t'=t}^T\overline{\lambda}_{-jt'}.
    \end{cases}
\end{align*}
So when $\overline{\pi}_{jt}> 0$, $F$ is uniquely maximized by setting:
\begin{align*}
    \overline{\omega}_{j t} = \omega_j + \sum_{t'=t}^T\overline{\lambda}_{-jt'}, 
\end{align*}
which matches (\ref{eq:optimal-omega-j}). Next, we have:
\small
\begin{align*}
    \frac{\partial F}{\partial \overline{v}_{jt}} = \frac{\overline{\pi}_{jt}\overline{u}_{jt}}{\overline{v}_{jt}^2} \left[v_j + \frac{\sum_{t'=t}^T \overline{\lambda}_{-jt'}(\tilde{r}^2_{-jt'} + \delta_{-jt'})}{2} + \frac{\overline{b}_{jt}^2\left(\omega_j + \sum_{t'=t}^T \overline{\lambda}_{-jt'}\right)}{2} - \overline{b}^2_{jt} \sum_{t'=t}^T\overline{\lambda}_{-jt'}\tilde{r}_{-jt'}\right] - \frac{\overline{\pi}_{jt}}{\overline{v}_{jt}}\left[u_j + \frac{(T-t+1)}{2}\right].
\end{align*}
\normalsize
For any given $\overline{b}_{jt},\overline{u}_{jt} $, define:
\small
\begin{align*}
    \overline{v}_{jt}\circ(\overline{u}_{jt}, \overline{b}_{jt}) := \frac{\overline{u}_{jt}}{u_j + \frac{(T-t+1)}{2}}\left[v_j + \frac{\sum_{t'=t}^T \overline{\lambda}_{-jt'}(\tilde{r}^2_{-jt'} + \delta_{-jt'})}{2} + \frac{\overline{b}_{jt}^2\left(\omega_j + \sum_{t'=t}^T \overline{\lambda}_{-jt'}\right)}{2}- \overline{b}^2_{jt} \sum_{t'=t}^T\overline{\lambda}_{-jt'}\tilde{r}_{-jt'}\right].
\end{align*}
\normalsize
Then:
\begin{align*}
\frac{\partial F}{\partial \overline{v}_{jt}}
    \begin{cases}
        > 0, &\text{if } \overline{v}_{jt} < \overline{v}_{jt}\circ(\overline{u}_{jt}, \overline{b}_{jt}), \\
        = 0, &\text{if } \overline{v}_{jt} = \overline{v}_{jt}\circ(\overline{u}_{jt}, \overline{b}_{jt}), \\
        < 0, &\text{if } \overline{v}_{jt} > \overline{v}_{jt}\circ(\overline{u}_{jt}, \overline{b}_{jt}).
    \end{cases}
\end{align*}
So $F$ is uniquely maximized by setting $\overline{v}_{jt} = \overline{v}_{jt}\circ(\overline{u}_{jt}, \overline{b}_{jt})$. Next we have:
\begin{align*}
    \frac{\partial F}{\partial \overline{u}_{j t}}\bigg|_{\overline{v}_{jt} = \overline{v}_{jt}\circ(\overline{b}_{jt},\overline{u}_{jt})} &= \overline{\pi}_{kt}\psi'(\overline{u}_{kt})\left(u_k + \frac{T-t+1}{2} -\overline{u}_{kt} \right) - \frac{\overline{\pi}_{kt}}{\overline{u}_{kt}}\left(u_k + \frac{T-t+1}{2}\right) +\overline{\pi}_{kt}.
\end{align*}
Therefore an identical argument as in part (ii) shows that $F$ is uniquely maximized when we set:
\begin{align*}
    \overline{u}_{jt} &= u_j + \frac{T-t+1}{2} 
\end{align*}
which matches (\ref{eq:optimal-u-j}). Plugging in the optimal values of $\overline{u}_{jt}$ into $\overline{v}_{jt}\circ(\overline{u}_{jt}, \overline{b}_{jt})$, we get:
\begin{align*}
    \overline{v}_{jt} = v_j + \frac{\sum_{t'=t}^T \overline{\lambda}_{-jt'}(\tilde{r}^2_{-jt'} + \delta_{-jt'})}{2} + \frac{\overline{b}_{jt}^2\left(\omega_j + \sum_{t'=t}^T \overline{\lambda}_{-jt'}\right)}{2}- \overline{b}^2_{jt} \sum_{t'=t}^T\overline{\lambda}_{-jt'}\tilde{r}_{-jt'}, 
\end{align*}
which is equivalent to:
\begin{align*}
    \overline{v}_{jt} = v_j + \frac{\sum_{t'=t}^T \overline{\lambda}_{-jt'}(\tilde{r}^2_{-jt'} + \delta_{-jt'})}{2} - \frac{\overline{\omega}_{jt}\overline{b}_{jt}^2}{2}, 
\end{align*}
at the optimal value of $\overline{\omega}_{jt}$, and which matches (\ref{eq:optimal-v-j}) at the optimal value of $\overline{b}_{jt}$. Finally, an identical argument as in parts (i) and (ii) gives that the $\overline{\pi}_{jt}$ that uniquely maximizes $F$ matches the value (\ref{eq:optimal-pi-j}), which is positive so long as $\pi_{jt}>0$.
\end{proof}
\subsection{Estimating \texorpdfstring{$\mu_0$}{mu0} and \texorpdfstring{$\lambda_0$}{lambda0} with Empirical Bayes}
\label{app:empirical-bayes}

In most cases, $\E[y_1]\neq 0$ and $\Var(y_1)\neq 1$; however, $\mathbf{y}_{1:T}$ is nonstationary so standardizing the data is not straightforward. In this case, it is desirable to include an intercept $\mu_0$ and initial scale $\lambda_0$ term in (\ref{eq:mu_t}) and (\ref{eq:lambda_t}):
\begin{align*}
    \mu_{t} &:= \mu_0 + \sum_{i = 1}^{J^*+L^*} \mu_{i t} := \sum_{j = 1}^{J^*} b_j\mathbbm{1}\{t\geq\tau_j\} + \sum_{\ell = J^* + 1}^{J^*+L^*} b_\ell\mathbbm{1}\{t\geq\tau_\ell\}, \\
    \lambda_t &:= \lambda_0\prod_{i=1}^{J^*+K^*} \lambda_{it} := \prod_{j=1}^{J^*} 
 s_{j}^{\mathbbm{1}\{t\geq\tau_{j}\}}\prod_{k=J^*+L^*+1}^{N} 
 s_{k}^{\mathbbm{1}\{t\geq\tau_{k}\}}.
\end{align*}
We treat $\mu_0$ and $\lambda_0$ as unknown nuisance parameters that must be estimated. One approach is to just set $\mu_0 = 0$ and $\lambda_0 = 1$ in Algorithm \ref{alg:mich}. Then, if $\mu_0 \neq 0$ or $\lambda_0 \neq 1$, Algorithm \ref{alg:mich} will attempt to detect a change-point at $t=1$ and center $\mu_{t}$ and $\lambda_t$ around $\mu_0$ and $\lambda_0$ respectively for all $t$ up until the first real change-point. The problem with fitting $\mu_0$ and $\lambda_0$ this way is that it requires MICH to search for a change-point whose location is known. Searching for an additional change-point requires MICH to estimate  $\mathcal{O}(T)$ additional parameters, when we only actually needed to fit two parameters $\mu_0$ and $\lambda_0$. For example, in the case of $\mu_0 \neq 0$ and $\lambda_0 \neq 1$, we end up having to estimate $5T$ additional parameters $\{\overline{b}_t, \overline{\tau}_t, \overline{u}_t, \overline{v}_t, \overline{\pi}_t\}_{t=1}^T$. Another problem with this approach is that we risk missing a real change-point, e.g. if the are $J^*>0$ true changes and we set $J=J^*$ to fit MICH, we will end up using one component of $J$ to fit $\mu_0$ and $\lambda_0$ leaving only $J^*-1$ components to fit the remaining $J^*$ true change-points. 

To circumvent these challenges, we adopt an empirical Bayes (EB) approach for estimating $\mu_0$ and $\lambda_0$. We seek the values $\mu_0$ and $\lambda_0$ that maximize the evidence $\log p(\mathbf{y}_{1:T}\:;\mu_0,\lambda_0)$. Since the evidence is now a function of $\mu_0$ and $\lambda_0$, so is the ELBO:
\begin{align}
    \text{ELBO}(q\:;\mu_0,\lambda_0) &:= \int q(\boldsymbol{\Theta}) \log \frac{ p(\mathbf{y}_{1:T},\boldsymbol{\Theta}\:;\mu_0,\lambda_0)}{q(\boldsymbol{\Theta})} \; d\boldsymbol{\Theta} \\
    &\;= \log p(\mathbf{y}_{1:T}\:;\mu_0,\lambda_0) - \text{KL}( q \:\lVert\: p)
\end{align}
Note that in the expression above, $\text{ELBO}(q;\mu_0,\lambda_0)$ is equivalent to $\log p(\mathbf{y}_{1:T};\mu_0,\lambda_0)$ plus a term that does not depend on $\mu_0$ or $\lambda_0$. Therefore, for fixed $q$, the $\mu_0$ and $\lambda_0$ that maximize the evidence are equivalent to: 
\begin{align}
    \argmax{\{\mu_0,\lambda_0\}\in\mathbb{R}\times\mathbb{R}_{++}} \; \text{ELBO}(q\:;\mu_0,\lambda_0),
\end{align}
which in turn is equivalent to:
\begin{align}\label{eq:EB-max}
    \argmax{\{\mu_0,\lambda_0\}\in\mathbb{R}\times\mathbb{R}_{++}} \; \E_{q} \left[\log p\left(\mathbf{y}_{1:T} \:|\: \boldsymbol{\Theta};\mu_0,\lambda_0\right)\right]. 
\end{align}
\cite{Wang20} use this same approach, and as they note, optimizing the ELBO over $\mu_0$ and $\lambda_0$ constitutes the M-step of an EM algorithm where the E-step is approximate (\citealp{Dempster77, Heskes03, Neal98}) do to the use of $q$. Define $\tilde{r}_t$, $\overline{\lambda}_t$ and $\delta_t$ as in (\ref{eq:mod-resid})-(\ref{eq:delta}) along with: 
\begin{align}
    \tilde{r}_{-0t} &:= \tilde{r}_t + \mu_0 \\
    \overline{\lambda}_{-0t} &:= \lambda_0^{-1}\overline{\lambda}_t
\end{align}
then (\ref{eq:EB-max}) can be rewritten as:
\begin{align}\label{eq:EB-max-simple}
    \argmax{\{\mu_0,\lambda_0\}\in\mathbb{R}\times\mathbb{R}_{++}} \; \frac{T }{2}\log\lambda_0 - \frac{\lambda_0}{2}\sum_{t=1}^{T} \overline{\lambda}_{-0t}\left[\left(\tilde{r}_{-0t}-\mu_0\right)^2 + \delta_t\right].
\end{align}
The solutions to (\ref{eq:EB-max-simple}) is given by:
\begin{align}
    \hat{\mu}_0 &= \frac{\sum_{t=1}^{T} \overline{\lambda}_{-0t}\tilde{r}_{-0t}}{\sum_{t=1}^{T} \overline{\lambda}_{-0t}} \label{eq:EB-max-solution-mu} \\
    \hat{\lambda}_0 &= \left[\frac{\sum_{t=1}^{T} \overline{\lambda}_{-0t}[\left(\tilde{r}_{-0t}-\hat{\mu}_0\right)^2 + \delta_t]}{T}\right]^{-1}. \label{eq:EB-max-solution-lambda}
\end{align}
We incorporate the estimation step for $\mu_0$ and $\lambda_0$ into Algorithm \ref{alg:mich-eb}.

\begin{algorithm}[!h]
\label{alg:mich-eb}
\caption{Empirical Bayes Variational Approximation to MICH Posterior.}

\small
\SetAlgoLined
  Inputs: $\mathbf{y}_{1:T},\:L,\:K,\:J,\:\omega_0,\:u_0,v_0,\boldsymbol{\pi}_{1:T}$\:; \\
  Initialize: $\mu_0,\:\lambda_0,\:\overline{\boldsymbol{\theta}}_{1:N}$\:;
  
  \Repeat {Convergence} {
  \For{$j=1$ \KwTo $J$} {
      $\tilde{r}_{-jt} := y_t -\mu_0- \sum_{j' \neq j} \E_{q_{j'}}[\lambda_{j't} \mu_{j't}] / \E_{q_{j'}}[\lambda_{j't}]- \sum_{\ell =J+ 1}^{J+L} \E_{q_\ell}[\mu_{\ell t}]$ \tcp*{j\textsuperscript{th} partial mean residual} 
      $\overline{\lambda}_{-jt} := \lambda_0\prod_{j' \neq j}\E_{q_{j'}}[\lambda_{j't}] \prod_{k=J+L+1}^N \E_{q_k}[\lambda_{kt}]$ \tcp*{j\textsuperscript{th} partial scale residual}
      %%$\delta_{-jt} := \sum_{\ell=1}^L  \Var(\mu_{\ell t} ) + \sum_{j' \neq j}[\E[\lambda_{j't} \mu^2_{j't}] / \E[\lambda_{j't}] - (\E[\lambda_{j't} \mu_{j't}] / \E[\lambda_{j't}])^2 ]$ \tcp*{j\textsuperscript{th} variance correction term} 
      Compute $\delta_{-j t}$, $\tilde{v}_{jt}$, and $\tilde{\pi}_{jt}$ by (\ref{eq:delta-j}), (\ref{eq:v-j-corrected}), and (\ref{eq:pi-j-corrected}) \tcp*{variance corrected priors} 
      $\overline{\boldsymbol{\theta}}_j := \texttt{meanvar-scp}(\tilde{\mathbf{r}}_{-j,1:T} \:;\:\overline{\boldsymbol{\lambda}}_{-j,1:T}, \omega_0, u_0, \tilde{\mathbf{v}}_{j,1:T}, \tilde{\boldsymbol{\boldsymbol{\pi}}}_{j,1:T})$ \tcp*{update meanvar-scp parameters}
    }
    \For{$\ell=J+1$ \KwTo $J+L$} {
      $\tilde{r}_{-\ell t} := y_t -\mu_0- \sum_{j =1}^J \E_{q_j}[\lambda_{jt} \mu_{jt}] / \E_{q_j}[\lambda_{jt}]- \sum_{\ell' \neq \ell} \E_{q_{\ell'}}[\mu_{\ell' t}]$ \tcp*{l\textsuperscript{th} partial mean residual}
      $\overline{\lambda}_{t} := \lambda_0\prod_{j=1}^J\E_{q_j}[\lambda_{jt}]\prod_{k=J+L+1}^N \E_{q_k}[\lambda_{kt}]$ \tcp*{precision of residual} 
      $\overline{\boldsymbol{\theta}}_\ell := \texttt{mean-scp}(\tilde{\mathbf{r}}_{-\ell,1:T} \:;\: \overline{\boldsymbol{\lambda}}_{1:T}, \omega_0, \boldsymbol{\boldsymbol{\pi}}_{1:T})$ \tcp*{update mean-scp parameters}
    }
    \For{$k=J+L+1$ \KwTo $N$} {
      $\tilde{r}_{t} := y_t  -\mu_0- \sum_{j=1}^J \E_{q_j}[\lambda_{jt} \mu_{jt}] / \E_{q_j}[\lambda_{jt}]- \sum_{\ell = J+1}^{J+L} \E_{q_\ell}[\mu_{\ell t}] $ \tcp*{mean residual} 
      $\overline{\lambda}_{-kt} := \lambda_0\prod_{j=1}^J\E_{q_j}[\lambda_{jt}]\prod_{k' \neq k} \E_{q_{k'}}[\lambda_{k't}] $ \tcp*{k\textsuperscript{th} partial scale residual} 
      %% $\delta_{t} := \sum_{\ell=1}^L  \Var(\mu_{\ell t} ) + \sum_{j=1}^J[\E[\lambda_{jt} \mu^2_{jt}] / \E[\lambda_{jt}] - (\E[\lambda_{jt} \mu_{jt}] / \E[\lambda_{jt}])^2 ];$ \tcp*{variance correction term}\\
      Compute $\delta_t$, $\tilde{v}_{kt}$, and $\tilde{\pi}_{kt}$ by (\ref{eq:delta}), (\ref{eq:v-k-corrected}), and (\ref{eq:pi-k-corrected}) \tcp*{variance corrected priors} 
      $\overline{\boldsymbol{\theta}}_k := \texttt{var-scp}(\tilde{\mathbf{r}}_{1:T} 
      \:;\:\overline{\boldsymbol{\lambda}}_{-k,1:T}, u_0, \tilde{\mathbf{v}}_{k,1:T}, \tilde{\boldsymbol{\boldsymbol{\pi}}}_{k,1:T})$ \tcp*{update var-scp parameters}
    }
    $\tilde{r}_{-0t} := y_t - \sum_{j=1}^J \E_{q_j}[\lambda_{jt} \mu_{jt}] / \E_{q_j}[\lambda_{jt}]- \sum_{\ell = J+1}^{J+L} \E_{q_\ell}[\mu_{\ell t}] $\:; \\
    $\overline{\lambda}_{-0t} := \prod_{j=1}^J\E_{q_j}[\lambda_{jt}]\prod_{k=J+L+1}^N \E_{q_k}[\lambda_{kt}]\;$; \\
    Compute $\mu_0$ and $\lambda_0$ by (\ref{eq:EB-max-solution-mu}) and (\ref{eq:EB-max-solution-lambda}) \tcp*{EB step} 
  }
  \Return{Model Parameters: $\mu_0,\:\lambda_0,\:\overline{\boldsymbol{\theta}}_{1:N}$}.
\end{algorithm}

Note that Algorithm \ref{alg:mich-eb} embeds Algorithm \ref{alg:mich} with the additional step of maximizing the ELBO with respect to $\mu_0$ and $\lambda_0$ while holding the approximate posterior $q$ fixed. Therefore, Algorithm \ref{alg:mich-eb} also defines a coordinate ascent procedure. The objective function in (\ref{eq:EB-max-simple}) is continuously differentiable with respect to $\mu_0$ and $\lambda_0$, so if we can show that $\hat{\mu}_0$ and $\hat{\lambda}_0$ are the unique maximizers of (\ref{eq:EB-max-simple}), then the conditions of Proposition 2.7.1 of \cite{Bertsekas97} will still hold even with the added coordinate ascent step of maximizing the ELBO with respect to $\mu_0$ and $\lambda_0$. 

To prove the uniqueness of $\hat{\mu}_0$ and $\hat{\lambda}_0$, first note that the term:
\begin{align*}
    -\frac{1}{2}\sum_{t=1}^{T} \overline{\lambda}_{-0t}\left(\tilde{r}_{-0t}-\mu_0\right)^2
\end{align*}
is strictly concave as a function of $\mu_0$ and is maximized at $\hat{\mu}_0$. Therefore, for any $\mu_0\in\mathbb{R}\setminus\{\hat{\mu}_0\}$ we have:
\begin{align*}
    -\frac{1}{2}\sum_{t=1}^{T} \overline{\lambda}_{-0t}\left(\tilde{r}_{-0t}-\hat{\mu}_0\right)^2 > -\frac{1}{2}\sum_{t=1}^{T} \overline{\lambda}_{-0t}\left(\tilde{r}_{-0t}-\mu_0\right)^2
\end{align*}
and thus for any $\lambda_0>0$,
\begin{align*}
    \frac{T }{2}\log\lambda_0-\frac{\lambda_0}{2}\sum_{t=1}^{T} \overline{\lambda}_{-0t}\left(\tilde{r}_{-0t}-\hat{\mu}_0\right)^2 > \frac{T }{2}\log\lambda_0-\frac{\lambda_0}{2}\sum_{t=1}^{T} \overline{\lambda}_{-0t}\left(\tilde{r}_{-0t}-\mu_0\right)^2.
\end{align*}
Since the $\log$ is strictly concave, the left-hand side of the above inequality is uniquely maximized by $\hat{\lambda}_0$. So for any $\lambda_0 > 0$, we have 
\begin{align*}
    \frac{T }{2}\log\hat{\lambda}_0-\frac{\hat{\lambda}_0}{2}\sum_{t=1}^{T} \overline{\lambda}_{-0t}\left(\tilde{r}_{-0t}-\hat{\mu}_0\right)^2 &> \frac{T }{2}\log\lambda_0-\frac{\lambda_0}{2}\sum_{t=1}^{T} \overline{\lambda}_{-0t}\left(\tilde{r}_{-0t}-\hat{\mu}_0\right)^2 \\
    &> \frac{T }{2}\log\lambda_0-\frac{\lambda_0}{2}\sum_{t=1}^{T} \overline{\lambda}_{-0t}\left(\tilde{r}_{-0t}-\mu_0\right)^2.
\end{align*}
Proving that $\hat{\mu}_0$ and $\hat{\lambda}_0$ are the global maximizers of (\ref{eq:EB-max-simple}). Therefore, Proposition \ref{prop:coord-ascent} still holds with the included maximization step for $\mu_0$ and $\lambda_0$ and Algorithm \ref{alg:mich-eb} will also converge to a stationary point. 

\subsection{MICH for Multivariate Mean Changes}
\label{app:multi-mich}

In order to simplify the exposition of the MICH model in Section \ref{sec:mich}, we solely focused on detecting multiple mean and variance change-points for univariate data. However, when $d>1$, it is still possible for MICH to identify multiple breaks in the mean, assuming that: 
\begin{align}
    \mathbf{y}_{t} &\overset{\text{ind}.} {\sim}\mathcal{N}_d(\boldsymbol{\mu}_t,\boldsymbol{\Lambda}^{-1}_t), \label{eq:y-multi}\\
    \boldsymbol{\mu}_{t} &:= \boldsymbol{\mu}_0+ \sum_{\ell = 1}^{L^*} \boldsymbol{\mu}_{\ell t}, \label{eq:mu_t-multi} \\
    \boldsymbol{\mu}_{\ell t} &:= \mathbf{b}_\ell\mathbbm{1}\{t \geq\tau_\ell\}, \\
    \tau_\ell &\overset{\text{ind}.}{\sim} \text{Categorical}(\boldsymbol{\pi}_{1:T}) \\
    \mathbf{b}_\ell &\overset{\text{ind}.}{\sim} \mathcal{N}_d(\mathbf{0},\omega_0^{-1}\mathbf{I}_d), \label{eq:b-multi-prior}
\end{align}
where $\boldsymbol{\mu}_0 \in \mathbb{R}^d$ is an unknown intercept and $\boldsymbol{\Lambda}_{1:T}$ is a known sequence of precision matrices. We again define the variable and parameter blocks $\boldsymbol{\theta}_\ell := \{\mathbf{b}_\ell, \tau_\ell\}$ and $\overline{\boldsymbol{\theta}}_\ell := \{\overline{\mathbf{b}}_{\ell t}, \overline{\boldsymbol{\Omega}}_{\ell t}, \overline{\boldsymbol{\pi}}_{\ell t}\}_{t=1}^T$. Then Algorithm \ref{alg:mich-multi} shows how to modify Algorithm \ref{alg:mich-eb} to handle multivariate mean changes of the form (\ref{eq:mu_t-multi}).

\begin{algorithm}
\label{alg:mich-multi}
\caption{Variational Bayes Approximation to MICH Posterior for Mulivariate Mean Changes.}

\footnotesize
\SetAlgoLined
  Inputs: $\mathbf{y}_{1:T},\:\boldsymbol{\Lambda}_{1:T},\:L,\:\omega_0,\boldsymbol{\pi}_{1:T}\:;$ \\
  Initialize: $\boldsymbol{\mu}_0,\:\{\overline{\boldsymbol{\theta}}_\ell\}_{\ell=1}^L\:$;
  
  \Repeat {Convergence} {
    \For{$\ell=1$ \KwTo $L$} {
      $\overline{\mathbf{r}}_{-\ell t} := \mathbf{y}_t -\boldsymbol{\mu}_0 - \sum_{\ell' \neq \ell}^L  \sum_{t'=1}^t \overline{\mathbf{b}}_{\ell' t'} \overline{\pi}_{\ell' t'}$ \tcp*{l\textsuperscript{th} partial mean residual}
      $\overline{\boldsymbol{\theta}}_\ell := \texttt{mean-scp}(\overline{\mathbf{r}}_{-\ell,1:T} \:;\: \boldsymbol{\Lambda}_{1:T}, \omega_{0}, \boldsymbol{\boldsymbol{\pi}}_{1:T})$ \tcp*{update mean parameters}
    }
    $\overline{\mathbf{r}}_{t} := \mathbf{y}_t- \sum_{\ell=1}^L  \sum_{t'=1}^t \overline{\mathbf{b}}_{\ell t'} \overline{\pi}_{\ell t'}$ \tcp*{mean residual}
    $\boldsymbol{\mu}_0 :=\left(\sum_{t=1}^T\boldsymbol{\Lambda}_t\right)^{-1}\sum_{t=1}^T\boldsymbol{\Lambda}_t\overline{\mathbf{r}}_t $\tcp*{update intercept}
  }
  \Return{Posterior Parameters: $\boldsymbol{\mu}_0,\:\{\overline{\boldsymbol{\theta}}_\ell\}_{\ell=1}^L$}.
\end{algorithm}

As with Algorithms \ref{alg:mich} and \ref{alg:mich-eb}, Algorithm \ref{alg:mich-multi}  begins by initializing some values for $\overline{\boldsymbol{\theta}}_{1:L}$, which are then used to approximate the residual $\overline{\mathbf{r}}_{t}$. Each $\overline{\boldsymbol{\theta}}_\ell$ is then updated iteratively by fitting the mean-scp model to the partial residual $\overline{\mathbf{r}}_{-\ell t}$ in place of $\mathbf{y}_t$. Fitting each mean-scp model returns a distribution, so given the value of $\overline{\boldsymbol{\theta}}_{1:L}$ at the current iteration of Algorithm \ref{alg:mich-multi}, we have:
\begin{align}
    \mathbf{b}_\ell \:|\:  \overline{\mathbf{r}}_{-\ell,1:T},  \boldsymbol{\Lambda}_{1:T}, \tau_\ell = t &\sim \mathcal{N}_d\left(\overline{\mathbf{b}}_{\ell t}, \overline{\boldsymbol{\Omega}}^{-1}_{\ell t}\right). \label{eq:mich-multi-b-post}
\end{align}
Let $\{q_\ell\}_{\ell=1}^L$ stand in for the $L$ conditional distributions in (\ref{eq:mich-multi-b-post}) with $q_\ell(\tau_\ell = t) = \overline{\pi}_{\ell t}$. If we define $q:=\prod_{\ell=1}^L q_\ell$, then $q$ is a variational approximation to the true posterior $p:=p(\boldsymbol{\theta}_{1:L}\:\mathbf{y}_{1:T}; \boldsymbol{\mu}_0,\boldsymbol{\Lambda}_{1:T})$ that converges to a stationary point of the ELBO:
\begin{align*}
    \text{ELBO}(q\:;\:\boldsymbol{\mu}_0,\boldsymbol{\Lambda}_{1:T}) &:= \int q(\boldsymbol{\theta}_{1:L}) \log \frac{ p(\mathbf{y}_{1:T},\boldsymbol{\theta}_{1:L};\boldsymbol{\mu}_0,\boldsymbol{\Lambda}_{1:T})}{q(\boldsymbol{\theta}_{1:L})} \; d\boldsymbol{\theta}_{1:L} \\
    &\;= \log p(\mathbf{y}_{1:T};\boldsymbol{\mu}_0,\boldsymbol{\Lambda}_{1:T}) - \text{KL}( q \:\lVert\: p).
\end{align*}
This follows directly from the following result: 

\begin{proposition} 
\label{prop:vb-multi}
Assume the model (\ref{eq:y-multi})-(\ref{eq:b-multi-prior}) with $\omega_0, \pi_{ t} > 0$. Let $\{q_\ell\}_{\ell =1}^L$ be an arbitrary collection of distributions such that $q:=\prod_{\ell=1}^L q_\ell$ and each coordinate of $\E_{q_\ell}[\boldsymbol{\mu}_{\ell t}]$ and $\E_{q_\ell}[\lVert\boldsymbol{\mu}_{\ell t}\rVert_2^2]$ are finite. Define the residual $\overline{\mathbf{r}}_t := \mathbf{y}_t - \sum_{\ell=1}^L \E_{q_\ell}[\boldsymbol{\mu}_{\ell t}]$ and the partial residual $\overline{\mathbf{r}}_{-\ell t} := \overline{\mathbf{r}}_t + \E_{q_\ell}[\boldsymbol{\mu}_{\ell t}]$. Then:
\begin{align*}
    \text{\normalfont{arg}}\max_{q_\ell} \text{\normalfont ELBO}\left(q;\boldsymbol{\mu}_0,\boldsymbol{\Lambda}_{1:T}\right)
\end{align*}
is equivalent to the mean-scp posterior in (\ref{eq:gamma-post-cat1}) and (\ref{eq:b-smcp}) with parameters: 
\begin{align*}
    \overline{\boldsymbol{\theta}}_\ell = \normalfont{\texttt{mean-scp}}(\overline{\mathbf{r}}_{-\ell,1:T} \:;\: \boldsymbol{\Lambda}_{1:T}, \omega_{\ell}, \boldsymbol{\pi}_{\ell,1:T}).
\end{align*}
Furthermore, when each $q_\ell$ takes this functional form, then $\text{\normalfont ELBO}\left(q;\boldsymbol{\mu}_0,\boldsymbol{\Lambda}_{1:T}\right)$ is equivalent to a continuously differentiable function $F(\overline{\boldsymbol{\theta}}_{1:L};\boldsymbol{\mu}_0,\boldsymbol{\Lambda}_{1:T})$, and $\overline{\boldsymbol{\theta}}_\ell$ is the unique maximizer of $F$ holding $\{\overline{\boldsymbol{\theta}}_{\ell'}\}_{\ell'\neq \ell}$ fixed.
\end{proposition}

The proof of Proposition \ref{prop:vb-multi} follows an identical argument to the proof of Proposition \ref{prop:vb} (i). We simply replace the univariate variables $b_{\ell}$ and $\lambda_t$ with their multivariate analogues $\mathbf{b}_{\ell}$ and $\boldsymbol{\Lambda}_t$. For the purpose of monitoring convergence (see Appendix \ref{app:convergence}), we characterize the function $F$:
\begin{align*}
    F(\overline{\boldsymbol{\theta}}_{1:L};\boldsymbol{\mu}_0,\boldsymbol{\Lambda}_{1:T}) &\quad:= \E_q\left[\log p(\mathbf{y}_{1:T},\boldsymbol{\mu}_0,\boldsymbol{\Lambda}_{1:T})\right] - \sum_{\ell=1}^L \text{KL}(q_\ell\:\lVert\: p_\ell)\\
    &\underset{\boldsymbol{\mu}_0,\boldsymbol{\theta}_{1:L}}{\propto} \frac{1}{2} \sum_{t=1}^T \log |\boldsymbol{\Lambda_t}| - \frac{1}{2} \sum_{t=1}^{T} \E_q\left[\left\lVert \boldsymbol{\Lambda}_t^{\frac{1}{2}}(\mathbf{y}_t -\boldsymbol{\mu}_t)\right\rVert_2^2\right] \\
    &\quad\quad\quad- \E_{q_\ell}\left[\sum_{t=1}^T \mathbbm{1}\{\tau_\ell = t\}\left(\log \frac{q(\mathbf{b}_\ell|\tau_\ell =t)}{p(\mathbf{b}_\ell|\tau_\ell =t)}+  \log \frac{\overline{\pi}_{\ell t}}{\pi_{\ell t}}\right) \right] \\
    &\quad=\frac{1}{2} \sum_{t=1}^T \log |\boldsymbol{\Lambda_t}| - \frac{1}{2} \sum_{t=1}^{T} \left( \left\lVert \boldsymbol{\Lambda}_t^{\frac{1}{2}}\overline{\mathbf{r}}_t\right\rVert_2^2 + \E_q\left[\left\lVert \boldsymbol{\Lambda}_t^{\frac{1}{2}}(\boldsymbol{\mu}_t -\boldsymbol{\overline{\mu}}_t)\right\rVert_2^2\right]\right) \\
    &\quad\quad -\frac{1}{2}\sum_{\ell = 1}^L\overline{\pi}_{\ell t}\left[\log|\overline{\boldsymbol{\Omega}}_{\ell t} - d \log \omega_\ell - E_{q_\ell}[\lVert\overline{\boldsymbol{\Omega}}_{\ell t}^{\frac{1}{2}}(\mathbf{b}_{\ell} - \overline{\mathbf{b}}_{\ell t})\rVert_2^2|\tau_\ell] +\omega_\ell E_{q_\ell}[\lVert\mathbf{b}_{\ell}\rVert_2^2|\tau_\ell] +\log \frac{\overline{\pi}_{\ell t}}{\pi_{\ell t}}\right] \\
    &\quad= \frac{1}{2} \sum_{t=1}^T \log |\boldsymbol{\Lambda_t}| - \frac{1}{2} \sum_{t=1}^{T} \left\{ \left\lVert \boldsymbol{\Lambda}_t^{\frac{1}{2}}\overline{\mathbf{r}}_t\right\rVert_2^2 + \sum_{\ell=1}^L\left(\E_{q_\ell}\left[\left\lVert \boldsymbol{\Lambda}_t^{\frac{1}{2}}\boldsymbol{\mu}_{\ell t}\right\rVert_2^2\right] - \left\lVert \boldsymbol{\Lambda}_t^{\frac{1}{2}}\overline{\boldsymbol{\mu}}_{\ell t}\right\rVert_2^2\right)\right\} \\
    &\quad\quad -\frac{1}{2}\sum_{\ell = 1}^L\overline{\pi}_{\ell t}\left[\log|\overline{\boldsymbol{\Omega}}_{\ell t}| - d \log \omega_\ell - d +\omega_\ell \left(\text{trace}\left(\overline{\boldsymbol{\Omega}}_{\ell t}^{-1}\right) + \lVert\overline{\mathbf{b}}_{\ell t}\rVert_2^2\right) +\log \frac{\overline{\pi}_{\ell t}}{\pi_{\ell t}} \right].
\end{align*}
Note that we can write:
\begin{align*}
    \E_{q_\ell}\left[\left\lVert \boldsymbol{\Lambda}_t^{\frac{1}{2}}\boldsymbol{\mu}_{\ell t}\right\rVert_2^2\right] &= \sum_{t'=1}^t \overline{\pi}_{\ell t'} \sum_{1\leq i,i'\leq d} \left[\boldsymbol{\Lambda}_t\right]_{ii'}\left(\left[\overline{\boldsymbol{\Omega}}^{-1}_{\ell t'}\right]_{ii'} + \left[\overline{\mathbf{b}}_{\ell t'}\right]_{i}\left[\overline{\mathbf{b}}_{\ell t'}\right]_{i'}\right) \\
    &= \sum_{t'=1}^t \overline{\pi}_{\ell t'} \left( \left\lVert \boldsymbol{\Lambda}_t^{\frac{1}{2}}\overline{\mathbf{b}}_{\ell t'}\right\rVert_2^2 +  \mathbf{1}_d^T\left[\boldsymbol{\Lambda}_t \odot\overline{\boldsymbol{\Omega}}^{-1}_{\ell t'}\right]\mathbf{1}_d\right),
\end{align*}
where $\odot$ denotes the Hadamard product. So:
\begin{align*}
    F(\overline{\boldsymbol{\theta}}_{1:L};\boldsymbol{\mu}_0,\boldsymbol{\Lambda}_{1:T}) &\underset{\boldsymbol{\mu}_0,\boldsymbol{\theta}_{1:L}}{\propto} \frac{1}{2} \sum_{t=1}^T \log |\boldsymbol{\Lambda}_t| - \frac{1}{2} \sum_{t=1}^{T} \left\{ \left\lVert \boldsymbol{\Lambda}_t^{\frac{1}{2}}\overline{\mathbf{r}}_t\right\rVert_2^2 - \sum_{\ell=1}^L \left\lVert \boldsymbol{\Lambda}_t^{\frac{1}{2}}\overline{\boldsymbol{\mu}}_{\ell t}\right\rVert_2^2\right\} \\
    &\quad\quad - \frac{1}{2} \sum_{t=1}^{T}\sum_{t'=1}^t\sum_{\ell=1}^L \overline{\pi}_{\ell t'} \left( \left\lVert \boldsymbol{\Lambda}_t^{\frac{1}{2}}\overline{\mathbf{b}}_{\ell t'}\right\rVert_2^2 + \mathbf{1}_d^T\left[\boldsymbol{\Lambda}_t \odot\overline{\boldsymbol{\Omega}}^{-1}_{\ell t'}\right]\mathbf{1}_d\right) \\
    &\quad\quad -\frac{1}{2}\sum_{\ell = 1}^L\overline{\pi}_{\ell t}\left[\log|\overline{\boldsymbol{\Omega}}_{\ell t}| - d \log \omega_\ell - d +\omega_\ell \left(\text{trace}\left(\overline{\boldsymbol{\Omega}}_{\ell t}^{-1}\right) + \lVert\overline{\mathbf{b}}_{\ell t}\rVert_2^2\right) + \log \frac{\overline{\pi}_{\ell t}}{\pi_{\ell t}}\right].
\end{align*}
As argued in Appendix \ref{app:empirical-bayes}, the empirical Bayes step for updating $\boldsymbol{\mu}_0$ in Algorithm \ref{alg:mich-multi} also constitutes a coordinate ascent step for maximizing the ELBO with respect to $\boldsymbol{\mu}_0$ seeing as:
\begin{align*}
     \E_{q} \left[\log p\left(\mathbf{y}_{1:T} \:|\: \boldsymbol{\Theta};\boldsymbol{\mu}_0,\boldsymbol{\Lambda}_{1:T}\right)\right] &\underset{\boldsymbol{\mu}_0,\boldsymbol{\theta}_{1:L}}{\propto} -\frac{1}{2} \left\lVert\left( \sum_{t=1}^T \boldsymbol{\Lambda}_t\right)^{\frac{1}{2}}\boldsymbol{\mu}_0\right\rVert_2^2 +\left\langle\boldsymbol{\mu}_0, \sum_{t=1}^T\boldsymbol{\Lambda}_t\overline{\mathbf{r}}_t \right\rangle,
\end{align*}
which is convex as a function of $\boldsymbol{\mu}_0$ and uniquely maximized on $\mathbb{R}^d$ at:
\begin{align*}
    \hat{\boldsymbol{\mu}}_0 := \left(\sum_{t=1}^T\boldsymbol{\Lambda}_t\right)^{-1}\sum_{t=1}^T\boldsymbol{\Lambda}_t\overline{\mathbf{r}}_t
\end{align*}
So once again, Algorithm \ref{alg:mich-multi} constitutes a coordinate ascent procedure. It is possible that the sequence of precision matrices is unknown. In this case, we can replace $\boldsymbol{\Lambda}_{1:T}$ with a consistent estimator $\hat{\boldsymbol{\Lambda}}_{1:T}$ in Algorithm \ref{alg:mich-multi}. For instance, if $\boldsymbol{\Lambda}_t = \boldsymbol{\Lambda} \sforall t \in [T]$, then $\Var(y_t - y_{t-1}) = 2\boldsymbol{\Lambda}^{-1}$ for each $t > 1$. We can define $\Tilde{y}_t := y_t - y_{t-1}$, then:
\begin{align*}
    \hat{\boldsymbol{\Lambda}}_T := \left[\frac{\sum_{t=1}^{T-1} \left(\Tilde{y}_t - \frac{\sum_{t=1}^{T-1}\Tilde{y}_t}{T-1}\right)\left(\Tilde{y}_t - \frac{\sum_{t=1}^{T-1}\Tilde{y}_t}{T-1}\right)^T}{2(T-2)}\right]^{-1}
\end{align*}
will be a consistent estimator of $\boldsymbol{\Lambda}$ that we can plug directly into Algorithm \ref{alg:mich-multi}.

\newpage
\section{Computational Details}
\subsection{Posterior Parameters}
\label{app:posterior-parameters}

\subsubsection{Mean-SCP}

The full posterior for the mean-scp model is given by:
\begin{align}
    \mathbf{b} \:|\: \tau = t, \: \mathbf{y}_{1:T} &\sim \mathcal{N}_d\left(\overline{\mathbf{b}}_{t}, \overline{\boldsymbol{\Omega}}_t^{-1}\right) \\
    \tau \:|\: \mathbf{y}_{1:T} &\sim \text{Categorical}(\overline{\boldsymbol{\pi}}_{1:T})  \\
    \overline{\boldsymbol{\Omega}}_t &= \omega_0\mathbf{I}_d + \sum_{t'=t}^{T} \boldsymbol{\Lambda}_{t'} \label{eq:mean-scp-post-omega}\\
    \overline{\mathbf{b}}_t &= \overline{\boldsymbol{\Omega}}_t^{-1}\sum_{t'=t}^{T} \boldsymbol{\Lambda}_{t'} \mathbf{y}_{t'} \label{eq:mean-scp-post-b}\\
    \overline{\pi}_t &\propto \pi_t|\overline{\boldsymbol{\Omega}}_t|^{-
    \frac{1}{2}}\exp\left[\frac{\lVert \overline{\boldsymbol{\Omega}}_t^{\frac{1}{2}} \overline{\mathbf{b}}_t\rVert_2^2}{2}\right].  \label{eq:mean-scp-post-pi}
\end{align}

\subsubsection{Var-SCP}

The full posterior for the var-scp model is given by:
\begin{align}
    s \:|\: \tau = t, \: \mathbf{y}_{1:T} &\sim \text{Gamma}\left(\overline{u}_{t}, \overline{v}_{t}\right) \\
    \tau \:|\: \mathbf{y}_{1:T}&\sim \text{Categorical}(\overline{\boldsymbol{\pi}}_{1:T}) \\
    \overline{u}_{t} &= u_0 + \frac{T - t + 1}{2} \label{eq:var-scp-post-u} \\
    \overline{v}_{t} &= v_0 + \frac{1}{2} \sum_{t'=t}^{T} \omega_{t'}y_{t'}^2 \label{eq:var-scp-post-v} \\
    \overline{\pi}_t &\propto  \frac{\pi_t\Gamma(\overline{u}_{t})}{\overline{v}_{t}^{\overline{u}_{t}}}\exp\left(- \frac{1}{2}\sum_{t'=1}^{t-1} \omega_{t'}y_{t'}^2\right). \label{eq:var-scp-post-pi}
\end{align}

Note that the summation in the last line may be ill-defined if $t = 1$. We use the convention that a sum is equal to zero if the indexing set is empty.

\subsubsection{MeanVar-SCP}

The full posterior for the meanvar-scp model is given by:
\begin{align}
    b \:|\: s, \tau = t, \mathbf{y}_{1:T} &\sim \mathcal{N}(\overline{b}_t, (\overline{\omega}_t s)^{-1}) \\
    s \:|\: \tau = t, \mathbf{y}_{1:T} &\sim \text{Gamma}(\overline{u}_t, \overline{v}_t) \\
    \tau \:|\: \mathbf{y}_{1:T} &\sim \text{Categorical}(\overline{\boldsymbol{\pi}}_{1:T}) \\
    \overline{\omega}_t &= \omega_0 + \sum_{t' = t}^{T} \omega_{t'} \label{eq:meanvar-scp-post-omega} \\
    \overline{b}_t &=  \sum_{t'=t}^{T} \frac{\omega_{t'} y_{t'}}{\overline{\omega}_t} \\
    \overline{u}_t &= u_0 + \frac{T  - t + 1}{2} \label{eq:meanvar-scp-post-b}\\
    \overline{v}_t &= v_0 - \frac{\overline{\omega}_t\overline{b}^2_t}{2} + \frac{1}{2} \sum_{t'=t}^{T} \omega_{t'}y_{t'}^2 \label{eq:meanvar-scp-post-v} \\
    \overline{\pi}_t &\propto \frac{\pi_t\Gamma(\overline{u}_t)}{\overline{v}_t^{\overline{u}_t}\overline{\omega}_t^{1/2} } \exp\left(-\frac{1}{2}\sum_{t'=1}^{t-1} \omega_{t'}y^2_{t'}\right). \label{eq:meanvar-scp-post-pi}
\end{align}

\subsection{Priors on Change-Point Locations}

\label{app:prior}

In the absence of any prior knowledge about the location of the change-point $\tau \in [T]$, one natural choice is the uniform prior, $\pi_t = T^{-1}$. While this choice of $\boldsymbol{\pi}_{1:T}$ satisfies the assumption of Theorems \ref{theorem:smcp}-\ref{theorem:alpha-mixing}
that for some constant $C_\pi > 0$:
\begin{align}
   \min_{t\in[T]} |\log \pi_{t}| \leq C_\pi \log T \label{eq:prior-bd}
\end{align}
and therefore guarantees asymptotic consistency for each of the single change-point models in Section \ref{sec:scp}, the uniform prior may reduce detection power or even result in false positives when $T$ is small. To see this, we first show that the uniform prior does not induce a uniform posterior in the absence of a change-point. In Figure \ref{fig:post-probs-plot}, we draw 10,000 replicates of $\mathbf{y}_{1:T}$ for $T = 50$ under the null model, i.e. where $y_t \overset{\text{i.i.d.}}{\sim}\mathcal{N}(0,1)$, then fit each SCP model to $\mathbf{y}_{1:T}$ and use the average of the resulting posterior probabilities $\overline{\boldsymbol{\pi}}_{1:T}$ to produce MCMC estimates of $\E[\overline{\boldsymbol{\pi}}_{1:T}]$ for each of SCP models.

\begin{figure}[!h]
    \centering
    \includegraphics[scale = 0.27]{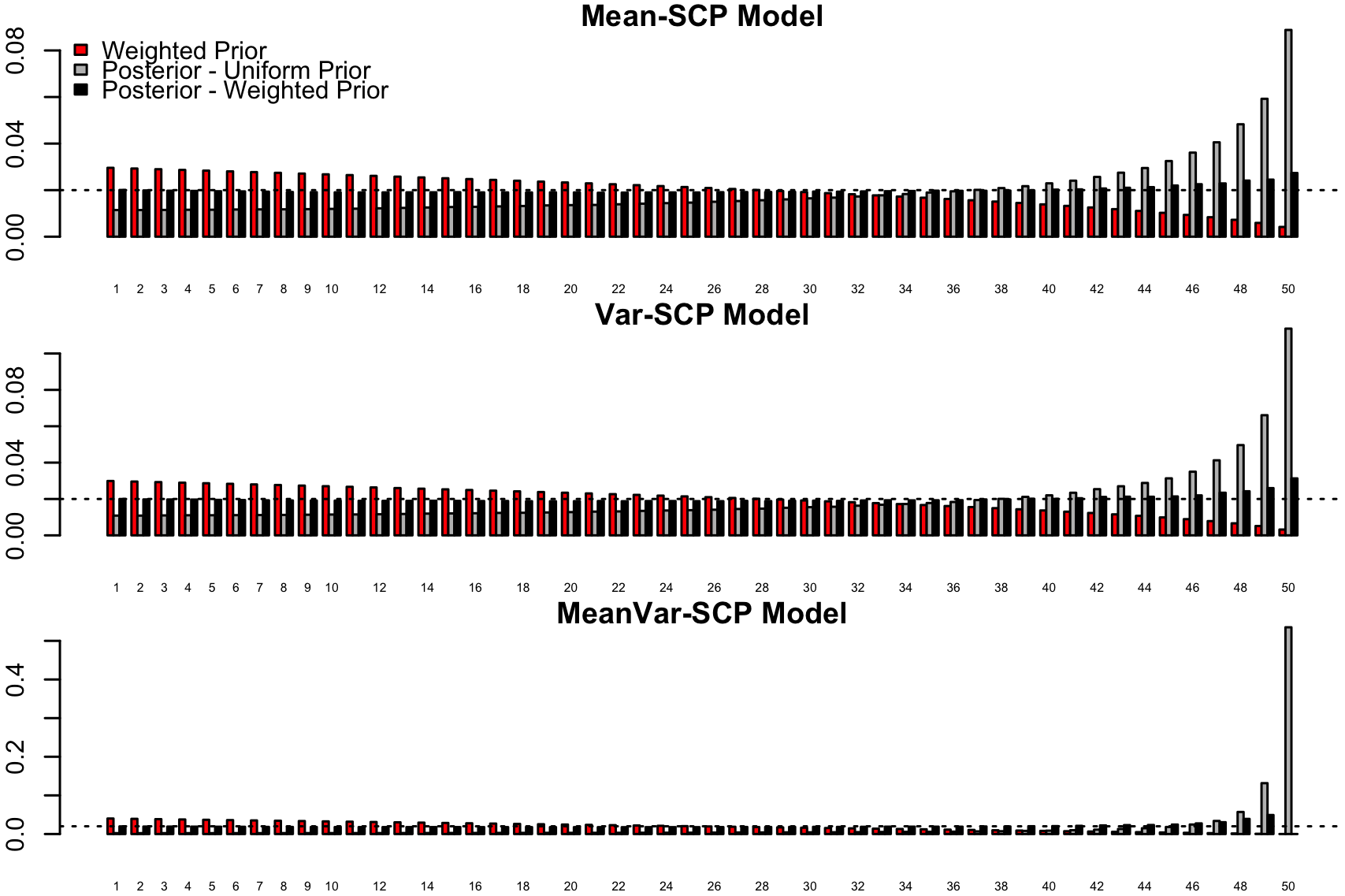}
    \caption{\textbf{Posterior Probabilities under Null Model}. MCMC estimates of $\E[\overline{\boldsymbol{\pi}}_{1:T}]$ under Null Model with $T = 50$. Gray bars show $\E[\overline{\boldsymbol{\pi}}_{1:T}]$ under the uniform prior and black bars show $\E[\overline{\boldsymbol{\pi}}_{1:T}]$ under the weighted prior (red bars).}
    \label{fig:post-probs-plot}
\end{figure}

Figure \ref{fig:post-probs-plot} shows an exponential increase in the posterior probabilities as $t$ approaches $T$. With a uniform prior, the SCP models may fail to detect a true change when the signal is weak or in small samples due to the large weight placed on times near $T$. For the var-scp and meanvar-scp models in particular, we see that the probabilities may be large enough to incorrectly detect a change-point in the vicinity of time $T$, even when no change is present. To rectify this behavior we propose selecting $\boldsymbol{\pi}_{1:T}$ so that under the null model we have:
\begin{align}\label{eq:prior-cond}
    \E\left[\log\overline{\pi}_t - \log \overline{\pi}_{t+1} \right] = 0.
\end{align}
In Figure \ref{fig:post-probs-plot} we also plot the weighted priors $\boldsymbol{\pi}_{1:T}$ that satisfy (\ref{eq:prior-cond}) and MCMC estimates of $\E[\overline{\boldsymbol{\pi}}_{t}]$ under this choice $\boldsymbol{\pi}_{1:T}$. We clearly see that these estimates adhere much more closely to the uniform dashed line at $T^{-1}$. Note that (\ref{eq:prior-cond}) also implies a uniform condition on the posterior probabilities since for any $r > t$:
\begin{align*}
    \E\left[\log\overline{\pi}_t - \log \overline{\pi}_r \right] &= \sum_{i=0}^{r-t-1} \E\left[\log\overline{\pi}_{t+i} - \log \overline{\pi}_{t+i+1} \right] = 0.
\end{align*}
We now show how to calculate $\pi_t$ so that (\ref{eq:prior-cond}) holds for each of the SCP models. In each case we show that, if $t_0$ is the true location of the change point

\subsubsection{Mean-SCP Prior}

For the mean-scp model in Section \ref{sec:smcp} we have:
\small
\begin{align*}
    \E\left[\log\overline{\pi}_t - \log \overline{\pi}_{t+1} \right] &= \log \pi_t - \log \pi_{t+1} - \frac{1}{2} \log \left|\omega_0\mathbf{I}_d + \sum_{t'=t}^{T} \boldsymbol{\Lambda}_{t'}\right| + \frac{1}{2} \log \left|\omega_0\mathbf{I}_d + \sum_{t'=t+1}^{T} \boldsymbol{\Lambda}_{t'}\right|\\
    &\quad + \frac{1}{2} \E\left[\left\lVert \left[\omega_0\mathbf{I}_d + \sum_{t'=t}^{T} \boldsymbol{\Lambda}_{t'}\right]^{-\frac{1}{2}} \sum_{t'=t}^{T} \boldsymbol{\Lambda}_{t'} \mathbf{y}_{t'}\right\rVert_2^2 \right] - \frac{1}{2} \E\left[\left\lVert \left[\omega_0\mathbf{I}_d + \sum_{t'=t+1}^{T} \boldsymbol{\Lambda}_{t'}\right]^{-\frac{1}{2}} \sum_{t'=t+1}^{T} \boldsymbol{\Lambda}_{t'} \mathbf{y}_{t'}\right\rVert_2^2 \right].
\end{align*}
\normalsize
Letting $\omega_0 \to 0$ and noting that $\sum_{t'=t}^{T} \boldsymbol{\Lambda}_{t'} \mathbf{y}_{t'} \sim \mathcal{N}_d\left(\mathbf{0}, \sum_{t'=t}^{T} \boldsymbol{\Lambda}_{t'}\right)$ under the null model, we get:
\small
\begin{align*}
    \E\left[\log\overline{\pi}_t - \log \overline{\pi}_{t+1} \right] &= \log \pi_t - \log \pi_{t+1} - \frac{1}{2} \log \left|\sum_{t'=t}^{T} \boldsymbol{\Lambda}_{t'}\right| + \frac{1}{2} \log \left|\sum_{t'=t+1}^{T} \boldsymbol{\Lambda}_{t'}\right|\\
    &\quad + \frac{1}{2} \E\left[\left\lVert \left[ \sum_{t'=t}^{T} \boldsymbol{\Lambda}_{t'}\right]^{-\frac{1}{2}} \sum_{t'=t}^{T} \boldsymbol{\Lambda}_{t'} \mathbf{y}_{t'}\right\rVert_2^2  \right] - \frac{1}{2} \E\left[\left\lVert \left[\sum_{t'=t+1}^{T} \boldsymbol{\Lambda}_{t'}\right]^{-\frac{1}{2}} \sum_{t'=t+1}^{T} \boldsymbol{\Lambda}_{t'} \mathbf{y}_{t'}\right\rVert_2^2  \right]\\
    &= \log \pi_t - \log \pi_{t+1} - \frac{1}{2} \log \left|\sum_{t'=t}^{T} \boldsymbol{\Lambda}_{t'}\right| + \frac{1}{2} \log \left|\sum_{t'=t+1}^{T} \boldsymbol{\Lambda}_{t'}\right|.
\end{align*}
\normalsize
If $\boldsymbol{\Lambda}_t = \boldsymbol{\Lambda}$ for all $t$, we can further simplify to:
\begin{align*}
    \E\left[\log\overline{\pi}_t - \log \overline{\pi}_{t+1} \right] &= \log \pi_t - \log \pi_{t+1} + \frac{d}{2} \log\left(\frac{T-t}{T-t+1}\right).
\end{align*}
Therefore, when (\ref{eq:prior-cond}) holds, we can set $\log \pi_1 = 0$ and solve for $\log \pi_t$ for $t > 1$ using the recurrence relation above, then normalize the sequence $\boldsymbol{\pi}_{1:T}$ to get:
\begin{align*}
    \log \pi_t &= \frac{d}{2}\log\left(\frac{T-t+1}{T}\right) - \log\left[\sum_{t'=1}^T \left(\frac{T-t+1}{T}\right)^{\frac{d}{2}}\right] \\
    &\geq - \left(1 + \frac{d}{2}\right)\log T. \tag{$\frac{T-t+1}{T} \leq 1 \sforall t \in [T]$ }
\end{align*}
So for fixed $d$ this choice of $\pi_t$ satisfies (\ref{eq:prior-bd}) with $C_\pi = 1 + d/2$.

\subsubsection{Var-SCP Prior}

For the var-scp model in Section \ref{sec:sscp} we have:
\small
\begin{align*}
    \E\left[\log\overline{\pi}_t - \log \overline{\pi}_{t+1} \right] &= \log \pi_t- \log \pi_{t+1} +\log \Gamma\left(u_0 +\frac{ T-t+1}{2}\right) - \log \Gamma\left(u_0 +\frac{ T-t}{2}\right) + \frac{1}{2} \E[\omega_ty_t^2] \\
    &\quad \: + \left(u_0 + \frac{T - t}{2}\right)\E\left[\log\left[v_0 +\frac{1}{2}\left(\sum_{t'=t+1}^T \omega_{t'}y_{t'}^2\right)\right]\right] \\
     &\quad \:- \left(u_0 + \frac{T - t +1}{2}\right)\E\left[\log\left[v_0 +\frac{1}{2}\left(\sum_{t'=t}^T \omega_{t'}y_{t'}^2\right)\right]\right].
\end{align*}
\normalsize
Letting $u_0,v_0 \to 0$ and noting that $\omega_{t}y^2_{t} \sim \chi^2_1$ under the null model, we get:
\small
\begin{align*}
    \E\left[\log\overline{\pi}_t - \log \overline{\pi}_{t+1} \right] &= \log \pi_t- \log \pi_{t+1} +\log \Gamma\left(\frac{T-t+1}{2}\right) - \log \Gamma\left(\frac{T-t}{2}\right) + \frac{1 + \log 2}{2}\\
    &\quad \: + \left(\frac{T - t}{2}\right)\E\left[\log\left(\sum_{t'=t+1}^T \omega_{t'}y_{t'}^2\right)\right] - \left(\frac{T - t +1}{2}\right)\E\left[\log\left(\sum_{t'=t}^T \omega_{t'}y_{t'}^2\right)\right].
\end{align*}
\normalsize
Since $\sum_{t'=t}^T \omega_{t'}y_{t'}^2 \sim \chi^2_{T-t+1}$, we have:
\begin{align*}
    \E\left[\log\left(\sum_{t'=t}^T \omega_{t'}y_{t'}^2\right)\right] = \psi\left(\frac{T-t+1}{2}\right) + \log 2
\end{align*}
where $\psi$ is the digamma function. Therefore:
\begin{align*}
    \E\left[\log\overline{\pi}_t - \log \overline{\pi}_{t+1} \right] &= \log \pi_t- \log \pi_{t+1} +\log \Gamma\left(\frac{T-t+1}{2}\right) - \log \Gamma\left(\frac{T-t}{2}\right) + \frac{1}{2}\\
    &\quad \: + \left(\frac{T - t}{2}\right)\psi\left(\frac{T-t}{2}\right)  - \left(\frac{T - t +1}{2}\right)\psi\left(\frac{T-t+1}{2}\right).
\end{align*}
So again we have defined a recurrence relation for calculating $\pi_t$. Recall Binet's first formula for the log gamma function for $x > 0$ (see e.g. p. 249 of \citealp{Whittaker96}): 
\begin{align*}
    \log \Gamma(x) = \left(x - \frac{1}{2}\right) \log x - x + \frac{1}{2} \log 2 \pi + \int_{0}^\infty\left(\frac{1}{2} - \frac{1}{t} + \frac{1}{e^t - 1} \right)\frac{e^{-tx}}{t} \;dt.
\end{align*}
Since in Theorem \ref{theorem:sscp} we only consider indexes $t$ such that $T-t+1 > \log^{\varepsilon} T$, we actually only need $\log \pi_t > -C_\pi \log T$ for just these indexes. Since $T-t+1 > \log^{\varepsilon} T \implies T-t+1 \to \infty$, we have the approximations:
\begin{align*}
    \log \Gamma\left(\frac{T-t+1}{2}\right) &\approx \left(\frac{T-t}{2}\right)\log\left(\frac{T - t +1}{2}\right) - \frac{T - t +1}{2} +\frac{1}{2} \log 2 \pi \\
    \psi\left(\frac{T-t+1}{2}\right)  &\approx \log \left(\frac{T-t+1}{2}\right) - \frac{1}{T-t+1}.
\end{align*}
Then we can write:
\begin{align*}
    \E\left[\log\overline{\pi}_t - \log \overline{\pi}_{t+1} \right] &\approx \log \pi_t- \log \pi_{t+1} + \frac{1}{2}\log\left(\frac{T - t}{T-t+1}\right).
\end{align*}
This is the same recurrence relation from the mean-scp model with $d=1$, so again the choice of $\pi_t$ implied by the var-scp recurrence relation satisfies (\ref{eq:prior-bd}).

\subsubsection{MeanVar-SCP Prior}

For the meanvar-scp model in Section \ref{sec:smscp} we have:
\small
\begin{align*}
    \E\left[\log\overline{\pi}_t - \log \overline{\pi}_{t+1} \right] &= \log \pi_t- \log \pi_{t+1} + \frac{1}{2} \E[\omega_ty_t^2] \\
    &\quad \: + \frac{1}{2}\log\left(\sum_{t'=t+1}^T \omega_{t'}\right) - \frac{1}{2}\log\left(\sum_{t'=t}^T \omega_{t'}\right) + \log \Gamma\left(\frac{u_0 + T-t+1}{2}\right) - \log \Gamma\left(\frac{u_0 + T-t}{2}\right)  \\
    &\quad \: + \left(u_0 + \frac{T - t}{2}\right)\E\left[\log\left(v_0 +\frac{1}{2}\sum_{t'=t+1}^T \omega_{t'}y_{t'}^2- \frac{\left(\sum_{t'=t+1}^T \omega_{t'}y_{t'}\right)^2}{2(\omega_0 + \sum_{t'=t+1}^T \omega_{t'})}\right)\right] \\
     &\quad \:- \left(u_0 + \frac{T - t +1}{2}\right)\E\left[\log\left(v_0 +\frac{1}{2}\sum_{t'=t}^T \omega_{t'}y_{t'}^2 - \frac{\left(\sum_{t'=t}^T \omega_{t'}y_{t'}\right)^2}{2(\omega_0 + \sum_{t'=t}^T \omega_{t'})}\right)\right].
\end{align*}
\normalsize
Letting $\omega_0, u_0,v_0 \to 0$ and noting that $\omega_{t}y^2_{t} \sim \chi^2_1$ under the null model, we get:
\small
\begin{align*}
    \E\left[\log\overline{\pi}_t - \log \overline{\pi}_{t+1} \right] &= \log \pi_t- \log \pi_{t+1} + \frac{1 + \log 2}{2} \\
    &\quad \: + \frac{1}{2}\log\left(\sum_{t'=t+1}^T \omega_{t'}\right) - \frac{1}{2}\log\left(\sum_{t'=t}^T \omega_{t'}\right) + \log \Gamma\left(\frac{ T-t+1}{2}\right) - \log \Gamma\left(\frac{T-t}{2}\right)  \\
    &\quad \: + \left(\frac{T - t}{2}\right)\E\left[\log\left(\sum_{t'=t+1}^T \omega_{t'}y_{t'}^2- \frac{\left(\sum_{t'=t+1}^T \omega_{t'}y_{t'}\right)^2}{\sum_{t'=t+1}^T \omega_{t'}}\right)\right] \\
     &\quad \:- \left(\frac{T - t +1}{2}\right)\E\left[\log\left(\sum_{t'=t}^T \omega_{t'}y_{t'}^2 - \frac{\left(\sum_{t'=t}^T \omega_{t'}y_{t'}\right)^2}{\sum_{t'=t}^T \omega_{t'}}\right)\right].
\end{align*}
\normalsize
Noting that the last line above is not well-defined when $t = T$, we can simply set $\pi_T = 0$. Since $\hat{\tau}_{\text{MAP}} \neq T$ in Theorem \ref{theorem:smscp}, this assumption is without loss of generality. For all other $t \in [T-1]$, since:
\begin{align*}
    \sum_{t'=t}^T \omega_{t'}y_{t'}^2 - \frac{\left(\sum_{t'=t}^T \omega_{t'}y_{t'}\right)^2}{\sum_{t'=t}^T \omega_{t'}} \sim \chi^2_{T-t}
\end{align*}
then we have:
\small
\begin{align*}
    \E\left[\log\overline{\pi}_t - \log \overline{\pi}_{t+1} \right] &= \log \pi_t- \log \pi_{t+1} + \frac{1}{2} + \frac{1}{2}\log\left(\sum_{t'=t+1}^T \omega_{t'}\right) - \frac{1}{2}\log\left(\sum_{t'=t}^T \omega_{t'}\right) \\
    &\quad \:+ \log \Gamma\left(\frac{ T-t+1}{2}\right) - \log \Gamma\left(\frac{T-t}{2}\right)  \\
    &\quad \: + \left(\frac{T - t}{2}\right)\psi\left(\frac{T-t-1}{2}\right) - \left(\frac{T - t +1}{2}\right)\psi\left(\frac{T-t}{2}\right).
\end{align*}
\normalsize
So again we have defined a recurrence relation for calculating $\pi_t$. Since $\boldsymbol{\tau}_{1:T}$ are known parameters, it is without loss of generality to assume that $\omega_t = 1$, then we have:   
then we have:
\small
\begin{align*}
    \E\left[\log\overline{\pi}_t - \log \overline{\pi}_{t+1} \right] &= \log \pi_t- \log \pi_{t+1} + \frac{1}{2} + \frac{1}{2}\log\left(\frac{T-t}{T-t+1}\right) + \log \Gamma\left(\frac{ T-t+1}{2}\right) - \log \Gamma\left(\frac{T-t}{2}\right)  \\
    &\quad \: + \left(\frac{T - t}{2}\right)\psi\left(\frac{T-t-1}{2}\right) - \left(\frac{T - t +1}{2}\right)\psi\left(\frac{T-t}{2}\right).
\end{align*}
\normalsize
Noting that the right hand side is just a sum of the recurrence relations for the mean-scp and var-scp models, then we can use the same argument that we used in those previous cases to establish that the choice of $\pi_t$ implied by the meanvar-scp recurrence relation satisfies (\ref{eq:prior-bd}).
\subsection{Merging Duplicate Components}
\label{app:merge-procedure}

As discussed in Section \ref{sec:merge-procedure}, Algorithm \ref{alg:mich} occasionally reaches a stationary point where a single change-point is split across multiple $\tau_i$'s. This behavior is plainly visible in Figure \ref{fig:duplicate}. We have generated $\mathbf{y}_{1:T}$ from Simulation \ref{sim:main} with $J^* = 10$, $T = 1000$, $\Delta_T = 50$, $C = \sqrt{200}$. After fitting the oracle version of MICH with $J = J^*$, we estimate two change-points $\hat{\tau}_2 = 179$ and $\hat{\tau}_3 = 184$ that clearly capture the the same true change at $\tau_2 = 179$, causing the model to miss the change-point at $\tau_{10} = 886$.

\begin{figure}[!h]
    \centering
    \includegraphics[scale = 0.175]{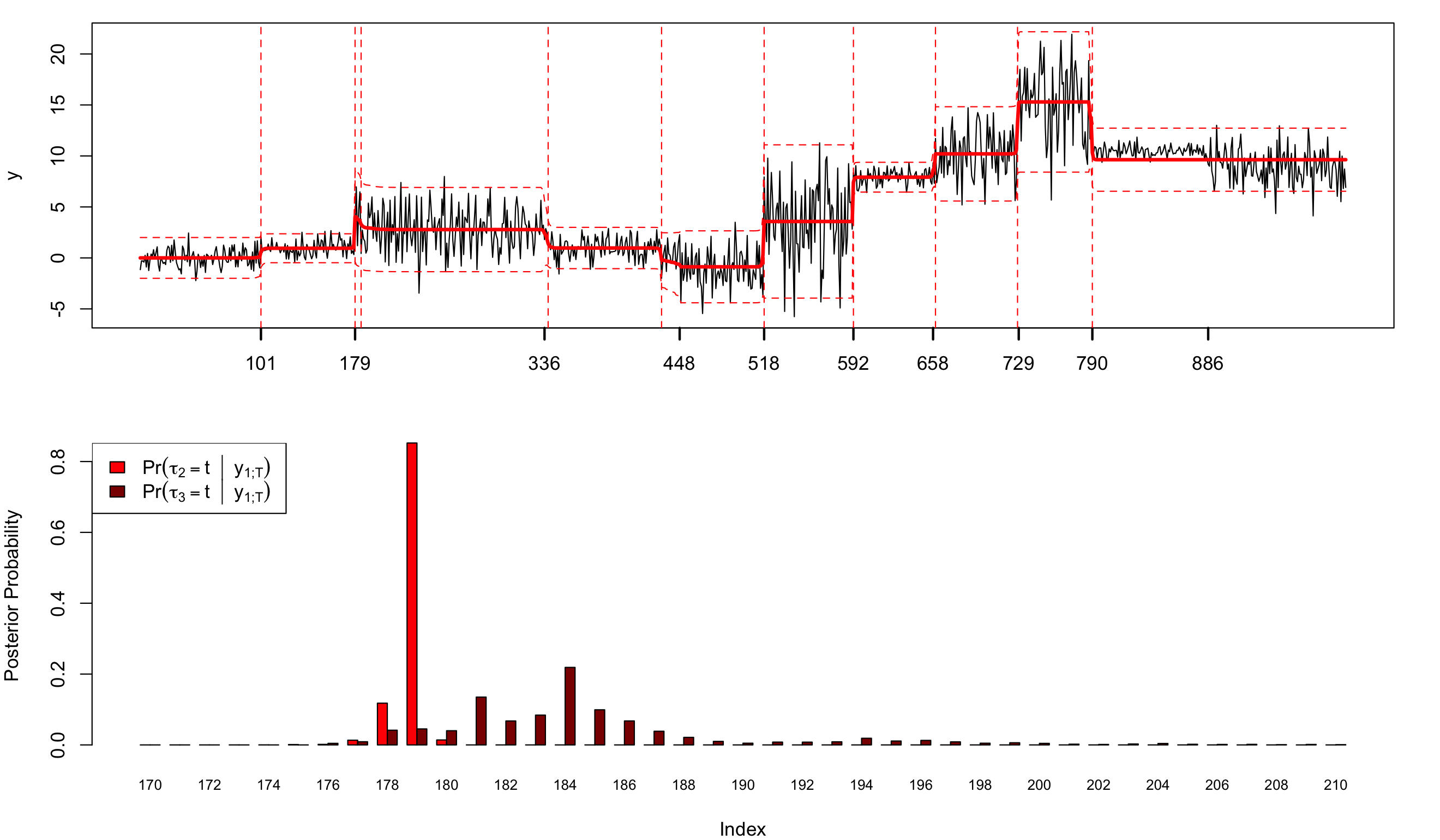}
    \caption{\textbf{MICH without Merging}. \textbf{Top}: Fit of $\mathbf{y}_{1:T}$ (black line) for $J^* = J = 10$ with true change points $\boldsymbol{\tau}_{1:J^*}$ (x-axis ticks), estimated changes $\hat{\boldsymbol{\tau}}_{1:\hat{J}}$ (vertical dashed red lines), and estimated mean (solid red line) and variance (dashed red lines) signals. \textbf{Bottom}: Estimated posterior distributions for $\hat{\tau}_2$ and $\hat{\tau}_3$.}
    \label{fig:duplicate}
\end{figure}

We propose a modification that helps Algorithm \ref{alg:mich} move out of these undesirable stationary points. If $\tau_i$ and $\tau_{i'}$ correspond to change-points of the same class, then by (\ref{eq:mean-field}) we have:
\begin{align}
    q(\tau_i = \tau_{i'}) = \sum_{t=1}^T q(\tau_{i'} = t | \tau_i = t ) q(\tau_i = t) = \sum_{t=1}^T q_{i'}(\tau_{i'} = t)q_i(\tau_i = t) = \langle\overline{\boldsymbol{\pi}}_{i',1:T}, \overline{\boldsymbol{\pi}}_{i,1:T} \rangle.
\end{align}
We therefore propose merging components $i$ and $i'$ if $\langle\overline{\boldsymbol{\pi}}_{i',1:T}, \overline{\boldsymbol{\pi}}_{i,1:T} \rangle$ exceeds some threshold $\beta > 0$. We use a greedy algorithm that merges the components $(i_1,i_2) := \arg\max_{i\neq i'} \langle\overline{\boldsymbol{\pi}}_{i',1:T}, \overline{\boldsymbol{\pi}}_{i,1:T}\rangle$ if $\langle\overline{\boldsymbol{\pi}}_{i_1,1:T}, \overline{\boldsymbol{\pi}}_{i_2,1:T} \rangle \geq \beta$. We continue merging components until the set $\{i,i' : i \neq i', \; \langle\overline{\boldsymbol{\pi}}_{i',1:T}, \overline{\boldsymbol{\pi}}_{i,1:T} \rangle\geq\beta\}$ is empty. When components $i$ and $i'$ both capture true change-points, then $\overline{\boldsymbol{\pi}}_{i,1:T}$ and $\overline{\boldsymbol{\pi}}_{i',1:T}$ tend to be sparse. Thus, when components $i$ and $i'$ correspond to distinct change-points, $\langle\overline{\boldsymbol{\pi}}_{i,1:T}, \overline{\boldsymbol{\pi}}_{i',1:T} \rangle$ tends to be vanishingly small. When components $i$ and $i'$ capture the same change-point, $\langle\overline{\boldsymbol{\pi}}_{i,1:T}, \overline{\boldsymbol{\pi}}_{i',1:T}\rangle$ will clearly be nonzero in most cases, so any small value of $\beta$ will correctly merge the true duplicates. After merging the duplicate components, we then remove the redundant components and use the current fit as an initialization for Algorithm \ref{alg:mich}, then use the automatic procedure described in Appendix \ref{app:LKJ-choice} to increase $L$, $K$, and $J$ to the desired totals.

To select a default for $\beta$, we consider the worst case scenario where $t^*_i$ and $t^*_{i+1}$ are consecutive change-points separated by the rate of the localization error in Theorems \ref{theorem:smcp}-\ref{theorem:smscp},, i.e. $|t^*_i - t^*_{i+1}| = a_T\log T$ for some $a_T\to\infty$. Suppose components $i$ and $i+1$ of MICH identify $t^*_i$ and $t^*_{i+1}$ respectively. Given a localization error $\epsilon_T = \mathcal{O}(\log T)$, for large enough $T$ we have $\mathbb{B}_{\epsilon_T}^{1}(t^*_i) \cap \mathbb{B}_{\epsilon_T}^{1}(t^*_{i+1}) = \emptyset$. In the proof of Corollary \ref{cor:cred-sets}, we show that $\overline{\pi}_{it} \leq T^{-2}$ with high probability for $t \not\in \mathbb{B}_{\epsilon_T}^{1}(t^*_i)$. Though this result is for the single change-point setting, in empirical results for the multiple change-point setting, we observe a similar rate of decay for the posterior probabilities that fall outside of the window defined by the localization error. Thus, with high probability:
\begin{align}
    \sum_{t \in \mathbb{B}_{\epsilon_{T}}^{1}(t^*_{i+1})} \overline{\pi}_{it} \overline{\pi}_{(i+1)t} \leq \frac{2\epsilon_T}{T^2}
\end{align}
A symmetric argument gives an identical bound for $t \in \mathbb{B}_{\epsilon_T}^{1}(t^*_{i})$, and for the remaining indices we have $\overline{\pi}_{it} \overline{\pi}_{(i+1)t} \leq T^{-4}$; therefore, the merge probability in this case is at most $\mathcal{O}(T^{-2}\log T)$, so any $\beta \gg \log T/T^2$ should avoid incorrectly merging components. Recall that for some choice of small $\delta >0$, we say that the $i$\textsuperscript{th} estimator $\hat{\tau}_{\text{MAP},i}$ as defined in (\ref{eq:map}) detects a change-point if $|\mathcal{CS}(\alpha,\overline{\boldsymbol{\pi}}_{i,1:T})|\leq \log^{1+\delta}T$, where the credible set $\mathcal{CS}(\alpha,\overline{\boldsymbol{\pi}}_{i,1:T})$ is defined as in (\ref{eq:cs}). If we set $\beta = \log^{1+\delta} T /T^2$, then we should avoid merging any truly distinct components with high probability as $T \to \infty$. For the case we observe in Figure \ref{fig:duplicate}, the merge probability $\langle \overline{\boldsymbol{\pi}}_{2,1:T},\overline{\boldsymbol{\pi}}_{3,1:T}\rangle \approx 0.044$ far exceeds $ \log^{1+\delta} T /T^2$ when $T = 1000$ and $\delta =1.01$. Therefore, our merge procedure combines $\hat{\tau}_2$ and $\hat{\tau}_3$. The resulting fit is now able to capture the change at $\tau_{10}$, as seen in Figure \ref{fig:merge}.

\begin{figure}[!h]
    \centering
    \includegraphics[scale=0.25]{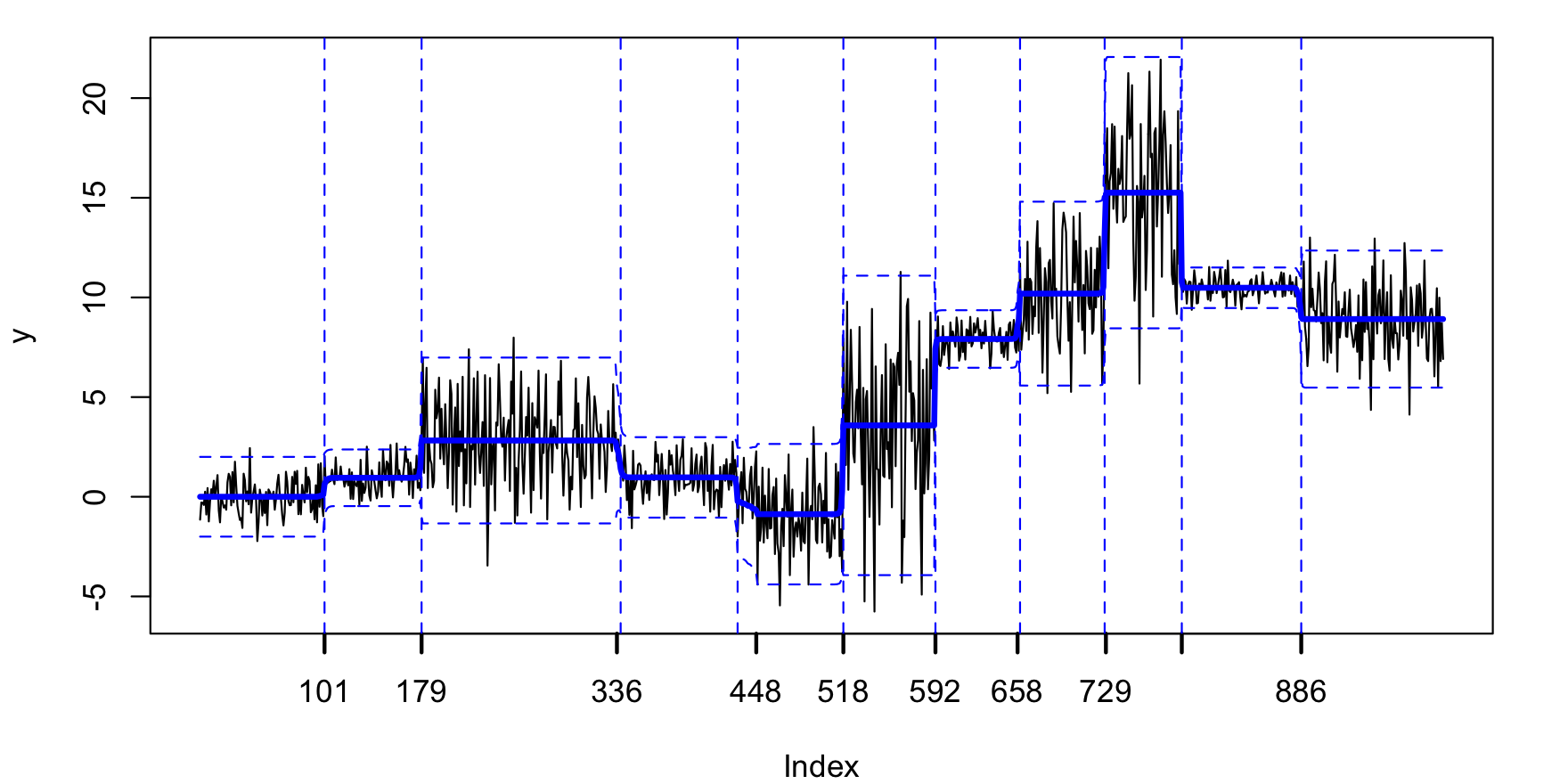}
    \caption{\textbf{MICH with Merging}. Fit of $\mathbf{y}_{1:T}$ (black line) from Figure \ref{fig:duplicate} using merge procedure with $\beta= \log^{1+\delta} T /T^2$ and $J^* = J = 10$. Plot shows true change points $\boldsymbol{\tau}_{1:J^*}$ (x-axis ticks), estimated changes $\hat{\boldsymbol{\tau}}_{1:\hat{J}}$ (vertical dashed blue lines), and estimated mean (solid blue line) and variance (dashed blue lines) signals.}
    \label{fig:merge}
\end{figure}

The merge procedure we have proposed may still run into issues when the model includes redundant components. As previously noted, when component $i$ does not capture a change-point, then $\overline{\boldsymbol{\pi}}_{i,1:T}$ tends to be diffuse with $\overline{\pi}_{it} \approx T^{-1}$. If component $i'$ does correspond to a change and $\overline{\boldsymbol{\pi}}_{i',1:T}$ is sparse, then we may end up with a situation where $\langle\overline{\boldsymbol{\pi}}_{i,1:T}, \overline{\boldsymbol{\pi}}_{i',1:T}\rangle \approx T^{-1}$. Under the default value of $\beta$ we have proposed, components $i$ and $i'$ will be erroneously merged. To avoid this outcome, we only merge components that we believe detect true change-points, i.e. we restrict merge candidates to only include the model components with $\alpha$-level credible sets that satisfy the detection rule with $\alpha = 0.9$ by default. The procedure tends to be insensitive to the choice of $\alpha$, so long as it is moderately large. For example, when $T = 100$, we have $\log^2 T \approx 25$, so for a diffuse $\overline{\boldsymbol{\pi}}_{i,1:T}$, we would rule component $i$ out as a merge candidate for any $\alpha \geq 0.25$.
\subsection{Choice of $L$, $K$, and $J$}
\label{app:LKJ-choice}

As in Section \ref{sec:LKJ}, assume that $L^*, K^*, J^* \geq 0$ are the true numbers of mean-only, variance-only, and joint mean and variance changes in $\mathbf{y}_{1:T}$. In the standard case where we only consider a single class of change-points, e.g. when $J^* > 0$ and $L^*=K^*=0$, we automatically select $J$ by starting from the null model with $J=0$. We then increment $J$ until the ELBO stops increasing. For finite samples, the ELBO does not necessarily increase monotonically on $[J^*]$, so we continue the search for an additional $\log T$ steps after local maximum is found. To increase the speed of the search, we can use the fitted parameters from the $J-1$ component model to initialize Algorithm \ref{alg:mich} after incrementing $J$. We can also restart the Algorithm \ref{alg:mich} from the null initialization after the ELBO decreases to prevent this incremental search procedure from leading to a suboptimal stationary point. To search over $L$, $K$, and $J$ simultaneously, we individually increment $L$, $K$, and $J$ and pick the direction that leads to the largest increase in $\text{ELBO}(q)$.

\subsection{ELBO and Convergence Criterion}
\label{app:convergence}

Because Algorithms \ref{alg:mich}-\ref{alg:mich-multi} define coordinate ascent procedures for maximizing the ELBO, then the ELBO is strictly increases after each iteration. Therefore, one possibility for a stopping rule is to terminate the procedure once the absolute or percentage increase in the ELBO falls below some error tolerance $\epsilon > 0$. Propositions \ref{prop:vb} and \ref{prop:vb-multi} show how to calculate the ELBO as a function $F$ of the posterior parameters (see (\ref{eq:F})) and guaranty the convergence of these parameters. Thus, the ELBO is guaranteed to converge as well. In practice, we set $\epsilon = 10^{-5}$ and terminate when the percentage increase in the ELBO falls below $100\epsilon\%$. See Appendix \ref{app:tol_sensitivity} for a sensitivity analysis for the choice of $\epsilon$. 

\newpage
\section{Simulations}
\label{app:simulations}

\subsection{Evaluation Metrics}

For a set of change-points $\boldsymbol{\tau}_{1:J^*}$ and their estimates $\hat{\boldsymbol{\tau}}_{1:\hat{J}}$, the Hausdorff distance is given by:
\begin{align}
    d_H(\hat{\boldsymbol{\tau}}_{1:\hat{J}}, \boldsymbol{\tau}_{1:J^*}) : =  \max_{0\leq j\leq J^*+1} \min_{0\leq i\leq\hat{J}+1} |\tau_j -\hat{\tau}_i| + \max_{0\leq j\leq \hat{J}+1} \min_{0\leq i\leq J^*+1} |\hat{\tau}_j -\tau_i|. \label{eq:hausdorff}
\end{align}
\cite{Futschik14} specifies the false positive sensitive location error (FPSLE) and false negative sensitive location error (FNSLE) as:
\begin{align}
    d_{\text{FPSLE}}(\hat{\boldsymbol{\tau}}_{1:\hat{J}}\lVert \boldsymbol{\tau}_{1:J}^*) &:= \frac{1}{2(\hat{J}+1)} \sum_{j=1}^{\hat{J}+1} |\hat{\tau}_{j-1} - \tau_{i_j - 1}| + |\hat{\tau}_{j} - \tau_{i_j}|, \label{eq:fpsle}\\
    \{i_j\}_{j=1}^{J^*+1} &:=\left\{i\in[J^*+1] : \tau_{i_j-1} < (\hat{\tau}_{j-1}+\hat{\tau}_{j})/2 \leq \tau_{i_j} \sforall j \in [J^*+1]\right\}. \notag \\
    d_{\text{FNSLE}}(\hat{\boldsymbol{\tau}}_{1:\hat{J}}\lVert \boldsymbol{\tau}_{1:J^*}) &:= d_{\text{FPSLE}}(\boldsymbol{\tau}_{1:J^*}\lVert \hat{\boldsymbol{\tau}}_{1:J^*}) \label{eq:fnsle}
\end{align}
Lastly, we calculate the coverage conditional on detection (CCD) by identifying every $\tau^{(m)}_i$ in the $m \in [M]$ replicated of the simulation that is within a $w:=\min\{T^{1/2}/2, 15\}$ sized window of some estimated $\hat{\tau}^{(m)}_j$. We then collect the $\alpha$ level credible sets $\mathcal{CS}(\alpha, \hat{\tau}^{(m)}_j)$ for each identified $\hat{\tau}^{(m)}_j$ in the window $w$ and determine what proportion cover their respective true changes $\tau^{(m)}_i$:
\begin{align}
    \text{CCD} := \frac{\sum_{m=1}^M\sum_{j=1}^{J^*} \mathbbm{1}\{\exists \; i \text{ s.t. } |\hat{\tau}^{(m)}_i - \tau^{(m)}_j| \leq w \text{ and } \tau^{(m)}_j \in \mathcal{CS}(\hat{\tau}^{(m)}_i)\}}{\sum_{m=1}^M\sum_{j=1}^{J^*} \mathbbm{1}\{\exists \; i \text{ s.t. } |\hat{\tau}^{(m)}_i - \tau^{(m)}_j| \leq w\}}. \label{eq:ccd}
\end{align}

\subsection{Software}
\label{app:simulations_software}

We fit H-SMUCE for $\alpha \in \{0.1,0.5\}$ using the \texttt{stepFit} function from the $\textbf{stepR}$ package\footnote{\url{https://cran.r-project.org/package=stepR}.} with the default settings. For the MOSUM procedure, we set the significance level $\alpha = 0.1$ and use the \texttt{multiscale.bottomUp} and \texttt{multiscale.localPrune} functions from the \textbf{mosum} package\footnote{\url{https://cran.r-project.org/package=mosum}.} to fit the model and automatically select the bandwidth parameter using the respective bottom-up merging and local pruning procedures described in \cite{Cho22}. As per the guidance in \cite{Meier21}, we set \texttt{var.est.method = "mosum.min"} to account for joint mean and variance changes. We fit PELT using the \texttt{cpt.meanvar} function with the default values from the \textbf{changepoint} package\footnote{\url{https://cran.r-project.org/package=changepoint}.} and fit NOT using the \texttt{not} function with from the \textbf{not} package\footnote{\url{https://cran.r-project.org/package=not}.} with \texttt{contrast = "pcwsConstMeanVar"}. We fit NSP using the \texttt{nsp\_selfnorm} function from the \textbf{nsp} package\footnote{\url{https://cran.r-project.org/package=nsp}.} with $M=$ 1,000 intervals and a detection threshold $\lambda_\alpha$ equal to the $\alpha =0.1$ threshold specified for the heterogeneous noise setting in Section 5.3 of \cite{Fryzlewicz24}. To calculate the conditional coverage as in (\ref{eq:ccd}), we use the midpoints of the sets returned by NSP in place of $\hat{\boldsymbol{\tau}}_{1:\hat{J}}$.

\subsection{Complete Simulation \ref{sim:main} Results}
\label{app:main_sim_results}

\begin{table}[htbp!] 
    \scriptsize
    \centering
    \begin{tabular}{l || l || r r r r | r r | r}
        \multicolumn{1}{c||}{Setting} & \multicolumn{1}{c||}{Method} & $|J^* - J|$ & Hausdorff & FPSLE & FNSLE & CI Len. & CCD & Time (s) \\ \hline\hline 
        $T =$ 100 & MICH-Auto & \textsuperscript{\dag\ddag}0.052 & \textsuperscript{\ddag}1.015 & \textsuperscript{\dag\ddag}0.308 & \textsuperscript{\ddag}0.201 & 1.482 & 0.972 & 0.250\\
        $J^* =$ 2 & MICH-Ora & \textsuperscript{*}0.000 & \textsuperscript{*}0.510 & \textsuperscript{*}0.089 & \textsuperscript{*}0.092 & 1.354 & 0.972 & 0.010\\
        Min. Space = 15 & H-SMUCE (0.1) & \textsuperscript{}0.134 & \textsuperscript{}2.981 & \textsuperscript{}0.721 & \textsuperscript{}1.439 & 23.315 & 0.966 & 0.188\\
         & H-SMUCE (0.5) & \textsuperscript{}0.059 & \textsuperscript{}1.644 & \textsuperscript{}0.460 & \textsuperscript{}0.412 & 16.192 & 0.986 & 0.061\\
         & MOSUM BUM & \textsuperscript{}0.083 & \textsuperscript{}2.940 & \textsuperscript{}0.790 & \textsuperscript{}0.617 & 12.043 & 0.956 & 0.018\\
         & MOSUM LP & \textsuperscript{}0.144 & \textsuperscript{}2.083 & \textsuperscript{}0.751 & \textsuperscript{}0.402 & 6.532 & 0.988 & 0.016\\
         & PELT & \textsuperscript{}0.105 & \textsuperscript{\dag}0.832 & \textsuperscript{}0.494 & \textsuperscript{\dag}0.148 &  &  & 0.004\\
         & NOT & \textsuperscript{}0.073 & \textsuperscript{}1.082 & \textsuperscript{}0.401 & \textsuperscript{}0.195 &  &  & 0.051\\
         & NSP & \textsuperscript{}0.712 & \textsuperscript{} & \textsuperscript{} & \textsuperscript{} & 36.062 & 1.000 & 3.661\\ \hline
        $T =$ 100 & MICH-Auto & \textsuperscript{\ddag}0.111 & \textsuperscript{\ddag}1.472 & \textsuperscript{\ddag}0.224 & \textsuperscript{\ddag}0.337 & 1.106 & 0.988 & 0.437\\
        $J^* =$ 5 & MICH-Ora & \textsuperscript{*}0.005 & \textsuperscript{}0.798 & \textsuperscript{}0.090 & \textsuperscript{}0.108 & 1.095 & 0.988 & 0.089\\
        Min. Space = 15 & H-SMUCE (0.1) & \textsuperscript{}1.643 & \textsuperscript{}17.449 & \textsuperscript{}3.318 & \textsuperscript{}5.469 & 15.405 & 0.727 & 0.022\\
         & H-SMUCE (0.5) & \textsuperscript{}0.576 & \textsuperscript{}8.520 & \textsuperscript{}1.064 & \textsuperscript{}1.879 & 12.792 & 0.901 & 0.061\\
         & MOSUM BUM & \textsuperscript{}0.922 & \textsuperscript{}13.416 & \textsuperscript{}1.881 & \textsuperscript{}2.892 & 24.490 & 0.998 & 0.021\\
         & MOSUM LP & \textsuperscript{}0.278 & \textsuperscript{}3.050 & \textsuperscript{}0.484 & \textsuperscript{}0.498 & 13.182 & 0.998 & 0.024\\
         & PELT & \textsuperscript{}0.182 & \textsuperscript{}0.773 & \textsuperscript{}0.236 & \textsuperscript{}0.072 &  &  & 0.003\\
         & NOT & \textsuperscript{\dag}0.043 & \textsuperscript{*\dag}0.589 & \textsuperscript{*\dag}0.083 & \textsuperscript{*\dag}0.055 &  &  & 0.050\\
         & NSP & \textsuperscript{}3.462 & \textsuperscript{} & \textsuperscript{} & \textsuperscript{} & 35.079 & 1.000 & 3.676\\ \hline \hline
    \end{tabular}
    \caption{\small \textbf{Simulation \ref{sim:main} Results}. Average statistics over 5,000 replicates from Simulation \ref{sim:main}. Bias $|J^* - J|$ assesses each model's ability to estimate correct number of changes (the lower the better). FPSLE and FNSLE assess each model's ability to accurately estimate the locations of the changes (the lower the better). The average confidence/credible set length (CI Len.) and coverage conditional on detection (CCD) summarize each method's ability to provide uncertainty quantification. MICH credible sets are constructed for $\alpha = 0.1$ and detections are determined with $\delta = 0.5$. Key: \textsuperscript{*}Best overall. \textsuperscript{\textdagger}Best excluding MICH-Ora. \textsuperscript{\ddag}Best among methods with uncertainty quantification excluding MICH-Ora.}
    \label{tab:main_sim_1}
\end{table}

\begin{table}[htbp!] 
    \scriptsize
    \centering
    \begin{tabular}{l || l || r r r r | r r | r}
    \multicolumn{1}{c||}{Setting} & \multicolumn{1}{c||}{Method} & $|J^* - J|$ & Hausdorff & FPSLE & FNSLE & CI Len. & CCD & Time (s) \\ \hline\hline 
    $T =$ 500 & MICH-Auto & \textsuperscript{\ddag}0.043 & \textsuperscript{\ddag}4.194 & \textsuperscript{\ddag}1.347 & \textsuperscript{\ddag}1.723 & 3.445 & 0.940 & 1.562\\
    $J^* =$ 2 & MICH-Ora & \textsuperscript{*}0.005 & \textsuperscript{}2.864 & \textsuperscript{*}0.626 & \textsuperscript{}0.836 & 3.313 & 0.942 & 0.066\\
    Min. Space = 15 & H-SMUCE (0.1) & \textsuperscript{}0.060 & \textsuperscript{}4.916 & \textsuperscript{}1.376 & \textsuperscript{}2.561 & 53.617 & 0.977 & 0.036\\
     & H-SMUCE (0.5) & \textsuperscript{}0.163 & \textsuperscript{}15.251 & \textsuperscript{}6.029 & \textsuperscript{}3.176 & 37.655 & 0.982 & 0.070\\
     & MOSUM BUM & \textsuperscript{}0.407 & \textsuperscript{}25.061 & \textsuperscript{}11.913 & \textsuperscript{}5.050 & 25.435 & 0.976 & 0.031\\
     & MOSUM LP & \textsuperscript{}0.460 & \textsuperscript{}21.648 & \textsuperscript{}10.404 & \textsuperscript{}5.067 & 17.821 & 0.987 & 0.033\\
     & PELT & \textsuperscript{}0.024 & \textsuperscript{*\dag}2.784 & \textsuperscript{}1.035 & \textsuperscript{*\dag}0.522 &  &  & 0.009\\
     & NOT & \textsuperscript{\dag}0.018 & \textsuperscript{}3.070 & \textsuperscript{\dag}0.907 & \textsuperscript{}0.579 &  &  & 0.139\\
     & NSP & \textsuperscript{}0.223 & \textsuperscript{} & \textsuperscript{} & \textsuperscript{} & 91.648 & 1.000 & 10.447\\ \hline
    $T =$ 500 & MICH-Auto & \textsuperscript{}0.031 & \textsuperscript{\ddag}3.896 & \textsuperscript{\ddag}1.160 & \textsuperscript{\ddag}0.994 & 3.588 & 0.941 & 1.548\\
    $J^* =$ 2 & MICH-Ora & \textsuperscript{*}0.001 & \textsuperscript{*}2.930 & \textsuperscript{*}0.579 & \textsuperscript{}0.614 & 3.466 & 0.942 & 0.061\\
    Min. Space = 30 & H-SMUCE (0.1) & \textsuperscript{\ddag}0.024 & \textsuperscript{}4.596 & \textsuperscript{}1.198 & \textsuperscript{}1.167 & 54.803 & 0.996 & 0.037\\
     & H-SMUCE (0.5) & \textsuperscript{}0.170 & \textsuperscript{}16.009 & \textsuperscript{}6.249 & \textsuperscript{}2.829 & 37.796 & 0.991 & 0.070\\
     & MOSUM BUM & \textsuperscript{}0.411 & \textsuperscript{}24.442 & \textsuperscript{}11.231 & \textsuperscript{}4.716 & 24.985 & 0.980 & 0.031\\
     & MOSUM LP & \textsuperscript{}0.451 & \textsuperscript{}21.891 & \textsuperscript{}10.071 & \textsuperscript{}5.140 & 18.404 & 0.987 & 0.033\\
     & PELT & \textsuperscript{}0.022 & \textsuperscript{\dag}2.974 & \textsuperscript{}1.004 & \textsuperscript{*\dag}0.562 &  &  & 0.009\\
     & NOT & \textsuperscript{\dag}0.016 & \textsuperscript{}3.166 & \textsuperscript{\dag}0.908 & \textsuperscript{}0.595 &  &  & 0.139\\
     & NSP & \textsuperscript{}0.141 & \textsuperscript{} & \textsuperscript{} & \textsuperscript{} & 92.570 & 1.000 & 10.679\\ \hline
    $T =$ 500 & MICH-Auto & \textsuperscript{\ddag}0.145 & \textsuperscript{\ddag}7.281 & \textsuperscript{\ddag}1.070 & \textsuperscript{}2.064 & 2.253 & 0.949 & 2.777\\
    $J^* =$ 5 & MICH-Ora & \textsuperscript{}0.040 & \textsuperscript{}4.999 & \textsuperscript{}0.648 & \textsuperscript{}1.039 & 2.197 & 0.950 & 0.541\\
    Min. Space = 15 & H-SMUCE (0.1) & \textsuperscript{}0.468 & \textsuperscript{}14.182 & \textsuperscript{}1.910 & \textsuperscript{}5.083 & 39.884 & 0.916 & 0.038\\
     & H-SMUCE (0.5) & \textsuperscript{}0.193 & \textsuperscript{}9.794 & \textsuperscript{}1.712 & \textsuperscript{\ddag}1.988 & 28.567 & 0.956 & 0.070\\
     & MOSUM BUM & \textsuperscript{}0.459 & \textsuperscript{}18.986 & \textsuperscript{}4.921 & \textsuperscript{}2.569 & 30.530 & 0.990 & 0.055\\
     & MOSUM LP & \textsuperscript{}0.581 & \textsuperscript{}23.289 & \textsuperscript{}5.017 & \textsuperscript{}4.664 & 16.407 & 0.995 & 0.044\\
     & PELT & \textsuperscript{}0.039 & \textsuperscript{*\dag}3.052 & \textsuperscript{}0.584 & \textsuperscript{}0.364 &  &  & 0.009\\
     & NOT & \textsuperscript{*\dag}0.018 & \textsuperscript{}3.124 & \textsuperscript{*\dag}0.447 & \textsuperscript{*\dag}0.346 &  &  & 0.143\\
     & NSP & \textsuperscript{}1.694 & \textsuperscript{} & \textsuperscript{} & \textsuperscript{} & 75.884 & 1.000 & 13.599\\ \hline
    $T =$ 500 & MICH-Auto & \textsuperscript{}0.087 & \textsuperscript{\ddag}6.261 & \textsuperscript{\ddag}0.886 & \textsuperscript{}1.203 & 2.387 & 0.943 & 2.597\\
    $J^* =$ 5 & MICH-Ora & \textsuperscript{}0.021 & \textsuperscript{}5.132 & \textsuperscript{}0.638 & \textsuperscript{}0.836 & 2.320 & 0.945 & 0.488\\
    Min. Space = 30 & H-SMUCE (0.1) & \textsuperscript{}0.158 & \textsuperscript{}9.345 & \textsuperscript{}1.159 & \textsuperscript{}2.134 & 41.832 & 0.978 & 0.038\\
     & H-SMUCE (0.5) & \textsuperscript{\ddag}0.077 & \textsuperscript{}6.743 & \textsuperscript{}1.129 & \textsuperscript{\ddag}0.746 & 28.893 & 0.994 & 0.070\\
     & MOSUM BUM & \textsuperscript{}0.323 & \textsuperscript{}13.757 & \textsuperscript{}3.184 & \textsuperscript{}1.510 & 24.593 & 0.993 & 0.056\\
     & MOSUM LP & \textsuperscript{}0.571 & \textsuperscript{}24.485 & \textsuperscript{}4.817 & \textsuperscript{}4.754 & 16.955 & 0.996 & 0.045\\
     & PELT & \textsuperscript{}0.038 & \textsuperscript{*\dag}3.211 & \textsuperscript{}0.592 & \textsuperscript{*\dag}0.365 &  &  & 0.009\\
     & NOT & \textsuperscript{*\dag}0.017 & \textsuperscript{}3.314 & \textsuperscript{*\dag}0.466 & \textsuperscript{}0.373 &  &  & 0.144\\
     & NSP & \textsuperscript{}1.416 & \textsuperscript{} & \textsuperscript{} & \textsuperscript{} & 74.979 & 1.000 & 14.193\\ \hline
    $T =$ 500 & MICH-Auto & \textsuperscript{}0.638 & \textsuperscript{\ddag}15.118 & \textsuperscript{\ddag}1.225 & \textsuperscript{}2.839 & 1.681 & 0.962 & 5.588\\
    $J^* =$ 10 & MICH-Ora & \textsuperscript{}0.197 & \textsuperscript{}10.309 & \textsuperscript{}0.811 & \textsuperscript{}1.265 & 1.707 & 0.963 & 2.383\\
    Min. Space = 15 & H-SMUCE (0.1) & \textsuperscript{}2.051 & \textsuperscript{}31.127 & \textsuperscript{}3.503 & \textsuperscript{}7.765 & 29.709 & 0.816 & 0.038\\
     & H-SMUCE (0.5) & \textsuperscript{}0.809 & \textsuperscript{}15.602 & \textsuperscript{}1.323 & \textsuperscript{}2.814 & 23.366 & 0.892 & 0.071\\
     & MOSUM BUM & \textsuperscript{\ddag}0.546 & \textsuperscript{}20.429 & \textsuperscript{}2.552 & \textsuperscript{\ddag}2.230 & 42.849 & 0.999 & 0.084\\
     & MOSUM LP & \textsuperscript{}0.819 & \textsuperscript{}24.649 & \textsuperscript{}2.617 & \textsuperscript{}3.600 & 16.647 & 0.998 & 0.064\\
     & PELT & \textsuperscript{}0.064 & \textsuperscript{*\dag}2.729 & \textsuperscript{}0.300 & \textsuperscript{}0.198 &  &  & 0.009\\
     & NOT & \textsuperscript{*\dag}0.015 & \textsuperscript{}2.813 & \textsuperscript{*\dag}0.218 & \textsuperscript{*\dag}0.197 &  &  & 0.144\\
     & NSP & \textsuperscript{}5.520 & \textsuperscript{} & \textsuperscript{} & \textsuperscript{} & 66.036 & 1.000 & 15.708\\ \hline
    $T =$ 500 & MICH-Auto & \textsuperscript{}0.296 & \textsuperscript{}10.875 & \textsuperscript{}0.868 & \textsuperscript{}1.468 & 1.697 & 0.955 & 5.287\\
    $J^* =$ 10 & MICH-Ora & \textsuperscript{}0.079 & \textsuperscript{}8.373 & \textsuperscript{}0.620 & \textsuperscript{}0.839 & 1.687 & 0.956 & 2.025\\
    Min. Space = 30 & H-SMUCE (0.1) & \textsuperscript{}1.026 & \textsuperscript{}29.855 & \textsuperscript{}2.499 & \textsuperscript{}4.651 & 31.285 & 0.918 & 0.039\\
     & H-SMUCE (0.5) & \textsuperscript{\ddag}0.128 & \textsuperscript{\ddag}6.923 & \textsuperscript{\ddag}0.495 & \textsuperscript{\ddag}0.706 & 24.504 & 0.986 & 0.071\\
     & MOSUM BUM & \textsuperscript{}0.194 & \textsuperscript{}9.114 & \textsuperscript{}0.983 & \textsuperscript{}0.736 & 41.213 & 1.000 & 0.093\\
     & MOSUM LP & \textsuperscript{}0.769 & \textsuperscript{}26.445 & \textsuperscript{}2.559 & \textsuperscript{}3.398 & 16.852 & 0.999 & 0.065\\
     & PELT & \textsuperscript{}0.060 & \textsuperscript{*\dag}2.839 & \textsuperscript{}0.306 & \textsuperscript{*\dag}0.206 &  &  & 0.009\\
     & NOT & \textsuperscript{*\dag}0.013 & \textsuperscript{}2.897 & \textsuperscript{*\dag}0.230 & \textsuperscript{}0.208 &  &  & 0.145\\
     & NSP & \textsuperscript{}4.968 & \textsuperscript{} & \textsuperscript{} & \textsuperscript{} & 65.723 & 1.000 & 16.480\\ \hline \hline
    \end{tabular}
    \caption{\small \textbf{Simulation \ref{sim:main} Results (Continued from Table \ref{tab:main_sim_1})}. Key: \textsuperscript{*}Best overall. \textsuperscript{\textdagger}Best excluding MICH-Ora. \textsuperscript{\ddag}Best among methods with uncertainty quantification excluding MICH-Ora.}
    \label{tab:main_sim_2}
\end{table}

\begin{table}[htbp!] 
    \scriptsize
    \centering
    \begin{tabular}{l || l || r r r r | r r | r}
    \multicolumn{1}{c||}{Setting} & \multicolumn{1}{c||}{Method} & $|J^* - J|$ & Hausdorff & FPSLE & FNSLE & CI Len. & CCD & Time (s) \\ \hline\hline 
    $T =$ 1000 & MICH-Auto & \textsuperscript{}0.042 & \textsuperscript{}8.421 & \textsuperscript{}2.826 & \textsuperscript{}3.106 & 5.229 & 0.933 & 2.668\\
    $J^* =$ 2 & MICH-Ora & \textsuperscript{*}0.004 & \textsuperscript{}5.933 & \textsuperscript{*}1.258 & \textsuperscript{}1.512 & 5.097 & 0.934 & 0.146\\
    Min. Space = 30 & H-SMUCE (0.1) & \textsuperscript{\ddag}0.028 & \textsuperscript{\ddag}8.031 & \textsuperscript{\ddag}2.323 & \textsuperscript{\ddag}2.121 & 89.516 & 0.995 & 0.047\\
     & H-SMUCE (0.5) & \textsuperscript{}0.182 & \textsuperscript{}31.377 & \textsuperscript{}12.809 & \textsuperscript{}5.619 & 63.084 & 0.993 & 0.078\\
     & MOSUM BUM & \textsuperscript{}0.197 & \textsuperscript{}34.276 & \textsuperscript{}13.073 & \textsuperscript{}8.990 & 41.426 & 0.982 & 0.059\\
     & MOSUM LP & \textsuperscript{}0.658 & \textsuperscript{}53.709 & \textsuperscript{}27.190 & \textsuperscript{}13.506 & 25.985 & 0.978 & 0.052\\
     & PELT & \textsuperscript{}0.011 & \textsuperscript{*\dag}4.885 & \textsuperscript{}1.378 & \textsuperscript{*\dag}0.930 &  &  & 0.017\\
     & NOT & \textsuperscript{\dag}0.007 & \textsuperscript{}5.032 & \textsuperscript{\dag}1.309 & \textsuperscript{}0.957 &  &  & 0.244\\
     & NSP & \textsuperscript{}0.120 & \textsuperscript{} & \textsuperscript{} & \textsuperscript{} & 142.515 & 1.000 & 16.688\\ \hline
    $T =$ 1000 & MICH-Auto & \textsuperscript{}0.036 & \textsuperscript{}8.061 & \textsuperscript{}2.730 & \textsuperscript{}2.577 & 5.565 & 0.924 & 2.680\\
    $J^* =$ 2 & MICH-Ora & \textsuperscript{*}0.003 & \textsuperscript{}5.754 & \textsuperscript{*}1.221 & \textsuperscript{}1.394 & 5.385 & 0.925 & 0.145\\
    Min. Space = 50 & H-SMUCE (0.1) & \textsuperscript{\ddag}0.016 & \textsuperscript{\ddag}7.575 & \textsuperscript{\ddag}2.045 & \textsuperscript{\ddag}1.475 & 90.608 & 0.999 & 0.047\\
     & H-SMUCE (0.5) & \textsuperscript{}0.182 & \textsuperscript{}31.610 & \textsuperscript{}12.730 & \textsuperscript{}5.594 & 62.755 & 0.995 & 0.078\\
     & MOSUM BUM & \textsuperscript{}0.214 & \textsuperscript{}35.690 & \textsuperscript{}13.667 & \textsuperscript{}8.758 & 41.798 & 0.989 & 0.060\\
     & MOSUM LP & \textsuperscript{}0.679 & \textsuperscript{}54.869 & \textsuperscript{}27.766 & \textsuperscript{}13.830 & 26.994 & 0.982 & 0.053\\
     & PELT & \textsuperscript{}0.014 & \textsuperscript{*\dag}5.247 & \textsuperscript{}1.522 & \textsuperscript{*\dag}1.006 &  &  & 0.017\\
     & NOT & \textsuperscript{\dag}0.008 & \textsuperscript{}5.404 & \textsuperscript{\dag}1.340 & \textsuperscript{}1.032 &  &  & 0.244\\
     & NSP & \textsuperscript{}0.081 & \textsuperscript{} & \textsuperscript{} & \textsuperscript{} & 142.772 & 1.000 & 16.833\\ \hline
    $T =$ 1000 & MICH-Auto & \textsuperscript{}0.133 & \textsuperscript{}12.730 & \textsuperscript{}2.045 & \textsuperscript{}3.331 & 3.529 & 0.939 & 5.172\\
    $J^* =$ 5 & MICH-Ora & \textsuperscript{}0.059 & \textsuperscript{}10.500 & \textsuperscript{}1.417 & \textsuperscript{}2.516 & 3.432 & 0.94 & 1.184\\
    Min. Space = 30 & H-SMUCE (0.1) & \textsuperscript{}0.118 & \textsuperscript{\ddag}11.494 & \textsuperscript{\ddag}1.563 & \textsuperscript{}2.819 & 63.992 & 0.981 & 0.048\\
     & H-SMUCE (0.5) & \textsuperscript{\ddag}0.113 & \textsuperscript{}15.943 & \textsuperscript{}3.181 & \textsuperscript{\ddag}1.697 & 43.95 & 0.994 & 0.078\\
     & MOSUM BUM & \textsuperscript{}0.210 & \textsuperscript{}24.164 & \textsuperscript{}4.614 & \textsuperscript{}4.365 & 54.762 & 0.987 & 0.117\\
     & MOSUM LP & \textsuperscript{}0.790 & \textsuperscript{}58.547 & \textsuperscript{}13.451 & \textsuperscript{}12.453 & 23.49 & 0.991 & 0.067\\
     & PELT & \textsuperscript{}0.024 & \textsuperscript{*\dag}5.395 & \textsuperscript{}0.968 & \textsuperscript{}0.669 &  &  & 0.017\\
     & NOT & \textsuperscript{*\dag}0.009 & \textsuperscript{}5.463 & \textsuperscript{*\dag}0.763 & \textsuperscript{*\dag}0.658 &  &  & 0.251\\
     & NSP & \textsuperscript{}1.035 & \textsuperscript{} & \textsuperscript{} & \textsuperscript{} & 111.56 & 1 & 21.27\\ \hline
    $T =$ 1000 & MICH-Auto & \textsuperscript{}0.094 & \textsuperscript{}10.987 & \textsuperscript{}1.850 & \textsuperscript{}2.263 & 3.683 & 0.933 & 4.844\\
    $J^* =$ 5 & MICH-Ora & \textsuperscript{}0.039 & \textsuperscript{}10.076 & \textsuperscript{}1.314 & \textsuperscript{}1.956 & 3.603 & 0.935 & 1.092\\
    Min. Space = 50 & H-SMUCE (0.1) & \textsuperscript{\ddag}0.043 & \textsuperscript{\ddag}9.224 & \textsuperscript{\ddag}1.209 & \textsuperscript{}1.512 & 64.677 & 0.995 & 0.048\\
     & H-SMUCE (0.5) & \textsuperscript{}0.102 & \textsuperscript{}14.312 & \textsuperscript{}2.718 & \textsuperscript{\ddag}1.489 & 43.846 & 0.995 & 0.08\\
     & MOSUM BUM & \textsuperscript{}0.146 & \textsuperscript{}17.999 & \textsuperscript{}3.562 & \textsuperscript{}2.300 & 49.614 & 0.997 & 0.123\\
     & MOSUM LP & \textsuperscript{}0.787 & \textsuperscript{}59.874 & \textsuperscript{}13.231 & \textsuperscript{}12.152 & 23.97 & 0.991 & 0.068\\
     & PELT & \textsuperscript{}0.026 & \textsuperscript{*\dag}5.672 & \textsuperscript{}0.999 & \textsuperscript{*\dag}0.705 &  &  & 0.017\\
     & NOT & \textsuperscript{*\dag}0.009 & \textsuperscript{}5.793 & \textsuperscript{*\dag}0.821 & \textsuperscript{}0.715 &  &  & 0.252\\
     & NSP & \textsuperscript{}0.783 & \textsuperscript{} & \textsuperscript{} & \textsuperscript{} & 112.256 & 1 & 22.079\\ \hline
    $T =$ 1000 & MICH-Auto & \textsuperscript{}0.411 & \textsuperscript{}21.607 & \textsuperscript{}1.831 & \textsuperscript{}3.539 & 2.539 & 0.942 & 11.42\\
    $J^* =$ 10 & MICH-Ora & \textsuperscript{}0.194 & \textsuperscript{}18.071 & \textsuperscript{}1.457 & \textsuperscript{}2.410 & 2.525 & 0.943 & 4.796\\
    Min. Space = 30 & H-SMUCE (0.1) & \textsuperscript{}0.645 & \textsuperscript{}28.032 & \textsuperscript{}2.228 & \textsuperscript{}4.696 & 47.512 & 0.943 & 0.05\\
     & H-SMUCE (0.5) & \textsuperscript{\ddag}0.140 & \textsuperscript{\ddag}11.707 & \textsuperscript{\ddag}1.078 & \textsuperscript{\ddag}1.136 & 35.031 & 0.988 & 0.08\\
     & MOSUM BUM & \textsuperscript{}0.498 & \textsuperscript{}31.897 & \textsuperscript{}2.930 & \textsuperscript{}4.704 & 81.425 & 0.999 & 0.194\\
     & MOSUM LP & \textsuperscript{}0.998 & \textsuperscript{}57.885 & \textsuperscript{}6.822 & \textsuperscript{}8.713 & 21.921 & 0.997 & 0.089\\
     & PELT & \textsuperscript{}0.034 & \textsuperscript{*\dag}5.093 & \textsuperscript{}0.551 & \textsuperscript{}0.434 &  &  & 0.017\\
     & NOT & \textsuperscript{*\dag}0.008 & \textsuperscript{}5.208 & \textsuperscript{*\dag}0.461 & \textsuperscript{*\dag}0.433 &  &  & 0.252\\
     & NSP & \textsuperscript{}3.734 & \textsuperscript{} & \textsuperscript{} & \textsuperscript{} & 91.141 & 1 & 26.14\\ \hline
    $T =$ 1000 & MICH-Auto & \textsuperscript{}0.252 & \textsuperscript{}17.505 & \textsuperscript{}1.552 & \textsuperscript{}2.251 & 2.665 & 0.936 & 11.365\\
    $J^* =$ 10 & MICH-Ora & \textsuperscript{}0.110 & \textsuperscript{}16.280 & \textsuperscript{}1.268 & \textsuperscript{}1.795 & 2.643 & 0.937 & 4.442\\
    Min. Space = 50 & H-SMUCE (0.1) & \textsuperscript{}0.285 & \textsuperscript{}21.714 & \textsuperscript{}1.583 & \textsuperscript{}2.744 & 49.112 & 0.978 & 0.05\\
     & H-SMUCE (0.5) & \textsuperscript{\ddag}0.053 & \textsuperscript{\ddag}8.219 & \textsuperscript{\ddag}0.767 & \textsuperscript{\ddag}0.666 & 35.455 & 0.997 & 0.081\\
     & MOSUM BUM & \textsuperscript{}0.153 & \textsuperscript{}15.538 & \textsuperscript{}1.454 & \textsuperscript{}1.597 & 83.603 & 1 & 0.212\\
     & MOSUM LP & \textsuperscript{}0.974 & \textsuperscript{}60.501 & \textsuperscript{}6.737 & \textsuperscript{}8.601 & 22.335 & 0.996 & 0.09\\
     & PELT & \textsuperscript{}0.034 & \textsuperscript{*\dag}5.451 & \textsuperscript{}0.602 & \textsuperscript{*\dag}0.471 &  &  & 0.017\\
     & NOT & \textsuperscript{*\dag}0.009 & \textsuperscript{}5.555 & \textsuperscript{*\dag}0.507 & \textsuperscript{}0.479 &  &  & 0.253\\
     & NSP & \textsuperscript{}3.208 & \textsuperscript{} & \textsuperscript{} & \textsuperscript{} & 90.513 & 1 & 27.481\\ \hline\hline
    \end{tabular}
    \caption{\small \textbf{Simulation \ref{sim:main} Results (Continued from Table \ref{tab:main_sim_2})}. Key: \textsuperscript{*}Best overall. \textsuperscript{\textdagger}Best excluding MICH-Ora. \textsuperscript{\ddag}Best among methods with uncertainty quantification excluding MICH-Ora.}
    \label{tab:main_sim_3}
\end{table}

\begin{figure}[!h]
    \centering
    \includegraphics[scale =0.42]{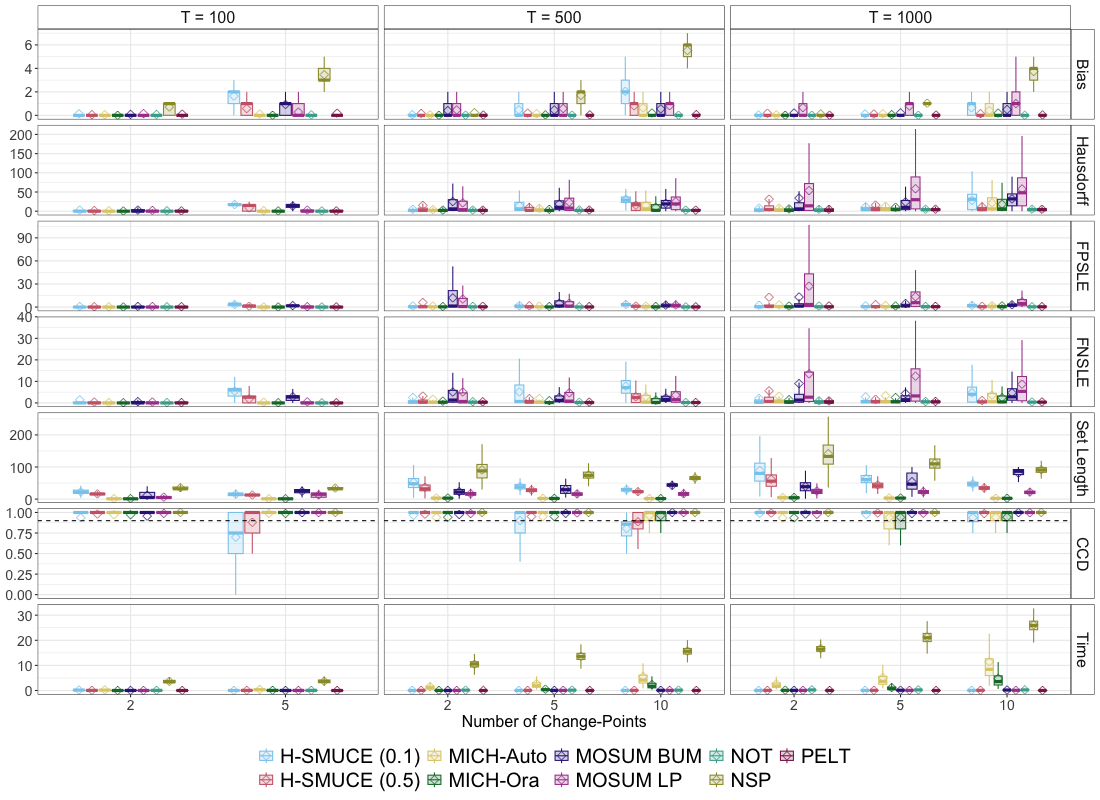}
    \caption{\textbf{Simulation \ref{sim:main} Results -- Small $\Delta_T$.} Box-plots of evaluation metrics for 5,000 replicates of Simulation \ref{sim:main} in small $\Delta_T$ setting with $\Delta_{100} = 15, \Delta_{500} = 15, \Delta_{1000} = 30$. Diamonds ($\Diamond$) display mean of each statistic. Bias $|J^* - J|$ assesses each model's ability to estimate correct number of changes (lower is better). Hausdorff, FPSLE, and FNSLE assess each model's ability to accurately estimate the locations of the changes (lower is better). Set Length and CCD report the average size and coverage of credible/confidence sets for methods that provide uncertainty quantification (dashed line indicates nominal coverage level for $\alpha = 0.1$). Time reports run-time for each method in seconds.}
    \label{fig:low_full_sim_plot}
\end{figure}

\begin{figure}[!h]
    \centering
    \includegraphics[scale =0.42]{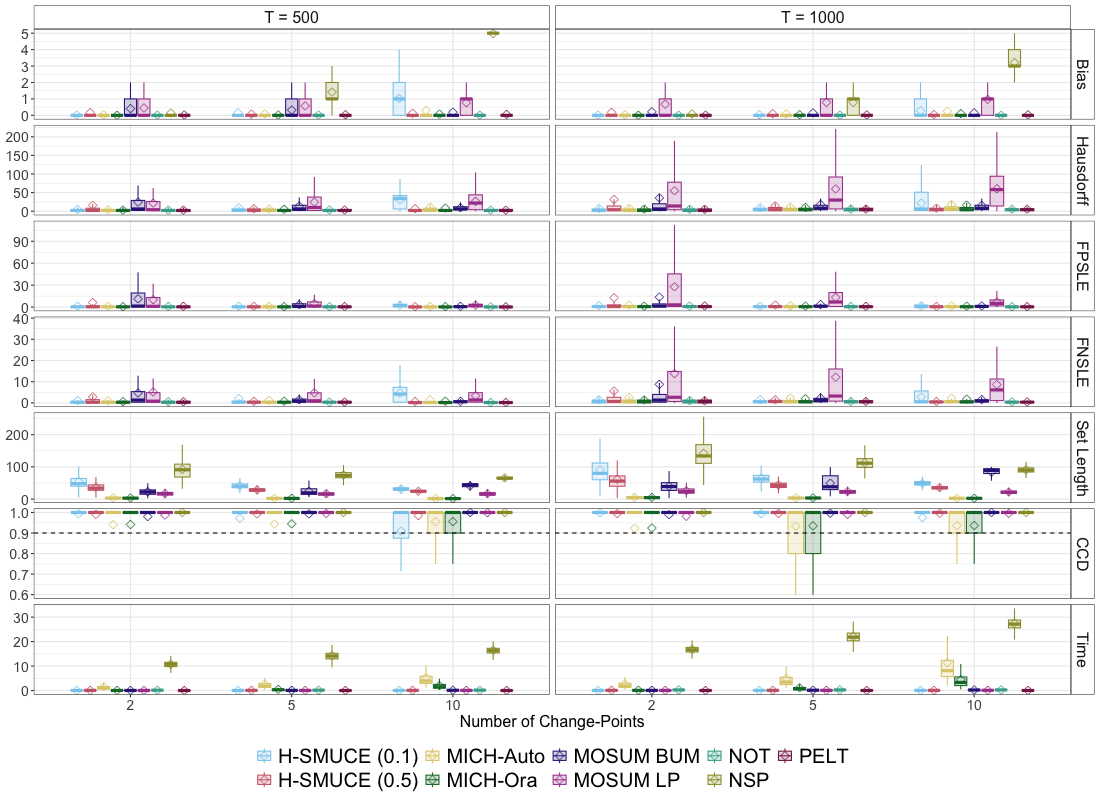}
    \caption{\textbf{Simulation \ref{sim:main} Results -- Large $\Delta_T$.} Box-plots of evaluation metrics for 5,000 replicates of Simulation \ref{sim:main} in large $\Delta_T$ setting with $\Delta_{500} = 30, \Delta_{1000} = 50$. Diamonds ($\Diamond$) display mean of each statistic. Bias $|J^* - J|$ assesses each model's ability to estimate correct number of changes (lower is better). Hausdorff, FPSLE, and FNSLE assess each model's ability to accurately estimate the locations of the changes (lower is better). Set Length and CCD report the average size and coverage of credible/confidence sets for methods that provide uncertainty quantification (dashed line indicates nominal coverage level for $\alpha = 0.1$). Time reports run-time for each method in seconds.}
    \label{fig:high_full_sim_plot}
\end{figure}

Figures \ref{fig:low_full_sim_plot} and \ref{fig:high_full_sim_plot} as well as Tables \ref{tab:main_sim_1}-\ref{tab:main_sim_3} display the results of the simulation study from  Section \ref{sec:simulations}. First note that every MICH model achieves the nominal coverage level of $90\%$ in each setting, while at the same time returning credible sets that are smaller than those returned by H-SMUCE, MOSUM, and NSP by a factor of 10-50. This empirical result reflects the the desirable localization properties we established in Section \ref{sec:localization}. Next, we see that MICH-Ora is virtually unbiased for $J^*$ in every setting (as we should expect), and in all but the most difficult setting with $T=$ 1,000 and $J^* = 10$, MICH-Ora either is the best performer as measured by the Hausdorff statistic, FPSLE, and FNSLE, or is off by a negligible factor. MICH-Ora is typically infeasible as we do not have knowledge of $J^*$, but we can always fit MICH-Auto, which performs admirably. In all but the most difficult settings ($T=$ 500, $J^* = 10$ and $T=$ 1,000, $J^* = 10$), we see that MICH-Auto is virtually unbiased for $J^*$ and performs either better or just as well as any other method (though the FNSLE of MICH appears to lag slightly behind that of PELT and NOT). 

MICH-Auto performs particularly well compared to the other methods that also provide uncertainty quantification. NSP and MOSUM both displaying significant bias for $J^*$ in each setting, and MOSUM consistently has the largest Hausdorff statistics, while and the regions returned by NSP are up to 50 times larger than the credible sets returned by MICH. MICH significantly outperforms H-SMUCE in the low observation settings with $T=100$. In particular, for the challenging setting where $T=100$ and $J^*=5$, MICH is essentially unbiased whereas H-SMUCE is off by 0.6-1.6 change-points and has performance metrics that are greater than those of MICH by a factor of 10. The performance of H-SMUCE begins to approach that of MICH as $T$ increases. When $T= $ 1,000 and $J^* = 10$, H-SMUCE with $\alpha = 0.5$ moderately outperforms MICH; however, MICH-Auto still outperforms H-SMUCE at the desired significance level $\alpha = 0.1$, particularly when $\Delta_T = 30$. In general, the performance of H-SMUCE heavily depends on the parameter $\alpha$. When $J^*/T$ is large, H-SMUCE with $\alpha = 0.5$ dominates the fit at $\alpha=0.1$ and vice versa when $J^*/T$ is small. 

Although MICH converges much faster than a Gibbs sampler implementation of the model (\ref{eq:y-mich})-(\ref{eq:s-mich}), the time needed to fit MICH still lags behind all other models except for NSP. The sensitivity analysis in Appendix \ref{app:tol_sensitivity} shows that increasing the convergence criterion to $\epsilon=10^{-5}$ brings the time to fit MICH in line with the other methods without sacrificing much in the way of performance. It is also possible to improve the performance of MICH in more challenging settings by further decreasing $\epsilon$. 

\subsection{Heavy-Tailed Simulation}

In this section we repeat Simulation \ref{sim:main}, only now we relax the Gaussianity assumption and adopt heavy-tailed distributions for $e_t$. Specifically, we still generate the change-points and jump sizes $\{\tau_j, \mu_j,\sigma_j\}$ as in Simulation \ref{sim:main}, however we now have:
\begin{align}
    y_t:= \sum_{j=0}^J \mu_j\mathbbm{1}_{\{\tau_j \leq t < \tau_{j+1}\}} + \frac{e_t}{\text{SD}(e_t)} \sum_{j=0}^J \sigma_j\mathbbm{1}_{\{\tau_j \leq t < \tau_{j+1}\}}, \label{eq:heavy_sim}
\end{align}
where $e_t$ follows a heavy tailed distribution. We fit MICH with error tolerance $\epsilon =10^{-8}$, determine detections with $\delta = 0.5$, and construct credible sets for $\alpha = 0.1$.

\begin{figure}[!h]
    \centering
    \includegraphics[scale =0.42]{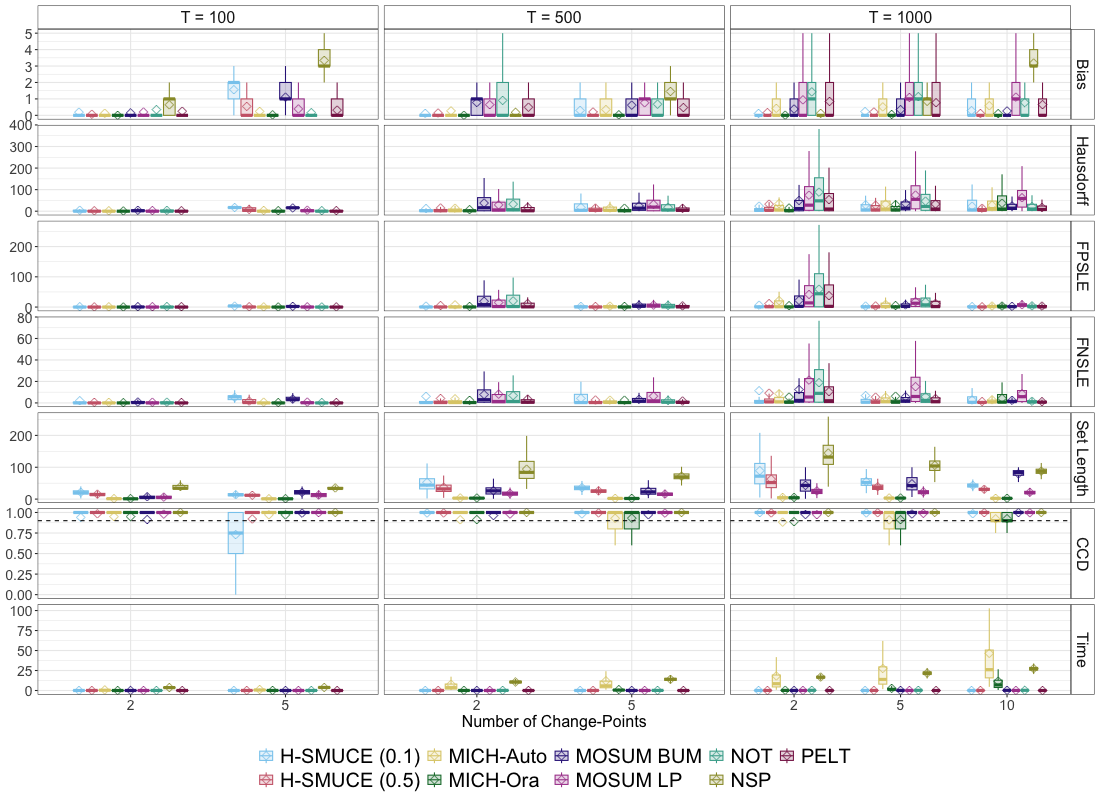}
    \caption{\small \textbf{Simulation \ref{sim:main} Results with $t_4$ Errors}. Box-plots of evaluation metrics for 5,000 replicates of $y_t$ generated according to (\ref{eq:heavy_sim}) with $e_t\sim t_4$. Diamonds ($\Diamond$) display mean of each statistic. Bias $|J^* - J|$ assesses each model's ability to estimate correct number of changes (lower is better). FPSLE and FNSLE assess each model's ability to accurately estimate the locations of the changes (lower is better). Set Length and CCD report the average size and coverage of credible/confidence sets for methods that provide uncertainty quantification (dashed line indicates nominal coverage level for $\alpha = 0.1$). Time reports run-time for each method in seconds.}
    \label{fig:t_sim_plot}
\end{figure}

\begin{table}[htbp!] 
    \scriptsize
    \centering
    \begin{tabular}{l || l || r r r r | r r | r}
        \multicolumn{1}{c||}{Setting} & \multicolumn{1}{c||}{Method} & $|J^* - J|$ & Hausdorff & FPSLE & FNSLE & CI Len. & CCD & Time (s) \\ \hline\hline 
        $T =$ 100 & MICH-Auto & \textsuperscript{}0.128 & \textsuperscript{\dag\ddag}2.056 & \textsuperscript{}0.814 & \textsuperscript{\ddag}0.546 & 1.578 & 0.950 & 0.768\\
        $J^* =$ 2 & MICH-Ora & \textsuperscript{*}0.001 & \textsuperscript{*}1.106 & \textsuperscript{*}0.233 & \textsuperscript{*}0.256 & 1.435 & 0.950 & 0.017\\
        Min. Space = 15 & H-SMUCE (0.1) & \textsuperscript{}0.194 & \textsuperscript{}4.615 & \textsuperscript{}1.154 & \textsuperscript{}2.064 & 22.139 & 0.967 & 0.021\\
         & H-SMUCE (0.5) & \textsuperscript{\dag\ddag}0.068 & \textsuperscript{}2.075 & \textsuperscript{\dag\ddag}0.545 & \textsuperscript{}0.589 & 15.799 & 0.988 & 0.059\\
         & MOSUM BUM & \textsuperscript{}0.145 & \textsuperscript{}4.298 & \textsuperscript{}1.271 & \textsuperscript{}0.845 & 8.095 & 0.912 & 0.017\\
         & MOSUM LP & \textsuperscript{}0.202 & \textsuperscript{}2.998 & \textsuperscript{}1.161 & \textsuperscript{}0.667 & 6.735 & 0.980 & 0.016\\
         & PELT & \textsuperscript{}0.251 & \textsuperscript{}2.125 & \textsuperscript{}1.172 & \textsuperscript{\dag}0.394 &  &  & 0.002\\
         & NOT & \textsuperscript{}0.348 & \textsuperscript{}3.711 & \textsuperscript{}1.651 & \textsuperscript{}0.728 &  &  & 0.050\\
         & NSP & \textsuperscript{}0.639 & \textsuperscript{} & \textsuperscript{} & \textsuperscript{} & 38.189 & 1.000 & 3.673\\ \hline
        $T =$ 100 & MICH-Auto & \textsuperscript{\ddag}0.240 & \textsuperscript{\ddag}3.541 & \textsuperscript{\ddag}0.496 & \textsuperscript{\ddag}0.707 & 1.153 & 0.972 & 1.273\\
        $J^* =$ 5 & MICH-Ora & \textsuperscript{*}0.031 & \textsuperscript{*}2.425 & \textsuperscript{*}0.281 & \textsuperscript{}0.354 & 1.191 & 0.971 & 0.137\\
        Min. Space = 15 & H-SMUCE (0.1) & \textsuperscript{}1.560 & \textsuperscript{}17.128 & \textsuperscript{}3.191 & \textsuperscript{}5.173 & 14.028 & 0.767 & 0.023\\
         & H-SMUCE (0.5) & \textsuperscript{}0.542 & \textsuperscript{}8.188 & \textsuperscript{}1.023 & \textsuperscript{}1.739 & 12.144 & 0.936 & 0.059\\
         & MOSUM BUM & \textsuperscript{}1.109 & \textsuperscript{}14.953 & \textsuperscript{}2.347 & \textsuperscript{}3.530 & 21.713 & 0.993 & 0.021\\
         & MOSUM LP & \textsuperscript{}0.389 & \textsuperscript{}4.861 & \textsuperscript{}0.744 & \textsuperscript{}0.793 & 12.939 & 0.997 & 0.024\\
         & PELT & \textsuperscript{}0.318 & \textsuperscript{}2.795 & \textsuperscript{}0.508 & \textsuperscript{}0.391 &  &  & 0.003\\
         & NOT & \textsuperscript{\dag}0.162 & \textsuperscript{\dag}2.508 & \textsuperscript{\dag}0.315 & \textsuperscript{*\dag}0.349 &  &  & 0.050\\
         & NSP & \textsuperscript{}3.344 & \textsuperscript{} & \textsuperscript{} & \textsuperscript{} & 35.476 & 1.000 & 3.805\\ \hline
        $T =$ 500 & MICH-Auto & \textsuperscript{}0.269 & \textsuperscript{}13.950 & \textsuperscript{}7.238 & \textsuperscript{\ddag}3.825 & 3.715 & 0.911 & 7.796\\
        $J^* =$ 2 & MICH-Ora & \textsuperscript{*}0.005 & \textsuperscript{*}7.122 & \textsuperscript{*}2.013 & \textsuperscript{*}2.465 & 3.426 & 0.914 & 0.134\\
        Min. Space = 30 & H-SMUCE (0.1) & \textsuperscript{\dag\ddag}0.121 & \textsuperscript{\dag\ddag}12.533 & \textsuperscript{\dag\ddag}3.173 & \textsuperscript{}6.110 & 52.502 & 0.996 & 0.032\\
         & H-SMUCE (0.5) & \textsuperscript{}0.142 & \textsuperscript{}14.904 & \textsuperscript{}5.038 & \textsuperscript{}4.176 & 37.831 & 0.994 & 0.067\\
         & MOSUM BUM & \textsuperscript{}0.777 & \textsuperscript{}37.873 & \textsuperscript{}19.136 & \textsuperscript{}7.766 & 29.209 & 0.962 & 0.034\\
         & MOSUM LP & \textsuperscript{}0.646 & \textsuperscript{}29.246 & \textsuperscript{}14.954 & \textsuperscript{}8.169 & 18.879 & 0.981 & 0.034\\
         & PELT & \textsuperscript{}0.496 & \textsuperscript{}18.114 & \textsuperscript{}11.671 & \textsuperscript{\dag}3.535 &  &  & 0.009\\
         & NOT & \textsuperscript{}0.908 & \textsuperscript{}32.736 & \textsuperscript{}20.255 & \textsuperscript{}6.824 &  &  & 0.141\\
         & NSP & \textsuperscript{}0.184 & \textsuperscript{} & \textsuperscript{} & \textsuperscript{} & 92.255 & 1.000 & 10.633\\ \hline
        $T =$ 500 & MICH-Auto & \textsuperscript{}0.351 & \textsuperscript{}15.142 & \textsuperscript{}3.157 & \textsuperscript{\ddag}2.581 & 2.396 & 0.929 & 12.118\\
        $J^* =$ 5 & MICH-Ora & \textsuperscript{*}0.034 & \textsuperscript{}13.872 & \textsuperscript{}2.151 & \textsuperscript{}2.570 & 2.333 & 0.931 & 1.146\\
        Min. Space = 30 & H-SMUCE (0.1) & \textsuperscript{}0.299 & \textsuperscript{}17.776 & \textsuperscript{}2.015 & \textsuperscript{}4.192 & 36.077 & 0.986 & 0.034\\
         & H-SMUCE (0.5) & \textsuperscript{\dag\ddag}0.205 & \textsuperscript{\ddag}14.518 & \textsuperscript{*\dag\ddag}1.850 & \textsuperscript{}2.722 & 25.868 & 0.995 & 0.067\\
         & MOSUM BUM & \textsuperscript{}0.621 & \textsuperscript{}24.285 & \textsuperscript{}5.814 & \textsuperscript{}3.268 & 25.656 & 0.972 & 0.058\\
         & MOSUM LP & \textsuperscript{}0.762 & \textsuperscript{}31.671 & \textsuperscript{}6.614 & \textsuperscript{}6.278 & 16.190 & 0.993 & 0.045\\
         & PELT & \textsuperscript{}0.474 & \textsuperscript{*\dag}13.335 & \textsuperscript{}3.873 & \textsuperscript{*\dag}1.777 &  &  & 0.009\\
         & NOT & \textsuperscript{}0.679 & \textsuperscript{}18.282 & \textsuperscript{}5.422 & \textsuperscript{}2.288 &  &  & 0.145\\
         & NSP & \textsuperscript{}1.453 & \textsuperscript{} & \textsuperscript{} & \textsuperscript{} & 71.626 & 1.000 & 14.008\\ \hline
        $T =$ 1000 & MICH-Auto & \textsuperscript{}0.426 & \textsuperscript{}34.942 & \textsuperscript{}21.144 & \textsuperscript{\dag\ddag}8.455 & 5.352 & 0.884 & 18.172\\
        $J^* =$ 2 & MICH-Ora & \textsuperscript{*}0.011 & \textsuperscript{*}14.949 & \textsuperscript{*}4.573 & \textsuperscript{*}5.457 & 5.228 & 0.889 & 0.347\\
        Min. Space = 50 & H-SMUCE (0.1) & \textsuperscript{\dag\ddag}0.117 & \textsuperscript{\dag\ddag}23.580 & \textsuperscript{\dag\ddag}6.219 & \textsuperscript{}11.368 & 88.268 & 0.999 & 0.043\\
         & H-SMUCE (0.5) & \textsuperscript{}0.167 & \textsuperscript{}32.728 & \textsuperscript{}11.590 & \textsuperscript{}9.107 & 63.666 & 0.995 & 0.076\\
         & MOSUM BUM & \textsuperscript{}0.384 & \textsuperscript{}49.666 & \textsuperscript{}22.002 & \textsuperscript{}12.455 & 45.575 & 0.974 & 0.063\\
         & MOSUM LP & \textsuperscript{}0.942 & \textsuperscript{}74.261 & \textsuperscript{}41.627 & \textsuperscript{}21.220 & 26.295 & 0.973 & 0.054\\
         & PELT & \textsuperscript{}0.846 & \textsuperscript{}53.419 & \textsuperscript{}37.553 & \textsuperscript{}10.486 &  &  & 0.017\\
         & NOT & \textsuperscript{}1.427 & \textsuperscript{}88.372 & \textsuperscript{}59.160 & \textsuperscript{}18.857 &  &  & 0.248\\
         & NSP & \textsuperscript{}0.121 & \textsuperscript{} & \textsuperscript{} & \textsuperscript{} & 142.987 & 1.000 & 16.825\\ \hline
        $T =$ 1000 & MICH-Auto & \textsuperscript{}0.513 & \textsuperscript{}32.891 & \textsuperscript{}8.601 & \textsuperscript{\ddag}5.213 & 3.648 & 0.911 & 27.516\\
        $J^* =$ 5 & MICH-Ora & \textsuperscript{*}0.053 & \textsuperscript{}30.453 & \textsuperscript{}5.017 & \textsuperscript{}6.259 & 3.486 & 0.913 & 2.723\\
        Min. Space = 50 & H-SMUCE (0.1) & \textsuperscript{}0.231 & \textsuperscript{}29.755 & \textsuperscript{*\dag\ddag}3.270 & \textsuperscript{}6.804 & 55.013 & 0.995 & 0.044\\
         & H-SMUCE (0.5) & \textsuperscript{\dag\ddag}0.204 & \textsuperscript{*\dag\ddag}27.892 & \textsuperscript{}3.658 & \textsuperscript{}5.408 & 38.754 & 0.996 & 0.076\\
         & MOSUM BUM & \textsuperscript{}0.346 & \textsuperscript{}36.451 & \textsuperscript{}7.124 & \textsuperscript{}5.690 & 48.742 & 0.981 & 0.124\\
         & MOSUM LP & \textsuperscript{}1.091 & \textsuperscript{}74.866 & \textsuperscript{}18.211 & \textsuperscript{}15.104 & 23.076 & 0.990 & 0.070\\
         & PELT & \textsuperscript{}0.770 & \textsuperscript{}33.104 & \textsuperscript{}11.830 & \textsuperscript{*\dag}4.062 &  &  & 0.016\\
         & NOT & \textsuperscript{}1.136 & \textsuperscript{}46.674 & \textsuperscript{}16.812 & \textsuperscript{}5.790 &  &  & 0.254\\
         & NSP & \textsuperscript{}0.850 & \textsuperscript{} & \textsuperscript{} & \textsuperscript{} & 105.509 & 1.000 & 21.752\\ \hline
        $T =$ 1000 & MICH-Auto & \textsuperscript{}0.583 & \textsuperscript{}27.740 & \textsuperscript{}3.449 & \textsuperscript{}2.881 & 2.584 & 0.923 & 46.445\\
        $J^* =$ 10 & MICH-Ora & \textsuperscript{*}0.105 & \textsuperscript{}36.927 & \textsuperscript{}3.347 & \textsuperscript{}4.249 & 2.550 & 0.924 & 11.605\\
        Min. Space = 50 & H-SMUCE (0.1) & \textsuperscript{}0.284 & \textsuperscript{}23.178 & \textsuperscript{}1.545 & \textsuperscript{}2.732 & 43.142 & 0.986 & 0.046\\
         & H-SMUCE (0.5) & \textsuperscript{\dag\ddag}0.114 & \textsuperscript{*\dag\ddag}13.832 & \textsuperscript{*\dag\ddag}0.951 & \textsuperscript{\ddag}1.308 & 30.987 & 0.997 & 0.077\\
         & MOSUM BUM & \textsuperscript{}0.264 & \textsuperscript{}24.142 & \textsuperscript{}2.293 & \textsuperscript{}2.420 & 79.544 & 0.999 & 0.214\\
         & MOSUM LP & \textsuperscript{}1.109 & \textsuperscript{}62.404 & \textsuperscript{}7.519 & \textsuperscript{}8.509 & 21.484 & 0.996 & 0.093\\
         & PELT & \textsuperscript{}0.650 & \textsuperscript{}17.279 & \textsuperscript{}3.264 & \textsuperscript{*\dag}1.225 &  &  & 0.016\\
         & NOT & \textsuperscript{}0.774 & \textsuperscript{}21.392 & \textsuperscript{}3.897 & \textsuperscript{}1.576 &  &  & 0.254\\
         & NSP & \textsuperscript{}3.183 & \textsuperscript{} & \textsuperscript{} & \textsuperscript{} & 87.515 & 1.000 & 27.390\\ \hline\hline
    \end{tabular}
    \caption{\small \textbf{Simulation \ref{sim:main} Results with $t_4$ Errors}. Average statistics over 5,000 replicates of $y_t$ generated according to (\ref{eq:heavy_sim}) with $e_t\sim t_4$. MICH is fit with error tolerance $\epsilon =10^{-8}$, credible sets are constructed for $\alpha = 0.1$, and detections are determined with $\delta = 0.5$. Key: \textsuperscript{*}Best overall. \textsuperscript{\textdagger}Best excluding MICH-Ora. \textsuperscript{\ddag}Best among methods with uncertainty quantification excluding MICH-Ora.}
    \label{tab:t_sim}
\end{table}

Figure \ref{fig:t_sim_plot} and Table \ref{tab:t_sim} show the results for this study with $e_t \sim t_4$. Across all metrics, we see that MICH-Ora is (unsurprisingly) the best performer or off from the best by a negligible factor. In the small sample settings with $T=100$, MICH-Auto is the next best performing method across all metrics (or again off by a negligible factor). As with Simulation \ref{sim:main}, MICH also returns the smallest credible sets among methods that provide uncertainty quantification by an order of magnitude, while still meeting the nominal coverage level. 

Figure \ref{fig:t_sim_plot} shows that MICH, PELT, NOT and H-SMUCE with $\alpha = 0.5$ are unbiased for the number of change-points in the median simulation. However, averaging the bias across all 5,000 replicates, H-SMUCE with $\alpha = 0.5$ significantly outperforms the other methods, particularly as $T$ and $J^*$ increase. The bias for PELT and NOT is notably affected by the presence of heavy-tailed errors, with both showing bias on the order of about an entire change-point for $T \geq 500$. The bias of MICH-Auto is also inflated as compared to Simulation \ref{sim:main} due to the presence of heavy-tailed errors, but is still roughly on the order of H-SMUCE with $\alpha = 0.1$. MICH-Auto performs notably better when $T=100$ and $J^* = 5$, where H-SMUCE is biased by nearly two whole change-points, once again showing that the performance of H-SMUCE depends heavily on the choice of $\alpha$. We also see that PELT and MICH-Auto outperform the other methods in terms of FNSLE, with the performance gap increasing as $T$ and $J^*$ increase (though H-SMUCE begins to catch up around $T=1000$). On the other hand, H-SMUCE beats all other methods with respect to the FPSLE, which is expected as H-SMUCE is designed to control the family-wise error rate for overestimating the number of change-points. 

\begin{figure}[!h]
    \centering
    \includegraphics[scale =0.42]{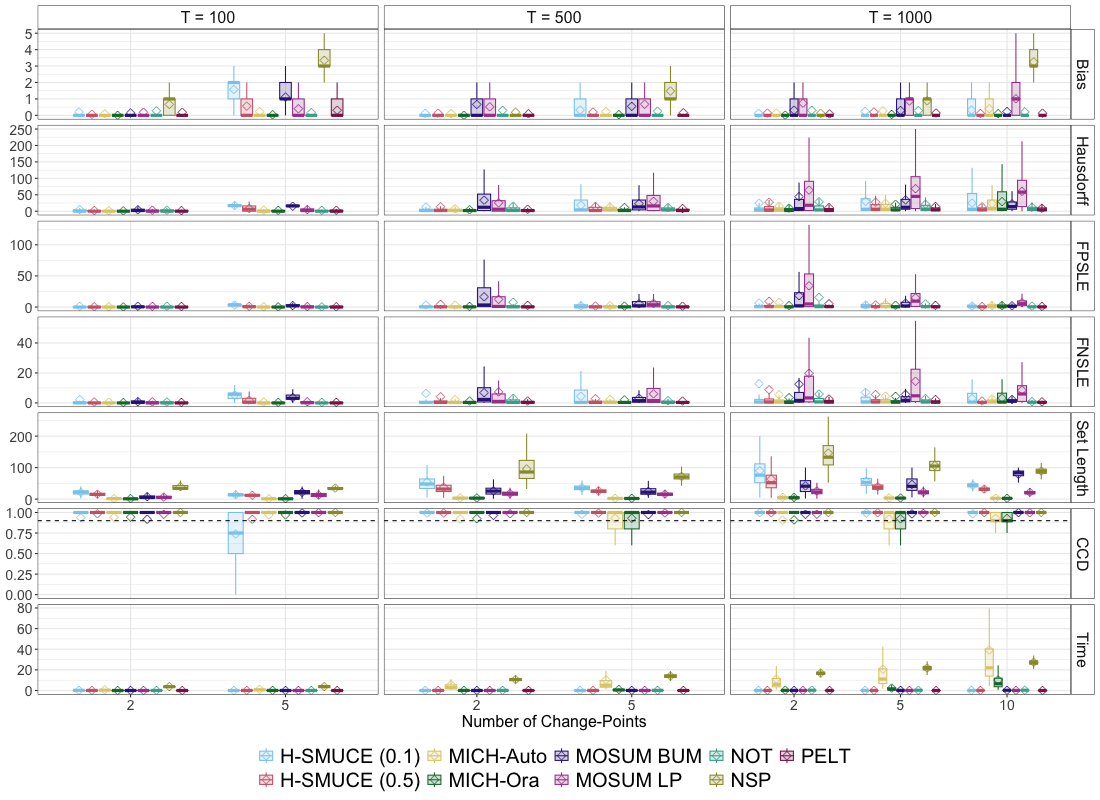}
    \caption{\small \textbf{Simulation \ref{sim:main} Results with Laplace Errors}. Box-plots of evaluation metrics for 5,000 replicates of $y_t$ generated according to (\ref{eq:heavy_sim}) with $e_t\sim \text{Laplace}(0,1)$. Diamonds ($\Diamond$) display mean of each statistic. Bias $|J^* - J|$ assesses each model's ability to estimate correct number of changes (lower is better). FPSLE and FNSLE assess each model's ability to accurately estimate the locations of the changes (lower is better). Set Length and CCD report the average size and coverage of credible/confidence sets for methods that provide uncertainty quantification (dashed line indicates nominal coverage level for $\alpha = 0.1$). Time reports run-time for each method in seconds.}
    \label{fig:laplace_sim_plot}
\end{figure}

\begin{table}[htbp!] 
    \scriptsize
    \centering
    \begin{tabular}{l || l || r r r r | r r | r}
        \multicolumn{1}{c||}{Setting} & \multicolumn{1}{c||}{Method} & $|J^* - J|$ & Hausdorff & FPSLE & FNSLE & CI Len. & CCD & Time (s) \\ \hline\hline 
        $T =$ 100 & MICH-Auto & \textsuperscript{}0.097 & \textsuperscript{\ddag}1.916 & \textsuperscript{}0.689 & \textsuperscript{\ddag}0.511 & 1.623 & 0.942 & 0.722\\
        $J^* =$ 2 & MICH-Ora & \textsuperscript{*}0.001 & \textsuperscript{*}1.006 & \textsuperscript{*}0.209 & \textsuperscript{*}0.232 & 1.440 & 0.944 & 0.017\\
        Min. Space = 15 & H-SMUCE (0.1) & \textsuperscript{}0.205 & \textsuperscript{}4.808 & \textsuperscript{}1.204 & \textsuperscript{}2.193 & 22.355 & 0.966 & 0.182\\
         & H-SMUCE (0.5) & \textsuperscript{\dag\ddag}0.066 & \textsuperscript{}2.054 & \textsuperscript{\dag\ddag}0.517 & \textsuperscript{}0.620 & 15.994 & 0.989 & 0.060\\
         & MOSUM BUM & \textsuperscript{}0.143 & \textsuperscript{}4.358 & \textsuperscript{}1.287 & \textsuperscript{}0.861 & 8.306 & 0.916 & 0.019\\
         & MOSUM LP & \textsuperscript{}0.191 & \textsuperscript{}2.971 & \textsuperscript{}1.148 & \textsuperscript{}0.703 & 6.851 & 0.977 & 0.017\\
         & PELT & \textsuperscript{}0.189 & \textsuperscript{\dag}1.756 & \textsuperscript{}0.916 & \textsuperscript{\dag}0.332 &  &  & 0.004\\
         & NOT & \textsuperscript{}0.274 & \textsuperscript{}3.187 & \textsuperscript{}1.341 & \textsuperscript{}0.633 &  &  & 0.050\\
         & NSP & \textsuperscript{}0.664 & \textsuperscript{} & \textsuperscript{} & \textsuperscript{} & 38.522 & 1.000 & 3.752\\ \hline
        $T =$ 100 & MICH-Auto & \textsuperscript{\ddag}0.229 & \textsuperscript{\ddag}3.551 & \textsuperscript{\ddag}0.459 & \textsuperscript{\ddag}0.673 & 1.164 & 0.971 & 1.193\\
        $J^* =$ 5 & MICH-Ora & \textsuperscript{*}0.037 & \textsuperscript{*}2.534 & \textsuperscript{*}0.292 & \textsuperscript{}0.373 & 1.196 & 0.970 & 0.138\\
        Min. Space = 15 & H-SMUCE (0.1) & \textsuperscript{}1.571 & \textsuperscript{}17.184 & \textsuperscript{}3.227 & \textsuperscript{}5.196 & 14.072 & 0.775 & 0.022\\
         & H-SMUCE (0.5) & \textsuperscript{}0.561 & \textsuperscript{}8.339 & \textsuperscript{}1.071 & \textsuperscript{}1.807 & 12.225 & 0.933 & 0.059\\
         & MOSUM BUM & \textsuperscript{}1.116 & \textsuperscript{}15.081 & \textsuperscript{}2.343 & \textsuperscript{}3.555 & 21.766 & 0.995 & 0.021\\
         & MOSUM LP & \textsuperscript{}0.400 & \textsuperscript{}4.947 & \textsuperscript{}0.751 & \textsuperscript{}0.804 & 13.077 & 0.998 & 0.024\\
         & PELT & \textsuperscript{}0.312 & \textsuperscript{}2.826 & \textsuperscript{}0.497 & \textsuperscript{}0.393 &  &  & 0.003\\
         & NOT & \textsuperscript{\dag}0.164 & \textsuperscript{\dag}2.568 & \textsuperscript{\dag}0.320 & \textsuperscript{*\dag}0.356 &  &  & 0.049\\
         & NSP & \textsuperscript{}3.351 & \textsuperscript{} & \textsuperscript{} & \textsuperscript{} & 35.426 & 1.000 & 3.794\\ \hline
        $T =$ 500 & MICH-Auto & \textsuperscript{\ddag}0.109 & \textsuperscript{\ddag}7.427 & \textsuperscript{\ddag}3.053 & \textsuperscript{\ddag}2.249 & 3.889 & 0.922 & 5.737\\
        $J^* =$ 2 & MICH-Ora & \textsuperscript{*}0.004 & \textsuperscript{*}4.346 & \textsuperscript{*}1.060 & \textsuperscript{}1.300 & 3.565 & 0.925 & 0.120\\
        Min. Space = 30 & H-SMUCE (0.1) & \textsuperscript{}0.126 & \textsuperscript{}12.461 & \textsuperscript{}3.169 & \textsuperscript{}6.352 & 54.024 & 0.996 & 0.034\\
         & H-SMUCE (0.5) & \textsuperscript{}0.130 & \textsuperscript{}13.503 & \textsuperscript{}4.409 & \textsuperscript{}4.116 & 38.126 & 0.994 & 0.067\\
         & MOSUM BUM & \textsuperscript{}0.663 & \textsuperscript{}33.526 & \textsuperscript{}16.477 & \textsuperscript{}6.830 & 28.573 & 0.964 & 0.033\\
         & MOSUM LP & \textsuperscript{}0.514 & \textsuperscript{}25.460 & \textsuperscript{}12.408 & \textsuperscript{}7.334 & 18.493 & 0.984 & 0.034\\
         & PELT & \textsuperscript{\dag}0.096 & \textsuperscript{\dag}5.982 & \textsuperscript{\dag}2.773 & \textsuperscript{*\dag}1.281 &  &  & 0.009\\
         & NOT & \textsuperscript{}0.302 & \textsuperscript{}15.272 & \textsuperscript{}7.752 & \textsuperscript{}3.063 &  &  & 0.141\\
         & NSP & \textsuperscript{}0.190 & \textsuperscript{} & \textsuperscript{} & \textsuperscript{} & 93.633 & 1.000 & 10.663\\ \hline
        $T =$ 500 & MICH-Auto & \textsuperscript{}0.231 & \textsuperscript{\ddag}12.443 & \textsuperscript{}2.037 & \textsuperscript{\ddag}2.296 & 2.412 & 0.930 & 10.229\\
        $J^* =$ 5 & MICH-Ora & \textsuperscript{*}0.044 & \textsuperscript{}11.116 & \textsuperscript{}1.529 & \textsuperscript{}2.030 & 2.326 & 0.930 & 1.059\\
        Min. Space = 30 & H-SMUCE (0.1) & \textsuperscript{}0.324 & \textsuperscript{}18.355 & \textsuperscript{}2.104 & \textsuperscript{}4.466 & 36.545 & 0.983 & 0.034\\
         & H-SMUCE (0.5) & \textsuperscript{\ddag}0.201 & \textsuperscript{}14.193 & \textsuperscript{\ddag}1.714 & \textsuperscript{}2.815 & 26.010 & 0.994 & 0.068\\
         & MOSUM BUM & \textsuperscript{}0.547 & \textsuperscript{}22.860 & \textsuperscript{}5.275 & \textsuperscript{}3.137 & 25.162 & 0.974 & 0.057\\
         & MOSUM LP & \textsuperscript{}0.676 & \textsuperscript{}30.030 & \textsuperscript{}5.881 & \textsuperscript{}6.026 & 16.151 & 0.994 & 0.045\\
         & PELT & \textsuperscript{\dag}0.156 & \textsuperscript{*\dag}8.100 & \textsuperscript{*\dag}1.444 & \textsuperscript{*\dag}1.277 &  &  & 0.009\\
         & NOT & \textsuperscript{}0.254 & \textsuperscript{}11.369 & \textsuperscript{}2.331 & \textsuperscript{}1.544 &  &  & 0.145\\
         & NSP & \textsuperscript{}1.479 & \textsuperscript{} & \textsuperscript{} & \textsuperscript{} & 71.943 & 1.000 & 13.989\\ \hline
        $T =$ 1000 & MICH-Auto & \textsuperscript{}0.138 & \textsuperscript{\ddag}15.618 & \textsuperscript{}7.330 & \textsuperscript{\ddag}5.498 & 5.711 & 0.907 & 11.056\\
        $J^* =$ 2 & MICH-Ora & \textsuperscript{*}0.012 & \textsuperscript{*}9.327 & \textsuperscript{*}2.448 & \textsuperscript{}3.538 & 5.361 & 0.910 & 0.300\\
        Min. Space = 50 & H-SMUCE (0.1) & \textsuperscript{}0.129 & \textsuperscript{}24.483 & \textsuperscript{\ddag}6.374 & \textsuperscript{}12.819 & 89.159 & 0.998 & 0.044\\
         & H-SMUCE (0.5) & \textsuperscript{}0.139 & \textsuperscript{}27.311 & \textsuperscript{}9.192 & \textsuperscript{}8.781 & 64.094 & 0.997 & 0.076\\
         & MOSUM BUM & \textsuperscript{}0.313 & \textsuperscript{}44.261 & \textsuperscript{}18.231 & \textsuperscript{}12.493 & 43.684 & 0.978 & 0.061\\
         & MOSUM LP & \textsuperscript{}0.742 & \textsuperscript{}64.184 & \textsuperscript{}34.084 & \textsuperscript{}19.776 & 26.329 & 0.979 & 0.054\\
         & PELT & \textsuperscript{\dag}0.087 & \textsuperscript{\dag}11.203 & \textsuperscript{\dag}5.275 & \textsuperscript{*\dag}2.599 &  &  & 0.016\\
         & NOT & \textsuperscript{}0.308 & \textsuperscript{}28.575 & \textsuperscript{}15.584 & \textsuperscript{}5.883 &  &  & 0.248\\
         & NSP & \textsuperscript{\ddag}0.128 & \textsuperscript{} & \textsuperscript{} & \textsuperscript{} & 144.308 & 1.000 & 16.852\\ \hline
        $T =$ 1000 & MICH-Auto & \textsuperscript{}0.272 & \textsuperscript{\ddag}23.641 & \textsuperscript{}4.491 & \textsuperscript{\ddag}4.538 & 3.751 & 0.921 & 20.625\\
        $J^* =$ 5 & MICH-Ora & \textsuperscript{*}0.058 & \textsuperscript{}21.102 & \textsuperscript{}3.031 & \textsuperscript{}4.401 & 3.562 & 0.923 & 2.376\\
        Min. Space = 50 & H-SMUCE (0.1) & \textsuperscript{}0.244 & \textsuperscript{}29.937 & \textsuperscript{\ddag}3.327 & \textsuperscript{}7.040 & 56.356 & 0.995 & 0.044\\
         & H-SMUCE (0.5) & \textsuperscript{\ddag}0.203 & \textsuperscript{}27.532 & \textsuperscript{}3.484 & \textsuperscript{}5.558 & 39.431 & 0.997 & 0.077\\
         & MOSUM BUM & \textsuperscript{}0.291 & \textsuperscript{}34.249 & \textsuperscript{}6.233 & \textsuperscript{}5.667 & 47.183 & 0.983 & 0.123\\
         & MOSUM LP & \textsuperscript{}0.871 & \textsuperscript{}68.413 & \textsuperscript{}14.869 & \textsuperscript{}14.461 & 22.882 & 0.989 & 0.069\\
         & PELT & \textsuperscript{\dag}0.141 & \textsuperscript{*\dag}14.238 & \textsuperscript{*\dag}2.747 & \textsuperscript{*\dag}2.263 &  &  & 0.016\\
         & NOT & \textsuperscript{}0.288 & \textsuperscript{}21.117 & \textsuperscript{}5.135 & \textsuperscript{}2.853 &  &  & 0.253\\
         & NSP & \textsuperscript{}0.883 & \textsuperscript{} & \textsuperscript{} & \textsuperscript{} & 106.239 & 1.000 & 21.744\\ \hline
        $T =$ 1000 & MICH-Auto & \textsuperscript{}0.405 & \textsuperscript{}23.629 & \textsuperscript{}2.403 & \textsuperscript{}2.633 & 2.632 & 0.928 & 39.192\\
        $J^* =$ 10 & MICH-Ora & \textsuperscript{*}0.108 & \textsuperscript{}29.129 & \textsuperscript{}2.442 & \textsuperscript{}3.220 & 2.604 & 0.929 & 10.139\\
        Min. Space = 50 & H-SMUCE (0.1) & \textsuperscript{}0.321 & \textsuperscript{}24.737 & \textsuperscript{}1.660 & \textsuperscript{}2.997 & 44.138 & 0.984 & 0.047\\
         & H-SMUCE (0.5) & \textsuperscript{\dag\ddag}0.113 & \textsuperscript{\ddag}13.226 & \textsuperscript{*\dag\ddag}0.914 & \textsuperscript{\ddag}1.262 & 31.837 & 0.998 & 0.077\\
         & MOSUM BUM & \textsuperscript{}0.234 & \textsuperscript{}22.341 & \textsuperscript{}2.047 & \textsuperscript{}2.307 & 79.207 & 0.999 & 0.214\\
         & MOSUM LP & \textsuperscript{}1.007 & \textsuperscript{}60.324 & \textsuperscript{}6.877 & \textsuperscript{}8.516 & 21.553 & 0.996 & 0.093\\
         & PELT & \textsuperscript{}0.133 & \textsuperscript{*\dag}9.076 & \textsuperscript{}0.999 & \textsuperscript{*\dag}0.723 &  &  & 0.016\\
         & NOT & \textsuperscript{}0.211 & \textsuperscript{}11.869 & \textsuperscript{}1.386 & \textsuperscript{}0.899 &  &  & 0.255\\
         & NSP & \textsuperscript{}3.251 & \textsuperscript{} & \textsuperscript{} & \textsuperscript{} & 88.510 & 1.000 & 27.316\\ \hline\hline
    \end{tabular}
    \caption{\small \textbf{Simulation \ref{sim:main} Results with Laplace Errors}. Average statistics over 5,000 replicates of $y_t$ generated according to (\ref{eq:heavy_sim}) with $e_t\sim\text{Laplace}(0,1)$. MICH credible sets are constructed for $\alpha = 0.1$ and detections are determined with $\delta = 0.5$. \textsuperscript{*}Best overall. Key: \textsuperscript{\textdagger}Best excluding MICH-Ora. \textsuperscript{\ddag}Best among methods with uncertainty quantification excluding MICH-Ora.}
    \label{tab:laplace_sim}
\end{table}

Figure \ref{fig:laplace_sim_plot} and Table \ref{tab:laplace_sim} show the results for this study with $e_t \sim \text{Laplace}(0,1)$. These results are mostly in line with what we observed in Simulation \ref{sim:main}, though we still see that each method performs worse due to the presence of non-Gaussian data. In particular, NOT performs significantly worse that PELT, whereas there was no notable difference between these methods in Simulation \ref{sim:main}. We again see that H-SMUCE is somewhat robust to the presence of non-Gaussian data when $\alpha = 0.5$, especially in terms of the bias and FPSLE. The performance of MICH-Auto is also somewhat robust, having bias that is comparable to PELT (or better when $T=100$), and Hausdorff and FNSLE statistics that are the lowest among the methods that provide uncertainty quantification in all but the last setting where $J^*=10$ and $T=1000$.

\subsection{Dependent Data Simulation}

In this section we repeat Simulation \ref{sim:main}, only now we relax the independence assumption and allow for serial correlation in the error terms $e_t \overset{\text{i.i.d.}}{\sim} \mathcal{N}(0,1)$. Specifically, we still generate the change-points and jump sizes $\{\tau_j, \mu_j,\sigma_j\}$ as Simulation \ref{sim:main}, however we now generate $y_t$ from the MA(2) process: 
\begin{align}
    y_t &:= \sum_{j=0}^J \mu_j\mathbbm{1}_{\{\tau_j \leq t < \tau_{j+1}\}} + \frac{e_t + \theta e_{t-1} + \theta^2 e_{t-2}}{1+\theta^2 + \theta^4} \sum_{j=0}^J \sigma_j\mathbbm{1}_{\{\tau_j \leq t < \tau_{j+1}\}},  \label{eq:dep_sim}
\end{align}
To account for dependence, we set $\delta = 1.1$ for the detection rule (see Section \ref{cor:cred-sets}). 

\begin{figure}[!h]
    \centering
    \includegraphics[scale =0.42]{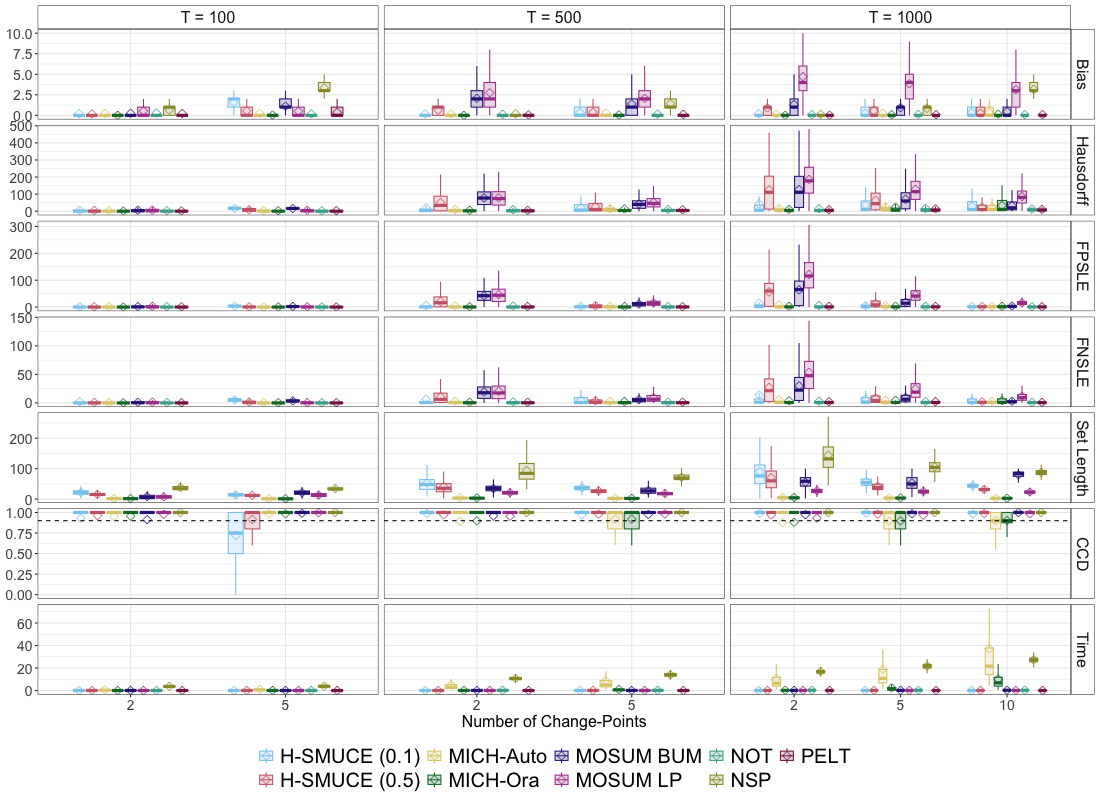}
    \caption{\small \textbf{Simulation \ref{sim:main} Results with MA(2) Errors -- Low Dependence}. Box-plots of evaluation metrics for 5,000 replicates of $y_t$ generated according to (\ref{eq:heavy_sim}) with MA(2) errors (\ref{eq:dep_sim}) and $\theta = 0.3$. Diamonds ($\Diamond$) display mean of each statistic. Bias $|J^* - J|$ assesses each model's ability to estimate correct number of changes (lower is better). FPSLE and FNSLE assess each model's ability to accurately estimate the locations of the changes (lower is better). Set Length and CCD report the average size and coverage of credible/confidence sets for methods that provide uncertainty quantification (dashed line indicates nominal coverage level for $\alpha = 0.1$). Time reports run-time for each method in seconds.}
    \label{fig:ma3_sim_plot}
\end{figure}

\begin{table}[htbp!] 
    \scriptsize
    \centering
    \begin{tabular}{l || l || r r r r | r r | r}
        \multicolumn{1}{c||}{Setting} & \multicolumn{1}{c||}{Method} & $|J^* - J|$ & Hausdorff & FPSLE & FNSLE & CI Len. & CCD & Time (s) \\ \hline\hline 
        $T =$ 100 & MICH-Auto & \textsuperscript{}0.096 & \textsuperscript{\ddag}1.902 & \textsuperscript{}0.681 & \textsuperscript{\ddag}0.504 & 1.620 & 0.942 & 0.344\\
        $J^* =$ 2 & MICH-Ora & \textsuperscript{*}0.001 & \textsuperscript{*}1.016 & \textsuperscript{*}0.212 & \textsuperscript{*}0.237 & 1.440 & 0.944 & 0.012\\
        Min. Space = 15 & H-SMUCE (0.1) & \textsuperscript{}0.205 & \textsuperscript{}4.808 & \textsuperscript{}1.204 & \textsuperscript{}2.193 & 22.355 & 0.966 & 0.194\\
         & H-SMUCE (0.5) & \textsuperscript{\dag\ddag}0.066 & \textsuperscript{}2.054 & \textsuperscript{\dag\ddag}0.517 & \textsuperscript{}0.620 & 15.994 & 0.989 & 0.060\\
         & MOSUM BUM & \textsuperscript{}0.143 & \textsuperscript{}4.358 & \textsuperscript{}1.287 & \textsuperscript{}0.861 & 8.306 & 0.916 & 0.018\\
         & MOSUM LP & \textsuperscript{}0.191 & \textsuperscript{}2.971 & \textsuperscript{}1.148 & \textsuperscript{}0.703 & 6.851 & 0.977 & 0.016\\
         & PELT & \textsuperscript{}0.189 & \textsuperscript{\dag}1.756 & \textsuperscript{}0.916 & \textsuperscript{\dag}0.332 &  &  & 0.004\\
         & NOT & \textsuperscript{}0.274 & \textsuperscript{}3.187 & \textsuperscript{}1.341 & \textsuperscript{}0.633 &  &  & 0.049\\
         & NSP & \textsuperscript{}0.664 & \textsuperscript{} & \textsuperscript{} & \textsuperscript{} & 38.522 & 1.000 & 3.689\\ \hline
        $T =$ 100 & MICH-Auto & \textsuperscript{\ddag}0.234 & \textsuperscript{\ddag}3.607 & \textsuperscript{\ddag}0.467 & \textsuperscript{\ddag}0.685 & 1.165 & 0.971 & 0.542\\
        $J^* =$ 5 & MICH-Ora & \textsuperscript{*}0.037 & \textsuperscript{*}2.519 & \textsuperscript{*}0.289 & \textsuperscript{}0.370 & 1.196 & 0.970 & 0.081\\
        Min. Space = 15 & H-SMUCE (0.1) & \textsuperscript{}1.571 & \textsuperscript{}17.184 & \textsuperscript{}3.227 & \textsuperscript{}5.196 & 14.071 & 0.775 & 0.022\\
         & H-SMUCE (0.5) & \textsuperscript{}0.561 & \textsuperscript{}8.336 & \textsuperscript{}1.071 & \textsuperscript{}1.806 & 12.225 & 0.933 & 0.059\\
         & MOSUM BUM & \textsuperscript{}1.117 & \textsuperscript{}15.083 & \textsuperscript{}2.343 & \textsuperscript{}3.555 & 21.766 & 0.995 & 0.021\\
         & MOSUM LP & \textsuperscript{}0.400 & \textsuperscript{}4.952 & \textsuperscript{}0.751 & \textsuperscript{}0.804 & 13.078 & 0.998 & 0.024\\
         & PELT & \textsuperscript{}0.312 & \textsuperscript{}2.826 & \textsuperscript{}0.497 & \textsuperscript{}0.393 &  &  & 0.003\\
         & NOT & \textsuperscript{\dag}0.164 & \textsuperscript{\dag}2.568 & \textsuperscript{\dag}0.320 & \textsuperscript{*\dag}0.356 &  &  & 0.049\\
         & NSP & \textsuperscript{}3.351 & \textsuperscript{} & \textsuperscript{} & \textsuperscript{} & 35.423 & 1.000 & 3.716\\ \hline
        $T =$ 500 & MICH-Auto & \textsuperscript{\ddag}0.114 & \textsuperscript{\ddag}7.954 & \textsuperscript{}3.285 & \textsuperscript{\ddag}2.570 & 3.911 & 0.922 & 1.951\\
        $J^* =$ 2 & MICH-Ora & \textsuperscript{*}0.005 & \textsuperscript{*}4.225 & \textsuperscript{*}1.004 & \textsuperscript{*}1.265 & 3.561 & 0.925 & 0.068\\
        Min. Space = 30 & H-SMUCE (0.1) & \textsuperscript{}0.126 & \textsuperscript{}12.462 & \textsuperscript{\ddag}3.169 & \textsuperscript{}6.352 & 54.044 & 0.996 & 0.035\\
         & H-SMUCE (0.5) & \textsuperscript{}0.130 & \textsuperscript{}13.503 & \textsuperscript{}4.409 & \textsuperscript{}4.116 & 38.138 & 0.994 & 0.068\\
         & MOSUM BUM & \textsuperscript{}0.663 & \textsuperscript{}33.562 & \textsuperscript{}16.491 & \textsuperscript{}6.839 & 28.593 & 0.964 & 0.033\\
         & MOSUM LP & \textsuperscript{}0.514 & \textsuperscript{}25.474 & \textsuperscript{}12.415 & \textsuperscript{}7.337 & 18.502 & 0.984 & 0.033\\
         & PELT & \textsuperscript{\dag}0.096 & \textsuperscript{\dag}5.982 & \textsuperscript{\dag}2.773 & \textsuperscript{\dag}1.281 &  &  & 0.009\\
         & NOT & \textsuperscript{}0.302 & \textsuperscript{}15.276 & \textsuperscript{}7.756 & \textsuperscript{}3.064 &  &  & 0.138\\
         & NSP & \textsuperscript{}0.190 & \textsuperscript{} & \textsuperscript{} & \textsuperscript{} & 93.639 & 1.000 & 10.464\\ \hline
        $T =$ 500 & MICH-Auto & \textsuperscript{}0.237 & \textsuperscript{\ddag}12.585 & \textsuperscript{}2.083 & \textsuperscript{\ddag}2.365 & 2.401 & 0.930 & 3.533\\
        $J^* =$ 5 & MICH-Ora & \textsuperscript{*}0.049 & \textsuperscript{}11.135 & \textsuperscript{}1.510 & \textsuperscript{}2.038 & 2.331 & 0.930 & 0.538\\
        Min. Space = 30 & H-SMUCE (0.1) & \textsuperscript{}0.325 & \textsuperscript{}18.366 & \textsuperscript{}2.106 & \textsuperscript{}4.468 & 36.548 & 0.983 & 0.034\\
         & H-SMUCE (0.5) & \textsuperscript{\ddag}0.201 & \textsuperscript{}14.204 & \textsuperscript{\ddag}1.715 & \textsuperscript{}2.817 & 26.012 & 0.994 & 0.068\\
         & MOSUM BUM & \textsuperscript{}0.547 & \textsuperscript{}22.856 & \textsuperscript{}5.276 & \textsuperscript{}3.137 & 25.162 & 0.974 & 0.056\\
         & MOSUM LP & \textsuperscript{}0.676 & \textsuperscript{}30.033 & \textsuperscript{}5.882 & \textsuperscript{}6.025 & 16.151 & 0.994 & 0.045\\
         & PELT & \textsuperscript{\dag}0.156 & \textsuperscript{*\dag}8.090 & \textsuperscript{*\dag}1.443 & \textsuperscript{*\dag}1.276 &  &  & 0.009\\
         & NOT & \textsuperscript{}0.254 & \textsuperscript{}11.359 & \textsuperscript{}2.330 & \textsuperscript{}1.543 &  &  & 0.142\\
         & NSP & \textsuperscript{}1.479 & \textsuperscript{} & \textsuperscript{} & \textsuperscript{} & 71.942 & 1.000 & 13.756\\ \hline
        $T =$ 1000 & MICH-Auto & \textsuperscript{}0.135 & \textsuperscript{\ddag}16.079 & \textsuperscript{}7.369 & \textsuperscript{\ddag}5.904 & 5.730 & 0.907 & 3.460\\
        $J^* =$ 2 & MICH-Ora & \textsuperscript{*}0.013 & \textsuperscript{*}9.802 & \textsuperscript{*}2.711 & \textsuperscript{}3.804 & 5.369 & 0.910 & 0.156\\
        Min. Space = 50 & H-SMUCE (0.1) & \textsuperscript{}0.129 & \textsuperscript{}24.482 & \textsuperscript{\ddag}6.373 & \textsuperscript{}12.818 & 89.142 & 0.998 & 0.044\\
         & H-SMUCE (0.5) & \textsuperscript{}0.139 & \textsuperscript{}27.309 & \textsuperscript{}9.191 & \textsuperscript{}8.781 & 64.085 & 0.997 & 0.076\\
         & MOSUM BUM & \textsuperscript{}0.313 & \textsuperscript{}44.259 & \textsuperscript{}18.230 & \textsuperscript{}12.492 & 43.682 & 0.978 & 0.060\\
         & MOSUM LP & \textsuperscript{}0.742 & \textsuperscript{}64.188 & \textsuperscript{}34.068 & \textsuperscript{}19.772 & 26.329 & 0.979 & 0.053\\
         & PELT & \textsuperscript{\dag}0.088 & \textsuperscript{\dag}11.281 & \textsuperscript{\dag}5.307 & \textsuperscript{*\dag}2.612 &  &  & 0.017\\
         & NOT & \textsuperscript{}0.309 & \textsuperscript{}28.693 & \textsuperscript{}15.649 & \textsuperscript{}5.909 &  &  & 0.243\\
         & NSP & \textsuperscript{\ddag}0.128 & \textsuperscript{} & \textsuperscript{} & \textsuperscript{} & 144.278 & 1.000 & 16.494\\ \hline
        $T =$ 1000 & MICH-Auto & \textsuperscript{}0.275 & \textsuperscript{\ddag}24.141 & \textsuperscript{}4.621 & \textsuperscript{\ddag}4.634 & 3.758 & 0.921 & 6.587\\
        $J^* =$ 5 & MICH-Ora & \textsuperscript{*}0.070 & \textsuperscript{}22.931 & \textsuperscript{}3.263 & \textsuperscript{}4.818 & 3.573 & 0.923 & 1.143\\
        Min. Space = 50 & H-SMUCE (0.1) & \textsuperscript{}0.244 & \textsuperscript{}29.938 & \textsuperscript{\ddag}3.327 & \textsuperscript{}7.040 & 56.347 & 0.995 & 0.045\\
         & H-SMUCE (0.5) & \textsuperscript{\ddag}0.203 & \textsuperscript{}27.533 & \textsuperscript{}3.484 & \textsuperscript{}5.558 & 39.433 & 0.997 & 0.077\\
         & MOSUM BUM & \textsuperscript{}0.291 & \textsuperscript{}34.248 & \textsuperscript{}6.233 & \textsuperscript{}5.667 & 47.208 & 0.983 & 0.121\\
         & MOSUM LP & \textsuperscript{}0.871 & \textsuperscript{}68.418 & \textsuperscript{}14.877 & \textsuperscript{}14.459 & 22.882 & 0.989 & 0.068\\
         & PELT & \textsuperscript{\dag}0.141 & \textsuperscript{*\dag}14.239 & \textsuperscript{*\dag}2.747 & \textsuperscript{*\dag}2.263 &  &  & 0.017\\
         & NOT & \textsuperscript{}0.288 & \textsuperscript{}21.123 & \textsuperscript{}5.149 & \textsuperscript{}2.854 &  &  & 0.249\\
         & NSP & \textsuperscript{}0.883 & \textsuperscript{} & \textsuperscript{} & \textsuperscript{} & 106.242 & 1.000 & 21.354\\ \hline
        $T =$ 1000 & MICH-Auto & \textsuperscript{}0.426 & \textsuperscript{}24.651 & \textsuperscript{}2.486 & \textsuperscript{}2.870 & 2.631 & 0.928 & 14.026\\
        $J^* =$ 10 & MICH-Ora & \textsuperscript{}0.141 & \textsuperscript{}30.665 & \textsuperscript{}2.551 & \textsuperscript{}3.474 & 2.612 & 0.929 & 4.776\\
        Min. Space = 50 & H-SMUCE (0.1) & \textsuperscript{}0.321 & \textsuperscript{}24.770 & \textsuperscript{}1.662 & \textsuperscript{}3.002 & 44.139 & 0.984 & 0.047\\
         & H-SMUCE (0.5) & \textsuperscript{*\dag\ddag}0.114 & \textsuperscript{\ddag}13.259 & \textsuperscript{*\dag\ddag}0.916 & \textsuperscript{\ddag}1.266 & 31.837 & 0.998 & 0.078\\
         & MOSUM BUM & \textsuperscript{}0.234 & \textsuperscript{}22.338 & \textsuperscript{}2.048 & \textsuperscript{}2.307 & 79.205 & 0.999 & 0.210\\
         & MOSUM LP & \textsuperscript{}1.007 & \textsuperscript{}60.334 & \textsuperscript{}6.879 & \textsuperscript{}8.521 & 21.544 & 0.996 & 0.091\\
         & PELT & \textsuperscript{}0.133 & \textsuperscript{*\dag}9.068 & \textsuperscript{}0.997 & \textsuperscript{*\dag}0.723 &  &  & 0.017\\
         & NOT & \textsuperscript{}0.210 & \textsuperscript{}11.856 & \textsuperscript{}1.384 & \textsuperscript{}0.898 &  &  & 0.250\\
         & NSP & \textsuperscript{}3.251 & \textsuperscript{} & \textsuperscript{} & \textsuperscript{} & 88.508 & 1.000 & 26.815\\ \hline\hline
    \end{tabular}
    \caption{\footnotesize \textbf{Simulation \ref{sim:main} Results with MA(2) Errors -- Low Dependence}. Average statistics over 5,000 replicates of $y_t$ generated according to MA(2) process (\ref{eq:dep_sim}) with $\theta = 0.3$. MICH credible sets are constructed for $\alpha = 0.1$ and detections are determined with $\delta = 1.1$. Key: \textsuperscript{*}Best overall. \textsuperscript{\textdagger}Best excluding MICH-Ora. \textsuperscript{\ddag}Best among methods with uncertainty quantification excluding MICH-Ora.}
    \label{tab:ma3_sim}
\end{table}

Figure \ref{fig:ma3_sim_plot} and Table \ref{tab:ma3_sim} show the results of this study in the low-dependence setting where $\theta = 0.3$. Again, MOSUM and NSP lag behind the other methods in the study, displaying larger bias across every setting where $J^* > 2$ and consistently having the worst localization errors in terms of the Hausdorff, FPSLE, and FNSLE statistics. MICH-Auto appears to be somewhat robust to the presence of dependence in this study, being unbiased in the median simulation and achieving the desired nominal coverage while still returning much smaller credible sets. MICH-Auto also has the lowest average bias (or is off by a negligible factor) in all settings except for when $T=1000$ and $J^*=10$. As before, the bias of MICH-Auto is notably lower than H-SMUCE when $\alpha = 0.1$ and $T=100$ (when $J^* = 5$, H-SMUCE is biased by nearly 2 whole change-points and does not even achieve 80\% coverage). Notably, the gap between the bias of MICH and PELT has narrowed as well, and as in the heavy-tailed simulation, the bias of NOT is heavily inflated in the presence of correlated data. Across all settings except the last where $T = 1000$ and $J^*=10$, MICH-Auto has the lowest FNSLE among methods that provide uncertainty quantification, and despite not providing any guarantees for false discovery control, the FPSLE of MICH-Auto is on par with H-SMUCE across each setting, even when $\alpha = 0.5$.

\begin{figure}[!h]
    \centering
    \includegraphics[scale =0.42]{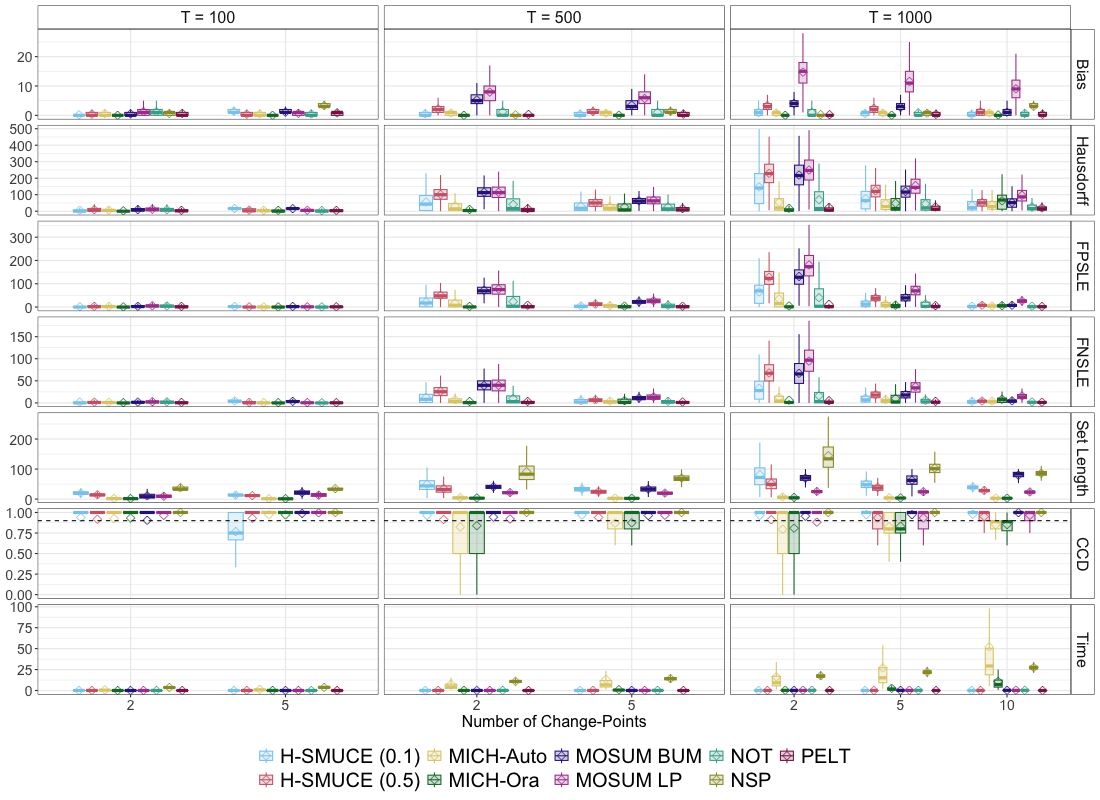}
    \caption{\small \textbf{Simulation \ref{sim:main} Results with MA(2) Errors -- High Dependence}. Box-plots of evaluation metrics for 5,000 replicates of $y_t$ generated according to (\ref{eq:heavy_sim}) with MA(2) errors (\ref{eq:dep_sim}) and $\theta = 0.7$. Diamonds ($\Diamond$) display mean of each statistic. Bias $|J^* - J|$ assesses each model's ability to estimate correct number of changes (lower is better). FPSLE and FNSLE assess each model's ability to accurately estimate the locations of the changes (lower is better). Set Length and CCD report the average size and coverage of credible/confidence sets for methods that provide uncertainty quantification (dashed line indicates nominal coverage level for $\alpha = 0.1$). Time reports run-time for each method in seconds.}
    \label{fig:ma7_sim_plot}
\end{figure}

\begin{table}[htbp!] 
    \scriptsize
    \centering
    \begin{tabular}{l || l || r r r r | r r | r}
        \multicolumn{1}{c||}{Setting} & \multicolumn{1}{c||}{Method} & $|J^* - J|$ & Hausdorff & FPSLE & FNSLE & CI Len. & CCD & Time (s) \\ \hline\hline 
        $T =$ 100 & MICH-Auto & \textsuperscript{}0.649 & \textsuperscript{}6.945 & \textsuperscript{}3.019 & \textsuperscript{\ddag}1.541 & 2.186 & 0.927 & 0.799\\
        $J^* =$ 2 & MICH-Ora & \textsuperscript{*}0.002 & \textsuperscript{*}2.366 & \textsuperscript{*}0.534 & \textsuperscript{*}0.603 & 1.616 & 0.935 & 0.017\\
        Min. Space = 15 & H-SMUCE (0.1) & \textsuperscript{\dag\ddag}0.183 & \textsuperscript{\ddag}5.697 & \textsuperscript{\dag\ddag}1.459 & \textsuperscript{}2.045 & 20.591 & 0.960 & 0.021\\
         & H-SMUCE (0.5) & \textsuperscript{}0.422 & \textsuperscript{}9.174 & \textsuperscript{}3.023 & \textsuperscript{}1.903 & 14.437 & 0.922 & 0.059\\
         & MOSUM BUM & \textsuperscript{}0.526 & \textsuperscript{}9.584 & \textsuperscript{}3.415 & \textsuperscript{}1.944 & 12.063 & 0.904 & 0.018\\
         & MOSUM LP & \textsuperscript{}1.236 & \textsuperscript{}11.757 & \textsuperscript{}5.304 & \textsuperscript{}2.817 & 10.499 & 0.965 & 0.018\\
         & PELT & \textsuperscript{}0.660 & \textsuperscript{\dag}5.227 & \textsuperscript{}2.752 & \textsuperscript{\dag}1.077 &  &  & 0.002\\
         & NOT & \textsuperscript{}1.244 & \textsuperscript{}9.390 & \textsuperscript{}4.630 & \textsuperscript{}2.383 &  &  & 0.050\\
         & NSP & \textsuperscript{}0.627 & \textsuperscript{} & \textsuperscript{} & \textsuperscript{} & 37.036 & 1.000 & 3.631\\ \hline
        $T =$ 100 & MICH-Auto & \textsuperscript{}0.477 & \textsuperscript{\ddag}5.184 & \textsuperscript{}0.841 & \textsuperscript{\ddag}0.834 & 1.263 & 0.972 & 1.350\\
        $J^* =$ 5 & MICH-Ora & \textsuperscript{*}0.024 & \textsuperscript{}4.999 & \textsuperscript{*}0.619 & \textsuperscript{}0.765 & 1.217 & 0.971 & 0.122\\
        Min. Space = 15 & H-SMUCE (0.1) & \textsuperscript{}1.384 & \textsuperscript{}16.187 & \textsuperscript{}2.777 & \textsuperscript{}4.547 & 13.575 & 0.802 & 0.022\\
         & H-SMUCE (0.5) & \textsuperscript{\dag\ddag}0.392 & \textsuperscript{}6.963 & \textsuperscript{\ddag}0.839 & \textsuperscript{}1.331 & 11.624 & 0.940 & 0.058\\
         & MOSUM BUM & \textsuperscript{}1.095 & \textsuperscript{}15.245 & \textsuperscript{}2.388 & \textsuperscript{}3.586 & 21.753 & 0.990 & 0.021\\
         & MOSUM LP & \textsuperscript{}0.700 & \textsuperscript{}6.101 & \textsuperscript{}1.140 & \textsuperscript{}0.865 & 13.450 & 0.996 & 0.024\\
         & PELT & \textsuperscript{}0.854 & \textsuperscript{}4.632 & \textsuperscript{}1.103 & \textsuperscript{}0.591 &  &  & 0.003\\
         & NOT & \textsuperscript{}0.459 & \textsuperscript{*\dag}3.800 & \textsuperscript{\dag}0.664 & \textsuperscript{*\dag}0.503 &  &  & 0.050\\
         & NSP & \textsuperscript{}3.348 & \textsuperscript{} & \textsuperscript{} & \textsuperscript{} & 35.082 & 1.000 & 3.750\\ \hline
        $T =$ 500 & MICH-Auto & \textsuperscript{}0.771 & \textsuperscript{\ddag}31.421 & \textsuperscript{\ddag}17.343 & \textsuperscript{\ddag}7.289 & 5.438 & 0.826 & 7.655\\
        $J^* =$ 2 & MICH-Ora & \textsuperscript{*}0.005 & \textsuperscript{*}10.715 & \textsuperscript{*}2.872 & \textsuperscript{*}3.309 & 3.818 & 0.842 & 0.130\\
        Min. Space = 30 & H-SMUCE (0.1) & \textsuperscript{}0.598 & \textsuperscript{}55.754 & \textsuperscript{}21.987 & \textsuperscript{}12.754 & 49.851 & 0.967 & 0.033\\
         & H-SMUCE (0.5) & \textsuperscript{}2.173 & \textsuperscript{}101.424 & \textsuperscript{}50.963 & \textsuperscript{}26.170 & 35.314 & 0.916 & 0.067\\
         & MOSUM BUM & \textsuperscript{}5.252 & \textsuperscript{}115.660 & \textsuperscript{}72.317 & \textsuperscript{}39.368 & 41.552 & 0.943 & 0.069\\
         & MOSUM LP & \textsuperscript{}7.923 & \textsuperscript{}114.569 & \textsuperscript{}77.307 & \textsuperscript{}39.577 & 21.692 & 0.923 & 0.067\\
         & PELT & \textsuperscript{}0.312 & \textsuperscript{\dag}15.736 & \textsuperscript{\dag}8.058 & \textsuperscript{\dag}3.349 &  &  & 0.009\\
         & NOT & \textsuperscript{}1.229 & \textsuperscript{}41.383 & \textsuperscript{}23.690 & \textsuperscript{}10.016 &  &  & 0.140\\
         & NSP & \textsuperscript{\dag\ddag}0.187 & \textsuperscript{} & \textsuperscript{} & \textsuperscript{} & 90.458 & 1.000 & 10.869\\ \hline
        $T =$ 500 & MICH-Auto & \textsuperscript{}0.727 & \textsuperscript{\ddag}26.750 & \textsuperscript{}6.203 & \textsuperscript{\ddag}4.144 & 3.121 & 0.871 & 13.441\\
        $J^* =$ 5 & MICH-Ora & \textsuperscript{*}0.042 & \textsuperscript{}26.191 & \textsuperscript{}4.134 & \textsuperscript{}5.120 & 2.727 & 0.874 & 0.964\\
        Min. Space = 30 & H-SMUCE (0.1) & \textsuperscript{\dag\ddag}0.385 & \textsuperscript{}30.606 & \textsuperscript{\ddag}5.018 & \textsuperscript{}5.008 & 33.757 & 0.972 & 0.034\\
         & H-SMUCE (0.5) & \textsuperscript{}1.267 & \textsuperscript{}51.333 & \textsuperscript{}12.698 & \textsuperscript{}6.821 & 25.297 & 0.940 & 0.067\\
         & MOSUM BUM & \textsuperscript{}3.476 & \textsuperscript{}62.738 & \textsuperscript{}22.576 & \textsuperscript{}11.011 & 34.266 & 0.963 & 0.076\\
         & MOSUM LP & \textsuperscript{}5.726 & \textsuperscript{}67.765 & \textsuperscript{}27.195 & \textsuperscript{}13.782 & 20.206 & 0.967 & 0.069\\
         & PELT & \textsuperscript{}0.429 & \textsuperscript{*\dag}15.934 & \textsuperscript{*\dag}3.642 & \textsuperscript{*\dag}2.135 &  &  & 0.009\\
         & NOT & \textsuperscript{}1.034 & \textsuperscript{}25.837 & \textsuperscript{}7.444 & \textsuperscript{}3.623 &  &  & 0.144\\
         & NSP & \textsuperscript{}1.338 & \textsuperscript{} & \textsuperscript{} & \textsuperscript{} & 69.955 & 1.000 & 14.180\\ \hline
        $T =$ 1000 & MICH-Auto & \textsuperscript{}0.809 & \textsuperscript{\ddag}54.233 & \textsuperscript{\ddag}34.467 & \textsuperscript{\ddag}12.779 & 7.590 & 0.800 & 16.332\\
        $J^* =$ 2 & MICH-Ora & \textsuperscript{*}0.008 & \textsuperscript{*}17.220 & \textsuperscript{*}4.823 & \textsuperscript{}6.128 & 5.385 & 0.812 & 0.351\\
        Min. Space = 50 & H-SMUCE (0.1) & \textsuperscript{}1.007 & \textsuperscript{}147.138 & \textsuperscript{}64.881 & \textsuperscript{}33.285 & 82.532 & 0.967 & 0.044\\
         & H-SMUCE (0.5) & \textsuperscript{}3.373 & \textsuperscript{}229.631 & \textsuperscript{}128.264 & \textsuperscript{}67.854 & 52.581 & 0.914 & 0.078\\
         & MOSUM BUM & \textsuperscript{}4.037 & \textsuperscript{}218.963 & \textsuperscript{}131.434 & \textsuperscript{}66.943 & 71.812 & 0.954 & 0.130\\
         & MOSUM LP & \textsuperscript{}14.724 & \textsuperscript{}249.611 & \textsuperscript{}180.790 & \textsuperscript{}94.650 & 25.236 & 0.883 & 0.295\\
         & PELT & \textsuperscript{\dag}0.216 & \textsuperscript{\dag}23.093 & \textsuperscript{\dag}11.806 & \textsuperscript{*\dag}5.094 &  &  & 0.017\\
         & NOT & \textsuperscript{}0.984 & \textsuperscript{}69.270 & \textsuperscript{}40.036 & \textsuperscript{}16.226 &  &  & 0.246\\
         & NSP & \textsuperscript{\ddag}0.244 & \textsuperscript{} & \textsuperscript{} & \textsuperscript{} & 145.666 & 1.000 & 17.566\\ \hline
        $T =$ 1000 & MICH-Auto & \textsuperscript{}0.752 & \textsuperscript{\ddag}46.834 & \textsuperscript{\ddag}12.380 & \textsuperscript{\ddag}7.002 & 4.553 & 0.828 & 26.944\\
        $J^* =$ 5 & MICH-Ora & \textsuperscript{*}0.056 & \textsuperscript{}48.197 & \textsuperscript{}8.102 & \textsuperscript{}10.103 & 4.003 & 0.833 & 2.699\\
        Min. Space = 50 & H-SMUCE (0.1) & \textsuperscript{\ddag}0.612 & \textsuperscript{}76.134 & \textsuperscript{}15.667 & \textsuperscript{}10.590 & 51.115 & 0.976 & 0.044\\
         & H-SMUCE (0.5) & \textsuperscript{}2.398 & \textsuperscript{}124.001 & \textsuperscript{}38.747 & \textsuperscript{}18.840 & 38.674 & 0.936 & 0.077\\
         & MOSUM BUM & \textsuperscript{}2.751 & \textsuperscript{}119.866 & \textsuperscript{}40.524 & \textsuperscript{}19.161 & 64.335 & 0.973 & 0.166\\
         & MOSUM LP & \textsuperscript{}11.548 & \textsuperscript{}155.310 & \textsuperscript{}73.463 & \textsuperscript{}36.615 & 24.708 & 0.934 & 0.376\\
         & PELT & \textsuperscript{\dag}0.328 & \textsuperscript{*\dag}25.897 & \textsuperscript{*\dag}5.986 & \textsuperscript{*\dag}3.571 &  &  & 0.017\\
         & NOT & \textsuperscript{}0.870 & \textsuperscript{}44.822 & \textsuperscript{}13.281 & \textsuperscript{}6.086 &  &  & 0.252\\
         & NSP & \textsuperscript{}0.732 & \textsuperscript{} & \textsuperscript{} & \textsuperscript{} & 103.073 & 1.000 & 22.018\\ \hline
        $T =$ 1000 & MICH-Auto & \textsuperscript{}0.815 & \textsuperscript{}39.548 & \textsuperscript{}5.072 & \textsuperscript{}4.138 & 3.000 & 0.850 & 52.075\\
        $J^* =$ 10 & MICH-Ora & \textsuperscript{*}0.125 & \textsuperscript{}62.608 & \textsuperscript{}6.214 & \textsuperscript{}8.088 & 2.958 & 0.852 & 10.903\\
        Min. Space = 50 & H-SMUCE (0.1) & \textsuperscript{\dag\ddag}0.343 & \textsuperscript{\ddag}35.573 & \textsuperscript{\ddag}3.197 & \textsuperscript{\ddag}3.491 & 39.565 & 0.979 & 0.046\\
         & H-SMUCE (0.5) & \textsuperscript{}1.228 & \textsuperscript{}52.427 & \textsuperscript{}7.785 & \textsuperscript{}4.330 & 28.883 & 0.950 & 0.078\\
         & MOSUM BUM & \textsuperscript{}1.053 & \textsuperscript{}53.169 & \textsuperscript{}7.635 & \textsuperscript{}4.758 & 79.771 & 0.998 & 0.226\\
         & MOSUM LP & \textsuperscript{}9.038 & \textsuperscript{}96.257 & \textsuperscript{}26.425 & \textsuperscript{}15.748 & 23.804 & 0.966 & 0.393\\
         & PELT & \textsuperscript{}0.374 & \textsuperscript{*\dag}20.128 & \textsuperscript{*\dag}2.597 & \textsuperscript{*\dag}1.605 &  &  & 0.017\\
         & NOT & \textsuperscript{}0.619 & \textsuperscript{}25.452 & \textsuperscript{}3.680 & \textsuperscript{}2.130 &  &  & 0.253\\
         & NSP & \textsuperscript{}3.058 & \textsuperscript{} & \textsuperscript{} & \textsuperscript{} & 85.841 & 1.000 & 27.522\\ \hline\hline
    \end{tabular}
    \caption{\footnotesize \textbf{Simulation \ref{sim:main} Results with MA(2) Errors -- High Dependence}. Average statistics over 5,000 replicates of $y_t$ generated according to MA(2) process (\ref{eq:dep_sim}) with $\theta = 0.7$. MICH credible sets are constructed for $\alpha = 0.1$ and detections are determined with $\delta = 1.1$. Key: \textsuperscript{*}Best overall. \textsuperscript{\textdagger}Best excluding MICH-Ora. \textsuperscript{\ddag}Best among methods with uncertainty quantification excluding MICH-Ora.}
    \label{tab:ma7_sim}
\end{table}

Figure \ref{fig:ma7_sim_plot} and Table \ref{tab:ma7_sim} show the results of this study in the high-dependence setting where $\theta = 0.7$. Here we see the performance of all methods beginning to break down. Although MICH-Ora remains unbiased, the average bias of MICH-Auto is now between 0.5 and 1 change-points. Though the credible sets are larger than in the other settings, MICH-Auto now fails to meet the desired 0.9 coverage when $T>100$, indicating we should likely increase $\delta$ beyond 1.1 when the serial correlation in $\mathbf{y}_{1:T}$ is larger. Despite the break down in performance, the bias of MICH-Auto is still on the order of NOT and H-SMUCE (in the small-sample setting with $T=100$, the bias of MICH-Auto is on the order of PELT or even smaller, but PELT has roughly half the bias of MICH-Auto when $T \geq 500$). As before, the performance of H-SMUCE depends heavily on the choice of $\alpha$, only now the difference is even more pronounced, e.g. when $T =100$ and $J^* = 5$, both MICH-Auto and H-SMUCE with $\alpha = 0.5$ are biased by approximately half a change-point, but H-SMUCE with $\alpha = 0.1$ is biased by 1.5 change-points and when $T =1000$ and $J^* = 2$, both MICH-Auto and H-SMUCE with $\alpha = 0.1$ are biased by approximately a single change-point, but H-SMUCE with $\alpha = 0.1$ is biased by over 3 change-points. 

Among the methods that provide uncertainty quantification, MICH-Auto has the lowest or near lowest localization errors as measured by the FNSLE and FPSLE across each setting. This difference is particularly pronounced in some settings, e.g. when $J^*=2$ and $T=500$ or $T=1000$, the FNSLE and FPSLE for MICH-Auto are roughly half that of the next best method (H-SMUCE with $\alpha = 0.5$). While the localization error of PELT remains robust even in the presence of high serial correlation, MICH-Auto is now consistently in line with or beating NOT. To Summarize, MICH-Auto is robust in the presence of low serial correlation, but begins to suffer as the strength of the dependence increases. Despite this, MICH is consistently one of the most robust performers among the methods that provide uncertainty quantification. 

\subsection{Multivariate Simulation}

In this section we adapt Simulation \ref{sim:main} for the case of multivariate mean changes to test the performance of Algorithm \ref{alg:mich-multi}.

\setcounter{algocf}{1}
\begin{simulation}[!h]
\begin{enumerate}
\small
    \itemsep0em 
    \item Fix the number of observations $T$, the dimension $d$, the number of change-points $L^*$, the proportion of active dimensions $p \in (0,1]$, the minimum spacing condition $\Delta_T$, and a constant $C > 0$.
    \item Draw $\boldsymbol{\tau}_{1:L^*}$ uniformly from $[T]$ subject to the minimum spacing condition $|\tau_{\ell+1} - \tau_\ell| \geq \Delta_T$ with $\tau_0 = 1$ and $\tau_{L^*+1} = T+1$. 
    \item Draw $\{U_i\}_{i=1}^d \sim \text{Uniform}(-2,2)$ and set $s_i := 2^{U_i}$ and $\boldsymbol{\Sigma} = \text{diag}(\mathbf{s}_{1:d})$.
    \item Let $A$ be a set of $d_0:=\lfloor p d\rfloor$ active coordinates drawn uniformly at random from $[d]$.
    \item Set $\boldsymbol{\mu}_0 :=0$, and for each $i \in [d]$ draw $\xi_{\ell, i} \sim \text{Bernoulli}(0.5)$ and set: 
    \vspace{-5pt}
    \begin{align*}
        \mu_{\ell, i} &:= \mu_{\ell - 1,i}  + \frac{C(1-2\xi_{\ell,i})s_i \mathbbm{1}_{\{i \in A\}}}{\sqrt{\min\{\tau_{\ell+1} - \tau_\ell, \tau_\ell - \tau_{\ell-1}\}}}.
    \end{align*}
    \vspace{-15pt}
    \item Draw $\mathbf{y}_t \overset{\text{ind.}}{\sim} \mathcal{N}_d\left(\sum_{j=0}^{L^*} \boldsymbol{\mu}_\ell\mathbbm{1}_{\{\tau_\ell \leq t < \tau_{\ell+1}\}}, \boldsymbol{\Sigma}\right)$.
    \vspace{-10pt}
\normalsize
\end{enumerate}
\caption{Multivariate Mean Simulation Study}
\label{sim:multi}
\end{simulation}

For each replicate of Simulation \ref{sim:multi}, we fit the oracle version of MICH with $L = L^*$ (MICH-Ora) and the version of the model with $L$ selected based on the ELBO (MICH-Auto). In each case, we fit MICH with error tolerance $\epsilon=10^{-5}$, and use $\delta = 0.5$ and $\alpha =0.1$ for detecting changes and constructing credible sets. We compare the performance of MICH to three other methods capable of detecting multivariate mean changes, including: i) the E-Divisive method of \cite{James15}, ii) the Two-Way
MOSUM ($\ell^2$-HD) method of \cite{Li23}, and iii) the informative sparse projection (Inspect) method of \cite{Wang17}. To fit E-Divisive, we use the \texttt{e.divisive} function from the \textbf{ecp} package.\footnote{\url{https://CRAN.R-project.org/package=ecp}.} Unlike MICH, E-Divisive requires knowledge of $\Delta_T$, so we set the argument \texttt{min.size} $=\Delta_T$ (though setting \texttt{min.size = 1} would provide a fairer comparison to MICH). We set the penalty argument \texttt{alpha = 2} to restrict E-Divisive to searching for mean changes, and set \texttt{R = 499} for the maximum number of random permutations as per the examples in \cite{James15}. To fit $\ell^2$-HD, we use the \texttt{ts\_hdchange} and \texttt{hdchange} functions  from the \textbf{L2hdchange} package\footnote{\url{https://CRAN.R-project.org/package=L2hdchange}.} with the default arguments. Lastly, to fit Inspect, we use the \texttt{inspect} function from the \textbf{InspectChangepoint} package\footnote{\url{https://cran.r-project.org/package=InspectChangepoint}.} with the penalty argument \texttt{lambda} set to the recommended default $\log(0.5d\log T)$. Note that for some replicates of Simulation \ref{sim:main}, this penalty is too small and leads to a compilation error in \texttt{inspect}. In this case, we take the ad-hoc approach of incrementing \texttt{lambda} until \texttt{inspect} compiles.

\begin{figure}[!h]
    \centering
    \includegraphics[scale = 0.31]{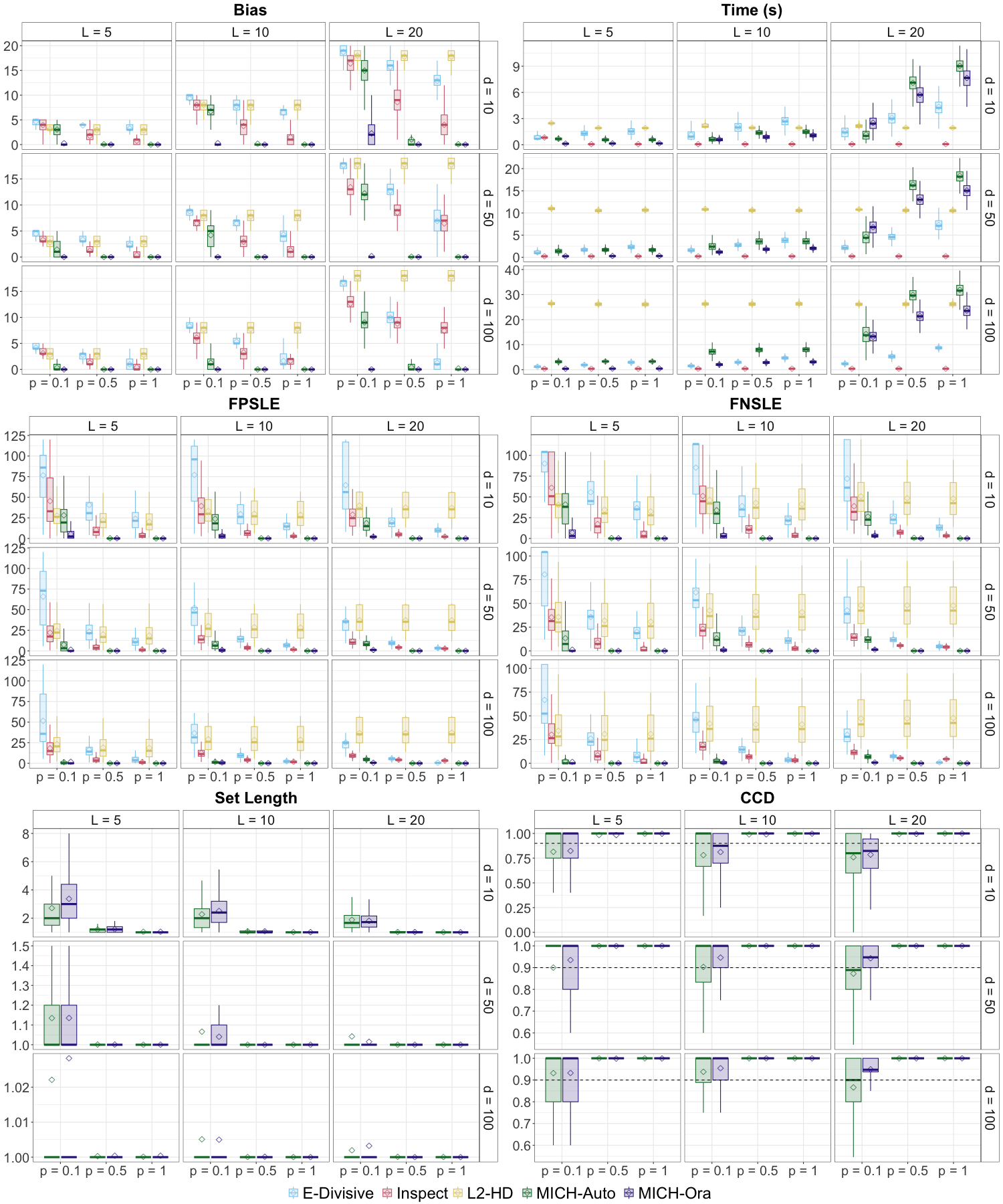}
    \caption{\small \textbf{Simulation \ref{sim:multi} Results}. Box-plots of evaluation metrics for 5,000 replicates of $\mathbf{y}_t$ from Simulation \ref{sim:multi} with $T = 250$, $\Delta_T = 10$, and $C = \sqrt{10}$. Diamonds ($\Diamond$) display mean of each statistic. Bias $|J^* - J|$ assesses each model's ability to estimate correct number of changes (lower is better). FPSLE and FNSLE assess each model's ability to accurately estimate the locations of the changes (lower is better). Set Length and CCD report the average size and coverage of MICH credible sets (dashed line indicates nominal coverage level for $\alpha = 0.1$). Time reports run-time for each method in seconds.}
    \label{fig:multi_sim_plot}
\end{figure}

Figure \ref{fig:multi_sim_plot} shows the results for 5,000 replicates of Simulation \ref{sim:multi} with $T = 250$, $\Delta_T = 10$, and $C = \sqrt{10}$, as well as for each combination of $d\in\{10,50,100\}$, $L^*\in\{5,10,20\}$, and $p \in \{0.1,0.5,1\}$. First note that as expected, MICH-Ora is unbiased across all settings and has virtually no localization error as measured by the FPSLE and FNSLE, even in the sparse setting with $p = 0.1$. Next, we see that outside of MICH-Ora, MICH-Auto consistently has the lowest bias, FPSLE, and FNSLE of all methods tested across each setting. As expected, the performance of MICH-Auto improves as both the dimension of the data $d$ and the proportion of active series $p$ increase. In fact, MICH-Auto is unbiased and has virtually no localization error for each value of $d$ once $p \geq 0.5$. Even in the sparse setting with $p=0.1$, MICH-Auto still beats the other methods, including Inspect, which is specifically designed to detect sparse changes. Somewhat surprising is the fact that the performance of $\ell^2$-HD, which is designed to detect dense signals, does not improve as $p$ increases. 

Lastly, we see that the computational cost of both MICH and E-Divisive is increasing with $d$ and $L^*$. The time needed to fit $\ell^2$-HD is constant across $L^*$ but increasing in $d$ and is typically much higher than the other methods only except when $L^*=20$. The time needed to fit Inspect is constant across each setting. Nonetheless, even when $L^* = 20$ and $d = 100$, both MICH-Auto and MICH-Ora still fit in under 30 seconds, and when $L^*=5$, they both run just as fast as Inspect. 

\subsubsection{Spatial Correlation}

Based on the estimation strategy for $\boldsymbol{\Sigma}$ outlined in Appendix \ref{app:empirical-bayes}, MICH should be robust to the presence of general correlation structures in $\boldsymbol{\Sigma}$. We test this empirically by repeating Simulation \ref{sim:multi} in with spatial correlation. We modify Step 3 of Simulation \ref{sim:multi} so that for some $\rho \in (-1,1)$ and any $i,j\in[d]$:
\begin{align}
    [\boldsymbol{\Sigma}]_{ij}:= \rho^{|i-j|}s_is_j \label{eq:spatial}
\end{align}

\begin{figure}[!h]
    \centering
    \includegraphics[scale = 0.31]{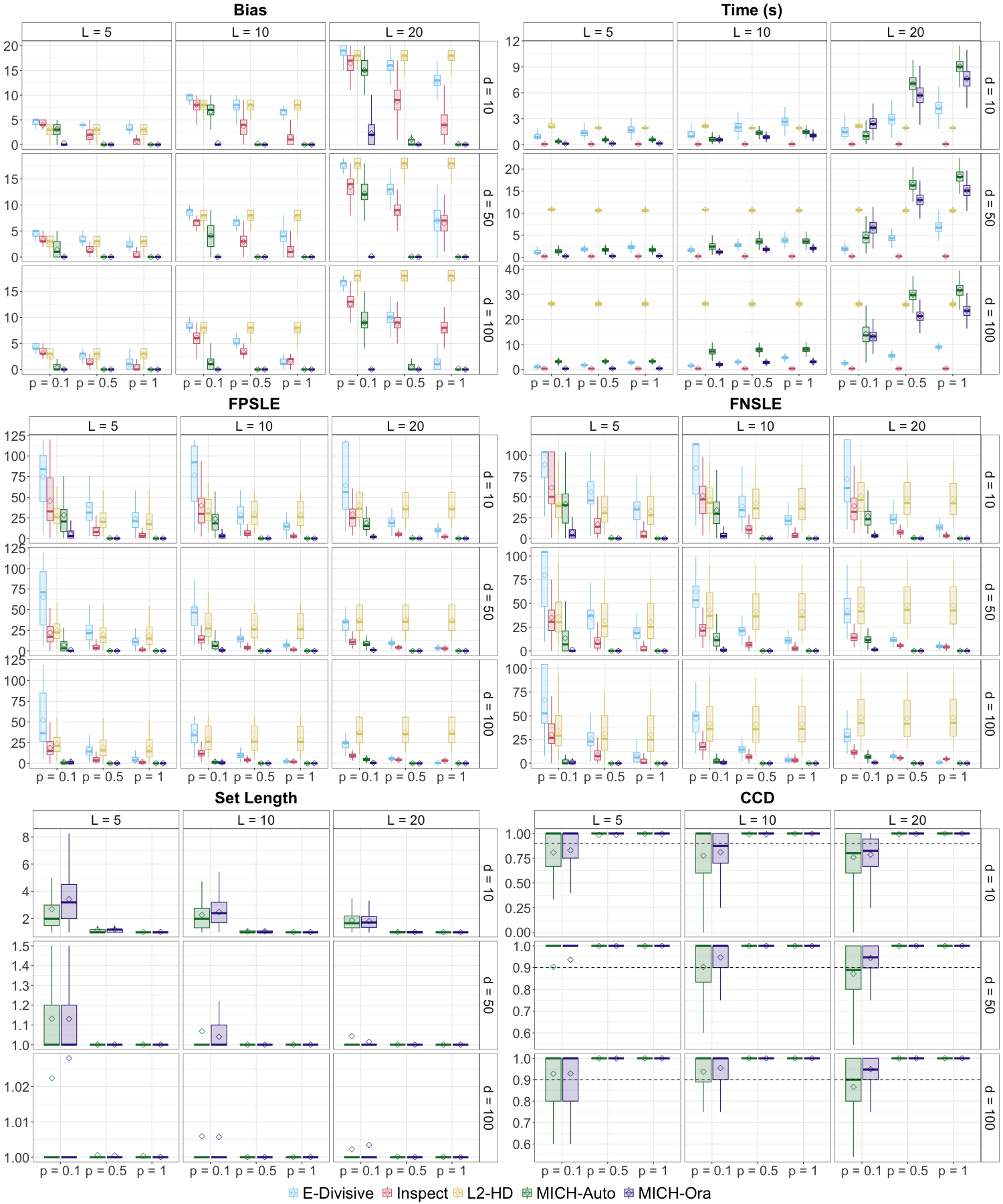}
    \caption{\small \textbf{Simulation \ref{sim:multi} Results -- Spatial Correlation}. Box-plots of evaluation metrics for 5,000 replicates of $\mathbf{y}_t$ from Simulation \ref{sim:multi} with $\boldsymbol{\Sigma}$ generated according to (\ref{eq:spatial}) with $T = 250$, $\Delta_T = 10$, and $C = \sqrt{10}$. Diamonds ($\Diamond$) display mean of each statistic. Bias $|J^* - J|$ assesses each model's ability to estimate correct number of changes (lower is better). FPSLE and FNSLE assess each model's ability to accurately estimate the locations of the changes (lower is better). Set Length and CCD report the average size and coverage of MICH credible sets (dashed line indicates nominal coverage level for $\alpha = 0.1$). Time reports run-time for each method in seconds.}
    \label{fig:corr_multi_sim_plot}
\end{figure}

Figure \ref{fig:corr_multi_sim_plot} shows the results for 5,000 replicates from Simulation \ref{sim:multi} under (\ref{eq:spatial}) with $\rho = 0.7$, $T = 250$, $\Delta_T = 10$, and $C = \sqrt{20}$. Overall, each method we test is robust to the presence of spatial correlation and our results do not noticeably differ from the results of Simulation \ref{sim:multi} with $\boldsymbol{\Sigma} = \text{diag}(\mathbf{s}_{1:d})$.

\subsubsection{Vanishing Signals}

As noted in Section \ref{sec:smscp-theory}, the mean-scp model uses a $\ell^2$-norm based aggregation to measure signal strength, meaning that MICH should be best suited for dense signals with many small jumps. To test this, we modify Step 5 of Simulation \ref{sim:multi} so that the jumps in $\boldsymbol{\mu}_{1:T}$ are decreasing as the dimension $d$ and the proportion of active series $p$ increase:
\begin{align}
    \mu_{\ell, i} &:= \mu_{\ell - 1,i}  + \frac{C(1-2\xi_{\ell,i})s_i \mathbbm{1}_{\{i \in A\}}}{\sqrt{d_0\min\{\tau_{\ell+1} - \tau_\ell, \tau_\ell - \tau_{\ell-1}\}}}. \label{eq:adaptive}
\end{align}

\begin{figure}[!h]
    \centering
    \includegraphics[scale = 0.31]{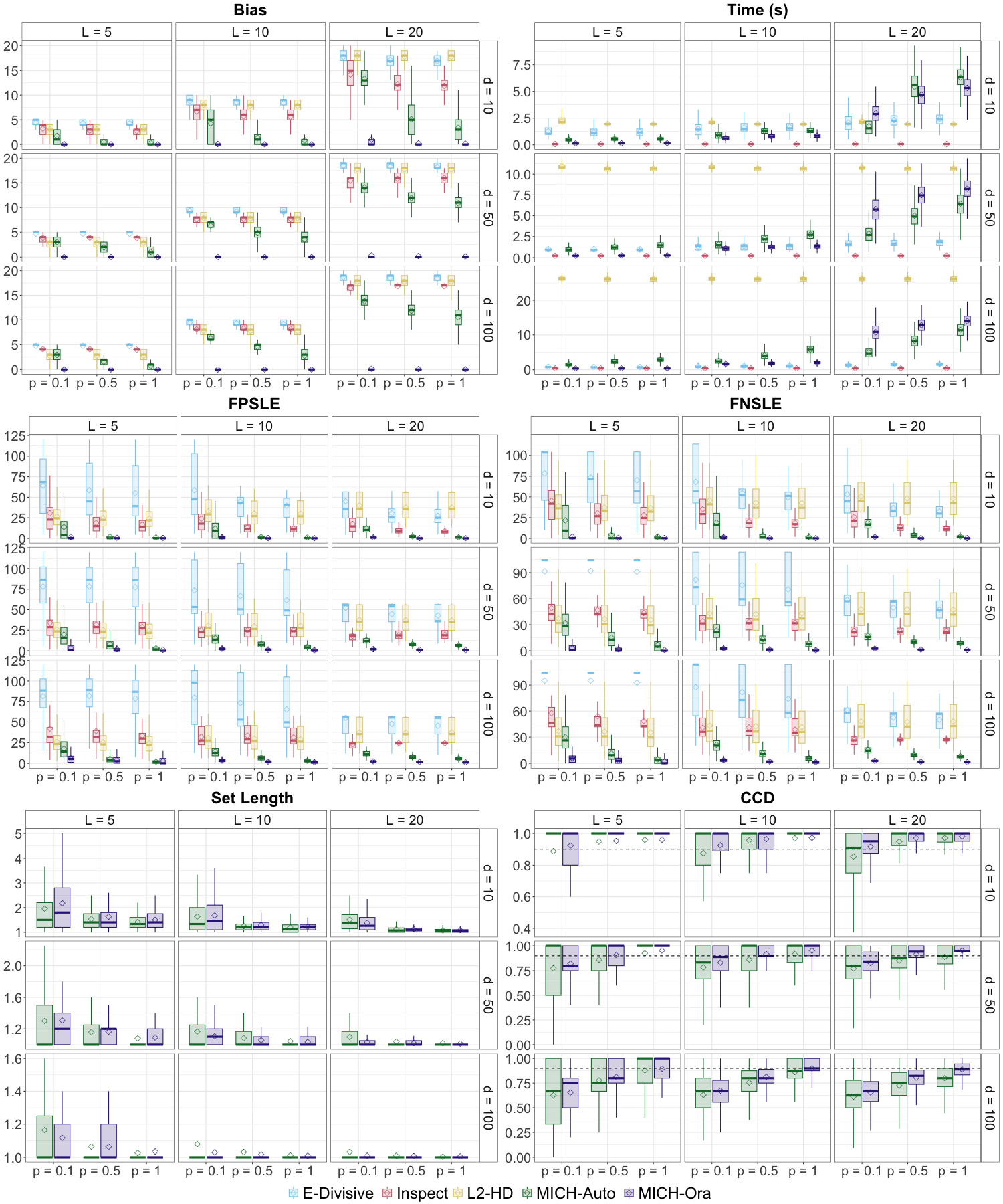}
    \caption{\small \textbf{Simulation \ref{sim:multi} Results -- Vanishing Signal}. Box-plots of evaluation metrics for 5,000 replicates of $\mathbf{y}_t$ from Simulation \ref{sim:multi} with $\boldsymbol{\mu}_{1:T}$ generated according to (\ref{eq:adaptive}) with $T = 250$, $\Delta_T = 10$, and $C = 5$. Diamonds ($\Diamond$) display mean of each statistic. Bias $|J^* - J|$ assesses each model's ability to estimate correct number of changes (lower is better). FPSLE and FNSLE assess each model's ability to accurately estimate the locations of the changes (lower is better). Set Length and CCD report the average size and coverage of MICH credible sets (dashed line indicates nominal coverage level for $\alpha = 0.1$). Time reports run-time for each method in seconds.}
    \label{fig:adapt_multi_sim_plot}
\end{figure}

Figure \ref{fig:adapt_multi_sim_plot} shows the results for 5,000 replicates from Simulation \ref{sim:multi} under (\ref{eq:adaptive}) with $T = 250$, $\Delta_T = 10$, and $C = 5$. As in Figure \ref{fig:multi_sim_plot}, we see that MICH outperforms the other methods across all metrics (though the bias of MICH is now roughly on the same order as $\ell^2$-HD for small $p$). In Figure \ref{fig:multi_sim_plot}, we saw the performance of MICH, Inspect and E-Divisive all improved as $p$ increased. Under (\ref{eq:adaptive}), the bias, FPSLE, and FNSLE for MICH-Auto are still decreasing as $p$ increases, but Figure \ref{fig:adapt_multi_sim_plot} shows that this is no longer the case for either Inspect or E-Divisive. This provides evidence for our hypothesis that MICH is best suited for detecting many small jumps in $\boldsymbol{\mu}_{1:T}$ among the methods we test. Notably, the evaluation metrics for MICH-Auto and Inspect were roughly on the same order in Figure \ref{fig:multi_sim_plot} for $p \geq 0.5$, but under (\ref{eq:adaptive}) we see a gap opening up between Inspect and MICH, with the performance on Inspect suffering much more in the presence of a vanishing signal. 

\subsection{Sensitivity Analysis}

\subsubsection{Sensitivity Analysis: Prior Parameters}
\label{app:prior_sensitivity}

To evaluate the effect of the prior parameters $\omega_0$, $u_0$, and $v_0$ on the output of MICH, we repeat Simulation \ref{sim:main} but now vary $\omega_0, u_0, v_0 \in \{10^{-1}, 10^{-3}, 10^{-5}\}$ when fitting MICH. Table \ref{tab:ora_prior_sa} shows the results of this study for the oracle version of MICH with $J=J^*$. In every setting, every metric is virtually unchanged regardless of the value for $\omega_0$, $u_0$, and $v_0$, indicating that Algorithm \ref{alg:mich} is not sensitive to the choice of $\omega_0$, $u_0$, $v_0$ for reasonably small values. Table \ref{tab:auto_prior_sa} shows similar results for this study using the automatic version of MICH where $J$ is selected using the value of the ELBO (see Appendix \ref{app:LKJ-choice}). There is a small decrease in performance for too small or large values of $\omega_0$, $u_0$, and $v_0$, but the magnitude of the difference is not significant, e.g. the difference in bias is less than half a change-point at most, and the evaluation metrics are within a standard error of each other (not displayed). Thus, the choice of prior parameters $\omega_0$, $u_0$, and $v_0$ does not seem make a meaningful difference and the default value of $10^{-3}$ should be appropriate in most use cases.

\subsubsection{Sensitivity Analysis: Detection Rule}
\label{app:delta_sensitivity}

According to the detection rule introduced in Section \ref{sec:cred-sets}, the $i^{\text{th}}$ component of MICH detects a change-point at significance level $\alpha\in(0,1)$ if $|\mathcal{CS}(\alpha, \overline{\boldsymbol{\pi}}_{1:T,i})| \leq \log^{1+\delta} T$ for some $\delta > 0$, where $\mathcal{CS}(\alpha, \overline{\boldsymbol{\pi}}_{1:T,i})$ is the $\alpha$ level credible set for the $i^{\text{th}}$ model component. Corollary \ref{cor:cred-sets} states that in the single change-point setting, as $T$ increases, the false negative rate (FNR) converges to zero for any $\delta > 0$, so long as the elements of $\mathbf{y}_{1:T}$ are independent. For dependent data, the FNR still goes to zero for $\delta > 1$. The analyst can test for dependence in the data and choose $\delta$ accordingly, or simply set $\delta > 1$, which is robust to dependence in $\mathbf{y}_{1:T}$ and will always result in the FNR going to zero as $T\to \infty$ in the single change-point setting. To evaluate the effect of $\delta$ on the output of MICH, we repeat the simulation study in Section \ref{sec:simulations} but now vary $\delta \in \{0.1, 0.5, 1.1\}$.

We should expect false negatives, and thus the FNSLE, to decrease as $\delta$ increases, thereby making the detection rule less restrictive. Tables \ref{tab:ora_delta_sa} and \ref{tab:auto_delta_sa} show that this is the case when $\delta$ increases from 0.1 to 0.5, but also that the FNSLE is essentially unchanged when $\delta$ increases from 0.5 to 1.1. We also see that increasing $\delta$ does not increase the number of false positives in any meaningful way. The FPSLE actually decreases when $\delta$ increases from $0.1$ to $0.5$ and only slightly increases when $\delta$ increases from $0.5$ to $1.1$ for MICH-Auto. In each case cases, the coverage conditional on detection is virtually unaffected by the choice of detection rule and the credible sets returned by MICH only increase slightly in size as $\delta$ increases. We chose the moderately robust $\delta = 0.5$ when fitting MICH in Simulation \ref{sim:main} seeing as the values of $\mathbf{y}_{1:T}$ are independent, but the results in Tables \ref{tab:ora_delta_sa} and \ref{tab:auto_delta_sa} show that the performance of MICH is virtually unchanged under the robust choice $\delta=1.1$. In fact, when $J=J^*$, Table \ref{tab:ora_delta_sa} shows that the oracle version of MICH tends performs the best with $\delta=1.1$. This suggest that a MICH user that is uncertain about the dependence structure of $\mathbf{y}_{1:T}$ can simply set the default $\delta = 1.1$ without necessarily having to worry about a tradeoff with performance. 

\subsubsection{Sensitivity Analysis: Convergence Criterion}
\label{app:tol_sensitivity}

Algorithm \ref{alg:mich} stops once the percentage increase in the ELBO falls below $100\epsilon\%$ for some $\epsilon > 0$ (see Appendix \ref{app:convergence}). To evaluate the effect of $\epsilon$ on the output of MICH, we repeat Simulation \ref{sim:main} but now vary $\epsilon \in \{10^{-3}, 10^{-5}, 10^{-7}, 10^{-10}\}$ when fitting MICH. Tables \ref{tab:ora_tol_sa} and \ref{tab:auto_tol_sa} show the respective results of this study for the oracle version of MICH with $J=J^*$ and the automatic version of MICH where $J$ is selected using the value of the ELBO. As one might expect, the performance of MICH generally improves as $\epsilon$ decreases, but the number of iterations and thus time until convergence increases. This is particularly true for larger values of $J^*$. However, Tables \ref{tab:ora_tol_sa} and \ref{tab:auto_tol_sa} also show that the gains from decreasing $\epsilon$ are diminishing, as the relative improvement of each metric is greater going from $\epsilon = 10^{-3}$ to $\epsilon = 10^{-5}$ than from $\epsilon = 10^{-7}$ to $\epsilon = 10^{-10}$, even though the increase in time to convergence is of the same order in either case. 
\newpage
\section{MICH for Poisson Data}
\label{app:poisson}

Consider the following model for a change-point in the rate of a Poisson sequence:
\begin{align}
    y_t\:|\: \lambda_t  &\overset{\text{ind.}}{\sim} \text{Poisson}(\lambda_t), \label{eq:pois-y}\\
    \lambda_t &:= \omega_t s^{\mathbbm{1}\{t\geq\tau\}}, \\
    \tau &\overset{\text{ind.}}{\sim} \text{Categorical}(\boldsymbol{\pi}_{1:T}),   \\
    s &\overset{\text{ind.}}{\sim} \text{Gamma}(u_0,v_0). \label{eq:pois-s}
    \vspace{-5pt}
\end{align}
As with the var-scp and meanvar-scp models, $\boldsymbol{\omega}_{1:T}$ represents the known component of $\boldsymbol{\lambda}_{1:T}$. Under the model (\ref{eq:pois-y})-(\ref{eq:pois-s}), we have: 
\begin{align}
    s \:|\:\mathbf{y}_{1:T},\tau = t &\sim \text{Gamma}(\overline{u}_{t}, \overline{v}_t), \\
    \tau \:|\:\mathbf{y}_{1:T} &\sim \text{Categorical}(\overline{\boldsymbol{\pi}}_{1:T}),  \\
    \overline{u}_{t} &= u_0 + \sum_{t=\tau}^T y_t,  \label{eq:pois-u-bar} \\
    \overline{v}_{t} &= u_0 + \sum_{t=\tau}^T \omega_t, \\
    \overline{\pi}_t &\propto \frac{\pi_t\Gamma(\overline{u}_t)}{\overline{v}_t^{\overline{u}_t}\exp[\sum_{t'=1}^{t-1}\omega_{t'}]}. \label{eq:pois-pi-bar}
\end{align}
We can easily generalize to multiple change-points by once again adopting the MICH framework:
\begin{align}
    y_t\:|\: \lambda_t  &\overset{\text{ind.}}{\sim} \text{Poisson}(\lambda_t), \label{eq:multi-pois-y} \\
    \lambda_t &:= \prod_{\ell=1}^{L} \lambda_{\ell t} := \prod_{\ell=1}^{L} s_{\ell}^{\mathbbm{1}\{t\geq\tau_{\ell}\}} \\
    \tau_\ell &\overset{\text{ind.}}{\sim} \text{Categorical}(\boldsymbol{\pi}_{1:T}),   \\
    s_\ell &\overset{\text{ind.}}{\sim} \text{Gamma}(u_0,v_0). \label{eq:multi-pois-s} 
     \vspace{-5pt}
\end{align}
Defining a function a function $\texttt{poiss-scp}(\mathbf{y}_{1:T};\boldsymbol{\omega}_{1:T}, u_0,v_0, \boldsymbol{\pi}_{1:T})$ that returns the posterior parameters (\ref{eq:pois-u-bar})-(\ref{eq:pois-pi-bar}), we can once again design a simple backfitting algorithm to efficiently approximate the posterior of model (\ref{eq:multi-pois-y})-(\ref{eq:multi-pois-s}). Algorithm \ref{alg:mich-poisson} shows in detail how we modify our approach for count data. 

\begin{algorithm}
\label{alg:mich-poisson}
\caption{Variational Bayes Approximation to MICH Posterior for Poisson Data.}

\footnotesize
\SetAlgoLined
  Inputs: $\mathbf{y}_{1:T},\:L,\:u_0,v_0,\boldsymbol{\pi}_{1:T}\:;$ \\
  Initialize: $\lambda_0,\:\{\overline{\lambda}_{\ell,1:T},\overline{\boldsymbol{\theta}}_\ell\}_{\ell=1}^L\:$;
  
  \Repeat {Convergence} {
    \For{$\ell=1$ \KwTo $L$} {
      $\overline{\lambda}_{-\ell t} := \prod_{\ell' \neq \ell}^L \overline{\lambda}_{\ell' t}$ \tcp*{l\textsuperscript{th} rate partial residual}
      $\overline{\boldsymbol{\theta}}_\ell := \texttt{poiss-scp}(\overline{\mathbf{y}}_{1:T} \:;\: \overline{\boldsymbol{\lambda}}_{-\ell,1:T}, u_{0}, v_0, \boldsymbol{\boldsymbol{\pi}}_{1:T})$ \tcp*{update parameters} 
      $\overline{\lambda}_{\ell t} := \sum_{t'=1}^{t-1} \overline{\pi}_{\ell t'} + \sum_{t'=t}^T \overline{\pi}_{\ell t'}\overline{v}_{\ell t'}^{-1}\overline{u}_{\ell t'}$ \tcp*{update l\textsuperscript{th} rate}
    }
    $\lambda_0 := \left(\sum_{t=1}^T \prod_{\ell=1}^L \overline{\lambda}_{\ell t}\right)^{-1}\sum_{t=1}^T y_t$\tcp*{EB update $\lambda_0$}
  }
  \Return{Posterior Parameters: $\lambda_0,\:\{\overline{\boldsymbol{\theta}}_\ell\}_{\ell=1}^L$}.
\end{algorithm}

\newpage

\section{Plots and Figures}
\label{app:plots}

\begin{table}[hbt!]
    \scriptsize
    \centering
    \begin{tabular}{l || c || r r r r | r r | r}
        \multicolumn{1}{c||}{Setting} & $\omega_0,u_0,v_0$ & $|J^* - J|$ & Hausdorff & FPSLE & FNSLE & CI Len. & CCD & Time (s) \\ \hline\hline 
        $T =$ 100 & $10^{-5}$ & 0.000 & 0.511 & 0.089 & 0.092 & 1.355 & 0.973 & 0.011\\
        $J^* =$ 2 & $10^{-3}$ & 0.000 & 0.510 & 0.089 & 0.092 & 1.354 & 0.972 & 0.010\\
        Min. Space = 15 & $10^{-1}$ & 0.000 & 0.536 & 0.095 & 0.096 & 1.363 & 0.970 & 0.009\\ \hline
        $T =$ 100 & $10^{-5}$ & 0.010 & 0.888 & 0.101 & 0.126 & 1.096 & 0.988 & 0.091\\
        $J^* =$ 5 & $10^{-3}$ & 0.005 & 0.798 & 0.090 & 0.108 & 1.095 & 0.988 & 0.089\\
        Min. Space = 15 & $10^{-1}$ & 0.015 & 1.388 & 0.163 & 0.209 & 1.181 & 0.987 & 0.074\\ \hline
        $T =$ 500 & $10^{-5}$ & 0.005 & 2.854 & 0.624 & 0.832 & 3.313 & 0.942 & 0.067\\
        $J^* =$ 2 & $10^{-3}$ & 0.005 & 2.864 & 0.626 & 0.836 & 3.313 & 0.942 & 0.066\\
        Min. Space = 15 & $10^{-1}$ & 0.005 & 2.869 & 0.618 & 0.829 & 3.327 & 0.942 & 0.068\\ \hline
        $T =$ 500 & $10^{-5}$ & 0.001 & 2.926 & 0.578 & 0.614 & 3.466 & 0.941 & 0.063\\
        $J^* =$ 2 & $10^{-3}$ & 0.001 & 2.930 & 0.579 & 0.614 & 3.466 & 0.942 & 0.061\\
        Min. Space = 30 & $10^{-1}$ & 0.001 & 2.916 & 0.560 & 0.594 & 3.473 & 0.940 & 0.060\\ \hline
        $T =$ 500 & $10^{-5}$ & 0.044 & 5.091 & 0.663 & 1.086 & 2.196 & 0.950 & 0.552\\
        $J^* =$ 5 & $10^{-3}$ & 0.040 & 4.999 & 0.648 & 1.039 & 2.197 & 0.950 & 0.541\\
        Min. Space = 15 & $10^{-1}$ & 0.051 & 5.978 & 0.788 & 1.368 & 2.264 & 0.950 & 0.606\\ \hline
        $T =$ 500 & $10^{-5}$ & 0.021 & 4.998 & 0.616 & 0.804 & 2.314 & 0.945 & 0.506\\
        $J^* =$ 5 & $10^{-3}$ & 0.021 & 5.132 & 0.638 & 0.836 & 2.320 & 0.945 & 0.488\\
        Min. Space = 30 & $10^{-1}$ & 0.033 & 6.148 & 0.773 & 1.124 & 2.366 & 0.943 & 0.532\\ \hline
        $T =$ 500 & $10^{-5}$ & 0.202 & 10.360 & 0.822 & 1.278 & 1.710 & 0.963 & 2.443\\
        $J^* =$ 10 & $10^{-3}$ & 0.197 & 10.309 & 0.811 & 1.265 & 1.707 & 0.963 & 2.383\\
        Min. Space = 15 & $10^{-1}$ & 0.232 & 12.853 & 1.065 & 1.817 & 1.784 & 0.962 & 2.640\\ \hline
        $T =$ 500 & $10^{-5}$ & 0.082 & 8.499 & 0.627 & 0.852 & 1.686 & 0.956 & 2.091\\
        $J^* =$ 10 & $10^{-3}$ & 0.079 & 8.373 & 0.620 & 0.839 & 1.687 & 0.956 & 2.025\\
        Min. Space = 30 & $10^{-1}$ & 0.130 & 12.017 & 0.924 & 1.367 & 1.765 & 0.954 & 2.145\\ \hline
        $T =$ 1000 & $10^{-5}$ & 0.004 & 5.926 & 1.257 & 1.511 & 5.097 & 0.934 & 0.149\\
        $J^* =$ 2 & $10^{-3}$ & 0.004 & 5.933 & 1.258 & 1.512 & 5.097 & 0.934 & 0.146\\
        Min. Space = 30 & $10^{-1}$ & 0.004 & 5.830 & 1.221 & 1.467 & 5.112 & 0.934 & 0.148\\ \hline
        $T =$ 1000 & $10^{-5}$ & 0.003 & 5.753 & 1.221 & 1.394 & 5.385 & 0.925 & 0.152\\
        $J^* =$ 2 & $10^{-3}$ & 0.003 & 5.754 & 1.221 & 1.394 & 5.385 & 0.925 & 0.145\\
        Min. Space = 50 & $10^{-1}$ & 0.003 & 5.689 & 1.211 & 1.390 & 5.395 & 0.926 & 0.141\\ \hline
        $T =$ 1000 & $10^{-5}$ & 0.060 & 10.505 & 1.394 & 2.527 & 3.428 & 0.940 & 1.249\\
        $J^* =$ 5 & $10^{-3}$ & 0.059 & 10.500 & 1.417 & 2.516 & 3.432 & 0.940 & 1.184\\
        Min. Space = 30 & $10^{-1}$ & 0.078 & 12.812 & 1.709 & 3.261 & 3.482 & 0.940 & 1.377\\ \hline
        $T =$ 1000 & $10^{-5}$ & 0.041 & 10.292 & 1.341 & 2.011 & 3.606 & 0.934 & 1.160\\
        $J^* =$ 5 & $10^{-3}$ & 0.039 & 10.076 & 1.314 & 1.956 & 3.603 & 0.935 & 1.092\\
        Min. Space = 50 & $10^{-1}$ & 0.058 & 12.265 & 1.550 & 2.579 & 3.635 & 0.935 & 1.226\\ \hline
        $T =$ 1000 & $10^{-5}$ & 0.192 & 18.048 & 1.464 & 2.412 & 2.530 & 0.943 & 5.229\\
        $J^* =$ 10 & $10^{-3}$ & 0.194 & 18.071 & 1.457 & 2.410 & 2.525 & 0.943 & 4.796\\
        Min. Space = 30 & $10^{-1}$ & 0.250 & 22.049 & 1.792 & 3.156 & 2.591 & 0.943 & 5.757\\ \hline
        $T =$ 1000 & $10^{-5}$ & 0.114 & 16.699 & 1.297 & 1.851 & 2.646 & 0.937 & 4.818\\
        $J^* =$ 10 & $10^{-3}$ & 0.110 & 16.280 & 1.268 & 1.795 & 2.643 & 0.937 & 4.442\\
        Min. Space = 50 & $10^{-1}$ & 0.164 & 21.902 & 1.715 & 2.621 & 2.710 & 0.938 & 5.149\\ \hline\hline
    \end{tabular}
    \caption{\small \textbf{Prior Parameter Sensitivity Analysis: MICH-Ora}. Simulation \ref{sim:main} repeated with $\omega_0, u_0, v_0 \in \{10^{-1}, 10^{-3}, 10^{-5}\}$ for oracle version of MICH with $J=J^*$ and $\alpha = 0.1$, $\delta = 0.5$ and $\epsilon=10^{-7}$.}
    \label{tab:ora_prior_sa}
\end{table}

\begin{table}[hbt!]
    \scriptsize
    \centering
    \begin{tabular}{l || c || r r r r | r r | r}
        \multicolumn{1}{c||}{Setting} & $\omega_0,u_0,v_0$ & $|J^* - J|$ & Hausdorff & FPSLE & FNSLE & CI Len. & CCD & Time (s) \\ \hline\hline 
        $T =$ 100 & $10^{-5}$ & 0.057 & 1.233 & 0.357 & 0.328 & 1.477 & 0.973 & 0.253\\
        $J^* =$ 2 & $10^{-3}$ & 0.052 & 1.015 & 0.308 & 0.201 & 1.482 & 0.972 & 0.250\\
        Min. Space = 15 & $10^{-1}$ & 0.051 & 0.964 & 0.301 & 0.172 & 1.484 & 0.970 & 0.314\\ \hline
        $T =$ 100 & $10^{-5}$ & 0.265 & 3.469 & 0.529 & 0.905 & 1.113 & 0.988 & 0.434\\
        $J^* =$ 5 & $10^{-3}$ & 0.111 & 1.472 & 0.224 & 0.337 & 1.106 & 0.988 & 0.437\\
        Min. Space = 15 & $10^{-1}$ & 0.069 & 1.267 & 0.177 & 0.178 & 1.200 & 0.987 & 0.682\\ \hline
        $T =$ 500 & $10^{-5}$ & 0.058 & 5.332 & 1.723 & 2.908 & 3.426 & 0.942 & 1.605\\
        $J^* =$ 2 & $10^{-3}$ & 0.043 & 4.194 & 1.347 & 1.723 & 3.445 & 0.940 & 1.562\\
        Min. Space = 15 & $10^{-1}$ & 0.041 & 3.885 & 1.293 & 1.202 & 3.465 & 0.940 & 1.934\\ \hline
        $T =$ 500 & $10^{-5}$ & 0.039 & 4.960 & 1.480 & 1.749 & 3.585 & 0.940 & 1.622\\
        $J^* =$ 2 & $10^{-3}$ & 0.031 & 3.896 & 1.160 & 0.994 & 3.588 & 0.941 & 1.548\\
        Min. Space = 30 & $10^{-1}$ & 0.034 & 3.627 & 1.161 & 0.755 & 3.599 & 0.940 & 1.883\\ \hline
        $T =$ 500 & $10^{-5}$ & 0.236 & 10.565 & 1.484 & 3.527 & 2.242 & 0.949 & 2.829\\
        $J^* =$ 5 & $10^{-3}$ & 0.145 & 7.281 & 1.070 & 2.064 & 2.253 & 0.949 & 2.777\\
        Min. Space = 15 & $10^{-1}$ & 0.122 & 5.993 & 1.024 & 1.363 & 2.323 & 0.950 & 4.922\\ \hline
        $T =$ 500 & $10^{-5}$ & 0.161 & 9.829 & 1.380 & 2.523 & 2.368 & 0.943 & 2.678\\
        $J^* =$ 5 & $10^{-3}$ & 0.087 & 6.261 & 0.886 & 1.203 & 2.387 & 0.943 & 2.597\\
        Min. Space = 30 & $10^{-1}$ & 0.082 & 5.056 & 0.878 & 0.771 & 2.436 & 0.943 & 4.806\\ \hline
        $T =$ 500 & $10^{-5}$ & 0.979 & 21.372 & 1.853 & 4.699 & 1.670 & 0.962 & 5.839\\
        $J^* =$ 10 & $10^{-3}$ & 0.638 & 15.118 & 1.225 & 2.839 & 1.681 & 0.962 & 5.588\\
        Min. Space = 15 & $10^{-1}$ & 0.452 & 12.335 & 1.103 & 1.955 & 1.784 & 0.962 & 12.029\\ \hline
        $T =$ 500 & $10^{-5}$ & 0.557 & 17.452 & 1.486 & 3.007 & 1.671 & 0.955 & 5.390\\
        $J^* =$ 10 & $10^{-3}$ & 0.296 & 10.875 & 0.868 & 1.468 & 1.697 & 0.955 & 5.287\\
        Min. Space = 30 & $10^{-1}$ & 0.230 & 8.156 & 0.792 & 0.899 & 1.779 & 0.954 & 11.692\\ \hline
        $T =$ 1000 & $10^{-5}$ & 0.058 & 11.358 & 3.687 & 5.428 & 5.202 & 0.933 & 2.773\\
        $J^* =$ 2 & $10^{-3}$ & 0.042 & 8.421 & 2.826 & 3.106 & 5.229 & 0.933 & 2.668\\
        Min. Space = 30 & $10^{-1}$ & 0.038 & 7.049 & 2.275 & 1.950 & 5.208 & 0.933 & 3.113\\ \hline
        $T =$ 1000 & $10^{-5}$ & 0.049 & 11.034 & 3.520 & 4.718 & 5.516 & 0.923 & 2.812\\
        $J^* =$ 2 & $10^{-3}$ & 0.036 & 8.061 & 2.730 & 2.577 & 5.565 & 0.924 & 2.680\\
        Min. Space = 50 & $10^{-1}$ & 0.040 & 7.346 & 2.578 & 1.875 & 5.578 & 0.925 & 3.026\\ \hline
        $T =$ 1000 & $10^{-5}$ & 0.213 & 19.667 & 2.894 & 6.339 & 3.501 & 0.938 & 5.751\\
        $J^* =$ 5 & $10^{-3}$ & 0.133 & 12.730 & 2.045 & 3.331 & 3.529 & 0.939 & 5.172\\
        Min. Space = 30 & $10^{-1}$ & 0.111 & 10.119 & 2.031 & 2.132 & 3.579 & 0.939 & 8.205\\ \hline
        $T =$ 1000 & $10^{-5}$ & 0.151 & 16.987 & 2.523 & 4.347 & 3.654 & 0.933 & 5.496\\
        $J^* =$ 5 & $10^{-3}$ & 0.094 & 10.987 & 1.850 & 2.263 & 3.683 & 0.933 & 4.844\\
        Min. Space = 50 & $10^{-1}$ & 0.095 & 9.072 & 1.819 & 1.495 & 3.741 & 0.935 & 7.562\\ \hline
        $T =$ 1000 & $10^{-5}$ & 0.630 & 31.001 & 2.592 & 5.969 & 2.500 & 0.943 & 13.930\\
        $J^* =$ 10 & $10^{-3}$ & 0.411 & 21.607 & 1.831 & 3.539 & 2.539 & 0.942 & 11.420\\
        Min. Space = 30 & $10^{-1}$ & 0.319 & 15.432 & 1.672 & 2.183 & 2.631 & 0.943 & 21.954\\ \hline
        $T =$ 1000 & $10^{-5}$ & 0.416 & 26.435 & 2.222 & 4.142 & 2.631 & 0.936 & 13.491\\
        $J^* =$ 10 & $10^{-3}$ & 0.252 & 17.505 & 1.552 & 2.251 & 2.665 & 0.936 & 11.365\\
        Min. Space = 50 & $10^{-1}$ & 0.240 & 12.327 & 1.478 & 1.323 & 2.763 & 0.938 & 20.687\\ \hline\hline
    \end{tabular}
    \caption{\small \textbf{Prior Parameter Sensitivity Analysis: MICH-Auto}. Simulation \ref{sim:main} repeated with $\omega_0, u_0, v_0 \in \{10^{-1}, 10^{-3}, 10^{-5}\}$ for MICH with $J$ selected using the ELBO and $\alpha = 0.1$, $\delta = 0.5$ and $\epsilon=10^{-7}$.}
    \label{tab:auto_prior_sa}
\end{table}

\begin{table}[hbt!]
    \scriptsize
    \centering
    \begin{tabular}{l || c || r r r r | r r | r}
        \multicolumn{1}{c||}{Setting} & \multicolumn{1}{c||}{$\delta$} & $|J^* - J|$ & Hausdorff & FPSLE & FNSLE & CI Len. & CCD & Time (s) \\ \hline\hline 
        $T =$ 100 & 0.1 & 0.000 & 0.510 & 0.089 & 0.092 & 1.354 & 0.972 & 0.010\\
        $J^* =$ 2 & 0.5 & 0.000 & 0.510 & 0.089 & 0.092 & 1.354 & 0.972 & 0.010\\
        Min. Space = 15 & 1.1 & 0.000 & 0.512 & 0.090 & 0.093 & 1.359 & 0.972 & 0.010\\ \hline
        $T =$ 100 & 0.1 & 0.006 & 0.796 & 0.089 & 0.108 & 1.094 & 0.988 & 0.089\\
        $J^* =$ 5 & 0.5 & 0.005 & 0.798 & 0.090 & 0.108 & 1.095 & 0.988 & 0.089\\
        Min. Space = 15 & 1.1 & 0.005 & 0.800 & 0.090 & 0.108 & 1.099 & 0.988 & 0.089\\ \hline
        $T =$ 500 & 0.1 & 0.014 & 4.434 & 1.084 & 1.571 & 3.240 & 0.942 & 0.066\\
        $J^* =$ 2 & 0.5 & 0.005 & 2.864 & 0.626 & 0.836 & 3.313 & 0.942 & 0.066\\
        Min. Space = 15 & 1.1 & 0.002 & 2.873 & 0.608 & 0.746 & 3.367 & 0.942 & 0.066\\ \hline
        $T =$ 500 & 0.1 & 0.013 & 4.633 & 1.008 & 1.424 & 3.377 & 0.941 & 0.061\\
        $J^* =$ 2 & 0.5 & 0.001 & 2.930 & 0.579 & 0.614 & 3.466 & 0.942 & 0.061\\
        Min. Space = 30 & 1.1 & 0.000 & 2.991 & 0.615 & 0.607 & 3.484 & 0.942 & 0.061\\ \hline
        $T =$ 500 & 0.1 & 0.052 & 5.277 & 0.657 & 1.179 & 2.156 & 0.950 & 0.541\\
        $J^* =$ 5 & 0.5 & 0.040 & 4.999 & 0.648 & 1.039 & 2.197 & 0.950 & 0.541\\
        Min. Space = 15 & 1.1 & 0.030 & 5.393 & 0.755 & 1.048 & 2.295 & 0.950 & 0.541\\ \hline
        $T =$ 500 & 0.1 & 0.029 & 5.229 & 0.630 & 0.907 & 2.292 & 0.945 & 0.488\\
        $J^* =$ 5 & 0.5 & 0.021 & 5.132 & 0.638 & 0.836 & 2.320 & 0.945 & 0.488\\
        Min. Space = 30 & 1.1 & 0.017 & 5.334 & 0.680 & 0.847 & 2.360 & 0.945 & 0.488\\ \hline
        $T =$ 500 & 0.1 & 0.226 & 9.988 & 0.749 & 1.300 & 1.648 & 0.963 & 2.386\\
        $J^* =$ 10 & 0.5 & 0.197 & 10.309 & 0.811 & 1.265 & 1.707 & 0.963 & 2.383\\
        Min. Space = 15 & 1.1 & 0.174 & 10.823 & 0.897 & 1.269 & 1.814 & 0.963 & 2.385\\ \hline
        $T =$ 500 & 0.1 & 0.094 & 8.230 & 0.594 & 0.841 & 1.661 & 0.956 & 2.024\\
        $J^* =$ 10 & 0.5 & 0.079 & 8.373 & 0.620 & 0.839 & 1.687 & 0.956 & 2.025\\
        Min. Space = 30 & 1.1 & 0.071 & 8.524 & 0.641 & 0.847 & 1.726 & 0.956 & 2.024\\ \hline
        $T =$ 1000 & 0.1 & 0.061 & 22.091 & 5.559 & 9.074 & 4.620 & 0.934 & 0.146\\
        $J^* =$ 2 & 0.5 & 0.004 & 5.933 & 1.258 & 1.512 & 5.097 & 0.934 & 0.146\\
        Min. Space = 30 & 1.1 & 0.001 & 4.967 & 1.031 & 1.112 & 5.169 & 0.934 & 0.146\\ \hline
        $T =$ 1000 & 0.1 & 0.069 & 24.465 & 6.145 & 9.911 & 4.830 & 0.924 & 0.145\\
        $J^* =$ 2 & 0.5 & 0.003 & 5.754 & 1.221 & 1.394 & 5.385 & 0.925 & 0.145\\
        Min. Space = 50 & 1.1 & 0.001 & 5.248 & 1.039 & 1.100 & 5.442 & 0.925 & 0.145\\ \hline
        $T =$ 1000 & 0.1 & 0.091 & 14.629 & 1.806 & 3.643 & 3.316 & 0.940 & 1.183\\
        $J^* =$ 5 & 0.5 & 0.059 & 10.500 & 1.417 & 2.516 & 3.432 & 0.940 & 1.184\\
        Min. Space = 30 & 1.1 & 0.046 & 10.730 & 1.554 & 2.384 & 3.594 & 0.940 & 1.183\\ \hline
        $T =$ 1000 & 0.1 & 0.066 & 13.631 & 1.668 & 2.861 & 3.508 & 0.935 & 1.093\\
        $J^* =$ 5 & 0.5 & 0.039 & 10.076 & 1.314 & 1.956 & 3.603 & 0.935 & 1.092\\
        Min. Space = 50 & 1.1 & 0.029 & 10.675 & 1.497 & 1.990 & 3.731 & 0.935 & 1.092\\ \hline
        $T =$ 1000 & 0.1 & 0.235 & 18.436 & 1.392 & 2.594 & 2.438 & 0.943 & 4.794\\
        $J^* =$ 10 & 0.5 & 0.194 & 18.071 & 1.457 & 2.410 & 2.525 & 0.943 & 4.796\\
        Min. Space = 30 & 1.1 & 0.158 & 19.411 & 1.674 & 2.442 & 2.738 & 0.943 & 4.793\\ \hline
        $T =$ 1000 & 0.1 & 0.143 & 16.207 & 1.207 & 1.897 & 2.575 & 0.937 & 4.440\\
        $J^* =$ 10 & 0.5 & 0.110 & 16.280 & 1.268 & 1.795 & 2.643 & 0.937 & 4.442\\
        Min. Space = 50 & 1.1 & 0.089 & 16.903 & 1.367 & 1.813 & 2.772 & 0.937 & 4.439\\ \hline\hline
    \end{tabular}
    \caption{\small \textbf{Detection Rule Sensitivity Analysis: MICH-Ora}. Simulation \ref{sim:main} repeated with $\delta \in \{0.1, 0.5, 1.1\}$ for oracle version of MICH with $J=J^*$ and $\alpha = 0.1$, $\epsilon = 10^{-7}$ and $\omega_0=u_0=v_0=10^{-3}$.}
    \label{tab:ora_delta_sa}
\end{table}

\begin{table}[hbt!]
    \scriptsize
    \centering
    \begin{tabular}{l || c || r r r r | r r | r}
        \multicolumn{1}{c||}{Setting} & \multicolumn{1}{c||}{$\delta$} & $|J^* - J|$ & Hausdorff & FPSLE & FNSLE & CI Len. & CCD & Time (s) \\ \hline\hline 
        $T =$ 100 & 0.1 & 0.044 & 0.917 & 0.270 & 0.183 & 1.425 & 0.972 & 0.248\\
        $J^* =$ 2 & 0.5 & 0.052 & 1.015 & 0.308 & 0.201 & 1.482 & 0.972 & 0.250\\
        Min. Space = 15 & 1.1 & 0.065 & 1.145 & 0.371 & 0.224 & 1.680 & 0.972 & 0.249\\ \hline
        $T =$ 100 & 0.1 & 0.110 & 1.465 & 0.222 & 0.336 & 1.102 & 0.988 & 0.436\\
        $J^* =$ 5 & 0.5 & 0.111 & 1.472 & 0.224 & 0.337 & 1.106 & 0.988 & 0.437\\
        Min. Space = 15 & 1.1 & 0.114 & 1.487 & 0.228 & 0.338 & 1.125 & 0.988 & 0.436\\ \hline
        $T =$ 500 & 0.1 & 0.048 & 5.847 & 1.726 & 2.569 & 3.292 & 0.940 & 1.556\\
        $J^* =$ 2 & 0.5 & 0.043 & 4.194 & 1.347 & 1.723 & 3.445 & 0.940 & 1.562\\
        Min. Space = 15 & 1.1 & 0.058 & 5.101 & 1.782 & 1.878 & 3.843 & 0.940 & 1.556\\ \hline
        $T =$ 500 & 0.1 & 0.036 & 5.543 & 1.477 & 1.811 & 3.426 & 0.941 & 1.547\\
        $J^* =$ 2 & 0.5 & 0.031 & 3.896 & 1.160 & 0.994 & 3.588 & 0.941 & 1.548\\
        Min. Space = 30 & 1.1 & 0.050 & 4.908 & 1.652 & 1.161 & 4.089 & 0.941 & 1.548\\ \hline
        $T =$ 500 & 0.1 & 0.143 & 7.392 & 1.023 & 2.166 & 2.182 & 0.949 & 2.773\\
        $J^* =$ 5 & 0.5 & 0.145 & 7.281 & 1.070 & 2.064 & 2.253 & 0.949 & 2.777\\
        Min. Space = 15 & 1.1 & 0.160 & 8.019 & 1.254 & 2.097 & 2.463 & 0.949 & 2.773\\ \hline
        $T =$ 500 & 0.1 & 0.081 & 6.122 & 0.821 & 1.256 & 2.324 & 0.943 & 2.596\\
        $J^* =$ 5 & 0.5 & 0.087 & 6.261 & 0.886 & 1.203 & 2.387 & 0.943 & 2.597\\
        Min. Space = 30 & 1.1 & 0.099 & 6.782 & 1.007 & 1.236 & 2.529 & 0.943 & 2.597\\ \hline
        $T =$ 500 & 0.1 & 0.637 & 14.946 & 1.183 & 2.856 & 1.646 & 0.962 & 5.595\\
        $J^* =$ 10 & 0.5 & 0.638 & 15.118 & 1.225 & 2.839 & 1.681 & 0.962 & 5.588\\
        Min. Space = 15 & 1.1 & 0.639 & 15.385 & 1.263 & 2.845 & 1.736 & 0.962 & 5.593\\ \hline
        $T =$ 500 & 0.1 & 0.294 & 10.788 & 0.851 & 1.477 & 1.669 & 0.955 & 5.283\\
        $J^* =$ 10 & 0.5 & 0.296 & 10.875 & 0.868 & 1.468 & 1.697 & 0.955 & 5.287\\
        Min. Space = 30 & 1.1 & 0.304 & 10.993 & 0.888 & 1.472 & 1.750 & 0.955 & 5.286\\ \hline
        $T =$ 1000 & 0.1 & 0.093 & 24.914 & 7.028 & 11.020 & 4.653 & 0.933 & 2.667\\
        $J^* =$ 2 & 0.5 & 0.042 & 8.421 & 2.826 & 3.106 & 5.229 & 0.933 & 2.668\\
        Min. Space = 30 & 1.1 & 0.064 & 9.591 & 3.768 & 2.968 & 6.194 & 0.933 & 2.667\\ \hline
        $T =$ 1000 & 0.1 & 0.090 & 26.596 & 7.366 & 11.168 & 4.871 & 0.924 & 2.675\\
        $J^* =$ 2 & 0.5 & 0.036 & 8.061 & 2.730 & 2.577 & 5.565 & 0.924 & 2.680\\
        Min. Space = 50 & 1.1 & 0.057 & 9.762 & 3.731 & 2.658 & 6.372 & 0.924 & 2.676\\ \hline
        $T =$ 1000 & 0.1 & 0.145 & 16.646 & 2.307 & 4.384 & 3.352 & 0.939 & 5.166\\
        $J^* =$ 5 & 0.5 & 0.133 & 12.730 & 2.045 & 3.331 & 3.529 & 0.939 & 5.172\\
        Min. Space = 30 & 1.1 & 0.155 & 14.595 & 2.576 & 3.460 & 3.912 & 0.939 & 5.166\\ \hline
        $T =$ 1000 & 0.1 & 0.103 & 14.807 & 2.114 & 3.276 & 3.527 & 0.933 & 4.845\\
        $J^* =$ 5 & 0.5 & 0.094 & 10.987 & 1.850 & 2.263 & 3.683 & 0.933 & 4.844\\
        Min. Space = 50 & 1.1 & 0.119 & 12.905 & 2.357 & 2.412 & 4.050 & 0.933 & 4.845\\ \hline
        $T =$ 1000 & 0.1 & 0.416 & 21.802 & 1.758 & 3.689 & 2.454 & 0.942 & 11.422\\
        $J^* =$ 10 & 0.5 & 0.411 & 21.607 & 1.831 & 3.539 & 2.539 & 0.942 & 11.420\\
        Min. Space = 30 & 1.1 & 0.424 & 22.692 & 1.999 & 3.582 & 2.736 & 0.942 & 11.421\\ \hline
        $T =$ 1000 & 0.1 & 0.244 & 17.352 & 1.466 & 2.315 & 2.592 & 0.936 & 11.369\\
        $J^* =$ 10 & 0.5 & 0.252 & 17.505 & 1.552 & 2.251 & 2.665 & 0.936 & 11.365\\
        Min. Space = 50 & 1.1 & 0.276 & 18.346 & 1.689 & 2.276 & 2.857 & 0.936 & 11.366\\ \hline\hline
    \end{tabular}
    \caption{\small \textbf{Detection Rule Sensitivity Analysis: MICH-Auto}. Simulation \ref{sim:main} repeated with $\delta \in \{0.1, 0.5, 1.1\}$ for MICH with $J$ selected using the ELBO and $\alpha = 0.1$, $\epsilon = 10^{-7}$ and $\omega_0=u_0=v_0=10^{-3}$.}
    \label{tab:auto_delta_sa}
\end{table}

\begin{table}[hbt!]
    \scriptsize
    \centering
    \begin{tabular}{l || l || r r r r | r r | r}
        \multicolumn{1}{c||}{Setting} & \multicolumn{1}{c||}{$\epsilon$} & $|J^* - J|$ & Hausdorff & FPSLE & FNSLE & CI Len. & CCD & Time (s) \\ \hline\hline 
        $T =$ 100 & $10^{-10}$ & 0.000 & 0.510 & 0.089 & 0.092 & 1.354 & 0.973 & 0.021\\
        $J^* =$ 2 & $10^{-7}$ & 0.000 & 0.510 & 0.089 & 0.092 & 1.354 & 0.972 & 0.010\\
        Min. Space = 15 & $10^{-5}$ & 0.002 & 0.701 & 0.133 & 0.146 & 1.491 & 0.969 & 0.007\\
         & $10^{-3}$ & 0.001 & 0.595 & 0.111 & 0.124 & 1.370 & 0.972 & 0.004\\ \hline
        $T =$ 100 & $10^{-10}$ & 0.005 & 0.875 & 0.102 & 0.124 & 1.095 & 0.988 & 0.298\\
        $J^* =$ 5 & $10^{-7}$ & 0.005 & 0.798 & 0.090 & 0.108 & 1.095 & 0.988 & 0.089\\
        Min. Space = 15 & $10^{-5}$ & 0.058 & 2.129 & 0.231 & 0.322 & 1.167 & 0.987 & 0.029\\
         & $10^{-3}$ & 0.031 & 1.317 & 0.144 & 0.195 & 1.120 & 0.988 & 0.010\\ \hline
        $T =$ 500 & $10^{-10}$ & 0.005 & 2.872 & 0.623 & 0.796 & 3.314 & 0.942 & 0.209\\
        $J^* =$ 2 & $10^{-7}$ & 0.005 & 2.864 & 0.626 & 0.836 & 3.313 & 0.942 & 0.066\\
        Min. Space = 15 & $10^{-5}$ & 0.032 & 5.308 & 1.505 & 2.680 & 3.262 & 0.941 & 0.024\\
         & $10^{-3}$ & 0.022 & 3.860 & 1.020 & 2.018 & 3.468 & 0.941 & 0.008\\ \hline
        $T =$ 500 & $10^{-10}$ & 0.001 & 2.904 & 0.573 & 0.600 & 3.463 & 0.942 & 0.191\\
        $J^* =$ 2 & $10^{-7}$ & 0.001 & 2.930 & 0.579 & 0.614 & 3.466 & 0.942 & 0.061\\
        Min. Space = 30 & $10^{-5}$ & 0.020 & 4.719 & 1.094 & 1.856 & 3.399 & 0.936 & 0.021\\
         & $10^{-3}$ & 0.010 & 3.658 & 0.807 & 1.321 & 3.564 & 0.941 & 0.008\\ \hline
        $T =$ 500 & $10^{-10}$ & 0.028 & 4.716 & 0.620 & 0.925 & 2.193 & 0.950 & 2.786\\
        $J^* =$ 5 & $10^{-7}$ & 0.040 & 4.999 & 0.648 & 1.039 & 2.197 & 0.950 & 0.541\\
        Min. Space = 15 & $10^{-5}$ & 0.129 & 8.945 & 1.117 & 2.377 & 2.171 & 0.954 & 0.136\\
         & $10^{-3}$ & 0.159 & 9.245 & 1.299 & 2.684 & 2.526 & 0.952 & 0.029\\ \hline
        $T =$ 500 & $10^{-10}$ & 0.014 & 4.749 & 0.598 & 0.729 & 2.312 & 0.945 & 2.345\\
        $J^* =$ 5 & $10^{-7}$ & 0.021 & 5.132 & 0.638 & 0.836 & 2.320 & 0.945 & 0.488\\
        Min. Space = 30 & $10^{-5}$ & 0.094 & 8.843 & 1.056 & 1.842 & 2.273 & 0.947 & 0.118\\
         & $10^{-3}$ & 0.096 & 8.501 & 1.076 & 1.853 & 2.505 & 0.945 & 0.029\\ \hline
        $T =$ 500 & $10^{-10}$ & 0.151 & 9.131 & 0.730 & 1.092 & 1.673 & 0.963 & 12.738\\
        $J^* =$ 10 & $10^{-7}$ & 0.197 & 10.309 & 0.811 & 1.265 & 1.707 & 0.963 & 2.383\\
        Min. Space = 15 & $10^{-5}$ & 0.334 & 12.067 & 0.870 & 1.683 & 1.642 & 0.963 & 0.539\\
         & $10^{-3}$ & 0.631 & 18.670 & 1.459 & 3.048 & 2.026 & 0.963 & 0.096\\ \hline
        $T =$ 500 & $10^{-10}$ & 0.053 & 7.419 & 0.557 & 0.724 & 1.671 & 0.956 & 11.077\\
        $J^* =$ 10 & $10^{-7}$ & 0.079 & 8.373 & 0.620 & 0.839 & 1.687 & 0.956 & 2.025\\
        Min. Space = 30 & $10^{-5}$ & 0.188 & 11.515 & 0.779 & 1.221 & 1.665 & 0.958 & 0.479\\
         & $10^{-3}$ & 0.412 & 19.159 & 1.336 & 2.316 & 1.910 & 0.956 & 0.096\\ \hline
        $T =$ 1000 & $10^{-10}$ & 0.004 & 6.030 & 1.265 & 1.501 & 5.092 & 0.934 & 0.606\\
        $J^* =$ 2 & $10^{-7}$ & 0.004 & 5.933 & 1.258 & 1.512 & 5.097 & 0.934 & 0.146\\
        Min. Space = 30 & $10^{-5}$ & 0.069 & 18.706 & 5.278 & 9.711 & 4.688 & 0.933 & 0.044\\
         & $10^{-3}$ & 0.030 & 8.792 & 2.352 & 4.739 & 5.389 & 0.936 & 0.013\\ \hline
        $T =$ 1000 & $10^{-10}$ & 0.004 & 5.867 & 1.251 & 1.467 & 5.386 & 0.925 & 0.635\\
        $J^* =$ 2 & $10^{-7}$ & 0.003 & 5.754 & 1.221 & 1.394 & 5.385 & 0.925 & 0.145\\
        Min. Space = 50 & $10^{-5}$ & 0.058 & 17.304 & 4.788 & 8.373 & 4.794 & 0.936 & 0.041\\
         & $10^{-3}$ & 0.017 & 8.137 & 2.142 & 3.332 & 5.652 & 0.926 & 0.013\\ \hline
        $T =$ 1000 & $10^{-10}$ & 0.034 & 8.706 & 1.140 & 1.820 & 3.413 & 0.941 & 6.506\\
        $J^* =$ 5 & $10^{-7}$ & 0.059 & 10.500 & 1.417 & 2.516 & 3.432 & 0.940 & 1.184\\
        Min. Space = 30 & $10^{-5}$ & 0.166 & 19.672 & 2.453 & 5.542 & 3.207 & 0.944 & 0.242\\
         & $10^{-3}$ & 0.216 & 21.359 & 2.999 & 6.910 & 3.940 & 0.941 & 0.050\\ \hline
        $T =$ 1000 & $10^{-10}$ & 0.023 & 9.010 & 1.185 & 1.627 & 3.598 & 0.935 & 5.499\\
        $J^* =$ 5 & $10^{-7}$ & 0.039 & 10.076 & 1.314 & 1.956 & 3.603 & 0.935 & 1.092\\
        Min. Space = 50 & $10^{-5}$ & 0.134 & 20.090 & 2.469 & 4.803 & 3.381 & 0.940 & 0.224\\
         & $10^{-3}$ & 0.154 & 20.168 & 2.690 & 5.233 & 3.970 & 0.936 & 0.050\\ \hline
        $T =$ 1000 & $10^{-10}$ & 0.127 & 14.369 & 1.140 & 1.797 & 2.467 & 0.943 & 27.027\\
        $J^* =$ 10 & $10^{-7}$ & 0.194 & 18.071 & 1.457 & 2.410 & 2.525 & 0.943 & 4.796\\
        Min. Space = 30 & $10^{-5}$ & 0.343 & 23.980 & 1.736 & 3.540 & 2.399 & 0.946 & 0.965\\
         & $10^{-3}$ & 0.661 & 38.122 & 2.992 & 6.499 & 3.001 & 0.944 & 0.170\\ \hline
        $T =$ 1000 & $10^{-10}$ & 0.057 & 12.157 & 0.974 & 1.269 & 2.601 & 0.937 & 23.742\\
        $J^* =$ 10 & $10^{-7}$ & 0.110 & 16.280 & 1.268 & 1.795 & 2.643 & 0.937 & 4.442\\
        Min. Space = 50 & $10^{-5}$ & 0.245 & 23.454 & 1.669 & 2.851 & 2.532 & 0.942 & 0.873\\
         & $10^{-3}$ & 0.470 & 37.722 & 2.766 & 5.118 & 3.010 & 0.939 & 0.172\\ \hline\hline
    \end{tabular}
    \caption{\small \textbf{Convergence Criterion Sensitivity Analysis: MICH-Ora}. Simulation \ref{sim:main} repeated with $\epsilon \in \{10^{-3}, 10^{-5}, 10^{-7}, 10^{-10}\}$ for oracle version of MICH with $J=J^*$ and $\alpha = 0.1$, $\delta = 0.5$ and $\omega_0=u_0=v_0=10^{-3}$.}
    \label{tab:ora_tol_sa}
\end{table}

\begin{table}[hbt!]
    \scriptsize
    \centering
    \begin{tabular}{l || l || r r r r | r r | r}
        \multicolumn{1}{c||}{Setting} & \multicolumn{1}{c||}{$\epsilon$} & $|J^* - J|$ & Hausdorff & FPSLE & FNSLE & CI Len. & CCD & Time (s) \\ \hline\hline 
        $T =$ 100 & $10^{-10}$ & 0.053 & 1.035 & 0.315 & 0.205 & 1.489 & 0.972 & 0.973\\
        $J^* =$ 2 & $10^{-7}$ & 0.052 & 1.015 & 0.308 & 0.201 & 1.482 & 0.972 & 0.250\\
        Min. Space = 15 & $10^{-5}$ & 0.045 & 1.183 & 0.349 & 0.335 & 1.540 & 0.968 & 1.529\\
         & $10^{-3}$ & 0.029 & 0.782 & 0.222 & 0.163 & 1.434 & 0.972 & 0.012\\ \hline
        $T =$ 100 & $10^{-10}$ & 0.109 & 1.441 & 0.222 & 0.327 & 1.107 & 0.988 & 2.032\\
        $J^* =$ 5 & $10^{-7}$ & 0.111 & 1.472 & 0.224 & 0.337 & 1.106 & 0.988 & 0.437\\
        Min. Space = 15 & $10^{-5}$ & 0.221 & 3.328 & 0.437 & 0.693 & 1.135 & 0.987 & 0.085\\
         & $10^{-3}$ & 0.191 & 2.478 & 0.381 & 0.648 & 1.117 & 0.988 & 0.021\\ \hline
        $T =$ 500 & $10^{-10}$ & 0.041 & 3.810 & 1.245 & 1.408 & 3.437 & 0.941 & 11.412\\
        $J^* =$ 2 & $10^{-7}$ & 0.043 & 4.194 & 1.347 & 1.723 & 3.445 & 0.940 & 1.562\\
        Min. Space = 15 & $10^{-5}$ & 0.083 & 7.500 & 2.569 & 4.841 & 3.294 & 0.939 & 0.185\\
         & $10^{-3}$ & 0.061 & 4.725 & 1.597 & 3.258 & 3.502 & 0.941 & 0.039\\ \hline
        $T =$ 500 & $10^{-10}$ & 0.030 & 3.699 & 1.078 & 0.849 & 3.591 & 0.941 & 12.830\\
        $J^* =$ 2 & $10^{-7}$ & 0.031 & 3.896 & 1.160 & 0.994 & 3.588 & 0.941 & 1.548\\
        Min. Space = 30 & $10^{-5}$ & 0.054 & 5.929 & 1.815 & 2.882 & 3.445 & 0.935 & 0.173\\
         & $10^{-3}$ & 0.035 & 4.163 & 1.326 & 1.743 & 3.624 & 0.940 & 0.038\\ \hline
        $T =$ 500 & $10^{-10}$ & 0.130 & 6.762 & 1.004 & 1.801 & 2.250 & 0.949 & 24.995\\
        $J^* =$ 5 & $10^{-7}$ & 0.145 & 7.281 & 1.070 & 2.064 & 2.253 & 0.949 & 2.777\\
        Min. Space = 15 & $10^{-5}$ & 0.267 & 11.900 & 1.592 & 3.888 & 2.157 & 0.952 & 0.375\\
         & $10^{-3}$ & 0.301 & 11.240 & 1.648 & 4.244 & 2.392 & 0.951 & 0.082\\ \hline
        $T =$ 500 & $10^{-10}$ & 0.071 & 5.508 & 0.786 & 0.972 & 2.380 & 0.944 & 22.970\\
        $J^* =$ 5 & $10^{-7}$ & 0.087 & 6.261 & 0.886 & 1.203 & 2.387 & 0.943 & 2.597\\
        Min. Space = 30 & $10^{-5}$ & 0.185 & 11.381 & 1.514 & 2.783 & 2.263 & 0.946 & 0.359\\
         & $10^{-3}$ & 0.160 & 9.086 & 1.294 & 2.325 & 2.460 & 0.944 & 0.082\\ \hline
        $T =$ 500 & $10^{-10}$ & 0.575 & 13.959 & 1.102 & 2.555 & 1.674 & 0.963 & 47.349\\
        $J^* =$ 10 & $10^{-7}$ & 0.638 & 15.118 & 1.225 & 2.839 & 1.681 & 0.962 & 5.588\\
        Min. Space = 15 & $10^{-5}$ & 0.758 & 17.512 & 1.405 & 3.452 & 1.612 & 0.962 & 0.937\\
         & $10^{-3}$ & 1.037 & 22.320 & 1.958 & 4.938 & 1.838 & 0.962 & 0.191\\ \hline
        $T =$ 500 & $10^{-10}$ & 0.249 & 9.762 & 0.760 & 1.256 & 1.686 & 0.955 & 42.987\\
        $J^* =$ 10 & $10^{-7}$ & 0.296 & 10.875 & 0.868 & 1.468 & 1.697 & 0.955 & 5.287\\
        Min. Space = 30 & $10^{-5}$ & 0.442 & 14.708 & 1.167 & 2.197 & 1.647 & 0.957 & 0.891\\
         & $10^{-3}$ & 0.610 & 19.439 & 1.615 & 3.144 & 1.799 & 0.955 & 0.196\\ \hline
        $T =$ 1000 & $10^{-10}$ & 0.030 & 7.518 & 2.265 & 2.294 & 5.199 & 0.934 & 32.231\\
        $J^* =$ 2 & $10^{-7}$ & 0.042 & 8.421 & 2.826 & 3.106 & 5.229 & 0.933 & 2.668\\
        Min. Space = 30 & $10^{-5}$ & 0.110 & 21.442 & 7.267 & 12.519 & 4.714 & 0.933 & 0.259\\
         & $10^{-3}$ & 0.056 & 9.254 & 3.376 & 5.157 & 5.340 & 0.931 & 0.071\\ \hline
        $T =$ 1000 & $10^{-10}$ & 0.027 & 6.985 & 2.252 & 1.786 & 5.527 & 0.925 & 33.539\\
        $J^* =$ 2 & $10^{-7}$ & 0.036 & 8.061 & 2.730 & 2.577 & 5.565 & 0.924 & 2.680\\
        Min. Space = 50 & $10^{-5}$ & 0.094 & 19.263 & 6.250 & 9.726 & 4.850 & 0.934 & 0.248\\
         & $10^{-3}$ & 0.042 & 8.142 & 3.163 & 3.373 & 5.626 & 0.926 & 0.071\\ \hline
        $T =$ 1000 & $10^{-10}$ & 0.102 & 11.307 & 1.703 & 2.776 & 3.495 & 0.940 & 61.906\\
        $J^* =$ 5 & $10^{-7}$ & 0.133 & 12.730 & 2.045 & 3.331 & 3.529 & 0.939 & 5.172\\
        Min. Space = 30 & $10^{-5}$ & 0.218 & 20.143 & 2.802 & 6.012 & 3.199 & 0.943 & 0.620\\
         & $10^{-3}$ & 0.273 & 20.226 & 3.283 & 7.396 & 3.743 & 0.940 & 0.154\\ \hline
        $T =$ 1000 & $10^{-10}$ & 0.071 & 10.023 & 1.567 & 1.919 & 3.665 & 0.934 & 62.226\\
        $J^* =$ 5 & $10^{-7}$ & 0.094 & 10.987 & 1.850 & 2.263 & 3.683 & 0.933 & 4.844\\
        Min. Space = 50 & $10^{-5}$ & 0.176 & 19.951 & 2.814 & 4.897 & 3.369 & 0.938 & 0.598\\
         & $10^{-3}$ & 0.186 & 17.059 & 2.889 & 4.778 & 3.859 & 0.934 & 0.152\\ \hline
        $T =$ 1000 & $10^{-10}$ & 0.366 & 20.076 & 1.654 & 3.220 & 2.513 & 0.942 & 121.808\\
        $J^* =$ 10 & $10^{-7}$ & 0.411 & 21.607 & 1.831 & 3.539 & 2.539 & 0.942 & 11.420\\
        Min. Space = 30 & $10^{-5}$ & 0.524 & 26.372 & 2.131 & 4.684 & 2.379 & 0.945 & 1.692\\
         & $10^{-3}$ & 0.787 & 36.019 & 3.124 & 7.224 & 2.786 & 0.943 & 0.386\\ \hline
        $T =$ 1000 & $10^{-10}$ & 0.197 & 15.685 & 1.300 & 1.913 & 2.636 & 0.937 & 119.104\\
        $J^* =$ 10 & $10^{-7}$ & 0.252 & 17.505 & 1.552 & 2.251 & 2.665 & 0.936 & 11.365\\
        Min. Space = 50 & $10^{-5}$ & 0.333 & 22.139 & 1.786 & 3.073 & 2.518 & 0.941 & 1.639\\
         & $10^{-3}$ & 0.507 & 30.850 & 2.661 & 4.889 & 2.856 & 0.937 & 0.387\\ \hline\hline
    \end{tabular}
    \caption{\small \textbf{Convergence Criterion Sensitivity Analysis: MICH-Auto}. Simulation \ref{sim:main} repeated with $\epsilon \in \{10^{-3}, 10^{-5}, 10^{-7}, 10^{-10}\}$ for MICH with $J$ selected using the ELBO and $\alpha = 0.1$, $\delta = 0.5$ and $\omega_0=u_0=v_0=10^{-3}$.}
    \label{tab:auto_tol_sa}
\end{table}

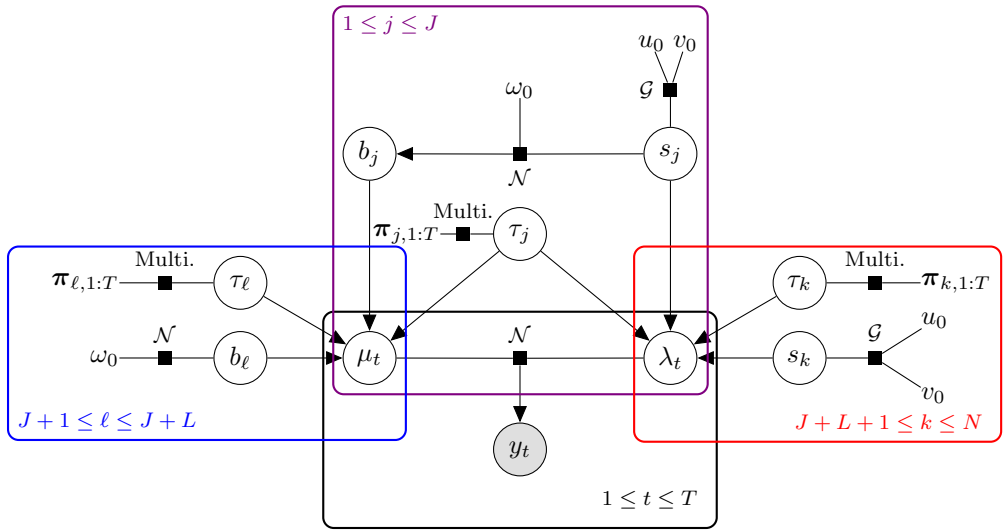
\begin{figure}[!h]
    \centering
    \begin{tikzpicture}[x=2cm,y=2cm]
  % Nodes

  % DGP
  \node[obs]                   (y)      {$y_t$} ; 
  \factor[above=0.76cm of y] {y-f} {above:$\mathcal{N}$} {} {} ; %
  \node[latent, above=0.5cm of y, xshift = -2cm]    (mu)     {$\mu_t$} ; 
  \node[latent, above=0.5cm of y, xshift = 2cm]    (lambda) {$\lambda_t$} ; 

  % SMCP
  \node[latent, left=0.5 of mu] (b_l) {$b_\ell$} ;
  \factor[left=of b_l, xshift=0.25cm] {b_l-f} {above:$\mathcal{N}$} {} {} ; %
  \node[const, left=of b_l, xshift=0.75cm] (tau_l) {$\omega_0$} ; %

  \node[latent, left=0.5 of mu, yshift=1cm] (gamma_l) {$\tau_\ell$} ;
  \factor[left=of gamma_l, xshift=0.25cm] {gamma_l-f} {above:Multi.} {} {} ; %
  \node[const, left=of gamma_l, xshift=0.75cm] (pi_l) {$\boldsymbol{\pi}_{\ell,1:T}$} ; %

  % SSCP
  \node[latent, right=0.5 of lambda] (s_k) {$s_k$} ;
  \factor[right=of s_k, xshift=-0.25cm] {s_k-f} {above:$\mathcal{G}$} {} {} ; %
  \node[const, right=of s_k, xshift=-0.75cm, yshift = 0.5cm] (u_k) {$u_0$} ; %
  \node[const, right=of s_k, xshift=-0.75cm, yshift = -0.5cm] (v_k) {$v_0$} ; %

  \node[latent, right=0.5 of lambda, yshift=1cm] (gamma_k) {$\tau_k$} ;
  \factor[right=of gamma_k, xshift=-0.25cm] {gamma_k-f} {above:Multi.} {} {} ; %
  \node[const, right=of gamma_k, xshift=-0.75cm] (pi_k) {$\boldsymbol{\pi}_{k,1:T}$} ; %

  % SMSCP
  \node[latent, above=2cm of mu] (b_j) {$b_j$} ;
  \node[latent, above=2cm of lambda] (s_j) {$s_j$} ;
  \factor[above=0.2 of s_j, xshift=0cm] {s_j-f} {left:$\mathcal{G}$} {} {} ; %
  \node[const, above=0.5 of s_j, xshift=-0.25cm] (u_j) {$u_0$} ; %
  \node[const, above=0.5 of s_j, xshift=0.2cm] (v_j) {$v_0$} ; %
  
  \node[latent, above=1.2cm of y-f] (gamma_j) {$\tau_j$} ;
  \factor[above=0.61cm of gamma_j] {b_j-f} {below:$\mathcal{N}$} {} {} ; %
  \node[const, above=0.65cm of b_j-f] (tau_j) {$\omega_0$} ; %
  \factor[left=of gamma_j , xshift=0.5cm] {gamma_j-f} {above:Multi.} {} {} ; %
  \node[const, left=of gamma_j, xshift=1.3cm] (pi_j) {$\boldsymbol{\pi}_{j,1:T}$} ; %

  % Connect the nodes
  \edge {y-f} {y} ;
  \edge[-] {mu, lambda} {y-f} ; %

  \edge{gamma_l, b_l} {mu}
  \edge[-] {b_l-f} {b_l} ; %
  \edge[-] {gamma_l-f} {gamma_l} ; %
  \edge[-] {tau_l} {b_l-f} ; %
  \edge[-] {pi_l} {gamma_l-f} ; %

  \edge{gamma_k, s_k} {lambda}
  \edge[-] {s_k-f} {s_k} ; %
  \edge[-] {gamma_k-f} {gamma_k} ; %
  \edge[-] {u_k, v_k} {s_k-f} ; %
  \edge[-] {pi_k} {gamma_k-f} ; %

  \edge{gamma_j, b_j} {mu}
  \edge{b_j-f} {b_j} ; %
  \edge[-] {tau_j,s_j} {b_j-f} ; %
  \edge[-] {gamma_j-f} {gamma_j} ; %
  \edge[-] {pi_j} {gamma_j-f} ; %
  \edge{gamma_j, s_j} {lambda}
  \edge[-] {s_j-f} {s_j} ; %
  \edge[-] {u_j, v_j} {s_j-f} ; %
  
  % Plates
  \plate[thick, inner sep=.25cm]{y} {(y)(mu)(lambda)} {$1 \leq t \leq T$} ;
  \tikzset{plate caption/.append style={above right =0pt of #1.north west}}
  \plate[thick,color=violet]{mu lambda} {(mu)(lambda)(gamma_j)(gamma_j-f)(pi_j)(b_j)(b_j-f)(tau_j)(s_j)(s_j-f)(u_j)(v_j)}{\textcolor{violet}{$1 \leq j \leq J$}} ;
  \tikzset{plate caption/.append style={below left = 5pt of #1.south east}}
  \plate[thick,color=red]{lambda} {(lambda)(gamma_k)(gamma_k-f)(pi_k)(s_k)(s_k-f)(u_k)(v_k)}{\textcolor{red}{$J+L+1 \leq k \leq N$}} ;
  \tikzset{plate caption/.append style={below=10pt of #1.south west, xshift=0.75cm}}
  \plate[thick,color=blue]{mu} {(mu)(gamma_l)(gamma_l-f)(pi_l)(b_l)(b_l-f)(tau_l)}{\textcolor{blue}{$J+1 \leq \ell \leq J+ L$}} ;
\end{tikzpicture}
    \caption{\textbf{MICH Plate Diagram}. Graphical model depicting the directed acyclic graph specified by (\ref{eq:dgp}) and (\ref{eq:mu_t})-(\ref{eq:lambda_t}).}
    \label{fig:plate-diagram}
\end{figure}

\begin{figure}[!h]
    \centering
    \includegraphics[scale=0.2]{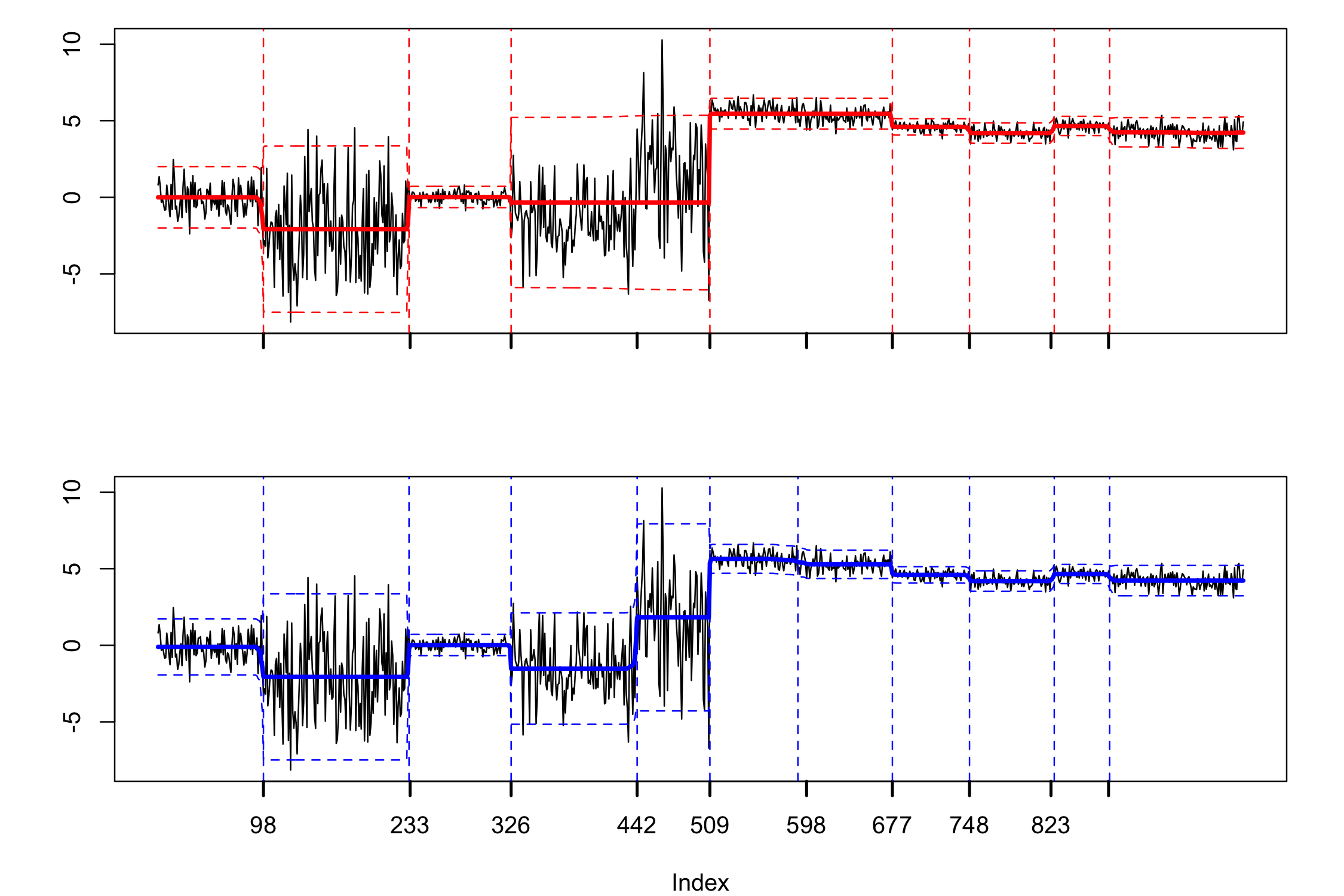}
    \caption{\textbf{MICH with Reversing}. Both plots show $\mathbf{y}_{1:T}$ (black line) generated from Simulation \ref{sim:main} with $J^* = 10$, $T = 1000$, $\Delta_T = 50$, $C = \sqrt{200}$, the true change-points (x-axis ticks), estimated change-points (vertical dashed lines), and estimated mean (solid line) and variance (dashed lines) signals. \textbf{Top}: MICH fit to $\mathbf{y}_{1:T}$ with $J=10$ selected using ELBO and $\hat{J}=8$. This fit misses the change-points at 442 and 598. \textbf{Bottom}: MICH fit to $\mathbf{y}_{T:1}$ with $J=10$ selected using ELBO and $\hat{J} = 10$. This fit recovers all of the change-points.}
    \label{fig:reverse-fit}
\end{figure}

\begin{figure}[!h]
    \centering
    \includegraphics[scale=0.36]{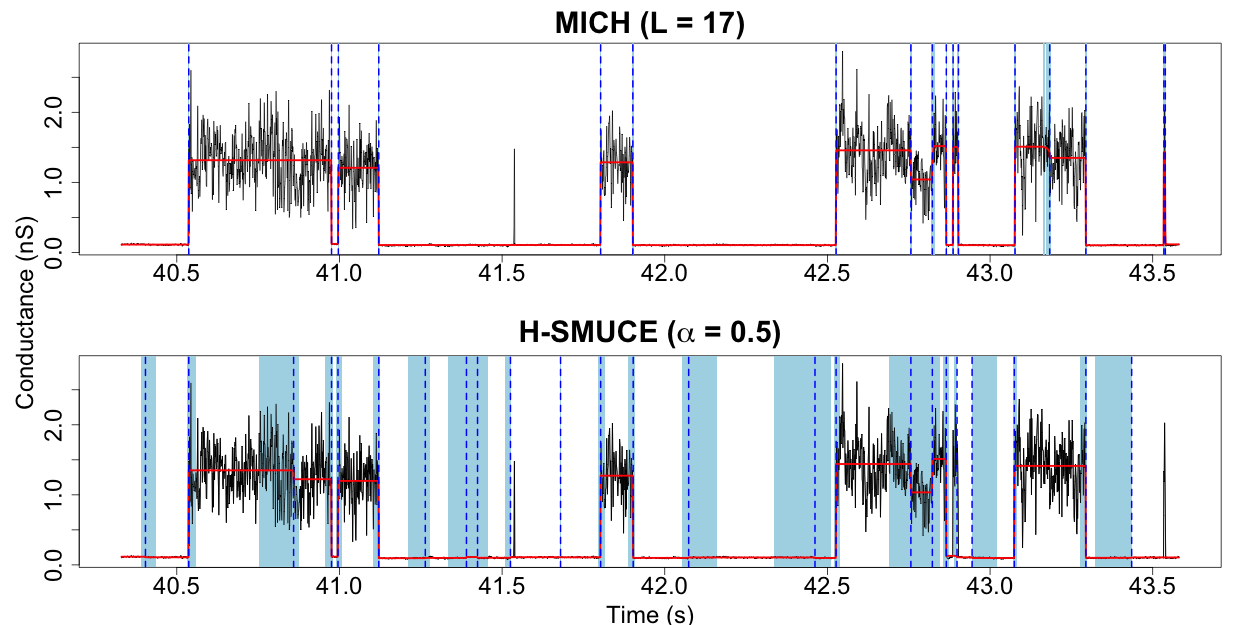}
    \caption{\textbf{Additional Ion Channel Model Fits}. Estimated mean signal (\textcolor{red}{\raisebox{0.5ex}{\rule{0.3cm}{1.5pt}}}) and change-points (\textcolor{blue}{\makebox[10pt][l]{\hdashrule[0.5ex]{10pt}{1.5pt}{4pt 1pt}}}) for 2,956 subsampled ionic current recordings (\textcolor{black}{\raisebox{0.5ex}{\rule{0.3cm}{1.5pt}}}) from a PorB porin collected by the Steinem lab. \textbf{Top}: MICH fit with $L = 17$ components and $\hat{J} = 17$ with 95\% credible sets shaded (\textcolor{cyan}{\raisebox{0.25ex}{\rule{0.3cm}{4pt}}}). \textbf{Bottom}: H-SMUCE fit with $\hat{J} = 24$ and 50\% confidence intervals shaded (\textcolor{cyan}{\raisebox{0.25ex}{\rule{0.3cm}{4pt}}})..}
    \label{fig:ion2}
\end{figure}

\begin{figure}[!h]
    \centering
    \includegraphics[scale=0.35]{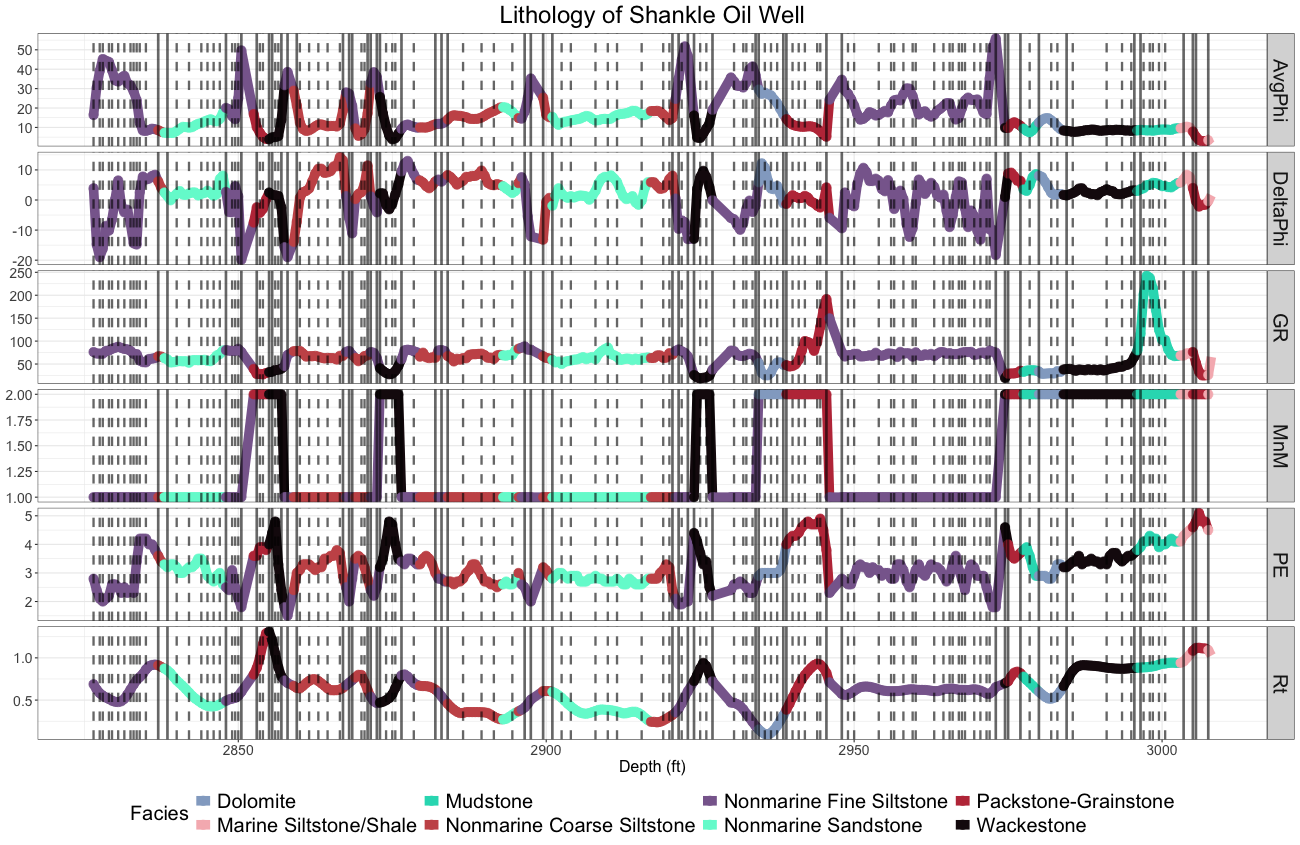}
    \caption{\textbf{\texttt{InspectChangepoint} Fit of Shankle Well Log}. $\hat{L}=140$ estimated changes (\textcolor{black}{\makebox[10pt][l]{\hdashrule[0.5ex]{10pt}{1.5pt}{4pt 1pt}}}) in the facies of the Shankle well from method of \cite{Wang17}. 75 estimated changes within one index of a true changes (\raisebox{0.5ex}{\rule{0.3cm}{1pt}}) and 32 true changes within one index of an estimated change. Note the Marine-Nonmarine layer must be omitted to fit \texttt{InspectChangepoint}.}
    \label{fig:well_inspect}
\end{figure}

\begin{figure}[!h]
    \centering
    \includegraphics[scale=0.35]{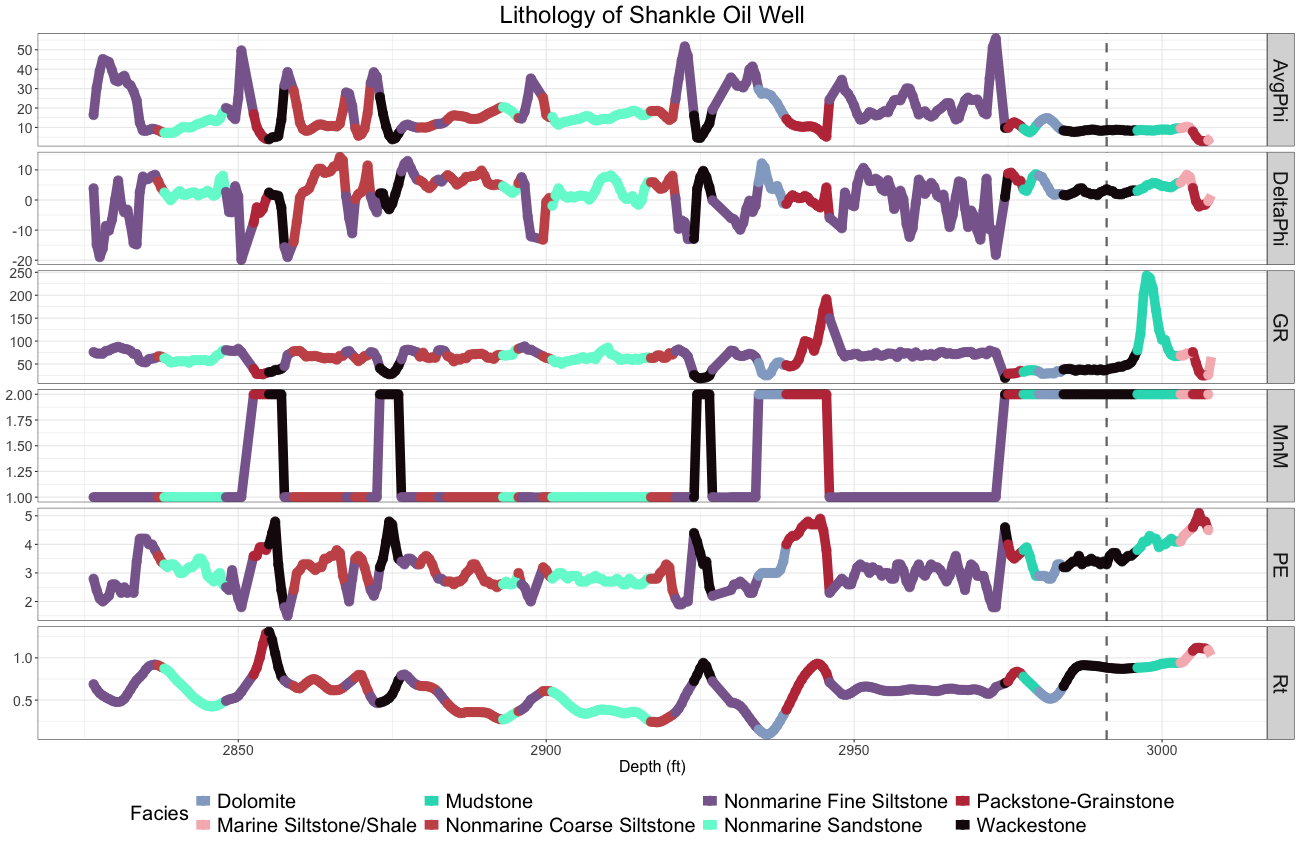}
    \caption{\textbf{\texttt{L2hdchange} Fit of Shankle Well Log}. $\hat{L}=1$ estimated change (\textcolor{black}{\makebox[10pt][l]{\hdashrule[0.5ex]{10pt}{1.5pt}{4pt 1pt}}}) in the facies of the Shankle well from method of \cite{Li23}.}
    \label{fig:well_l2}
\end{figure}

\end{document}